\documentclass[12pt,beltcrest]{ociamthesis} 
\pdfoutput=1

\usepackage{mathtools} 
\usepackage{amssymb} 
\usepackage{bbold} 
\usepackage{amsthm} 
\usepackage{stmaryrd} 

\usepackage{relsize} 
\usepackage{microtype} 
\usepackage{csquotes} 
\usepackage{enumerate}
\usepackage{cancel}

\usepackage{graphicx} 
\usepackage[usenames,dvipsnames]{xcolor} 
\usepackage{tikz} 
\usepackage{circuitikz} 
\usetikzlibrary{
	arrows,
	shapes,
	decorations,
	intersections,
	backgrounds,
	positioning,
	circuits.ee.IEC
	}

		\newcounter{theorem_c} 
		\numberwithin{theorem_c}{chapter} 
		\numberwithin{equation}{chapter} 
		\newcommand\numberthis{\addtocounter{equation}{1}\tag{\theequation}} 

		\theoremstyle{plain} 
		\newtheorem{theorem}[theorem_c]{Theorem}
		
		\newtheorem{lemma}[theorem_c]{Lemma}
		\newtheorem{corollary}[theorem_c]{Corollary}
		
		\newtheoremstyle{exampstyle}
		  {2mm} 
		  {2mm} 
		  {\itshape} 
		  {} 
		  {\bfseries} 
		  {.} 
		  {.5em} 
		  {} 

		\theoremstyle{exampstyle}
		\newtheorem{definition}[theorem_c]{Definition}
		\newtheorem{example}[theorem_c]{Example}
		\newtheorem{remark}[theorem_c]{Remark}

	\newcommand{\emptyArg}{\,\underline{\hspace{6px}}\,} 
	\newcommand{\inlineQuote}[1]{\textquotedblleft #1\textquotedblright} 
	\newcommand{\goodchi}{\protect\raisebox{2pt}{$\chi$}} 
	\newcommand{\goodrho}{\protect\raisebox{2pt}{$\rho$}} 
	\newcommand{\goodvdots}{\protect\raisebox{7pt}{\vdots}} 
	\newcommand{\verteq}{\rotatebox{90}{$\,=$}}

	\newcommand{\Powerset}[1]{\mathcal{P}(#1)} 
	\newcommand{\naturals}{\mathbb{N}} 
	\newcommand{\integers}{\mathbb{Z}} 
	\newcommand{\circleGroup}{\mathbb{T}} 
	\newcommand{\torusGroup}[1]{\circleGroup^{#1}} 
	\newcommand{\rationals}{\mathbb{Q}} 
	\newcommand{\reals}{\mathbb{R}} 
	\newcommand{\complexs}{\mathbb{C}} 
	\newcommand{\integersMod}[1]{\mathbb{Z}_{#1}} 
	\newcommand{\splitComplexs}{\mathbb{C}[\sqrt{1}]} 
	\newcommand{\finiteField}[1]{\mathbb{F}_{\!\!#1}} 
	\newcommand{\modclass}[2]{#1 \; (\text{mod } #2)} 
	\newcommand{\restrict}[2]{\left. #1 \right\vert_{#2}} 
	\newcommand{\domain}[1]{\operatorname{dom}#1} 
	\newcommand{\support}[1]{\operatorname{supp}#1} 
	\newcommand{\inject}{\hookrightarrow} 
	\newcommand{\nonstd}[1]{\,^\star #1}
	\newcommand{\starNaturals}{\nonstd{\naturals}} 
	\newcommand{\starIntegers}{\nonstd{\integers}} 
	\newcommand{\starRationals}{\nonstd{\rationals}} 
	\newcommand{\starComplexs}{\nonstd{\complexs}} 
	\newcommand{\starReals}{\nonstd{\reals}} 
	
	\newcommand{\normalSubgroup}{\trianglelefteq}

	\newcommand{\iffdef}{\stackrel{def}{\iff}} 
	\newcommand{\suchthat}[2]{\left\{#1 \: \middle\vert \: #2\right\}} 

	\newcommand{\stdpartSym}{\operatorname{st}}
	\newcommand{\stdpart}[1]{\stdpartSym(#1)}
	\newcommand{\indexSet}[1]{\{1,...,\dim{#1}\}}
	\newcommand{\truncate}[1]{\bar{#1}}
	\newcommand{\liftSym}[1]{\operatorname{lift}_{#1}}
	\newcommand{\lift}[2]{\liftSym{#2}[#1]}

	\newcommand{\starIntegersMod}[1]{{\nonstd{\integersMod{#1}}}}
	\newcommand{\starIntegersModPow}[2]{{\nonstd{\integersMod{#1}^{#2}}}}

		\newcommand{\ket}[1]{\vert #1 \rangle} 
		\newcommand{\bra}[1]{\langle #1 \vert} 
		\newcommand{\braket}[2]{\langle #1 \vert #2 \rangle} 
		\newcommand{\innerprod}[2]{\left( #1 , #2 \right)}
		\newcommand{\Trace}[1]{\operatorname{Tr}\,#1} 
		\newcommand{\Norm}[2]{|| #2 ||_{#1}} 
		\newcommand{\decohSym}{\operatorname{dec}} 
		\newcommand{\decoh}[1]{\decohSym_{#1}} 
		\newcommand{\CPMdoubled}[1]{\textbf{double}\left[#1\right]}

		\newcommand{\im}[1]{\operatorname{im}#1} 
		\newcommand{\Annihil}[1]{\operatorname{Ann}[#1]} 

		\newcommand{\Bounded}[1]{\operatorname{B}\left[#1\right]} 
		\newcommand{\LtwoSym}{\operatorname{L}^2} 
		\newcommand{\Ltwo}[1]{\LtwoSym[#1]} 
		\newcommand{\ltwoSym}{\ell^2} 
		\newcommand{\ltwo}[1]{\ltwoSym[#1]} 

		\newcommand{\SpaceH}{\mathcal{H}} 
		\newcommand{\SpaceG}{\mathcal{G}}
		\newcommand{\SpaceK}{\mathcal{K}}

		\newcommand{\GroupG}{\mathbb{G}}
		\newcommand{\GroupH}{\mathbb{H}}

		\newcommand{\isom}{\cong} 
		\newcommand{\id}[1]{id_{#1}} 

		\newcommand{\tensor}{\otimes} 
		\newcommand{\tensorUnit}{I} 

		\newcommand{\Hom}[3]{\operatorname{Hom}_{\,#1}\left[#2,#3\right]} 
		\newcommand{\Automs}[2]{\operatorname{Aut}_{\,#1}\left[#2\right]} 
		\newcommand{\UnitaryOps}[1]{U\left[#1\right]} 
		\newcommand{\Isoms}[3]{\operatorname{Iso}_{\,#1}\left[#2,#3\right]} 

		\newcommand{\SetCategory}{\operatorname{Set}} 

		\newcommand{\GrpCategory}{\operatorname{Grp}} 
		\newcommand{\AbCategory}{\operatorname{Ab}} 
		\newcommand{\fAbGrpCategory}{\operatorname{fAbGrp}} 
		\newcommand{\CMonCategory}{\operatorname{CMon}} 
		\newcommand{\RMatCategory}[1]{#1\operatorname{-Mat}} 
		\newcommand{\fRfreeModCategory}[1]{\RMatCategory{#1}} 
		\newcommand{\VectCategory}[1]{#1\operatorname{-Vect}} 
		\newcommand{\HilbCategory}{\operatorname{Hilb}} 
		\newcommand{\sHilbCategory}{\operatorname{sHilb}} 
		\newcommand{\fHilbCategory}{\operatorname{fHilb}} 
		\newcommand{\fdHilbCategory}{\fHilbCategory} 
		\newcommand{\RelCategory}{\operatorname{Rel}} 
		\newcommand{\fRelCategory}{\operatorname{fRel}} 
		\newcommand{\fPFunCategory}{\operatorname{fPFun}} 
		\newcommand{\fSetCategory}{\operatorname{fSet}} 
		\newcommand{\fStochCategory}{\operatorname{fStoch}} 
		\newcommand{\starHilbCategory}{^\star\!\HilbCategory} 
		\newcommand{\starHilbCategoryNearStd}{\starHilbCategory^{(std)}}

		\newcommand{\CategoryC}{\mathcal{C}}
		\newcommand{\CategoryD}{\mathcal{D}}

		\newcommand{\obj}[1]{\operatorname{obj} \, #1} 
		\newcommand{\OpCategory}[1]{#1^{\operatorname{op}}} 
		\newcommand{\CPMCategory}[1]{\operatorname{CPM}[#1]} 
		\newcommand{\CPStarCategory}[1]{\operatorname{CP}^\ast[#1]} 
		\newcommand{\CausalSubCategory}[1]{#1_{\hbox{\input{symbols/smallTraceSym.tex}}\!}}
		\newcommand{\DoubledCategory}[1]{\operatorname{Double}\left[#1\right]}
		\newcommand{\KaroubiEnvelope}[1]{\operatorname{Split}\left[#1\right]} 
		\newcommand{\DaggerKaroubiEnvelope}[1]{\operatorname{Split}^{\dagger}\left[#1\right]} 
		\newcommand{\CoherentGroupsCategory}[1]{\operatorname{QG}\left[#1\right]} 
		\newcommand{\WpAbQuantumGroupsCategory}[1]{\operatorname{wpAbQG}\left[#1\right]} 
		\newcommand{\underlyingGroup}[1]{\left\llbracket#1\right\rrbracket}
		\newcommand{\dualGroupWRTMonoid}[2]{#1^{\wedge_{#2}}}
		\newcommand{\CoherentGroupRepMonad}[1]{\emptyArg \otimes #1}
		\newcommand{\RepCategory}[1]{\operatorname{Rep}\left[#1\right]}
		\newcommand{\UnitaryRepCategory}[1]{\operatorname{Rep}^\dagger\left[#1\right]}
		\newcommand{\classicalSubcategory}[1]{#1_{K}} 

	\newcommand{\classicalStates}[1]{K(#1)} 
	\newcommand{\phaseGroup}[1]{P(#1)} 
	\newcommand{\phasegate}[1]{P_{#1}}
	
	\newcommand{\CPstates}[1]{\mathcal{L}[#1]}

	\newcommand{\sheafOfEventsSym}{\mathcal{E}} 
	\newcommand{\sheafOfEvents}[1]{\sheafOfEventsSym[#1]} 
	\newcommand{\restrictionMap}[2]{\operatorname{res}^{#1}_{#2}} 
	\newcommand{\distributionFunctorSym}[1]{\mathcal{D}_{#1}} 
	\newcommand{\distributionFunctor}[2]{\distributionFunctorSym{#1}[#2]} 
	\newcommand{\presheafOfDistributionsSym}[1]{\distributionFunctorSym{#1}\sheafOfEventsSym} 
	\newcommand{\presheafOfDistributions}[2]{\presheafOfDistributionsSym{#1}[#2]} 
	\newcommand{\supportSubpresheafSym}{\mathbb{S}} 
	\newcommand{\supportSubpresheaf}[1]{\supportSubpresheafSym[#1]} 

	\newcommand{\Xcolour}{Red}
	
	\newcommand{\Zcolour}{YellowGreen}

	\newcommand{\Xaltcolour}{Purple}
	
	\newcommand{\Zaltcolour}{Cyan}
	
	\newcommand{\Dcolour}{black!80}

	\newcommand{\Xbwcolour}{black!80}
	\newcommand{\XbwdotSym}{\hbox{\input{symbols/DdotSym.tex}}\!\!} 
	\newcommand{\XbwmultSym}{\!\hbox{\input{symbols/DmultSym.tex}}\!\!} 
	\newcommand{\XbwunitSym}{\!\hbox{\input{symbols/DunitSym.tex}}\!\!} 

	\newcommand{\Zbwcolour}{white}
	\newcommand{\ZbwdotSym}{\hbox{\input{symbols/ZbwdotSym.tex}}\!\!} 

	\newcommand{\Ybwcolour}{black!15}
	\newcommand{\YbwdotSym}{\hbox{\input{symbols/YbwdotSym.tex}}\!\!} 

	\newcommand{\Wbwcolour}{black!50}
	\newcommand{\WbwdotSym}{\hbox{\input{symbols/WbwdotSym.tex}}\!} 

	\newcommand{\trace}[1]{\hbox{\input{symbols/traceSym.tex}}\!_{#1}} 

	\tikzset{
	  rectangle with rounded corners north west/.initial=4pt,
	  rectangle with rounded corners south west/.initial=4pt,
	  rectangle with rounded corners north east/.initial=4pt,
	  rectangle with rounded corners south east/.initial=4pt,
	}
	\makeatletter
	\pgfdeclareshape{rectangle with rounded corners}{
	  \inheritsavedanchors[from=rectangle] 
	  \inheritanchorborder[from=rectangle]
	  \inheritanchor[from=rectangle]{center}
	  \inheritanchor[from=rectangle]{north}
	  \inheritanchor[from=rectangle]{south}
	  \inheritanchor[from=rectangle]{west}
	  \inheritanchor[from=rectangle]{east}
	  \inheritanchor[from=rectangle]{north east}
	  \inheritanchor[from=rectangle]{south east}
	  \inheritanchor[from=rectangle]{north west}
	  \inheritanchor[from=rectangle]{south west}
	  \backgroundpath{
	    \southwest \pgf@xa=\pgf@x \pgf@ya=\pgf@y
	    \northeast \pgf@xb=\pgf@x \pgf@yb=\pgf@y
	    \pgfkeysgetvalue{/tikz/rectangle with rounded corners north west}{\pgf@rectc}
	    \pgfsetcornersarced{\pgfpoint{\pgf@rectc}{\pgf@rectc}}
	    \pgfpathmoveto{\pgfpoint{\pgf@xa}{\pgf@ya}}
	    \pgfpathlineto{\pgfpoint{\pgf@xa}{\pgf@yb}}
	    \pgfkeysgetvalue{/tikz/rectangle with rounded corners north east}{\pgf@rectc}
	    \pgfsetcornersarced{\pgfpoint{\pgf@rectc}{\pgf@rectc}}
	    \pgfpathlineto{\pgfpoint{\pgf@xb}{\pgf@yb}}
	    \pgfkeysgetvalue{/tikz/rectangle with rounded corners south east}{\pgf@rectc}
	    \pgfsetcornersarced{\pgfpoint{\pgf@rectc}{\pgf@rectc}}
	    \pgfpathlineto{\pgfpoint{\pgf@xb}{\pgf@ya}}
	    \pgfkeysgetvalue{/tikz/rectangle with rounded corners south west}{\pgf@rectc}
	    \pgfsetcornersarced{\pgfpoint{\pgf@rectc}{\pgf@rectc}}
	    \pgfpathclose
	 }
	}
	\makeatother

	\tikzset{->-/.style={decoration={markings,mark=at position #1 with {\arrow{>}}},postaction={decorate}}}
	\tikzset{-<-/.style={decoration={markings,mark=at position #1 with {\arrow{<}}},postaction={decorate}}}

	\tikzstyle{every picture}=[baseline=-0.25em,scale=0.5]
	\pgfdeclarelayer{edgelayer}
	\pgfdeclarelayer{nodelayer}
	\pgfsetlayers{edgelayer,nodelayer,main}

	\tikzstyle{box} = [draw,shape=rectangle,inner sep=2pt,minimum height=6mm,minimum width=6mm,fill=white] 
	\tikzstyle{boxlarge} = [draw,shape=rectangle,inner sep=2pt,minimum height=1.5cm,minimum width=8mm,fill=white] 
	\tikzstyle{boxLarge} = [draw,shape=rectangle,inner sep=2pt,minimum height=2cm,minimum width=10mm,fill=white] 
	\tikzstyle{boxsmall} = [draw,shape=rectangle,inner sep=2pt,minimum height=3mm,minimum width=3mm,fill=white] 
	\tikzstyle{dot} = [inner sep=0mm,minimum width=3mm,minimum height=3mm,draw,shape=circle,text depth=-0.1mm]
	\tikzstyle{Zbwdot} = [dot, fill=\Zbwcolour]
	\tikzstyle{Xbwdot} = [dot, fill=\Xbwcolour]
	\tikzstyle{Ybwdot} = [dot, fill=\Ybwcolour]
	\tikzstyle{Wbwdot} = [dot, fill=\Wbwcolour]
	\tikzstyle{antipode} = [boxsmall] 

	\tikzstyle{state} = [draw, rectangle with rounded corners,
	  rectangle with rounded corners north west=8pt,
	  rectangle with rounded corners south west=8pt,
	  rectangle with rounded corners north east=0pt,
	  rectangle with rounded corners south east=0pt,
	,inner sep=2pt,minimum height=6mm,minimum width=6mm,fill=white]
	\tikzstyle{statelarge} = [draw, rectangle with rounded corners,
	  rectangle with rounded corners north west=8pt,
	  rectangle with rounded corners south west=8pt,
	  rectangle with rounded corners north east=0pt,
	  rectangle with rounded corners south east=0pt,
	,inner sep=2pt,minimum height=1.5cm,minimum width=8mm,fill=white]
	\tikzstyle{stateLarge} = [draw, rectangle with rounded corners,
	  rectangle with rounded corners north west=8pt,
	  rectangle with rounded corners south west=8pt,
	  rectangle with rounded corners north east=0pt,
	  rectangle with rounded corners south east=0pt,
	,inner sep=2pt,minimum height=2cm,minimum width=8mm,fill=white]
	\tikzstyle{effect} = [draw, rectangle with rounded corners,
	  rectangle with rounded corners north west=0pt,
	  rectangle with rounded corners south west=0pt,
	  rectangle with rounded corners north east=8pt,
	  rectangle with rounded corners south east=8pt,
	,inner sep=2pt,minimum height=6mm,minimum width=6mm,fill=white]
	\tikzstyle{scalar}=[diamond,draw,inner sep=1pt,font=\small,fill=white]

	\tikzstyle{cdnode}=[fill=white]
	\tikzstyle{labelnode}=[fill=white]
	\tikzstyle{tightlabelnode}=[fill=white,inner sep = 0.1mm]
	\tikzstyle{none}=[inner sep=0pt]
	\tikzstyle{whiteline}=[-, line width=4pt, draw=white]

	\tikzstyle{trace}=[circuit ee IEC,thick,ground,scale=2.5]
	\tikzstyle{cotrace}=[circuit ee IEC,thick,ground,rotate=180,scale=2.5]
	\tikzstyle{upground}=[circuit ee IEC,thick,ground,rotate=90,scale=2.5]
	\tikzstyle{downground}=[circuit ee IEC,thick,ground,rotate=-90,scale=2.5]

	\tikzstyle{doubled} = [line width=1.8pt] 
	
	\tikzstyle{empty diagram}=[draw=gray!40!white,dashed,shape=rectangle,minimum width=1cm,minimum height=1cm]

\usepackage{silence}
\title{Categorical Quantum Dynamics}
\author{Stefano Gogioso}
\college{Trinity College}  	  
\degree{Doctor of Philosophy} 
\degreedate{Michaelmas 2016}  

\begin{document}

\baselineskip=18pt plus1pt

\setcounter{secnumdepth}{3}
\setcounter{tocdepth}{2} 	 

\maketitle

\begin{center}
Per Aspera Ad Astra
\end{center}

\begin{romanpages}           
\tableofcontents             
\end{romanpages}             







\chapter{Introduction} 
\label{section_Introduction}

\vspace{-0.7cm}
\section{Summary of this work}

Since their original introduction, strongly complementary observables have been a fundamental ingredient of the ZX calculus, one of the most successful fragments of Categorical Quantum Mechanics (CQM). In this thesis, we show that strong complementarity plays a vastly greater role in quantum theory.

Firstly, we use strong complementarity to introduce dynamics and symmetries within the framework of CQM, which we also extend to infinite-dimensional separable Hilbert spaces: these were long-missing features, which open the way to a wealth of new applications. The coherent treatment presented in this work also provides a variety of novel insights into the dynamics and symmetries of quantum systems: examples include the extremely simple characterisation of symmetry-observable duality, the connection of strong complementarity with the Weyl Canonical Commutation Relations, the generalisations of Feynman's clock construction, the existence of time observables and the emergence of quantum clocks.

Secondly, we show that strong complementarity is a key resource for quantum algorithms and protocols. We provide the first fully diagrammatic, theory-independent proof of correctness for the quantum algorithm solving the Hidden Subgroup Problem, and show that strong complementarity is the feature providing the quantum advantage. In quantum foundations, we use strong complementarity to derive the exact conditions relating non-locality to the structure of phase groups, within the context of Mermin-type non-locality arguments. Our non-locality results find further application to quantum cryptography, where we use them to define a quantum-classical secret sharing scheme with provable device-independent security guarantees.

All in all, we argue that strong complementarity is a truly powerful and versatile building block for quantum theory and its applications, and one that should draw a lot more attention in the future.

\newpage
\section{Background literature}

\subsection{Categorical Quantum Mechanics}

This work takes its roots in the framework of \textbf{categorical quantum mechanics} (CQM) \cite{Abramsky2009,Coecke2015}, which its extends and refines in a number of aspects. The general motivation behind the application of category-theoretic tools lies in the intuition that the features distinguishing quantum theory from classical physics can be understood in terms of the way quantum processes compose, sequentially and in parallel (we refer to this as an \textbf{operational}, or \textbf{process-theoretic}, description of quantum theory). The framework of symmetric monoidal categories is particularly suited for operational descriptions, and comes with a natural diagrammatic formalism \cite{Joyal1991,Baez2010,Coecke2015,Kissinger2012,Selinger2007,Coecke2012d}, combining the rigour of the category theory with the versatility of graphical manipulation. The diagrammatic formalism has quickly become a major selling point of CQM, and receives plenty of interest on its own \cite{Coecke2011,Backens2014,Hadzihasanovic2015,Kissinger2012,Kissinger2015}.

The operational and process-theoretic description of quantum theory given by CQM \cite{Abramsky2012,Coecke2015} is in a certain sense antipodal to the traditional Hilbert space formulation, which heavily relies on explicit complex-linear structure. Several intermediate approaches exist, such as quantum circuits and generalised probabilistic theories \cite{Hardy2001,Barrett2007}, and have been applied to a diversity of topics in quantum information and foundations. The framework of \textbf{operational probabilistic theories} (OPTs) is perhaps the closest, in spirit, to CQM: it was developed in \cite{Chiribella2010,Chiribella2011,Chiribella2015} with the aim of obtaining an informational derivation of quantum theory, it comes with a graphical calculus and it has a strong operational and process-theoretic flavour to it; work to connect OPTs to CQM is currently being undertaken \cite{Tull2016,Gogioso2016f}. Despite the diagrammatic and operational similarities, the approach to quantum theory of OPTs is quite different from that of CQM: the former relies on explicit probabilistic structure, and is mainly axiomatic in nature, while the latter focuses on categorical structures (with no explicit summations or probabilities), and is mostly constructive. The processes which OPTs are concerned with are physical processes\footnote{Such as quantum channels, aka completely positive trace-preserving maps.}, or processes appearing in their convex decompositions\footnote{Such as sub-normalised pure states, appearing in the convex decomposition of density matrices.}. On the contrary, CQM is concerned with the study of arbitrary processes with interesting categorical properties, which are used as building blocks of physical processes. In a nutshell, OPTs take physical processes and impose axioms on the way they can be probabilistically decomposed, while CQM studies how certain building blocks compose and interact to form larger processes of physical interest.

\subsection{Dagger compact structure}

The categorical environment of choice for pure-state quantum theory in CQM is that of dagger compact symmetric monoidal categories: the dagger structure is the categorical abstraction of operator adjunction\footnote{And of bra/ket duality as a special case.} in the Hilbert space formalism, while the compact structure corresponds to operator-state duality. Within this environment, the most common building blocks used by CQM are certainly $\dagger$-Frobenius algebras, which provide the coherent/pure versions of the classical operations of copy, delete and match, and are a key component in the treatment of classicality \cite{Coecke2008,Coecke2010}. Special commutative $\dagger$-Frobenius algebras correspond to orthonormal bases \cite{Coecke2013b}, or equivalently to non-degenerate observables in the Hilbert space formalism. 

The $\dagger$-compact structure can furthermore be used to provide a categorical environment for the operational treatment of mixed-state quantum theory: the CPM construction \cite{Selinger2007,Coecke2012env} represents quantum channels as a diagrammatic version of their Kraus decomposition in the Hilbert space formalism\footnote{A similarity which becomes even more apparent in the more general CP construction \cite{Coecke2012d}.}, and a partial trace naturally arises. In the CPM formalism, $\dagger$-Frobenius algebras can be used to define decoherence maps for all orthonormal bases, and the ensuing mixed quantum-classical formalism can be given categorical dignity via the CP* construction \cite{Coecke2014a,Cunningham2015}, which connects the operational picture to the study of C*-algebras and algebraic quantum theory.

\subsection{Some application of CQM}

Throughout the years, the categorical and diagrammatic formalisms have yielded many novel characterisations of quantum structures, with applications to quantum foundations, information and computation. In quantum foundations, the CQM framework has been applied to study features of causality \cite{Coecke2013b,Coecke2015,Coecke2016} and non-locality \cite{Coecke2012c,Gogioso2015c}, both operationally and in connection with the sheaf-theoretic framework for contextuality of \cite{Abramsky2011,Abramsky2012a,Abramsky2015} (which is also applicable to OPTs \cite{Chiribella2016}). In quantum information and computation, the framework has been applied to the study of quantum algorithms and protocols \cite{Coecke2011,Vicary2012a,Zeng2014,Vicary2012,Verdon2016,Coecke2016a,Zamdzhiev2012}, measurement-based and cluster-state quantum computing \cite{Duncan2015,Horsman2011}, complementarity \cite{Musto2015,Duncan2016,Zeng2014}, and the information theoretic characterisation of quantum theory \cite{Heunen2016}. Relational models for non-deterministic classical computation have also been explored using tools from CQM, both as toy models for quantum theory \cite{Pavlovic2009,Abramsky2013,Gogioso2015b,Marsden,Coecke2016,Coecke2012a,Backens2015} and in their own right as models of computation \cite{Pavlovic2009,Bar2014}.

\subsection{The ZX calculus}

One of the most successful and intriguing fragments of CQM is the \textbf{ZX calculus}. First introduced in \cite{Coecke2011}, the ZX calculus is a diagrammatic graphical calculus for multi-qubit systems, designed to reason formally about quantum algorithms and protocols, as well as to derive rigorous results on quantum foundations and information. Its simple but rigorous presentation makes the ZX calculus ideal for diagrammatic reasoning and automated proof-checking \cite{Kissinger2012,Kissinger2015}. The ZX calculus has been shown to be \textit{universal} and \textit{sound} for pure qubit quantum mechanics \cite{Coecke2011}, as well as \textit{complete} for pure qubit \textit{stabiliser} quantum mechanics \cite{Backens2014,Backens2015}. However, the ZX calculus is known to be incomplete for pure qubit quantum mechanics \cite{DeWitt2014}.

Since its introduction, the ZX calculus has found plenty of applications in quantum information: it was used to describe the fundamental gates of quantum circuits \cite{Coecke2011,Duncan2015}, to formalise a number of quantum protocols \cite{Coecke2011,Zamdzhiev2012,Coecke2016a}, to describe the logical structure of information flow in topological cluster-state quantum computing \cite{Horsman2011}, and to provide a categorical model of Spekkens's toy theory \cite{Coecke2012a,Backens2015}.

\subsection{Applications of strong complementarity}

In several applications of the ZX calculus, key steps are performed by a specific set of rules relating the Z and X observables, namely the \textbf{bialgebra law} and \textbf{coherence laws}. It is possible to show that classical structures satisfying these rules also satisfy the Hopf law \cite{Duncan2016,Kissinger2012}, and that the associated observables are \textbf{complementary} (or \textbf{mutually unbiased}): as a consequence, the property defined by the the bialgebra and coherence laws is often referred to as \textbf{strong complementarity}. 

Complementarity plays an important role in the correctness and security of certain quantum protocols, but its classification in arbitrary dimensions has proven to be remarkably tricky; strongly complementary pairs of non-degenerate observables on finite-dimensional Hilbert spaces, on the other hand, are completely classified by finite abelian groups \cite{Kissinger2012}, and hence much easier to work with. 

A number of applications of CQM rely on strong complementarity as their active ingredient: most relevant are its appearance as an abstract version of the quantum Fourier transform in group-theoretic quantum algorithms \cite{Vicary2012a,Zeng2014,Zeng2015,Gogioso2015d}, and the role it plays in connecting Mermin-type non-locality scenarios to the structure of phase groups in abstract process theories \cite{Coecke2012c,Coecke2010a,Gogioso2015}.

\newpage
\section{Brief synopsis of this work}

\subsection{A primer of CQM}
In Chapter \ref{chapter_CQM}, we provide a primer of the framework of Categorical Quantum Mechanics. We begin by introducing symmetric monoidal categories as a general model for process theories, and we characterise dagger compact structure in relation to inner products and operator-state duality, thus capturing what we believe to be the essential operational and structural features of pure-state quantum theory. 

In order to model mixed-state behaviour, we proceed to present environment structure and the CPM construction. We introduce $\dagger$-Frobenius algebras in relation to observables, classicality and the coherent manipulation of data\footnote{
	When talking about \inlineQuote{coherent data}, or the \inlineQuote{coherent encoding/manipulation} of classical data, we will be talking about classical data which has been encoded into quantum systems via orthonormal bases, and is manipulated using $\complexs$-linear extensions of classical deterministic functions. We will argue in Chapter \ref{chapter_CQM} that coherent data plays a huge role in quantum computing and in the foundations of quantum theory: it is used to construct oracles, entangled states, as well as many other building blocks used in the design of quantum protocols. Coherent data enjoys all those features of $\complexs$-linearity which coalesce to provide quantum advantage and non-classical behaviour, while at the same time allowing for a good deal of classical intuition to play a positive role in designing quantum algorithms. 
}, and we use them to construct measurements, preparations and decoherence maps in the context of the CP* construction. Finally, we provide a recap of the sheaf-theoretic framework for non-locality and contextuality, and connect it to our portrayal of process theories.

\subsection{Coherent dynamics and symmetries}
In Chapter \ref{chapter_CQD} we use strong complementarity to introduce symmetries and dynamics within the framework of CQM. Our approach relies on a coherent treatment of group theory and representation theory, similar in spirit to the way in which $\dagger$-Frobenius algebras provide a coherent treatment of classical data manipulation. 

We define a new notion of \inlineQuote{coherent group}, based on strongly complementary pairs of quasi-special $\dagger$-Frobenius algebras and modelling an abstract coherent counterpart of classical groups. By analogy with the relatable case of finite periodic lattices, we prove general results about symmetry-observable duality, the Weyl canonical commutation relations and a weak form of the uncertainty principle, and we establish that coherent groups show the same fundamental structural and operational features that would be expected of position/momentum pairs. 

We consider representations of coherent groups as a model of coherent symmetric systems, in analogy with the traditional identification of classical symmetric systems with classical group representations, and we characterise them as the Eilenberg-Moore algebras of a certain monad, with equivariant maps arising as the Eilenberg-Moore morphisms between them. 

We present a new framework for the treatment of infinite-dimensional separable Hilbert spaces in CQM, and explicitly construct a coherent group corresponding to the position/momentum pair for the textbook case of wavefunctions on a 1-dimensional continuous space with periodic boundary conditions (i.e. with translation group isomorphic to the circle group $S^1$). 

Finally, we apply all the techniques developed in the remainder of the Chapter to the study of quantum dynamics, with an explicit treatment of continuous periodic, discrete and discrete periodic dynamics of finite-dimensional and separable quantum systems. We talk about quantum dynamics, Hamiltonians and Schr\"{o}dinger's Equation. We use the symmetry/observable duality properties of our coherent framework to provide simple diagrammatic proofs of Stone's Theorem on 1-parameter unitary groups and von Neumann's Mean Ergodic Theorem, and we give an abstract proof of correctness for the the Feynman clock construction. We conclude the Section by turning our attention to the description of synchronised dynamical systems, the existence of time observables, and the emergence of quantum clocks.

\subsection{Strong complementarity in quantum algorithms}

In Chapter \ref{chapter_algos} we move away from quantum foundations, and we present two applications of strong complementarity to quantum algorithms. 

Firstly, we put to work the connection between strong complementarity and the quantum Fourier transform in a fully diagrammatic, theory-independent proof of correctness for the quantum subroutine of the algorithm solving the Hidden Subgroup Problem (HSP). In doing so, we definitively prove that strong complementarity is the structural feature providing the quantum advantage in the HSP. As an immediate application of our theory-independent approach, we conclude that Simon's problem can be efficiently solved in real quantum theory.

Secondly, we investigate the relationship between strong complementarity and phase groups, and we formulate a broad generalisation of Mermin's non-locality argument for GHZ states. Our results provide an exact characterisation of the relationship between phase groups and non-locality, bringing the research programme of \cite{Coecke2012c,Coecke2010a} to a close. We relate our findings to the framework of All-vs-Nothing arguments, and we put them to work in the definition of a device-independent quantum-classical secret sharing scheme, extending the classical scheme of Hillery, Bu\v{z}ek and Berthiaume \cite{Hillery1999}.







\newcommand{\morphismName}[1]{\left\lfloor#1\right\rfloor}

\chapter{Categorical Quantum Mechanics}
\label{chapter_CQM}

\section{Symmetric monoidal categories}
\label{section_SMCs}

\subsection{Objects as physical systems}

Categorical quantum mechanics is first and foremost a theory of systems and processes, composing sequentially and in parallel. \textbf{Symmetric monoidal categories} (henceforth SMCs) provide a very general framework to model such systems and processes, and come with a natural graphical/diagrammatic language \cite{Joyal1991,Joyal1996}; because of this, they are often referred to as \textbf{process theories} in the literature. 

The objects of the SMC are taken to correspond to physical systems, and are diagrammatically denoted by wires:
\begin{equation}\label{eqn_identityMorphism}
\begin{multlined}
\begin{tikzpicture}
	\begin{pgfonlayer}{nodelayer}
		\node [style=cdnode] (0) at (-2, 0) {$A$};
		\node [style=cdnode] (1) at (2, 0) {$A$};
	\end{pgfonlayer}
	\begin{pgfonlayer}{edgelayer}
		\draw [style=-] (0) to (1);
	\end{pgfonlayer}
\end{tikzpicture}
\end{multlined}
\end{equation}

\noindent The identity $\id{A}$ on an object/system $A$ is associated to the process of \inlineQuote{doing nothing} to the system, and is also denoted by the same undecorated wire which represents the system in Diagram~\ref{eqn_identityMorphism}. This free confusion between objects and identity morphisms is fairly common in category theory, and has an interesting physical interpretation: possession of a static system is the same thing as the process of continuing to do nothing with it.

\subsection{Sequential composition of processes}

The morphisms of the SMC are taken to correspond to processes connecting physical systems: a morphism $A \rightarrow B$ embodies a process taking a state in system $A$ as its input and returning a state in system $B$ as its output. The \textbf{sequential composition} of processes (with compatible intermediate systems) is embodied by composition of morphisms in the category, and assumed to be associative. Diagrammatically, a morphism/process $f:A \rightarrow B$ is denoted by a labelled box:
\begin{equation}\label{eqn_morphism}
\begin{multlined}
\begin{tikzpicture}
	\begin{pgfonlayer}{nodelayer}
		\node [style=box] (0) at (0, 0) {$f$};
		\node [style=cdnode] (1) at (-3, 0) {$A$};
		\node [style=cdnode] (2) at (3, 0) {$B$};
	\end{pgfonlayer}
	\begin{pgfonlayer}{edgelayer}
		\draw [style=-] (1) to (0);
		\draw [style=-] (0) to (2);
	\end{pgfonlayer}
\end{tikzpicture}
\end{multlined}
\end{equation}
Composition $g \circ f: A \rightarrow C$ of two morphisms $f: A \rightarrow B$ and $g: B \rightarrow C$ is denoted by composition of boxes along the intermediate wire:
\begin{equation}\label{eqn_morphismsComposition}
\begin{multlined}
\input{pictures/chapter2/morphismsComposition.tikz}
\end{multlined}
\end{equation}
Composition of multiple processes is obtained by repeated pairwise composition. Also, the labels for systems (input, output or intermediate) might be omitted when clear from context and/or when not relevant.

\subsection{Parallel composition of processes}

The description of parallel composition of processes also requires the description of parallel composition of systems: given two processes $f: A \rightarrow B$ and $g: C \rightarrow D$, their \textbf{parallel composition}, denoted $f \otimes g$, transforms the joint system formed by $A$ and $C$, denoted $A \otimes C$, into the joint system formed by $B$ and $D$, denoted $B \otimes D$. A minimal notion of parallel composition and joint systems is captured by \textit{monoidal categories}, and $\otimes$ is called the \textbf{tensor product}. Diagrammatically, a joint system $A \otimes C$ is denoted by parallel wires for $A$ and $C$, and parallel composition of processes is denoted by parallel boxes:
\begin{equation}\label{eqn_morphismsParallelComposition}
\begin{multlined}
\input{pictures/chapter2/morphismsParallelComposition.tikz}
\end{multlined}
\end{equation}
For a monoidal category $\CategoryC$, the tensor product is a functor $\otimes:\CategoryC \times \CategoryC \rightarrow \CategoryC$, and hence parallel composition and sequential composition distribute over each other:
\begin{equation}\label{eqn_morphismsParallelSequentialComposition}
\begin{multlined}
\input{pictures/chapter2/morphismsParallelSequentialComposition.tikz}
\end{multlined}
\end{equation}
While the diagrammatic formalism for categories was 1-dimensional, with processes represented by a line of wires and boxes, the diagrammatic formalism for monoidal categories is 2-dimensional, with processes represented by parallel wires and boxes on them. Boxes need not be confined to single wires, but can connect multiple wires together: there could be processes $h: B \otimes D \rightarrow E \otimes F$ which cannot be obtained as parallel composition of processes $B \rightarrow D$ and $E \rightarrow F$, and which need to be represented by boxes spanning multiple wires. 
\begin{equation}\label{eqn_morphismsParallelCompositionExample}
\begin{multlined}
\input{pictures/chapter2/morphismsParallelCompositionExample.tikz}
\end{multlined}
\end{equation}
Morphisms which arise from parallel composition are said to be \textbf{separable}, and may (but need not) be represented by parallel boxes on parallel wires, as done in Equation~\ref{eqn_morphismsParallelComposition}. Conversely, processes represented by parallel boxes on parallel wires are necessarily separable.

Adequately capturing the notion of joint system turns out to be slightly problematic: the axioms of monoidal categories do not guarantee that $A \otimes B$ and $B \otimes A$ will be the same system, or even isomorphic: the linear order which systems are parallely composed in is relevant. To obviate this problem, one might introduce a natural isomorphism $\sigma_{A,B}$, called the \textbf{symmetry isomorphism}, which swaps $A$ and $B$, so that $A \otimes B$ and $B \otimes A$ are effectively the same system, and the linear order becomes irrelevant; this leads to the definition of \textit{symmetric monoidal categories}. In the graphical presentation, the symmetry isomorphism $\sigma_{A,B}$ is represented as a crossing of the wires and the naturality condition guarantees that processes can be made to \inlineQuote{slide} through the wire crossing:
\begin{equation}\label{eqn_braidingNaturality}
\begin{multlined}
\input{pictures/chapter2/braidingNaturality.tikz}
\end{multlined}
\end{equation}
The symmetry isomorphism comes with the additional requirement that $\sigma_{B,A} = \sigma_{A,B}^{-1}$:
\begin{equation}\label{eqn_braidingTwice}
\begin{multlined}
\input{pictures/chapter2/braidingTwice.tikz}
\end{multlined}
\end{equation}

\subsection{States, effects and scalars}

The tensor product is associative, and comes with a \textbf{tensor unit} $\tensorUnit$: categorically, this means that there are natural isomorphisms $\alpha_{A,B,C}: (A \otimes B) \otimes C \rightarrow A \otimes (B \otimes C)$ (the \textbf{associator}), $\rho_{A}: A \otimes \tensorUnit \rightarrow A$ (the \textbf{right unitor}), and $\lambda_{A}: \tensorUnit \otimes A \rightarrow A$ (the \textbf{left unitor}). The tensor unit models the trivial system, and processes from/to the tensor unit have a special status. Processes $\tensorUnit \rightarrow A$ are called the \textbf{states} of $A$ and processes $A \rightarrow \tensorUnit$ are called the \textbf{effects} of $A$. Processes $\tensorUnit \rightarrow \tensorUnit$ are called the \textbf{scalars} of the SMC, and they form a commutative monoid under composition\footnote{Both sequential and parallel: for scalars $x,y : \tensorUnit \rightarrow \tensorUnit$, sequential composition ($y \circ x$) coincides with parallel composition followed by a unitor ($\lambda_\tensorUnit (x \otimes y)$ and/or $\rho_\tensorUnit (x \otimes y)$).}, with the identity $\id{\tensorUnit}$ as its unit. Diagrammatically, the wire for the tensor unit is almost always omitted\footnote{One can always use unitors and their inverses to insert/remove tensor units as needed to correctly match input/output systems for processes in the diagrams. In this work, unitors and associators are always omitted, as their application is obvious from context.}. Hence, states are processes with no input wires, effects are processes with no output wires, and scalars are processes with no wires attached at all:
\begin{equation}\label{eqn_statesEffectsScalars}
\begin{multlined}
\input{pictures/chapter2/statesEffectsScalars.tikz}
\end{multlined}
\end{equation}
For any systems $A$ and $B$, the processes $A \rightarrow B$ come with a canonical monoid action of the scalars on them, given by the tensor product:
\begin{equation}\label{eqn_scalarsAction}
\begin{multlined}
\input{pictures/chapter2/scalarsAction.tikz}
\end{multlined}
\end{equation}
Diagrammatically, the scalar $1$ (our alternative notation for $\id{\tensorUnit}$ from now on) is a special case: being the identity of the trivial system and acting as the identity on processes, it is usually not represented at all, or equivalently it is represented by an empty diagram.

States, effects and scalars provide a point of connection between SMCs and set-theoretic formulations. A process $f: A \rightarrow B$ in a SMC can be seen in three different ways as a set function: (i)  mapping states $\psi : \tensorUnit \rightarrow A$ of $A$ to states $f \circ \psi : \tensorUnit \rightarrow B$ of $B$; (ii) mapping effects $b : B \rightarrow \tensorUnit$ of $B$ to effects $b \circ f : A \rightarrow \tensorUnit$ of $A$; (iii)  mapping pairs $(\psi,b)$ of a state on $A$ and an effect on $B$ to scalars $b \circ f \circ \psi$.
We say that a SMC has \textbf{enough states} if any two processes $f,g:A \rightarrow B$ are equal whenever they are equal as functions on the states of $A$ . Dually, we say that a SMC has \textbf{enough effects} if any two processes $f,g:A \rightarrow B$ are equal whenever they are equal as functions on the effects of $B$. We say that a SMC has \textbf{enough states and effects} if any two processes $f,g:A \rightarrow B$ are equal whenever they are equal as functions mapping pairs of a state $A$ and an effect of $B$ to scalars. 

\subsection{Examples of symmetric monoidal categories}
\label{subsection_ExamplesSMCs}

\paragraph{The category of sets.} The category $\SetCategory$ having sets as objects and functions as morphisms is symmetric monoidal, with function composition as sequential composition of processes. The tensor product on objects is given by the Cartesian product $A \times B$ of sets, with the singleton set $\mathbb{1}$ as tensor unit, while the parallel composition of morphisms is given by product of functions $f \times g := \big(a,b\big) \mapsto \big(f(a),g(b)\big)$. States of an object $A$ in $\SetCategory$ are exactly the elements of set $A$. Because $\mathbb{1}$ is terminal, there is a unique effect $A \rightarrow \mathbb{1}$ for each $A$, and hence a unique scalar $\id{\mathbb{1}}$; hence $\SetCategory$ has enough states, but not enough effects. Every category with finite products and terminal object similarly forms a symmetric monoidal category (but need not have enough states, e.g. the category of groups). 

\paragraph{The category of Hilbert spaces.} The category $\HilbCategory$ having Hilbert spaces as objects and bounded linear maps as morphisms is symmetric monoidal, with function composition as sequential composition of processes. The tensor product on objects is that of Hilbert spaces $\ltwo{A} \otimes \ltwo{B} \isom \ltwo{A \times B}$, with $\complexs$ as tensor unit. The tensor product on morphisms is the Kronecker product. The states of a Hilbert space $\ltwo{A}$ are exactly the vectors in $\ltwo{A}$, where vector $\ket{\psi}$ is seen as the map $\complexs \rightarrow \ltwo{A}$ sending $\xi \mapsto \xi \ket{\psi}$; the effects of a Hilbert space $\ltwo{A}$ are exactly the continuous linear functionals on it; the scalars are the complex numbers $\complexs$, with $\otimes$ as multiplication. The category $\HilbCategory$ has enough states and effects.

\paragraph{Categories of matrices over semirings.} The category of finite-dimensional Hilbert spaces $\fdHilbCategory$ (a full subcategory of $\HilbCategory$) presents no issues of continuity, and its construction can be straightforwardly extended from $\complexs$ to an arbitrary commutative semiring $R$. The category $\fRfreeModCategory{R}$ of free, finite-dimensional semimodules over $R$ has objects in the form $R^X$, where $X$ is a finite set, and morphisms $R^X \rightarrow R^Y$ are the $Y \times X$ matrices with values in $R$. Sequential composition is matrix composition\footnote{A semiring has the minimal requirements for matrix composition, namely an addition and a multiplication with appropriate distributivity.}, and parallel composition is Kronecker product of matrices. The states of $R^X$ are $X$-indexed, $R$-valued vectors; the effects of $R^X$ are $R$-linear functionals $R^X \rightarrow R$; the scalars form the semiring $R$, with both $\circ$ and $\otimes$ coinciding with the semiring multiplication. The category $\fRfreeModCategory{R}$ always has enough states and effects. This is a large family of categories, which includes a number of examples of interest for CQM and related disciplines.
\begin{enumerate}
	\item[(i)] The category $\fdHilbCategory$ of finite-dimensional Hilbert spaces\footnote{The same as the category $\VectCategory{\complexs}$ of finite-dimensional complex vector spaces.}, for $R = \complexs$. This is the traditional arena of pure-state quantum theory.
	\item[(ii)] The category $\VectCategory{\reals}$ of finite-dimensional real vector spaces, for $R = \reals$. This is the arena of pure-state real quantum theory.
	\item[(iii)] The category $\fRelCategory$ of finite sets and relations between them, for $R = \mathbb{B}$, the booleans. This is the arena of non-deterministic computation, and provides a toy model for pure-state quantum theory.
	\item[(iv)] The category of finite-dimensional convex cones over simplices, for $R = \reals^+$. This is the arena of classical probabilistic systems.
	\item[(v)] The category of multi-sets and \inlineQuote{multi-relations}, for $R= \naturals$. 
\end{enumerate}
In computer science, semirings are often used to model resources in computation.
\begin{enumerate}
	\item[(vi)] The \textit{boolean semiring} $(\mathbb{B},\vee,\wedge)$ is used to model non-deterministic computation (related to the complexity class NP).
	\item[(vii)] The \textit{probability semiring} $(\reals^{+},+,\times)$ is used to model probabilistic computation (related to the complexity classes BPP and PP).
	\item[(viii)] The \textit{natural numbers} $(\naturals,+,\times)$ is used to model counting problems (related to the complexity class $\#$P).
	\item[(ix)] The \textit{tropical semirings} $(M, \min , + )$---where $(M,+,0)$ is a totally ordered monoid with an absorbing maximal element $\infty$---is used in the Floyd-Warshall algorithm \cite{floyd1962shortest} for shortest-path finding in directed graphs (and a number of other optimisation problems solvable within the framework of \textit{tropical geometry} \cite{simon1988tropical,pin1998tropical,mikhalkin2004amoebas,speyer2009tropical}).
	\item[(x)] The \textit{Viterbi semiring} $([0,1], \max, \times)$ is used by the Viterbi algorithm \cite{viterbi1967decoding} to find the most likely sequence of hidden states in a Hidden Markov Model.
	\item[(xi)] Locales (related to intuitionistic logic) and commutative unital quantales (related to linear logic \cite{Yetter1990} and generalised metrics for topological spaces \cite{Kopperman1988}) both find a number of applications in programming semantics and other areas of computer science \cite{Abramsky1987,Abramsky1993,Flagg1997}.
\end{enumerate}
Because of these examples, the category $\RMatCategory{R}$ can be interpreted as modelling  classical computation with resources encoded by a semiring $R$. However, the $R=\reals^+,\mathbb{B},\reals$ cases also suggest an interpretation of $\RMatCategory{R}$ as modelling classical non-deterministic systems, with the semiring $R$ encoding a notion of non-determinism: when using this interpretation, we will refer to $\RMatCategory{R}$ as the category of \textbf{classical $R$-probabilistic systems}. It should be noted that both the category $\fPFunCategory$ of finite sets and partial functions (modelling partial deterministic classical computation) and the category $\fSetCategory$ of sets and total functions (modelling deterministic classical computation) are subcategories of $\RMatCategory{R}$ for all choices of semiring $R$.

\section{Dagger-compact categories}
\label{section_DaggerCompactSMCs}

\subsection{Dagger, isometries and unitaries}

Some of the most iconic features of quantum theory, such as unitaries and the bra-ket notation, depend on a notion of inner product on states. Because composition of states and effects already produces scalars, a categorical way to introduce an inner product in a SMC is to guarantee state-effect duality, in a way compatible with parallel and sequential composition of processes. This approach leads to \textbf{dagger symmetric monoidal categories} (henceforth $\dagger$-SMC) \cite{Selinger2007}, symmetric monoidal categories $\CategoryC$ equipped with an involutive functor $\dagger: \CategoryC \rightarrow \OpCategory{\CategoryC}$ of SMCs, the \textbf{dagger}, which is furthermore the identity on objects. Concretely, to each process $f: A \rightarrow B$ in a $\dagger$-SMC $\CategoryC$ is associated a process $f^{\dagger}: B \rightarrow A$ in $\CategoryC$, called the \textbf{adjoint} of $f$, in a way such that $(f^{\dagger})^{\dagger}= f$. Being a functor of SMCs further implies that $\id{A}^\dagger = \id{A}$, that $(f \circ g)^\dagger = g^\dagger \circ f^\dagger$, and that $(f \otimes g)^\dagger = f^\dagger \otimes g^\dagger$. There is a further requirement that $\sigma_{A,B}^\dagger = \sigma_{A,B}^{-1}$. For an operational characterisation of the Hermitian adjoint in quantum theory, see \cite{Selby2016}. 

Although there is no mention of linearity, a $\dagger$-SMC comes with an \textbf{inner product}, where two states $\phi$ and $\psi$ are sent to the the scalar $\phi^\dagger \circ \psi$. We will say that a process $f: A \rightarrow B$ is an \textbf{isometry} if $f^\dagger  f = \id{A}$, and that it is a \textbf{unitary} if it is furthermore invertible (equivalently, if $f^\dagger$ is also an isometry, i.e. $f f^\dagger = \id{B}$). Isometries preserve the inner product of states:
\begin{equation}\label{eqn_IsometriesInnerProd}
\Big(\phi^\dagger f^\dagger\Big) \Big(f \psi\Big) = \phi^\dagger \Big( f^\dagger f \Big) \psi = \phi^\dagger   \id{A} \psi = \phi^\dagger \psi
\end{equation}

The category $\HilbCategory$ is a $\dagger$-SMC, with dagger given by the Hermitian adjoints. More generally, any choice of involutive semiring isomorphism $^\dagger: R \rightarrow R$ endows the category $\RMatCategory{R}$ with the structure of a $\dagger$-SMC.

\subsection{Dagger compact structure}

The dagger abstractly captures state-effect duality, but says nothing about operator-state duality, the other important correspondence in quantum theory. The latter requires \textit{compact closed} structure \cite{Doplicher1989,Kelly1980a,Baez1995}, and is a much more restrictive property. Dagger structure taken together with compact closed structure leads to the definition of dagger compact categories \cite{Abramsky2004,Abramsky2005}. Here, we will (loosely) follow the presentation of compact closed categories given by \cite{Selinger2010,Selinger2009}.

We say that a SMC is \textbf{compact closed} if every object $A$ comes with a \textbf{dual object} $A^\ast$, a \textbf{cup} $\eta_{A}: \tensorUnit \rightarrow A^\ast \otimes A$ and a \textbf{cap} $\varepsilon_A: A \otimes A^\ast \rightarrow \tensorUnit$
\begin{equation}\label{eqn_cupCap}
\begin{multlined}
\input{pictures/chapter2/cupCap.tikz}
\end{multlined}
\end{equation}
which satisfy the following \textbf{yanking equations} (where the right hand sides are just the identity morphisms $\id{A}$ and $\id{A^\ast}$):
\begin{equation}\label{eqn_yankingEquations}
\begin{multlined}
\input{pictures/chapter2/yankingEquations.tikz}
\end{multlined}
\end{equation}
Duals are required to distribute (contravariantly) over the tensor product $(A \otimes B)^\ast = A^\ast \otimes B^\ast$, and to respect the tensor unit $\tensorUnit^\ast = \tensorUnit$ (technically, they are canonically isomorphic rather than equal, but the distinction can be safely ignored here). 

These definitions actually hold in a generic monoidal category. In a SMC, the symmetry isomorphism gives one more cup and one more cap, themselves satisfying yanking equations:
\begin{equation}\label{eqn_cupCapSwap}
\begin{multlined}
\input{pictures/chapter2/cupCapSwap.tikz}
\end{multlined}
\end{equation}

In a $\dagger$-SMC, one more cup should arise in principle as $\varepsilon_{A}^\dagger$, and one more cap as $\eta_{A}^\dagger$: however, one imposes additional coherence conditions ensuring that these new cup and cap coincide with those of Equation~\ref{eqn_cupCapSwap}, i.e. that $\varepsilon_{A}^\dagger = \sigma_{A^\ast\!,A} \eta_{A}$ and $\eta_{A}^\dagger = \varepsilon_{A} \sigma_{A,A^\ast}^\dagger$. A compact closed $\dagger$-SMC satisfying these two additional requirements is called a \textbf{dagger compact category} \cite{Abramsky2004,Abramsky2005}. Note that the little arrow markings flip upside-down when taking the adjoint:
\begin{equation}\label{eqn_cupCapDagger}
\begin{multlined}
\input{pictures/chapter2/cupCapDagger.tikz}
\end{multlined}
\end{equation}

Because of the yanking equations, compact closed structure provides a form of \textbf{operator-state duality}, a bijection\footnote{When the SMC is enriched, this always is an isomorphism in the appropriate category.} between processes $A \rightarrow B$ and states $\tensorUnit \rightarrow A^\ast \otimes B$ given as follows:
\begin{equation}\label{eqn_operatorStateDuality}
\begin{multlined}
\input{pictures/chapter2/operatorStateDuality.tikz}
\end{multlined}
\end{equation}
In a dagger compact category, the adjoint induces an inner product on operators via operator-state duality, which we will refer to as the \textbf{Hilbert-Schmidt inner product} (because that is what it is called in $\fdHilbCategory$).

The dagger structure induces an involutive symmetry on processes, sending a process $f: A \rightarrow B$ to its adjoint $B \rightarrow A$. Similarly, the compact closed structure induces another involutive symmetry, sending $f: A \rightarrow B$ to its \textbf{transpose} $f^T: B^\ast \rightarrow A^\ast$ defined as follows, and vice versa:
\begin{equation}\label{eqn_transpose}
\begin{multlined}
\input{pictures/chapter2/transpose.tikz}
\end{multlined}
\end{equation}
Correspondence~\ref{eqn_operatorStateDuality} in particular bijects the states of the dual object $A^\ast$ with the effects of $A$: as a consequence, the transpose $f^T$ can be seen as a process acting on effects rather than states, sending an effect $b : B \rightarrow \tensorUnit$ on $B$ to the effect $b \circ f : A \rightarrow \tensorUnit$ on $A$ obtained by pre-composition with $f: A \rightarrow B$. 

The adjoint and transpose symmetries commute, and together they generate a symmetry group on processes isomorphic to $\integersMod{2} \times \integersMod{2}$, with the following \textbf{conjugate} symmetry $f \mapsto f^\ast := (f^\dagger)^T = (f^T)^\dagger$ as the fourth group element:
\begin{equation}\label{eqn_conjugate}
\begin{multlined}
\input{pictures/chapter2/conjugate.tikz}
\end{multlined}
\end{equation}
Like the transpose, the conjugate can also be seen as a process acting on effects. 

A \textbf{self-duality structure} on a compact closed SMC is a family of isomorphisms $h_A : A \rightarrow A^\ast$ satisfying four coherence conditions laid out in \cite{Selinger2010}. With these, we can define a new pair of \textbf{symmetric cup} and \textbf{symmetric cap}:
\begin{equation}\label{eqn_cupCapSD}
\begin{multlined}
\input{pictures/chapter2/cupCapSD.tikz}
\end{multlined}
\end{equation}
The first two conditions on $h_A$ concern the relationship of duals with the tensor product: $h_\tensorUnit = \id{\tensorUnit}$ and $h_{A \otimes B}$ = $\sigma_{A^\ast,B^\ast} \big(h_{A} \otimes h_{B} \big)$, where the symmetry isomorphism is required since $h_{A \otimes B} : A \otimes B \rightarrow B^\ast \otimes A^\ast$. The third coherence condition is symmetry:
\begin{equation}\label{eqn_symmetryEquationsSD}
\begin{multlined}
\input{pictures/chapter2/symmetryEquationsSD.tikz}
\end{multlined}
\end{equation}
The last condition relates the symmetric cup/cap on an object $A$ to the symmetric cup/cap on the dual object $A^\ast$:
\begin{equation}\label{eqn_dualSystemsCupsSD}
\begin{multlined}
\input{pictures/chapter2/dualSystemsCupsSD.tikz}
\end{multlined}
\end{equation}
In a dagger compact category, Equation~\ref{eqn_cupCapDagger} can be seen to state that the two caps are the adjoints of the two cups; for a self-duality structure in a compact closed category, on the other hand, Equation~\ref{eqn_symmetryEquationsSD} states that there is only one symmetric cup and one symmetric cap. It is then natural to require that self-duality structures in a dagger compact category lead to a symmetric cap which is adjoint to the symmetric cup\footnote{Which also validates the graphical notation we've chosen.}: this is equivalent to unitarity of the $h_A$ isomorphisms, which is imposed as an additional coherence condition on self-duality structures in dagger compact categories.

\subsection{Examples of dagger compact categories}

\paragraph{Finite-dimensional Hilbert spaces.} The category $\fdHilbCategory$ of finite-dimensional Hilbert spaces is dagger compact, but the larger category $\HilbCategory$ of Hilbert spaces is not (it has dual objects, and state-effect duality, but it cannot have a cup or cap). For a finite-dimensional Hilbert space $A$, the dual object is defined to be its dual space $A^\ast$, the Hilbert space of linear functionals $A \rightarrow \complexs$. The cap $\varepsilon_{A}$ is obtained from the inner product, as the linear map $A \otimes A^\ast$ which sends a state $\ket{\psi}$ of $A$ and an effect $\bra{a}$ of $A$ (i.e. a state of $A^\ast$) to the complex number $\braket{a}{\psi}$; the other cap is obtained by first applying a symmetry isomorphism $\sigma_{A,A^\ast}^{-1}$, and the two cups are obtained by taking adjoints. In terms of an orthonormal basis $\ket{e_j}_{j=1}^{\dim{A}}$ (any basis), the cup and cap can be written as follows, where states of $A^\ast$ are represented by effects of $A$, and effects of $A^\ast$ by states of $A$:
\begin{equation}\label{eqn_cupCapBasis}
\begin{multlined}
\input{pictures/chapter2/cupCapBasis.tikz}
\end{multlined}
\end{equation}


Choice of basis $\ket{e^{(a)}_j}_{j}$ induces a self-duality structure $h_A$, and the symmetric cup and cap given by Equation~\ref{eqn_cupCapSD} for this self-duality structure take the following form:
\begin{equation}\label{eqn_cupCapSDBasis}
\begin{multlined}
\input{pictures/chapter2/cupCapSDBasis.tikz}
\end{multlined}
\end{equation}

\noindent Because they satisfy the yanking equations, substituting the symmetric cup and cap in Equations~\ref{eqn_transpose} and~\ref{eqn_conjugate} induce another transpose symmetry and another conjugate symmetry on processes (this time sending $f: A \rightarrow B$ to $f^{T_h}: B \rightarrow A$ and $f^{\ast_h}: A \rightarrow B$). Contrary to the transpose and conjugate symmetries induced by the dagger compact structure, these symmetries are basis-dependent. 
Similarly, the symmetric cup and cap of Equation~\ref{eqn_cupCapSD} (and related transpose and conjugate) are basis-dependent, while the cup and cap of Equation~\ref{eqn_cupCap} (and related transpose and conjugate) are basis-independent.

\paragraph{Categories of matrices over semirings.} Now consider the $\dagger$-SMC $\fRfreeModCategory{R}$, where $R$ is a commutative semiring with involution. Because the objects are (essentially) in the form $R^n$ for all natural numbers $n$, one can choose the canonical orthonormal basis $\ket{e_j^{(n)}} := 1 \mapsto (\delta_{ij})_{i=1}^n$, and Equation~\ref{eqn_cupCapSDBasis} readily give a family of symmetric cups and caps (they only involve the scalars $0$ and $1$, which behave the same way in any semiring). The dual space to the free $R$-semimodule $R^n$ is the free $R$-semimodule of $R$-linear maps $R^n \rightarrow R$, itself isomorphic to $R^n$. Self-duality structures and symmetric cups/caps can be defined from choices of basis exactly like in the complex case.
It should be noted that the existence of a dagger compact structure is related to the inner product, but has nothing to do with its non-degeneracy: the latter depends on the specific choice of involution, while the former exists for all choices of involution.

\subsection{The matrix algebra}

We have seen in Equation~\ref{eqn_operatorStateDuality} that compact closed structure implements operator-state duality: it is natural to ask whether composition of operators can always be internalised by a process acting on the corresponding states, and this question leads to the definition of the \textit{matrix algebra}. 

If $f: A \rightarrow B$ is a process in a compact closed category, we will denote the corresponding state by $\morphismName{f}: \tensorUnit \rightarrow A^\ast \otimes B$. The state $\morphismName{g \circ f}: \tensorUnit \rightarrow A^\ast \otimes C$ corresponding to the composition of two processes $f: A \rightarrow B$ and $g: B \rightarrow C$ can be obtained as follows in terms of the corresponding states $\morphismName{f}$ and $\morphismName{g}$:
\begin{equation}\label{eqn_morphismNameComposition}
\begin{multlined}
\input{pictures/chapter2/morphismNameComposition.tikz}
\end{multlined}
\end{equation}
For processes $A \rightarrow A$, composition is a monoid operation, with the identity $\id{A}$ as its unit. There is an internal monoid on object $A^\ast \otimes A$ corresponding to composition: 
\begin{equation}\label{eqn_matrixAlgebra}
\begin{multlined}
\input{pictures/chapter2/matrixAlgebra.tikz}
\end{multlined}
\end{equation}
Similarly we can use the other cup and cap to construct a comultiplication and counit. The multiplication, unit, comultiplication and counit in a SMC form a Frobenius algebra (more on Frobenius algebras later), known as the \textbf{matrix algebra}. In a $\dagger$-SMC, they form a symmetric $\dagger$-Frobenius algebra (which in $\fdHilbCategory$ corresponds to the C*-algebra $\Bounded{A}$ of bounded operators on Hilbert space $A$).


\newcommand{\discardingmap}{discarding map}
\newcommand{\Discardingmap}{Discarding map}
\newcommand{\discardingmaps}{discarding maps}
\newcommand{\Discardingmaps}{discarding maps}
\newcommand{\normalised}{normalised}
\newcommand{\Normalised}{Normalised}
\section{Environments, causality and purification}
\label{section_Environments}

\subsection{Environment structures}

Dagger compact categories capture a number of structures typical of pure-state quantum theory, such as inner product, state-effect duality and operator-state duality. However, one additional component is necessary to make the leap from pure-state to mixed-state, and that component is the \textit{environment structure}.

An \textbf{environment structure} for a SMC consists of a family of processes $\trace{A}: A \rightarrow \tensorUnit$ for all objects, the \textbf{discarding maps}, such that:
\begin{equation}\label{eqn_traces}
\begin{multlined}
\input{pictures/chapter2/traces.tikz}
\end{multlined}
\end{equation} 
The empty diagram on the RHS of the right equation is the scalar 1. Given an environment structure, a process $f: A \rightarrow B$ is said to be \textbf{normalised} if performing the process and then discarding the output is the same as discarding the input:
\begin{equation}\label{eqn_normalised}
\begin{multlined}
\input{pictures/chapter2/normalised.tikz}
\end{multlined}
\end{equation}
Discarding maps are also called \textbf{traces} by those of the Hilbert-space persuasion, in which case normalised processes should be referred to as \textbf{trace-preserving}. In the effectus community, the discarding map $\trace{X}$ is also known as the \textit{truth} (or \textit{ground}) on $X$, and normalised processes are known as \textit{total}.

Thanks to Equation~\ref{eqn_traces}, normalised processes in a SMC $\CategoryC$ with environment structure $\hbox{\input{symbols/traceSym.tex}}\!$ are guaranteed to form a sub-SMC $\CausalSubCategory{\CategoryC}$, the \textbf{normalised subcategory}. In the normalised subcategory, the tensor unit $\tensorUnit$ is terminal: there is a unique effect, namely $\trace{X}$, on any system $X$, and there is a unique scalar $\id{\tensorUnit}$. This means that the discarding maps truly lose all information about the system, and thus define a sensible notion of discarding.\footnote{This statement can be given categorical dignity by saying that $\hbox{\input{symbols/traceSym.tex}}\!$ is a monoidal natural transformation from the identity functor to the terminal endofunctor (of SMCs), the one sending all objects to $\tensorUnit$ and all processes to the scalar $\id{\tensorUnit}$.}

\subsection{Discarding maps in categories of matrices}


The simplest example of environment structures is given by the category $\SetCategory$ of sets and total functions, where the tensor unit (the singleton set $\mathbb{1}$) is already terminal: the discarding map $\trace{X}$ on a set $X$ is the unique total function $X \rightarrow \mathbb{1}$, and $\SetCategory = \CausalSubCategory{\SetCategory}$. The category $\fSetCategory$ of finite sets and total functions is a subcategory of $\RMatCategory{R}$ for any semiring $R$, and endows the latter with its environment structure. If $\ket{\psi}: \mathbb{1} \rightarrow R^X$ is a state in $\RMatCategory{R}$, then $\trace{X} \ket{\psi}$ is the sum $\sum_{x \in X} \braket{x}{\psi}$ of all entries of column vector $\ket{\psi}$, and normalised states are exactly those with entries summing to the multiplicative unit $1$ of semiring $R$. Similarly, the normalised processes $R^{X} \rightarrow R^{Y}$ are those with $Y \times X$ matrix having normalised vectors as columns. The following examples are of interest in the applications we will consider here:
\begin{enumerate}
\item[(i)] in $\RMatCategory{\reals^+}$, the normalised subcategory is the category $\fStochCategory$ of finite sets and stochastic maps between them (with probability distributions as states);
\item[(iia)] in $\RMatCategory{\mathbb{B}}$, the normalised states of $\mathbb{B}^X$ are the non-empty subsets of $X$, and the normalised processes $\mathbb{B}^X \rightarrow \mathbb{B}^Y$ are the relations $f \subseteq X \times Y$ which are total, i.e. such that $\domain{f} = X$;
\item[(iib)] in the subcategory $\fPFunCategory \leq \RMatCategory{\mathbb{B}}$ of finite sets and partial functions, the normalised processes are the total functions, yielding $\CausalSubCategory{\fPFunCategory} = \fSetCategory$;
\item[(iii)] in $\RMatCategory{\naturals}$, the normalised subcategory is the category $\fSetCategory$ of finite sets and functions.
\end{enumerate}

\subsection{The CPM construction}

We have seen that the discarding maps inherited from $\fSetCategory$ give the desired notion of normalised states and processes in the case of $\RMatCategory{\reals^+}$, the category modelling classical probabilistic systems. In the case of quantum theory, on the other hand, discarding is done by inner products and traces, and we need a different construction.

In the traditional formulation of quantum mechanics, the transition from pure-state to mixed-state sees the Hilbert space formalism replaced by the operator formalism: the states on a Hilbert space $A$ are taken to be the positive self-adjoint operators $\rho: A \rightarrow A$, forming the $\reals^+$-semimodule (aka convex cone) $\CPstates{A}$, and more general processes $\Phi: \CPstates{A} \rightarrow \CPstates{B}$ are taken to be \textit{completely positive (CP) maps}, sending positive self-adjoint operators $\rho : A \rightarrow A$ to positive self-adjoint operators $\Phi(\rho): B \rightarrow B$. By \textbf{Kraus' Theorem}, the \textbf{CP maps} $\CPstates{A} \rightarrow \CPstates{B}$ are exactly those in the following form, known as \textbf{Kraus decomposition}:
\begin{equation}\label{eqn_KrausDecomposition}
\Phi(\rho) = \sum_{i=1}^{\dim{A}} \sum_{j=1}^{\dim{B}} B_{ij} \, \rho \, B_{ij}^\dagger
\end{equation}
where the linear maps $B_{ij}: A \rightarrow B$ are known as the \textbf{Kraus operators}\footnote{Note that the Kraus decomposition is only unique up to unitary transformations of the Kraus operators $B_{ij} \mapsto B'_{kl} :=  \sum_i\sum_j u_{klij}B_{ij}$.}. As a special case of Kraus' Theorem, one recovers a decomposition of positive self-adjoint operators in terms of pure states (not necessarily normalised):
\begin{equation}\label{eqn_KrausDecompositionStates}
\rho = \sum_{j=1}^{\dim{B}} \ket{\psi_j}\bra{\psi_j}
\end{equation}

In the operator formalism, the discarding map on system $\CPstates{A}$ is given by the \textit{trace}, sending state $\rho \in \CPstates{A}$ to $\Trace{\rho} \in \reals^+$. \textit{Normalised states} are exactly those in the form of Equation \ref{eqn_KrausDecompositionStates} with $\sum_j \braket{\psi_j}{\psi_j} = 1$; traditionally, the pure states $\ket{\psi_j}$ are chosen to be orthogonal (appealing to the spectral theorem) and then renormalised to yield a convex mixture of orthonormal states $\rho = \sum_j p_j \, \ket{\phi_j}\bra{\phi_j}$, where $p_j \in \reals^{+}$ sum to $1$ and are interpreted as probabilities. Normalised CP maps are called \textit{trace-preserving}, and by Kraus' Theorem they are exactly those satisfying $\sum_i \sum_j B_{ij}^\dagger B_{ij} = \id{A}$.

Categories of completely positive maps, also known as \textbf{CPM categories} \cite{Selinger2007,Coecke2012env}, can be constructed for all dagger compact categories, in a process which mimics the way in which the operator model of mixed-state quantum mechanics (corresponding to the CPM category $\CPMCategory{\fdHilbCategory}$) is constructed from the Hilbert space model of pure-state quantum mechanics (corresponding to the dagger compact category $\fdHilbCategory$). Given a dagger compact category $\CategoryC$, the corresponding CPM category $\CPMCategory{\CategoryC}$ is defined as the subcategory of $\CategoryC$ having objects in the form $A^\ast \otimes A$ and morphisms in the following form: 
\begin{equation}\begin{multlined}\label{diagram_CPMmorphismUnboxed}
\input{pictures/chapter2/CPMmorphismUnboxed.tikz}
\end{multlined}\end{equation}
where $f: A \rightarrow E \otimes B $ is a morphism of $\CategoryC$, and $f^\ast: A^\ast \rightarrow B^\ast \otimes E^\ast$ is its conjugate (obtained via the dagger compact structure). In $\CPMCategory{\CategoryC}$, the object $A^\ast \otimes A$ is simply written $A$, so that the process depicted in Diagram \ref{diagram_CPMmorphismUnboxed} is considered a process $A \rightarrow B$ in $\CPMCategory{\CategoryC}$. Processes in the CPM category are called \textbf{completely positive (CP) maps}. The system $E$ in Diagram \ref{diagram_CPMmorphismUnboxed} is often interpreted as an \textbf{environment} which is operationally inaccessible, and hence must be \inlineQuote{discarded} after the process has taken place. 
In the case of $\CPMCategory{\fdHilbCategory}$, i.e. in the operator model of mixed-state quantum mechanics, Diagram \ref{diagram_CPMmorphismUnboxed} can be seen as an alternative formulation of Kraus decomposition. 

Diagrammatic reasoning about categories of completely positive maps often involves two distinct SMCs: the original category $\CategoryC$ and the CPM category $\CPMCategory{\CategoryC}$. As a consequence, a stylistic convention is adopted where systems and processes of the CPM category $\CPMCategory{\CategoryC}$ are denoted by thicker wires, boxes and decorations. For example, the \inlineQuote{doubled} version $f^\ast \otimes f$ of a process $f: A \rightarrow B \otimes E$ will be denoted as $f$ with thicker wires and box:  
\begin{equation}\begin{multlined}\label{eqn_CPMdoubledNotation}
\input{pictures/chapter2/CPMdoubledNotation.tikz}
\end{multlined}\end{equation}
The caps from compact closed structure play a particularly important role in the definition of the CPM category, and are given their own decoration:
\begin{equation}\begin{multlined}\label{eqn_CPMdiscardingMaps}
\input{pictures/chapter2/CPMdiscardingMaps.tikz}
\end{multlined}\end{equation}
The CP map $\CPMdoubled{f}$ defined by Equation \ref{eqn_CPMdoubledNotation} is called the \textbf{double} of process $f$, while the CP map $\trace{A}$ defined by Equation \ref{eqn_CPMdiscardingMaps} is called the \textbf{discarding map} on system $A$. The discarding maps $\trace{A}$ form a canonical environment structure for $\CPMCategory{\CategoryC}$. In mixed-state quantum mechanics, the double of a linear map $f$ is the CP map $\CPMdoubled{f}:= \rho \mapsto f \circ \rho \circ f^\dagger$, while the discarding map $\trace{A}$ sends a positive state $\rho \in \CPstates{A}$ to its trace $\Trace{\rho} \in \CPstates{\complexs} \isom \reals^+$.

CPM categories $\CPMCategory{\CategoryC}$ are dagger compact, and the rules of diagrammatic reasoning for dagger compact categories apply to them. The compact structure for $\CPMCategory{\CategoryC}$ is given by the doubles of the cups and caps of $\CategoryC$, while the adjoint of a process in the form of Diagram \ref{diagram_CPMmorphismUnboxed} is given by first taking the adjoint in $\CategoryC$, and then using the following equation for the adjoint of the discarding map:
\begin{equation}\begin{multlined}\label{eqn_CPMdiscardingMapsAdjoints}
\input{pictures/chapter2/CPMdiscardingMapsAdjoints.tikz}
\end{multlined}
\end{equation}
Because the doubled processes $\CPMdoubled{f}$ and the discarding maps $\trace{A}$ are well-defined CP maps, it is legitimate to rephrase the very definition of the CPM category by saying that its processes are exactly those in the following form:
\begin{equation}\begin{multlined}\label{eqn_CPMmorphism}
\input{pictures/chapter2/CPMmorphism.tikz}
\end{multlined}\end{equation}
This means that doubled processes and discarding maps are enough to \textit{express} all CP maps. However, in order to \textit{prove results about} CP maps, we need a graphical axiom relating a generic CPM category $\CPMCategory{\CategoryC}$ to the corresponding original category $\CategoryC$. The required relationship is encoded by the following \textbf{CPM axiom}, which characterises the action of discarding maps in $\CPMCategory{\CategoryC}$ in terms of the dagger structure of $\CategoryC$:
\begin{equation}\begin{multlined}\label{eqn_CPMaxiom}
\input{pictures/chapter2/CPMaxiom.tikz}
\end{multlined}\end{equation}
It is possible to state an exact correspondence between CPM categories and dagger compact categories satisfying the CPM axiom above.
\begin{theorem}[\textbf{CPM Categories} \cite{Coecke2008a,Coecke2012env,Coecke2012d}]\hfill \\
Let $\CategoryC$ be a dagger compact category with an environment structure $(\trace{A})_{a \in \obj{\CategoryC}}$ satisfying Equation \ref{eqn_CPMdiscardingMapsAdjoints}. Let $\CategoryD$ be another dagger compact category, together with a functor $\Phi: \CategoryD \rightarrow \CategoryC$ of dagger compact categories which is a bijection on objects (so that the compact closed structure of $\CategoryC$ is exactly the image under $\Phi$ of the compact closed structure of $\CategoryD$). Assume that all morphisms in $\CategoryC$ take the form of Equation \ref{eqn_CPMmorphism} for some morphism $f$ in the image of $\Phi$, and that the CPM Axiom is satisfied. Then there is a (necessarily unique) isomorphism of dagger compact categories $F: \CategoryC \rightarrow \CPMCategory{\CategoryD}$ such that $F \big(\Phi(f)\big) = \CPMdoubled{f}$ for all morphisms $f$ of $\CategoryD$ and $F(\trace{\Phi(A)}) = \trace{A}$ for all objects $A$ of $\CategoryD$.
\end{theorem}

\subsection{Purification}

Observe that CP maps in $\CPMCategory{\CategoryC}$ arising as doubles of morphisms in $\CategoryC$ form a dagger compact subcategory $\DoubledCategory{\CategoryC}$: they are closed under parallel and sequential composition and dagger, and the cups and caps of $\CPMCategory{\CategoryC}$ arise themselves as doubled processes. In the case of mixed-state quantum mechanics, the \textbf{doubled category} $\DoubledCategory{\fdHilbCategory}$ corresponds to linear maps of Hilbert spaces \textit{up to global phase}: it is in fact this category, rather than $\fdHilbCategory$, that models pure-state quantum mechanics. More generally, we want to see $\DoubledCategory{\CategoryC}$ as a sub-category of pure processes within a larger category $\CPMCategory{\CategoryC}$ of mixed processes: for this to be meaningful, we need the doubled category to satisfy some core operational characteristics of purity in mixed-state quantum theory. 

Purity is a feature arising at the interface between quantum theory and thermodynamics~\cite{Chiribella2010,Chiribella2015}: pure processes can broadly be interpreted as not involving any probabilistic mixing due to non-trivial interactions with a discarded environment. Former work on purity has taken the following approach: purity is defined as a property, and the \textit{purification axiom} is formulated as the requirement that all processes be expressible as a combination of pure processes and discarding maps. Because CPM categories already come with the guarantee that all processes are expressible as combinations of doubled processes and discarding maps, we will turn things the other way around. We will say that a CPM category satisfies the \textbf{Purification axiom} if doubled processes (which we want to interpret as pure) satisfy the following abstract version of purity:
\begin{equation}\begin{multlined}\label{eqn_CPMpurificationAxiom}
\begin{multlined}
\input{pictures/chapter2/CPMpurificationAxiom.tikz}
\end{multlined}
\end{multlined}\end{equation}
with the additional requirement that $\trace{E} \circ \CPMdoubled{\psi} = 1$ (i.e. $\psi$ is a normalised state), or equivalently that $\braket{\psi}{\psi} = 1$ (by the CPM axiom). Operationally, the Purification axiom can be given the following interpretation: the only way a process $(\id{B} \otimes \trace{E}) \circ \CPMdoubled{f}$ involving the discarding of an environment $E$ can result in a pure process~$g$ is if the environment being discarded is disconnected altogether from the pure process (i.e. the interaction with the environment is \textbf{trivial}). In a category which satisfies the purification axiom, we will also refer to the doubled processes as \textbf{pure processes}, and to the doubled subcategory as \textbf{pure subcategory}.

The Purification axiom is a non-trivial requirement for CPM categories, and there are many inequivalent examples of CPM categories violating it. 
\begin{theorem}[\textbf{CPM categories violating the Purification axiom}]\hfill\\
Let $R$ be a commutative semiring with involution in which the multiplicative unit $1$ is additively idempotent\footnote{Examples include all locales (e.g. the booleans $\mathbb{B}$), all tropical semirings (e.g. the tropical semiring $(\reals,\min,+)$ or the Viterbi semiring $([0,1],\max,\times) \isom (\reals^+,\min,+)$) and all commutative unital quantales.}. Then $\CPMCategory{\RMatCategory{R}}$ violates the Purification Axiom. 
\end{theorem}
\begin{proof}
Let $X$ be any finite set with at least two elements, and let $\mathcal{P}^+(X)$ be the set of non-empty subsets of $X$. To obtain a counterexample to the Purification axiom \ref{eqn_CPMpurificationAxiom}, we construct an $f: R^X \rightarrow R^X \otimes R^{\mathcal{P}^+(X)}$ such that $(\id{R^X} \otimes \trace{R^{\mathcal{P}^+(X)}}) \circ \CPMdoubled{f} = \CPMdoubled{\id{R^X}}$ but $f$ does not decompose as $\id{R^X} \otimes \psi$ for any state $\psi: R \rightarrow R^{\mathcal{P}^+(X)}$. Indeed, consider the map $f$ defined as follows:
\begin{equation}
f := \sum_{U \in \mathcal{P}^+(X)} \sum_{x \in U} \ket{x} \otimes \ket{U} \otimes \bra{x}
\end{equation}
Note that this map cannot decompose as $\id{R^X} \otimes \psi$ for any state $\psi$:
\begin{equation}
f \ket{x} = \ket{x} \otimes \Big(\sum_{U \in \mathcal{P}^+(X)} \delta_{x \in U} \ket{U} \Big)
\end{equation}
Now consider its doubled version $\CPMdoubled{f}$, and discard the system $R^{\mathcal{P}^+(X)}$ to obtain the following morphism $R^X \rightarrow R^X$ of $\CPMCategory{\RMatCategory{R}}$, which we write down explicitly as a morphism $R^X \otimes R^X \rightarrow R^X \otimes R^X$ of $\RMatCategory{R}$ by using the expression $\trace{R^{\mathcal{P}^+(X)}} = \sum_{U \in \mathcal{P}^+(X)} \CPMdoubled{\bra{U}}$ for the discarding map on $R^{\mathcal{P}^+(X)}$:
\begin{align}
&\sum_{x,y \in X} \sum_{U \in \mathcal{P}^+(X)} \delta_{x,y \in U} \big(\ket{x} \otimes \ket{y}\big) \otimes \big(\bra{x} \otimes \bra{y}\big) \nonumber \\
= &\sum_{x,y \in X} \big(\sum_{U \in \mathcal{P}^+(X)} \delta_{x,y \in U}\big) \big(\ket{x} \otimes \ket{y}\big) \otimes \big(\bra{x} \otimes \bra{y}\big) \nonumber\\
= &\sum_{x,y \in X} \big(\ket{x} \otimes \ket{y}\big) \otimes \big(\bra{x} \otimes \bra{y}\big)  = \CPMdoubled{\id{R^X}}
\end{align}
where the second equality follows from the fact that $\sum_{U \in \mathcal{P}^+(X)} \delta_{x,y \in U} = 1$, since $1$ is additively idempotent and there is at least one $U$ such that $x,y \in U$ (furthermore, we have used the fact that $0^\dagger 0 = 0$ and $1^\dagger 1 = 1$, so that $\delta_{x,y \in U}^\dagger \delta_{x,y \in U} = \delta_{x,y \in U}$). This concludes our proof.
\end{proof}



\section{Coherent data manipulation}
\label{section_CoherentDataManipulation}

\subsection{Dagger Frobenius algebras}
Frobenius algebras are a fundamental ingredient of quantum theory, where they are intimately related to the notion of observable. Consider a monoid $(A,\!\hbox{\input{symbols/ZbwmultSym.tex}}\!\!,\!\hbox{\input{symbols/ZbwunitSym.tex}}\!\!)$ on an object $A$ of a $\dagger$-SMC $\CategoryC$: a binary operation $\!\hbox{\input{symbols/ZbwmultSym.tex}}\!\!: A \otimes A \rightarrow A$ (the \textbf{multiplication}) which is associative and has the state $\!\hbox{\input{symbols/ZbwunitSym.tex}}\!\!: \tensorUnit \rightarrow A$ (the \textbf{unit}) as its bilateral unit. We will refer to this fact by saying that $\!\hbox{\input{symbols/ZbwmultSym.tex}}\!\!$ and $\!\hbox{\input{symbols/ZbwunitSym.tex}}\!\!$ satisfy \textbf{associativity}, or the \textbf{associative law}, and the \textbf{left/right unit laws}. Then the adjoints $\!\hbox{\input{symbols/ZbwcomultSym.tex}}\!\! := (\!\hbox{\input{symbols/ZbwmultSym.tex}}\!\!)^\dagger: A \rightarrow A \otimes A$ (the \textbf{comultiplication}) and $\!\hbox{\input{symbols/ZbwcounitSym.tex}}\!\! := (\!\hbox{\input{symbols/ZbwunitSym.tex}}\!\!)^\dagger : A \rightarrow \tensorUnit$ (the \textbf{counit}) automatically form a comonoid on $A$ in $\CategoryC$ (i.e. they satisfy \textbf{coassociativity} and the \textbf{left/right counit laws}). A \textbf{$\dagger$-Frobenius algebra} on an object $A$ in $\CategoryC$ is one such pair of monoid and comonoid on $A$, which are furthermore related by the following \textbf{Frobenius law}:
\begin{equation}\begin{multlined}\label{eqn_FrobeniusLaw}
\input{pictures/chapter2/FrobeniusLaw.tikz}
\end{multlined}\end{equation}
A $\dagger$-Frobenius algebra is said to be \textbf{special} if the comultiplication $\!\hbox{\input{symbols/ZbwcomultSym.tex}}\!\!$ is an isometry, and \textbf{commutative} if the monoid and comonoid are commutative: 
\begin{equation}\begin{multlined}\label{eqn_specialCommutativeLaws}
\input{pictures/chapter2/specialCommutativeLaws.tikz}
\end{multlined}\end{equation}
More generally, a \textbf{quasi-special} $\dagger$-Frobenius algebra is one with comultiplication $\!\hbox{\input{symbols/ZbwcomultSym.tex}}\!\!$ which is an isometry up to a \textbf{normalisation factor} $N$, where $N$ is in the form $n^\dagger n$ for some invertible scalar $n$:
\begin{equation}\begin{multlined}\label{eqn_quasiSpecialLaw}
\input{pictures/chapter2/quasiSpecialLaw.tikz}
\end{multlined}\end{equation}
Speciality corresponds to the particular case of quasi-speciality in which $N=1$. By normalisation, any quasi-special Frobenius algebra corresponds to a unique special Frobenius algebra: as such, quasi-speciality can be used in place of speciality to simplify some definitions and results, without causing any issue of interpretation.

Because several combinations of these properties will play a role in this work, we introduce a number of short-hands for $\dagger$-Frobenius algebras:
\begin{center}
\begin{tabular}{c | c | c}
\textbf{$\dagger$-Frobenius algebras} 	& commutative 		& arbitrary 		\\ \hline
special    								& $\dagger$-SCFA 	& $\dagger$-SFA		\\ \hline
quasi-special  							& $\dagger$-qSCFA	& $\dagger$-qSFA	\\ \hline
arbitrary 								& $\dagger$-CFA		& $\dagger$-FA		\\ \hline
\end{tabular}
\end{center}

\subsection{Quantum observables}

The importance of $\dagger$-SCFAs in the foundations of quantum mechanics comes from the fact that they correspond to orthonormal bases, i.e. non-degenerate quantum observables. Key to this correspondence is the notion of \textbf{classical states} for a $\dagger$-FA~$\hbox{\input{symbols/ZbwdotSym.tex}}\!\!$, those states $\psi$ which are copied/adjoined/deleted by $\hbox{\input{symbols/ZbwdotSym.tex}}\!\!$ in the following sense:
\begin{equation}\begin{multlined}\label{eqn_classicalState}
\resizebox{\textwidth}{!}{\input{pictures/chapter2/classicalState.tikz}}
\end{multlined}\end{equation}
Note that the RHS of the delete condition is the scalar 1.
\begin{theorem}[\cite{Coecke2013b}]
In $\fdHilbCategory$, the classical states for a $\dagger$-SCFA $\hbox{\input{symbols/ZbwdotSym.tex}}\!\!$ always form an orthonormal basis\footnote{It should also be noted that the copy and delete conditions are sufficient to characterise classical states in the case of $\fdHilbCategory$.}. Furthermore, any orthonormal basis arises this way for a unique $\dagger$-SCFA. More generally, if $\hbox{\input{symbols/ZbwdotSym.tex}}\!\!$ is a $\dagger$-qSCFA, with normalisation factor $N$, then the classical states for $\hbox{\input{symbols/ZbwdotSym.tex}}\!\!$ form an orthogonal basis, each state having norm $\sqrt{N}$. Furthermore, any orthogonal basis where all states have the same norm $\sqrt{N}$ arises this way for a unique $\dagger$-qSCFA.
\end{theorem}
\noindent The concept of classical states forming a basis is generalised to arbitrary $\dagger$-SMC by the notion of enough classical states. A $\dagger$-FA $\hbox{\input{symbols/ZbwdotSym.tex}}\!\!$ on an object $A$ is said to have \textbf{enough classical states} if its classical states separate morphisms from $A$ (i.e. any two morphisms $f,g: A \rightarrow B$ are equal whenever $f \circ \psi = g \circ \psi$ for all classical states $\psi$ of $\hbox{\input{symbols/ZbwdotSym.tex}}\!\!$). Because of the copy condition, a $\dagger$-FA with enough classical states is always necessarily commutative\footnote{To see this, just compose both sides of the commutativity equation with an arbitrary classical state and copy the state through, obtaining the same result on both sides.}. 

This algebraic characterisation of quantum observables is not limited to the non-degenerate case of orthonormal bases, but can be extended to the more general case of complete families of orthogonal projectors (i.e. to possibly degenerate quantum observables). To do so, one considers \textbf{symmetric} $\dagger$-Frobenius algebras\footnote{Commutative $\dagger$-FAs are a special case of symmetric $\dagger$-FAs.}, i.e. those satisfying the following equation:
\begin{equation}\begin{multlined}\label{eqn_balancedSymmetric}
\input{pictures/chapter2/balancedSymmetric.tikz}
\end{multlined}\end{equation}
Independently of their relevance to Theorem \ref{thm_FrobeniusCStar} below, symmetric $\dagger$-Frobenius algebras have the desirable feature that the inner product structure (the cup and cap) that they define is symmetric. As a consequence, the adjoin condition holds both in the formulation of Equation \ref{eqn_classicalState} and in the \inlineQuote{conjugate} formulation, the one having the state on the other side of the symmetric cap.

\begin{theorem}[\cite{Vicary2011}]\label{thm_FrobeniusCStar}
In $\fdHilbCategory$, symmetric $\dagger$-SFAs are in bijective correspondence with C*-algebras, and hence with complete families of orthogonal projectors.
\end{theorem}

\noindent The core observation in the proof of Theorem \ref{thm_FrobeniusCStar} is that the following map $A \rightarrow A^\ast \otimes A$ is in fact a monoid homomorphism from the algebra $(A,\!\hbox{\input{symbols/ZbwmultSym.tex}}\!\!,\!\hbox{\input{symbols/ZbwunitSym.tex}}\!\!)$ to the matrix algebra on $A^\ast \otimes A$:
\begin{equation}\begin{multlined}\label{diagram_FrobeniusToCStar}
\input{pictures/chapter2/FrobeniusToCStar.tikz}
\end{multlined}\end{equation}
Elements of the algebra $(A,\!\hbox{\input{symbols/ZbwmultSym.tex}}\!\!,\!\hbox{\input{symbols/ZbwunitSym.tex}}\!\!)$ are sent to operators $p: A \rightarrow A$ (or, to be precise, to the corresponding states $\morphismName{p} : \tensorUnit \rightarrow A^\ast \otimes A$ under compact closure): in fact, as Theorem \ref{thm_FrobeniusCStar} shows, they are sent to a C*-algebra of operators $A \rightarrow A$. Elements which are central for the algebra get sent to operators which commute with all other operators in the image, elements which are self-adjoint get sent to self-adjoint operators, and elements which are idempotent get sent to idempotent operators:
\begin{equation}\begin{multlined}\label{eqn_centralSelfadjointProjections}
\resizebox{\textwidth}{!}{\input{pictures/chapter2/centralSelfadjointProjections.tikz}}
\end{multlined}\end{equation}
In particular, the central self-adjoint idempotents of the algebra  $(A,\!\hbox{\input{symbols/ZbwmultSym.tex}}\!\!,\!\hbox{\input{symbols/ZbwunitSym.tex}}\!\!)$ are mapped to the central projectors of the image C*-algebra. The non-zero minimal\footnote{Under the partial order defined by letting $p \preceq q$ if and only if $p = q + s$ for some $s$ which is mapped to a positive self-adjoint operator.} central self-adjoint idempotents are mapped to the unique complete family of orthogonal projectors $(p_j)_j$ corresponding to the C*-algebra (i.e. the one given by the \textit{Artin–--Wedderburn theorem for finite-dimensional C*-algebras}).

\subsection{A brief digression on observables}
\label{subsection_briefDigressionObservables}

In the traditional presentation of pure-state quantum mechanics, observables are identified with self-adjoint operators. This is mainly because the latter have two extremely useful features: (i) the eigenspaces of a self-adjoint operator form a complete family of orthogonal subspaces; (ii) the eingenvalues of a self-adjoint operator are real, and automatically possess the linear structure necessary to treat them as the values of a random variable. As a consequence of these features, it is always possible to see a self-adjoint operator as defining a unique measurement with real-valued outcomes, with each eigenspace corresponding to a definite outcome for the measurement.

This identification is slick and full of useful consequences\footnote{For example, it leads to the identification of $\bra{\psi}H\ket{\psi}$ as the expected value of the measurement corresponding to $H$ on a pure state $\ket{\psi}$.}, but at the end of the day it is no more than a trick, or a fortuitous coincidence. To explain why one should not take the identification of observables and self-adjoint operators too seriously, we lay down the following issues.
\begin{itemize}
	\item Self-adjoint operators in quantum mechanics correspond to measurements with real-valued outcomes. While many examples of naturally real-valued observables exist in the physical literature\footnote{For example, position/momentum of unbounded particles, energy, number, etc.}, this is in no way general: to fit other measurements within this framework, a potentially unnatural identification of measurement outcomes with real values will be needed. 
	\item Even vector-valued measurements cannot be accommodated directly by self-adjoint operators: instead, another slick trick is needed, where families of self-adjoint operators are considered, each operator associated to an individual vector coordinate.
	\item While measurement of states in probabilistic theories always result in probability distributions over the classical outcomes, they need not yield a real-valued random variable, and a notion of expectation might not be well defined for them. On the other hand, self-adjoint operators always yield a notion of expected value on the real values associated to the measurement outcomes, which might be spurious at best and misleading at worst.
	\item The identification is not a defining property of self-adjoint operators, instead relying on the Spectral Theorem (a heavyweight result in linear algebra) for its entire justification. The identification of observables as the self-adjoint generators of unitary symmetries requires Stone's Theorem (another heavyweight result).
\end{itemize}

\noindent As an example of a situation in which self-adjoint operators are unsuitable, we consider the position observable for a periodic lattice, which we can think of as valued in the translation group $G = \prod_{d=1}^D\integersMod{n_d}$ (just like the usual position observable for a 1-dim wavefunction is valued in the translation group $\reals$). Because there is no group homomorphism $\prod_{d=1}^D\integersMod{n_d} \rightarrow \reals^D$, there is no natural way to identify the position observable with a family $(x_d)_{d=1}^D$ of self-adjoint operators.

We could try to extract one such identification by considering the family of self-adjoint operators $(\textbf{x}_d)_{d=1}^D$ which exponentiates to the unitary representation $(V_{\goodchi_p})_{p \in G}$ of the boost symmetry in the following way:
\begin{equation}
(V_{\goodchi_p}\big)_d = e^{2\pi i \frac{p_d \textbf{x}_d}{n_d}}
\end{equation}
This attempt, which in the traditional case shows that the position observable generates the boost symmetry, cannot work here: there are infinitely many equivalent families $(\textbf{x}_d)_{d=1}^D$ which satisfy the equation above, corresponding to the infinitely many possible choices of representatives in $\integers$ (which is a subset of the reals) for the congruence classes modulo $n_d$. Furthermore, the periodic nature of positions means that there is no well-defined notion of expected value, and the one arising from any given choice of representatives in $\reals$ is misleading. 

The objections above will turn out to be extremely pertinent to our work, and we will therefore reject the identification of observables and self-adjoint operators. When talking about observables, we will be referring to the CQM notion of observables as (special) $\dagger$-Frobenius algebras.

\subsection{Coherent data manipulation}

Further to their correspondence with quantum observables, $\dagger$-Frobenius algebras find concrete use as fundamental building blocks of quantum algorithms and protocols \cite{Vicary2012a,Boixo2012,Coecke2016a,Gogioso2016d}. When designing quantum protocols, classical data is often encoded into quantum systems using orthonormal bases. In this context, the four processes forming a special $\dagger$-SCFA can be seen as the abstract, \inlineQuote{coherent} versions of some basic data manipulation primitives:
\begin{enumerate}
	\item[(a)] the comultiplication $\!\!\hbox{\input{symbols/ZbwcomultSym.tex}}\!\!\!= \ket{\psi_x} \mapsto \ket{\psi_x} \otimes \ket{\psi_x}$ acts as coherent copy;
	\item[(b)] the counit $\!\!\hbox{\input{symbols/ZbwcounitSym.tex}}\!\!\! = \ket{\psi_x} \mapsto 1$ acts as coherent deletion;
	\item[(c)] the multiplication $\!\!\hbox{\input{symbols/ZbwmultSym.tex}}\!\!\! = \ket{\psi_x} \otimes \ket{\psi_y} \mapsto \delta_{xy}\ket{\psi_x}$ acts as a coherent matching;
	\item[(d)] the unit $\!\!\hbox{\input{symbols/ZbwunitSym.tex}}\!\!\! = \sum_x \ket{\psi_x}$ acts as coherent superposition (up to normalisation).
\end{enumerate} 
These \inlineQuote{coherent} operations seldom appear alone, but are instead composed amongst themselves and with other primitives to form unitaries and CP maps appearing in quantum algorithms and protocols (as shown in the next Chapters). 

When talking about \textbf{coherent data}, we will be thinking of classical data, valued in some finite set $X$, which has been \textbf{coherently encoded} into an orthonormal basis $\ket{x}_{x \in X}$ of some finite-dimensional Hilbert space $\SpaceH$. Having fixed coherent encodings of its input and output data into orthonormal bases $\ket{x}_{x \in X}$ and $\ket{y}_{y \in Y}$ of Hilbert spaces $\SpaceH$ and $\SpaceK$, when talking about the \textbf{coherent} encoding of a classical function $F: X \rightarrow Y$ we will be thinking of its $\complexs$-linear map extension to $\SpaceH \rightarrow \SpaceK$:
\begin{equation}\label{coherentFunction}
f := \sum_{x \in X} \sum_{y \in Y} \ket{F(x)}\bra{x}
\end{equation}
The term \textbf{coherent} is chosen as the opposite of \textbf{decohered}, a term which we will use to denote classical data and processes (which we see as arising from quantum operations by appropriate decoherence). This choice of nomenclature is self-consistent: the original classical function $F$ (the decohered version) is obtained back by decohering the $\complexs$-linear map $f$ from Equation \ref{coherentFunction} (the coherent version) in the orthonormal bases that were used to encode the classical input and output data.

Coherent encodings of classical functions play a huge role in quantum computing and in the foundations of quantum theory: they are used to construct oracles (e.g. in the Deutsch-Josza algorithm, in Grover's algorithm and in the algorithm solving the abelian HSP), Frobenius algebras (see above), groups algebras (see below), entangled states (e.g. unnormalised Bell states, GHZ states and W states), and many other building blocks used to design quantum protocols. This is because coherent data enjoys all those features of $\complexs$-linearity (such as superposition, interference and non-commutative observables) which coalesce to provide quantum advantage and non-classical behaviour, while at the same time allowing for a good deal of classical intuition to play a positive role in designing quantum algorithms. A significant part of the development of CQM has been devoted to capturing the defining algebraic and diagrammatic properties of coherent manipulation of classical data, with the aim of designing quantum protocols and reasoning about them without explicit reference to the $\complexs$-linear structure itself. 

The defining properties of $\dagger$-SCFAs---one of the most fundamental structures in CQM---can indeed be seen as characterising the topological flow of classical information. As such, $\dagger$-SCFAs in arbitrary $\dagger$-SMCs are often interpreted as modelling coherent copy/deletion/matching operations on data which has been coherently encoded using their classical states, and are known as \textbf{classical structures} in the literature (a characterisation which becomes exact when the $\dagger$-SCFA has enough classical states). Because non-degenerate observables in quantum mechanics correspond exactly to all the possible ways that classical data can be coherently encoded into a quantum system, $\dagger$-qSCFAs are chosen in CQM to be the generalisation of non-degenerate observables and commutative C*-algebras from finite-dimensional quantum systems to arbitrary $\dagger$-SMCs. As a natural consequence, symmetric $\dagger$-qSFAs are chosen to be the generalisation of generic observables and C*-algebras.\footnote{The generalisation further extends to quasi-special $\dagger$-Frobenius algebras, which are are interpreted as the unnormalised versions of observables.}

If classical states for a symmetric $\dagger$-qSFA are interpreted as coherently encoding classical data in an object of a $\dagger$-SMC, it becomes interesting to understand which processes can be interpreted as coherently encoding classical functions on that data. Classical states are defined by three conditions: they are coherently copied, adjoined and deleted. Similarly, we should expect the coherent versions of classical functions to respect the same coherent copy, adjoin and delete operations that define the states, so that they will end up mapping classical states to classical states. As a consequence, we will say that a process $f: A \rightarrow B$ is \textbf{$\hbox{\input{symbols/ZbwdotSym.tex}}\!\!$-to-$\hbox{\input{symbols/YbwdotSym.tex}}\!\!$ classical}, where $\hbox{\input{symbols/ZbwdotSym.tex}}\!\!$ and $\hbox{\input{symbols/YbwdotSym.tex}}\!\!$ are symmetric $\dagger$-qSFAs on $A$ and $B$ respectively, if it satisfies the following three conditions:
\begin{equation}\begin{multlined}\label{eqn_classicalMap}
\resizebox{\textwidth}{!}{\input{pictures/chapter2/classicalMap.tikz}}
\end{multlined}\end{equation}
A $\hbox{\input{symbols/ZbwdotSym.tex}}\!\!$-to-$\hbox{\input{symbols/YbwdotSym.tex}}\!\!$ classical process always sends classical states of $\hbox{\input{symbols/ZbwdotSym.tex}}\!\!$ to classical states of $\hbox{\input{symbols/YbwdotSym.tex}}\!\!$. In $\fdHilbCategory$, the $\hbox{\input{symbols/ZbwdotSym.tex}}\!\!$-to-$\hbox{\input{symbols/YbwdotSym.tex}}\!\!$ classical process between $\dagger$-qSCFAs are exactly those taking the form of Equation \ref{coherentFunction}, where $\ket{x}_{x \in X}$ and $\ket{y}_{y \in Y}$ are the orthonormal bases associated with $\hbox{\input{symbols/ZbwdotSym.tex}}\!\!$ and $\hbox{\input{symbols/YbwdotSym.tex}}\!\!$ respectively.

\subsection{Bell states and effects}

The adjoin condition for classical processes involves a setup which is strongly reminiscent of the operator-state duality induced by the compact closed structure:
\begin{equation}\begin{multlined}\label{diagram_FrobeniusOperatorStateDuality}
\input{pictures/chapter2/FrobeniusOperatorStateDuality.tikz}
\end{multlined}\end{equation}
This is not a coincidence. Indeed, any symmetric $\dagger$-FA on a system $A$ in a $\dagger$-SMC induces a \textbf{symmetric cup} (also known as a \textbf{Bell state}) and a \textbf{symmetric cap} (also known as a \textbf{Bell effect}) on $A$:
\begin{equation}\begin{multlined}\label{diagram_FrobeniusCupCap}
\input{pictures/chapter2/FrobeniusCupCap.tikz}
\end{multlined}\end{equation}
Just like the symmetric cup and cap for a self-duality structure, the ones above satisfy the yanking equations (because of Frobenius law and unit laws), and Equations \ref{eqn_symmetryEquationsSD} (because of symmetry). 
When the category is dagger compact, the symmetric cup and cap defined by a $\dagger$-FA on a system $A$ also satisfy Equations \ref{eqn_cupCapSD} for the \textbf{self-duality isomorphism} $h_A : A \rightarrow A^\ast$ defined as follows:
\begin{equation}\begin{multlined}\label{diagram_FrobeniusSelfDualityIsom}
\input{pictures/chapter2/FrobeniusSelfDualityIsom.tikz}
\end{multlined}\end{equation}

\subsection{Observables in classical physics}

From the discussion above, it might sound like $\dagger$-FAs are restricted to modelling quantum observables, and are unsuitable to model classical observables: luckily, this couldn't be further from the truth. 

An observable on a classical system $X$ can be though to be a partition $X = \sqcup_{i \in I} X_i$ of its set of \textbf{microstates} $X$ into non-empty disjoint subsets $(X_i)_{i \in I}$, the \textbf{macrostates}, indexed by some classical outcome set $I$. However, a complete picture should also take into account the fact that different microstates $x \in X_i$ within the same macrostate $X_i$ are connected by certain internal symmetries, which can be modelled by a connected groupoid \cite{Bar2014}. Overall, we can think of an observable on a classical system $X$ as a groupoid\footnote{In this subsection, and only in this subsection, the symbol $\oplus$ is used to denote the disjoint union of groupoids, which is the co-product in the category of groupoids. Every groupoid is expressible in a unique way as the co-product of its connected components.} $\oplus_{i \in I} G_i$, where $G_i$ are connected groupoids and the macrostates are taken to be the underlying sets $X_i := |G_i|$ (i.e. are obtained by forgetting the internal symmetries).  As it turns out, this picture of classical observables coincides with that of (symmetric) $\dagger$-SFAs in the dagger compact category $\RelCategory$ of sets and relations.

\begin{theorem}[\textbf{Observables for classical systems \cite{Pavlovic2009,Coecke2014a}}]\hfill\\
The $\dagger$-SFAs $\hbox{\input{symbols/ZbwdotSym.tex}}\!\!$ on an object $X$ of $\RelCategory$ are the groupoids $G := \oplus_{i \in I} G_i$ on the set $X$: 
\begin{enumerate} 
	\item[(a)] the algebra multiplication $\!\hbox{\input{symbols/ZbwmultSym.tex}}\!\!$ is a partial function $X \times X \rightharpoonup X$, corresponding to the groupoid multiplication: 
	\begin{equation}
		\!\hbox{\input{symbols/ZbwmultSym.tex}}\!\! \circ (\ket{x} \otimes \ket{y}) = 
		\begin{cases}
			x \cdot y \text{ in } G_i &\text{ if both } x,y \in G_i \text{ for some $i$,}\\
			0 &\text{ otherwise};
		\end{cases}
	\end{equation} 
	\item[(b)] the algebra unit is the union $\!\hbox{\input{symbols/ZbwunitSym.tex}}\!\! = \bigcup_{i \in I} \sum_{u \text{ unit of }G_i} \ket{u}$ of all the units\footnote{Recall that elements in a groupoid can have distinct left and right units, so that even connected groupoids can have more than one unit.} for all the connected groupoids $(G_i)_{i \in I}$.  
\end{enumerate}
The $\hbox{\input{symbols/ZbwdotSym.tex}}\!\!$-classical states take the form $\ket{X_i} := \bigcup_{x \in G_i} \ket{x}$. Also, $\hbox{\input{symbols/ZbwdotSym.tex}}\!\!$ is necessarily symmetric.
\end{theorem}
\begin{proof}
The correspondence between $\dagger$-SFAs and groupoids is a result of \cite{Pavlovic2009,Coecke2014a}. In particular, $\hbox{\input{symbols/ZbwdotSym.tex}}\!\!$ is necessarily symmetric: $x \cdot y = u$ for some unit $u$ implies that $x,y,u \in G_i$ for some $i \in I$, and hence that $y \cdot x = v$ for some other unit $v \in G_i$. Now consider a $\hbox{\input{symbols/ZbwdotSym.tex}}\!\!$-classical state $\ket{\psi}$, with $\psi \subseteq X$. The copy condition is the requirement that $x,y \in \psi$ if and only if $x \cdot y$ is defined and $x \cdot y \in \psi$; the delete condition is the requirement that $u \in \psi$ for at least one unit $u$; the adjoin condition is the requirement that $x \in \psi$ if and only if $x^{-1} \in \psi$. Hence the classical states are exactly those in the form $\ket{X_i}$ for $X_i := |G_i|$.
\end{proof}

\subsection{Canonical Frobenius algebras}

Finally, a brief remark about the relationship between $\dagger$-Frobenius algebras and the CPM construction. In this Section, we have characterised $\dagger$-Frobenius algebras in $\fdHilbCategory$ as quantum observables, while in the previous Section we have noted that the real category modelling pure-state quantum theory is the pure subcategory of $\CPMCategory{\fdHilbCategory}$, rather than $\fdHilbCategory$. So a question arises: how does our characterisation of $\dagger$-Frobenius algebras in $\fdHilbCategory$ affect the pure subcategory of $\CPMCategory{\fdHilbCategory}$? The answer turns out to be rather simple. 

Because of their diagrammatic definition, the $\dagger$-FAs $(A,\!\hbox{\input{symbols/ZbwmultSym.tex}}\!\!,\!\hbox{\input{symbols/ZbwunitSym.tex}}\!\!,\!\hbox{\input{symbols/ZbwcomultSym.tex}}\!\!,\!\hbox{\input{symbols/ZbwcounitSym.tex}}\!\!)$ from a dagger compact $\CategoryC$ give rise to $\dagger$-FAs in the doubled subcategory of $\CPMCategory{\CategoryC}$: 
\begin{equation}
\big(A,\CPMdoubled{\!\hbox{\input{symbols/ZbwmultSym.tex}}\!\!},\CPMdoubled{\!\hbox{\input{symbols/ZbwunitSym.tex}}\!\!},\CPMdoubled{\!\hbox{\input{symbols/ZbwcomultSym.tex}}\!\!},\CPMdoubled{\!\hbox{\input{symbols/ZbwcounitSym.tex}}\!\!}\big)
\end{equation}
The $\dagger$-FAs in $\CPMCategory{\CategoryC}$ that arise this way are said to be \textbf{canonical}. In this work, we will only consider canonical $\dagger$-FAs when working with CPM categories.

\section{Measurements, decoherence and classicality}

\subsection{Probabilistic theories}

We have seen before that a generalised notion of finite classical systems, where probabilities are replaced by a generic involutive\footnote{The involution can simply be the identity $\id{R}: R \rightarrow R$. When talking about the probabilistic case $R = \reals^+$, we will always implicitly assume that the involution is the identity.} semiring $R$, can be modelled by the dagger compact category $\RMatCategory{R}$. We will refer to these as \textbf{classical $R$-probabilistic systems}, or simply \textbf{classical probabilistic systems} in the case $R = \reals^+$. 

The category $\fSetCategory$ of finite sets and functions, modelling finite deterministic classical systems, is always a sub-SMC of $\RMatCategory{R}$, which it endows with the following environment structure:
\begin{equation}
\trace{R^X} := \Big(\sum_{x \in X} p_x \ket{x}\Big) \mapsto \sum_{x \in X} p_x
\end{equation}
The normalised states in $\RMatCategory{R}$ are the \textbf{$R$-distributions}, the states $\sum_{x \in X} p_x \ket{x}$ such that $\sum_{x \in X} p_x = 1$, and the normalised processes are the \textbf{$R$-stochastic maps}, the linear maps $R^X \rightarrow R^Y$ which send $R$-distributions on $X$ to $R$-distributions on $Y$. In the case of classical probabilistic systems, normalised states are probability distributions on finite sets, and normalised processes are stochastic maps.

The category $\RMatCategory{R}$ is enriched in the category of commutative monoids (or $\CMonCategory$-enriched), by which we mean that the processes $R^X \rightarrow R^Y$ between fixed systems $R^X$ and $R^Y$ form a commutative monoid, and that composition, tensor product and dagger all respect the commutative monoid structure. Specifically, the addition $f+g: R^X \rightarrow R^Y$ between morphism $f,g: R^X \rightarrow R^Y$ in $\RMatCategory{R}$ is given by addition of matrices, and the zero element $0: R^X \rightarrow R^Y$ is given by the zero matrix. Furthermore, the tensor product is linear, i.e. it distributes over the addition and respects the zero element:
\begin{equation}
\begin{cases}
f \otimes (g+h) &= f\otimes g + f \otimes h \\ 
(g+h) \otimes f &= g\otimes f + h \otimes f
\end{cases} 
\hspace{3cm} 
\begin{cases}
f \otimes 0 &= 0 \\
0 \otimes f &= 0
\end{cases}
\end{equation}
We will use \textbf{distributively $\CMonCategory$-enriched} to refer to a SMC which is $\CMonCategory$-enriched with linear tensor product. When talking about distributively $\CMonCategory$-enriched $\dagger$-SMCs, we will furthermore require that the dagger be linear.

It is important to note that the scalars of a distributively $\CMonCategory$-enriched SMC always form a commutative semiring, which in a distributively $\CMonCategory$-enriched $\dagger$-SMC further comes with an involution given by the dagger. Using enrichment, the discarding maps $\trace{R^X}$ can be expressed as follows:
\begin{equation}\begin{multlined}\label{eqn_discardingMapsClassical}
\input{pictures/chapter2/discardingMapsClassical.tikz}
\end{multlined}\end{equation}

When working in the foundations of quantum theory, the existence and behaviour of classical systems are often entirely taken for granted, and manifest themselves in a variety of ways across the different frameworks and formalisms. In a purely process-oriented framework, such as the one underlying this work, classical systems and processes should be explicitly modelled by the physical theory under investigation, together with their interface with other systems. On these lines, we distil four requirements that any such theory should respect when modelled by a SMC $\CategoryC$:
\begin{enumerate}
	\item[1.] the SMC $\CategoryC$ has $\RMatCategory{R}$ (or a category equivalent to it) as a full sub-SMC, where $R$ is the semiring encoding the desired notion of non-determinism (we will refer to this as the \textbf{classical subcategory}, and to its objects as the \textbf{classical systems});
	\item[2.] the SMC $\CategoryC$ is distributively $\CMonCategory$-enriched, with scalars forming the semiring $R$, and the enrichment of $\CategoryC$ extends the enrichment defined above for $\RMatCategory{R}$;
	\item[3.] the SMC $\CategoryC$ comes with a choice of environment structure, extending the environment structure defined above for $\RMatCategory{R}$.
\end{enumerate}
We will refer to a SMC satisfying the requirements above as a \textbf{$R$-probabilistic theory}, or simply as a \textbf{probabilistic theory} in the case $R=\reals^+$. 
$R$-probabilistic theories automatically come with a number of handy features built in, amongst which: marginalisation, conditioning, classical control, and convex combination. 
In this work, we will restrict ourselves to the special case of $R$-probabilistic CP* categories, but a full discussion of $R$-probabilistic theories --- in connection to the framework of Operational Probabilistic Theories \cite{Chiribella2010,Chiribella2011} --- has appeared in \cite{Gogioso2016f}.

\section{The CP* construction}

\subsection{The CP* construction and quantum theory}


We now wish to construct a $R$-probabilistic theory from scratch, using Frobenius algebras to define decoherence maps. Let $\CPMCategory{\CategoryC}$ be a CPM category, and let $\hbox{\input{symbols/ZbwdotSym.tex}}\!\!$ be a canonical symmetric $\dagger$-SFA (i.e. a symmetric $\dagger$-SFA in $\CategoryC$) on some system $A$. The \textbf{decoherence map} $\decoh{\hbox{\input{symbols/ZbwdotSym.tex}}\!\!}$ associated to $\hbox{\input{symbols/ZbwdotSym.tex}}\!\!$ is the process $A \rightarrow A$ in $\CPMCategory{\CategoryC}$ defined as follows:
\begin{equation}\begin{multlined}\label{diagram_decoherence}
\input{pictures/chapter2/decoherence.tikz}
\end{multlined}\end{equation}
Decoherence maps associated to symmetric $\dagger$-SFAs are always normalised (because of speciality), idempotent (because of associativity and speciality), and self-adjoint (because of Frobenius law, unit laws and symmetry). In the quantum case of $\CPMCategory{\fdHilbCategory}$, decoherence maps take the following form, in terms of the complete family of orthogonal projectors $(p_x)_{x \in X}$ associated to $\hbox{\input{symbols/ZbwdotSym.tex}}\!\!$:
\begin{equation}\begin{multlined}\label{eqn_decoherenceQM}
\input{pictures/chapter2/decoherenceQM.tikz}
\end{multlined}\end{equation}
That is, the decoherence maps defined by symmetric canonical $\dagger$-SFAs in $\CPMCategory{\fdHilbCategory}$ are exactly the decoherence maps that are traditionally associated with quantum observables (seen as complete families of orthogonal projectors).

The \textbf{Karoubi envelope}\footnote{Also known as \textit{idempotent completion}, or \textit{Cauchy completion}.} of a SMC $\CategoryD$, which we denote by $\KaroubiEnvelope{\CategoryD}$, is the SMC defined as follows:
\begin{enumerate}
	\item[(i)] the objects of $\KaroubiEnvelope{\CategoryD}$ are the pairs $(A,e)$ of an object $A$ of $\CategoryD$ and an idempotent morphism $e: A \rightarrow A$;
	\item[(ii)] the morphisms $(A,e) \rightarrow (B,e')$ in $\KaroubiEnvelope{\CategoryD}$ are the morphisms $f: A \rightarrow B$ in $\CategoryD$ which satisfy $e' \circ f \circ e = f$, i.e. which are invariant under pre-composition with $e$ and under post-composition with $e'$;
	\item[(iii)] composition is inherited from $\CategoryD$, while the identity on object $(A,e)$ is the morphism $e: (A,e) \rightarrow (A,e)$.
\end{enumerate}
Because $\CategoryD$ is a SMC, the Karoubi envelope $\KaroubiEnvelope{\CategoryD}$ is also a SMC, and contains $\CategoryD$ as the full sub-SMC spanned by the objects in the form $(A,\id{A})$.

We now consider the Karoubi envelope $\KaroubiEnvelope{\CPMCategory{\CategoryC}}$ and restrict our attention the the full sub-SMC\footnote{Because decoherence maps obtained from symmetric $\dagger$-SFAs are self-adjoint, the subcategory is in fact a $\dagger$-SMC.} having objects in the form $(A,\decoh{\hbox{\input{symbols/ZbwdotSym.tex}}\!\!})$, where $\hbox{\input{symbols/ZbwdotSym.tex}}\!\!$ is a canonical symmetric $\dagger$-SFA $\hbox{\input{symbols/ZbwdotSym.tex}}\!\!$ on $A$. The processes $(A,\decoh{\hbox{\input{symbols/ZbwdotSym.tex}}\!\!}) \rightarrow (B,\decoh{\hbox{\input{symbols/YbwdotSym.tex}}\!\!})$ are exactly those satisfying the following condition: 
\begin{equation}\begin{multlined}\label{eqn_CPStarMorphism}
\input{pictures/chapter2/CPStarMorphism.tikz}
\end{multlined}\end{equation}

\noindent The full sub-SMC of $\KaroubiEnvelope{\CPMCategory{\CategoryC}}$ defined above is known in the literature as a \textbf{CP* category} \cite{Coecke2014a,Cunningham2015}, and denoted by $\CPStarCategory{\CategoryC}$. Traditionally, CP* categories are constructed as a generalisation of the category of finite-dimensional C*-algebras, using the correspondence with symmetric $\dagger$-SFAs on $\fdHilbCategory$ proven by \cite{Vicary2011}: this is known as the \textit{CP* construction}, and the resulting category is exactly the same as the one we constructed above using decoherence maps and the Karoubi envelope. In the quantum case, these two equivalent ways of constructing $\CPStarCategory{\fdHilbCategory}$ reflect different perspectives on the quantum-classical interface:
\begin{enumerate}
	\item[(a)] the decoherence construction we presented gives an operational perspective, showing that $\CPStarCategory{\fdHilbCategory}$ is the category of super-selected finite-dimensional quantum systems\footnote{This can be seen from Equation \ref{eqn_decoherenceQM}, which expresses the decoherence in terms of the complete family of projectors associated with a quantum observable. The projectors determine the super-selection sectors associated with the observable, and the morphisms between different super-selected quantum systems are exactly the CP maps that respect the super-selection sectors.};
	\item[(b)] the CP* construction gives an algebraic/logical perspective, showing that $\CPStarCategory{\fdHilbCategory}$ is the category of finite-dimensional C*-algebras.
\end{enumerate}

Amongst the many objects of the CP* category $\CPStarCategory{\fdHilbCategory}$, two particular groups stand out: (a) the objects associated with $\dagger$-SCFAs (or commutative C*-algebras), corresponding to quantum system with 1-dimensional super-selection sectors; (b) the objects associated with the matrix algebras, corresponding to quantum systems with a single super-selection sector (i.e. which are trivially super-selected). The objects associated with the matrix algebras span a full sub-SMC which is equivalent to $\CPMCategory{\fdHilbCategory}$, and as a consequence $\CPStarCategory{\fdHilbCategory}$ is interpreted as an extension of mixed-state quantum theory. The objects associated with $\dagger$-SCFAs, on the other hand, span a full sub-SMC which is equivalent to $\RMatCategory{\reals^+}$: indeed, if we denote by $\classicalStates{\hbox{\input{symbols/ZbwdotSym.tex}}\!\!}$ the set of classical states of a $\dagger$-SCFA $\hbox{\input{symbols/ZbwdotSym.tex}}\!\!$, then in $\CPMCategory{\fdHilbCategory}$ the decoherence map $\decoh{\hbox{\input{symbols/ZbwdotSym.tex}}\!\!}$ takes the following form:
\begin{equation}
\decoh{\hbox{\input{symbols/ZbwdotSym.tex}}\!\!} = \rho \mapsto \sum_{x \in \classicalStates{\hbox{\input{symbols/ZbwdotSym.tex}}\!\!}} \ket{x}\bra{x} \, \rho \, \ket{x} \bra{x}
\end{equation}
Furthermore, the CPM category $\CPMCategory{\fdHilbCategory}$ is distributively $\CMonCategory$-enriched, and both the enrichment and the discarding maps transfer to $\CPStarCategory{\fdHilbCategory}$: as a consequence, $\CPStarCategory{\fdHilbCategory}$ is a probabilistic theory, with classical systems given by objects associated with $\dagger$-SCFAs.

\subsection{The CP* construction and R-probabilistic theories}

We now go back to $\KaroubiEnvelope{\CPMCategory{\CategoryC}}$ for a generic dagger compact category $\CategoryC$. First of all, we tackle a technical issue with the CP* construction: in the quantum case, the CP* category includes the CPM category as the full subcategory given by the quantum systems with trivial super-selection (those associated with the matrix algebras). In general, however, there is no guarantee that the objects associated with the matrix algebras will span a full subcategory of $\CPStarCategory{\CategoryC}$ isomorphic to $\CPMCategory{\CategoryC}$. Here is where our construction and the CP* construction diverge: the latter aims to study a generalisation of the category of finite-dimensional C*-algebras, while our aim is to construct an $R$-probabilistic theory which \textbf{extends} a given physical theory modelled by $\CPMCategory{\CategoryC}$. As a consequence, \underline{we modify the definition of the CP* category}.

From now on, we \underline{redefine} $\CPStarCategory{\CategoryC}$ to be the full sub-category of $\KaroubiEnvelope{\CPMCategory{\CategoryC}}$ spanned by objects in the form $(A,\id{A})$ and objects in the form $(A,\hbox{\input{symbols/ZbwdotSym.tex}}\!\!)$, with $\hbox{\input{symbols/ZbwdotSym.tex}}\!\!$ a canonical symmetric $\dagger$-SFA on $A$. We refer to objects in the form $(A,\id{A})$ as the \textbf{CPM systems}, and to objects in the form $(A,\decoh{\hbox{\input{symbols/ZbwdotSym.tex}}\!\!})$ as the \textbf{decohered systems}. Following established conventions, we denote CPM systems $(A,\id{A})$ as $A$, and decohered systems $(A,\decoh{\hbox{\input{symbols/ZbwdotSym.tex}}\!\!})$ as $(A,\hbox{\input{symbols/ZbwdotSym.tex}}\!\!)$. The CPM systems always span a full subcategory isomorphic to $\CPMCategory{\CategoryC}$, and we use the doubled notation from CPM categories to denote them and the morphisms between them. Again following established conventions, we use single wires for decohered systems, and single borders for morphisms solely involving decohered systems, while we retain doubled notation for morphisms involving both decohered systems and CPM systems. 

If $\hbox{\input{symbols/ZbwdotSym.tex}}\!\!$ is a canonical symmetric $\dagger$-SFA on an object $A$, the decoherence map $\decoh{\hbox{\input{symbols/ZbwdotSym.tex}}\!\!}$ is always a process $\decoh{\hbox{\input{symbols/ZbwdotSym.tex}}\!\!}: A \rightarrow A$ in $\CPStarCategory{\CategoryC}$. Because of idempotence, however, it is also a process $A \rightarrow (A,\hbox{\input{symbols/ZbwdotSym.tex}}\!\!)$ and a process $(A,\hbox{\input{symbols/ZbwdotSym.tex}}\!\!) \rightarrow A$: we will refer to the former as the \textbf{measurement} in $\hbox{\input{symbols/ZbwdotSym.tex}}\!\!$, and the latter as the \textbf{preparation} in $\hbox{\input{symbols/ZbwdotSym.tex}}\!\!$. The single and doubled notation distinguish between the different cases:
\begin{equation}\begin{multlined}\label{diagram_preparationMeasurementCPStar}
\input{pictures/chapter2/preparationMeasurementCPStar.tikz}
\end{multlined}\end{equation}
Because of idempotence, the measurement and preparation for $\hbox{\input{symbols/ZbwdotSym.tex}}\!\!$ satisfy the abstract properties defining measurement-preparation pairs in $R$-probabilistic theories \cite{Gogioso2016f}, save for the fact that we don't yet have an $R$-probabilistic theory in our hands (and, even if we did, $(A,\hbox{\input{symbols/ZbwdotSym.tex}}\!\!)$ need not always be a classical system). However, we have already seen that $\CPStarCategory{\fdHilbCategory}$ is a probabilistic theory, and we can begin by checking that measurements and preparations as defined above yield the usual notions in the case of quantum theory, when $(A,\hbox{\input{symbols/ZbwdotSym.tex}}\!\!)$ is a classical system:
\begin{equation}\begin{multlined}\label{eqns_preparationMeasurementQM}
\input{pictures/chapter2/preparationMeasurementQM.tikz}
\end{multlined}\end{equation}

When saying that $\CPStarCategory{\CategoryC}$ is an \textbf{$R$-probabilistic CP* category}, we will mean that it satisfies the following requirements:
\begin{enumerate}
	\item[(i)] the category $\CPStarCategory{\CategoryC}$ is distributively $\CMonCategory$-enriched, with $R$ as its involutive semiring of scalars\footnote{Equivalently, we can ask for $\CPMCategory{\CategoryC}$ to be enriched, as the two categories mutually inherit enrichment and discarding maps.}.
	\item[(ii)] the \textbf{classical systems} are defined to be those decohered systems $(A,\hbox{\input{symbols/ZbwdotSym.tex}}\!\!)$ where $\hbox{\input{symbols/ZbwdotSym.tex}}\!\!$ is a $\dagger$-SCFA with enough classical states, and such that the $\hbox{\input{symbols/ZbwdotSym.tex}}\!\!$-classical states are orthonormal and form a finite set;
	\item[(iii)] for each $n \in \naturals$, there is some classical system $(A,\hbox{\input{symbols/ZbwdotSym.tex}}\!\!)$ such that the $\dagger$-SCFA $\hbox{\input{symbols/ZbwdotSym.tex}}\!\!$ has exactly $n$ classical states.
\end{enumerate}
Indeed, processes $(A,\hbox{\input{symbols/ZbwdotSym.tex}}\!\!) \rightarrow (B,\hbox{\input{symbols/YbwdotSym.tex}}\!\!)$ between two classical systems in the CP* category form an $R$-module which is isomorphic to the $R$-module of processes $R^{\classicalStates{\hbox{\input{symbols/YbwdotSym.tex}}\!\!}} \rightarrow R^{\classicalStates{\hbox{\input{symbols/ZbwdotSym.tex}}\!\!}}$ in the category $\RMatCategory{R}$ of classical $R$-probabilistic systems: 
\begin{enumerate}
\item[(i)] firstly, every process $f: (A,\hbox{\input{symbols/ZbwdotSym.tex}}\!\!) \rightarrow (B,\hbox{\input{symbols/YbwdotSym.tex}}\!\!)$ is determined by the $R$-valued matrix $\big(\bra{y}f\ket{x}\big)_{x \in \classicalStates{\hbox{\input{symbols/ZbwdotSym.tex}}\!\!}}^{y \in \classicalStates{\hbox{\input{symbols/YbwdotSym.tex}}\!\!}}$ obtained by testing against classical states of the two $\dagger$-SCFAs:
	\begin{equation}\begin{multlined}\label{eqns_CPStarProcessesRModule1}
	\hspace{-1.25cm}\input{pictures/chapter2/CPStarProcessesRModule1.tikz}
	\end{multlined}\end{equation}
\item[(ii)] secondly, for every matrix $\big(F_x^y\big)_{x \in \classicalStates{\hbox{\input{symbols/ZbwdotSym.tex}}\!\!}}^{y \in \classicalStates{\hbox{\input{symbols/YbwdotSym.tex}}\!\!}}$ there is a unique process $(A,\hbox{\input{symbols/ZbwdotSym.tex}}\!\!) \rightarrow (B,\hbox{\input{symbols/YbwdotSym.tex}}\!\!)$ corresponding to it:
	\begin{equation}\begin{multlined}\label{eqns_CPStarProcessesRModule2}
	\input{pictures/chapter2/CPStarProcessesRModule2.tikz}
	\end{multlined}\end{equation}
\end{enumerate}
Hence the full sub-SMC of an $R$-probabilistic CP* category spanned by the classical systems is equivalent to $\RMatCategory{R}$, and our definition of classical systems for a CP* category is consistent with the nomenclature used in $R$-probabilistic theories.

\section{Non-locality and contextuality}

While generally of interest to understand the background of this work, this section is only directly relevant to the proof of non-locality for generalised Mermin-type arguments, and to the proof of device-independent security for the related quantum-classical secret sharing protocols.

\subsection{The sheaf-theoretic framework} 

In the context of this work, \textbf{contextuality} and \textbf{non-locality} will be used interchangeably, and will be understood in the sense of the sheaf-theoretic framework of \cite{Abramsky2011}. Consider the abstract setup of a Bell test:
\begin{enumerate} 
\item[(i)] $N$ parties are given devices $B_1,...,B_N$ which might share some global state $\rho$;
\item[(ii)] each device $B_j$ takes an input, the \textbf{measurement choice}, freely chosen by party $j$ from some finite set $M_j$;
\item[(iii)] upon receiving input $m_j \in M_j$, the device $B_j$ produces some output $o_j$, the \textbf{measurement outcome}, in some finite set $O_j$;
\item[(iv)] no signalling is possible between the devices from before the first input is given to after the last outputs has been produced. 
\end{enumerate}
The sheaf-theoretic framework characterises the distribution on joint outputs conditional on joint inputs from the point of view of sheaf theory, showing that non-locality and contextuality are related to the (non-) existence of global sections for a particular presheaf. The framework doesn't rely on any concrete description of the state $\rho$ or the devices $B_1,...,B_N$, focusing instead on the distributional properties of joint measurement outcomes $\underline{o} := (o_1,...,o_N)$ conditional to the joint measurement choice $\underline{m} := (m_1,...,m_N)$.

The framework begins by identifying a finite set $\mathcal{X}$ of inputs, which in the Bell test setup above (the one used in this work) would be $\mathcal{X} = \sqcup_{j=1}^N M_j$. The disjoint union preserves information about which party each measurement is associated to, so we will adopt the notation $m_j$ for generic elements of $\mathcal{X}$, where $m$ is the measurement and $j$ is the party. For each subset $U \subseteq \mathcal{X}$, the family of all potential\footnote{Not all subsets of measurements need be compatible in each concrete scenario: see below for the definition of measurement contexts.} \textbf{joint outcomes} takes the following form:
\begin{equation}
	\sheafOfEvents{U} := \prod_{m_j \in U} O_j
\end{equation}
The powerset $\Powerset{\mathcal{X}}$ is a poset (hence a poset category) under inclusion $V \subseteq U$ of subsets. We can define a functor $\sheafOfEventsSym: \OpCategory{\Powerset{\mathcal{X}}} \rightarrow \SetCategory$, i.e. a \textbf{presheaf}, by setting:
\begin{enumerate}
	\item[(i)] if $U \in  \Powerset{\mathcal{X}} $, then we define $ \sheafOfEvents{U} := \prod_{m_j \in U} O_j$ as above;
	\item[(ii)] if $V \subseteq U $, then we define $\sheafOfEvents{V \subseteq U}:=\restrictionMap{U}{V}$ to be the following \textbf{restriction map} $U \stackrel{\SetCategory}{\longrightarrow} V$:
	\begin{equation}
		\restrictionMap{U}{V} = s \mapsto \restrict{s}{V}
	\end{equation}
\end{enumerate}
A \textbf{section $s$ over $U$} (or \textbf{$U$-section}) is one in the following form:
\begin{equation}
	s = \suchthat{(m_j,s(m_j))}{m_j \in U} \in  \prod_{m_j \in U} O_j
\end{equation}
The restriction map sends a section $s$ over $U$ to its restriction $\restrict{s}{V}$ over $V$:
\begin{equation}
	\restrict{s}{V} = \suchthat{(m_j,s(m_j))}{m_j \in V} \in  \prod_{m_j \in V} O_j
\end{equation}

The definition of the set of possible joint inputs requires further consideration: it is a fundamental feature of quantum mechanics that not all measurements on a system are compatible, and we shouldn't expect different measurement choices in each $M_j$ to have a consistent assignment of outcomes. Instead, the framework requires us to specify a set $\mathcal{M}$ of \textbf{measurement contexts}, subsets $C \subseteq \mathcal{X}$ of measurements which are mutually compatible (and therefore have a well-defined notion of joint outcome). Even though more general setups are allowed, we will assume that our measurement contexts all take the form $C = \{m_1,...,m_N\}$ for $m_j \in M_j$, which we will denote by $\underline{m}$: each party chooses exactly one input for their device, but we allow the possibility that not all combinations of inputs might be allowed/interesting. The only requirement is that $\cup_{C \in \mathcal{M}} C = \mathcal{X}$, i.e. that $\mathcal{M}$ be a \textbf{global cover} of $\mathcal{X}$ (each measurement choice for each player appears in at least one measurement context), where we consider $\mathcal{X}$ to be endowed with the discrete topology. One can also define the \textbf{local covers} for any $U \subseteq \mathcal{X}$ as the families $(U_i)_{i \in I}$ such that $\cup_{i \in I} U_i = U$.

The choice of the discrete topology on $\mathcal{X}$ makes $\Powerset{\mathcal{X}}$ its locale of open subsets, and one can define a notion of \textbf{sheaf} on it. Because it is defined in terms of sections\footnote{Compatibility of local sections amounts to compatibility over the intersection of the domains, and hence compatible local sections can always be glued together.}, the presheaf $\sheafOfEventsSym$ is in fact a sheaf on the locale $\Powerset{\mathcal{X}}$, and we shall refer to it as the \textbf{sheaf of events}. The sheaf of events $\sheafOfEventsSym$ and the measurement cover $\mathcal{M}$ are the two ingredients required to define a \textbf{measurement scenario} $(\sheafOfEventsSym,\mathcal{M})$: the former gives the joint measurement outcomes conditional on all possible measurement choices, while the latter specifies the compatible joint measurement choices.

The next step in the framework sees the introduction of generalised notions of probabilities and distributions. In quantum mechanics, probabilities can be seen as taking values in the commutative semiring $\reals^+ := (\reals^+,+,0,\cdot,1)$ of non-negative reals (in fact they fall within the interval $[0,1]$, a consequence in the semiring $\reals^+$ of the normalisation condition requiring that probabilities add up to $1$). In other circumstances, one may be interested in the \textbf{possibilities} associated with events, living in the commutative semiring $\mathbb{B} = (\{0,1\},\vee,0,\wedge,1)$ of the booleans. In the sheaf-theoretic treatment of contextuality, one works with an arbitrary commutative semiring $R = (|R|,+,0,\cdot,1)$. 

Given a set $X$, an \textbf{$R$-distribution} on $X$ is a function $d:X \rightarrow R$ which has finite \textbf{support} $\support{d} := \suchthat{s \in X}{d(s) \neq 0}$ and such that:
\begin{equation}
	\sum_{s \in \support{d}} d(s) = 1
\end{equation}  
One can then define a functor $\distributionFunctorSym{R}: \SetCategory \rightarrow \SetCategory$ as follows:
\begin{enumerate}
	\item[(i)] for any set $X$, define $\distributionFunctor{R}{X}$ to be the set of $R$-distributions over $X$;
	\item[(ii)] for any function $f: X \rightarrow Y$, define $\distributionFunctor{R}{f}: \distributionFunctor{R}{X} \rightarrow \distributionFunctor{R}{Y}$ to be the following function:
	\begin{equation}
		\distributionFunctor{R}{f} = d \mapsto \left[ t \mapsto \sum_{f(s) = t} d(s) \right]
	\end{equation}
\end{enumerate}
Composing this functor with the sheaf of events yields the \textbf{presheaf of distributions} $\presheafOfDistributionsSym{R}:\OpCategory{\Powerset{\mathcal{X}}} \rightarrow \SetCategory$, which captures the structure of $R$-distributions on joint measurement outcomes under marginalisation. The presheaf sends each set $U$ of measurements (the objects of the presheaf category $\Powerset{\mathcal{X}}$) to the set $\presheafOfDistributions{R}{U}$ of \textbf{$R$-distributions on $U$-sections}, and sends any inclusion $V \subseteq U$ (the morphisms of the presheaf category $\Powerset{\mathcal{X}}$) to the corresponding marginalisation of distributions:
\begin{equation}
	\presheafOfDistributions{R}{V \subseteq U} = d \mapsto \restrict{d}{V} := \left[ t \mapsto \sum_{\restrict{s}{V} = t} d(s)\right]
\end{equation}
We will refer to $\restrict{d}{V}$ as the \textbf{marginal} of $d$. Indeed, $\restrict{d}{V}$ can be manipulated into taking the following, familiar form:
\begin{equation}
	t \in \sheafOfEvents{V} \implies \restrict{d}{V}(t) := \hspace{-0.75cm}\sum_{s \in \sheafOfEvents{V} \text{ s.t. } \restrict{s}{V} = t} \hspace{-0.75cm} d(s)
\end{equation}

In quantum mechanics, if $C$ is a set of compatible measurements on some state $\rho$, then there is a probability distribution $d \in \presheafOfDistributions{\reals^+}{C}$ on the joint outcomes of the measurements, and the typical contextuality argument involves showing that the probability distributions on different contexts cannot be obtained, in a no-signalling scenario, as marginals of some non-contextual hidden variable. In the sheaf-theoretic framework, a \textbf{(no-signalling) empirical model} is defined to be a compatible family of distributions $(\zeta_C)_{C \in \mathcal{M}}$ for the global cover $\mathcal{M}$ of measurement contexts\footnote{A compatible family of distributions $(a_C)_{C \in \mathcal{M}}$ is one such that $\restrict{a_C}{C \cap C'} = \restrict{a_{C'}}{C \cap C'}$ for all possible pairs of measurement contexts $C,C' \in \mathcal{M}$.}; the usual no-signalling property is shown in \cite{Abramsky2011} to be a special case of the compatibility condition. In other literature (usually treating probabilistic models), empirical models for Bell tests are usually given explicitly as conditional (probability) distributions, in a format akin to the following:
\begin{equation}
\zeta_{\,\underline{m}}\big(\underline{o}\big) := \mathbb{P}\big[\,\underline{o}\,\big\vert\, \underline{m}\,\big]
\end{equation}
where $\underline{m} = (m_1,...,m_N) \in \mathcal{M}$ are the measurement contexts used by the scenario and $\underline{o} \in \prod_{j} O_j$ are the joint outcomes. This is the format we will use in the last section of this work. In the probabilistic case, empirical models for a fixed scenario form a polytope. However, this need not be the same as the no-signalling polytope which is traditionally studied in quantum information theory, because the set of measurement contexts need not include all possible combinations of all possible measurements for each party (i.e. it need not always be the case that $\mathcal{M} = \prod_j M_j$, although it is necessarily the case that $\mathcal{M} \subseteq \prod_j M_j$). 

A \textbf{global section} for an empirical model\footnote{From now on, no-signalling is implicitly assumed.} $(\zeta_C)_{C \in \mathcal{M}}$ is a distribution $d \in \presheafOfDistributions{R}{\mathcal{X}}$ over the joint outcomes of all measurements which marginalises to the distributions specified by the empirical model:
\begin{equation}
	\restrict{d}{C}=\zeta_C \text{ for all } C \in \mathcal{M}
\end{equation} 
\noindent The fundamental observation behind the sheaf-theoretic framework is that the existence of a global section for an empirical model is equivalent to the existence of a \textbf{non-contextual hidden variable model} (also known as a \textbf{local hidden variable model}). Concretely, the existence of a global section $d$ means that there is a finite set $\Lambda$, an $R$-distribution $q(\lambda): \Lambda \rightarrow R$ and a family of functions $f_j^{\lambda}: M_j \rightarrow O_j$ such that:
\begin{equation}
\zeta_{\,\underline{m}}\big(\underline{o}\big) = \sum_{\lambda \in \Lambda} q(\lambda) \prod_{j} \delta_{f_j^{\lambda}(m_j) = o_j}
\end{equation}
We will say that an empirical model $(\zeta_C)_{C \in \mathcal{M}}$ is \textbf{contextual} (or \textbf{non-local}) if it does not admit a global section. 

Contextuality of probabilistic models is interesting in itself, but more refined notions can be obtained by relating $\reals^+$ to two other semirings: the reals, modelling signed probabilities, and the booleans, modelling possibilities. Observe that the construction $\distributionFunctorSym{R}$ is functorial in $R$, so that for any morphism of semirings $r: R \rightarrow R'$ we can define the following:
\begin{equation}
	\distributionFunctor{r}{U} := \left(d:U \rightarrow R\right) \mapsto \left(r \circ d: U \rightarrow R'\right)
\end{equation}
\noindent In particular, there is an injective morphism of semirings $i^+: \reals^+ \inject R$ sending $x \in \reals^+$ to $x \in \reals$, and a surjective morphism of semirings $p: \reals^+ \rightarrow \mathbb{B}$ sending $0 \mapsto 0$ and $x \mapsto 1$ for all $x > 0$ (the latter mapping is well defined for all positive semirings).  

If $(\zeta_C)_{C \in \mathcal{M}}$ is a probabilistic empirical model, i.e. one in the semiring $\reals^+$, then $(\zeta_C)_{C \in \mathcal{M}}$ can be seen as an empirical model $(i^+ \circ \zeta_C)_{C \in \mathcal{M}}$ in the semiring $\reals$: regardless of whether $(\zeta_C)_{C \in \mathcal{M}}$ was contextual or not over $\reals^+$, it can be shown \cite{Abramsky2011} that over the reals $\reals$ it always admits a global section. On the other hand, any probabilistic empirical model $(\zeta_C)_{C \in \mathcal{M}}$ can be assigned a corresponding possibilistic empirical model $(p\circ \zeta_C)_{C \in \mathcal{M}}$ in the semiring $\mathbb{B}$ of the booleans (and each boolean function $p \circ \zeta_C$ can equivalently be seen as the characteristic function of the subset $\support{\zeta_C} \subseteq \sheafOfEvents{C}$). 

Note that contextuality is a contravariant property with respect to change of semiring: if $(\zeta_C)_{C \in \mathcal{M}}$ is an empirical model in a semiring $R$ and $r: R \rightarrow R'$ is a morphism of semirings, then contextuality of $(r \circ \zeta_C)_{C \in \mathcal{M}}$ implies contextuality of $(\zeta_C)_{C \in \mathcal{M}}$ (because a global section $d$ of the latter is mapped to a global section $r \circ d$ of the former). We will say that a probabilistic empirical model $(\zeta_C)_{C \in \mathcal{M}}$ is \textbf{possibilistically contextual} if the corresponding possibilistic model $(p \circ \zeta_C)_{C \in \mathcal{M}}$ is contextual (as opposed to \textbf{probabilistically contextual}, which we use to say that $(\zeta_C)_{C \in \mathcal{M}}$ is contextual over $\reals^+$). Because of contravariance, possibilistic contextuality implies probabilistic contextuality, but the opposite is not true: the Bell model given in \cite{Abramsky2011} is probabilistically contextual but not possibilistically contextual.

Seeing distributions $d \in \presheafOfDistributions{\mathbb{B}}{U}$ as indicator functions of the subsets $\support{d} \subseteq \sheafOfEvents{U}$ endows them with a partial order:
\begin{equation}
	\label{eqn_SCimpliesC}
	d' \preceq d \text{ if and only if } \support{d'} \subseteq \support{d}
\end{equation}
\noindent The existence of a global section $d \in \presheafOfDistributions{\mathbb{B}}{U}$ for a possibilistic empirical model $(\zeta_C)_{C \in \mathcal{M}}$ implies that:
\begin{equation}
	\label{eqn_StrongContextualityCondition}
	\restrict{d}{C} \preceq \zeta_C \text{ for all } C \in \mathcal{M}
\end{equation}
We say that a possibilistic empirical model $(\zeta_C)_{C \in \mathcal{M}}$ is \textbf{strongly contextual} if there is no distribution $d \in \presheafOfDistributions{\mathbb{B}}{\mathcal{X}}$ such that Equation \ref{eqn_StrongContextualityCondition} holds. In particular, the GHZ model given in \cite{Abramsky2011}, corresponding to Mermin's original non-locality argument, is strongly contextual. Because of Equation \ref{eqn_SCimpliesC}, strong contextuality implies contextuality, but the opposite is not true: the possibilistic Hardy model give in \cite{Abramsky2011} is possibilistically contextual, but not strongly contextual. We will say that a probabilistic empirical model is strongly contextual if the associated possibilistic empirical model is strongly contextual, yielding the following hierarchy of notions of contextuality for probabilistic empirical models:
\begin{equation}
	\text{probabilistically contextual } \Leftarrow \text{ possibilistically contextual } \Leftarrow \text{ strongly contextual}
\end{equation}

\subsection{Contextuality in R-probabilistic theories}

The relevance of the sheaf-theoretic framework to this work stems from the following result: in any $R$-probabilistic theory, Bell tests give rise to no-signalling empirical models (with $R$-distributions) in the sheaf-theoretic framework, which can be used to prove contextuality/non-locality of the tests independently of the specific theory. As mentioned before, we will be interested in the CP* case, but the definitions below straightforwardly extend to arbitrary $R$-probabilistic theories.

\begin{definition}
A \textbf{Bell test} in an $R$-probabilistic CP* category is a process in the following form, for some normalised state $\rho$ and some normalised processes $B_1,...,B_N$:
\begin{equation}\label{eqn_BellTest}
	\input{pictures/chapter2/BellTest.tikz}
\end{equation}
We have denoted by $M_j$ the set of classical states for each classical input system $(\SpaceH_j,\hbox{\input{symbols/YbwdotSym.tex}}\!\!_j)$, and by $O_j$ the set of classical states for each classical output system $(\SpaceK_j,\hbox{\input{symbols/DdotSym.tex}}\!\!_j)$.
\end{definition}
\begin{theorem}[\textbf{Bell tests and sheaf-theoretic non-locality \cite{Gogioso2016f}}]\hfill\\ \label{thm_BellTestEM} Bell tests in $R$-probabilistic CP* categories give rise to $R$-valued no-signalling empirical models $(\zeta_{\,\underline{m}})_{\underline{m} \in \mathcal{M}}$ as follows, for any measurement cover $\mathcal{M}$:
\begin{equation}\label{eqn_BellTestEM}
	\input{pictures/chapter2/BellTestEM.tikz}
\end{equation}\end{theorem}

\begin{proof}
All we need to show is that the states in Equation \ref{eqn_BellTestEM} (indexed by the measurement contexts $\underline{m} \in \mathcal{M}$) satisfy no-signalling and are normalised. Marginalising over party $j$ yields the following state, which we want to prove independent of $m_j$:
\begin{equation}\label{diagram_BellTestEMproof1}
	\input{pictures/chapter2/BellTestEMproof1.tikz}
\end{equation}
The discarding map on the classical output systems $(\SpaceK,\hbox{\input{symbols/DdotSym.tex}}\!\!_j)$ can be written as $\trace{\,(\SpaceK_j,\hbox{\input{symbols/DdotSym.tex}}\!\!_j)} = \sum_{o_j} \bra{o_j}$ in terms of the classical states $(\ket{o_j})_{o_j \in O_j}$ of $\hbox{\input{symbols/DdotSym.tex}}\!\!_j$. Hence the marginalised state in Diagram \ref{diagram_BellTestEMproof1} can be rewritten as follows:
\begin{equation}\label{diagram_BellTestEMproof2}
	\input{pictures/chapter2/BellTestEMproof2.tikz}
\end{equation}
Because the $\hbox{\input{symbols/DdotSym.tex}}\!\!_j$ measurement, the process $B_j$ and the $\hbox{\input{symbols/YbwdotSym.tex}}\!\!_j$ preparation are all normalised, we conclude that the state of Diagram \ref{diagram_BellTestEMproof1} is independent of $m_j$:
\begin{equation}\label{eqn_BellTestEMproof3}
	\resizebox{\textwidth}{!}{\input{pictures/chapter2/BellTestEMproof3.tikz}}
\end{equation}
Marginalising over all outputs leaves us with $\trace{}\circ \rho$, which equals $1$ (independently of the measurement context $\underline{m}$) since $\rho$ is normalised. Hence the state of Diagram \ref{diagram_BellTestEMproof2} is also an $R$-distribution, completing our proof that Bell tests always give rise to no-signalling empirical models.
\end{proof}

\begin{theorem}[\textbf{Locality of $R$-probabilistic theories over fields}] \hfill\\
\label{thm_localityFields}
If $R$ is a field, then all $R$-probabilistic CP* categories are local. 
\end{theorem}
\begin{proof}
Theorem 5.4 from Ref. \cite{Abramsky2011} states that all no-signalling empirical models over the signed-probability field $\reals$ admit a local hidden variable model in terms of signed probabilities. Although the original result was proven for $\reals$, close inspection reveals that it holds for no-signalling empirical models over any field $R$: as a consequence, Bell-type measurement scenarios in $R$-probabilistic theories where $R$ is a field give rise to no-signalling empirical models admitting local hidden variable models. Finally, $R$-probabilistic theories have a sub-SMC of finite $R$-probabilistic classical systems, with all $R$-distributions as normalised states and all $R$-stochastic maps as normalised processes: as a consequence, all local hidden variable models valued in $R$ can be realised in any and all of them.
\end{proof}

\section{Some toy models of quantum theory}

In this Section, we present a number of toy models of quantum theory constructed using the framework for $R$-probabilistic CP* categories presented above. These examples are taken from the very recent \cite{Gogioso2017FQT}. 

\subsection{Theories of wavefunctions over semirings}

Note that two different linear structures intervene in the definition of quantum theory: the $\complexs$-linear structure of wavefunctions, modelling superposition, interference and phases, and the $\reals^+$-linear structure of probability distributions over classical systems. We have already seen that the framework of $R$-probabilistic theories replaces the probability semiring $\reals^+$ with a more general commutative semiring $R$ as a model of classical non-determinism. In this Section, we construct a large class of toy models of quantum theory by considering theories of wavefunctions with amplitudes valued in some commutative semiring with involution $S$, generalising the field with involution $\complexs$ traditionally used in quantum mechanics. To do so, we consider the dagger compact category $\RMatCategory{S}$ (dagger and compact closed structure will be defined using the involution), and we require classical non-determinism to arise via the Born rule, as embodied by the CP* construction. The corresponding quantum-classical theory will therefore be modelled by $\CPStarCategory{\RMatCategory{S}}$, and the main result of this Section (Theorem \ref{thm_RprobabilisticCPStarCategories}) will show that this is an $R$-probabilistic theory (where $R$ the sub-semiring of positive elements of $S$, see Definition \ref{def_semiringPositiveElements} below).

The category $\RMatCategory{S}$ for a commutative semiring with involution $S$ is defined as in the previous Section, but it comes with additional structure. Indeed, we can define dagger and compact closed structures on $\RMatCategory{S}$ exactly as done in $\fdHilbCategory$ (which is $\RMatCategory{\complexs}$), with conjugation taken using the involution $^\ast$ of $S$ in place of complex conjugation. Each object $S^X$ in $\RMatCategory{S}$ comes with at least one orthonormal basis $\ket{x}_{x \in X}$, as well as an associated special commutative $\dagger$-Frobenius algebra $\hbox{\input{symbols/ZbwdotSym.tex}}\!\!_X$:
\begin{equation}
\input{pictures/chapter2/classicalStructureFQT.tikz}
\end{equation}
For any group structure $G = (X,\cdot,1)$ on any finite set $X$, one also obtains an associated $\dagger$-Frobenius algebra $\hbox{\input{symbols/DdotSym.tex}}\!\!_G$ on $S^X$ by linearly extending the group multiplication and unit:
\begin{equation}
\input{pictures/chapter2/groupStructureFQT.tikz}
\end{equation}
The $\dagger$-Frobenius algebra is commutative if and only if the group is, and it always satisfies the following:
\begin{equation}
\input{pictures/chapter2/almostQuasiSpecialityFQT.tikz}
\end{equation}
Unfortunately, $\hbox{\input{symbols/DdotSym.tex}}\!\!_G$ is not quasi-special (a.k.a. normalisable) unless the scalar $|G|$ takes the form $z_G^\ast z_G$ for some $z_G \in S$ which is multiplicatively invertible: when this is the case, however, we have a legitimate strongly complementary pair $(\hbox{\input{symbols/ZbwdotSym.tex}}\!\!_X,\hbox{\input{symbols/DdotSym.tex}}\!\!_G)$ in $\RMatCategory{S}$ corresponding to the finite group $G$ (see next Chapter). When $G$ is abelian these strongly complementary pairs can be used (under additional conditions) to implement quantum protocols such as the algorithm to solve the abelian Hidden Subgroup Problem \cite{Vicary2012a,Gogioso2016d} or generalised Mermin-type arguments \cite{Gogioso2015,Gogioso2016e} (see Chapter \ref{chapter_algos}). This also means that certain objects in $\RMatCategory{S}$ support fragments of the ZX calculus\footnote{To be precise, they always support the spider rules (but cups/caps for the two algebras are generally distinct), the bialgebra rules, the Hopf laws (with non-trivial antipode), the copy rules, and a generalised version of the $\pi$-copy rules (see \cite{Backens2014}). A Hadamard unitary can be defined if and only if the $S$-valued unitary multiplicative characters for $G$ form an orthonormal basis for $S^X$, and in this case the colour-change rules are also supported (taking care, where relevant, to consider the adjoint of the Hadamard in place of the Hadamard itself).} \cite{Coecke2011,Backens2014}, opening the way to the application of well-established diagrammatic techniques to reason in these categories.

In quantum theory, the probabilistic semiring $\reals^+$ arises as a sub-semiring of $\complexs$ fixed by complex conjugation, namely the sub-semiring of those elements $z \in \complexs$ taking the form $z = x^\ast x$: this is, essentially, a hallmark of the Born rule. In general commutative semirings with involution, elements in the form $x^\ast x$ need not be closed under addition, but it is true their closure under addition always form a semiring.
\begin{definition}\label{def_semiringPositiveElements}
Let $S$ be a commutative semiring with involution. Then we define its \textbf{sub-semiring $R$ of positive elements} in $S$ to be the closure under addition in $S$ of the subset $\suchthat{x^\ast x}{ x \in S}$.
\end{definition}

\noindent When classical non-determinism is introduced via the Born rule, quantum theory naturally gives rise to a probabilistic theory. Similarly, it is possible to prove that any theory of wavefunctions valued in a commutative semirings with involution $S$ gives rise to an $R$-probabilistic theory, where $R$ is the corresponding sub-semiring of positive elements. 
\begin{theorem}\label{thm_RprobabilisticCPStarCategories}
Let $S$ be a commutative semiring with involution, and let $R$ be its sub-semiring of positive elements. Then $\CPStarCategory{\RMatCategory{S}}$ is $R$-probabilistic under the $\CMonCategory$-enrichment inherited from $\RMatCategory{S}$.
\end{theorem}
\begin{proof}

In order for $\CPStarCategory{\RMatCategory{S}}$ to be $R$-probabilistic under the $\CMonCategory$-enrichment of $\RMatCategory{S}$, we need to show that it satisfies the following three conditions:
\begin{enumerate}
	\item[(i)] there is a full sub-SMC $\classicalSubcategory{\CPStarCategory{\RMatCategory{S}}}$ which is equivalent to $\RMatCategory{R}$;
	\item[(ii)] the $\CMonCategory$-enrichment of $\RMatCategory{S}$ must restrict to a well-defined enrichment for $\CPStarCategory{\RMatCategory{S}}$, coinciding on $\classicalSubcategory{\CPStarCategory{\RMatCategory{S}}}$ with the enrichment of $\RMatCategory{R}$;
	\item[(iii)] the SMC $\CPStarCategory{\RMatCategory{S}}$ comes with an environment structure which restricts to the the canonical one from $\RMatCategory{R}$ on the full subcategory $\classicalSubcategory{\CPStarCategory{\RMatCategory{S}}}$.
\end{enumerate}

\noindent Firstly, we show that the $\CMonCategory$-enrichment of $\RMatCategory{S}$ restricts to a well-defined $\CMonCategory$-enrichment for $\CPStarCategory{\RMatCategory{S}}$. Because $\RMatCategory{S}$ is a category of matrices, this is in turn true if and only if the scalars of $\CPStarCategory{\RMatCategory{S}}$ are closed under addition in $\RMatCategory{S}$, i.e. if and only if they form a sub-semiring of $S$ (they are always necessarily closed under multiplication). To see that this is true, it suffices to show that the scalars of $\CPStarCategory{\RMatCategory{S}}$ form exactly the sub-semiring $R$ of positive elements of $S$ (we have to show it anyway, if we want our theory to be $R$-probabilistic!). Indeed, the generic scalar of $\CPStarCategory{\RMatCategory{S}}$ takes the form $\trace{\,S^D} \circ \CPMdoubled{\ket{\psi}} = \sum_{d=1}^D p_d^\ast p_d$ for a generic state $\ket{\psi} := \sum_{d=1}^D  p_d \ket{d}$ of $\RMatCategory{S}$.

For condition (i), consider the full-subcategory $\classicalSubcategory{\CPStarCategory{\RMatCategory{S}}}$ of $\CPStarCategory{\RMatCategory{S}}$ spanned by those objects in the form $(S^X,\decoh{\hbox{\input{symbols/ZbwdotSym.tex}}\!\!_X})$, where $X$ is a finite set, $\hbox{\input{symbols/ZbwdotSym.tex}}\!\!_X$ is the special commutative $\dagger$-Frobenius algebra on $S^X$ associated with the orthonormal basis $\ket{x}_{x \in X}$, and $\decoh{\hbox{\input{symbols/ZbwdotSym.tex}}\!\!_X}: S^X \rightarrow S^X$ is the decoherence map for $\hbox{\input{symbols/ZbwdotSym.tex}}\!\!_X$, which is a self-adjoint idempotent normalised CP map. Morphisms $(S^X,\decoh{\hbox{\input{symbols/ZbwdotSym.tex}}\!\!_X}) \rightarrow (S^Y,\decoh{\hbox{\input{symbols/ZbwdotSym.tex}}\!\!_Y})$ are exactly those in the following form, where $(f_{xy})_{x \in X, y \in Y}$ is an arbitrary matrix of scalars (i.e. elements of $R$):
\begin{equation}
\sum_{y \in Y}\sum_{x \in X} \CPMdoubled{\ket{y}}\; f_{xy}\; \CPMdoubled{\bra{x}}
\end{equation}
As a consequence, $\classicalSubcategory{\CPStarCategory{\RMatCategory{S}}}$ is equivalent to $\RMatCategory{R}$. Condition (ii) is satisfied as well. For condition (iii), it suffices to consider the canonical environment structure given by the CP* construction. Because decoherence maps are normalised, this environment structure restricts to the canonical one on $\classicalSubcategory{\CPStarCategory{\RMatCategory{S}}}$.
\end{proof}

Note that the scalars of $\CPStarCategory{\RMatCategory{S}}$ are the elements of $R$, and that the pure scalars are those in the form $\xi^\ast \xi$ for some $\xi \in S$: as a consequence, not all scalars of $\CPStarCategory{\RMatCategory{S}}$ need be pure (in contrast to what happens with ordinary quantum theory). In what follows, we will try as much as possible to construct theories where all scalars are pure, but there are examples (such as the case of $p$-adic quantum theory) where this cannot be achieved. When all scalars are pure, the requirement that $|G| = z_G^\ast z_G$ is always automatically satisfied for all finite groups $G$, and we only need to care about $|G|$ being invertible as a scalar in $S$ (a fact which always holds true whenever $S$ is a semi-field/field and $|G|$ is non-zero in $S$). We will now proceed to construct a number of toy models within this framework.

\subsection{Real quantum theory}

The simplest non-conventional example is given by the ring $\reals$ of signed reals (with the trivial involution), which yields the \textbf{probability semiring} $\reals^+$ as its sub-semiring of positive elements; in particular, all positive elements are pure scalars. The corresponding probabilistic theory $\CPStarCategory{\RMatCategory{\reals}}$ is known as \textbf{real quantum theory} \cite{Jordan1934,Baez2012,Belenchia2012,Wilce2016}: it is arguably the most well-studied of the quantum-like theories, and the closest to ordinary quantum theory. Thus said, real quantum theory can be distinguished from ordinary quantum theory because it fails to be \textit{locally tomographic} \cite{Araki1980,Wootters1990,Chiribella2010}, i.e. bipartite (mixed) states in real quantum theory cannot in general be distinguished by product measurements alone. Equivalently, one can check that the CP maps $\CPMdoubled{\sigma_x} + \CPMdoubled{\sigma_z} - \CPMdoubled{\id{\reals^2}}$ and $\CPMdoubled{\sigma_y}$ on $\reals^2$ in $\CPMCategory{\RMatCategory{\reals}}$ cannot be distinguished by applications to mixed states of $\reals^2$ alone, because the latter are described by density matrices which are always real symmetric\footnote{By $\sigma_x$, $\sigma_y$ and $\sigma_z$ we have denoted the complex qubit Pauli matrices, which give rise to real CP maps on $\reals^2$ when doubled.}.

The group of phases in $\reals$ is $\{\pm 1\} \isom \integersMod{2}$, and non-trivial interference is possible in real quantum theory. For example, each of the Pauli $X$ eigenstates $\ket{\pm} := \frac{1}{\sqrt{2}}(\ket{0}\pm\ket{1})$ of the qubit $\reals[\integersMod{2}]$ in real quantum theory yields the uniform distribution when measured in the Pauli $Z$ basis $\ket{0},\ket{1}$, but their superposition $\frac{1}{\sqrt{2}}(\ket{+}+\ket{-})$ yields the outcome corresponding to $\ket{0}$ with certainty.

Finally, Bell's theorem goes through in real quantum theory (as it only involves states and measurements on the $ZX$ greater circle of the Bloch sphere), and the latter is a non-local probabilistic theory (because the states and processes of real quantum theory are a subset of those of quantum theory).

\subsection{Relational quantum theory}

Examples of an entirely different nature are given by considering distributive lattices $\Omega$ (with the trivial involution), which yield themselves back as their sub-semirings of positive elements (because of multiplicative idempotence); in particular, all positive elements are pure scalars. Distributive lattices seem to be almost as far as one can go from the probabilistic semiring $\reals^+$, but the category $\RMatCategory{\Omega}$ has been studied extensively as a toy model of quantum theory (especially in the boolean case $\Omega = \mathbb{B}$) \cite{Pavlovic2009,Abramsky2012,Evans2009,Zeng2015,Coecke2012a}, and the corresponding CPM category has also received some attention on its own \cite{Marsden,Gogioso2015b}. We refer to the corresponding $\Omega$-probabilistic (or \textbf{possibilistic}) theory as \textbf{relational quantum theory}.

The group of phases in $\Omega$ is the singleton $\{1\}$, and no interference is possible in relational quantum theory. Relational quantum theories also feature very few quantum-to-classical transitions: there is a unique basis on each system, namely the one given by the elements of the underlying set. They are local tomographic on pure states, but they fail to be tomographic altogether on mixed states: for example, the pure state $\ket{\psi}\bra{\psi}$ for $\ket{\psi} := \ket{0}+\ket{1}$ and the mixed state $\ket{0}\bra{0} + \ket{1}\bra{1}$ are distinct, but cannot be distinguished by measurement. In fact, a characteristic trait of relational quantum theories is exactly that superposition and mixing are essentially indistinguishable (because of idempotence) \cite{Abramsky2012,Marsden,Gogioso2015b}, and this can be used to show that relational quantum theories are entirely local \cite{Abramsky2012,Gogioso2015b}.

\subsection{Hyperbolic quantum theory}

Turning our attention back to real algebras, we can consider the commutative ring of \textbf{split complex numbers} $\splitComplexs := \reals[X]/(X^2-1)$, a two-dimensional real algebra. Split complex numbers take the form $(x+j y)$, where $x,y \in \reals$ and $j^2 = 1$; in particular, they have non-trivial zero-divisors in the form $a(1\pm j)$, because $(1+j)(1-j)=1-j^2 = 0$. They come with the involution $(x + j y)^\ast := x - jy$, which yields the \textbf{signed-probability ring} $\reals$ as sub-semiring of positive elements; in particular, all positive elements are pure scalars. We refer to the corresponding $\reals$-probabilistic (or \textbf{quasi-probabilistic}) theory $\CPStarCategory{\RMatCategory{\splitComplexs}}$ as \textbf{hyperbolic quantum theory}\footnote{Clifford referred to functions of split complex numbers as \inlineQuote{functions of a motor variable} \cite{Clifford1871}, so we could say that hyperbolic quantum theory is the theory of \textbf{wavefunctions of a motor variable} (how does \textbf{motor quantum theory} sound?).} \cite{Khrennikov2003,Khrennikov2010,Nyman2011}. Because scalars form a field, Theorem \ref{thm_localityFields} (and the original Theorem 5.4 from \cite{Abramsky2011}) implies that hyperbolic quantum theory is a local theory. 

The group of phases in $\splitComplexs$ consists of the elements with square norm $1$, i.e. the elements in the form $x+jy$ which lie on the following unit hyperbola of the real plane:
\begin{equation}
1=(x+jy)^\ast(x+jy)=x^2-y^2
\end{equation} 
In fact, the natural geometry for the split complex numbers is that of the real plane endowed with the Lorentzian metric $-dy^2+dx^2$, i.e. that of the Minkowski plane. Just like multiplication by phases in $\complexs$ forms the circle group $U(1)$ of rotations around the origin for the Euclidean plane, multiplication by phases in $\splitComplexs$ forms the group $SO(1,1)$ of orthochronous homogeneous Lorentz transformations for the Minkowski plane, and we have the isomorphism of Lie groups $\integersMod{2} \times \reals \isom SO(1,1)$ given by $(s,\theta) \mapsto s (\cosh(\theta)+j\sinh(\theta))$.

hyperbolic quantum theory is a local theory, in the sense that every empirical model arising in hyperbolic quantum theory admits a local hidden variable model in terms of signed probabilities (the notion of classical non-determinism for hyperbolic quantum theory) \cite{Abramsky2011}. While signed probabilities might at first sound unphysical, an operational interpretation exists in terms of unsigned probabilities on signed events \cite{Abramsky2014}\footnote{Where the sign of the events themselves cannot be observed, yielding an epistemic restriction which is not too far removed from the one which originally motivated Spekkens's toy model \cite{Spekkens2007,Catani2017}}.

\subsection{Parity quantum theory}

A simple variation on relational quantum theory (over the booleans) is given by using symmetric difference of sets, instead of union, as the superposition operation. This leads us to consider the finite field with two elements $\integersMod{2}:=(\{0,1\},+,0,\times,1)$, with trivial involution, in place of the booleans $\mathbb{B}:=(\{0,1\},\vee,0,\times,1)$, also with trivial involution. The multiplication is the same, but addition is now non-idempotent, and superposition is no longer the same as mixing. The \textbf{parity semiring} $\integersMod{2}$ yields itself back as its sub-semiring of positive elements (in particular, all positive elements are pure scalars), and we refer to the corresponding $\integersMod{2}$-probabilistic theory $\CPStarCategory{\RMatCategory{\integersMod{2}}}$ as \textbf{parity quantum theory}. 

\begin{remark}
Because the involution is trivial, parity quantum theory as defined here pretty much coincides with the $\integersMod{2}$ case of modal quantum theory \cite{Schumacher2012,Schumacher2016}, but it should be noted that the philosophical interpretation of $\integersMod{2}$-valued probabilities is significantly different. In modal quantum theory, the interest is in generating possibilistic tables by using finite fields, subsequently interpreting all zero values as the boolean $0$ and all non-zero values as the boolean $1$. In parity quantum theory, the non-determinism itself is interpreted to be natively $\integersMod{2}$-valued, and no attempt is made to translate the resulting empirical models into possibilistic ones. Indeed, such an interpretation would not be natural within our semiring-oriented framework, as no semiring homomorphism can exists from any finite field into the booleans. 
\end{remark}

The group of phases in $\integersMod{2}$ is the singleton $\{1\}$, but interference is still possible in parity quantum theory: this somewhat counter-intuitive situation is made possible by the fact that $1$ is its own additive inverse in $\integersMod{2}$, so that triviality of the group of phases is slightly deceptive. Indeed, consider the four two-qubit states below, which form an orthonormal basis for $\integersMod{2}^2$:
\begin{align}
\ket{\psi_{012}} := \ket{00}+\ket{01}+\ket{10} \hspace{2cm} \ket{\psi_{123}} := \ket{01}+\ket{10}+\ket{11} \nonumber \\
\ket{\psi_{230}} := \ket{10}+\ket{11}+\ket{00} \hspace{2cm} \ket{\psi_{301}} := \ket{11}+\ket{00}+\ket{01}
\end{align}
For example, we have $\ket{10} = \ket{\psi_{012}}+\ket{\psi_{123}}+\ket{\psi_{230}}$. When measured in the computational basis $\ket{00}, \ket{01}, \ket{10}, \ket{11}$, the normalised states $\ket{01}$, $\ket{10}$ and $\ket{\psi_{012}}$ all have non-zero $\integersMod{2}$-probability of yielding an outcome in the set $\{01,10\}$, but their superposition $\ket{01}+\ket{10}+\ket{\psi_{012}} = \ket{00}$ (also a normalised state) has zero $\integersMod{2}$-probability of yielding an outcome in that set. 

$R$-probabilistic theories can be similarly constructed for modal quantum theory  over any other finite field $\finiteField{p^n}$ \cite{Schumacher2012,Schumacher2016}, by taking $S := \finiteField{p^n}$ with the trivial involution. However, these theories have a lot of non-pure scalars---namely the $(p^n-1)/2$ non-squares in $\finiteField{p^n}$---and their phases are close to trivial---namely they are $\{\pm 1\}$ if $p>2$ and $\{1\}$ if $p=2$. Instead, we will consider a more sophisticated construction based on quadratic extensions of finite fields, which we call \inlineQuote{finite-field quantum theory}. 
 
What will make finite-field quantum theory extremely attractive for CQM is the fact that it is a local theory (by Theorem \ref{thm_localityFields}), in which it is nonetheless possible to formulate non-trivial quantum algorithms (such as the one solving the abelian Hidden Subgroup Problem), as well as non-trivial \inlineQuote{non-locality} arguments (such as generalised Mermin-type arguments). This is in stark contrast with the more traditional toy models based on relational quantum theory, in which the quantum Fourier transform cannot be performed for non-trivial groups \cite{Gogioso2015d},precluding the implementation of algorithms based on it, and in which all Mermin-type arguments are necessarily trivial \cite{Coecke2012c,Gogioso2015} (see Chapter \ref{chapter_algos}).

\subsection{Finite-field quantum theory}

Consider a finite field $\finiteField{p^n}$ (with $p$ odd), and let $\epsilon$ be a generator for the cyclic group $\finiteField{p^n}^\times$ of invertible (aka non-zero) elements in $\finiteField{p^n}$ (i.e. a primitive element for $\finiteField{p^n}$). We consider the ring $\finiteField{p^n}[\sqrt{\epsilon}] := \finiteField{p^n}[X^2-\epsilon]$, equipped with the involution $(x+y\sqrt{\epsilon})^\ast := (x-y\sqrt{\epsilon})$: because $\epsilon$ is a primitive element, $\finiteField{p^n}(\sqrt{\epsilon}) \isom \finiteField{p^{2n}}$ is a field. We are in fact working with the quadratic extension of fields $\finiteField{p^n}(\sqrt{\epsilon}) / \finiteField{p^n}$, equipped with the usual involution and (squared) norm:
\begin{equation}
\big|x+y\sqrt{\epsilon}\big|^2 = (x-y\sqrt{\epsilon})(x+y\sqrt{\epsilon}) = x^2 - \epsilon y^2
\end{equation}
The sub-field $\finiteField{p^n}$ (given by the elements in the form $x+0\sqrt{\epsilon}$) is the sub-semiring of positive elements (and we will shortly see that all positive elements are pure scalars). 

The phases in $\finiteField{p^{n}}(\sqrt{\epsilon})$ are the points $(x,y)$ of the $\finiteField{p^n}^2$ plane lying on the unit hyperbola $x^2 - \epsilon y^2 = 1$, which does not factor as a product of two lines because $\epsilon$ is a primitive element. The following iconic result of Galois theory due to Hilbert can be used to characterise them (see e.g. \cite{Hilbert1998} for a proof). 
\begin{theorem}[\textbf{Hilbert's Theorem 90}]\hfill\\
Let $L/K$ be a finite cyclic field extension, and let $\sigma: L \rightarrow L$ be a generator for its cyclic Galois group. Then the multiplicative group of elements $\xi \in L$ of unit relative norm $N_{L/K}(\xi)=1$ is isomorphic to the quotient group $L^\times / K^\times$.
\end{theorem}
\begin{corollary}\label{cor_phasesFFQT}
The phases in $\finiteField{p^{n}}(\sqrt{\epsilon})$ form the cyclic group $\finiteField{p^{2n}}^\times/\finiteField{p^n}^\times \isom \integersMod{p^n+1}$.
\end{corollary}
\begin{proof}
We have a quadratic extension $\finiteField{p^{n}}(\sqrt{\epsilon})/\finiteField{p^n}$, with 2-element Galois group generated by the involution $\sigma~:=~\xi~\mapsto~\xi^\ast$, and corresponding field norm $N_{\finiteField{p^{n}}(\sqrt{\epsilon}) / \finiteField{p^n}}(\xi) := \xi^\ast \xi$. By Hilbert's~Theorem~90, the multiplicative group of those $\xi \in \finiteField{p^{2n}}$ such that $\xi^\ast \xi = 1$ is isomorphic to the quotient group $\finiteField{p^{n}}(\sqrt{\epsilon})^\times/\finiteField{p^n}^\times$. But $\finiteField{p^{n}}(\sqrt{\epsilon})^\times \isom \finiteField{p^{2n}}^\times$ is cyclic with $p^{2n}-1$ elements, and $\finiteField{p^n}^\times$ has $p^n-1$ elements: hence the quotient is cyclic with $(p^{2n}-1)/(p^n-1) = p^n+1$ elements, i.e. it is $\integersMod{p^n+1}$. 
\end{proof}
\noindent Another interesting consequence of Hilbert's Theorem 90 is the fact that the positive elements in finite-field quantum theory are all pure scalars.
\begin{lemma}\label{lem_purescalarsFFQT}
All scalars in $\CPStarCategory{\RMatCategory{\finiteField{p^{n}}(\sqrt{\epsilon})}}$ are pure.
\end{lemma}
\begin{proof}
Because $\finiteField{p^{n}}(\sqrt{\epsilon})$ is a field, we have that $a^\ast a = b^\ast b$ if and only if $a = \xi b$ for some $\xi$ such that $\xi^\ast \xi = 1$, i.e. for some phase $\xi$. Equality up to phase is an equivalence relation on elements of $\finiteField{p^{n}}(\sqrt{\epsilon})$ (because phases form a group under multiplication), and there are exactly $p^n+1$ phases by Corollary \ref{cor_phasesFFQT}: as a consequence, there are exactly $(p^{2n}-1)/(p^n+1)= p^n-1$ non-zero pure scalars in $\CPStarCategory{\RMatCategory{\finiteField{p^{n}}(\sqrt{\epsilon})}}$, i.e. all the scalars are in fact pure (since the zero scalar always is).
\end{proof}

While finite-field quantum theory and parity quantum theory might not have as direct a physical interpretation as hyperbolic quantum theory and relational quantum theory, they have the major advantage of dealing with wavefunction valued over a field, so that objects are finite-dimensional vector spaces (equipped with a non-standard inner product, in the case of finite-field quantum theory). This opens the door for a systematic study of quantum systems in these theories using standard tools from finite geometry. Further investigation in this direction is left to future work.

\newcommand{\ord}[1]{{\operatorname{ord} #1}}
\newcommand{\sgn}[2]{{\operatorname{sgn}_{#1} #2}}
\subsection{p-adic quantum theory}

We now look at the construction of \textbf{$p$-adic quantum mechanics} \cite{Vladimirov1989,Ruelle1989,Khrennikov1991,Khrennikov1993,Palmer2016,Palmer2016a}, where $R:=Q_p$ is the field of $p$-adic numbers, and $S$ is some quadratic extension. In this Section, we will use the notation $Q_p$ to denote the $p$-adic numbers, and $Z_p$ to denote the $p$-adic integers, to distinguish them from the finite field $\integersMod{p}$ of integers modulo $p$; note that this convention is different from the one used in many texts on $p$-adic arithmetic, where $\integersMod{p}$ is used for the $p$-adic integers (and $\rationals_p$ for the $p$-adic numbers). 

When $p > 2$, the $p$-adic numbers $x:= p^\ord{x}\sum_{i=0}^{+\infty} x_i p^i$ fall within four distinct quadratic classes---jointly labelled by the parity of the order $\ord{x} \in \integers$ and by the quadratic class of the first non-zero digit $x_0 \in \integersMod{p}^\times$---and there are three corresponding inequivalent quadratic extensions. This means that there is no way to obtain all positive elements as pure scalars by a single quadratic extension. This would seem to indicate that mixed states play a necessary role in the emergence of $p$-adic probabilities, which cannot all be obtained from pure states alone: the potential physical significance of this observation might become the topic of future work.

We consider the quadratic extension $S:=Q_p(\sqrt{\epsilon})$, where $p \geq 3$ and $\epsilon$ is a primitive element in the field $\integersMod{p}$ of integers modulo $p$, and we follow the presentation of \cite{Ruelle1989}. A generic element of $Q_p(\sqrt{\epsilon})$ takes the form $c + s \sqrt{\epsilon}$, for $c,s \in Q_p$, and its square norm is $|c+s\sqrt{\epsilon}|^2 = (c-s\sqrt{\epsilon})(c+s\sqrt{\epsilon}) = c^2 - \epsilon s^2$. Whether an element $x \in Q_p$ can be written in this form, i.e. whether it is a pure scalar in $\CPStarCategory{\RMatCategory{Q_p(\sqrt{\epsilon})}}$, is determined by the \textit{sign function} $\sgn{\epsilon}{x}$, which takes the value $+1$ if $x = c^2 - \epsilon s^2$ for some $c,s \in Q_p$, and the value $-1$ otherwise. An explicit form for the sign function (in the $p \neq 2$ case) is given by Equation (2.34) of \cite{Ruelle1989}, which specialised to our case ($\tau = \epsilon$ and $\ord{\tau} = 0$) reads $\sgn{\epsilon}{x} = (-1)^{\ord{x}}$. Hence the pure scalars in $\CPStarCategory{\RMatCategory{Q_p(\sqrt{\epsilon})}}$ are exactly the $p$-adic numbers $x$ with even order $\ord{x}$; closure of this set under addition yields $R := Q_p$ as sub-semiring (field, in fact) of positive elements in $S:= Q_p(\sqrt{\epsilon})$.

The phases in $p$-adic quantum theory are those $\xi := (c + s \sqrt{\epsilon}) \in Q_p(\sqrt{\epsilon})$ such that $\xi^\ast \xi = c^2 - \epsilon s^2 = 1$. In \cite{Ruelle1989} (Equation (4.35) of Section IV.C, and Equation (C12b) of Appendix C.3) it is shown that phases form a multiplicative group $C_\epsilon$ isomorphic to the additive group $\integersMod{p+1} \times p Z_p$---here $(\integersMod{p+1},+,0)$ are the integers modulo $p+1$, while $(p Z_p, +, 0)$ is the additive subgroup of $Z_p$ formed by those $p$-adic integers which are divisible by $p$. In particular, $C_\epsilon$ is an infinite group with the cardinality of the continuum, and each \inlineQuote{sheet} $pZ_p$ is a profinite\footnote{And hence both compact and totally disconnected.} torsion-free group, which is best understood by looking at the descending normal series $p Z_p \triangleright p^2 Z_p \triangleright ... \triangleright p^m Z_p \triangleright ...$ and considering the finite cyclic quotients  $p^n Z_p / p^m Z_p \isom \integersMod{p^{m-n}}$.

\begin{remark}
Similar considerations apply to the the construction of $p$-adic quantum theory for the other two quadratic extensions $Q_p(\sqrt{p})$ and $Q_p(\sqrt{p\epsilon})$ available in the case of $p \geq 3$ (although the cases $p=3$ and $p \geq 5$ have to be treated separately), as well as the seven quadratic extensions available in the case of $p=2$. The phase groups take a similar (but not identical) form to the one presented here, and the full details can be found in  \cite{Ruelle1989} (Section IV.C and Appendices C.3, C.4).
\end{remark}

\subsection{Tropical quantum theory}

Relational quantum theory involves semirings which are both additively and multiplicatively idempotent, parity quantum theory involves a semiring which is only multiplicatively idempotent, and ordinary quantum theory involves a semiring which is neither additively nor multiplicatively idempotent. We now give examples of theories with wavefunctions based in semirings which are additively idempotent but not multiplicatively idempotent, namely the tropical semirings  \cite{simon1988tropical,Maslov1987,Simon1994,pin1998tropical,mikhalkin2004amoebas,speyer2009tropical}. 
\begin{definition}
A \textbf{tropical semiring} is the commutative semiring $(M,\min,\infty,+,0)$ obtained from a totally ordered commutative monoid $(M,+,0,\leq)$ having an absorbing element $\infty$ which is larger than all elements in the monoid. In the tropical semiring, $\min$ is the addition, $\infty$ is the additive unit, $+$ is the multiplication and $0$ is the multiplicative unit. The nomenclature is extended to semirings isomorphic to the explicitly min-plus semirings used above (e.g. max-plus formulations, or the Viterbi semiring).
\end{definition}
\noindent Examples of tropical semirings appearing in the literature  include the tropical reals $(\reals\sqcup\{\infty\},\min,\infty,+,0)$, the tropical integers $(\integers\sqcup\{\infty\},\min,\infty,+,0)$, the tropical naturals $(\naturals\sqcup\{\infty\},\min,\infty,+,0)$, and the Viterbi semiring $([0,1],\max,0,\cdot,1)$ (which is a tropical semiring because it is isomorphic to the explicitly min-plus semiring $(\reals^+\sqcup\{\infty\},\min,\infty,+,0)$ via the semiring homomorphism $x \mapsto -\log x$). In fact, there is an easy characterisation of which commutative semirings arise as tropical semirings (the proof is omitted as it is a straightforward check).
\begin{lemma}
A commutative semiring $(S,+,0,\cdot,1)$ is a tropical semiring if and only if for all $a,b \in S$ we have $a = a+b$ or $b = a+b$ (in which case we can set $\min(a,b) = a+b$).
\end{lemma}
\noindent From now on, we will revert back to usual semiring notation, and we will rely on the result above to connect with the min-plus notation typical of tropical geometry \cite{speyer2009tropical}. We will, however, remember that tropical semirings come with a total order respected by the multiplication, and we will occasionally use $\min$, $\max$ and $\leq$.

\begin{lemma}\label{lem_tropicalSemiringPositiveEls}
The only involution possible on a tropical semiring $(S,+,0,\cdot,1)$ is the trivial one, with sub-semiring of positive elements $(\suchthat{x^2}{x\in S},+,0,\cdot,1)$.
\end{lemma}
\begin{proof}
Let $^\ast$ be an involution for the tropical semiring $S$: $x \leq y$ implies that $x = x+y$, so that $x^\ast = x^\ast+y^\ast$ and $x^\ast \leq y^\ast$. But then $x \leq x^\ast$ implies $x^\ast \leq (x^\ast)^\ast = x$ (and similarly for $x^\ast \leq x$), so that $x^\ast = x$ is the trivial involution. Now consider the tropical semiring with trivial involution, so that the positive elements are exactly those in the form $x^2$ for some $x \in S$. But in a tropical semiring we have that $x^2+y^2 = (x+y)^2$ (as Speyer and Sturmfels put it, \inlineQuote{the Freshman's dream holds in tropical arithmetic.} \cite{speyer2009tropical}): hence the squares are closed under addition $+$, and form a sub-semiring.
\end{proof}
\noindent If $S$ is a tropical semiring and $R:=(\suchthat{x^2}{x \in S},+,0,\cdot 1)$ is its sub-semiring of positive elements, we refer to the $R$-probabilistic theory $\CPStarCategory{\RMatCategory{S}}$ as \textbf{tropical quantum theory}. 

Just as in the case of relational quantum theory, the group of phases in a tropical semiring $S$ is always trivial (because $x^2 = 1$ implies $x=1$ in any totally ordered monoid $(S,\cdot,1,\leq)$), and there is no interference. Similarly, there is a unique orthonormal basis on each system, the only unitaries/invertible maps are permutations, and superposition cannot be distinguished from mixing by measurements alone. Tropical quantum theory does not admit any implementation of the algorithm for the abelian Hidden Subgroup Problem, nor does it admit any generalised Mermin-type non-locality arguments.

The parallelisms with relational quantum theory become less surprising when one realises that tropical quantum theory is another generalisation of quantum theory over the booleans, which form a totally ordered distributive lattice, and hence are a particular case of tropical semiring. (Proof of the following result is omitted, as it is a straightforward check.)
\begin{lemma} 
Any totally ordered distributive lattice $(\Omega,\vee,\bot,\wedge,\top)$ is a tropical semiring $(\Omega,\wedge,\top,\vee,\bot)$; conversely, every tropical semiring $(S,+,0,\cdot,1)$ which has $1$ as least element and such that $x^2 = x$ for all $x \in S$ is a totally ordered distributive lattice $(S, \cdot,1,+,0)$.
\end{lemma}
\noindent In the light of the result above, we expect tropical quantum theory to be local, exactly like relational quantum theory, but further investigation of this question is left to future work.







\chapter{Categorical Quantum Dynamics}
\label{chapter_CQD}

\section{Introduction}
\label{section_CQDintro}

\subsection{A coherent approach to quantum symmetries}

We have seen in the previous Chapter that the importance of Frobenius algebras to quantum foundations and quantum algorithms stems from their connection to the coherent manipulation of classical data: in quantum foundations, coherent operations precede classical operations, as the latter can be obtained from the former by decoherence; in quantum algorithms, coherence is one of the resources involved in providing quantum advantage (e.g. see the next Chapter). 

We claim that quantum symmetries and dynamics should similarly be understood by studying the coherent versions of primitive notions from classical symmetries and dynamics. We will see how strong complementarity, an algebraic notion born to capture the special relation between a vector basis and its corresponding Fourier basis, embodies the coherent counterpart of finite-dimensional group theory. Thanks to our coherent approach, we are able to obtain clear and intuitive proofs for a number of known results, as well as a wealth of new insights. The simple and accessible case of periodic lattice symmetry (i.e. finite abelian group symmetry) will be used in the first half of the chapter to showcase some of the salient features of our approach, but the results proven will always be as general as possible. 

From the point of view of a mathematician, a symmetry of a system is simply the action of a group on it: a set of reversible transformations of the system into itself, closed under composition and inversion. From the point of view of a physicist, however, symmetries often have different standings depending on their specific origin: there are intrinsic symmetries of a system, such as the $U(n)$ symmetry of an $n$-dimensional quantum system, and there are symmetries induced by the presence of some underlying structure, such as the symmetry induced by space-time on quantum fields.  We will take the mathematician's point of view, and define a symmetry to be a group action on a system (representations, when linear structure is present). However, we will pay respect to the physicist's point of view by investigating the physical significance of the mathematical constructs we introduce.

In the context of this work, quantum dynamics will be treated as the special case of quantum symmetries generated by a group which is 1-dimensional in some suitable sense, an approach similar, in spirit, to the one behind Noether's theorem. Depending on the context, by quantum dynamics we might mean discrete periodic dynamics (corresponding to a $\integersMod{n}$ symmetry), continuous periodic dynamics (corresponding to a $S^1$ symmetry), discrete dynamics (corresponding to a $\integers$ symmetry) or continuous dynamics (corresponding to a $\reals$ symmetry). The results we will prove hold for any notion of dynamics which can be modelled by a doubly well-pointed coherent group, and in particular they will be immediately applicable to discrete periodic, discrete and continuous periodic dynamics of finite-dimensional and separable quantum systems. Unfortunately, the continuous case of $\reals$ is not going to be treated explicitly in this work: the necessary techniques were only developed recently, and there has been no time to accommodate them in this Thesis. However, we are certain that the results derived here will straightforwardly extend to the continuous case, with minimum modifications necessary.  

In fact, a significant number of concrete examples in our treatment of dynamics will focus on discrete periodic dynamics of finite-dimensional quantum systems: this is a simple, accessible family of examples, which nonetheless offers the full spectrum of features traditionally associated with dynamics (in fact, comparing and contrasting the discrete periodic case with the traditional continuous case yields some interesting new insights on the latter). Moreover, the discrete periodic case relates well to the problem of time observables, an interesting open question in the philosophy of quantum theory. Thus said, the same examples can be readily generalised to the discrete, continuous periodic and continuous cases.

\subsection{Synopsis of this Chapter}

\subsubsection{Coherent Groups}

In Section \ref{section_quantumGroups} we investigate the basic structures and properties recurring throughout our coherent treatment of group symmetries. We begin our investigation with the concrete case of wavefunctions on periodic lattices, where the momentum observable arises naturally from a coherent treatment of the lattice translation symmetry. We identify strong complementarity as the relevant algebraic property relating the position and momentum observables, and we define a notion of coherent groups to capture the basic abstract structures intervening in the coherent approach to group theory.

Just like group algebras can be used to \inlineQuote{embed} groups into categories of vector spaces, coherent groups can be used to embed groups into arbitrary $\dagger$-SMCs, and this generalisation is related to non-commutative geometry and algebraic quantum theory. In the finite-dimensional quantum case of $\fdHilbCategory$, coherent groups generalise group algebras by using an arbitrary quantum observable (a symmetric $\dagger$-qSFA) as the point structure, instead of the non-degenerate quantum observable (a $\dagger$-SCFA) used in the case of a group algebra. This corresponds to a possibly non-commutative C*-algebra, as opposed to the commutative C*-algebra associated with the standard basis, and coherent groups in $\fdHilbCategory$ reduce to a special case of compact quantum groups.

There are three main advantages to working with an abstract, diagrammatic, theory-independent formulation of group algebras such as coherent groups. Firstly, the abstract character of our definitions allows us to focus on the core structural and operational features of group algebras which play a role in quantum foundations, quantum algorithms and non-locality arguments, without getting distracted or waylaid by the rich structure of Hilbert space quantum mechanics. Secondly, the diagrammatic formulation makes important physical concepts such as position/momentum duality, quantum symmetries, Hamiltonians and dynamics available within the framework of CQM, and in turn allows methods from CQM to be applied to a much wider variety of physically interesting problems. Finally, the theory-independent approach means that our results (in both foundations and protocols) are not limited to conventional quantum mechanics, but are instead immediately applicable to a vast landscape of quantum-like theories (comprising toy models, variations, and extensions of quantum mechanics).

That said, the joint aim of the results in this Section is to show that coherent groups provide a suitable generalisation of group algebras (and, more generally, certain finite-dimensional compact quantum groups) to arbitrary $\dagger$-symmetric monoidal categories. 
\begin{enumerate}
	\item[(i)] \textbf{Theorem \ref{thm_coherentGroupsCompactQuantumGroups} (p.\pageref{thm_coherentGroupsCompactQuantumGroups})} shows that coherent groups on finite-dimensional Hilbert spaces are a special, well-behaved case of compact quantum groups.
	\item[(ii)] \textbf{Theorems \ref{thm_QuantumGroupAreGroupsOnPoints} (p.\pageref{thm_QuantumGroupAreGroupsOnPoints}), \ref{thm_QuantumGroupHomsAreGroupHomsOnPoints} (p.\pageref{thm_QuantumGroupHomsAreGroupHomsOnPoints}) and \ref{thm_underlyingGroupFunctor} (p.\pageref{thm_underlyingGroupFunctor})} show how coherent groups can be used to encode groups into arbitrary $\dagger$-SMCs.
	\item[(iii)] \textbf{Theorem \ref{thm_WellPointedQuantumGroupAreGroupAlgebras} (p.\pageref{thm_WellPointedQuantumGroupAreGroupAlgebras})} shows that well-pointed coherent groups generalise group algebras on categories of finite-dimensional vector spaces (and certain more general categories of matrices over commutative semirings with involution).
\end{enumerate}

\subsubsection{Wavefunctions on a periodic lattice}

In Section \ref{section_wavefunctionsPeriodicLattice} we go back to wavefunctions on a periodic lattice, and we show that finite abelian coherent groups capture the salient abstract features of position/momentum observables and their relation to translation/boost symmetry. The narrative of this Section focuses on periodic lattices as a concrete and accessible example, but the results we obtain are fully general. Specifically, we will prove that symmetry-observable duality, the Weyl Canonical Commutation Relations and the weak uncertainty principle are results that hold for all coherent groups, not just for the ones we identify with periodic lattice symmetry.

The joint aim of the results in this Section is to show that the observable associated with the group structure in a coherent group can be suitable interpreted as a momentum observable. Namely, we expect to have symmetry-observable dualities for translation-momentum and boost-position, we expect the position/momentum pair to satisfy the Weyl Canonical Commutation Relations, and we expect some form of the Uncertainty Principle to hold. 
\begin{enumerate}
	\item[(i)] It is expected that the momentum eigenstates on a finite periodic lattice generate the lattice translation group $G$, and are invariant states under lattice translations. \textbf{Theorems \ref{thm_latticeTranslationSymmetryAction} (p.\pageref{thm_latticeTranslationSymmetryAction}) and \ref{thm_multiplicativeCharInvariance} (p.\pageref{thm_multiplicativeCharInvariance})} show that the classical states for the group structure of a coherent group generate lattice translations, and are invariant states for lattice translations.
	\item[(ii)] It is expected that the position eigenstates on a finite periodic lattice generate the momentum boost group $G^\wedge$, and are invariant under boosts. \textbf{Theorems \ref{thm_latticeBoostSymmetryAction} (p.\pageref{thm_latticeBoostSymmetryAction}) and \ref{thm_pointsBoostInvariance} (p.\pageref{thm_pointsBoostInvariance})} show that the position eigenstates generate a group symmetry $G^\wedge$ on the classical states of the group structure, and are invariant states for this symmetry.
	\item[(iii)] It is expected that the position/momentum pair on a finite periodic lattice satisfy the Weyl Canonical Commutation Relations. \textbf{Theorem \ref{thm_WeylCCRs} (p.\pageref{thm_WeylCCRs})} shows that the position observable and the observable associated with the groups structure satisfy the Weyl Canonical Commutation Relations.
	\item[(iv)] We argue that the full form of the Uncertainty Principle is too strong a requirement for a position momentum pair, as there are theories with reasonable notions of position and momentum observables which cannot satisfy it. \textbf{Theorem \ref{thm_hyperbolicUncertaintyPrincipleFails} (p.\pageref{thm_hyperbolicUncertaintyPrincipleFails})} provides a concrete example of failure of the Uncertainty Principle for a position/momentum pair in hyperbolic quantum theory, due to the fact that the former is a local theory.
	\item[(v)] We argue that position/momentum pairs should satisfy a weaker form of the Uncertainty Principle, postulating that states of complete knowledge about position have completely indeterminate momentum, and vice versa. \textbf{Theorem \ref{thm_uncertaintyPrinciple} (p.\pageref{thm_uncertaintyPrinciple})} shows that the position observable and the observable associated with the group structure satisfy this weaker form of the Uncertainty Principle.
\end{enumerate}

\subsubsection{Systems with symmetries}

In Section \ref{section_systemsWithSymmetries} we extend our coherent approach to general symmetric systems. Following the identification of classical symmetries with unitary representations of groups, we define coherent symmetries as unitary representations of coherent groups. We provide a categorical characterisation of symmetric systems as objects of the Eilenberg-Moore category for a certain monad, and we extend symmetry-observable duality results from coherent groups to their representations. We conclude the Section with a digression on Stone's Theorem, which we rephrase within our framework.

The joint aim of the results in this Section is to show that the unitary representations of coherent groups model systems equipped with coherent symmetries, satisfying a generalised version of symmetry-observable duality. 
\begin{enumerate}
	\item[(i)] \textbf{Theorem \ref{thm_classicalFromQuantumReps} (p.\pageref{thm_classicalFromQuantumReps})} relates unitary representations of coherent groups to unitary representations of the classical groups they encode.
	\item[(ii)] \textbf{Subsection \ref{subsection_EMCategory} (p.\pageref{subsection_EMCategory})} explains how systems with symmetry governed by a coherent group have a natural interpretation as the objects of the Eilenberg-Moore category for a certain strong commutative monad, with equivariant maps as morphisms between them.
	\item[(iii)] \textbf{Theorems \ref{thm_symmetryObservableDuality} (p.\pageref{thm_symmetryObservableDuality}), \ref{thm_symmetryInvariantDuality} (p.\pageref{thm_symmetryInvariantDuality}) and \ref{thm_invariantStates} (p.\pageref{thm_invariantStates})} prove symmetry-observable for general systems with symmetry governed by a coherent group. 
	\item[(iv)] \textbf{Subsection \ref{subsection_StoneTheoremRevisited} (p.\pageref{subsection_StoneTheoremRevisited})} reformulates Stone's theorem on 1-parameter unitary groups in terms of projection-valued measures, and connects it to the results on symmetry-observable duality for symmetric systems established in the rest of the Section. 
\end{enumerate}

\subsubsection{Infinite-dimensional CQM}

In Section \ref{section_compactAbelian} we introduce the framework of infinite-dimensional CQM to deal with the coherent treatment of certain infinite groups in quantum mechanics. We explicitly cover the textbook example of position/momentum observables for 1-dimensional wavefunctions with periodic boundary conditions (with translation symmetry group $S^1$), but our techniques naturally extend to other compact and discrete abelian groups (e.g. the case of toroidal translation symmetry groups $\torusGroup{d}$ or lattice translation symmetry groups $\integers^d$). Unfortunately, infinite-dimensional CQM is a very young field, and treatment of locally compact symmetry groups (e.g. the continuous time-translation symmetry group $\reals$ or the continuous space-translation symmetry groups $\reals^d$) is left to future work. This Section is taken from \cite{Gogioso2016b}.

The joint aim of the results in this Section is to obtain a categorical formulation of separable Hilbert spaces and (possibly unbounded) linear maps between them which features the algebraic ingredients (namely strongly complementary pairs of $\dagger$-qSFAs) necessary to talk about some infinite groups (such as $\integers^d$ and $\torusGroup{d}$) in the context of our framework.
\begin{enumerate}
	\item[(i)] \textbf{Theorems \ref{thm_starHilbSMC} (p.\pageref{thm_starHilbSMC} and \ref{thmNS_SeparableInStarHilb} (p.\pageref{thmNS_SeparableInStarHilb}))} define a new $\dagger$-SMC $\starHilbCategory$ of non-standard separable Hilbert spaces and (possibly unbounded) maps between them, and relate it to the category of separable Hilbert spaces and bounded maps between them.
	\item[(ii)] \textbf{Theorems \ref{thmNS_ClassicalStructures} (p.\pageref{thmNS_ClassicalStructures}) and \ref{thmNS_CompactClosed} (p.\pageref{thmNS_CompactClosed})} show that $\starHilbCategory$ is dagger compact, and has unital $\dagger$-SCFAs.
	\item[(iii)] \textbf{Subsection \ref{subsubsection_WavefunctionsPeriodic} (p.\pageref{subsubsection_WavefunctionsPeriodic})} explicitly constructs the non-standard model for wavefunctions in a box with periodic boundary conditions (we cover the $1$-dimensional case explicitly, but the treatment straightforwardly extends to boxes with arbitrarily many dimensions). In particular, \textbf{Theorem \ref{thmNS_StrongComplementarity} (p.\pageref{thmNS_StrongComplementarity})} shows that there is a doubly well-pointed coherent group corresponding to the position/momentum pair for the wavefunction. 
\end{enumerate}

\subsubsection{Categorical Quantum Dynamics}

In Section \ref{section_finiteCyclic} we apply the results from Sections \ref{section_wavefunctionsPeriodicLattice} and \ref{section_systemsWithSymmetries} to the special case of discrete periodic dynamics, which we see as symmetries governed by finite cyclic groups (1-dim periodic lattices): symmetry-observable duality yields the Hamiltonian, while the defining equation of Eilenberg-Moore algebras turns into Schr\"{o}dinger's equation. We cover some construction of specific interest in quantum dynamics, such as Feynman's clock and von Neumann's Mean Ergodic Theorem, and we also look at synchronisation, the emergence of quantum clocks, and the existence of time observables. A first version of this Section appeared in \cite{Gogioso2015a}.


\newpage
\section{Coherent Groups}
\label{section_quantumGroups}

\subsection{Wavefunctions on periodic lattices}

The simplest Hilbert space endowed with a finite abelian group symmetry is the group algebra $\complexs[G]$, together with the regular representation of $G$:
\begin{equation}
U(g) := \ket{h} \mapsto \ket{h \oplus g}
\end{equation}
We will use $\oplus$ to denote the addition operation of a generic abelian group, and $0$ for the corresponding unit.\footnote{There are four good reasons to use $\oplus$ for the addition in abelian groups, at least within the context of this work. Firstly, the notation is reminiscent of the XOR operation, which is the addition in the abelian group $\integersMod{2}^N$ of $N$-bit strings and is the most common abelian group operation appearing in quantum computing and protocols. Secondly, while the multiplicative notation would allow a uniform treatment of the abelian and non-abelian cases, it would also result very unfamiliar in the specific instances we are interested in, where group elements are treated as vectors (e.g. the translations of a periodic lattice). Thirdly, the notation can cause no confusion with the direct sum of Hilbert spaces, the most common meaning of $\oplus$ in the context of quantum theory, as direct sums will play no role whatsoever in this work. Finally, the other common additive symbol is $+$, which is already used for superposition in pure-state quantum theory and for convex combination in mixed-state quantum theory and more in general in the context of $R$-probabilistic theories.}
When we interpret a finite abelian group $G = \prod_{d=1}^{D} \integersMod{n_d}$ as the translation group for a periodic lattice $\Lambda$, the group algebra has the very concrete interpretation of a quantum system corresponding to a wavefunction on the lattice. Because of this, we can expect the position and momentum observables to play an important role in the structural characterisation of the group algebra (and indeed this will be the case). More in general, when talking about a quantum system with lattice symmetry we will mean a system $\SpaceH$ which comes equipped with a unitary representation of the translation group $\prod_{d=1}^{D} \integersMod{n_d}$. From an operational perspective, we can imagine the following scenario: we have a system $\SpaceH$, we can transform it reversibly by translation of a lattice $\Lambda$, but we know nothing more of its internal structure than what the transformations tell us.

The essence of a group algebra is embodied by the interplay between two kinds of structures: there is the classical data of the group, embedded into the system via a distinguished orthonormal basis (the \textbf{standard basis}), and there is the group structure on that classical data. We have already seen that, from a coherent perspective, the classical data embedded in the system is represented by some $\dagger$-SCFA~$\hbox{\input{symbols/ZbwdotSym.tex}}\!\!_G$, which corresponds to a non-degenerate quantum observable, with 1-dimensional projectors specified by the elements of the translation group. If the latter are identified with the lattice sites (e.g. by fixing a distinguished site), then $\hbox{\input{symbols/ZbwdotSym.tex}}\!\!_G$ corresponds exactly to the position observable for a wavefunction on the lattice:
\begin{equation}\label{coherentDataG}
\begin{multlined}
\input{pictures/chapter3/coherentDataG.tikz}
\end{multlined}
\end{equation}

What about the group structure? Do we gain anything by employing a coherent approach in this case? As it will turn out in the rest of this work, we do (and quite a lot). We being by considering the coherent versions of the group multiplication~$\oplus$ (also known as \textit{addition}) and group inverse~$\ominus$ (as well as the group unit~$0$, which is already embedded as a distinguished element of the standard basis):
\begin{equation}\label{coherentGroupOpsG}
\begin{multlined}
\input{pictures/chapter3/coherentGroupOps.tikz}
\end{multlined}
\end{equation}
Because of the group structure, these processes come with a number of interesting structural properties; to begin with, the group multiplication and unit form a monoid.

The coherent group inverse is an involution, and it satisfies the following equation known as \textbf{Hopf's law}:
\begin{equation}\label{HopfsLaw1}
\begin{multlined}
\input{pictures/chapter3/HopfsLaw.tikz}
\end{multlined}
\end{equation}
The relation of Hopf's law to the group inverse can be seen by applying both sides of each equation to any $\hbox{\input{symbols/ZbwdotSym.tex}}\!\!$-classical state $\ket{g}$: on the LHS of the first equation/RHS of the second equation, the state is copied, one copy is inverted, and both copies are added, yielding $\ket{(\ominus g) \oplus g}$ for the first equation and $\ket{g \oplus (\ominus g)}$ for the second equation; on the RHS of the first equation/LHS of the second equation, the state is coherently deleted, and replaced with the state $\ket{0}$; all in all, the equations read $(\ominus g) \oplus g = 0 = g \oplus (\ominus g)$, which is the very definition of group inverse. 

Finally, the coherent group unit is an element of the standard basis (i.e. $\!\hbox{\input{symbols/DunitSym.tex}}\!\!$ is a $\hbox{\input{symbols/ZbwdotSym.tex}}\!\!$-classical state), and application of the coherent multiplication to elements of the standard basis yields elements of the standard basis (i.e. $\!\hbox{\input{symbols/DmultSym.tex}}\!\!$ is a $\hbox{\input{symbols/ZbwdotSym.tex}}\!\!$-classical\footnote{Henceforth, we will simply use \textbf{$\hbox{\input{symbols/ZbwdotSym.tex}}\!\!$-classical} when referring to $\hbox{\input{symbols/ZbwdotSym.tex}}\!\!$-to-$\hbox{\input{symbols/ZbwdotSym.tex}}\!\!$ classical processes.} process):
\begin{equation}\label{strongComplementarity1}
\begin{multlined}
\resizebox{\textwidth}{!}{\input{pictures/chapter3/strongComplementarity.tikz}}
\end{multlined}
\end{equation}
The top row contains the three conditions (recall from the previous chapter: copy, adjoin and delete) for the coherent group unit to be a $\hbox{\input{symbols/ZbwdotSym.tex}}\!\!$-classical state, while the bottom row contains the three conditions (again: copy, adjoin and delete) for the coherent group multiplication to be a $\hbox{\input{symbols/ZbwdotSym.tex}}\!\!$-classical process (which, in particular, maps $\hbox{\input{symbols/ZbwdotSym.tex}}\!\!$-classical states to $\hbox{\input{symbols/ZbwdotSym.tex}}\!\!$-classical states).

Perhaps the most important property, however, is one which isn't directly inspired by classical groups, and is instead unique to the coherent version of the operations: $\!\hbox{\input{symbols/DmultSym.tex}}\!\!$ and $\!\hbox{\input{symbols/DunitSym.tex}}\!\!$ form the multiplicative fragment of a $\dagger$-Frobenius algebra:
\begin{equation}\label{FrobeniusLawCoherentGroupOps}
\begin{multlined}
\input{pictures/chapter3/FrobeniusLawCoherentGroupOps.tikz}
\end{multlined}
\end{equation}
To be precise, they form a $\dagger$-qSCFA (commutativity is equivalent to the group being abelian) with normalisation factor $|G|$ (the composite $(\!\hbox{\input{symbols/DmultSym.tex}}\!\! \circ \!\hbox{\input{symbols/DcomultSym.tex}}\!\!)$ sends $\ket{g}$ to $\sum_{h \oplus k = g} \ket{h \oplus k} = |G| \cdot \ket{g}$). 

The statement of Frobenius law does not involve the coherent group inverse, and one might therefore imagine that a (commutative) monoid would also give rise to a $\dagger$-FA on its algebra. On the contrary, it turns out that Frobenius law can only be satisfied when $G$ is a group: for any fixed $g,h \in G$, the sum $\sum_{a\oplus b = g \oplus h} \ket{a} \otimes \ket{b}$ contains a term for $\ket{0} \otimes \ket{g \oplus h}$, which means that the sum $\sum_{c\oplus b = h} \ket{g \oplus c} \otimes \ket{b}$ must also contain a term $\ket{g \oplus c} \otimes \ket{g \oplus h}$ with $g \oplus c = 0$, which in turn means that $g$ must be invertible. Quasi-speciality also depends partially on the fact that $G$ is a group, rather than a monoid: the last equality in the proof above requires each element $g$ to have the same number of pairs $(h,k)$ such that $h \oplus k$, something which is true when $G$ is a group (there are always exactly $|G|$ many such pairs), but need not be true for a general monoid. The apparent absence of the group inverse from the proof of Frobenius law becomes less surprising when we realise that the coherent group inverse, also known as the \textbf{antipode}, can always be constructed by only using the $\dagger$-SCFA $\hbox{\input{symbols/ZbwdotSym.tex}}\!\!$ and the $\dagger$-qSCFA $\hbox{\input{symbols/DdotSym.tex}}\!\!$:
\begin{equation}\label{antipode1}
\begin{multlined}
\input{pictures/chapter3/antipode.tikz}
\end{multlined}
\end{equation}
To see that the equalities above hold, apply the three processes to the the state $\ket{g}$ and test them against the effect $\bra{h}$, for all $g,h \in G$: the group inverse on the left yields the scalar 1 if $h = \ominus g$ and the scalar 0 otherwise; the process in the middle yields the scalar 1 if $h \oplus g = 0$ and the scalar 0 otherwise; the process on the right yields the scalar 1 if $g \oplus h = 0$ and the scalar 0 otherwise.

Because $\hbox{\input{symbols/DdotSym.tex}}\!\!$ is a $\dagger$-qSCFA, it is legitimate to ask whether it corresponds to some interesting quantum observable. To find the orthonormal basis associated to it, we want to study the $\hbox{\input{symbols/DdotSym.tex}}\!\!$-classical states. Equivalently, we can study their adjoints, which are exactly the effects satisfying the following three conditions\footnote{The adjoin condition has been equivalently rewritten in terms of the group inverse, multiplying both sides of the original condition by the symmetric cup corresponding to $\hbox{\input{symbols/ZbwdotSym.tex}}\!\!$.}: 
\begin{equation}\label{multiplicativeChar1}
\begin{multlined}
\input{pictures/chapter3/multiplicativeChar.tikz}
\end{multlined}
\end{equation}
Writing $\chi(g) := \braket{\chi}{g}$ for any such effect $\bra{\chi}$ and any state $\ket{g}$ of the standard basis, we see that the adjoints of the $\hbox{\input{symbols/DdotSym.tex}}\!\!$ classical states correspond to maps $\chi: G \rightarrow \complexs$ satisfying $\chi(g \oplus h) = \chi(g) \cdot \chi(h)$, $\chi(0) = 1$, and $\chi(\ominus g) = \chi(g)^\ast$, i.e. to the multiplicative characters of the finite abelian group $G$. 

In order to understand what this means concretely, observe that the elements $g \in \prod_{d=1}^{D} \integersMod{n_d}$ of the translation group  for the lattice $\Lambda$ can be written in terms of components $g = (g_d)_{d=1}^D$, where $g_d \in \integersMod{n_d}$ for all $d=1,...,D$. This makes them look like vectors, and that's not far from true: because they form an abelian group, the translations can always be understood as forming a $\integers$-module. When $G$ is an abelian group, the multiplicative characters always form an abelian group, known as the \textbf{Pontryagin dual} of $G$ and denoted $G^\wedge$, under pointwise multiplication:
\begin{equation}
\chi \cdot \chi' := g \mapsto \chi(g) \chi'(g)
\end{equation}
The group unit of $G^\wedge$ is the \textbf{trivial character} $1 := g \mapsto 1$.
When $G$ is a finite abelian group, it is always isomorphic to its Pontryagin dual, but not in a canonical way (i.e. there are, in general, many equally legitimate choices of isomorphism $G \isom G^\wedge$; more about this later). In the case of $\prod_{d=1}^{D} \integersMod{n_d}$, the multiplicative characters can always be written in the following way (here $h \mapsto \chi_h$ is our specific choice of isomorphism $G \isom G^\wedge$):
\begin{equation}
\chi_h := g \mapsto e^{ -2 \pi i\, g \cdot h}
\end{equation}
where the \inlineQuote{inner product} $g \cdot h$ in the $\integers$-module $\prod_{d=1}^{D} \integersMod{n_d}$ is defined as follows (multiplication $g_d \cdot h_d$ is done modulo $n_d$):
\begin{equation}
g \cdot h :=  \sum_{d=1}^D \frac{g_d \cdot h_d}{n_d} 
\end{equation}
With this explicit characterisation in hand, we are able to write down the orthogonal basis corresponding to the $\hbox{\input{symbols/DdotSym.tex}}\!\!$ observable:
\begin{equation}
\ket{\chi_h} = \sum_{g \in G}  e^{ -2 \pi i \, g \cdot h} \ket{g}
\end{equation} 
These are the plane waves on a periodic lattice, and hence the quantum observable corresponding\footnote{Implicitly taking into account the fact that the basis states are not normalised.} to $\hbox{\input{symbols/DdotSym.tex}}\!\!$ is exactly the momentum observable. 

The result above is an iconic example of what will happen again and again in this Chapter: we treat the classical group symmetry coherently and the observable associated with the corresponding invariant comes out of the framework for free. This is because, contrary to the classical perspective, the coherent perspective is not rigid, and allows us to look at the same primitive from different angles: when looked from the point of view of the position observable, the coherent group multiplication behaves exactly like the classical group multiplication, but if we switch point of view we can also see as part of an observable in its own right, namely the momentum observable. 

This kind of direct connection between the coherent treatment of a symmetry and the observables corresponding to its conserved quantity is not limited to lattices, or to quantum mechanics; instead, it will be a general result of the framework introduced in this work. The most surprising aspect of this framework will be how it manages to turn few simple ingredients into a whole array of traditional cornerstones of quantum mechanics: we will re-discover familiar results such as position/momentum duality, time/energy duality and the Weyl canonical commutation relations, as well as special cases of Stone's theorem and von Neumann's mean ergodic theorem. Their new formulation in abstract, diagrammatic terms gives us new structural and operational understanding of the reasons behind their validity in quantum theory. At the same time, the abstract and algebraic nature of our definitions and proofs will extend the validity of our results to a much larger spectrum of quantum-like theories, including all those theories based on semiring-valued wavefunctions that we presented at the end of the last Chapter (e.g. real quantum theory, relational quantum theory, modal quantum theory, $p$-adic quantum theory, etc).

\subsection{Complementarity and strong complementarity}

Hopf's law from \ref{HopfsLaw1} and the \textbf{bialgebra law} from \ref{strongComplementarity1} (the leftmost equation in the bottom row) are well known in quantum algebra, and make their first appearance in CQM as part of the ZX calculus \cite{Coecke2011}. The ZX calculus focuses on the algebraic relation between the single-qubit Pauli X, Y and Z observables. Both Pauli X and Pauli Y are \textit{complementary}, or \textit{mutually unbiased} to Pauli Z: their eigenstates lie on the equator of the Bloch sphere, an this property is captured by Hopf's law. Amongst the observables complementary to Pauli Z, however, Pauli X plays a special role: if the eigenstates of Pauli Z are taken as the computational basis, the eigenstates of Pauli X are uniquely characterised by the fact that they form the corresponding Fourier basis (they are determined by the multiplicative characters of the finite abelian group $\integersMod{2}$). This special relationship between Pauli Z and Pauli X is known as \textit{strong complementarity}, and is captured by the bialgebra law (and some subset of the equations in \ref{strongComplementarity1}, depending on the specific work \cite{Coecke2011,Backens2014,Duncan2016,Gogioso2016d}). 

In this Subsection, we define complementarity and strong complementarity in $\dagger$-SMCs, and prove their general relationship to being mutual unbiased, group structure and the Fourier transform. 

\begin{definition}\label{def_complementarity}
	Two symmetric $\dagger$-qSFAs $\hbox{\input{symbols/ZbwdotSym.tex}}\!\!$ and $\hbox{\input{symbols/DdotSym.tex}}\!\!$ on the same object $\SpaceH$ of a $\dagger$-SMC are said to be \textbf{complementary} (or a \textbf{complementary pair}) if they satisfy \textbf{Hopf's Law}:
	\begin{equation}\label{hopfsLaw}
		\resizebox{\textwidth}{!}{\input{pictures/chapter3/HopfsLaw.tikz}}
	\end{equation}
	where the \textbf{antipode} $\hbox{\input{symbols/antipodeSym.tex}}\!\!: \SpaceH \rightarrow \SpaceH$ is the unitary defined as follows, which we furthermore require to be self-adjoint (or equivalently self-inverse) as part of this definition:
	\begin{equation}\label{antipode}
		\input{pictures/chapter3/antipode.tikz}
	\end{equation}
\end{definition}
\begin{remark}
Because the $\dagger$-qSFAs are chosen to be symmetric, the antipode can furthermore be written in the following additional ways:
	\begin{equation}\label{antipodeConj}
		\input{pictures/chapter3/antipodeConj.tikz}
	\end{equation}
\end{remark}
\begin{definition}\label{def_unbiasedStates}
	Let $\hbox{\input{symbols/ZbwdotSym.tex}}\!\!$ be a $\dagger$-qSFAs in a dagger compact category $\CategoryC$. Then a state $\psi$ is a \textbf{$\hbox{\input{symbols/ZbwdotSym.tex}}\!\!$-unbiased} state if the following holds in $\CPStarCategory{\CategoryC}$:
	\begin{equation}\label{ZunbiasedState}
		\input{pictures/chapter3/ZunbiasedState.tikz}
	\end{equation}
	Equation \ref{ZunbiasedState} can be unfolded into the following more general definition, which holds in an arbitrary $\dagger$-SMC:
	\begin{equation}
		\input{pictures/chapter3/ZphaseStateDefExplained.tikz}
	\end{equation}
\end{definition}
\noindent In $\fdHilbCategory$, the $\hbox{\input{symbols/ZbwdotSym.tex}}\!\!$-unbiased states are exactly those which, upon normalisation, yield the uniform distribution when measured in the $\hbox{\input{symbols/ZbwdotSym.tex}}\!\!$ observable.

\begin{lemma}\label{lem_complementarityUnbiased}
	Consider a pair of symmetric $\dagger$-qSFAs $\hbox{\input{symbols/ZbwdotSym.tex}}\!\!$ and $\hbox{\input{symbols/DdotSym.tex}}\!\!$ in a $\dagger$-SMC. If $(\hbox{\input{symbols/ZbwdotSym.tex}}\!\!,\hbox{\input{symbols/DdotSym.tex}}\!\!)$ is a complementary pair, then the $\hbox{\input{symbols/DdotSym.tex}}\!\!$-classical states are $\hbox{\input{symbols/ZbwdotSym.tex}}\!\!$-unbiased, and the $\hbox{\input{symbols/ZbwdotSym.tex}}\!\!$-classical states are $\hbox{\input{symbols/DdotSym.tex}}\!\!$-unbiased.
\end{lemma}
\begin{proof} 
	We prove that a $\hbox{\input{symbols/DdotSym.tex}}\!\!$-classical state $\goodchi$ is $\hbox{\input{symbols/ZbwdotSym.tex}}\!\!$-unbiased: 
	\begin{equation}\label{HopfLawUnbiasedStatesProof}
	\resizebox{\textwidth}{!}{\input{pictures/chapter3/HopfLawUnbiasedStatesProof.tikz}}
	\end{equation}
	The first equality is by the delete condition for $\hbox{\input{symbols/DdotSym.tex}}\!\!$-classical states, the second equality is Hopf's law (together with the self-adjoint requirement for the antipode), the third equality is by the copy condition and adjoint conditions for $\hbox{\input{symbols/DdotSym.tex}}\!\!$-classical states, the last equality is by Frobenius law and unit law for $\hbox{\input{symbols/ZbwdotSym.tex}}\!\!$. The proof for $\hbox{\input{symbols/ZbwdotSym.tex}}\!\!$-classical states is the same, with colours swapped.
\end{proof}
\noindent Hence complementarity always results in mutual the observables involved being mutually unbiased, regardless of the specific theory under consideration. We can also prove the converse, that being mutual unbiased implies complementarity (in the sense of Hopf's law and self-adjointness of the antipode), as long as at least one of the two $\dagger$-qSFAs has enough classical states. 
\begin{lemma}\label{lem_complementarityUnbiasedConverse}
	Consider a pair of symmetric $\dagger$-qSFAs $\hbox{\input{symbols/ZbwdotSym.tex}}\!\!$ and $\hbox{\input{symbols/DdotSym.tex}}\!\!$ in a $\dagger$-SMC. If $\hbox{\input{symbols/ZbwdotSym.tex}}\!\!$ has enough classical states, and the $\hbox{\input{symbols/ZbwdotSym.tex}}\!\!$-classical states are $\hbox{\input{symbols/DdotSym.tex}}\!\!$-unbiased, then $(\hbox{\input{symbols/ZbwdotSym.tex}}\!\!,\hbox{\input{symbols/DdotSym.tex}}\!\!)$ is a complementary pair. Similarly, if $\hbox{\input{symbols/DdotSym.tex}}\!\!$ has enough classical states, and the $\hbox{\input{symbols/DdotSym.tex}}\!\!$-classical states are $\hbox{\input{symbols/ZbwdotSym.tex}}\!\!$-unbiased, then $(\hbox{\input{symbols/ZbwdotSym.tex}}\!\!,\hbox{\input{symbols/DdotSym.tex}}\!\!)$ is a complementary pair.
\end{lemma}
\begin{proof} 
	If $\hbox{\input{symbols/ZbwdotSym.tex}}\!\!$ has enough classical states, then each equation in \ref{HopfLawUnbiasedStatesProof} holds as long as it holds when tested against $\psi^\dagger$, where $\psi$ is an arbitrary $\hbox{\input{symbols/ZbwdotSym.tex}}\!\!$-classical state. The first equation in the chain is seen to hold by the delete conditions for $\hbox{\input{symbols/ZbwdotSym.tex}}\!\!$-classical states and $\hbox{\input{symbols/DdotSym.tex}}\!\!$-classical states. The second equation is seen to hold by applying the copy and adjoin conditions for $\hbox{\input{symbols/ZbwdotSym.tex}}\!\!$-classical states to the RHS, and then using the fact that $\hbox{\input{symbols/ZbwdotSym.tex}}\!\!$-classical states are $\hbox{\input{symbols/DdotSym.tex}}\!\!$-unbiased by hypothesis; an application of Frobenius law and unit laws for $\hbox{\input{symbols/DdotSym.tex}}\!\!$ is necessary to bring the diagram in the form required by the definition of $\hbox{\input{symbols/DdotSym.tex}}\!\!$-unbiased states (in the form of its adjoint, to be precise). The third equation is seen to hold by the copy and adjoin conditions for $\hbox{\input{symbols/DdotSym.tex}}\!\!$-classical states. The last equation is seen to hold by Frobenius law and unit laws for $\hbox{\input{symbols/ZbwdotSym.tex}}\!\!$.
\end{proof}
\noindent Variants of Lemmas \ref{lem_complementarityUnbiased} and \ref{lem_complementarityUnbiasedConverse} have previously appeared in the literature \cite{Coecke2016a}.

\begin{definition}\label{def_strongComplementarity}
	Two symmetric $\dagger$-qSFAs $\hbox{\input{symbols/ZbwdotSym.tex}}\!\!$ and $\hbox{\input{symbols/DdotSym.tex}}\!\!$ on the same object $\SpaceH$ of a $\dagger$-SMC are said to be \textbf{strongly complementary} (or a \textbf{strongly complementary pair}) if they are complementary and furthermore satisfy the following equations\footnote{The empty diagram on the RHS of the top right equation is the scalar 1.}:
	\begin{equation}\label{strongComplementarity}
	\resizebox{\textwidth}{!}{\input{pictures/chapter3/strongComplementarity.tikz}}
	\end{equation}
\end{definition}
\begin{remark}\label{rmrk_strongComplementarityColorSwappedEqns}
	The central equations of the top and bottom rows  of \ref{strongComplementarity} are a consequence of Hopf's law, self-adjointness of the antipode, and the other four equations:
	\begin{equation}
	\resizebox{\textwidth}{!}{\input{pictures/chapter3/HopfsLawClassicalAdjunctionsConsequence.tikz}}
	\end{equation}
	\begin{equation}
	\resizebox{\textwidth}{!}{\input{pictures/chapter3/HopfsLawClassicalAdjunctionsConsequence2.tikz}}
	\end{equation}
	Their corresponding colour-swapped versions are proven similarly. This is why previous presentations of strong complementarity often include only the remaining four equations (together with Hopf's law and either one of: (i) self-adjointness of the antipode, or (ii) the central equation of the top row together with its colour-swapped version). As a consequence, strong complementarity as a property is symmetric in $\hbox{\input{symbols/ZbwdotSym.tex}}\!\!$ and $\hbox{\input{symbols/DdotSym.tex}}\!\!$ (i.e. we could have equivalently stated it with the colour-swapped versions of the equations above).
\end{remark}

In the specific case of quantum mechanics, the relation between strong complementarity and the quantum Fourier transform is given by the following results, some bits of which already appeared in \cite{Coecke2012c,Kissinger2012} (for the abelian case only).
\begin{lemma}\label{lem_SCHilbGroupAlgebra}
	Let $\hbox{\input{symbols/ZbwdotSym.tex}}\!\!$ and $\hbox{\input{symbols/DdotSym.tex}}\!\!$ be a $\dagger$-SCFA and a $\dagger$-FA on the same finite-dimensional Hilbert space $\SpaceH$. Then $\hbox{\input{symbols/ZbwdotSym.tex}}\!\!$ and $\hbox{\input{symbols/DdotSym.tex}}\!\!$ are strongly complementary iff there exists a finite group $G$ such that $(\!\hbox{\input{symbols/DmultSym.tex}}\!\!,\!\hbox{\input{symbols/DunitSym.tex}}\!\!)$ endows the set of $\hbox{\input{symbols/ZbwdotSym.tex}}\!\!$-classical states with the group structure of $G$. Concretely, this means that we can label the $\hbox{\input{symbols/ZbwdotSym.tex}}\!\!$-classical states as $(\ket{g})_{g \in G}$ in a way such that:
	\begin{equation}\label{multG}
		\input{pictures/chapter3/multG.tikz}
	\end{equation}
	If $\hbox{\input{symbols/ZbwdotSym.tex}}\!\!$ and $\hbox{\input{symbols/DdotSym.tex}}\!\!$ are strongly complementary, then $\hbox{\input{symbols/DdotSym.tex}}\!\!$ is a $\dagger$-qSFA with normalisation factor $|G|$, and it is commutative if and only if the group $G$ is. 
\end{lemma}
\begin{proof}
By definition of strong complementarity, if $(\hbox{\input{symbols/ZbwdotSym.tex}}\!\!,\hbox{\input{symbols/DdotSym.tex}}\!\!)$ is a strongly complementary pair, then the set of $\hbox{\input{symbols/ZbwdotSym.tex}}\!\!$-classical states is always endowed with the structure of some group $G$ which is finite because any $\dagger$-SCFA $\hbox{\input{symbols/ZbwdotSym.tex}}\!\!$ can only have finitely many classical states in $\fdHilbCategory$. Conversely, if the the set of $\hbox{\input{symbols/ZbwdotSym.tex}}\!\!$-classical states is endowed by $(\!\hbox{\input{symbols/DmultSym.tex}}\!\!,\!\hbox{\input{symbols/DunitSym.tex}}\!\!)$ with the structure of some group $G$, then in particular $\!\hbox{\input{symbols/DmultSym.tex}}\!\!$ is a $\hbox{\input{symbols/ZbwdotSym.tex}}\!\!$-classical process, and $\!\hbox{\input{symbols/DunitSym.tex}}\!\!$ is a $\hbox{\input{symbols/ZbwdotSym.tex}}\!\!$-classical state: this means that the Equations \ref{strongComplementarity} always hold when applied to $\hbox{\input{symbols/ZbwdotSym.tex}}\!\!$-classical states, and hence they hold altogether because any $\dagger$-SCFA in $\fdHilbCategory$ always has enough classical states. 

It is immediate to see that $(\classicalStates{\hbox{\input{symbols/ZbwdotSym.tex}}\!\!},\!\hbox{\input{symbols/DmultSym.tex}}\!\!,\!\hbox{\input{symbols/DunitSym.tex}}\!\!)$ is abelian if and only if $\!\hbox{\input{symbols/DmultSym.tex}}\!\!$ is commutative. Furthermore, the composite $\!\hbox{\input{symbols/DmultSym.tex}}\!\! \circ \!\hbox{\input{symbols/DcomultSym.tex}}\!\!$ sends an $\hbox{\input{symbols/ZbwdotSym.tex}}\!\!$-classical state $\ket{g}$ to $\sum_{hk = g} \ket{hk} = |G| \cdot \ket{g}$: because $\hbox{\input{symbols/ZbwdotSym.tex}}\!\!$ has enough classical states, then this implies that $\hbox{\input{symbols/DdotSym.tex}}\!\!$ is quasi-special, with normalisation factor $|G|$.
\end{proof}

\begin{lemma}\label{lem_SCHilbQFT}
	Let $\hbox{\input{symbols/ZbwdotSym.tex}}\!\!$ and $\hbox{\input{symbols/DdotSym.tex}}\!\!$ be a strongly complementary pair of a $\dagger$-SCFA and a $\dagger$-qSFA on the same finite-dimensional Hilbert space $\SpaceH$. Then the $\hbox{\input{symbols/DdotSym.tex}}\!\!$-classical states are labelled by the multiplicative characters $\goodchi: G \rightarrow S^1$ of the group $G:=(\classicalStates{\hbox{\input{symbols/ZbwdotSym.tex}}\!\!},\!\hbox{\input{symbols/DmultSym.tex}}\!\!,\!\hbox{\input{symbols/DunitSym.tex}}\!\!)$, and take the following form in terms of the $\hbox{\input{symbols/ZbwdotSym.tex}}\!\!$-classical states:
	\begin{equation}\label{XclassicalStatesMultChars}
		\ket{\goodchi} := \sum_{g \in G} \goodchi(g) \ket{g} 
	\end{equation}
	If $\hbox{\input{symbols/DdotSym.tex}}\!\!$ is commutative, then it has enough classical states, and measurement in the $\hbox{\input{symbols/DdotSym.tex}}\!\!$ observable provides the quantum Fourier transform:
	\begin{equation}\label{XmeasurementQFT}
		\input{pictures/chapter3/XmeasurementQFT.tikz}
	\end{equation}
\end{lemma} 
\begin{proof}
The adjoints of the $\hbox{\input{symbols/DdotSym.tex}}\!\!$-classical states satisfy Equations \ref{multiplicativeChar1}: when restricted to the $\hbox{\input{symbols/ZbwdotSym.tex}}\!\!$-classical states, the equations are equivalent to those defining the multiplicative characters $\goodchi \in G^\wedge$ of the finite abelian group $(G,\oplus,0)$ induced by $(\!\hbox{\input{symbols/DmultSym.tex}}\!\!,\!\hbox{\input{symbols/DunitSym.tex}}\!\!)$ on the $\hbox{\input{symbols/ZbwdotSym.tex}}\!\!$-classical states. Hence, they take the desired form. If $\hbox{\input{symbols/DdotSym.tex}}\!\!$ is commutative, then $G$ is abelian, and the $\hbox{\input{symbols/DdotSym.tex}}\!\!$-classical states form an orthogonal basis.

When $\hbox{\input{symbols/DdotSym.tex}}\!\!$ has enough classical states, $(\SpaceH, \hbox{\input{symbols/DdotSym.tex}}\!\!)$ is a classical system in $\CPStarCategory{\fdHilbCategory}$. The fact that $\!\hbox{\input{symbols/DcomultSym.tex}}\!\! = \sum_{\goodchi \in G^\wedge} \ket{\goodchi} \otimes \ket{\goodchi} \otimes \bra{\goodchi}$, together with Equation \ref{XclassicalStatesMultChars}, yields the desired result.
\end{proof}

More results about complementarity and strong complementarity, and their relation to the group of phase gates, will be presented in the next Chapter, in regards to Mermin-type non-locality arguments and the abelian Hidden Subgroup Problem.

\subsection{Coherent Groups}

In the $\fdHilbCategory$ results we have seen above, strong complementarity captures exactly the concept of group algebra. In one direction, a group algebra $\complexs[G]$ always give rise to a strongly complementary pair, by considering the $\dagger$-SCFA associated with the standard basis and the $\dagger$-qSFA generated by the coherent group multiplication and unit. In the other direction, consider a strongly complementary pair of a $\dagger$-SCFA $\hbox{\input{symbols/ZbwdotSym.tex}}\!\!$ and a $\dagger$-qSFA $\hbox{\input{symbols/DdotSym.tex}}\!\!$ on a finite-dimensional Hilbert space $\SpaceH$: by Theorem \ref{lem_SCHilbGroupAlgebra}, there always is a finite group $G$ and a unique isomorphism $\SpaceH \isom \complexs[G]$ which sends the $\hbox{\input{symbols/ZbwdotSym.tex}}\!\!$-classical states to the elements of the standard basis of $\complexs[G]$, and that isomorphism turns $\!\hbox{\input{symbols/DmultSym.tex}}\!\!$ into the coherent group multiplication on $\complexs[G]$.

We take strongly complementary pairs as our starting point to generalise the coherent treatment of group-theoretic primitives away from the quantum case.  
\begin{definition}
By a \textbf{coherent group} in $\dagger$-SMC we will mean a strongly complementary pair $(\hbox{\input{symbols/ZbwdotSym.tex}}\!\!,\hbox{\input{symbols/DdotSym.tex}}\!\!)$ of two symmetric $\dagger$-qSFAs. We will refer to $\hbox{\input{symbols/ZbwdotSym.tex}}\!\!$ as the \textbf{point structure}, and to its classical states as the \textbf{points} of the coherent group. We will refer to $\hbox{\input{symbols/DdotSym.tex}}\!\!$ as the \textbf{group structure}.
\end{definition}
\noindent We will use (coherent) \textbf{copy} and (coherent) \textbf{deletion} to refer to the comultiplication and counit of the point structure, and (coherent) \textbf{multiplication} and \textbf{unit} to refer to the multiplication and unit of the group structure. 

Although the underlying definition of the strongly complementary pair $(\hbox{\input{symbols/ZbwdotSym.tex}}\!\!,\hbox{\input{symbols/DdotSym.tex}}\!\!)$ is symmetric in $\hbox{\input{symbols/ZbwdotSym.tex}}\!\!$ and $\hbox{\input{symbols/DdotSym.tex}}\!\!$, this symmetry is broken in our definition of the coherent group $(\hbox{\input{symbols/ZbwdotSym.tex}}\!\!,\hbox{\input{symbols/DdotSym.tex}}\!\!)$ by assigning distinct names to the two $\dagger$-qSFAs involved. However, the symmetry can be recovered by considering the coherent group $(\hbox{\input{symbols/DdotSym.tex}}\!\!,\hbox{\input{symbols/ZbwdotSym.tex}}\!\!)$, which we refer to as the \textbf{dual} of $(\hbox{\input{symbols/ZbwdotSym.tex}}\!\!,\hbox{\input{symbols/DdotSym.tex}}\!\!)$. Because the two structures play different roles, we will be interested in different properties for each. For example, we will say that a coherent group is \textbf{well-pointed} (or that it has \textbf{enough points}) if the \textit{point} structure has enough classical states (in which case it is also necessarily commutative), and we will say that a coherent group is \textbf{finite} if it is well-pointed with finitely many points. On the other hand, we will say that a coherent group is \textbf{commutative}, or \textbf{abelian}, if the \textit{group} structure is commutative.

We will take coherent groups on a given $\dagger$-SMC $\CategoryC$ to be the objects of a new SMC $\CoherentGroupsCategory{\CategoryC}$, which we define as follows:
\begin{enumerate}
	\item[(i)] the objects of $\CoherentGroupsCategory{\CategoryC}$ are the coherent groups on objects of $\CategoryC$;
	\item[(ii)] if $(\hbox{\input{symbols/ZbwdotSym.tex}}\!\!,\hbox{\input{symbols/DdotSym.tex}}\!\!)$ and $(\hbox{\input{symbols/YbwdotSym.tex}}\!\!,\hbox{\input{symbols/WbwdotSym.tex}}\!)$ are coherent groups on objects $\SpaceH$ and $\SpaceK$ respectively, then the morphisms $f:(\hbox{\input{symbols/ZbwdotSym.tex}}\!\!,\hbox{\input{symbols/DdotSym.tex}}\!\!) \rightarrow (\hbox{\input{symbols/YbwdotSym.tex}}\!\!,\hbox{\input{symbols/WbwdotSym.tex}}\!)$ of $\CoherentGroupsCategory{\CategoryC}$, which we will refer to as \textbf{coherent group homomorphisms}, are the morphisms $f: \SpaceH \rightarrow \SpaceK$ which are both $\hbox{\input{symbols/ZbwdotSym.tex}}\!\!$-to-$\hbox{\input{symbols/YbwdotSym.tex}}\!\!$ classical and monoid homomorphisms $(\SpaceH,\!\hbox{\input{symbols/DmultSym.tex}}\!\!,\!\hbox{\input{symbols/DunitSym.tex}}\!\!) \rightarrow (\SpaceK,\hbox{\input{symbols/WbwmultSym.tex}}\!,\hbox{\input{symbols/WbwunitSym.tex}}\!)$
	\item[(iii)] the tensor product on morphisms is inherited from $\CategoryC$, while the tensor product on objects is given by taking the tensor product of the two point structures as the new point structure, and the tensor product of the two group structures as the new group structure $(\hbox{\input{symbols/ZbwdotSym.tex}}\!\!,\hbox{\input{symbols/DdotSym.tex}}\!\!) \otimes (\hbox{\input{symbols/YbwdotSym.tex}}\!\!,\hbox{\input{symbols/WbwdotSym.tex}}\!) := (\hbox{\input{symbols/ZbwdotSym.tex}}\!\! \otimes \hbox{\input{symbols/YbwdotSym.tex}}\!\!,\hbox{\input{symbols/DdotSym.tex}}\!\! \otimes \hbox{\input{symbols/WbwdotSym.tex}}\!)$.
\end{enumerate}
Unfortunately, $\CoherentGroupsCategory{\CategoryC}$ does not come with a natural dagger. However, dual coherent groups and the dagger of $\CategoryC$ can be combined into a well-defined involutive automorphism $^{\wedge}: \CoherentGroupsCategory{\CategoryC} \rightarrow \OpCategory{\CoherentGroupsCategory{\CategoryC}}$, which sends a coherent group $(\hbox{\input{symbols/ZbwdotSym.tex}}\!\!,\hbox{\input{symbols/DdotSym.tex}}\!\!)$ to its dual $(\hbox{\input{symbols/ZbwdotSym.tex}}\!\!,\hbox{\input{symbols/DdotSym.tex}}\!\!)^\wedge := (\hbox{\input{symbols/DdotSym.tex}}\!\!,\hbox{\input{symbols/ZbwdotSym.tex}}\!\!)$, and a coherent group homomorphism $f$ to $f^\wedge := f^\dagger$. There is also a faithful functor of SMCs $\CoherentGroupsCategory{\CategoryC} \rightarrow \CategoryC$ which sends each coherent group to its underlying object in $\CategoryC$, and is the identity on morphisms.  

The name \textit{quantum group} is used in the literature to refer to a variety of inequivalent algebraic objects \cite{Drinfeld1987,Kassel1995,Majid2000,Street2007,Woronowicz1987,Woronowicz1998}, sharing a common conceptual flavour but differing in their actual mathematical implementation. The coherent groups used in this work are closely related to quantum groups: first and foremost, they all feature some form of Hopf's law and bialgebra law, and they take direct inspiration from group algebras and their internal structure. Group algebras are often used to treat group theory within the linear framework of vector spaces, over $\complexs$ or some other field $k$: some famous results of this linear treatment are (in order of increasing generality) Fourier theory, Pontryagin duality and Tannaka-Krein dualities for modular categories. Quantum groups (at least in some definitions) can be seen as a direct generalisation of group algebras: they keep the algebraic skeleton, but lose the underlying classical structure given by the standard basis (which is the part that makes group algebras \inlineQuote{undesirably} rigid). This is conceptually akin to the way in which non-commutative spaces generalise traditional spaces in Non-commutative Geometry \cite{Connes1994}, or the way in which non-commutative C*-algebras and spectral bundles generalise classical observables and configuration spaces in Algebraic Quantum Theory \cite{Heunen2009}. 

We will proceed to show that coherent groups in $\fdHilbCategory$ are closely related to finite-dimensional compact quantum groups \cite{Woronowicz1987,Woronowicz1998}. In general, a \textbf{compact quantum group} is a pair $(C,\Delta)$ of a unital separable C*-algebra $C$ and a co-associative unital C*-algebra homomorphism $\Delta: C \rightarrow C \otimes C$ satisfying the additional conditions that the sets $(C \otimes 1) \cdot \Delta(C)$ and $(1 \otimes C) \cdot \Delta(C)$ are dense in $C \otimes C$. When the C*-algebra is finite dimensional, it can be written in the form $C = (\SpaceH,\!\hbox{\input{symbols/ZbwmultSym.tex}}\!\!,\!\hbox{\input{symbols/ZbwunitSym.tex}}\!\!)$ for a symmetric $\dagger$-SFA $\hbox{\input{symbols/ZbwdotSym.tex}}\!\!$. As a consequence, a compact quantum group in $\fdHilbCategory$ can be seen as a pair $(\SpaceH,\hbox{\input{symbols/ZbwdotSym.tex}}\!\!,\!\hbox{\input{symbols/DcomultSym.tex}}\!\!)$ a symmetric $\dagger$-SFA $\hbox{\input{symbols/ZbwdotSym.tex}}\!\!$ (or, without loss of generality, a $\dagger$-qSFA)  and a co-associative $\!\hbox{\input{symbols/DcomultSym.tex}}\!\!: \SpaceH \rightarrow \SpaceH \otimes \SpaceH$, satisfying the following conditions: 
\begin{itemize}
	\item The comultiplication $\!\hbox{\input{symbols/DmultSym.tex}}\!\!$ is $\hbox{\input{symbols/ZbwdotSym.tex}}\!\!$-to-$(\hbox{\input{symbols/ZbwdotSym.tex}}\!\!\otimes\hbox{\input{symbols/ZbwdotSym.tex}}\!\!)$ classical (i.e. a C*-algebra homomorphism $C \rightarrow C \otimes C$);
	\item the following linear endomorphisms of $\SpaceH \otimes \SpaceH$ are invertible\footnote{In finite dimensions, this is the same as saying that their image is dense.}:
	\begin{equation}\label{compactQuantumGroupDensityCondition}
		\input{pictures/chapter3/compactQuantumGroupDensityCondition.tikz}
	\end{equation}
\end{itemize}
Thanks to this equivalent characterisation, we can prove that coherent groups in $\fdHilbCategory$ arise as a special, extremely well-behaved case of compact quantum groups on finite-dimensional C*-algebras.
\begin{theorem}[\textbf{Coherent groups are compact quantum groups}]\label{thm_coherentGroupsCompactQuantumGroups} \hfill\\
There is a bijective correspondence between coherent groups $(\hbox{\input{symbols/ZbwdotSym.tex}}\!\!,\hbox{\input{symbols/DdotSym.tex}}\!\!)$ in $\fdHilbCategory$ with $\hbox{\input{symbols/ZbwdotSym.tex}}\!\!$ special and compact quantum groups $(C,\!\hbox{\input{symbols/DcomultSym.tex}}\!\!)$ satisfying the following conditions:
\begin{enumerate}
	\item[(i)] $C$ is a finite-dimensional C*-algebra, corresponding to a symmetric $\dagger$-SFA $\hbox{\input{symbols/ZbwdotSym.tex}}\!\!$;
	\item[(ii)] the comultiplication $\!\hbox{\input{symbols/DcomultSym.tex}}\!\!$ is part of a symmetric $\dagger$-qSFA $\hbox{\input{symbols/DdotSym.tex}}\!\!$;
	\item[(iii)] the unit $\!\hbox{\input{symbols/DunitSym.tex}}\!\!: \SpaceH \rightarrow \complexs$ is $\hbox{\input{symbols/ZbwdotSym.tex}}\!\!$-to-$(\complexs,\cdot,1)$ classical;
	\item[(iv)] the two maps depicted in \ref{compactQuantumGroupDensityCondition} are unitaries;
	\item[(v)] the following map is self-inverse:
		\begin{equation}\label{compactQuantumGroupAntipode}
		\input{pictures/chapter3/compactQuantumGroupAntipode.tikz}		
		\end{equation}
\end{enumerate} 
\end{theorem}
\begin{proof}
In one direction, consider a coherent group $(\hbox{\input{symbols/ZbwdotSym.tex}}\!\!,\hbox{\input{symbols/DdotSym.tex}}\!\!)$. Then $\hbox{\input{symbols/ZbwdotSym.tex}}\!\!$ induces a C*-algebra $C = (\SpaceH,\!\hbox{\input{symbols/ZbwmultSym.tex}}\!\!\!,\!\hbox{\input{symbols/ZbwunitSym.tex}}\!\!\!)$, with $\hbox{\input{symbols/DdotSym.tex}}\!\!$ co-associative and $\hbox{\input{symbols/ZbwdotSym.tex}}\!\!$-to-$(\hbox{\input{symbols/ZbwdotSym.tex}}\!\!\otimes\hbox{\input{symbols/ZbwdotSym.tex}}\!\!)$ classical (second row of Equations \ref{strongComplementarity} in the definition of strong complementarity). The comultiplication $\!\hbox{\input{symbols/DcomultSym.tex}}\!\!$ is part of the symmetric $\dagger$-qSFA $\hbox{\input{symbols/DdotSym.tex}}\!\!$, and the unit $\!\hbox{\input{symbols/DunitSym.tex}}\!\!$ is $\hbox{\input{symbols/ZbwdotSym.tex}}\!\!$-to-$(\complexs,\cdot,1)$ classical (first row of Equations \ref{strongComplementarity} in the definition of strong complementarity). The map depicted in \ref{compactQuantumGroupAntipode} is the antipode of the coherent group $(\hbox{\input{symbols/ZbwdotSym.tex}}\!\!,\hbox{\input{symbols/DdotSym.tex}}\!\!)$, and it is self-inverse as part of the definition of complementarity. The two maps depicted in \ref{compactQuantumGroupDensityCondition} are unitaries because of complementarity (see Theorem 9 from \cite{Zeng2014}). In particular, $\hbox{\input{symbols/DdotSym.tex}}\!\!$ is a co-associative C*-algebra homomorphism $C \rightarrow (C \otimes C)$ and the two maps depicted in \ref{compactQuantumGroupDensityCondition} are invertible, so that $(C,\!\hbox{\input{symbols/DcomultSym.tex}}\!\!)$ is a compact quantum group.

In the other direction, consider a compact quantum group $(C,\!\hbox{\input{symbols/DcomultSym.tex}}\!\!)$ on a finite-dimensional C*-algebra $C$ (condition (i)), which always takes the form $C = (\SpaceH,\!\hbox{\input{symbols/ZbwmultSym.tex}}\!\!\!,\!\hbox{\input{symbols/ZbwunitSym.tex}}\!\!\!)$ for some symmetric $\dagger$-SFA $\hbox{\input{symbols/ZbwdotSym.tex}}\!\!$, and assume that the four conditions (ii)-(v) are satisfied. Then $\!\hbox{\input{symbols/DmultSym.tex}}\!\!$ is part of a symmetric $\dagger$-qSFA $\hbox{\input{symbols/DdotSym.tex}}\!\!$ (by condition (ii)) which satisfies Equations \ref{strongComplementarity} (because $\!\hbox{\input{symbols/DcomultSym.tex}}\!\!$ is a C* homomorphism $C \rightarrow C \otimes C$ and by condition (iii)). Furthermore, $(\hbox{\input{symbols/ZbwdotSym.tex}}\!\!,\hbox{\input{symbols/DdotSym.tex}}\!\!)$ are complementary: they satisfy Hopf's law by condition (iv) and Theorem 9 froms \cite{Zeng2014}, and the associated antipode is self-inverse by condition (v). Hence $(\hbox{\input{symbols/ZbwdotSym.tex}}\!\!,\hbox{\input{symbols/DdotSym.tex}}\!\!)$ is a coherent group, with $\hbox{\input{symbols/ZbwdotSym.tex}}\!\!$ special.  
\end{proof}

The following results give a precise meaning to the idea that coherent groups can be used to embed groups in arbitrary $\dagger$-SMCs. We show that coherent groups always behave as groups on their points, and that well-pointed coherent groups provide a sound generalisation of group algebras to arbitrary $\dagger$-SMCs.
\begin{theorem}[\textbf{Underlying Group}]\label{thm_QuantumGroupAreGroupsOnPoints}\hfill\\
	Let $(\hbox{\input{symbols/ZbwdotSym.tex}}\!\!,\hbox{\input{symbols/DdotSym.tex}}\!\!)$ be a coherent group in a $\dagger$-SMC $\CategoryC$. Then the multiplication $\!\hbox{\input{symbols/DmultSym.tex}}\!\!$ and the unit $\!\hbox{\input{symbols/DunitSym.tex}}\!\!$ endow the points with the structure of a group $(\classicalStates{\hbox{\input{symbols/ZbwdotSym.tex}}\!\!}, \!\hbox{\input{symbols/DmultSym.tex}}\!\!,\!\hbox{\input{symbols/DunitSym.tex}}\!\!)$.
\end{theorem}  
\begin{proof}
The laws of strong complementarity show that $\!\hbox{\input{symbols/DunitSym.tex}}\!\! \in \classicalStates{\hbox{\input{symbols/ZbwdotSym.tex}}\!\!}$, and that $\!\hbox{\input{symbols/DmultSym.tex}}\!\!$ yields a well-defined function $\classicalStates{\hbox{\input{symbols/ZbwdotSym.tex}}\!\!} \times \classicalStates{\hbox{\input{symbols/ZbwdotSym.tex}}\!\!} \rightarrow \classicalStates{\hbox{\input{symbols/ZbwdotSym.tex}}\!\!}$ when restricted to $\hbox{\input{symbols/ZbwdotSym.tex}}\!\!$-classical states. By associative law and unit laws for $\hbox{\input{symbols/DdotSym.tex}}\!\!$, we conclude that $(\classicalStates{\hbox{\input{symbols/ZbwdotSym.tex}}\!\!},\!\hbox{\input{symbols/DmultSym.tex}}\!\!,\!\hbox{\input{symbols/DunitSym.tex}}\!\!)$ is a monoid. Furthermore, the laws of strong complementarity show that if $g \in \classicalStates{\hbox{\input{symbols/ZbwdotSym.tex}}\!\!}$ then $\hbox{\input{symbols/antipodeSym.tex}}\! \circ g \in \classicalStates{\hbox{\input{symbols/ZbwdotSym.tex}}\!\!}$, and hence the antipode $\hbox{\input{symbols/antipodeSym.tex}}\!$ yields a well-defined function $\classicalStates{\hbox{\input{symbols/ZbwdotSym.tex}}\!\!} \rightarrow \classicalStates{\hbox{\input{symbols/ZbwdotSym.tex}}\!\!}$, which by is furthermore a self-inverse bijection. Finally, Hopf's law implies that for any $\hbox{\input{symbols/ZbwdotSym.tex}}\!\!$-classical state $g \in \classicalStates{\hbox{\input{symbols/ZbwdotSym.tex}}\!\!}$ the $\hbox{\input{symbols/ZbwdotSym.tex}}\!\!$-classical state $\hbox{\input{symbols/antipodeSym.tex}}\! \circ g$ is an inverse in the monoid $(\classicalStates{\hbox{\input{symbols/ZbwdotSym.tex}}\!\!},\!\hbox{\input{symbols/DmultSym.tex}}\!\!,\!\hbox{\input{symbols/DunitSym.tex}}\!\!)$, which is therefore a group.
\end{proof}

\begin{theorem}[\textbf{Underlying Homomorphisms}]\label{thm_QuantumGroupHomsAreGroupHomsOnPoints}\hfill\\
Let $\mathbb{G} := (\hbox{\input{symbols/ZbwdotSym.tex}}\!\!,\hbox{\input{symbols/DdotSym.tex}}\!\!)$ and $\mathbb{H} := (\hbox{\input{symbols/YbwdotSym.tex}}\!\!,\hbox{\input{symbols/WbwdotSym.tex}}\!)$ be coherent groups on objects $\SpaceH$ and $\SpaceK$ of a $\dagger$-SMC. A coherent group homomorphism $f:\mathbb{G} \rightarrow \mathbb{H}$ gives rise to a well-defined group homomorphism $f: (\classicalStates{\hbox{\input{symbols/ZbwdotSym.tex}}\!\!},\!\hbox{\input{symbols/DmultSym.tex}}\!\!,\!\hbox{\input{symbols/DunitSym.tex}}\!\!) \rightarrow (\classicalStates{\hbox{\input{symbols/YbwdotSym.tex}}\!\!},\hbox{\input{symbols/WbwmultSym.tex}}\!,\hbox{\input{symbols/WbwunitSym.tex}}\!)$ when restricted to the points of $\mathbb{G}$. Furthermore, when $\mathbb{G}$ is well-pointed, any morphism $f: \SpaceH \rightarrow \SpaceK$ which gives rise to a well-defined group homomorphism $f: (\classicalStates{\hbox{\input{symbols/ZbwdotSym.tex}}\!\!},\!\hbox{\input{symbols/DmultSym.tex}}\!\!,\!\hbox{\input{symbols/DunitSym.tex}}\!\!) \rightarrow (\classicalStates{\hbox{\input{symbols/YbwdotSym.tex}}\!\!},\hbox{\input{symbols/WbwmultSym.tex}}\!,\hbox{\input{symbols/WbwunitSym.tex}}\!)$ is a coherent group homomorphism $f:\mathbb{G} \rightarrow \mathbb{H}$.
\end{theorem}
\begin{proof}
For the first part, consider a coherent group homomorphism $f: \mathbb{G} \rightarrow \mathbb{H}$. The fact that $f$ gives rise to a well-defined function $f: \classicalStates{\hbox{\input{symbols/ZbwdotSym.tex}}\!\!} \rightarrow \classicalStates{\hbox{\input{symbols/YbwdotSym.tex}}\!\!}$ is due to the fact that $f$ is a $\hbox{\input{symbols/ZbwdotSym.tex}}\!\!$-to-$\hbox{\input{symbols/YbwdotSym.tex}}\!\!$ classical function as part of the definition of coherent group homomorphism. The fact that $f$ is a group homomorphism $f: (\classicalStates{\hbox{\input{symbols/ZbwdotSym.tex}}\!\!},\!\hbox{\input{symbols/DmultSym.tex}}\!\!,\!\hbox{\input{symbols/DunitSym.tex}}\!\!) \rightarrow (\classicalStates{\hbox{\input{symbols/YbwdotSym.tex}}\!\!},\hbox{\input{symbols/WbwmultSym.tex}}\!,\hbox{\input{symbols/WbwunitSym.tex}}\!)$ follows from the fact that it is a monoid homomorphism as part of the definition of coherent group homomorphism.
Conversely, assume that $\mathbb{G}$ is well-pointed, and consider a morphism $f: \SpaceH \rightarrow \SpaceK$ which gives rise to a well-defined group homomorphism $f: (\classicalStates{\hbox{\input{symbols/ZbwdotSym.tex}}\!\!},\!\hbox{\input{symbols/DmultSym.tex}}\!\!,\!\hbox{\input{symbols/DunitSym.tex}}\!\!) \rightarrow (\classicalStates{\hbox{\input{symbols/YbwdotSym.tex}}\!\!},\hbox{\input{symbols/WbwmultSym.tex}}\!,\hbox{\input{symbols/WbwunitSym.tex}}\!)$. Because it is a function from the $\hbox{\input{symbols/ZbwdotSym.tex}}\!\!$-classical states to the $\hbox{\input{symbols/YbwdotSym.tex}}\!\!$-classical states, by well-pointedness we conclude that $f$ is $\hbox{\input{symbols/ZbwdotSym.tex}}\!\!$-to-$\hbox{\input{symbols/YbwdotSym.tex}}\!\!$-classical. Because it is a group homomorphism from $(\classicalStates{\hbox{\input{symbols/ZbwdotSym.tex}}\!\!},\!\hbox{\input{symbols/DmultSym.tex}}\!\!,\!\hbox{\input{symbols/DunitSym.tex}}\!\!)$ to $(\classicalStates{\hbox{\input{symbols/YbwdotSym.tex}}\!\!},\hbox{\input{symbols/WbwmultSym.tex}}\!,\hbox{\input{symbols/WbwunitSym.tex}}\!)$, it is in particular a monoid homomorphism $(\!\hbox{\input{symbols/DmultSym.tex}}\!\!,\!\hbox{\input{symbols/DunitSym.tex}}\!\!)$ to $(\hbox{\input{symbols/WbwmultSym.tex}}\!,\hbox{\input{symbols/WbwunitSym.tex}}\!)$. Hence $f$ is a coherent group homomorphism $\mathbb{G} \rightarrow \mathbb{H}$.
\end{proof}

\begin{theorem}[\textbf{Underlying Group Functor}]\label{thm_underlyingGroupFunctor}\hfill\\
Let $\CategoryC$ be a $\dagger$-SMC. The following defines a monoidal functor $\underlyingGroup{\emptyArg} : \CoherentGroupsCategory{\CategoryC} \rightarrow \GrpCategory$ from coherent groups to groups:
\begin{align}
	\underlyingGroup{(\hbox{\input{symbols/ZbwdotSym.tex}}\!\!,\hbox{\input{symbols/DdotSym.tex}}\!\!)}&:=(\classicalStates{\hbox{\input{symbols/ZbwdotSym.tex}}\!\!},\!\hbox{\input{symbols/DmultSym.tex}}\!\!,\!\hbox{\input{symbols/DunitSym.tex}}\!\!)\nonumber \\
	\underlyingGroup{f: (\hbox{\input{symbols/ZbwdotSym.tex}}\!\!,\hbox{\input{symbols/DdotSym.tex}}\!\!) \rightarrow (\hbox{\input{symbols/YbwdotSym.tex}}\!\!,\hbox{\input{symbols/WbwdotSym.tex}}\!)} &:= f \circ \emptyArg : \classicalStates{\hbox{\input{symbols/ZbwdotSym.tex}}\!\!} \rightarrow \classicalStates{\hbox{\input{symbols/YbwdotSym.tex}}\!\!}
\end{align}
We refer to $\underlyingGroup{\emptyArg}$ as the \textbf{underlying group functor}. It is faithful when restricted to well-pointed coherent groups. 
\end{theorem}
\begin{proof}
By Theorems \ref{thm_QuantumGroupAreGroupsOnPoints} and \ref{thm_QuantumGroupHomsAreGroupHomsOnPoints} above, we already know that $\underlyingGroup{\emptyArg}$ is a well defined functor: all we really need to show is that it is monoidal. Given two coherent groups $(\hbox{\input{symbols/ZbwdotSym.tex}}\!\!,\hbox{\input{symbols/DdotSym.tex}}\!\!)$ on an object $\SpaceH$ and $(\hbox{\input{symbols/YbwdotSym.tex}}\!\!,\hbox{\input{symbols/WbwdotSym.tex}}\!)$ on an object $\SpaceK$, the product coherent group on object $\SpaceH \otimes \SpaceK$ is given by $(\hbox{\input{symbols/ZbwdotSym.tex}}\!\! \otimes \hbox{\input{symbols/YbwdotSym.tex}}\!\!, \hbox{\input{symbols/DdotSym.tex}}\!\! \otimes \hbox{\input{symbols/WbwdotSym.tex}}\!)$. Hence we get the following underlying group for the product coherent group:
\begin{equation}  
\underlyingGroup{(\hbox{\input{symbols/ZbwdotSym.tex}}\!\! \otimes \hbox{\input{symbols/YbwdotSym.tex}}\!\!, \hbox{\input{symbols/DdotSym.tex}}\!\! \otimes \hbox{\input{symbols/WbwdotSym.tex}}\!)} = (\classicalStates{\hbox{\input{symbols/ZbwdotSym.tex}}\!\! \otimes \hbox{\input{symbols/YbwdotSym.tex}}\!\!}, \!\hbox{\input{symbols/DmultSym.tex}}\!\! \otimes \hbox{\input{symbols/WbwmultSym.tex}}\!, \!\hbox{\input{symbols/DunitSym.tex}}\!\! \otimes \hbox{\input{symbols/WbwunitSym.tex}}\!)
\end{equation}
If we can show that all classical states for $\hbox{\input{symbols/ZbwdotSym.tex}}\!\! \otimes \hbox{\input{symbols/YbwdotSym.tex}}\!\!$ are in the form $\psi \otimes \phi$ fo a $\hbox{\input{symbols/ZbwdotSym.tex}}\!\!$-classical state $\psi$ and a $\hbox{\input{symbols/YbwdotSym.tex}}\!\!$-classical state $\phi$, then we get $\classicalStates{\hbox{\input{symbols/ZbwdotSym.tex}}\!\! \otimes \hbox{\input{symbols/YbwdotSym.tex}}\!\!} = \classicalStates{\hbox{\input{symbols/ZbwdotSym.tex}}\!\!} \times \classicalStates{\hbox{\input{symbols/YbwdotSym.tex}}\!\!}$ and the product is monoidal as desired. The fact that the classical states of a product $\dagger$-FA are the products of the individual classical states is straightforward consequence of the following observation:
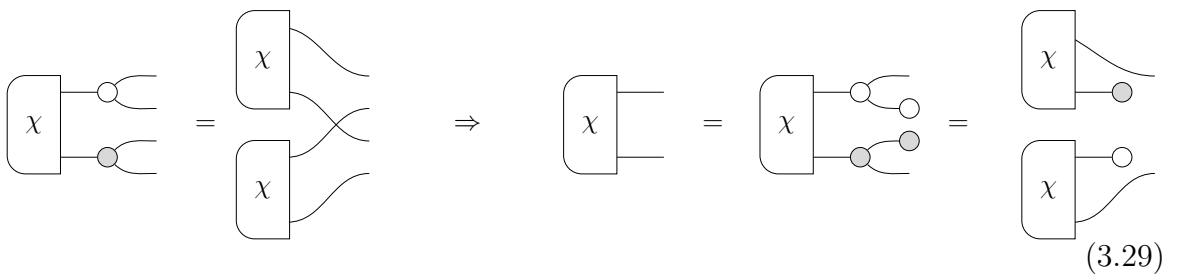
\begin{equation}\label{classicalStatesProduct}
	\resizebox{\textwidth}{!}{\input{pictures/chapter3/classicalStatesProduct.tikz}}
\end{equation}
Finally, the functor is faithful when restricted to well-pointed coherent groups because a coherent group homomorphism from a well-pointed coherent group is entirely defined by its action on the underlying group.
\end{proof}

\begin{theorem}[\textbf{Group algebras are well-pointed coherent groups}]\label{thm_WellPointedQuantumGroupAreGroupAlgebras}\hfill\\
Let $(R,^\dagger)$ be a cancellative involutive semiring such that the $\dagger$-SMC $\RMatCategory{R}$ has non-degenerate inner product. If $G$ is any finite group such that $|G| = n^\dagger n$, for some invertible $n \in R$\footnote{Semirings in which this is true for all groups $G$ include the rational, real and complex numbers.}, then the group algebra $R[G]$ in $\RMatCategory{R}$ always gives rise to a well-pointed coherent group $(\hbox{\input{symbols/ZbwdotSym.tex}}\!\!,\hbox{\input{symbols/DdotSym.tex}}\!\!)$ (which we also denote by $R[G]$, when no confusion can arise): the point structure $\hbox{\input{symbols/ZbwdotSym.tex}}\!\!$ is the $\dagger$-SCFA determined by the standard basis of the group algebra $R[G]$, and the group structure is given by the coherent group multiplication and unit on $R[G]$. Conversely, if $(\hbox{\input{symbols/ZbwdotSym.tex}}\!\!,\hbox{\input{symbols/DdotSym.tex}}\!\!)$ is a well-pointed coherent group in $\RMatCategory{R}$, then there is a coherent group isomorphism $(\hbox{\input{symbols/ZbwdotSym.tex}}\!\!,\hbox{\input{symbols/DdotSym.tex}}\!\!) \isom R[G]$ for some $G$.
\end{theorem}
\begin{proof}
Consider a group algebra $R[G]$, with standard orthonormal basis $\big(\ket{g}\big)_{g \in G}$. A $\dagger$-SCFA with the standard basis as its set of classical states is given by the following $R$-linear maps, defined on the standard basis:
\begin{align}
\!\hbox{\input{symbols/ZbwcomultSym.tex}}\!\! &:= \ket{g} \mapsto \ket{g} \otimes \ket{g} \nonumber \\
\!\hbox{\input{symbols/ZbwcounitSym.tex}}\!\! &:= \ket{g} \mapsto 1
\end{align}
If we take $(\!\hbox{\input{symbols/DmultSym.tex}}\!\!,\!\hbox{\input{symbols/DunitSym.tex}}\!\!)$ to be the coherent group multiplication and unit for the group algebra $R[G]$, then Frobenius law, the quasi-speciality law, the laws of complementarity and strong complementarity for the pair $(\hbox{\input{symbols/ZbwdotSym.tex}}\!\!,\hbox{\input{symbols/DdotSym.tex}}\!\!)$ all hold when applied to $\hbox{\input{symbols/ZbwdotSym.tex}}\!\!$-classical states. Because $\hbox{\input{symbols/ZbwdotSym.tex}}\!\!$ has enough classical states, Frobenius law and the quasi-speciality law hold in generality, making $(\!\hbox{\input{symbols/DmultSym.tex}}\!\!,\!\hbox{\input{symbols/DunitSym.tex}}\!\!,\!\hbox{\input{symbols/DcomultSym.tex}}\!\!,\!\hbox{\input{symbols/DcounitSym.tex}}\!\!)$ a $\dagger$-qSFA, and so do the laws of complementarity and strong complementarity, making $(\hbox{\input{symbols/ZbwdotSym.tex}}\!\!,\hbox{\input{symbols/DdotSym.tex}}\!\!)$ a well-pointed coherent group. 
Conversely, take a well-pointed coherent group $(\hbox{\input{symbols/ZbwdotSym.tex}}\!\!,\hbox{\input{symbols/DdotSym.tex}}\!\!)$ on an object $A$ of $\RMatCategory{R}$, and let $G := (\classicalStates{\hbox{\input{symbols/ZbwdotSym.tex}}\!\!},\!\hbox{\input{symbols/DmultSym.tex}}\!\!,\!\hbox{\input{symbols/DunitSym.tex}}\!\!)$ by the group given by Theorem \ref{thm_QuantumGroupAreGroupsOnPoints}. Because the semiring is cancellative and the inner product is non-degenerate, the points of the coherent group form an orthogonal basis for $A$, and have all the same square norm $M$ (where $M=m^\dagger m$ is the normalisation factor of the point structure). In particular, there can only be finitely many points, so $R[G]$ is a well-defined group algebra in $\RMatCategory{R}$ and $\id{A} = \frac{1}{N}\sum_{g \in \classicalStates{\hbox{\input{symbols/ZbwdotSym.tex}}\!\!}} g \circ g^\dagger$. Define the following linear map (where $g^\dagger: A \rightarrow \tensorUnit$ and $\ket{g}: \tensorUnit \rightarrow R[G]$):
\begin{equation}
U := \frac{1}{n} \sum_{g \in \classicalStates{\hbox{\input{symbols/ZbwdotSym.tex}}\!\!}} \ket{g} \circ g^\dagger : A \rightarrow R[G]
\end{equation}
Because $\id{A} = \frac{1}{N}\sum_{g \in \classicalStates{\hbox{\input{symbols/ZbwdotSym.tex}}\!\!}} g \circ g^\dagger$ and $\id{R[G]} = \sum_{g \in \classicalStates{\hbox{\input{symbols/ZbwdotSym.tex}}\!\!}} \ket{g} \bra{g}$, the linear map $U$ is a unitary. Furthermore, it is immediate to see that $U$ is a well-defined group homomorphism when restricted to the points of $(\hbox{\input{symbols/ZbwdotSym.tex}}\!\!,\hbox{\input{symbols/DdotSym.tex}}\!\!)$,   and hence a coherent group homomorphism by Theorem \ref{thm_QuantumGroupHomsAreGroupHomsOnPoints}.  
\end{proof}

As a closing remark, it should be noted that in $\RMatCategory{R}$, e.g. for $R=\reals^+,\reals, \complexs, \rationals$, all finite groups arise from (well-pointed) coherent groups, but that this is not true in general $\dagger$-SMCs. It is furthermore true that any two well-pointed coherent groups in $\RMatCategory{R}$ corresponding to the same finite group are connected by a coherent group isomorphism, but again this need not hold in general $\dagger$-SMCs. 


\newpage
\section{Wavefunctions on a periodic lattice}
\label{section_wavefunctionsPeriodicLattice}

We turn our attention back to the treatment of wavefunctions on a periodic lattice, but this time from the abstract perspective of coherent groups on some object $\SpaceH$ of some $\dagger$-SMC $\CategoryC$. For the rest of this section, we will work with a well-pointed abelian coherent group $\mathbb{G} := (\hbox{\input{symbols/ZbwdotSym.tex}}\!\!,\hbox{\input{symbols/DdotSym.tex}}\!\!)$ having finitely many points. By Theorem \ref{thm_QuantumGroupAreGroupsOnPoints}, $ G:= (\classicalStates{\hbox{\input{symbols/ZbwdotSym.tex}}\!\!},\!\hbox{\input{symbols/DmultSym.tex}}\!\!,\!\hbox{\input{symbols/DunitSym.tex}}\!\!)$ is a finite abelian group, and we fix an isomorphism $G \isom \prod_{d=1}^D \integersMod{n_d}$ allowing us to interpret $G$ as the translation group of a $D$-dimensional lattice $\Lambda$.

\subsection{What we look for in a momentum observable}
\label{subsection_whatWeLookForInMomentum}

The identification of the position observable with the point structure $\hbox{\input{symbols/ZbwdotSym.tex}}\!\!$ shouldn't come as a surprise: if we fix a distinguished site $\lambda_0$ on the lattice $\Lambda$, then lattice sites are naturally identified with elements of $G$, i.e. points of the coherent group, by the bijection $g \mapsto g(\lambda_0)$. If $\CPStarCategory{\CategoryC}$ is $R$-probabilistic, then the following is the desired lattice \textbf{position measurement} (where $R^G$ is the space of $R$-distributions over $G$):
\begin{equation}\label{latticePositionMeasurement}
	\input{pictures/chapter3/latticePositionMeasurement.tikz}
\end{equation}
The multiplicative fragment $(\!\hbox{\input{symbols/DmultSym.tex}}\!\!,\!\hbox{\input{symbols/DunitSym.tex}}\!\!)$ of the group structure for $\mathbb{G}$ acts as the (coherent) lattice translation on the points/lattice sites, with the unit $\!\hbox{\input{symbols/DunitSym.tex}}\!\!$ corresponding to the distinguished site $\lambda_0$. The position measurement, together with the corresponding preparation, can then be used to obtain the classical translation group on the lattice\footnote{Or, to be more precise, its extension to the space of $R$-distributions over the lattice.} from the coherent one, confirming that the outcomes of the position measurement defined above are naturally endowed with lattice structure:
\begin{equation}\label{latticeTranslationGroup}
	\input{pictures/chapter3/latticeTranslationGroup.tikz}
\end{equation}

From our experience with $\fdHilbCategory$, we expect the momentum observable to arise as the group structure $\hbox{\input{symbols/DdotSym.tex}}\!\!$. But what does it mean to be the momentum observable in an abstract setting such as ours? What are the structural and operational features that would allow us to conclude, beyond reasonable doubt, that $\hbox{\input{symbols/DdotSym.tex}}\!\!$ behaves as the momentum observable? To understand this, we look at some of the defining characteristics of position and momentum observables in the traditional formulation of quantum mechanics.

Consider a wavefunctions on the periodic lattice $\Lambda$, living in the group algebra $\complexs[G]$. The translation symmetry on $\Lambda$ is encoded by a unitary group action $(U_g)_{g \in G}$ of $G$ on $\complexs[G]$, and the momentum eigenstates $\frac{1}{\sqrt{|G|}}\ket{\goodchi}$ are the states invariant under this action. Hence the momentum eigenstates generate the group action in following sense:
\begin{equation}\label{eqn_latticeTranslationGroupActionDiagonalised}
U_g := \frac{1}{|G|} \sum_{\goodchi} \goodchi(g) \ket{\goodchi} \bra{\goodchi}
\end{equation}
The momentum eigenstates come themselves with an abelian group symmetry, the \textbf{boost symmetry} on $\Lambda$, which is encoded by a unitary group action $(V_{\goodchi})_{\goodchi \in G^\wedge}$ of the Pontryagin dual $G^\wedge$ of the translation symmetry group. The position eigenstates turn out to be exactly the states invariant under boost symmetry, and generate its group action the following sense:
\begin{equation}\label{eqn_latticeBoostGroupActionDiagonalised}
V_{\goodchi} := \sum_{g \in G} \goodchi(g) \ket{g} \bra{g}
\end{equation}
We will take this as our first structural description of the relationship between the position and momentum observables: there is a symmetry-observable duality, with the momentum observable corresponding to the translation symmetry (the symmetry of position eigenstates), and the position observable corresponding to the boost symmetry (the symmetry of momentum eigenstates).

In the continuous case of wavefunctions over $\reals^{n}$ (with positions $x \in \reals^n$ and momenta indexed by $p \in \reals^n$ as $\goodchi_p := x \mapsto e^{i p x }$), the relationship between the translation and boost symmetries is fully captured by the \textbf{Weyl Canonical Commutation Relations}:
\begin{equation}\label{eqn_WeylCCRs}
  V_p U_x = e^{i \hbar p \cdot x} \; U_x V_p
\end{equation}
Note that we chose the Weyl CCRs, which refer to the translation and boost symmetries, instead of the more common Heisenberg CCRs $[\textbf{x},\textbf{p}] = i \hbar \;\id{\SpaceH}$, which refer to the infinitesimal generators $\textbf{x}$ and $\textbf{p}$ for the symmetries (usually referred to as the position and momentum observables). There are two reasons for this choice. Firstly, the Heisenberg CCRs are known not to hold in finite dimensions, for any choice of operators $\textbf{x}$ and $\textbf{p}$:
\begin{equation}
\Trace{(\textbf{xp}-\textbf{px})} = \Trace{(\textbf{xp})}-\Trace{(\textbf{px})} = 0 \neq i\hbar \dim{\SpaceH} = \Trace{(i \hbar \; \id{\SpaceH})}
\end{equation}
On the contrary, the Weyl CCRs are easily formulated in our finite-dimensional setting:
\begin{equation}\label{eqn_WeylCCRsLattice}
V_{\goodchi} U_g = \goodchi(g)\;  U_g V_{\goodchi} 
\end{equation}
Secondly, the Heisenberg CCRs focus on infinitesimal generators $\textbf{x}$ and $\textbf{p}$, which are mere mathematical constructs arising from Stone's Theorem, while the Weyl CCRs focus on the symmetries associated with position and momentum, which have direct physical significance. The issues with identifying self-adjoint operators and observables are fleshed out in full detail in Subsection \ref{subsection_briefDigressionObservables} below: the conclusion will be that it makes no sense to look for an analogue of $\textbf{x}$ and $\textbf{p}$ in our abstract framework, and as a consequence the Heisenberg form of the CCRs should not be used. We will take the Weyl CCRs, together with an appropriately revised formulation of the Stone-von Neumann Theorem, as our second structural description of the relationship between the position and momentum observables. Despite our departure from the Heisenberg CCRs, we will still be able to formulate a suitable version of Stone's Theorem (in Subsection \ref{subsection_StoneTheoremRevisited} below), further strengthening our claims.

The Uncertainty Principle is unarguably the most iconic operational feature of position and momentum in quantum mechanics. The most common formulation of the principle is due to Kennard and Weyl, and involves the standard deviations $\sigma_x$ and $\sigma_p$ for the position and momentum observables of a wavefunction on the continuous 1-dimensional space $\reals$:
\begin{equation}\label{eqn_KennardWeylUncertaintyPrinciple}
\sigma_x \sigma_p \geq \frac{\hbar}{2}
\end{equation}
There is also an entropic formulation of the principle \cite{Beckner1975,Bialynicki1975}, which involves the entropies $H_x$ and $H_p$ of the probability distributions on outcomes of position and momentum measurements of a same state:
\begin{equation}\label{eqn_entropicUncertaintyPrinciple}
H_x + H_p \geq \log(e/2)
\end{equation}
Unfortunately, neither form of the uncertainty principle is suitable for our purposes: while some of its implications truly characterise the abstract relationship between position and momentum observables, other implications, such as the necessarily ensuing non-locality, seem to pertain more to quantum mechanics in general rather than to position and momentum specifically. Indeed there are quantum-like theories which have sensible notions of position/momentum observables\footnote{By which we mean coherent groups with each observable having an orthonormal basis of classical states, so that the position/momentum measurements are well defined with non-deterministic classical outcomes.} while at the same time being entirely local: a stunning example is given by hyperbolic quantum theory, a quasi-probabilistic theory which admits non-trivial examples of position/momentum duality (e.g. for the finite periodic lattices $\integersMod{2}^N$ and for the infinite lattices $\integers^N$), but at the same time fails to satisfy either formulation of the uncertainty principle (not a surprise, since the theory is local). 

Consider a qubit in hyperbolic quantum theory. Let $\hbox{\input{symbols/ZbwdotSym.tex}}\!\!$ be the $\dagger$-SCFA associated with the Pauli Z orthonormal basis $\ket{0},\ket{1}$, and $\hbox{\input{symbols/DdotSym.tex}}\!\!$ be the $\dagger$-qSCFA associated with the Pauli X orthogonal basis $\ket{\pm} := \ket{0} \pm \ket{1}$. Then $(\hbox{\input{symbols/ZbwdotSym.tex}}\!\!,\hbox{\input{symbols/DdotSym.tex}}\!\!)$ is a coherent group corresponding to the position/momentum pair for a $1$-dimensional periodic lattice with points $\integersMod{2}$: the position eigenstates $\ket{0},\ket{1}$ are unbiased for the momentum measurement, and conversely the (normalised) momentum eigenstates $\frac{1}{\sqrt{2}}\ket{+}, \frac{1}{\sqrt{2}}\ket{-}$ are unbiased for the position measurement. However, both the Kennard-Weyl and the entropic uncertainty principles fail.
\begin{theorem}[\textbf{Hyperbolic quantum theory fails the uncertainty principle}]
\label{thm_hyperbolicUncertaintyPrincipleFails}
There is a mixed state $\rho$ of the qubit in hyperbolic quantum theory which gives outcome $\ket{0}$ with certainty when measured in the Pauli Z observable, and outcome $\ket{+}$ with certainty when measured in the Pauli X observable.
\end{theorem}
\begin{proof}
Let $a := \sqrt{2}$ and $b := \frac{1}{\sqrt{2}}+j \frac{\sqrt{3}}{\sqrt{2}}$, and consider the pure qubit state $\ket{\psi} := a \ket{0} + b \ket{1}$, which is normalised: 
\begin{equation}
\braket{\psi}{\psi} = |a|^2 + |b|^2 = 2 + (\frac{1}{2} - \frac{3}{2}) = 1
\end{equation} 
Now consider the normalised mixed state $\rho := \frac{1}{2} \ket{\psi}\bra{\psi} + \frac{1}{2}\ket{1}\bra{1}$. Upon measurement in the Pauli Z observable, the state $\rho$ results in the outcome $\ket{0}$ with certainty:
\begin{align}
	\bra{0} \rho \ket{0} &= \frac{1}{2} |\braket{0}{\psi}|^2 + \frac{1}{2} |\braket{0}{1}|^2 = \frac{1}{2} \cdot 2 + \frac{1}{2}\cdot 0 = 1 \nonumber\\
	\bra{1} \rho \ket{1} &= \frac{1}{2} |\braket{1}{\psi}|^2 + \frac{1}{2} |\braket{1}{1}|^2 = \frac{1}{2} \cdot (-1) + \frac{1}{2}\cdot 1 = 0 
\end{align}
Upon measurement in the Pauli X observable, the state $\rho$ also results in the outcome $\ket{+}$ with certainty:
\begin{align}
	\frac{1}{2} \bra{+} \rho \ket{+} &= \frac{1}{4} |\braket{+}{\psi}|^2 + \frac{1}{4} |\braket{+}{1}|^2 \nonumber\\
	&= \frac{1}{4} (|a|^2 + ab^\ast + a^\ast b + |b|^2) + \frac{1}{4}\cdot 1 = \frac{1}{4} \big((2 + 2\frac{\sqrt{2}}{\sqrt{2}} -1)+ 1\big) = 1 \nonumber \\
	\frac{1}{2} \bra{-} \rho \ket{-} &= \frac{1}{4} |\braket{-}{\psi}|^2 + \frac{1}{4} |\braket{-}{1}|^2 \nonumber\\
	&= \frac{1}{4} (|a|^2 - ab^\ast - a^\ast b + |b|^2) + \frac{1}{4}\cdot 1 = \frac{1}{4} \big((2 - 2\frac{\sqrt{2}}{\sqrt{2}}-1) + 1\big) = 0
\end{align}
Hence the state $\rho$ is sharp in both the Pauli Z (lattice position) and Pauli X (lattice momentum) observables, and the conventional form of the uncertainty principle necessarily fails. 
\end{proof}
The example of hyperbolic quantum theory is not isolated: there are many other theories admitting sensible notions of position and momentum observables (finite-field and $p$-adic quantum theories another examples), but which at the same time are local, and hence necessarily fail to satisfy either formulation of the uncertainty principle. As a consequence, we will choose to take a restricted version of the uncertainty principle as our third structural description of the relationship between the position and momentum observables: namely, that states of definite position have completely indeterminate momentum, and vice versa that states of definite momentum have completely indeterminate position.

\subsection{Momenta generate translation symmetry}
\label{subsection_momentumTranslationSymmetry}

In the first part of this Section, we have set out a number of structural and operational criteria that would help identifying $\hbox{\input{symbols/DdotSym.tex}}\!\!$ with the momentum observable: (i) the relationship between the position/momentum observables and the translation/boost symmetries; (ii) the Weyl canonical commutation relations; (iii) the (restricted version of the) uncertainty principle. In the remainder of this Section we will prove that the observable $\hbox{\input{symbols/DdotSym.tex}}\!\!$ satisfies all those criteria, and conclude that it is indeed the lattice momentum observable we were looking for.

We begin by observing that the unitary translation symmetry action $(U_g)_{g \in G}$ on $\SpaceH$ is obtained by evaluating $\!\hbox{\input{symbols/DmultSym.tex}}\!\!$ on the points of the coherent group.
\begin{theorem}[\textbf{Momenta generate translations}]\label{thm_latticeTranslationSymmetryAction}\hfill\\
	Let $\mathbb{G} := (\hbox{\input{symbols/ZbwdotSym.tex}}\!\!,\hbox{\input{symbols/DdotSym.tex}}\!\!)$ be a coherent group on an object $\SpaceH$ of a $\dagger$-SMC $\CategoryC$, and let $G:=(\classicalStates{\hbox{\input{symbols/ZbwdotSym.tex}}\!\!},\!\hbox{\input{symbols/DmultSym.tex}}\!\!,\!\hbox{\input{symbols/DunitSym.tex}}\!\!)$. Define a family $(U_g)_{g \in G}$ of endomorphisms of $\SpaceH$ as follows: 
	\begin{equation}\label{latticeTranslationSymmetryAction}
		\input{pictures/chapter3/latticeTranslationSymmetryAction.tikz}
	\end{equation}
	Then $(U_g)_{g \in G}$ gives a unitary action of the group $G$ on $\SpaceH$, restricting to the left regular action of the translation group $G$ on the points of $\mathbb{G}$.
\end{theorem}
\begin{proof}
	To show that this defines a unitary group action, we need to check that $U_{g\oplus h} = U_h U_g$, that $U_0 = \id{\SpaceH}$, and that $U_g^\dagger U_g= \id{\SpaceH} = U_g U_h^\dagger$. The first claim follows by the associative law for $\hbox{\input{symbols/DdotSym.tex}}\!\!$, together with the fact that it acts as the group $G$ on the points of the coherent group: 
	\begin{equation}\label{latticeTranslationSymmetryActionProof1}
		\input{pictures/chapter3/latticeTranslationSymmetryActionProof1.tikz}
	\end{equation}
	The second claim follows from the unit law for $\hbox{\input{symbols/DdotSym.tex}}\!\!$:
	\begin{equation}\label{latticeTranslationSymmetryActionProof2}
		\input{pictures/chapter3/latticeTranslationSymmetryActionProof2.tikz}
	\end{equation}
	The third claim has a slightly more complicated proof, which involves Frobenius and unit laws for $\hbox{\input{symbols/DdotSym.tex}}\!\!$, the adjoin condition for $\hbox{\input{symbols/ZbwdotSym.tex}}\!\!$-classical states, and Hopf's law:
	\begin{equation}\label{latticeTranslationSymmetryActionProof3}
		\resizebox{\textwidth}{!}{\input{pictures/chapter3/latticeTranslationSymmetryActionProof3.tikz}}
	\end{equation}
	The proof that $U_g U_g^\dagger = \id{\SpaceH}$ goes along the same lines. 
\end{proof}
We deduce that, when $\CPStarCategory{\CategoryC}$ is $R$-probabilistic and $\mathbb{G}$ is well-pointed, the following defines the controlled unitary corresponding to the symmetry group action:
\begin{equation}\label{latticeTranslationSymmetryControlledUnitary}
	\input{pictures/chapter3/latticeTranslationSymmetryControlledUnitary.tikz}
\end{equation}

We define the \textbf{multiplicative characters} for the coherent group $\mathbb{G}$ to be the effects $\goodchi: \SpaceH \rightarrow \tensorUnit$ satisfying the following three equations:
\begin{equation}\label{multiplicativeChar}
	\input{pictures/chapter3/multiplicativeChar.tikz}
\end{equation}
When restricted to the points of $\mathbb{G}$, the three equation above mimic the defining properties of multiplicative characters for classical groups: $\goodchi(g\oplus h) = \goodchi(g) \goodchi(h)$, $\goodchi(0) = 1$, and $\goodchi(\ominus g)=\big(\goodchi(g)\big)^\dagger$ (where the last one made use of the adjoint condition for $\hbox{\input{symbols/ZbwdotSym.tex}}\!\!$-classical states); indeed, we already encountered them at the beginning of the chapter. It is easy to show that, just like in $\fdHilbCategory$, the multiplicative characters for $\mathbb{G}$ are exactly the adjoints of the $\hbox{\input{symbols/DdotSym.tex}}\!\!$-classical states (which we wish to interpret as momentum eigenstates): the first two equations in \ref{multiplicativeChar} correspond to the copy and delete conditions for the state $\goodchi^\dagger$, while the third one can be easily turned into the adjoin condition by applying the symmetric cup corresponding to $\hbox{\input{symbols/ZbwdotSym.tex}}\!\!$ (and recalling the definition of the antipode).

We now show that the $\hbox{\input{symbols/DdotSym.tex}}\!\!$-classical states are exactly the states invariant under the translation symmetry action on $\SpaceH$.
\begin{theorem}[\textbf{Momenta invariant under translation}]\label{thm_multiplicativeCharInvariance} \hfill\\
	Let $\mathbb{G} := (\hbox{\input{symbols/ZbwdotSym.tex}}\!\!,\hbox{\input{symbols/DdotSym.tex}}\!\!)$ be a coherent group on an object $\SpaceH$ of a $\dagger$-SMC $\CategoryC$, and let $G:=(\classicalStates{\hbox{\input{symbols/ZbwdotSym.tex}}\!\!},\!\hbox{\input{symbols/DmultSym.tex}}\!\!,\!\hbox{\input{symbols/DunitSym.tex}}\!\!)$. The adjoints of the multiplicative characters (i.e. the $\hbox{\input{symbols/DdotSym.tex}}\!\!$-classical states) are invariant under the translation symmetry action, up to a scalar:
	\begin{equation}\label{multiplicativeCharInvariance}
		\input{pictures/chapter3/multiplicativeCharInvariance.tikz}
	\end{equation}
	Furthermore, the scalar $\goodchi(g) := \goodchi \circ g$ satisfies $\goodchi(g)^\dagger \cdot \goodchi(g) = 1$. Finally, assume that $\mathbb{G}$ is well-pointed, and consider a state $\goodchi^\dagger$: if both (i) Equation \ref{multiplicativeCharInvariance} holds for all points $g$, and (ii) $\goodchi(\!\hbox{\input{symbols/DunitSym.tex}}\!\!)^\dagger \cdot \goodchi(\!\hbox{\input{symbols/DunitSym.tex}}\!\!) = 1$, then $\goodchi^\dagger$ must be a $\hbox{\input{symbols/DdotSym.tex}}\!\!$-classical state.
\end{theorem}
\begin{proof}
First we show that the adjoint $\goodchi^\dagger$ of a multiplicative character $\goodchi$ of the coherent group satisfies Equation \ref{multiplicativeCharInvariance}:
\begin{equation}\label{multiplicativeCharInvarianceProof}
	\input{pictures/chapter3/multiplicativeCharInvarianceProof.tikz}
\end{equation}
Then we show that the scalar $\goodchi(g)$ satisfies $\goodchi(g)^\dagger \cdot \goodchi(g) = 1$:
\begin{equation}\label{multiplicativeCharInvarianceProofScalar}
	\resizebox{\textwidth}{!}{\input{pictures/chapter3/multiplicativeCharInvarianceProofScalar.tikz}}
\end{equation}
Finally, assume that $\goodchi^\dagger$ is some state satisfying Equation \ref{multiplicativeCharInvariance}. Then we can quickly derive following three equations:
\begin{equation}\label{multiplicativeCharInvarianceProofConverse}
	\resizebox{\textwidth}{!}{\input{pictures/chapter3/multiplicativeCharInvarianceProofConverse.tikz}}
\end{equation}
The rightmost equation together with the requirement that $\goodchi(\!\hbox{\input{symbols/DunitSym.tex}}\!\!)^\dagger \cdot \goodchi(\!\hbox{\input{symbols/DunitSym.tex}}\!\!) = 1$ implies the delete condition for $\hbox{\input{symbols/DdotSym.tex}}\!\!$-classical states. Because $\mathbb{G}$ is well-pointed, the middle equation together with the delete condition implies the adjoin condition for $\hbox{\input{symbols/DdotSym.tex}}\!\!$-classical states, while the left equation implies the copy condition.
\end{proof}
We deduce that, when $\CPStarCategory{\CategoryC}$ is $R$-probabilistic, the adjoints of the multiplicative characters are invariant under the controlled unitary associated with the translation symmetry action:
\begin{equation}\label{multiplicativeCharInvarianceCPStar}
	\input{pictures/chapter3/multiplicativeCharInvarianceCPStar.tikz}
\end{equation}

Having checked that the $\hbox{\input{symbols/DdotSym.tex}}\!\!$-classical states are the invariant states for the translation symmetry action, the next task on our list is to show that they actually generate the symmetry action itself. In the quantum mechanics of wavefunctions on $\reals$, the self-adjoint momentum operator $\textbf{p}$ is traditionally obtained from the translation symmetry action $(U_x)_{x \in \reals}$ by using Stone's theorem on 1-parameter unitary groups \cite{Stone1930,Stone1932}:
\begin{equation}
U_x := e^{ix\textbf{p}} = \sum_{p} e^{i x p } \ket{p} \bra{p}
\end{equation}
We have already seen that taking infinitesimal generators does not yield a well defined self-adjoint momentum operator when working on periodic lattices, not even in the traditional quantum mechanical formalism. Hence we will aim for something like Equation \ref{eqn_latticeTranslationGroupActionDiagonalised}:
\begin{equation}\nonumber 
U_g := \frac{1}{|G|} \sum_{\goodchi} \goodchi(g) \ket{\goodchi} \bra{\goodchi}
\end{equation}
Unfortunately, we don't have the luxury of sums. We need to look for a more structural way of phrasing Equation \ref{eqn_latticeTranslationGroupActionDiagonalised}, one which we can formulate, and hopefully prove, in our more abstract framework. As it turns out, Theorem \ref{thm_latticeTranslationSymmetryAction} already provides us with such an alternative phrasing, since in quantum mechanics the following is true:
\begin{equation}\label{latticeTranslationSymmetryActionExplicit}
	\input{pictures/chapter3/latticeTranslationSymmetryActionExplicit.tikz}
\end{equation}
As a consequence, Theorem \ref{thm_latticeTranslationSymmetryAction} already proved that the momentum observable, seen as the $\dagger$-qSFA $\hbox{\input{symbols/DdotSym.tex}}\!\!$, generates the translation symmetry action.

\subsection{Positions generate boost symmetry}
\label{subsection_positionsBoostSymmetry}

Having completed our description of the connection between $\hbox{\input{symbols/DdotSym.tex}}\!\!$, our candidate momentum observable, and the translation symmetry action, we now characterise the dual relationship between the position observable, embodied by $\hbox{\input{symbols/ZbwdotSym.tex}}\!\!$, and the boost symmetry action. To begin with, we need to formalise what we mean by boost symmetry. 

In the quantum mechanical case of $\complexs[G]$, momentum eigenstates $\ket{\goodchi}$ correspond to multiplicative characters $\goodchi: G \rightarrow \complexs$. Any multiplicative character $\goodchi$ can be written as $\goodchi_h$ in the following form, for some (non-unique) $h = (h_d)_{d=1}^D \in G \isom \prod_{d=1}^{D} \integersMod{n_d}$:
\begin{equation}\label{eqn_multCharBoostForm}
\goodchi_h = g \mapsto e^{ 2 \pi \, i \;  g \cdot h}
\end{equation}
where $g \cdot h := \sum_{d=1}^D\frac{g_d \cdot h_d}{n_d}$, and each product $g_d \cdot h_d$ is taken modulo $n_d$. An element $h = (h_d)_{d=1}^D \in G$ is a vector, describing some direction and magnitude on the lattice: as a consequence we can think of $\ket{\goodchi}$ as the momentum eigenstate associated with \inlineQuote{moving with momentum $h$ on the lattice}. Contrary to the case of wavefunctions over $\reals$, there is no canonical choice for $h$ in the periodic lattice case, and hence it is best to work directly with the multiplicative character $\goodchi$. 

We've already seen that the multiplicative characters $\goodchi: G \rightarrow \complexs$ of the finite abelian group $G$ form the Pontryagin dual group $G^\wedge$ under pointwise multiplication, but what does this have to do with boosts? To see exactly what's going on, we need to observe the effect of pointwise multiplication $\goodchi_{h} \cdot \goodchi_{\delta h}$ on the explicit form given by Equation \ref{eqn_multCharBoostForm}:
\begin{equation}
\goodchi_{h} \cdot \goodchi_{\delta h} = g \mapsto e^{ 2 \pi \, i \;  g \cdot h} e^{ 2 \pi \, i \;  g \cdot \delta h} = e^{ 2 \pi \, i \;  g \cdot (h \oplus \delta h)}
\end{equation}
Hence, if we interpret $\goodchi_h$ as moving with momentum $h$ on the lattice, and $\goodchi_{\delta h}$ as moving with momentum $\delta h$ on the lattice, then $\goodchi_h \cdot \goodchi_{\delta h}$ should be interpreted as moving with momentum $h \oplus \delta h$ on the lattice. This is why the Pontryagin dual group structure on the multiplicative characters is referred to as the \textbf{boost symmetry}. 

Perhaps unsurprisingly, the Pontryagin duality between the translation symmetry group and the boost symmetry group in quantum mechanics lifts to a duality between corresponding coherent groups. Perhaps more surprisingly, the duality on the coherent group side has a much simpler characterisation than Pontryagin duality. Theorem \ref{thm_DualQuantumGroup} introduces the notion of dual coherent group, while Theorem \ref{lem_QGDualityPontryaginDuality} gives a first bout of legitimacy to the idea that duality of coherent groups generalises Pontryagin duality (the two coincide in the case of well-pointed abelian coherent groups in $\fdHilbCategory$). 

\begin{theorem}[\textbf{Dual coherent group}]\label{thm_DualQuantumGroup}\hfill\\
Let $\mathbb{G}:=(\hbox{\input{symbols/ZbwdotSym.tex}}\!\!,\hbox{\input{symbols/DdotSym.tex}}\!\!)$ be a coherent group on an object $\SpaceH$ of a $\dagger$-SMC $\CategoryC$. Then the \textbf{dual coherent group} $\mathbb{G}^\wedge:=(\hbox{\input{symbols/DdotSym.tex}}\!\!,\hbox{\input{symbols/ZbwdotSym.tex}}\!\!)$ is also a coherent group, with the adjoints $\goodchi^\dagger$ of the multiplicative character of $\mathbb{G}$ as its points. The multiplicative structure $\hbox{\input{symbols/ZbwdotSym.tex}}\!\!$ of $\mathbb{G}^\wedge$ acts by pointwise composition:
\begin{equation}\label{DualQuantumGroupArePointMultGroups}
	\input{pictures/chapter3/DualQuantumGroupArePointMultGroups.tikz}
\end{equation}
Finally, the following defines an involutive monoidal functor $\wedge:  \CoherentGroupsCategory{\CategoryC} \rightarrow  \OpCategory{\CoherentGroupsCategory{\CategoryC}}$ on coherent groups in a $\dagger$-SMC $\CategoryC$:
\begin{align}
	(\hbox{\input{symbols/ZbwdotSym.tex}}\!\!,\hbox{\input{symbols/DdotSym.tex}}\!\!)^\wedge &:= (\hbox{\input{symbols/DdotSym.tex}}\!\!,\hbox{\input{symbols/ZbwdotSym.tex}}\!\!) \nonumber \\
	f^\wedge & := f^\dagger
\end{align}
\end{theorem}
\begin{proof}
The conditions defining a coherent group $\mathbb{G}:=(\hbox{\input{symbols/ZbwdotSym.tex}}\!\!,\hbox{\input{symbols/DdotSym.tex}}\!\!)$ are symmetric in $\hbox{\input{symbols/ZbwdotSym.tex}}\!\!$ and $\hbox{\input{symbols/DdotSym.tex}}\!\!$ (see Definition \ref{def_strongComplementarity} and Remark \ref{rmrk_strongComplementarityColorSwappedEqns}), so $\mathbb{G}^\wedge:=(\hbox{\input{symbols/DdotSym.tex}}\!\!,\hbox{\input{symbols/ZbwdotSym.tex}}\!\!)$ satisfies the same conditions and is a coherent group. We have also already observed that the $\hbox{\input{symbols/DdotSym.tex}}\!\!$-classical states coincide with the multiplicative characters of $\mathbb{G}$, and the pointwise multiplication action of $\hbox{\input{symbols/ZbwdotSym.tex}}\!\!$ on them is a consequence of the copy condition for the points of $\mathbb{G}$, which are $\hbox{\input{symbols/ZbwdotSym.tex}}\!\!$-classical states. 
We now need to show that $^\wedge$ is an involutive monoidal functor. It is definitely involutive and monoidal, as long as we can show that it is a well-defined functor. We just showed that if $(\hbox{\input{symbols/ZbwdotSym.tex}}\!\!,\hbox{\input{symbols/DdotSym.tex}}\!\!)$ is a coherent group, then so is its dual $(\hbox{\input{symbols/DdotSym.tex}}\!\!,\hbox{\input{symbols/ZbwdotSym.tex}}\!\!)$, so the functor $^\wedge$ is well-defined on objects. Because the dagger is a functor, all we need to show is that if $f: (\hbox{\input{symbols/ZbwdotSym.tex}}\!\!,\hbox{\input{symbols/DdotSym.tex}}\!\!) \rightarrow (\hbox{\input{symbols/YbwdotSym.tex}}\!\!,\hbox{\input{symbols/WbwdotSym.tex}}\!)$ is a coherent group homomorphism, then $f^\dagger: (\hbox{\input{symbols/WbwdotSym.tex}}\!,\hbox{\input{symbols/YbwdotSym.tex}}\!\!) \rightarrow (\hbox{\input{symbols/DdotSym.tex}}\!\!,\hbox{\input{symbols/ZbwdotSym.tex}}\!\!)$ is also a coherent group homomorphism. Below we present the six equations that define $f$ as a coherent group homomorphism:
\begin{equation}\label{quantumGroupHomomorphism}
	\resizebox{\textwidth}{!}{\input{pictures/chapter3/quantumGroupHomomorphism.tikz}}
\end{equation}
Taking adjoints of the two leftmost equations for $f$ yields the two leftmost equations for $f^\dagger$, and similarly taking adjoints of the two rightmost equations for $f$ implies the two rightmost equations for $f^\dagger$. All we need to show is that the central two equations for $f$ imply the central two equations for $f^\dagger$. First we prove the adjoint condition for $f^\dagger$, using the central two equations for $f$:
\begin{equation}\label{quantumGroupDualHomomorphismProof1}
	\input{pictures/chapter3/quantumGroupDualHomomorphismProof1.tikz}
\end{equation}
Then we use the central two equations for $f$ and the top central equation for $f^\dagger$ we have just obtained to prove the bottom central equation for $f^\dagger$ (recall that the antipode is self-inverse):
\begin{equation}\label{quantumGroupDualHomomorphismProof2}
	\input{pictures/chapter3/quantumGroupDualHomomorphismProof2.tikz}
\end{equation}
This concludes the proof, showing that $^\wedge$ is an involutive monoidal functor.
\end{proof}

\noindent We will often refer to the underlying group $\underlyingGroup{\mathbb{G}^\wedge} = (\classicalStates{\hbox{\input{symbols/DdotSym.tex}}\!\!}, \!\hbox{\input{symbols/ZbwmultSym.tex}}\!\!, \!\hbox{\input{symbols/ZbwunitSym.tex}}\!\!)$ in terms of multiplicative characters of $\mathbb{G}$, in which case we will write its operation as $\cdot$ (\textit{pointwise multiplication}) and its unit as $\mathbb{1}$ (the \textit{trivial character}). 

\begin{lemma}\label{lem_QGDualityPontryaginDuality}
Let $\WpAbQuantumGroupsCategory{\fdHilbCategory}$ be the category of well-pointed, abelian coherent groups in $\fdHilbCategory$, a full subcategory of $\CoherentGroupsCategory{\fdHilbCategory}$. Then the functor $\underlyingGroup{\emptyArg}: \WpAbQuantumGroupsCategory{\fdHilbCategory} \rightarrow \fAbGrpCategory$ is well-defined, induces an equivalence of categories, and makes the following diagram commute:
\begin{equation}\label{pontryaginDualityCommutingSquare}
	\input{pictures/chapter3/pontryaginDualityCommutingSquare.tikz}
\end{equation}
This shows that Pontryagin duality $^\wedge$ for finite abelian groups corresponds exactly to the duality $^\wedge$ on well-pointed abelian coherent groups in $\fdHilbCategory$.
\end{lemma}
\begin{proof}
The underlying group functor $\underlyingGroup{\emptyArg}: \WpAbQuantumGroupsCategory{\fdHilbCategory} \rightarrow \fAbGrpCategory$ is full, faithful and essentially surjective because of Theorem \ref{thm_WellPointedQuantumGroupAreGroupAlgebras}. We've already observed that in $\fdHilbCategory$ the multiplicative characters of an abelian coherent group $(\hbox{\input{symbols/ZbwdotSym.tex}}\!\!,\hbox{\input{symbols/DdotSym.tex}}\!\!)$ are the multiplicative characters of the underlying group: hence $\underlyingGroup{\mathbb{G}^\wedge} = \underlyingGroup{\mathbb{G}}^{\wedge}$ in $\fdHilbCategory$, and the diagram commutes on objects. To see that it also commutes on morphisms, consider a coherent group homomorphism $f$ with $\underlyingGroup{f}: G \rightarrow H$ for some finite abelian groups $G,H$. Then we have that both $\underlyingGroup{f^\wedge}^{op}$ and $\underlyingGroup{f}^\wedge$ are morphisms $H^\wedge \rightarrow G^\wedge$, with the following explicit expressions:
\begin{align}
\underlyingGroup{f^\wedge}^{op} &= \goodchi \mapsto (f^\wedge \circ \goodchi^\dagger)^\dagger = \goodchi \circ f \nonumber \\
\underlyingGroup{f}^\wedge &= \goodchi \mapsto \goodchi \circ f
\end{align}
The two expressions coincide, showing that Diagram \ref{pontryaginDualityCommutingSquare} also commutes on morphisms. 
\end{proof}

\begin{remark}
We might try to generalise the Pontryagin duality side of Theorem \ref{lem_QGDualityPontryaginDuality} as follows. Given a group $G$ and another group $K$, we can define the \textbf{dual} $\dualGroupWRTMonoid{G}{K}$ of $G$ with respect to $K$ to be the group\footnote{Abelian when either one of $K$ or $G$ is commutative.} of homomorphisms $G \rightarrow K$ (the \textbf{$K$-valued multiplicative characters}) under pointwise multiplication (the inverse of $\goodchi: G \rightarrow K$ is given by $g \mapsto \goodchi(g^{-1})$). When $K = S^1$ and $G$ is abelian, $\dualGroupWRTMonoid{G}{S^1}$ is the usual Pontryagin dual of $G$. We can turn $^{\wedge_{M}}$ into a functor $^{\wedge_{M}}: \GrpCategory \rightarrow \GrpCategory$ by setting $f^{\wedge_{M}}:= \emptyArg \circ f$, exactly as in the usual $M = S^1$ case. 

Now consider coherent groups in a $\dagger$-SMC $\CategoryC$, and let $K$ be the group of \textbf{units} in $\CategoryC$, i.e. those scalars $x$ such that $x^\dagger x = 1$. The Diagram \ref{pontryaginDualityCommutingSquare}, where Pontryagin duality $^\wedge: \fAbGrpCategory \rightarrow \OpCategory{\fAbGrpCategory}$ is replaced by $^{\wedge_{K}}: \GrpCategory \rightarrow \GrpCategory$, still commutes on morphisms: we have $\underlyingGroup{f^\wedge} = \underlyingGroup{f}^{\wedge_{K}}$ for all coherent group homomrphisms $f$. However, the new Diagram need not commute on objects: we always have that $\underlyingGroup{\mathbb{G}^\wedge} \leq \underlyingGroup{\mathbb{G}}^{\wedge_{K}}$, but equality need not hold. Further investigation is left to future work.
\end{remark}

Recall that the group of position eigenstates under translation symmetry for the coherent group $\mathbb{G} = (\hbox{\input{symbols/ZbwdotSym.tex}}\!\!,\hbox{\input{symbols/DdotSym.tex}}\!\!)$ is given by the underlying group $G = \underlyingGroup{\mathbb{G}} = (\classicalStates{\hbox{\input{symbols/ZbwdotSym.tex}}\!\!},\!\hbox{\input{symbols/DmultSym.tex}}\!\!, \!\hbox{\input{symbols/DunitSym.tex}}\!\!)$: the discussion until this point makes it clear that the correct choice for the group of momentum eigenstates under boost symmetry should be the underlying group $\underlyingGroup{\mathbb{G}^\wedge} = (\classicalStates{\hbox{\input{symbols/DdotSym.tex}}\!\!},\!\hbox{\input{symbols/ZbwmultSym.tex}}\!\!, \!\hbox{\input{symbols/ZbwunitSym.tex}}\!\!)$ of the dual coherent group $\mathbb{G}^\wedge = (\hbox{\input{symbols/DdotSym.tex}}\!\!,\hbox{\input{symbols/ZbwdotSym.tex}}\!\!)$. 

Now that we have figured out what boost symmetry in our generalise setting should be, we need to confirm that it relates as expected to the position observable: need to show that positions generate our choice of boost symmetry (in the same sense as momenta generating translation symmetry in Theorem \ref{thm_latticeTranslationSymmetryAction}), and that they are the invariant states for boost symmetry (in the same sense as momenta being the invariant states of translation symmetry in Theorem \ref{thm_multiplicativeCharInvariance}).

\begin{theorem}[\textbf{Positions generate boosts}]\label{thm_latticeBoostSymmetryAction}\hfill\\
	Let $\mathbb{G} := (\hbox{\input{symbols/ZbwdotSym.tex}}\!\!,\hbox{\input{symbols/DdotSym.tex}}\!\!)$ be a coherent group on an object $\SpaceH$ of a $\dagger$-SMC $\CategoryC$, and let $X:=(\classicalStates{\hbox{\input{symbols/DdotSym.tex}}\!\!},\!\hbox{\input{symbols/ZbwmultSym.tex}}\!\!,\!\hbox{\input{symbols/ZbwunitSym.tex}}\!\!)$ be the underlying group for the dual $\mathbb{G}^\wedge$. Define a family $(V_{\goodchi})_{\goodchi^\dagger \in X}$ of endomorphisms of $\SpaceH$ as follows: 
	\begin{equation}\label{latticeBoostSymmetryAction}
		\input{pictures/chapter3/latticeBoostSymmetryAction.tikz}
	\end{equation}
	Then $(V_{\goodchi})_{\goodchi^\dagger \in X}$ gives a unitary action of the group $X$ on $\SpaceH$, restricting to the left regular action of the boost symmetry group $\underlyingGroup{\mathbb{G}^\wedge}$ on the points of the dual coherent group $\mathbb{G}^\wedge$. 
\end{theorem}
\begin{proof}
	This is nothing but Theorem \ref{thm_latticeTranslationSymmetryAction} applied to the dual coherent group $\mathbb{G}^\wedge$. 
\end{proof}

\begin{theorem}[\textbf{Positions invariant under boost}]\label{thm_pointsBoostInvariance}\hfill\\
	Let $\mathbb{G} := (\hbox{\input{symbols/ZbwdotSym.tex}}\!\!,\hbox{\input{symbols/DdotSym.tex}}\!\!)$ be a coherent group on an object $\SpaceH$ of a $\dagger$-SMC $\CategoryC$, and let $X:=(\classicalStates{\hbox{\input{symbols/DdotSym.tex}}\!\!},\!\hbox{\input{symbols/ZbwmultSym.tex}}\!\!,\!\hbox{\input{symbols/ZbwunitSym.tex}}\!\!)$ be the underlying group of the dual coherent group $\mathbb{G}^\wedge$. The points of $\mathbb{G}$ are invariant under the boost symmetry action, up to a scalar:
	\begin{equation}\label{pointsBoostInvariance}
		\input{pictures/chapter3/pointsBoostInvariance.tikz}
	\end{equation}
	Furthermore, the scalar $\goodchi(g) = \goodchi \circ g$ satisfies $\goodchi(g)^\dagger \cdot \goodchi(g) = 1$. Finally, assume that the dual coherent group $\mathbb{G}^\wedge$ is well-pointed, and consider a state $g$: if both (i) Equation \ref{pointsBoostInvariance} holds for all multiplicative characters $\goodchi$ of $\mathbb{G}$, and (ii) $\!\hbox{\input{symbols/ZbwcounitSym.tex}}\!\!(g)^\dagger \cdot \!\hbox{\input{symbols/ZbwcounitSym.tex}}\!\!(g) = 1$, then $g$ must be a point of $\mathbb{G}$.
\end{theorem}
\begin{proof}
	This is nothing but Theorem \ref{thm_multiplicativeCharInvariance} applied to the dual coherent group $\mathbb{G}^\wedge$.
\end{proof}

Looking at the assumptions of Theorem \ref{thm_pointsBoostInvariance}, we see a problem arise: while position eigenstates (the points of the coherent group $\mathbb{G}$) are always invariant under boost symmetry and generate it,  Theorem \ref{thm_pointsBoostInvariance} does not guarantee that they will be the \textit{only} invariant states unless the dual coherent group $\mathbb{G}^\wedge$ is itself well-pointed. This is always true for a well-pointed abelian coherent group $\mathbb{G}$ in $\fdHilbCategory$, because the points of the dual $\mathbb{G}^\wedge$ correspond exactly to the multiplicative characters of the underlying abelian group $\underlyingGroup{\mathbb{G}}$: the latter form a basis, and hence $\mathbb{G}^\wedge$ is itself well-pointed. However, it need not be true in general.

We will say that a coherent group $\mathbb{G}$ is \textbf{doubly well-pointed} if both $\mathbb{G}$ and $\mathbb{G}^\wedge$ are well-pointed. It is worth noting that a doubly well-pointed coherent group is necessarily abelian, and that the dual of a doubly well-pointed coherent group is itself doubly well-pointed. From a physical perspective, well-pointedness models the requirement that there must be enough position eigenstates to distinguish different processes from $\SpaceH$: processes are entirely determined by what they do to the position eigenstates. Double well-pointedness models the additional requirement that there must also be enough momentum eigenstates to distinguish different processes from $\SpaceH$, i.e. that processes are also entirely determined by what they do to the momentum eigenstates. In the light of these developments, we will henceforth require the coherent group modelling wavefunctions on a periodic lattice to be doubly well-pointed (which, as we mentioned, is a special case of well-pointed abelian). The same requirement will carry through to Sections \ref{section_finiteCyclic} and \ref{section_compactAbelian}.

When $\mathbb{G}$ is doubly well-pointed, we can write both the position measurement and the \textbf{momentum measurement} in $\CPStarCategory{\CategoryC}$ (which we assume to be $R$-probabilistic) as processes with output in an appropriate classical system: 
\begin{equation}\label{latticePositionMomentumMeasurements}
	\input{pictures/chapter3/latticePositionMomentumMeasurements.tikz}
\end{equation}

Before moving on to the Weyl CCRs, let's go over a brief summary of our work until this point. At the beginning of Section \ref{section_systemsWithSymmetries}, we have considered a well-pointed coherent group $\mathbb{G} = (\hbox{\input{symbols/ZbwdotSym.tex}}\!\!, \hbox{\input{symbols/DdotSym.tex}}\!\!)$ having finitely many points as an abstract model of wavefunctions on a periodic lattice. This is because the points $\classicalStates{\hbox{\input{symbols/ZbwdotSym.tex}}\!\!}$ of the coherent group form a finite abelian group $\underlyingGroup{G} = (\classicalStates{\hbox{\input{symbols/ZbwdotSym.tex}}\!\!}, \!\hbox{\input{symbols/DmultSym.tex}}\!\!, \!\hbox{\input{symbols/DunitSym.tex}}\!\!) \isom \prod_{d=1}^{D} \integersMod{n_d}$, which can be interpreted as a $D$-dimensional periodic lattice $\Lambda$ endowed with the group structure of translation symmetry.

While $\hbox{\input{symbols/ZbwdotSym.tex}}\!\!$ is the natural candidate for the position observable, identifying $\hbox{\input{symbols/DdotSym.tex}}\!\!$ with the momentum observable is more challenging. In Subsection \ref{subsection_whatWeLookForInMomentum}, we have compiled a list of operational and structural properties characterising the relationship between the position and momentum observables for wavefunctions on periodic lattices in the traditional formulation of quantum mechanics: 
\begin{enumerate}
	\item[(a)] that the momentum eigenstates are invariant under the translation symmetry action, and that they generate it; 
	\item[(b)] that the position eigenstates are invariant under the boost symmetry action, and that they generate it; 
	\item[(c)] that the Weyl Canonical Commutation Relations hold; 
	\item[(d)] that the weak form of the uncertainty principle (see Subsection \ref{subsection_whatWeLookForInMomentum}) holds.
\end{enumerate}
In Subsection~\ref{subsection_momentumTranslationSymmetry}, we have shown that the classical states $\classicalStates{\hbox{\input{symbols/DdotSym.tex}}\!\!}$ for the $\hbox{\input{symbols/DdotSym.tex}}\!\!$ observable, our putative momentum eigenstates, are exactly the invariant states for the translation symmetry action on wavefunctions, which that they furthermore generate. In Subsection \ref{subsection_positionsBoostSymmetry}, we have identified the boost symmetry in the underlying group $\underlyingGroup{\mathbb{G}^\wedge} = (\classicalStates{\hbox{\input{symbols/DdotSym.tex}}\!\!}, \!\hbox{\input{symbols/ZbwmultSym.tex}}\!\!, \!\hbox{\input{symbols/ZbwunitSym.tex}}\!\!)$ of the dual coherent group $\mathbb{G}^\wedge = (\hbox{\input{symbols/DdotSym.tex}}\!\!,\hbox{\input{symbols/ZbwdotSym.tex}}\!\!)$. We have then shown that the position eigenstates $\classicalStates{\hbox{\input{symbols/ZbwdotSym.tex}}\!\!}$ are invariant states for the boost symmetry action on wavefunctions, which they generate. In order to characterise the position eigenstates as \textit{exactly} the invariant states under the boost symmetry action, we strengthened our requirements on $\mathbb{G}$, assuming that it is doubly well-pointed.

We are half-way through: points (a) and (b) of our list are down, points (c) and (d) remain to be shown. These will be the topic of the next two Subsections.

\subsection{Weyl Canonical Commutation Relations}

We have already mentioned in Equation \ref{eqn_WeylCCRsLattice} that the Weyl CCRs for wavefunctions on periodic lattices should take the following form:
\begin{equation}\nonumber 
V_{\goodchi} U_g = \goodchi(g)\; U_g V_{\goodchi} 
\end{equation}
This is in direct analogy with the traditional Weyl CCRs from Equation \ref{eqn_WeylCCRs}:
\begin{equation}\nonumber 
V_p U_x = e^{i \hbar p \cdot x} \; U_x V_p
\end{equation}
There is a single difference between Equation \ref{eqn_WeylCCRs} and Equation \ref{eqn_WeylCCRsLattice}: in the former, the momentum eigenstates are labelled by the eigenvalues $p$ of the infinitesimal generator $\textbf{p}$, while in the latter the momentum eigenstates are labelled by the multiplicative characters $\goodchi$ of the translation symmetry group. As a consequence, the phase in Equation \ref{eqn_WeylCCRs} is written explicitly as $e^{i \hbar p \cdot x}$, while the phase in Equation \ref{eqn_WeylCCRsLattice} is obtained more naturally by evaluating the multiplicative character $\goodchi$ labelling the momentum eigenstate on the group element $g$ labelling the position eigenstate. Refer to Subsection \ref{subsection_briefDigressionObservables} for the reasons behind this choice.

Theorems \ref{thm_latticeTranslationSymmetryAction} and \ref{thm_latticeBoostSymmetryAction} already provide us with an abstract description of the unitaries $U_g$ and $V_{\goodchi}$, and of the scalar $\goodchi(g)$: all we need to show is that they respect Equation \ref{eqn_WeylCCRsLattice}.

\begin{theorem}[\textbf{Weyl Canonical Commutation Relations}]\label{thm_WeylCCRs}\hfill\\
Let $\mathbb{G}=(\hbox{\input{symbols/ZbwdotSym.tex}}\!\!,\hbox{\input{symbols/DdotSym.tex}}\!\!)$ be a coherent group in a $\dagger$-SMC $\CategoryC$, and define the families of unitaries $(U_g)_{g \in \underlyingGroup{\mathbb{G}}}$ and  $(V_{\goodchi})_{\goodchi^\dagger \in \underlyingGroup{\mathbb{G}^\wedge}}$ as in Equations \ref{latticeTranslationSymmetryAction} and \ref{latticeBoostSymmetryAction}. Then the Weyl Canonical Commutation Relations from Equation \ref{eqn_WeylCCRsLattice} hold:
\begin{equation}\label{WeylCCRs}
	\input{pictures/chapter3/WeylCCRs.tikz}
\end{equation}
\end{theorem}
\begin{proof}
The proof hinges on the bialgebra law from strong complementarity, on the definition of the antipode, and on the copy/adjoin conditions for $\hbox{\input{symbols/ZbwdotSym.tex}}\!\!$-classical and $\hbox{\input{symbols/DdotSym.tex}}\!\!$-classical states. The first equality below uses the definition of the antipode and the adjoin condition for the $\hbox{\input{symbols/DdotSym.tex}}\!\!$-classical state $\goodchi^\dagger$, while the second equality uses the bialgebra law:
\begin{equation}\label{WeylCCRsProof1}
	\resizebox{\textwidth}{!}{\input{pictures/chapter3/WeylCCRsProof1.tikz}}
\end{equation}
The first equality below uses the copy condition for the $\hbox{\input{symbols/ZbwdotSym.tex}}\!\!$-classical state $g$ and the $\hbox{\input{symbols/DdotSym.tex}}\!\!$-classical state $\hbox{\input{symbols/antipodeSym.tex}}\! \circ \goodchi$ (because the antipode is a function on $\hbox{\input{symbols/DdotSym.tex}}\!\!$-classical states), while the second equality uses again the definition of the antipode and the adjoin condition for the $\hbox{\input{symbols/DdotSym.tex}}\!\!$-classical state $\goodchi$:
\begin{equation}\label{WeylCCRsProof2}
	\resizebox{\textwidth}{!}{\input{pictures/chapter3/WeylCCRsProof2.tikz}}
\end{equation}
This completes our proof of the Weyl Canonical Commutation Relations.
\end{proof}

In the traditional presentation of quantum mechanics, the Weyl Canonical Commutation Relations make an important appearance as part of the Stone-von Neumann Theorem \cite{Stone1930,V.Neumann1931,Stone1932,V.Neumann1932}. Consider two jointly irreducible unitary representations $(U_t)_{t \in \reals}$ and $(V_s)_{s \in \reals}$ of the abelian group $(\reals,+,0)$ on some separable Hilbert Space $\SpaceH$: the Stone-von Neumann Theorem states that if they satisfy the Weyl Canonical Commutation Relations from Equation \ref{eqn_WeylCCRs} then they are jointly unitarily equivalent to the translation and boost symmetry actions on $\Ltwo{\reals}$, i.e. there is some unitary $W: \Ltwo{\reals} \rightarrow \SpaceH$ such that $W^\dagger U_t W = e^{i t \textbf{x}}$ and $W^\dagger V_s W = e^{i s \textbf{p}}$. 

We cannot expect such a direct and precise result in our case, for a variety of reasons. First and foremost, we have as many inequivalent notions of position and momentum as there are periodic lattices: Equation \ref{eqn_WeylCCRs} gives an explicit braiding relation, with $e^{i p x }$ as phase, while Equation \ref{eqn_WeylCCRsLattice} gives a braiding relation parametrized on the multiplicative characters, with $\goodchi(g)$ as phase. However, substituting an explicit form for $\goodchi(g)$ singles out a unique lattice position/momentum pair: e.g.  writing $\goodchi(g) = e^{i 2 \pi \frac{p g}{10}}$ in $\fdHilbCategory$ singles out the position/momentum pair for the 1-dimensional lattice $\integersMod{10}$, while writing $e^{i 2 \pi (\frac{p_1 g_1}{4} + \frac{p_2 g_2}{8}) }$ singles out the position/momentum pair for the 2-dimensional lattice $\integersMod{4} \times \integersMod{8}$. In fact, this is the same for the Stone-von Neumann Theorem in its modern form: the 1-dimensional $e^{i p x }$ case for the group $(\reals,+,0)$ was the first to be proven, but a straightforward generalisation exists for all the groups $(\reals^n,+,0)$, with phases given by $e^{i \sum_{j=1}^n p_j x_j}$.

There is a second, structural reason why we cannot expect a result as tight as the Stone-von Neumann Theorem in the general setting of coherent groups on $\dagger$-SMCs: both the original result and its generalisations to Mackey theory rely both on Pontryagin duality and on a considerable amount of continuous and integrable structure. In an general $\dagger$-SMC $\CategoryC$, the possible choices for a position/momentum pair are classified by the (doubly well-pointed) coherent groups, rather than the underlying groups: the functor $\underlyingGroup{\emptyArg}$ on doubly well-pointed coherent groups is faithful, but not necessarily full, and as a consequence it might not be possible to find a suitable subcategory of groups which classify the position/momentum pairs on systems of $\CategoryC$.\footnote{Contrary to $\fdHilbCategory$, where the functor $\underlyingGroup{\emptyArg}$ is full and faithful on doubly well-pointed coherent groups, and essentially surjective onto the subcategory of finite abelian groups.}

The search for a suitable extension of the Stone-von Neumann Theorem to coherent groups in arbitrary $\dagger$-SMCs is left to future work.

\subsection{Uncertainty principle}

As of this point, we are three fourths of the way to establishing that $\hbox{\input{symbols/DdotSym.tex}}\!\!$ is a suitable momentum observable in a coherent group: we have proven (a) the relationship between momentum observable and translation symmetry; (b) the relationship between position observable and boost symmetry; (c) the Weyl Canonical Commutation Relations. There is one final piece of evidence that we tasked ourselves with finding: item (d), the uncertainty principle. 

It was already remarked in Subsection \ref{subsection_whatWeLookForInMomentum} that the full Kennard-Weyl form of the uncertainty principle from Equation \ref{eqn_KennardWeylUncertaintyPrinciple} is too strong to be considered a characteristic trait of position/momentum duality: it essentially implies contextuality, an operational feature of quantum theory that we believe should not play a direct role in the abstract treatment of mechanics and dynamics. In this light, we proposed that a suitable compromise would involve restricting our attention to the position and momentum observables themselves, and show that position eigenstates have completely indeterminate momentum, and vice versa that momentum eigenstates have completely indeterminate position. But we already know this is going be true: Lemma \ref{lem_complementarityUnbiased} states that any complementary pair is mutually unbiased, and in particular so will be the pair of observables $\hbox{\input{symbols/ZbwdotSym.tex}}\!\!$ and $\hbox{\input{symbols/DdotSym.tex}}\!\!$ appearing in a coherent group (which by definition must be strongly complementary, and hence complementary).

\begin{theorem}[\textbf{Weak uncertainty principle}]\label{thm_uncertaintyPrinciple}\hfill \\
Let $\mathbb{G} = (\hbox{\input{symbols/ZbwdotSym.tex}}\!\!,\hbox{\input{symbols/DdotSym.tex}}\!\!)$ be a coherent group on an object $\SpaceH$ of a $\dagger$-SMC, and let $N_{\hbox{\input{symbols/ZbwdotSym.tex}}\!\!}$ and $N_{\hbox{\input{symbols/DdotSym.tex}}\!\!}$ be the normalisation factors of the $\dagger$-qSFA $\hbox{\input{symbols/ZbwdotSym.tex}}\!\!$ and $\hbox{\input{symbols/DdotSym.tex}}\!\!$ respectively. Then the points of $\mathbb{G}$, i.e. the position eigenstates $\classicalStates{\hbox{\input{symbols/ZbwdotSym.tex}}\!\!}$, are unbiased for the momentum observable $\hbox{\input{symbols/DdotSym.tex}}\!\!$. Conversely, the points of the dual coherent group $\mathbb{G}^\wedge$, i.e. the momentum eigenstates $\classicalStates{\hbox{\input{symbols/DdotSym.tex}}\!\!}$, are unbiased for the position observable $\hbox{\input{symbols/ZbwdotSym.tex}}\!\!$. If $\mathbb{G}$ is doubly well-pointed and \textbf{doubly finite} (by which we mean it has finitely many points and multiplicative characters), we get the following in $\CPStarCategory{\CategoryC}$ (which we assume to be $R$-probabilistic):
\begin{equation}\label{uncertaintyPrinciplePositionMeas}
	\input{pictures/chapter3/uncertaintyPrinciplePositionMeas.tikz}
\end{equation}
\begin{equation}\label{uncertaintyPrincipleMomentumMeas}
	\input{pictures/chapter3/uncertaintyPrincipleMomentumMeas.tikz}
\end{equation}
Note that the states on the RHS of the two equations above are the states of $R^{\underlyingGroup{\mathbb{G}}}$ and $R^{\underlyingGroup{\mathbb{G}^\wedge}}$ corresponding to elements $g \in \underlyingGroup{\mathbb{G}}$ and $\goodchi^\dagger \in \underlyingGroup{\mathbb{G}}$: they correspond to the normalised CPM states $\frac{1}{N_{\hbox{\input{symbols/ZbwdotSym.tex}}\!\!}}\CPMdoubled{g}$ and $\frac{1}{N_{\hbox{\input{symbols/DdotSym.tex}}\!\!}}\CPMdoubled{\goodchi^\dagger}$ respectively. 
\end{theorem}
\begin{proof}
Essentially Lemma \ref{lem_complementarityUnbiased}, taking into account the normalisation factors of the two $\dagger$-qSFAs and using the fact that they both have enough classical states (so that both $(\SpaceH,\hbox{\input{symbols/ZbwdotSym.tex}}\!\!)$ and $(\SpaceH,\hbox{\input{symbols/DdotSym.tex}}\!\!)$ are classical systems in $\CPStarCategory{\CategoryC}$).
\end{proof}

We have finally come to the end of our quest: we have shown that the point structure $\hbox{\input{symbols/ZbwdotSym.tex}}\!\!$ and group structure $\hbox{\input{symbols/DdotSym.tex}}\!\!$ in a coherent group $\mathbb{G} := (\hbox{\input{symbols/ZbwdotSym.tex}}\!\!,\hbox{\input{symbols/DdotSym.tex}}\!\!)$ possess the main operational and structural features that we would expect from a position/momentum pair. In particular, doubly well-pointed, doubly finite coherent groups  can always be interpreted as defining the position/momentum pair for wavefunctions on a periodic lattice, as made clear by the following summary of the work to this point.
\begin{enumerate}
	\item[(i)] The points $\classicalStates{\hbox{\input{symbols/ZbwdotSym.tex}}\!\!}$ of a finite coherent group are endowed with the group structure $\prod_{d=1}^{D} \integersMod{n_d}$ of some periodic lattice $\Lambda$. They can therefore be interpreted as position eigenstates wavefunctions on the lattice, and $\hbox{\input{symbols/ZbwdotSym.tex}}\!\!$ can be interpreted as the position observable.
	\item[(ii)] The processes from a well-pointed coherent group are entirely determined by their action on the position eigenstates, excluding the existence of additional underlying structure.
	\item[(iii)] The momentum observable $\hbox{\input{symbols/DdotSym.tex}}\!\!$ in a well-pointed coherent group generates the translation symmetry action on the wavefunctions, and its classical states $\classicalStates{\hbox{\input{symbols/DdotSym.tex}}\!\!}$ are the invariant states for that action.
	\item[(iv)] The position observable $\hbox{\input{symbols/ZbwdotSym.tex}}\!\!$ in a doubly well-pointed coherent group generates the boost symmetry action on the wavefunctions, and its classical states $\classicalStates{\hbox{\input{symbols/ZbwdotSym.tex}}\!\!}$ are the invariant states for that action.
	\item[(v)] The position observable $\hbox{\input{symbols/ZbwdotSym.tex}}\!\!$ and the momentum observable $\hbox{\input{symbols/DdotSym.tex}}\!\!$ in a coherent group always satisfy the Weyl Canonical Commutation Relations.
	\item[(vi)] The position and momentum observables are always mutually unbiased. In a doubly well-pointed, doubly finite coherent group we can define both a position measurement and a momentum measurement in the CP* category (assuming the latter is $R$-probabilistic), with outcomes in the classical systems $R^{\underlyingGroup{\mathbb{G}}}$ and $R^{\underlyingGroup{\mathbb{G}^\wedge}}$ respectively. Measuring the position of a (normalised) momentum eigenstate yields the uniform distribution on $R^{\underlyingGroup{\mathbb{G}}}$, proving that momentum eigenstates have completely indeterminate positions. Similarly, measuring the momentum of a (normalised) position eigenstate yields the uniform distribution on $R^{\underlyingGroup{\mathbb{G}^\wedge}}$, proving that position eigenstates have completely indeterminate momentum.      
\end{enumerate}

\newpage
\section{Systems with symmetries}
\label{section_systemsWithSymmetries}

Up until this moment, we have restricted our attention to the concrete example of wavefunctions on periodic lattices, which we have abstractly identified with doubly well-pointed, doubly finite coherent groups. However, almost none of the results we obtained requires well-pointedness or finiteness: the picture they paint is that coherent groups \textit{in general} have many of the structural properties of position/momentum pairs for wavefunctions on symmetric systems. There are some issues arising when the groups are well-pointed but not doubly well-pointed, such as in the case of non-abelian group algebras in $\fdHilbCategory$, which will be covered in future work. There are also some issues with the operational interpretation of non-finite coherent groups, which will be covered in Section \ref{section_compactAbelian}. However, the overall picture as it stands is solid enough, and throughout this Section we will interpret coherent groups as modelling a sensible notion of wavefunctions over symmetry groups.

Just like representations of classical groups yield the notion of physical systems with a classical symmetry, we expect that a suitable notion of representation for coherent groups will yield a suitable notion of physical system with a coherent symmetry. The topic of this section will be the definition and study of said representations.

\subsection{Unitary representations of coherent groups}

In the traditional formalism, a quantum system with periodic lattice symmetry is given by a unitary representation $(U_g)_{g \in G}$ of the translation symmetry group $G = \prod_{d=1}^{D} \integersMod{n_d}$ on a Hilbert space $\SpaceH$, and wavefunctions on a periodic lattice arise as the special case $\SpaceH = \complexs[G]$ with the regular action $g(\ket{h}) := \ket{g \oplus h}$. In the coherent approach, we consider unitary representations of coherent groups instead, and the physical intuition behind this choice goes as follows. Unitary representations of a group can be seen as controlled unitaries, where the controlling system is classical: in the case of periodic lattice symmetries, the controlling system is the periodic lattice itself (or, equivalently, the classical system of distributions over the lattice). In the passage from the classical to the coherent approach, we wish the states of the controlling system to be wavefunctions over the lattice instead of distributions: as our work to this point shows, this is the same as moving from a (finite abelian) group to a (doubly well-pointed, doubly finite) coherent group. 

But what should we take as a unitary representation of a coherent group $\mathbb{G}$ in a generic $\dagger$-SMC? One possible approach to figuring this out starts from the definition of representations of finite groups in $\fdHilbCategory$, and tries to replace all the classical bits and pieces with their coherent counterparts. Recall that a unitary representation $(U_g)_{g \in G}$ of a group\footnote{We denote the group multiplication by $m: G \times G \rightarrow G$, the group unit by $e: 1 \rightarrow G$, and the group inverse by $i: G \rightarrow G$.} $(G,m,e,i)$ on a finite-dimensional Hilbert space $\SpaceH$ is a function $\alpha: \SpaceH \times G \rightarrow \SpaceH$, such that $U_g = \alpha(\emptyArg,g)$ is linear for all $g \in G$ and with $\alpha$ satisfying the following three conditions:
\begin{align}
	\alpha\big(\emptyArg, m(g,h)\big) &= U_{gh} = U_h U_g = \alpha\big(\alpha(\emptyArg,h),g\big) \\
	\alpha(\emptyArg,e) &= U_e = \id{\SpaceH} \\
	\alpha(\emptyArg,i(g)) &= U_{g^{-1}} = U_g^{-1} = U_g^\dagger= \Big(\alpha(\emptyArg,g)\Big)^\dagger
\end{align}
We make the following modifications: instead of the group $G$ and function $\alpha: \SpaceH \times G \rightarrow \SpaceH$, we work with the group algebra $\complexs[G]$ (which we see as a well-pointed coherent group $\mathbb{G} := (\hbox{\input{symbols/ZbwdotSym.tex}}\!\!, \hbox{\input{symbols/DdotSym.tex}}\!\!)$) and a linear map $\alpha: \SpaceH \otimes \complexs[G] \rightarrow \SpaceH$. We replace the three conditions above with the following conditions involving the linear extensions $\!\hbox{\input{symbols/DmultSym.tex}}\!\!$, $\!\hbox{\input{symbols/DunitSym.tex}}\!\!$ and $\hbox{\input{symbols/antipodeSym.tex}}\!$ of the multiplication $m$, unit $e$ and inverse $i$:
\begin{align}
	\alpha\big(\emptyArg, \!\hbox{\input{symbols/DmultSym.tex}}\!\! \circ (\ket{g}\otimes\ket{h})\big) &= \alpha\big(\alpha(\emptyArg,\ket{g}),\ket{h}\big) \label{quantumGroupRepMultInline}\\
	\alpha(\emptyArg,\!\hbox{\input{symbols/DunitSym.tex}}\!\!) &= \id{\SpaceH} \label{quantumGroupRepUnitInline} \\
	\alpha(\emptyArg,\hbox{\input{symbols/antipodeSym.tex}}\!\ket{g}) &= \Big(\alpha(\emptyArg,\ket{g})\Big)^\dagger \label{quantumGroupRepAntipodeInline}
\end{align}
As a final step, we get rid of the evaluation over group elements (a very classical thing to do), and obtain a definition which solely involves the coherent group $\mathbb{G}$.

\begin{definition}\label{def_quantumGroupUnitaryRep}
Let $\mathbb{G} = (\hbox{\input{symbols/ZbwdotSym.tex}}\!\!,\hbox{\input{symbols/DdotSym.tex}}\!\!)$ be a coherent group on an object $\SpaceG$ of a $\dagger$-SMC $\CategoryC$, and let $\SpaceH$ be another object of $\CategoryC$. A \textbf{representation} of $\mathbb{G}$ on $\SpaceH$ is a process $\alpha: \SpaceH \otimes \SpaceG \rightarrow \SpaceH$ satisfying the following two requirements:
\begin{equation}\label{quantumGroupRepMult}
	\input{pictures/chapter3/quantumGroupRepMult.tikz}
\end{equation}
\begin{equation}\label{quantumGroupRepUnit}
	\input{pictures/chapter3/quantumGroupRepUnit.tikz}
\end{equation}
We say that a representation $\alpha$ is \textbf{unitary} if it satisfies the following additional requirement:
\begin{equation}\label{quantumGroupRepAntipode}
	\input{pictures/chapter3/quantumGroupRepAntipode.tikz}
\end{equation}
\end{definition}
Equations \ref{quantumGroupRepMult} and \ref{quantumGroupRepUnit} are straightforward graphical translations of Equations \ref{quantumGroupRepMultInline} and \ref{quantumGroupRepUnitInline}. Equation \ref{quantumGroupRepAntipode} sees, further to Equation \ref{quantumGroupRepAntipodeInline}, the introduction of the symmetric cup for the point structure. This is because in Equation \ref{quantumGroupRepAntipodeInline} we take the adjoint of the representation \textit{already evaluated} at $\ket{g}$, and hence from a compositional perspective we are taking the adjoint of $\ket{g}$ as well: this is achieved by using the symmetric cup, as prescribed by the adjoin condition for $\hbox{\input{symbols/ZbwdotSym.tex}}\!\!$-classical states, which leads to the graphical formulation of Equation \ref{quantumGroupRepAntipode}. 

To make sure that our definition of coherent symmetries is sensible, we need to check two things: (i) that our definition is consistent with the definition of coherent translation symmetry for wavefunctions on a periodic lattice, and (ii) that our definition yields back the classical symmetries by appropriate use of preparations/measurements. Just to clarify, by a (unitary) representation of a group $G$ on an object $\SpaceH$ of a $\dagger$-SMC we mean a family $(U_g)_{g \in G}$ of (unitary) processes $U_g: \SpaceH \rightarrow \SpaceH$ such that $U_{gh} = U_h U_g$ and $U_e = \id{\SpaceH}$. 

\begin{lemma}\label{lem_regularRep}
Let $\mathbb{G}:=(\hbox{\input{symbols/ZbwdotSym.tex}}\!\!,\hbox{\input{symbols/DdotSym.tex}}\!\!)$ be a coherent group on an object $\SpaceG$ of a $\dagger$-SMC. Then $\!\hbox{\input{symbols/DmultSym.tex}}\!\!: \SpaceG \otimes \SpaceG \rightarrow \SpaceG$ is a unitary representation of $\mathbb{G}$ on $\SpaceG$, which we will refer to as the \textbf{regular representation} of $\mathbb{G}$.
\end{lemma} 
\begin{proof}
The first requirement for a representation of $\mathbb{G}$ is a consequence of associative law, the second requirement for a representation of $\mathbb{G}$ is a consequence of unit law, and the additional requirement for a unitary representation of $\mathbb{G}$ is a consequence of Frobenius law and unit law (once the antipode is expanded in terms of symmetric cap/cap of the point/group structure of $\mathbb{G}$):
\begin{equation}\label{quantumGroupRegularRep}
	\resizebox{\textwidth}{!}{\input{pictures/chapter3/quantumGroupRegularRep.tikz}}
\end{equation}
This completes the proof, showing that $\!\hbox{\input{symbols/DmultSym.tex}}\!\!$ is indeed a unitary representation of the coherent group $\mathbb{G}$.
\end{proof}

\begin{theorem}[\textbf{Underlying group reps from coherent group reps}]\label{thm_classicalFromQuantumReps}\hfill\\
Let $\mathbb{G}:=(\hbox{\input{symbols/ZbwdotSym.tex}}\!\!,\hbox{\input{symbols/DdotSym.tex}}\!\!)$ be a coherent group on an object $\SpaceG$ of a $\dagger$-SMC $\CategoryC$, and let $\alpha: \SpaceH \otimes \SpaceG \rightarrow \SpaceH$ be a (unitary) representation of $\mathbb{G}$ on an system $\SpaceH$ in $\CategoryC$. Then the following defines a (unitary) representation $(U_g)_{g \in \underlyingGroup{\mathbb{G}}}$ of $\underlyingGroup{\mathbb{G}}$ on $\SpaceH$:
\begin{equation}\label{classicalFromQuantumReps}
	\input{pictures/chapter3/classicalFromQuantumReps.tikz}
\end{equation}
If $\CPStarCategory{\CategoryC}$ is $R$-probabilistic and $\mathbb{G}$ is well-pointed, then we can write the representation of the underlying group $\underlyingGroup{\mathbb{G}}$ as the following classically controlled process:
\begin{equation}\label{classicalFromQuantumRepsCPStar}
	\input{pictures/chapter3/classicalFromQuantumRepsCPStar.tikz}
\end{equation}
If $\alpha$ is a unitary representation, then the process above is in fact a controlled unitary.
\end{theorem}
\begin{proof}
The first claim follows straightforwardly from the requirements satisfied by a (unitary) representation $\alpha$ for $\mathbb{G}$ and from the definition of the underlying group $\underlyingGroup{\mathbb{G}} := (\classicalStates{\hbox{\input{symbols/ZbwdotSym.tex}}\!\!},\!\hbox{\input{symbols/DmultSym.tex}}\!\!,\!\hbox{\input{symbols/DunitSym.tex}}\!\!)$:
\begin{equation}\label{classicalFromQuantumRepsProof1}
	\input{pictures/chapter3/classicalFromQuantumRepsProof1.tikz}
\end{equation}
\begin{equation}\label{classicalFromQuantumRepsProof2}
	\input{pictures/chapter3/classicalFromQuantumRepsProof2.tikz}
\end{equation}
\begin{equation}\label{classicalFromQuantumRepsProof3}
	\input{pictures/chapter3/classicalFromQuantumRepsProof3.tikz}
\end{equation}
Recasting the representation in the form of Diagram \ref{classicalFromQuantumRepsCPStar} is also completely straightforward, but statement that it yields a controlled unitary when $\alpha$ is a unitary representation deserves graphical proof:
\begin{equation}\nonumber
	\resizebox{\textwidth}{!}{\input{pictures/chapter3/classicalFromQuantumRepsCPStarProof1.tikz}}
\end{equation}
\begin{equation}\nonumber
	\resizebox{\textwidth}{!}{\input{pictures/chapter3/classicalFromQuantumRepsCPStarProof2.tikz}}
\end{equation}
\begin{equation}\label{classicalFromQuantumRepsCPStarProof}
	\resizebox{\textwidth}{!}{\input{pictures/chapter3/classicalFromQuantumRepsCPStarProof3.tikz}}
\end{equation}
The other half of the proof of controlled unitarity goes along the exact same lines. 
\end{proof}

\subsection{The category of representations of a coherent group}
\label{subsection_EMCategory}

Our definition of a coherent group representation is a very concrete one, which we extrapolated directly from the definition of classical group representations. Instead, we would prefer a more categorical characterisation of coherent group representations. We begin by looking at the following commuting diagrams, involving a representation $\alpha$ of a coherent group $\mathbb{G}:=(\hbox{\input{symbols/ZbwdotSym.tex}}\!\!,\hbox{\input{symbols/DdotSym.tex}}\!\!)$ on a system $\SpaceH$ of a generic $\dagger$-SMC $\CategoryC$:
\begin{equation}\label{quantumGroupRepAlgebraDiagram}
	\resizebox{\textwidth}{!}{\input{pictures/chapter3/quantumGroupRepAlgebraDiagram.tikz}}
\end{equation}
To a category theorist, these two diagrams scream \inlineQuote{algebra of a monad}. This would indeed be a nice categorical definition, but which monad are we talking about? And what is the physical meaning of all of this?

When talking about a system $\SpaceH$ with periodic lattice symmetry from a coherent perspective, we are implicitly considering two systems: the system $\SpaceH$ itself, and the system $\SpaceG$ of wavefunctions on the periodic lattice that control the symmetry (i.e. we are thinking of a specific coherent group $\mathbb{G}$ on it). The concrete process of taking a system $\SpaceH$ and considering it jointly with the coherent controlling system $\SpaceG$ can be turned into a functor $T: \CategoryC \rightarrow \CategoryC$ (i.e. it can be made properly categorical) as follows:
\begin{align}
	T[\SpaceH] &:= \SpaceH \otimes \SpaceG \nonumber \\
	T[f: \SpaceH \rightarrow \SpaceH'] &:= f \otimes \id{\SpaceG} : \SpaceH \otimes \SpaceG \rightarrow \SpaceH' \otimes \SpaceG \label{eqns_quantumGroupMonadFunctor}
\end{align}
The functor $T$ can be made into a  \textit{monad} by considering the following natural transformations, known as the \textit{multiplication} $\mu: T^2 \rightarrow T$ and \textit{unit} $\eta: \id{\CategoryC} \rightarrow T$:
\begin{align}
	\mu_\SpaceH &:= \id{\SpaceH} \otimes \!\hbox{\input{symbols/DmultSym.tex}}\!\! : T[T[\SpaceH]] \rightarrow T[\SpaceH] \nonumber \\
	\eta_\SpaceH &:= \id{\SpaceH} \otimes \!\hbox{\input{symbols/DunitSym.tex}}\!\!: \SpaceH \rightarrow T[\SpaceH] \label{eqns_quantumGroupMonadMultUnit}
\end{align}
The fact that $(T,\mu,\eta)$ is a monad on $\CategoryC$ (in fact it is a \textit{commutative} monad~\cite{Kock1972}) is a direct consequence of the fact that $(\SpaceG,\!\hbox{\input{symbols/DmultSym.tex}}\!\!,\!\hbox{\input{symbols/DunitSym.tex}}\!\!)$ is an internal monoid in $\CategoryC$, and is summarised by the following graphical equations:
\begin{equation}\label{quantumGroupRepMonadDiagram}
	\resizebox{\textwidth}{!}{\input{pictures/chapter3/quantumGroupRepMonadDiagram.tikz}}
\end{equation}

Monads are the category-theoretic way of talking about abstract operations in algebra\footnote{Monads also arise in the context of functional programming \cite{Moggi1991} (albeit with a slightly different interpretation), in modal logic, and in a surprising variety of fields of mathematics.}. We can think of a monad $T$ as embodying the general principles of an algebraic structure, and of its \textit{algebras} $\alpha: T[A] \rightarrow A$ as the concrete realisations of said structure. For example, the \textit{group monad} on $\SetCategory$ sends a set $X$ to the underlying set $F[X]$ of the free group on $X$, and its algebras $\alpha: F[X] \rightarrow X$ are all the possible group structures on $X$. Similar constructions hold for a variety of algebraic theories. 

The monad $(T,\mu,\eta)$ we constructed in Equations \ref{eqns_quantumGroupMonadFunctor} and \ref{eqns_quantumGroupMonadMultUnit} turns out to embody the general structure of representations for a fixed coherent group $\mathbb{G}$ on an object. If $\SpaceH$ is an object of $\CategoryC$, an \textbf{Eilenberg-Moore algebra} (henceforth, \textbf{EM algebra}) of the monad $(T,\mu,\eta)$ is a map $\alpha: T[\SpaceH] \rightarrow \SpaceH$ such that $\alpha \circ T[\alpha] = \alpha \circ \mu_{\SpaceH}$ and $\alpha \circ \eta_{\SpaceH} = \id{\SpaceH}$. Expanding the definitions of $T$, $\mu$ and $\eta$, we see that the defining equations of an EM algebras for $(T,\mu,\eta)$ are nothing but the two commutative diagrams depicted in \ref{quantumGroupRepAlgebraDiagram}: hence the EM algebras for the monad above, which we will henceforth denote by $\CoherentGroupRepMonad{\mathbb{G}}$, are exactly the representations of the coherent group $\mathbb{G} := (\hbox{\input{symbols/ZbwdotSym.tex}}\!\!, \hbox{\input{symbols/DdotSym.tex}}\!\!)$.

Eilenberg-Moore algebras form a category, the \textbf{Eilenberg-Moore category}, which turns out to be really important for our treatment of coherent symmetries:
\begin{itemize}
	\item the objects of the Eilenberg-Moore category are the EM algebras for the monad;
	\item the morphisms $f: \alpha \rightarrow \beta$ in the Eilenberg-Moore category, where $\alpha: \SpaceH \otimes \SpaceG \rightarrow \SpaceH$ and $\beta: \SpaceH' \otimes \SpaceG \rightarrow \SpaceH'$ are EM algebras, are exactly the morphisms $f: \SpaceH \rightarrow \SpaceH'$ in $\CategoryC$ which make the following square commute:
	\begin{equation}\label{EMmorphismCommutingSquare}
		\input{pictures/chapter3/EMmorphismCommutingSquare.tikz}
	\end{equation}	
	\item composition and identities are inherited from $\CategoryC$.
\end{itemize}
The commuting square \ref{EMmorphismCommutingSquare} appearing in the definition of EM morphisms can equivalently be written as the following diagrammatic equation:
\begin{equation}\label{EMmorphism}
	\input{pictures/chapter3/EMmorphism.tikz}
\end{equation}	
But what is its physical meaning? Looking at Equation \ref{EMmorphism} it is pretty clear that EM morphisms are \textbf{equivariant maps} (if thinking of systems with symmetries), or \textbf{intertwiners} (if thinking of representations): they are the processes in the base theory $\CategoryC$ that respect the coherent symmetries that systems have been endowed with. Because of this, the Eilenberg-Moore category is the natural environment to talk about systems with coherent symmetry given by a fixed coherent group $\mathbb{G}$: since the latter were defined as the \textit{representations} of $\mathbb{G}$, we will refer to the Eilenberg-Moore category as $\RepCategory{\mathbb{G}}$.

Just like any other Eilenberg-Moore category, $\RepCategory{\mathbb{G}}$ comes with a forgetful functor $\RepCategory{\mathbb{G}} \rightarrow \CategoryC$ sending a representation (i.e. a symmetric system) $\alpha: \SpaceH \otimes \SpaceG \rightarrow \SpaceH$ to the \textbf{underlying system} (i.e. a system without symmetry) $\SpaceH$, and acting as the identity on processes. The forgetful functor comes with a left adjoint $\CategoryC \rightarrow \RepCategory{\mathbb{G}}$, sending a system $\SpaceH$ to the \textbf{free representation} $\id{\SpaceH} \otimes \!\hbox{\input{symbols/DmultSym.tex}}\!\! : (\SpaceH \otimes \SpaceG) \otimes \SpaceG \rightarrow \SpaceH \otimes \SpaceG$, and acting as the identity on morphisms. 

Because in this work we are concerned with \textit{unitary} representations of a coherent group, we will restrict our attention to the full sub-category of $\RepCategory{\mathbb{G}}$ given specified by the unitary representations: we will refer to this subcategory as the \textbf{unitary Eilenberg-Moore category}, and denote it by $\UnitaryRepCategory{\mathbb{G}}$.

\subsection{Symmetry-observable duality}

In the first part of this section, we have defined a system with coherent symmetry given by a coherent group $\mathbb{G}$ to be a representation $\alpha$ of $\mathbb{G}$. The representation $\alpha: \SpaceH \otimes \SpaceG \rightarrow \SpaceH$ generalises the coherent multiplication $\!\hbox{\input{symbols/DmultSym.tex}}\!\!: \SpaceG \otimes \SpaceG \rightarrow \SpaceG$, which endows the space of wavefunctions over the classical group $\underlyingGroup{\mathbb{G}}$ with the symmetry corresponding to the regular representation. We know that the regular representation $\!\hbox{\input{symbols/DmultSym.tex}}\!\!$ is tightly related to the momentum observable $\hbox{\input{symbols/DdotSym.tex}}\!\!$ of the coherent group $\mathbb{G}$, so a natural question arises: does every representation $\alpha: \SpaceH \otimes \SpaceG \rightarrow \SpaceH$ always yield some sort of momentum measurement for the underlying system $\SpaceH$? The answer turns out to be \inlineQuote{yes}, but first we need to understand what the coherent counterpart of a non-demolition measurement looks like.

\begin{definition}
Let $\CPStarCategory{\CategoryC}$ be an $R$-probabilistic CP* category, and let $\hbox{\input{symbols/ZbwdotSym.tex}}\!\!$ be a $\dagger$-SCFA with enough classical states, and finitely many so (i.e. $(\SpaceG,\hbox{\input{symbols/ZbwdotSym.tex}}\!\!)$ is a classical system). Then a \textbf{non-demolition measurement} on a system $\SpaceH$ with outcomes in $(\SpaceG,\hbox{\input{symbols/ZbwdotSym.tex}}\!\!)$ is a process $m$ in the following form:
\begin{equation}\label{nonDemolMeas}
	\input{pictures/chapter3/nonDemolMeas.tikz}
\end{equation}	
which satisfies the following three requirements:
\begin{equation}\label{nonDemolMeas1}
	\input{pictures/chapter3/nonDemolMeas1.tikz}
\end{equation}	
\begin{equation}\label{nonDemolMeas2}
	\input{pictures/chapter3/nonDemolMeas2.tikz}
\end{equation}	
\begin{equation}\label{nonDemolMeas3}
	\input{pictures/chapter3/nonDemolMeas3.tikz}
\end{equation}	
The associated \textbf{demolition measurement} takes the following form:
\begin{equation}\label{demolMeas}
	\input{pictures/chapter3/demolMeas.tikz}
\end{equation}	
\end{definition}

It's not hard to show that non-demolition measurements always take the familiar form of classically-indexed families of projectors.
\begin{lemma}\label{lem_NDmeasIsOrthogonalFamily}
Let $\CPStarCategory{\CategoryC}$ be an $R$-probabilistic CP* category, let $\hbox{\input{symbols/ZbwdotSym.tex}}\!\!$ be a $\dagger$-SCFA on an object $\SpaceG$ in $\CategoryC$ having enough classical points, and assume that $\classicalStates{\hbox{\input{symbols/ZbwdotSym.tex}}\!\!}$ is finite. Then the non-demolition measurements $m: \SpaceH \rightarrow \SpaceH \otimes (\SpaceG,\hbox{\input{symbols/ZbwdotSym.tex}}\!\!)$ are exactly those taking the following form:
\begin{equation}\label{nonDemolMeasExplicit}
	\input{pictures/chapter3/nonDemolMeasExplicit.tikz}
\end{equation}	
where $(P_x)_{x \in \classicalStates{\hbox{\input{symbols/ZbwdotSym.tex}}\!\!}}$ is a complete family of orthogonal projectors in $\CategoryC$ (i.e. we have $P_x P_x = P_x$, $P_x^\dagger = P_x$, $P_x P_y = 0$ for $y \neq x$ and $\sum_{x \in \classicalStates{\hbox{\input{symbols/ZbwdotSym.tex}}\!\!}} \CPMdoubled{P_x}$ is normalised\footnote{When $\CategoryC$ possesses linear structure compatible with that of $\CPMCategory{\CategoryC}$, this is the same as the familiar completeness requirement $\sum_{x} P_x = \id{\SpaceH}$.}).
\end{lemma}
\begin{proof}
We begin by defining putative projectors $P_x$ from a non-demolition measurement, and show that they indeed satisfy the requirements for a complete orthogonal family of projectors. We define $P_x$ by evaluating against an individual classical outcome $x \in \classicalStates{\hbox{\input{symbols/ZbwdotSym.tex}}\!\!}$ for the non-demolition measurement $m$:
\begin{equation}\label{nonDemolMeasExplicitProof1}
	\resizebox{\textwidth}{!}{\input{pictures/chapter3/nonDemolMeasExplicitProof1.tikz}}
\end{equation}	
The idempotence and orthogonality of projectors follow from the idempotence requirement for non-demolition measurements:
\begin{equation}\label{nonDemolMeasExplicitProof2}
	\resizebox{\textwidth}{!}{\input{pictures/chapter3/nonDemolMeasExplicitProof2.tikz}}
\end{equation}	
The completeness of the family of projectors follows from the normalisation requirement for non-demolition measurements:
\begin{equation}\label{nonDemolMeasExplicitProof3}
	\resizebox{\textwidth}{!}{\input{pictures/chapter3/nonDemolMeasExplicitProof3.tikz}}
\end{equation}	
The self-adjointness of the projectors follows from the self-adjointness requirement for non-demolition measurements:
\begin{equation}\label{nonDemolMeasExplicitProof4}
	\resizebox{\textwidth}{!}{\input{pictures/chapter3/nonDemolMeasExplicitProof4.tikz}}
\end{equation}	
The three equations above can similarly be used to show that the map on the RHS of Equation \ref{nonDemolMeas} defines a non-demolition measurement. 
\end{proof}

In order to figure out what the abstract, coherent counterpart of non-demolition measurements should be, we rewrite the three requirements \ref{nonDemolMeas1}, \ref{nonDemolMeas2} and \ref{nonDemolMeas3} so that the only non-pure maps appearing are preparation in $\hbox{\input{symbols/ZbwdotSym.tex}}\!\!$, measurements in $\hbox{\input{symbols/ZbwdotSym.tex}}\!\!$, and discarding maps:
\begin{equation}\label{nonDemolMeas1alt}
	\input{pictures/chapter3/nonDemolMeas1alt.tikz}
\end{equation}	
\begin{equation}\label{nonDemolMeas2alt}
	\input{pictures/chapter3/nonDemolMeas2alt.tikz}
\end{equation}	
\begin{equation}\label{nonDemolMeas3alt}
	\input{pictures/chapter3/nonDemolMeas3alt.tikz}
\end{equation}	
The equations above inspire the following definition of coherent non-demolition measurements.
\begin{definition}
Let $\CategoryC$ be a $\dagger$-SMC, and $\hbox{\input{symbols/ZbwdotSym.tex}}\!\!$ be a $\dagger$-qSFA on an object $\SpaceG$ of $\CategoryC$, with normalisation factor $N_{\hbox{\input{symbols/ZbwdotSym.tex}}\!\!}$. A $\hbox{\input{symbols/ZbwdotSym.tex}}\!\!$-valued \textbf{coherent non-demolition measurement} on an object $\SpaceH$ of $\CategoryC$ is a process $m: \SpaceH \rightarrow \SpaceH \otimes \SpaceG$ which satisfies the following three requirements:
\begin{equation}\label{nonDemolMeas1coherent}
	\input{pictures/chapter3/nonDemolMeas1coherent.tikz}
\end{equation}	
\begin{equation}\label{nonDemolMeas2coherent}
	\input{pictures/chapter3/nonDemolMeas2coherent.tikz}
\end{equation}	
\begin{equation}\label{nonDemolMeas3coherent}
	\input{pictures/chapter3/nonDemolMeas3coherent.tikz}
\end{equation}	
\end{definition}
\begin{remark}
Note that the isometry requirement involves the normalisation factor $N_{\hbox{\input{symbols/ZbwdotSym.tex}}\!\!}$ for the $\dagger$-qSFA $\hbox{\input{symbols/ZbwdotSym.tex}}\!\!$ on the RHS. This is because in the passage from Equation \ref{nonDemolMeas2} to Equation \ref{nonDemolMeas2coherent} we removed the measurement in $\hbox{\input{symbols/ZbwdotSym.tex}}\!\!$ from the LHS: when $\hbox{\input{symbols/ZbwdotSym.tex}}\!\!$ is a generic $\dagger$-qSCFA, rather than a $\dagger$-SCFA, this means that we also removed the scalar $\frac{1}{N_{\hbox{\input{symbols/ZbwdotSym.tex}}\!\!}}$ that comes with the measurement, and hence the RHS of Equations \ref{nonDemolMeas2alt} and \ref{nonDemolMeas2coherent} will need to carry an extra $N_{\hbox{\input{symbols/ZbwdotSym.tex}}\!\!}$ factor.
\end{remark}
\begin{lemma}\label{lem_coherentNDmeasYieldsNDmeas}
Let $\CPStarCategory{\CategoryC}$ be an $R$-probabilistic CP* category, and let $\hbox{\input{symbols/ZbwdotSym.tex}}\!\!$  be a $\dagger$-SCFA on an object $\SpaceG$ of $\CategoryC$ having enough classical states, and finitely many of them (so that $(\SpaceG,\hbox{\input{symbols/ZbwdotSym.tex}}\!\!)$ is a classical system). If $m: \SpaceH \rightarrow \SpaceH \otimes \SpaceG$ is a $\hbox{\input{symbols/ZbwdotSym.tex}}\!\!$-valued coherent non-demolition measurement on an object $\SpaceH$ of $\CategoryC$, then the following is a non-demolition measurement in $\CPStarCategory{\CategoryC}$:
\begin{equation}\label{nonDemolMeasInLemma}
	\input{pictures/chapter3/nonDemolMeas.tikz}
\end{equation}	
\end{lemma}
\begin{proof}
The proof is straightforward: (i) Equation \ref{nonDemolMeas1coherent} for the coherent non-demolition measurement implies Equation \ref{nonDemolMeas1alt}, which is equivalent to Equation \ref{nonDemolMeas1} for the non-demolition measurement; (ii) Equation \ref{nonDemolMeas2coherent} for the coherent non-demolition measurement implies Equation \ref{nonDemolMeas2alt}, which is equivalent to Equation \ref{nonDemolMeas2} for the non-demolition measurement; (iii) Equation \ref{nonDemolMeas3coherent} for the coherent non-demolition measurement implies Equation \ref{nonDemolMeas3alt}, which is equivalent to Equation \ref{nonDemolMeas3} for the non-demolition measurement.
\end{proof}
\begin{remark}
In Lemmas \ref{lem_NDmeasIsOrthogonalFamily} and \ref{lem_coherentNDmeasYieldsNDmeas} we have restricted our attention to $\dagger$-SCFAs, instead of considering more general $\dagger$-qSCFAs. This is merely for reasons of clarity, and the results hold just as well when $\hbox{\input{symbols/ZbwdotSym.tex}}\!\!$ is a $\dagger$-qSCFA (as long as preparations/measurements are appropriately normalised).
\end{remark}

Having defined coherent non-demolition measurements, we are in a position to answer our original question on the momentum measurement for a symmetric system $\alpha: \SpaceH \otimes \SpaceG \rightarrow \SpaceH$ in $\UnitaryRepCategory{\mathbb{G}}$. We wish to establish that $\alpha^\dagger$ yields the coherent momentum measurement\footnote{Henceforth, we will write \textit{coherent measurement} for \textit{coherent non-demolition measurement}, since demolition measurements are, almost by definition, not coherent (expect in trivial cases).} of a system $\alpha$ with periodic lattice symmetry (i.e. a unitary representation of a doubly well-pointed, doubly finite coherent group). We will do so by providing the following compelling evidence:
\begin{enumerate}
	\item[(i)] we will show (for all coherent groups) that $\alpha^\dagger$ is a $\hbox{\input{symbols/DdotSym.tex}}\!\!$-valued coherent measurement on the system $\SpaceH$\footnote{From this we can conclude (in the case of doubly well-pointed, doubly finite coherent groups) that the non-demolition and demolition measurements associated to $\alpha^\dagger$ have classical outcomes in the set $\classicalStates{\hbox{\input{symbols/DdotSym.tex}}\!\!}$ of possible momenta allowed by the lattice controlling the symmetry.};

	\item[(ii)] we will show (for all coherent groups) that any invariant for the symmetry must commute with $\alpha^\dagger$\footnote{The momentum measurement is indeed expected to have this property for the translation symmetry, just like it had in the case of wavefunctions on the periodic lattice.}; furthermore, we will show (for abelian coherent groups) that $\alpha^\dagger$ is itself an invariant for the symmetry \footnote{And hence so do the associated non-demolition and demolition measurements.};

	\item[(iii)] we will show (for all coherent groups) that states $\psi_{\goodchi}$ of $\SpaceH$ associated with a definite outcome $\goodchi^\dagger \in \classicalStates{\hbox{\input{symbols/DdotSym.tex}}\!\!}$ under $\alpha^\dagger$ transform as expected under the translation symmetry (i.e. translation by $g \in \classicalStates{\hbox{\input{symbols/ZbwdotSym.tex}}\!\!}$ sends $\psi_{\goodchi}$ to itself times the phase $\goodchi \circ g$).
\end{enumerate}
We will colloquially refer to $\alpha^\dagger$ as the \textbf{coherent momentum observable} on a system $\alpha$ with periodic lattice symmetry, generalising the coherent momentum observable $\hbox{\input{symbols/DdotSym.tex}}\!\!$ from the case of wavefunctions on the lattice. We exemplified things in the case of periodic lattices, in which case momentum measurement is appropriate, but when talking about a generic coherent group we will refer to $\alpha^\dagger$ as the \textbf{invariant} associated with the symmetry~$\alpha$.

The fact that the invariant for a symmetric system $\alpha$ can be obtained simply by considering the adjoint $\alpha^\dagger$ is unique to the coherent approach. Indeed, the non-demolition momentum measurement cannot be obtained from the representation of the classical lattice translation group, as they involve measurement/preparation in two complementary observables. In the coherent approach, on the other hand, all the information is preserved, and the distinction between a coherent symmetry $\alpha$ and its invariant $\alpha^\dagger$ is a mere matter of perspective.

\begin{theorem}[\textbf{Symmetry-observable duality}]\label{thm_symmetryObservableDuality}\hfill\\
Let $\mathbb{G} := (\hbox{\input{symbols/ZbwdotSym.tex}}\!\!,\hbox{\input{symbols/DdotSym.tex}}\!\!)$ be a coherent group on an object $\SpaceG$ of a $\dagger$-SMC $\CategoryC$, and let $\alpha: \SpaceH \otimes \SpaceG \rightarrow \SpaceH$ be a unitary representation of $\mathbb{G}$. Then $\alpha^\dagger: \SpaceH \rightarrow \SpaceH \otimes \SpaceG$ is a $\hbox{\input{symbols/DdotSym.tex}}\!\!$-valued coherent measurement on $\SpaceH$.
\end{theorem}
\begin{proof}
We begin by showing that the idempotence condition holds. The first equality below is by unitarity of the representation $\alpha$, while the second equality expands the antipode in its definition and uses the laws of Frobenius algebras: 
\begin{equation}\label{symmetryObservableDualityProof1}
	\resizebox{\textwidth}{!}{\input{pictures/chapter3/symmetryObservableDualityProof1.tikz}}
\end{equation}	
The rightmost diagram above is equal to the leftmost below by the multiplicativity condition for representation $\alpha$. The first equality below is again by unitarity of $\alpha$, and the second is again by definition of the antipode and the laws of Frobenius algebras:  
\begin{equation}\label{symmetryObservableDualityProof2}
	\resizebox{\textwidth}{!}{\input{pictures/chapter3/symmetryObservableDualityProof2.tikz}}
\end{equation}	
This completes the proof of the idempotence condition for $\alpha^\dagger$. We now move on to prove the isometry condition. The first equality below follows from the laws of Frobenius algebras, while the second equality below is by unitarity of the representation $\alpha$:
\begin{equation}\label{symmetryObservableDualityProof3}
	\resizebox{\textwidth}{!}{\input{pictures/chapter3/symmetryObservableDualityProof3.tikz}}
\end{equation}	
The rightmost diagram above is equal to the leftmost diagram below by the multiplicativity condition for representation $\alpha$. The first equality below follows from Hopf law, and the second equality below follows from the unit condition for representation $\alpha$ (together with the fact that $\!\hbox{\input{symbols/ZbwunitSym.tex}}\!\!$ is a $\hbox{\input{symbols/DdotSym.tex}}\!\!$-classical state, with squared norm $N_{\hbox{\input{symbols/DdotSym.tex}}\!\!}$):
\begin{equation}\label{symmetryObservableDualityProof4}
	\resizebox{\textwidth}{!}{\input{pictures/chapter3/symmetryObservableDualityProof4.tikz}}
\end{equation}	
This completes the proof of the isometry condition. Finally, we prove the self-adjointness condition. The first equality below follows from the laws of Frobenius algebras, the second equality by definition of the antipode and the third equality by unitarity of the representation $\alpha$:
\begin{equation}\label{symmetryObservableDualityProof5}
	\resizebox{\textwidth}{!}{\input{pictures/chapter3/symmetryObservableDualityProof5.tikz}}
\end{equation}	
Finally, the antipode is an involution, and hence the rightmost diagram above is nothing but $\alpha$ itself. This completes the proof of the self-adjointness condition.
\end{proof}

When $\CPStarCategory{\CategoryC}$ is $R$-probabilistic and $\mathbb{G}$ is doubly well-pointed and doubly finite, we can write the non-demolition and demolition measurements associated to $\alpha^\dagger$ as follows, with classical outputs in the set $\classicalStates{\hbox{\input{symbols/DdotSym.tex}}\!\!}$:
\begin{equation}\label{symmetryObservableDualityMeasts}
	\input{pictures/chapter3/symmetryObservableDualityMeasts.tikz}
\end{equation}

\begin{theorem}[\textbf{Symmetry-invariant duality}]\label{thm_symmetryInvariantDuality}\hfill\\
Let $\mathbb{G} := (\hbox{\input{symbols/ZbwdotSym.tex}}\!\!,\hbox{\input{symbols/DdotSym.tex}}\!\!)$ be a coherent group on an object $\SpaceG$ of a $\dagger$-SMC $\CategoryC$, and let $\alpha: \SpaceH \otimes \SpaceG \rightarrow \SpaceH$ be a unitary representation of $\mathbb{G}$. Any invariant $\Phi$ of $\alpha$ (and in particular every intertwiner $\Phi: \alpha \rightarrow \alpha$) must commute with~$\alpha^\dagger$: 
\begin{equation}\label{symmetryMaximalInvariant}
	\input{pictures/chapter3/symmetryMaximalInvariant.tikz}
\end{equation}	
Furthermore, if $\mathbb{G}$ is abelian, then $\alpha^\dagger: \SpaceH \rightarrow \SpaceH \otimes \SpaceG$ is itself an invariant for the symmetry $\alpha$:
\begin{equation}\label{symmetryInvariant}
	\input{pictures/chapter3/symmetryInvariant.tikz}
\end{equation}	
\end{theorem}
\begin{proof}
We begin by proving that any invariant $\Phi$ for $\alpha$ must commute with $\alpha^\dagger$. The first and last equalities below are by unitarity of the representation $\alpha$, while the central equality uses the hypothesis of invariance of $\Phi$:
\begin{equation}\label{symmetryMaximalInvariantProof}
	\resizebox{\textwidth}{!}{\input{pictures/chapter3/symmetryMaximalInvariantProof.tikz}}
\end{equation}	
This concludes the proof that any $\Phi$ invariant for $\alpha$ must commute with $\alpha^\dagger$. We then prove that, if $\mathbb{G}$ is an abelian coherent group (i.e. if $\!\hbox{\input{symbols/DmultSym.tex}}\!\!$ is commutative), then $\alpha^\dagger$ itself must be an invariant for the symmetry $\alpha$. The first equality below is by unitarity of the representation $\alpha$, while the second equality uses the multiplicativity condition: 
\begin{equation}\label{symmetryInvariantProof1}
	\resizebox{\textwidth}{!}{\input{pictures/chapter3/symmetryInvariantProof1.tikz}}
\end{equation}	
We use commutativity of $\!\hbox{\input{symbols/DmultSym.tex}}\!\!$ to obtain the leftmost diagram below from the topmost above. The first equality below is again by the multiplicativity condition, and the second equality is again by unitarity:
\begin{equation}\label{symmetryInvariantProof2}
	\resizebox{\textwidth}{!}{\input{pictures/chapter3/symmetryInvariantProof2.tikz}}
\end{equation}	
This concludes the proof that, when $\mathbb{G}$ is commutative, $\alpha^\dagger$ is an invariant for the symmetry $\alpha$. 
\end{proof}

The requirement that the coherent group is abelian for $\alpha^\dagger$ to be an invariant for $\alpha$ is a necessary one. Indeed, the multiplication $\!\hbox{\input{symbols/DmultSym.tex}}\!\!$ is itself a unitary representation, and we have the following consequence of Equation \ref{symmetryInvariant} holding for all $\alpha$:
\begin{equation}\label{symmetryInvariantWavefunctions}
	\input{pictures/chapter3/symmetryInvariantWavefunctions.tikz}
\end{equation}	
Using the equation above, we can prove that $\!\hbox{\input{symbols/DmultSym.tex}}\!\!$ is in fact commutative:
\begin{equation}\label{symmetryInvariantWavefunctionsProof}
	\resizebox{\textwidth}{!}{\input{pictures/chapter3/symmetryInvariantWavefunctionsProof.tikz}}
\end{equation}	

\begin{theorem}[\textbf{Invariant states}]\label{thm_invariantStates}\hfill\\
Let $\mathbb{G} := (\hbox{\input{symbols/ZbwdotSym.tex}}\!\!,\hbox{\input{symbols/DdotSym.tex}}\!\!)$ be a coherent group on an object $\SpaceG$ of a $\dagger$-SMC $\CategoryC$, and let $\alpha: \SpaceH \otimes \SpaceG \rightarrow \SpaceH$ be a unitary representation of $\mathbb{G}$. Let $\Psi$ be a state of $\SpaceH$ associated with a definite outcome $\goodchi^\dagger \in \classicalStates{\hbox{\input{symbols/DdotSym.tex}}\!\!}$ of the coherent momentum measurement $\alpha^\dagger$, i.e. one such that $\alpha^\dagger \Psi$ separates as $\Psi' \otimes \goodchi^\dagger$ for some state $\Psi'$: 
\begin{equation}\label{invariantState}
	\input{pictures/chapter3/invariantState.tikz}
\end{equation}	
Then we must necessarily have $\Psi = \Psi'$. Furthermore, $\Psi$ transforms as follows under the symmetry action $\alpha$:
\begin{equation}\label{invariantStateTransformation}
	\input{pictures/chapter3/invariantStateTransformation.tikz}
\end{equation}	
In terms of the classical action $\big(U_g := \alpha \circ (\emptyArg \otimes g)\big)_{g \in \underlyingGroup{\mathbb{G}}}$ of the underlying group, we have $U_g \Psi = \goodchi(g) \Psi$. Conversely, any state $\Psi$ satisfying the transformation law of Equation \ref{invariantStateTransformation} is associated to a definite outcome $\goodchi^\dagger$ of $\alpha^\dagger$, as in Equation \ref{invariantState}.
\end{theorem}
\begin{proof}
We begin by proving that $\Psi = \Psi'$. The first equality below is by the isometry condition for $\alpha^\dagger$, and the second equality by unitarity of $\alpha$:
\begin{equation}\label{invariantStateEqualityProof1}
	\resizebox{\textwidth}{!}{\input{pictures/chapter3/invariantStateEqualityProof1.tikz}}
\end{equation}	
The leftmost diagram below is obtained from the rightmost diagram above by the idempotence condition for $\alpha^\dagger$. The first equality below is by Equation \ref{invariantState}, the second equality is by Hopf's law, and the third equality follows because $\goodchi(\!\hbox{\input{symbols/DunitSym.tex}}\!\!) = 1$ and $\!\hbox{\input{symbols/ZbwcounitSym.tex}}\!\! \circ \!\hbox{\input{symbols/ZbwunitSym.tex}}\!\! = N_{\hbox{\input{symbols/DdotSym.tex}}\!\!}$:
\begin{equation}\label{invariantStateEqualityProof2}
	\resizebox{\textwidth}{!}{\input{pictures/chapter3/invariantStateEqualityProof2.tikz}}
\end{equation}	
This completes the proof that $\Psi = \Psi'$. To prove the transformation law of Equation \ref{invariantStateTransformation}, we use Equation \ref{invariantState} together with the fact that $\Psi = \Psi'$ and unitarity of the representation $\alpha$:
\begin{equation}\label{invariantStateTransformationProof}
	\resizebox{\textwidth}{!}{\input{pictures/chapter3/invariantStateTransformationProof.tikz}}
\end{equation}
The proof of the converse statement goes along the exact same lines. 
\end{proof}

\subsection{Stone's theorem revisited}
\label{subsection_StoneTheoremRevisited}

The standard result relating momentum to the translation symmetry of 1-dimensional wavefunctions is known as \textbf{Stone's theorem on 1-parameter unitary groups}: it states that the strongly continuous group homomorphisms $x \mapsto U_x$ (the 1-parameter unitary groups) from the additive reals $(\reals,+,0)$ to the unitary operators $\UnitaryOps{\SpaceH}$ over some separable Hilbert space $\SpaceH$ are exactly those in the form $U_x = \exp [i x \textbf{p}]$ for some (not necessarily bounded) self-adjoint operator $\textbf{p}$ on $\SpaceH$ (the traditional momentum observable). In Theorems \ref{thm_symmetryObservableDuality} and \ref{thm_symmetryInvariantDuality}, we saw that the relationship between translation and the momentum observable on periodic lattices is given, in our framework, by adjunction. Throughout this Subsection, we will work in the standard QM formalisms: our aim will be to recast Stone's Theorem for 1-dimensional wavefunctions in a form explicitly compatible with our formulation, i.e. one not involving an infinitesimal generator $\textbf{p}$.

\begin{theorem}[\textbf{Stone's Theorem} \cite{Stone1932}]\label{thm_StoneThmOneParamUGroups}\hfill\\
	Let $x \mapsto U_x$ be a strongly continuous group homomorphism $\reals \rightarrow \UnitaryOps{\SpaceH}$, where $\SpaceH$ is any Hilbert space. Then there exists a unique self-adjoint operator $\textbf{p}: \SpaceH \rightarrow \SpaceH$, not necessarily bounded, such that $U_x = \exp[i x \textbf{p} ]$ for all $x \in \reals$.
\end{theorem}

\begin{theorem}[\textbf{Spectral Theorem} \cite{Hall2013}]\label{thm_SpectralTheorem}\hfill\\
	Let $\textbf{p}: \SpaceH \rightarrow \SpaceH$ be a self-adjoint operator. Then there is a measurable space $Z$, a measure $\mu$ and a unitary isomorphism $V: \SpaceH \rightarrow \Ltwo{Z,\mu}$ such that $\textbf{p}' := V \textbf{p} V^\dagger$ is a multiplication operator:
	\begin{align*}
		\textbf{p}' : 	\Ltwo{Z,\mu} & \rightarrow 	\Ltwo{Z,\mu}\\
				\psi		 & \mapsto 	 	(z \mapsto p_z \psi(z)) \numberthis \label{eqn_SpectralTheorem} 
	\end{align*}
	We will refer to the measurable function $p: Z \rightarrow \reals$ as the \textbf{spectrum} of the operator $\textbf{p}$. If $\textbf{p}$ is bounded then $p$ is essentially bounded and we have $\Norm{}{\textbf{p}'} = \Norm{\infty}{p}$.
\end{theorem}

This is the usual way to derive the momentum spectrum for 1-dimensional wavefunctions: unfortunately, it turns out not to be canonical. This may seem a merely categorical flaw, but it is in fact related to an important physical fact: valuing momentum in the reals is necessarily subject to a choice of units of measurement.

In the course of this Section, we have established that the canonical space for the momenta associated with a symmetric space (governed by some classical abelian group symmetry $G$) is given by the Pontryagin dual $G^\wedge$: any attempt to faithfully value momentum in some other space $K$ is equivalent to a choice of group isomorphism $G^\wedge \isom K$ for some $K$. Similarly, when the translation symmetry is governed by $G = \reals$, we expect any valuation of the momentum in $K = \reals$ to be conditional on some choice of units of measurement, i.e. on fixing some isomorphism $\reals^\wedge \isom \reals$.

Units of measurement, seen as (continuous) group isomorphisms $G^\wedge \stackrel{\isom}{\rightarrow} K$, form a homogeneous space under (transitive and faithful) left regular action of the group automorphisms of $K$. The action corresponds to changing units, and for $K = \reals$ this is the usual multiplication by some non-zero real number. Thus the momentum operator and its spectrum obtained from Theorems \ref{thm_StoneThmOneParamUGroups} and \ref{thm_SpectralTheorem} are subject to an underlying choice of units of measurement $\reals^\wedge \stackrel{\isom}{\longrightarrow} \reals$: Lemma \ref{thm_StoneThmOneParamUGroups2} below make this statement precise.

\begin{lemma}\label{thm_StoneThmOneParamUGroups2}
	The continuous isomorphisms $\Isoms{\AbCategory}{\reals^\wedge}{\reals}$ form a homogeneous space under the (faithful and transitive) left regular action of $\Automs{\AbCategory}{\reals}$. Also $\Automs{\AbCategory}{\reals} \isom_{\AbCategory} (\reals^\times,\cdot,1)$, where $\alpha_c := x \mapsto c \cdot x $ is the continuous automorphism corresponding to a non-zero real $c$.
	As a consequence, the bijection of Theorem \ref{thm_StoneThmOneParamUGroups} is non-canonical, and there is instead a homogeneous space of bijections $U_x = \exp[i x \frac{1}{\hbar} \textbf{p} ]$ between strongly continuous group homomorphisms $(U_x)_{x \in \reals}$ and self-adjoint operators $\textbf{p}$, with fiber isomorphic to the homogeneous space $\Isoms{\AbCategory}{\reals^\wedge}{\reals}$ (except at the singular point $(U_x)_x = (\id{\SpaceH})_x$). Singling out one such bijection is equivalent to fixing a choice of isomorphism $\reals^\wedge \isom \reals$.
\end{lemma}
\begin{proof}
	The first two observations are standard checks. To see that the bijections form a homogeneous space, all we have to show is that there is an action of $\Automs{\AbCategory}{\reals}$ on them: the action of a $\frac{\hbar}{\hbar'} : \reals^\times$ on the space of bijections is given as follows: 
	\begin{equation}\label{eqn_StoneThmv2HomogeneousSpaceBijs}
		\frac{\hbar}{\hbar'} : U_x = \exp[i x \frac{1}{\hbar} \textbf{p} ] \mapsto U_x = \exp[i x \frac{1}{\hbar'} \textbf{p} ]
	\end{equation}
\end{proof}

Taking Theorems \ref{thm_StoneThmOneParamUGroups} and \ref{thm_SpectralTheorem} together, the \textbf{momentum spectrum} for a unitary symmetry $(U_x)_{x \in \reals}$ is usually defined to be $p: Z \rightarrow \reals$. However, it is a consequence of Lemma \ref{thm_StoneThmOneParamUGroups2} that this momentum spectrum is non-canonical, depending instead on a particular choice of unit of measurement: we will denote by $p^\hbar$ the spectrum associated with a particular bijection $U_x = \exp[i x \frac{1}{\hbar} \textbf{p} ]$. We can, however, define a canonical energy spectrum $\hat{p}: Z \rightarrow \reals^\wedge$.

\begin{theorem}[\textbf{Canonical energy spectrum}]\label{thm_CanonicalSpectrum}\hfill\\
	Let $(U_x)_{x\in\reals}$ be a strongly continuous group homomorphism $\reals \rightarrow \UnitaryOps{\SpaceH}$. Fix a bijection $U_x = \exp[i x \frac{1}{\hbar} \textbf{p} ]$, and obtain the\footnote{The decomposition is not really unique. However, the same $Z$ works for all $\hbar$, and the construction of $p^\hbar$ is contravariantly functorial with respect to the choice of $Z$. So we shall not worry about this any further.} spectral decomposition with $V: \SpaceH \rightarrow \Ltwo{Z,\mu}$ and $p^\hbar: Z \rightarrow \reals$. Define $\hat{p} : \Ltwo{Z,\mu} \rightarrow \reals^\wedge$ by:
	\begin{equation}\label{eqn_CanonicalEnergySpectrumDef}
		z  \mapsto (x \mapsto \exp[i \frac{1}{\hbar} p^\hbar_z x])
	\end{equation}
	Then $\hat{p}$ is independent of the choice of $\hbar \in \reals^\times$ (i.e. it is canonical) and we shall refer to it is as the \textbf{canonical momentum spectrum} of $(U_x)_x$.
\end{theorem}
\begin{proof}
	The action defined in Equation \ref{eqn_StoneThmv2HomogeneousSpaceBijs} sends $p^\hbar$ to $p^{\hbar'} = \frac{\hbar'}{\hbar} p^\hbar$. Thus Equation \ref{eqn_CanonicalEnergySpectrumDef} is invariant under the action of $\frac{\hbar}{\hbar'} \in \reals^\times$.
\end{proof}

\begin{remark}[\textbf{Non-demolition Momentum Measurement?}]\hfill\\
	Given a symmetric system $(U_x)_{x \in \reals}$ and its canonical momentum spectrum $\hat{p}$, we can \inlineQuote{construct} an operator $\hat{p}: \Ltwo{Z,\mu} \rightarrow \Ltwo{Z,\mu} \tensor \Ltwo{\reals^\wedge}^\star$ similar to the coherent non-demolition measurement by using delta functions:
	\begin{equation} \label{eqn_DynamicsSpectraReprise}
		\hat{p}: \int_Z a_z \ket{z} d\mu(z) \mapsto \int_Z a_z \ket{z}\tensor \ket{\hat{p}_z} d\mu(z)
	\end{equation}
 	We denoted by $\ket{z}$ the delta function at $z \in Z$ and by $\ket{\hat{p}_z}$ the delta function at $\hat{p}_z \in \reals^\wedge$. Subject to a choice $f: \reals^\wedge \stackrel{\isom}{\rightarrow} \reals$ of units of measurement, we can also \inlineQuote{recover} the momentum operator of Theorem \ref{thm_StoneThmOneParamUGroups} from the non-demolition momentum measurement \inlineQuote{constructed} above:
	\begin{equation} \label{eqn_TraditionalHamiltonian}
		V\textbf{p}V^\dagger = \left(\id{\Ltwo{Z,\mu}} \tensor \int_{\reals^\wedge} f(\chi) \bra{\chi} d\chi\right) \circ \hat{p}
	\end{equation}
	The operator $\int_{\reals^\wedge} f(\chi) \bra{\chi} d\chi$ is nothing but $f$ extended linearly on the basis of delta functions for $\reals^\wedge$.
\end{remark}

The non-demolition Hamiltonian above, however, is not fully rigorous, and we need take a different road to link Stone's Theorem with our periodic lattice symmetries. Recasting the results in terms of projection-valued measures provides a viable alternative.

\begin{lemma}\label{lemma_ProjectionValuedSpectrum}
	Let $X,Y$ be measurable spaces (with sigma-algebras $\Sigma_X$ and $\Sigma_Y$), $\mu$ a measure on $X$ and $f: X \rightarrow Y$ measurable. Then $f$ determines a projection-valued measure $\pi_f: \Sigma_Y \rightarrow \Bounded{\Ltwo{X,\mu}}$ by:
	\begin{equation}\label{eqn_ProjectionValuedSpectrum}
		\pi_f(U) = \text{projection onto subspace } \Ltwo{f^{-1}(U),\mu}
	\end{equation}
	for all $U \in \Sigma_Y$. If $V: \SpaceH \rightarrow \Ltwo{X,\mu}$ is a unitary, then $\pi_f$ can be seen (giving $V$ as understood) as a projection valued measure $\Sigma_Y \rightarrow \Bounded{\SpaceH}$ by considering $V^\dagger \pi_f V$.
\end{lemma}

\begin{theorem}[\textbf{Spectral Theorem, projection-valued}]\label{thm_SpectralTheoremProj}\hfill\\
	Let $\textbf{p}: \SpaceH \rightarrow \SpaceH$ be a self-adjoint operator. Let $V: \SpaceH \rightarrow \Ltwo{Z,\mu}$ and spectrum $p: Z \rightarrow \reals$ be given by Theorem \ref{thm_SpectralTheorem}. If $\pi_p$ is the projection-valued measure defined by Lemma \ref{lemma_ProjectionValuedSpectrum}, then we can reconstruct $\textbf{p}$ as:
	\begin{equation}
		\textbf{p} = \int_\reals \! \lambda \, d\pi_p(\lambda)
	\end{equation}
\end{theorem}

\begin{theorem}[\textbf{Stone's Theorem, projection-valued}]\label{thm_StoneThmOneParamUGroupsProjValued}\hfill\\
	Let $(U_x)_{x \in \reals}$ be a strongly continuous group homomorphism $\reals \rightarrow \UnitaryOps{\SpaceH}$. Let $V: \SpaceH \rightarrow \Ltwo{Z,\mu}$ unitary isomorphism and $\hat{p}: Z \rightarrow \reals^\wedge$ canonical momentum spectrum be given by Theorem \ref{thm_CanonicalSpectrum}. If $\pi_{\hat{p}}$ is the projection-valued measure defined by Lemma \ref{lemma_ProjectionValuedSpectrum}, then we can reconstruct $(U_x)_x$ as:
	\begin{equation}
		U_x = \int_{\reals^\wedge} \!\! \chi(t) \, d\pi_{\hat{p}}(\chi)
	\end{equation}		
\end{theorem}

Finally, the form of Stone's theorem on 1-parameter unitary groups given by Theorem \ref{thm_StoneThmOneParamUGroupsProjValued} can be extended to the periodic lattice symmetries described in this work, remembering that a symmetry $\alpha$ for a well-pointed abelian\footnote{Recall that well-pointed abelian coherent groups in $\fdHilbCategory$ are always doubly well-pointed and doubly finite.} coherent group $\mathbb{G}$ on a finite-dimensional Hilbert space $\SpaceH$ corresponds to a (necessarily strongly continuous) group homomorphisms $G \rightarrow \UnitaryOps{\SpaceH}$, where $G$ is the (finite abelian) underlying group.

\begin{theorem}[\textbf{Canonical momentum spectrum, finite abelian groups}]\label{thm_CanonicalSpectrumFinite}\hfill\\
	Let $(U_g)_{g \in G}$ be the strongly continuous group homomorphism $G \rightarrow \UnitaryOps{\SpaceH}$ corresponding to a representation $\alpha: \SpaceH \otimes \SpaceG \rightarrow \SpaceH$  of a doubly well-pointed abelian coherent group $\mathbb{G}$ in $\fdHilbCategory$, with $G := \underlyingGroup{\mathbb{G}}$ as its (finite abelian) underlying group. Let $\hat{p}:= \alpha^\dagger: \SpaceH \rightarrow \SpaceH \tensor \SpaceG$, and $Z$ be an orthonormal basis of eigenvalues for $\hat{p}$. Let $V: \SpaceH \rightarrow \Ltwo{Z,\mu}$ be the unitary corresponding to the basis, and define the canonical momentum spectrum $\hat{p}:Z \rightarrow G^\wedge$ via the multiplicative character basis of $\SpaceG$:
	\begin{equation}
		\hat{p}(\ket{z}) := \left[\left(\bra{z} \tensor \id{\SpaceG}\right)\hat{p} \ket{z}\right]^\dagger
	\end{equation} 
	Then the projection-valued measure $\pi_{\hat{p}}$ is independent of the choice of basis\footnote{Seen as having projections in $\Bounded{\SpaceH}$, its correct form would be $V^\dagger \pi_{\hat{p}} V$, where $p$ is dependent on the choice of $V$. The statement here is that the entire expression $V^\dagger \pi_{\hat{p}} V$ is independent of the choice of $V$.} and coincides with the complete family of orthogonal projectors defined by $\hat{p}$.
\end{theorem}

The measure provided by Theorem \ref{thm_CanonicalSpectrumFinite} can be extended linearly to obtain the invariant $\alpha^\dagger$ for the well-pointed abelian coherent group $\complexs[G]$, and similarly $(U_g)_{g \in G}$ can be extended linearly to obtain the representation $\alpha$. In the last section of this Chapter, we will see that this provides a direct link between the (finite abelian groups case of) Stone's Theorem and the symmetry/observable duality results for coherent groups presented in this Section.

\newpage
\section{Infinite-dimensional CQM}
\label{section_compactAbelian}

Throughout the past decade, the framework of CQM has achieved remarkable success in describing the foundations of finite-dimensional quantum theory, and the structures behind quantum information protocols and quantum computation. Unfortunately, attempts to extend the same techniques to the treatment of infinite-dimensional case have so far achieved limited success. Although the work of \cite{Abramsky2012b} on H$^\star$-algebras provides a characterisation of non-degenerate observables in arbitrary dimensions, the machinery needed to describe coherent groups for separable Hilbert spaces is inevitably lost: strongly complementary pairs do not exist in the category $\sHilbCategory$ of separable Hilbert spaces and bounded linear maps for infinite-dimensional spaces. This is a major issue for our coherent framework: it prevents us from being able to talk about one of the textbook examples of position/momentum pairs in quantum mechanics, that of 1-dimensional wavefunctions with periodic boundary conditions. The reason is simple: the translation symmetry group is the compact abelian Lie group $(\reals/{L\integers},+,0)$ (where $L$ is the length of the underlying space), while the boost symmetry groups is the infinite discrete abelian group $(\integers,+,0)$, and $\sHilbCategory$ doesn't allow us to formulate coherent groups on $\Ltwo{\integers}$ or $\Ltwo{\reals/{L\integers}}$.

In this Section, we resort to non-standard analysis \`{a} la Robinson \cite{Robinson1974} to tackle the issue of infinitesimal and infinite quantities behind unbounded operators, Dirac deltas and plane-waves: these are key ingredients of mainstream quantum mechanics which the categorical framework has thus failed to adequately capture, and we demonstrate how they can be used to recover a great deal of CQM machinery in infinite-dimensions. Applications of non-standard analysis to quantum theory already appeared in the past decades \cite{Ojima1993,Farrukh1975}, but in a different spirit and with different objectives in mind. In Subsection~\ref{subsection_NonStandardAnalysis}, we provide a basic summary of the non-standard techniques we will be using. In Subsection~\ref{subsection_StarHilbCategory}, we construct a category $\starHilbCategory$ of non-standard separable Hilbert spaces, and we relate it to the category $\sHilbCategory$ of standard separable Hilbert spaces and bounded linear maps. In Subsection~\ref{subsection_InfiniteDimCQM}, we use our newly defined category to extend CQM from finite to separable Hilbert spaces, and we treat the textbook case of position and momentum observable for 1-dimensional wavefunctions with periodic boundary conditions. 

The contents of this Section appeared as a standalone work in QPL 2016 \cite{Gogioso2016b}, and we will follow that treatment here. However, please note that the constructions presented have since been generalised by \cite{Gogioso2017} (which includes a new and extended definition of the category $\starHilbCategory$, as well as a number of explicit constructions).

\subsection{Non-standard analysis}
\label{subsection_NonStandardAnalysis}

\subsubsection{Non-standard models}

In this brief introduction to non-standard models, we follow the common lines in the presentations of the original \cite{Robinson1974} and the more recent \cite{Goldblatt1998}. Consider a (first or higher order) theory $\mathbb{T}$, with a standard model $M$: for example, we could consider the theory of natural numbers, with its standard model $\naturals$, or the theory of real numbers, with its standard model $\reals$. We now proceed to outline the \textbf{ultrapower construction}, which is used to produce a non-standard model $\nonstd{M}$ for the theory\footnote{Non-standard models are denoted by a prefix $^\star$, bearing no relation to complex conjugation.}.

Consider the set $M^\naturals$ of all sequences of elements in $M$, and extend all operations and relations of $\mathbb{T}$ to $M^\naturals$ by pointwise definition: any algebraic structure of $M$ transfers to $M^\naturals$ this way, but non-algebraic axioms in $\mathbb{T}$ (such as the existence of inverses in a field) need not transfer. For example, $\reals^\naturals$ is a commutative ring this way, but it is neither totally ordered nor a field. 

Now fix a non-principal ultrafilter $\mathcal{F}$ on the set $\naturals$, and define an equivalence relation $\equiv$ on $M^\naturals$ as follows:
\begin{equation}
(s_n)_{n \in \naturals} \equiv (t_n)_{n \in \naturals} \text{ iff } \suchthat{n \in \naturals}{s_n = t_n} \in \mathcal{F} 
\end{equation}
We can think of $\equiv$ as equating all sequences which \textit{agree almost everywhere (according to the ultrafilter $\mathcal{F}$)}, and we consider the quotient set $\nonstd{M} := M^\naturals / \equiv$ (known as the \textbf{ultrapower}). Because the equivalence relation $\equiv$ is defined in terms of pointwise equality, the operations and relations of $\mathbb{T}$---which we had already extended from $M$ to $M^\naturals$ by pointwise definition---descend to well-defined operations and relations on the quotient set $\nonstd{M}$; for example, relations in the quotient $\nonstd{M}$ hold if and only if their pointwise-defined counterparts hold in $M^\naturals$ almost everywhere (according to $\mathcal{F}$). But a lot more is true: because of the Transfer Theorem (see below), $\nonstd{M}$ is in fact a model of $\mathbb{T}$, which we refer to as the \textbf{non-standard model}. For example, $\nonstd{\integers}$ is a totally ordered ring, $\nonstd{\reals}$ is a totally ordered field, and $\nonstd{\complexs}$ is an algebraically closed field.

\begin{remark}
The non-standard model obtained via the ultrapower construction is not unique\footnote{Nor is it true that all non-standard models of $\mathbb{T}$ need arise this way.}. however, all the statements we will make and results we will prove will rely on the Transfer Theorem (see below), and they will apply to any non-standard model $\nonstd{M}$ obtained via the ultrapower construction. In fact, under the Continuum Hypothesis the choice of $\mathcal{F}$ is entirely irrelevant for the purposes of this work, as all the non-standard models of $\reals$ obtained by the ultrapower construction are isomorphic (and a similar statement applies to $\naturals$, $\integers$ and $\complexs$) \cite{Goldblatt1998}. 
\end{remark}

We can define a structure preserving map $\nonstd{\emptyArg}:M \rightarrow \nonstd{M}$ by sending $x \in M$ to the equivalence class $\nonstd{x} := [(x,x,x,....)] \in \nonstd{M}$ of the constant sequence $(x,x,x,...) \in M^\naturals$: this is a structure-preserving injective mapping (known as a \textbf{universe embedding}), and hence the standard model $M$ is embedded into the non-standard model $\nonstd{M}$. We refer to the elements of $\nonstd{M}$ in the form $\nonstd{x}$ as \textbf{standard}, and we will often freely confuse them with the corresponding elements of $M$ (i.e. we will often simply write $\nonstd{x}$ as $x$, when no confusion can arise). We refer to the elements of $\nonstd{M}$ at large as \textbf{non-standard}: if $M$ is infinite, then not all elements of $\nonstd{M}$ are in the form $\nonstd{x}$, and hence some \inlineQuote{truly non-standard} elements exist.
\begin{example}
Consider the sequence $s:=(n)_{n \in N} \in \naturals^{\naturals}$, and define $\omega := [s] \in \starNaturals$. Now take $m \in \naturals$, and consider $\nonstd{m} := [(m,m,m,....)]$: we have that the subset $\suchthat{n \in \naturals}{ m < s_n} = \suchthat{n \in \naturals}{m < n}$ is in any non-principal ultrafilter $\mathcal{F}$, and hence $m < \omega$. Thus in the non-standard model $\starNaturals$ there is an \textbf{infinite natural} $\omega$ which satisfies $m < \omega$ for all standard natural numbers $m \in \naturals$. 
\end{example}
\begin{example}
Consider the sequence $s:=\big(1/(n+1)\big)_{n \in N} \in \reals^{\naturals}$, and define $\epsilon := [s] \in \starReals$. Now take $0 <x \in \reals$, and consider $\nonstd{x} := [(x,x,x,...)]$: we have that the subset $\suchthat{n \in \naturals}{ 0 < \frac{1}{n+1} \leq x} = \suchthat{n \in \naturals}{n+1 \geq \lceil 1/x \rceil}$ is in any non-principal ultrafilter $\mathcal{F}$, and hence $0 < \epsilon \leq x$. Thus in the non-standard model $\starReals$ there is an \textbf{infinitesimal real} $\epsilon$ which satisfies $0 < \epsilon \leq x$ for all positive standard real numbers $0 < x \in \reals$.   
\end{example}

Now we consider a \textbf{standard} subset $A \subseteq M$, and we construct the a new non-standard subset $\nonstd{A} \subseteq \nonstd{M}$ as follows:
\begin{equation}
[(s_n)_{n \in \naturals}] \in \nonstd{A} \text{ iff } \suchthat{n \in \naturals}{s_n \in A} \in \mathcal{F}
\end{equation}
The set $\nonstd{A}$ contains $A$ as a subset, and we refer to it as the \textbf{enlargement} of $A$. Furthermore, functions $f: A \rightarrow B$ and relations $R \subseteq A \times B$ between standard sets extend to functions $\nonstd{f}: \nonstd{A} \rightarrow \nonstd{B}$ and relations $\nonstd{R} \subseteq \nonstd{A} \times \nonstd{B}$ between the corresponding enlargements; we refer to these as the \textbf{non-standard extensions} of the corresponding standard function $f$ and relation $R$. If $\mathcal{M} := (M,Rel_{\mathcal{M}},Fun_{\mathcal{M}})$ is the full relational structure associated with the standard model\footnote{I.e. the set $Rel_{\mathcal{M}}$ contains all finitary relations on $M$, and the set $Fun_{\mathcal{M}}$ contains all finitary (partial) functions on $M$.}, when talking about the \textbf{non-standard} model we will be considering the following structure:
\begin{equation}   
\nonstd{\mathcal{M}} := \Big(\nonstd{M},\suchthat{\nonstd{R}}{R \in Rel_{\mathcal{M}}}, \suchthat{\nonstd{f}}{f \in Fun_{\mathcal{M}}} \Big)
\end{equation}
When $M$ is infinite, the structure presented above is not full: we will refer to subsets, relations and functions appearing in $\nonstd{\mathcal{M}}$ as \textbf{internal}, and to all other subsets, relations and functions of $\nonstd{M}$ as \textbf{external}.

The fundamental result which relates the standard model $M$ to any non-standard model $\nonstd{M}$ obtained by the ultrapower construction is known as \textbf{Transfer Theorem}. The Transfer Theorem plays a central role in this work: all our proofs are carried out explicitly appealing to it, and are therefore blind to the underlying construction of non-standard models. Consider a sentence $\varphi$ in the language $\mathcal{L}_{\mathcal{M}}$ of the standard model $M$: constants, functions and relations are chosen from those of $\mathcal{M}$, and quantification is on standard subsets of $M$. Define the \textbf{$\ast$-transform} $\nonstd{\varphi}$ of $\varphi$ by replacing each constant $a \in M$ with $\nonstd{a}\in \nonstd{M}$, each relation $R \subseteq A \times B$ with $\nonstd{R} \subseteq \nonstd{A} \times \nonstd{B}$, each function $f: A \rightarrow B$ with $\nonstd{f}: \nonstd{A} \rightarrow \nonstd{B}$, and each standard set $A$ with its enlargement $\nonstd{A}$; in particular, quantification $\forall x \in A$ and $\exists x \in A$ over a standard set $A$ turns into quantification $\forall x \in \nonstd{A}$ and $\exists x \in \nonstd{A}$ over its enlargement. Then $\nonstd{\varphi}$ is a sentence in the language $\mathcal{L}_{\nonstd{\mathcal{M}}}$ of the non-standard model $\nonstd{M}$, and the following result holds.
\begin{theorem}[\textbf{Transfer Theorem}]\hfill\\
A sentence $\varphi$ holds in the standard model $M$ if and only if its $\ast$-transform $\nonstd{\varphi}$ holds in the non-standard model $\nonstd{M}$.
\end{theorem}
\noindent We now present a number of sample applications of the Transfer Theorem to the theory of non-standard naturals and reals.

\begin{example}
Consider the sentence defining predecessors in the natural numbers:
\begin{equation}
\forall n \in \naturals.\, \big[n \neq 0 \Rightarrow [\exists m \in \naturals .\, n = m+1]\big]
\end{equation}
By Transfer Theorem, the following sentence holds in the non-standard model $\starNaturals$:
\begin{equation}
\forall n \in \starNaturals.\, \big[n \neq 0 \Rightarrow [\exists m \in \starNaturals .\, n = m+1]\big]
\end{equation}
Hence all non-zero non-standard naturals have predecessors. 
\end{example}

\begin{example}
Consider the sentence defining the well-order property for the natural numbers, i.e. saying that every non-empty subset of $\naturals$ has a minimum:
\begin{equation}
	\forall A \subseteq \naturals.\, \Big[A \neq \emptyset \Rightarrow \big[ \exists m \in A. \forall a \in A. m \leq a\big] \Big]
\end{equation}
By Transfer Theorem, the following sentence holds in the non-standard model $\starNaturals$:
\begin{equation}
	\forall \nonstd{A} \subseteq \starNaturals. \, \Big[\nonstd{A} \neq \emptyset \Rightarrow \big[ \exists m \in \!\! \nonstd{\!A}. \forall a \in \!\! \nonstd{\!A}. m \leq a\big] \Big]
\end{equation}
Hence all non-empty internal subsets of $\starNaturals$ have a minimum. Now consider the subset $W \subset \! \starNaturals$ of all infinite non-standard naturals, i.e. $W := \suchthat{k \in \!\! \starNaturals}{\forall n \in \naturals . n < k}$. The subset $W$ cannot have a minimum: if $m \in W$ were such a minimum, then $m \neq 0$ and hence $m-1$ would exists (by the previous example); but then $m-1$ would be a standard natural, making $m$ itself standard and not infinite. Because $W$ is non-empty and has no minimum, we infer that it cannot be an internal subset.
\end{example}

\begin{example}
Consider the sentence defining multiplicative inverses in $\reals$:
\begin{equation}
	\forall x \in \reals.\, \big[x \neq 0 \Rightarrow [\exists y \in \reals .\, x \cdot y = 1 ] \big]
\end{equation}
By Transfer Theorem, the following sentence holds in the non-standard model $\starReals$:
\begin{equation}
	\forall x \in \!\! \starReals.\, \big[x \neq 0 \Rightarrow [\exists y \in \starReals .\, x \cdot y = 1 ] \big]
\end{equation}
Hence all non-zero non-standard reals have multiplicative inverses. This reasoning can be applied to all axioms making $\reals$ an ordered field, and hence $\starReals$ is an ordered field, with $\reals$ as a sub-field. In particular, the following holds by Transfer Theorem:
\begin{equation}
\forall x,y \in \!\!\starReals. \, [x \neq 0 \wedge y \neq 0] \Rightarrow [x < y \Rightarrow 1/x > 1/y] 
\end{equation}
Applying this to the infinitesimal real number $\epsilon$ implies that $1/ \epsilon > x$ for all $x \in \reals$, i.e. that $1/\epsilon$ is an infinite non-standard real number.
\end{example}

\begin{example}
Consider the sentence defining the sequence $s: \naturals \rightarrow \reals$ of partial sums for every sequence $f: \naturals \rightarrow \reals$ in the standard model $\reals$:
\begin{align}
\forall f: \naturals \rightarrow \reals . \exists s : \naturals \rightarrow \reals . \big[s(0) = f(0) \wedge [\forall m \in \naturals . s(m+1) = s(m) + f(m+1) ] \big] 
\end{align}
By Transfer Theorem, the following sentence holds in the non-standard model $\starReals$:
\begin{align}
\forall \nonstd{f}: \!\!\starNaturals \rightarrow \!\!\starReals . \exists \nonstd{s} : \!\!\starNaturals \rightarrow \!\!\starReals . \big[\nonstd{s}(0) = \nonstd{f}(0) \wedge [\forall m \in \!\!\starNaturals . \nonstd{s}(m+1) = \nonstd{s}(m) + \nonstd{f}(m+1) ] \big] 
\end{align}
Hence every internal sequence $\nonstd{f}: \starNaturals \rightarrow \starReals$ admits a corresponding internal sequence of partial sums $\nonstd{s}: \starNaturals \rightarrow \starReals$, i.e. the notation $\sum_{n=0}^{m} \nonstd{f}(n)$ is legitimate for all $m \in \starNaturals$. Similarly, $\nonstd{f}$ admits a corresponding internal sequence of partial products $\nonstd{p}$, i.e. the notation $\prod_{n=0}^{m} \nonstd{f}(n)$ is legitimate for all $m \in \starNaturals$. 
\end{example}

\begin{example}
For each $n \in \naturals$ we can define the lower set $n\!\downarrow := \suchthat{m \in \naturals}{m \leq n}$, and given any sequence $f: \naturals \rightarrow \reals$ we can define its truncation $f^{(n)} : n\!\!\downarrow\; \rightarrow \reals$ by setting $f^{(n)}(m) = f(m)$ when $m \leq n$ and leaving $f^{(n)}(m)$ undefined otherwise.  By Transfer Theorem, for each $\kappa \in \starNaturals$ there is a corresponding internal set $\kappa\!\downarrow := \suchthat{m \in \starNaturals}{m \leq \kappa}$, and each internal function $\nonstd{f}: \starNaturals \rightarrow \starReals$ has a corresponding internal truncation $\nonstd{f}^{(\kappa)}: \kappa\!\downarrow\; \rightarrow \starReals$. When talking about the \textbf{non-standard extension} of a standard sequence $f:=(a_n)_{n \in \naturals}$ \textbf{up to an infinite natural $\kappa$} we will mean the truncation $\nonstd{f}^{(\kappa)}: \kappa\!\downarrow\; \rightarrow \reals$, which we simply denote by $(a_n)_{n=0}^{\kappa}$.
\end{example}

\subsubsection{The structure of $\starNaturals$}

The \textbf{non-standard naturals} $\starNaturals$ form a totally ordered semiring, with the \textbf{standard naturals} $\naturals$ as an initial segment. As a totally ordered set, the non-standard naturals are order-isomorphic to $\naturals + \theta \times \integers$, where $\theta$ is a dense order with no maximum nor minimum. We refer to the standard naturals as \textbf{finite naturals}, and to the internal naturals in $\starNaturals - \naturals$ as \textbf{infinite naturals}: this is because any infinite natural $\kappa$ satisfies $\kappa > n$ for all $n \in \naturals$. We say that two non-standard naturals $n,m$ have the same \textbf{order of infinity} if they differ by a finite natural $|n-m| \in \naturals$: this gives an equivalence relation, and the set of equivalence classes is in order-preserving bijection with the totally ordered set $\Theta^+ := \{0\} + \theta$. The set $\Theta^+$ also inherits the additive monoid structure of $\starNaturals$, but not the full semiring structure.

By Transfer Theorem, many properties of $\naturals$ transfer to $\starNaturals$: for example, from the fact that every non-empty set of standard naturals has a minimum we conclude that every non-empty internal set of non-standard naturals also has a minimum, and arguments by induction can be carried out on non-empty internal subsets of $\starNaturals$. If $(a_n)_{n\in\naturals}$ is a sequence of natural numbers in the standard model, then we can consider the unique corresponding standard sequence $(a_n)_{n \in \starNaturals}$ in the non-standard model, coinciding with $(a_n)_{n \in \naturals}$ for all finite naturals. Furthermore, for any $m \in \naturals$ the naturals $s_m := \sum_{n=0}^m a_n$ and $p_m := \prod_{n=0}^m a_n$ exist in the standard model, and hence the non-standard naturals $s_m$ and $p_m$ exist in the non-standard model for all $m \in \starNaturals$.

The \textbf{non-standard integers} $\starIntegers$ similarly relate to the standard integers $\integers$: they form a totally ordered ring, with $\integers$ as a sub-ring and $\starNaturals$ as a sub-semiring. As a totally ordered set, they are order-isomorphic to $(\theta + \{0\} + \theta) \times \integers$: they contain the finite integers together two copies of the infinite naturals, one copy above all finite integers (the \textbf{positive infinities}) and one copy below all finite integers (the \textbf{negative infinities}). The set $\Theta := \theta + \{0\} + \theta$ of orders of infinity for $\starIntegers$ again inherits the total order and the additive group structure, but not the ring one.

\subsubsection{The structure of $\starReals$}
\label{subsubsection_StructureStarReals}

The \textbf{non-standard reals} $\starReals$ form an ordered field, with the \textbf{standard reals} $\reals$ as a sub-field and the non-standard integers $\starIntegers$ as a subring. They are a non-archimedean field, with a sub-ring $M_1$ of \textbf{infinitesimals}, smaller in absolute value than all positive standard reals. The non-zero infinitesimals have inverses, the \textbf{infinite reals}, larger in absolute value than all positive standard integers/reals. 

By using the finite integers $\integers \subset \starIntegers \subset \starReals$, it is possible to define the sub-ring\footnote{In fact, they form a $\reals$-vector subspace of $\starReals$.} $M_0$ of the \textbf{finite reals}, given by those $x \in \starReals$ such that $\exists \; n \in \integers \; |x| < n$. The sub-ring $M_1$ of infinitesimals is a two-sided ideal in $M_0$, and by using Dedekind cuts it is possible to show that the ring quotient $M_0 / M_1$ is isomorphic to $\reals$: we refer to the corresponding surjective ring homomorphism $\stdpart{\emptyArg}: M_0 \rightarrow \reals$ as the \textbf{standard part} (which is the identity on the subring $\reals \leq M_0$), and we denote the corresponding quotient equivalence relation on $M_0$ by $x \simeq y \iffdef |x-y| \text{ is an infinitesimal}$. The coset of $M_1$ surrounding any non-standard real $x \in \starReals$ is called the \textbf{monad} of $x$, and when $x$ is finite it contains exactly one standard real $\stdpart{x} \in \reals$.

The non-standard reals are Archimedean in a non-standard sense: by the transfer theorem, for any $x \in \starReals$ there is a unique $n \in \starNaturals$ such that $n \leq |x| < n+1$.\footnote{Equivalently, for every infinitesimal $\xi \in M_1$ there is a unique non-standard natural $n \in \starNaturals$ such that $1/(n+1) < |\xi| \leq 1/n$.} As a consequence, the non-standard rationals $\starRationals$ are a dense subfield of $\starReals$. Furthermore, the non-standard reals can be obtained in the familiar way by \inlineQuote{gluing} a copy of the (non-standard) unit interval between any two consecutive (non-standard) integers: as a totally ordered additive group, they are then isomorphic to $\Theta \times M_0$. 

Any sequence $(a_n)_{n \in \naturals}$ of reals definable in the standard model has a corresponding non-standard extension $(a_n)_{n \in \starNaturals}$ by the transfer theorem: it coincides with the original sequence on all finite naturals, but will not in general be valued in the standard reals on infinite naturals. It is possible to show that $\lim_{n \rightarrow \infty} a_n = a \in \reals$ in the standard model if and only if $a_n \simeq a$ for all infinite naturals $n$ in the non-standard model. Furthermore, $(a_n)_{n \in \naturals}$ is bounded (say by $|a_n| \leq z \in \reals^+$) in the standard model if and only if $a_n$ is a finite real (with $|a_n| \leq z$) for all infinite naturals.

Real-valued functions $f: I \rightarrow \reals$ in the standard model can similarly be extended by the transfer theorem to real-valued $f: \nonstd{I} \rightarrow \starReals$ in the non-standard model, coinciding with the original function on all standard reals in $\nonstd{I}$. Then $\lim_{x \rightarrow a} f(x) = c$ in the standard model if and only if in the non-standard model we have $f(x) \simeq c$ for all $x \simeq a$ (except perhaps at $x = a$). As a consequence, $f$ is continuous at $a \in I$ in the standard model if and only if in the non-standard model we have $f(x) \simeq f(a)$ whenever $x \simeq a$. 

The \textbf{non-standard complex numbers} $\starComplexs$ similarly extend $\complexs$ with infinitesimals and infinities: they also form a field, with both $\starReals$ and $\complexs$ as sub-fields. As an additive group, they are isomorphic to $\starReals^2$. We will transfer most notations from $\starReals$ to $\starComplexs$, when no confusion can arise.

\subsubsection{Non-standard Hilbert spaces}

The passage from standard to non-standard models has a two-fold effect on (complex) Hilbert spaces: (i) the scalars change from $\complexs$ to $\starComplexs$; (ii) the vectors change from sequences $(a_n)_{n \in \naturals^+}$ indexed by the standard naturals to sequences $(a_n)_{n \in \starNaturals^+}$ indexed by the non-standard naturals. Each standard Hilbert space $V$ has a non-standard counterpart $\nonstd{V}$: the non-standard space $\nonstd{V}$ contains all vectors of $V$, known as the \textbf{standard vectors}, as a $\complexs$-linear (but not $\starComplexs$-linear) subspace. The non-standard space $\nonstd{V}$ comes with a $\starComplexs$-valued inner product (extending the standard one on $V$), and an associated $\starReals^+$-valued norm. 

The vectors infinitesimally close to standard vectors are called \textbf{near-standard vectors}, and the vectors with infinitesimal norm are called \textbf{infinitesimal vectors}: both form $\complexs$-linear (and $M_0$-linear) subspaces $\nonstd{V_0}$ and $\nonstd{V_1}$ of $\nonstd{V}$. There is a $\complexs$-linear map $\stdpart{\emptyArg}: \nonstd{V}_0 \rightarrow V$, known as the \textbf{standard part}, which sends the near-standard vectors surjectively onto $V$, acts as the identity on standard vectors and has the infinitesimal vectors $V_1$ as kernel. The standard part defines an equivalence relation $\simeq$ on near-standard vectors, with $\ket{\psi} \simeq \ket{\phi} $ if and only if $\ket{\psi} - \ket{\phi}$ is an infinitesimal vector. 

An interesting class of non-standard vectors can be obtained by the transfer theorem. Consider a standard complex Hilbert space $V$ which is separable, i.e. comes with a complete orthonormal basis $\ket{e_{n}}_{n \in \naturals^+}$ which is countable\footnote{We index our vectors in the positive naturals $\naturals^+$ for reasons of convenience: this way a generic vector in a $d$-dimensional vector space is written cleanly as $\sum_{n=1}^d v_n \ket{e_n}$.}. If $(\psi_{n})_{n \in \naturals^+}$ is a standard sequence of complex numbers, then the vector $\ket{\psi^{(k)}} := \sum_{n=1}^{k} \psi_n \ket{e_n} \in V$ exists for all positive standard naturals $k \in \naturals^+$: by the transfer theorem, the vector $\ket{\psi^{(\kappa)}}$ exists in $\nonstd{V}$ for any infinite natural $\kappa$, where the corresponding non-standard sequence $(\psi_n)_{n \in \starNaturals^+}$ is used to provide values. In particular, the vector $\sum_{n=1}^{\kappa} \ket{e_n} \in \nonstd{V}$ exists, and has squared norm $\kappa \in \starReals^+$. 

The vectors of finite norm are known as \textbf{finite vectors} and form a $\complexs$-linear (and $M_0$-linear) subspace of $V$. However, this is where the second effect of non-standard analysis on Hilbert spaces comes into play: there exist finite vectors, such as $\ket{\phi} := \frac{1}{\sqrt{\kappa}} \sum_{n=1}^{\kappa} \ket{e_n}$, which are not near-standard. Indeed, any standard vector $\ket{\psi}$ is infinitesimally close to its truncation in the form $\ket{\psi^{(\kappa)}} := \sum_{n=1}^{\kappa} \psi_n \ket{e_n}$, where $\psi_\nu$ is infinitesimal for all infinite naturals $\nu$. We get the following lower bound for the squared norm of the difference $\ket{\phi}-\ket{\psi}$:
\begin{align}
\Big|\Big| \ket{\phi}-\ket{\psi} \Big|\Big|^2 \simeq \Big|\Big| \ket{\phi}-\ket{\psi^{(\kappa)}} \Big|\Big|^2 = &\sum_{n=1}^{\kappa} \frac{|1-\psi_n|}{\kappa}^2 \geq \sum_{n\geq\kappa/M} \frac{|1-\psi_n|^2}{\kappa} \nonumber \\
\geq &\sum_{n\geq\kappa/M} \frac{1-\epsilon}{\kappa} = (1-\epsilon)(1-\frac{1}{M}) \nonumber
\end{align}
for all $M \in \naturals^{+}$ and $\epsilon \in (0,1)$. This means that $\stdpart{\big|\big| \ket{\phi}-\ket{\psi^{(\kappa)}} \big|\big|} \geq 1$ for all standard vectors $\ket{\psi}$, and hence the vector $\ket{\phi}$ is finite but not near-standard. Finite vectors which are not near-standard are genuinely new, and can be used to do genuinely new things. This is what makes the non-standard approach to quantum mechanics so powerful: in $\starReals$ and $\starComplexs$ finite numbers are all near-standard, and correspond to standard numbers under infinitesimal equivalence, while in a non-standard Hilbert space one gets new things for free, such as normalised plane-waves in $\Ltwo{\integers}$ and Dirac-deltas in $\Ltwo{\reals/(L\integers)}$. These will be the fundamental building blocks of our work.

The transfer theorem can similarly be used to define non-standard linear operators (not necessarily continuous/bounded): if $(a_{nm})_{n, m \in \naturals^+}$ is a doubly-indexed sequence (a.k.a. a matrix) of complex numbers, then the linear operator $\sum_{m,n=0}^{\kappa} a_{mn} \, \ket{e_m}\bra{e_n} : \nonstd{V} \,\rightarrow \nonstd{V}$ exists for any infinite natural $\kappa$ (where $(a_{nm})_{n, m \in \starNaturals^+}$ is the unique internal non-standard sequence given by the transfer theorem). This is a remarkable result, but it comes with some tricky limitations which will be presented in the next section.

\subsection{The category $\starHilbCategory$}
\label{subsection_StarHilbCategory}
The main idea behind our construction is to legitimise, through non-standard analysis, notations such as $\sum_{n\in \naturals^+} \ket{e_n}\bra{e_n}$ for the identity operator, $\sum_{n\in\naturals^+} \ket{e_n}$ for the unit of an infinite-dimensional Frobenius algebra, $\sum_{n,m \in \naturals^+}\ket{e_n}a_{nm}\bra{a_m}$ for a general matrix $(a_{nm})_{n,m \in \naturals^+}$. The transfer theorem doesn't allow us to conclude the existence of sums strictly over $\naturals^+$ (nor over the entirety of $\starNaturals$), but it does allow us to sum up to some infinite natural $\kappa$: the sums $\sum_{n=1}^\kappa \ket{e_n}\bra{e_n}$, $\sum_{n=1}^\kappa \ket{e_n}$ and $\sum_{n,m=1}^\kappa \ket{e_n} a_{nm} \bra{e_m}$ all describe well-defined internal linear maps of non-standard Hilbert spaces. Unfortunately,  $P_\kappa := \sum_{n=1}^\kappa \ket{e_n}\bra{e_n}$ does not behave like the identity over the space of all internal linear maps, but rather it is as a subspace projector: in order to turn these projectors into identities, we use a construction similar to that of the Cauchy/idempotent\footnote{Projectors are self-adjoint idempotents.} completion. As it turns out, this procedure preserves all standard bounded operators, and enough non-standard ones to do many of the things we care about in categorical quantum mechanics.

\subsubsection{Definition of the category}
\label{subsubsection_StarHilbCategorydef}

We proceed to define the \textbf{category of non-standard separable Hilbert spaces}\footnote{We have complex Hilbert spaces in mind, but the construction is identical for real Hilbert spaces.}, which we will denote by $\starHilbCategory$. All proofs of results in this and future sections can be found in the Appendix. As objects we take separable (standard) Hilbert spaces together with a witness of separability, i.e. pairs $\SpaceH := \big(V, \ket{e_{n}}_{n=1}^\kappa\big)$ of a standard separable Hilbert space $V$ and a family of vectors $\ket{e_{n}}_{n=1}^\kappa$ defined as follows (for some non-standard natural $\kappa\in \starNaturals$).
\begin{enumerate}
	\item[(i)] If $V$ is finite-dimensional: we consider a finite orthonormal basis $\ket{e_{n}}_{n=1}^\kappa$, where $\kappa := \dim{V} \in \naturals$.
	\item[(ii)] If $V$ is infinite-dimensional: we fix some infinite natural $\kappa \in \,\starNaturals$ (meant to be the non-standard dimension), and we consider the unique extension (by the transfer theorem) up to $\kappa$ of a complete orthonormal basis $\ket{e_{n}}_{n\in \naturals^+}$ for $V$. 
\end{enumerate} 

\noindent For each object $\SpaceH := \big(V, \ket{e_{n}}_{n=1}^\kappa\big)$, let the \textbf{truncating projector} $P_\SpaceH: \SpaceH \rightarrow \SpaceH$ be the following internal linear map $\nonstd{V} \rightarrow \nonstd{V}$, where we refer to $\dim{\SpaceH} := \kappa \in \starNaturals$ as the \textbf{dimension} of object $\SpaceH$:
\begin{equation}\label{eqn_TruncatingProjector}
	P_\SpaceH := \sum_{n=1}^{\dim{\SpaceH}} \ket{e_n} \bra{e_n}.
\end{equation}

\noindent We also use notation $|\SpaceH| := V$ to refer to the standard separable Hilbert space underlying an object $\SpaceH$ of $\starHilbCategory$. The morphisms in the category $\starHilbCategory$ are then defined as follows:
\begin{equation}\label{eqn_Homset}
	\Hom{\starHilbCategory}{\SpaceH}{\SpaceG} := \suchthat{\;P_\SpaceG \circ F \circ P_\SpaceH\;}{\;F:\nonstd{|\SpaceH|} \,\rightarrow\, \nonstd{|\SpaceG|} \text{ internal linear map}}.
\end{equation}

\noindent Because the truncating projectors for $\SpaceH$ and $\SpaceG$ are internal linear maps, the composite $P_\SpaceG \circ F \circ P_\SpaceH$ is an internal linear map $\nonstd{|\SpaceH|} \, \rightarrow \nonstd{|\SpaceG|}$, which we shall denote by $\truncate{F}$. Composition of morphisms in $\starHilbCategory$ is simply composition of internal linear maps
\begin{equation}\label{eqn_Composition}
	\truncate{G} \cdot \truncate{F} := \truncate{G} \circ \truncate{F} = (P_\SpaceG \circ G \circ P_\SpaceH) \circ (P_\SpaceH \circ F  \circ P_\SpaceK) = P_\SpaceG \circ (G \circ  P_\SpaceH \circ F)  \circ P_\SpaceK,
\end{equation}
where we used associativity of composition and idempotence of truncating projectors. Idempotence of the projectors, in particular, means that they provide suitable identity morphisms. Indeed if we define
\begin{equation}
    \id{\SpaceH} := P_\SpaceH \circ \id{\nonstd{|\SpaceH|}} \circ P_\SpaceH = P_\SpaceH \circ P_\SpaceH = P_\SpaceH,
\end{equation}
it is straightforward to check that $\id{\SpaceG} \cdot  \truncate{F} = P_\SpaceG \circ P_\SpaceG \circ F \circ P_\SpaceH = P_\SpaceG \circ F \circ P_\SpaceH = \truncate{F}$, and similarly for $\truncate{F} \cdot \id{\SpaceH}$.

\noindent Now consider two naturals $\kappa, \nu \in \starNaturals$, and define the internal map
\begin{equation}\label{eqn_varsigma}
	\varsigma_{\kappa,\nu}(n,m) := (n-1)\nu + m,
\end{equation}
which is an internal bijection between $\{1,...,\kappa\}\times\{1,...,\nu\}$ and $\{1,...,\kappa \nu\}$. Also, we will simply write $\varsigma(n,m)$ when no confusion can arise. A tensor product can be defined on the objects of $\starHilbCategory$ as follows, with tensor unit $(\complexs,1)$:
\begin{equation}\label{eqn_TensorObjects}
	\Big(V, \ket{e_{n}}_{n=1}^\kappa\Big)
	\otimes \Big(W, \ket{f_{m}}_{m=1}^\nu\Big) 
	:= \Big(V \otimes W, \big(\ket{e_{n}} \otimes \ket{f_{m}}\big)_{\varsigma(n,m)=1}^{\kappa\nu}\Big).
\end{equation}

\noindent In order to define the tensor product on morphisms, we need to first note that morphisms $\truncate{F} : \SpaceH \rightarrow \SpaceG$ in $\starHilbCategory$ are uniquely determined by certain matrices $\indexSet{\SpaceG} \times \indexSet{\SpaceH} \rightarrow\, \starComplexs$:
\begin{equation}\label{eqn_MatrixRepresentation}
	\truncate{F} = P_\SpaceG \circ F \circ P_\SpaceH = 
	\sum_{m=1}^{\dim{\SpaceG}} \sum_{n=1}^{\dim{\SpaceH}} 
	\ket{f_{m}} \Big( \bra{f_{m}} F \ket{e_{n}} \Big) \bra{e_{n}}.
\end{equation}

\noindent We introduce the notation $\truncate{F}_{mn} := \bra{f_m} F \ket{e_n}$, and define the tensor product of two morphisms $\truncate{F} : \SpaceH \rightarrow \SpaceG$ and $\truncate{G} : \SpaceH' \rightarrow \SpaceG'$ to be the familiar tensor product of matrices:
\begin{equation}\label{eqn_TensorMorphisms}
	\truncate{F} \otimes \truncate{G} := 
	\sum_{\varsigma(m,m')=1}^{\dim{\SpaceG}\dim{\SpaceG'}}
	\sum_{\varsigma(n,n')=1}^{\dim{\SpaceH}\dim{\SpaceH'}}
	\ket{f_{m}} \otimes \ket{f'_{m'}}  \; \truncate{F}_{mn} \truncate{G}_{m'n'} \; \bra{e_{n}} \otimes \bra{e'_{n'}}.
\end{equation}

\noindent The map $\truncate{F} \otimes \truncate{G}$ is an internal linear map $\nonstd{|\SpaceH|} \,\otimes\, \nonstd{|\SpaceH'|} \rightarrow \nonstd{|\SpaceG|} \,\otimes\, \nonstd{|\SpaceG'|}$ by the transfer theorem. Also we have that $P_\SpaceH \otimes P_{\SpaceH'} = P_{\SpaceH \otimes \SpaceH'}$, and that $\truncate{F} \otimes \truncate{G} = P_{\SpaceG \otimes \SpaceG'} \circ \big(\truncate{F} \otimes \truncate{G}\big) \circ P_{\SpaceH \otimes \SpaceH'}$. Hence, $\truncate{F} \otimes \truncate{G}$ is a genuine morphism $\SpaceH \otimes \SpaceH' \rightarrow \SpaceG \otimes \SpaceG'$. It is straightforward to check that this results in a well defined tensor product, and the following braiding operator turns $\starHilbCategory$ into a symmetric monoidal category (SMC):
\begin{equation}\label{eqn_Braiding}
	\sigma_{\SpaceH \SpaceG} :=
	\sum_{\varsigma(n,m)=1}^{\dim{\SpaceH}\dim{\SpaceG}}
	\ket{f_m} \otimes \ket{e_n} \; \bra{e_n}\otimes \bra{f_m}.
\end{equation}

\noindent Finally, one can define a dagger on morphisms by taking the conjugate transpose on the \textbf{matrix representation} given by (\ref{eqn_MatrixRepresentation}), obtaining the following morphism (by the transfer theorem):
\begin{equation}\label{eqn_Dagger}
	(\truncate{F})^\dagger := 
	\sum_{n=1}^{\dim{\SpaceH}} 
	\sum_{m=1}^{\dim{\SpaceG}} 
	\ket{e_{n}} 
	\truncate{F}_{mn}^\star
	\bra{f_{m}}.
\end{equation}
It is straightforward to check that $(\truncate{F})^\dagger$ is a morphism $\SpaceG \rightarrow \SpaceH$ whenever $\truncate{F}$ is a morphism $\SpaceH \rightarrow \SpaceG$, that the dagger is functorial and that it satisfies all the compatibility requirements with the monoidal structure. The content of this section can thus be summarised by the following result.

\begin{theorem}\label{thm_starHilbSMC}
	The category $\starHilbCategory$ is a $\dagger$-SMC, with tensor product and dagger defined by (\ref{eqn_TensorObjects}, \ref{eqn_TensorMorphisms}, \ref{eqn_Dagger}).
\end{theorem} 

\subsubsection{Standard bounded linear maps in $\starHilbCategory$}
\label{subsubsection_StandardBoundedMaps}

In order to do categorical quantum mechanics in $\starHilbCategory$, we have to first establish its relationship with the more traditional arena of standard Hilbert spaces and bounded linear maps. By construction, we don't expect $\starHilbCategory$ to contain all of $\HilbCategory$, as the objects were explicitly chosen to be separable (rather than arbitrary) Hilbert spaces. We expect, however, that the full subcategory $\sHilbCategory$ of separable Hilbert spaces and bounded linear maps will be faithfully embedded in it. 

We will refer to morphisms $\ket{\psi} :\equiv \sum_{n=1}^{\dim{\SpaceH}} \psi_n \ket{e_n}: \starComplexs \rightarrow \SpaceH$ as \textbf{vectors} or \textbf{states} in $\SpaceH$, and the $\starComplexs$-valued inner product induced by the dagger can be written as $\braket{\phi}{\psi} = \sum_{n=1}^{\dim{\SpaceH}} \phi_n^\star \psi_n$. We will refer to vectors $\ket{\psi}$ having finite squared norm $\braket{\psi}{\psi}$ as \textbf{finite vectors}, and to vectors having infinitesimal squared norm as \textbf{infinitesimal vectors}. Difference by infinitesimal vectors gives rise to the following equivalence relation, corresponding to the notion of convergence of vectors in norm:
\begin{equation}
	\ket{\phi} \simeq \ket{\psi} \iffdef \ket{\phi}-\ket{\psi} \text{ is infinitesimal}.
\end{equation}
We will say that a morphism $\truncate{F}: \SpaceH \rightarrow \SpaceG$ in $\starHilbCategory$ is \textbf{continuous} if for any $\ket{\psi_\kappa},\ket{\phi_\kappa} :\, \starComplexs \rightarrow \SpaceH$ satisfying $\ket{\psi_\kappa} \simeq \ket{\phi_\kappa}$  we have $\truncate{F} \ket{\psi_\kappa} \simeq \truncate{F} \ket{\phi_\kappa}$. Furthermore, the \textbf{operator norm} on some homset $\Hom{\starHilbCategory}{\SpaceH}{\SpaceG}$ can be defined as follows\footnote{Both the $\sup$ and the square root are simply extended from $\reals^+$ to $\starReals^+$ by the transfer theorem, as usual. The definition of the operator norm is independent of the equivalence relation $\simeq$.}:
\begin{equation}
	|| \truncate{F} ||_{op} := \sup_{\braket{\psi}{\psi} = 1} \sqrt{\bra{\psi}\truncate{F}^\dagger \truncate{F} \ket{\psi} }.
\end{equation}
We will say that a morphism $\truncate{F}: \SpaceH \rightarrow \SpaceG$ is \textbf{bounded} if its operator norm $||\truncate{F}||_{op}$ is finite. Just as it happens in the case of standard Hilbert spaces, throughout this work we will confuse bounded and continuous, thanks to the following result.

\begin{lemma}\label{thmNS_Continuity}
	Let $\truncate{F}: \SpaceH \rightarrow \SpaceG$ be a morphism in $\starHilbCategory$. The following are equivalent:
	\begin{enumerate} 
		\item[(i)] the operator norm $||\truncate{F}||_{op}$ is finite;
		\item[(ii)] $\truncate{F} \ket{\xi_\kappa}:\, \starComplexs \rightarrow \SpaceG$ is infinitesimal whenever $\ket{\xi_\kappa} :\, \starComplexs \rightarrow \SpaceH$ is infinitesimal;
		\item[(iii)] if $\ket{\psi_\kappa},\ket{\phi_\kappa} :\, \starComplexs \rightarrow \SpaceH$ satisfy $\ket{\psi_\kappa} \simeq \ket{\phi_\kappa}$, then we have $\truncate{F} \ket{\psi_\kappa} \simeq \truncate{F} \ket{\phi_\kappa}$.
	\end{enumerate}
\end{lemma}
\begin{proof}
	\textit{(i) implies (ii)}: let $\zeta := \braket{\xi_\kappa}{\xi_\kappa}$ be an infinitesimal; then we have $\bra{\xi_\kappa} \truncate{F}^\dagger \truncate{F} \ket{\xi_\kappa} \leq \zeta ||\truncate{F}||_{op}$, which is infinitesimal since $||\truncate{F}||_{op}$ is finite. \textit{(ii) implies (i)}: if $||\truncate{F}||_{op}$ is infinite, then for some $\ket{\psi_\kappa}$ of unit norm we have $\bra{\psi_\kappa}\truncate{F}^\dagger \truncate{F} \ket{\psi_\kappa} = \theta$ infinite; but then $\bra{\psi_\kappa}\frac{1}{\sqrt{\theta}}\truncate{F}^\dagger \truncate{F} \frac{1}{\sqrt{\theta}}\ket{\psi_\kappa} = 1$ is not infinitesimal, with $\frac{1}{\sqrt{\theta}}\ket{\psi_\kappa}$ infinitesimal. \textit{(ii) equivalent to (iii)}: by linearity of $\truncate{F}$.
\end{proof}

\noindent The following equivalence relation embodies the notion of convergence in operator norm:
\begin{equation}\label{eqn_InfinitesimalOpEquiv}
	\truncate{F} \sim \truncate{F'} \iffdef ||\truncate{F} - \truncate{F'}||_{op} \text{ is infinitesimal}.
\end{equation}
This equivalence relation is $\complexs$-linear, by triangle inequality, and it commutes with the dagger. It also commutes with composition and tensor product, as long as we restrict ourselves to continuous operators.
\begin{lemma}\label{thmNS_InfinitesimalOpEquiv}
	Suppose that $\truncate{F}$, $\truncate{F'}$, $\truncate{G}$ and $\truncate{G'}$ are all continuous. Then the following statements hold:
	\begin{align}
		\truncate{G} \cdot \truncate{F} \sim \truncate{G'} \cdot \truncate{F'} \text{ whenever both } \truncate{F} \sim \truncate{F'} \text{ and } \truncate{G} \sim \truncate{G'}, \nonumber \\
		\truncate{G} \otimes \truncate{F} \sim \truncate{G'} \otimes \truncate{F'} \text{ whenever both } \truncate{F} \sim \truncate{F'} \text{ and } \truncate{G} \sim \truncate{G'}. 
	\end{align}
\end{lemma}
\begin{proof}
	Bi-linearity of composition and tensor product, together with the triangle inequality, imply that the only statements we need to prove are the following:
	\begin{align}
		||\truncate{G} \cdot \xi_\kappa||_{op} &\text{ infinitesimal whenever } \truncate{G} \text{ continuous and } \xi_\kappa \text{ infinitesimal;} \nonumber \\
		||\zeta_\kappa \cdot \truncate{F}||_{op} &\text{ infinitesimal whenever } \truncate{F} \text{ continuous and } ||\zeta_\kappa||_{op} \text{ infinitesimal;} \nonumber \\
		||\truncate{G} \otimes \xi_\kappa||_{op} &\text{ infinitesimal whenever } \truncate{G} \text{ continuous and } ||\xi_\kappa||_{op} \text{ infinitesimal.}
	\end{align}
	The first statement follows from the fact that $||\truncate{G} \xi_\kappa||_{op} \leq ||\truncate{G}||_{op} ||\xi_\kappa||_{op}$, which is infinitesimal because $||\truncate{G}||_{op}$ is finite. The second statement goes similarly. The third statement is slightly trickier. Let $\ket{\psi_\kappa}$ be unit norm, and write $\ket{\psi_\kappa} = \sum_{n} \ket{\phi_\kappa^{(n)}} \ket{e_{n}}$ (where $(\ket{e_{n}})_{n}$ is the chosen orthonormal basis for the domain of $\xi_\kappa$). Then we have the following
	\begin{equation}
		\bra{\psi_\kappa} (\truncate{G} \otimes \xi_\kappa)^\dagger (\truncate{G} \otimes \xi_\kappa) \ket{\psi_\kappa} \leq \sum_{n'} \sum_{n} \braket{\phi_\kappa^{(n)}}{\phi_\kappa^{(n)}} ||\truncate{G}||_{op} |(\xi_\kappa)_{n'n}|^2 \leq ||\truncate{G}||_{op} ||\xi_\kappa||_{op},
	\end{equation}
	where the last product is infinitesimal because $||\truncate{G}||_{op}$ is finite.
\end{proof}

We say that a morphism $\truncate{G}$ is \textbf{near-standard} (in the operator norm) if it satisfies $\truncate{G} \sim \truncate{f}$ for some standard bounded linear map $f$. From now on, we will always use lowercase letters to denote standard bounded linear maps. Near-standard morphisms are in particular continuous, and form a sub-$\dagger$-SMC of $\starHilbCategory$, which we shall denote by $\starHilbCategoryNearStd$. If $\omega$ is some infinite natural, we denote by $\starHilbCategoryNearStd_\omega$ the full sub-category of $\starHilbCategoryNearStd$ having objects which are either finite-dimensional or have dimension $\omega$. By Lemma \ref{thmNS_InfinitesimalOpEquiv} both $\starHilbCategoryNearStd$ and $\starHilbCategoryNearStd_\omega$ can be enriched to become a strict $\dagger$-symmetric monoidal 2-categories. This observation finally allows us to relate our newly introduced category $\starHilbCategory$ to the more familiar $\sHilbCategory$. 

\noindent We define a strict \textbf{standard part} functor $\stdpart{\emptyArg}:\, \starHilbCategoryNearStd \rightarrow \sHilbCategory$ as follows:
\begin{enumerate}
	\item[(i)] $\stdpart{V,\ket{e_{n}}_{n=1}^\kappa} := V$;
	\item[(ii)] $\stdpart{\truncate{F}} := $ the unique $f$ such that $f$ is standard bounded and $\truncate{F} \sim \truncate{f}$.
\end{enumerate}
\noindent We fix an infinite non-standard natural $\omega$, and define a weak \textbf{lifting functor}, denoted by $\liftSym{\omega}: \sHilbCategory \rightarrow\, \starHilbCategoryNearStd$, as follows (with functoriality only up to $\sim$):
\begin{enumerate}
	\item[(i)] $\lift{V}{\omega} := (V,(\ket{e_{n}})_{n})$, where the orthonormal bases are chosen in such a way as to respect tensor product of $\starHilbCategoryNearStd$ (see the \cite{Gogioso2016b} for details);
	\item[(ii)] we have that $\lift{f}{\omega} := \truncate{f}$ on morphisms, and Lemma \ref{thmNS_InfinitesimalOpEquiv} guarantees that $\lift{ G \cdot F}{\omega} \sim \lift{G}{\omega} \cdot \lift{F}{\omega}$.
\end{enumerate}
For any fixed infinite natural $\omega$, the standard part functor restricts to a functor $\stdpart{\emptyArg}:\, \starHilbCategoryNearStd_\omega \rightarrow \sHilbCategory$, and the lifting functor restricts to a well-defined weak functor $\liftSym{\omega}: \sHilbCategory \rightarrow\, \starHilbCategoryNearStd_\omega$.
\begin{theorem}\label{thmNS_SeparableInStarHilb}
	The following results relate $\starHilbCategoryNearStd_\omega$ and $\sHilbCategory$:
	\begin{enumerate}
		\item[(i)] $\stdpart{\emptyArg}$ is a strict full functor of $\dagger$-SMCs, which is surjective on objects; 
		\item[(ii)] $\liftSym{\omega}$ is a weak faithful functor from a $\dagger$-SMC to a $\dagger$-symmetric monoidal 2-category, which is essentially surjective on objects; its restriction to the subcategory $\fdHilbCategory$ is strictly functorial;
		\item[(iii)] $\stdpart{\lift{f}{\omega}} = f$, for all standard bounded morphisms $f$;
		\item[(iv)] $\stdpart{\lift{V}{\omega}} = V$, for all objects $V$ of $\sHilbCategory$;
		\item[(v)] For all objects $\SpaceH$ of $\starHilbCategoryNearStd_\omega$, there is a (unique) standard unitary $\truncate{u}_\SpaceH : \SpaceH \rightarrow \lift{\stdpart{\SpaceH}}{\omega}$ such that $\stdpart{\truncate{u}_\SpaceH} = \id{\stdpart{\SpaceH}}$.
		\item[(vi)] $\truncate{u}_\SpaceG^\dagger \lift{\stdpart{\truncate{F}}}{\omega} \truncate{u}_\SpaceH \sim \truncate{F}$ for all morphisms $\truncate{F}: \SpaceH \rightarrow \SpaceG$ in $\starHilbCategoryNearStd_\omega$
	\end{enumerate}
\end{theorem}
\begin{proof}
\textit{Existence and uniqueness of definition of $\stdpart{\truncate{F}}$.}
	By definition, if $\truncate{F}$ is near-standard, at least one standard bounded linear map $f'$ exists such that $\truncate{F} \sim \truncate{f'}$. Now take two such standard bounded linear maps $f'$ and $f''$: by transitivity we get that $f' \sim f''$, i.e. that $||\truncate{f'}-\truncate{f''}||_{op}$ is infinitesimal; define $g := f'-f''$, standard bounded linear map. By transfer theorem (both directions), $\sqrt{\bra{\psi}g^\dagger g\ket{\psi}}$ is bounded above (by a standard constant $c \in \reals^+$, for all standard $\ket{\psi}$ satisfying $\braket{\psi}{\psi}=1$), if and only if $\sqrt{\bra{\psi}g^\dagger g\ket{\psi}}$ is also bounded above (by the same standard constant $c$, for all internal $\ket{\psi}$ such that $\braket{\psi}{\psi}=1$). Because $g$ is standard and bounded, $\sqrt{\bra{\psi}g^\dagger g\ket{\psi}}$ and $\sqrt{\bra{\psi}\truncate{g}^\dagger \truncate{g}\ket{\psi}}$ are infinitesimally close: as a consequence, if $||\truncate{f'}-\truncate{f''}||_{op}$ is infinitesimal, then it is bounded above by all standard reals $c>0$, and hence by the transfer theorem so is $||f'-f''||_{op}$. This proves that $||f'-f''||_{op}=0$, and we conclude that $f'=f''$.

\textit{Choice of orthonormal bases for $\liftSym{\omega}$.}
	Up to equivalence of categories, we can consider $\sHilbCategory$ as having objects given by all finite (possibly empty) tensor products of the following basic objects: the finite-dimensional Hilbert spaces $\complexs^p$ for all primes $p$, and the separable space $\ltwo{\naturals^+}$. Choose any orthonormal basis for each of the basic objects; denote them by $\ket{e^{(p)}_n}_{n=1}^p$ and $\ket{e^{(\infty)}_n}_{n=1}^{\infty}$. On basic objects, define $\lift{\complexs^p}{\omega} := (\complexs^p,\ket{e^{(p)}_n}_{n=1}^p)$ and $\lift{\ltwo{\naturals^+}}{\omega} := (\ltwo{\naturals^+},\ket{e^{(\infty)}_n}_{n=1}^\omega)$. Extend the definition to finite tensor products by using the tensor product of $\starHilbCategory$ (or, equivalently, by using the bijection $\varsigma$ from Equation (\ref{eqn_varsigma}) to explicitly construct a basis). 

\textit{Proof that $f \mapsto \truncate{f}$ is an injection.}
	Let $f,g$ be standard bounded linear maps, defined by matrices $(a_{nm})_{n,m\in \naturals^+}$ and $(b_{nm})_{n,m \in \naturals^+}$ respectively. The matrices can be extended by the transfer theorem to non-standard indices, and $\truncate{f}$ and $\truncate{g}$ have matrices $(a_{nm})_{\varsigma(n,m)=1}^{\kappa\nu}$ and $(b_{nm})_{\varsigma(n,m)=1}^{\kappa\nu}$. If $\truncate{f}=\truncate{g}$, then we have that $(a_{nm})_{\varsigma(n,m)=1}^{\kappa\nu}=(b_{nm})_{\varsigma(n,m)=1}^{\kappa\nu}$ as matrices, and in particular $a_{nm}=b_{nm}$ for all $n,m \in \naturals^+$, proving that $f=g$ in the first place. 

\textit{Weak functoriality of $\liftSym{\omega}$.}
	We begin by covering weak functoriality of $\liftSym{\omega}$, as it makes an interesting point by itself. Note that $\lift{g}{\omega} \cdot \lift{f}{\omega} = P_\SpaceG \circ g \circ P_\SpaceH \circ f \circ P_\SpaceK$, and that $\lift{g\cdot f}{\omega} = P_\SpaceG \circ g \circ f \circ P_\SpaceK$. In the infinite-dimensional case, if $f,g$ are standard bounded linear maps, then the standard series $a_{ln} := \sum_{m=0}^\infty g_{lm}f_{mn}$ converges for all fixed $l,n$, and the non-standard complex number $\sum_{m=0}^\kappa g_{lm}f_{mn}$ is infinitesimally close to the standard complex number $a_{ln}$. Hence $g \circ P_\SpaceH \circ g \sim g \circ f$, when seen as internal morphisms of non-standard Hilbert spaces. In the finite-dimensional case, there is no issue of truncation, and $\liftSym{\omega}$ is strictly functorial.

\textit{Proof of the main results.} 
\textit{Proof of (i).} The map $\stdpart{\emptyArg}$ is well-defined and monoidally functorial by Lemma \ref{thmNS_InfinitesimalOpEquiv}. It is full by the proof of existence/uniqueness given above, and surjective on objects by construction of $\starHilbCategory$ and $\sHilbCategory$. 
\textit{Proof of (ii).} The map $\lift{\emptyArg}{\omega}$ is weakly functorial by the argument given at the beginning of this proof (strictly functorial when restricted to $\fdHilbCategory$), and monoidally so by Lemma \ref{thmNS_InfinitesimalOpEquiv} and the choice of orthonormal bases presented above. Faithfulness was proven above (by showing that $f \mapsto \truncate{f}$ is an injection), and essential surjectivity follows from point (v) below. 
\textit{Proof of (iii).} We know from above that $\liftSym{\omega}$ is faithful, i.e. that $f \mapsto \truncate{f}$ is an injection. If $f$ is a standard bounded linear map, then the morphism $\stdpart{\lift{f}{\omega}}$ of $\sHilbCategory$ is the unique standard bounded linear map which is infinitesimally close (in operator norm) to $\truncate{f}$, i.e. it is $f$ itself. 
\textit{Proof of (iv).} By definition of the two functors.
\textit{Proof of (v).} By (iii), one such standard unitary $\truncate{u}_\SpaceH$ exists, namely by taking $u := \id{\SpaceH}$. Uniqueness follows because any such unitary must be infinitesimally close to the standard unitary $\truncate{u}_\SpaceH$ define above, and at most one such standard linear map exists.
\textit{Proof of (vi).} The morphism $\truncate{u}_{\SpaceG}^\dagger \lift{\stdpart{\truncate{F}}}{\omega} \truncate{u}_\SpaceH$ is infinitesimally close to its image under $\stdpart{\emptyArg}$, which is $\stdpart{\truncate{F}}$ by points (iii) and (v) above. Similarly, $\truncate{F}$ is infinitesimally close to its image under $\stdpart{\emptyArg}$, which is also $\stdpart{\truncate{F}}$. We conclude by transitivity/symmetry of $\sim$. 
\end{proof}

\subsection{How to use $\starHilbCategory$ for standard purposes}

\noindent The essence of Theorem \ref{thmNS_SeparableInStarHilb} is that $\sHilbCategory$ is equivalent to the subcategory $\starHilbCategoryNearStd$ of $\starHilbCategory$ given by near-standard morphisms in the operator norm, as long as we take care to equate morphisms which are infinitesimally close. The equivalence allows one to prove results about $\sHilbCategory$ by working in $\starHilbCategory$ and taking advantage of the CQM machinery introduced in the next Section. In a typical scenario, one might follow the following procedure, which is conceptually akin to using the two directions of the transfer theorem to prove results of standard analysis using non-standard methods:
\begin{enumerate}
	\item[(i)] start from $\sHilbCategory$; 
	\item[(ii)] lift to $\starHilbCategoryNearStd$ via the lifting functor;
	\item[(iii)] work in $\starHilbCategory$ to obtain a near-standard result (living again in $\starHilbCategoryNearStd$); 
	\item[(iv)] descend again to $\sHilbCategory$ via the standard part functor. 
\end{enumerate}
When proving equalities in $\sHilbCategory$, it is in fact sufficient to lift both sides via $\liftSym{\omega}$, and prove the equality in $\starHilbCategory$ without further constraints (this is because both sides will necessarily be lifted to $\starHilbCategoryNearStd$).

The arbitrary choice of infinite natural $\omega$ in the \inlineQuote{lifting phase} might seem unnatural at first, as the objects $\lift{V}{\omega}$ and $\lift{V}{\omega'}$ are not isomorphic in $\starHilbCategory$ for different infinite naturals $\omega \neq \omega'$. However, this is not actually an issue: from the perspective of $\sHilbCategory$, the two spaces are equivalent for all intents and purposes, because any proof that can be performed in one space can also be performed in the other. The following result makes this statement categorically precise.

\begin{lemma}
\label{lem_infDimEquivalenceFromStdViewpoint}
Let $\omega,\omega' \in \starNaturals$ be infinite natural numbers. Then the categories $\starHilbCategoryNearStd_\omega$ and $\starHilbCategoryNearStd_{\omega'}$ are weakly equivalent over $\sHilbCategory$, in the sense that there is a weak functor $\Phi_{\omega,\omega'}: \starHilbCategoryNearStd_\omega \rightarrow \starHilbCategoryNearStd_{\omega'}$ such that:
\begin{align}
\Big(\Phi_{\omega',\omega} \circ \Phi_{\omega,\omega'}\Big)(\SpaceH) = \SpaceH &\text{ for all objects } \SpaceH \text{ of } \starHilbCategoryNearStd_\omega \nonumber \\
\Big(\Phi_{\omega',\omega} \circ \Phi_{\omega,\omega'}\Big)(\truncate{F}) \sim \truncate{F} &\text{ for all morphisms } \truncate{F} \text{ of } \starHilbCategoryNearStd_\omega \nonumber
\end{align}
In particular, from the standard point of view of $\sHilbCategory$ we have that $\stdpartSym \circ\; \Phi_{\omega,\omega'} = \stdpartSym$, so that the categories $\starHilbCategoryNearStd_\omega$ and $\starHilbCategoryNearStd_{\omega'}$ are indistinguishable for standard purposes.
\end{lemma}
\begin{proof}
Define the functor $\Phi_{\omega,\omega'}$ as follows: 
\begin{equation}
\Phi_{\omega,\omega'}(V,\ket{e_n}_{n=1}^{\omega}) := (V,\ket{e_n}_{n=1}^{\omega'}) \hspace{2cm} \Phi_{\omega,\omega'}(\truncate{F}) := \truncate{F}
\end{equation}
Because all morphisms $\truncate{F}$ of $\starHilbCategoryNearStd_{\omega}$ are near-standard, this is clearly a weak functor of symmetric monoidal 2-categories: it respects composition and tensor product of morphisms only up to infinitesimal equivalence $\sim$, because the truncation $\truncate{F}$ is performed with different truncating projectors in $\starHilbCategoryNearStd_{\omega}$ and $\starHilbCategoryNearStd_{\omega'}$. By their very definition, $\Phi_{\omega,\omega'}$ and $\Phi_{\omega',\omega}$ establish a weak equivalence between $\starHilbCategoryNearStd_{\omega}$ and $\starHilbCategoryNearStd_{\omega'}$, and the equation $\stdpartSym \circ\; \Phi_{\omega,\omega'} = \stdpartSym$ is also a straightforward check.
\end{proof}
\newpage
\noindent There are two main kinds of proofs that can be performed using $\starHilbCategory$.
\begin{itemize}
	\item In one kind of proof, we have standard maps which can be expressed as compositions of non-standard maps with nicer algebraic/diagrammatic properties. This is the case, for example, of the proof of the Weyl Canonical Commutation Relations: the time-translation unitary (standard) is expressed as composition of the Frobenius algebra multiplication for the momentum observable (standard) and a position eigenstate (not near-standard), while the the momentum-boost unitary (standard) is expressed as composition of the Frobenius algebra multiplication for the position observable (not near-standard) and a momentum eigenstate (standard). This is also the case in the proof of Stone's Theorem on 1-parameter unitary groups (in the case of continuous periodic dynamics, where a symmetric cup is used to turn the unitary dynamics into the observable associated with their invariant.
	\item In the other kind of proof, we have equalities between standard maps which involve limits: these are lifted to equalities between non-standard maps where the limits have been absorbed into appropriate limiting objects (not necessarily near-standard), allowing the proof to be carried out algebraically. This is the case, for example, of the proofs of von Neumann's Mean Ergodic Theorem and Stone's Theorem on 1-parameter unitary groups (in the case of continuous periodic dynamics) : in both cases, the limit of a sum of standard maps is replaced by composition with the (not near-standard) counit of a Frobenius algebra, so that the proof can be carried out algebraically.
\end{itemize}
\noindent Despite the remarks above, one should not necessarily discount $\starHilbCategory$ as just being a category of handy mathematical tricks: a number of objects of concrete interest in the everyday practice of quantum mechanics (such as the position/momentum observables and eigenstates) are native to that richer environment, and help confer it its own independent dignity. Of course, the existence of multiple inequivalent choices of infinite natural dimension raises concerns with the physical interpretation of these non-standard objects, but Lemma \ref{lem_infDimEquivalenceFromStdViewpoint} guarantees an essential equivalence of different choices of infinite natural dimension $\omega$ from the point of view of standard quantum theory. The same limiting objects (e.g. Diract deltas or plane-waves) for different choices of $\omega$ should effectively be treated as incarnations of the same conceptual objects corresponding to different choices of \inlineQuote{infinite cutoff}. This point of view is more evident in the recent work of \cite{Gogioso2017}.

\subsection{Infinite-dimensional categorical quantum mechanics}
\label{subsection_InfiniteDimCQM}

The main motivation behind our use of non-standard analysis comes from the work of \cite{Abramsky2012b} on commutative $H^\star$-algebras, a particular class of non-unital special commutative $\dagger$-Frobenius algebras\footnote{In \cite{Abramsky2012b}, non-unital special commutative $\dagger$-Frobenius algebras are simply referred to as \textit{Frobenius algebras}. We refer to them in full as special commutative $\dagger$-Frobenius algebras, and we will specify \textit{non-unital} or \textit{unital} explicitly.} (non-unital $\dagger$-SCFAs, in short). It is an established result that approximate units for the algebras exist in separable Hilbert spaces: we will show that, in our non-standard framework, they can be made truly unital.

\begin{theorem}[From \cite{Abramsky2012b}]
	A non-unital $\dagger$-SCFA $(\!\hbox{\input{symbols/ZbwcomultSym.tex}}\!\!,\!\hbox{\input{symbols/ZbwmultSym.tex}}\!\!\!)$ on an object $V$ of $\sHilbCategory$ is an $H^\star$-algebra if and only if it corresponds to an orthonormal basis $\ket{e_{n}}_{n\in \naturals^+}$ of $V$ such that $\!\hbox{\input{symbols/ZbwcomultSym.tex}}\!\! \circ \ket{e_n} = \ket{e_n}\ket{e_n}$.
\end{theorem}

\begin{theorem}[From \cite{Abramsky2012b,Ambrose1945}]
	A non-unital $\dagger$-SCFA $(\!\hbox{\input{symbols/ZbwcomultSym.tex}}\!\!,\!\hbox{\input{symbols/ZbwmultSym.tex}}\!\!\!)$ on an object $V$ of $\sHilbCategory$ is an $H^\star$-algebra if and only if there is a sequence $\ket{E_n}_{n \in \naturals^+}$ such that for all $\ket{a} : \complexs \rightarrow V$ we have:
	\begin{enumerate}
		\item[(i)] $\!\hbox{\input{symbols/ZbwmultSym.tex}}\!\! \circ (\ket{E_n} \otimes \ket{a})$ converges to $\ket{a}$;
		\item[(ii)] $(\id{V} \otimes \bra{a}) \circ \!\hbox{\input{symbols/ZbwcomultSym.tex}}\!\! \circ \ket{E_n}$ converges.
	\end{enumerate}
	If this is so, then we can take $\ket{E_{n}} =: \sum_{n'\leq n} \ket{e_{n'}}$.
\end{theorem}

\noindent The sequence $\ket{E_{n}}_{n \in \naturals^+}$ itself doesn't converge in $\sHilbCategory$, because the state $\sum_{n \in \naturals^+} \ket{e_{n}}$ would have infinite norm. In our non-standard context, however, the state $\sum_{n=1}^{\kappa} \ket{e_n}$ is a well-defined, internal state for $\SpaceH = (V,\ket{e_n}_{n=1}^\kappa)$. This opens the way to the definition of unital $\dagger$-SCFAs on all objects of $\starHilbCategory$.

\begin{theorem}\label{thmNS_ClassicalStructures}
	Let $\SpaceH = (V,\ket{e_n}_{n=1}^\kappa)$ be an object in $\starHilbCategory$, and $\ket{f_{n}}_{n \in \naturals^+}$ be a standard orthonormal basis for $V$. Then the following comultiplication and counit define a \textbf{weakly unital}, \textbf{weakly special} commutative $\dagger$-Frobenius algebra on $\SpaceH$ (i.e. one where the Unit and Speciality laws hold only up to $\sim$):
	\begin{equation}\label{nonStdComultCounit}
		\input{pictures/chapter3/nonStdComultCounit.tikz}
	\end{equation}	
	We refer to it as the \textbf{classical structure}\footnote{The terminology \textit{classical structure}, in the context of $\starHilbCategory$, will refer to weakly unital, weakly special, commutative $\dagger$-Frobenius algebras. This is in accordance with the weak functoriality of $\liftSym{\omega}$ seen in the previous section.} 
	for $\ket{f_n}_n$. When $\ket{f_n}_{n}$ is  the \textbf{chosen orthonormal basis} $\ket{e_n}_n$ for $\SpaceH$, the algebra is strictly unital and strictly special, i.e. a unital $\dagger$-SCFA.
\end{theorem}
\begin{proof}
	Associativity and Frobenius laws hold with strict equalities (not up to $\sim$, despite involving composition of standard bounded linear maps), exactly as shown in \cite{Abramsky2012b}. Commutativity also holds with strict equality. The only things left to check are a Unit law and the Speciality law.
	\begin{align}
		\!\hbox{\input{symbols/DrightcounitLawSym.tex}}\!\! &=(\id{\SpaceH} \otimes \sum\limits_{m=1}^\kappa\bra{f_m})\cdot \sum\limits_{n=1}^\kappa \ket{f_n}\otimes\ket{f_n}\;\bra{f_n} \nonumber \\
		& \sim \sum\limits_{m=1}^\kappa\sum\limits_{n=1}^\kappa \ket{f_n}\braket{f_m}{f_n}\bra{f_n} = \sum\limits_{n=1}^\kappa \ket{f_n}\bra{f_n} \sim \id{\SpaceH}
	\end{align}
	\begin{align}
		\!\hbox{\input{symbols/DspecialtyLawSym.tex}}\!\! &=(\sum\limits_{m=1}^\kappa\ket{f_m}\;\bra{f_m}\otimes\bra{f_m})\cdot (\sum\limits_{n=1}^\kappa \ket{f_n}\ket{f_n}\bra{f_n}) \nonumber \\
		&\sim \sum\limits_{m=1}^\kappa\sum\limits_{n=1}^\kappa \ket{f_m}\braket{f_m}{f_n}^2\bra{f_n} = \sum\limits_{n=1}^\kappa \ket{f_n}\bra{f_n} \sim \id{\SpaceH}
	\end{align}
	Finally, if $\ket{f_n}_n$ is the chosen orthonormal basis $\ket{e_n}_n$ for $\SpaceH$, then the $\sim$ in the previous equations are in fact $=$, and the classical structure is a strictly unital, strictly special commutative $\dagger$-Frobenius algebra\footnote{\label{footnote_SimIntoEq}The $\sim$ become $=$ because the identity takes the exact form $\id{\SpaceH} = \sum_{n=1}^\kappa \ket{f_n}\bra{f_n}$, rather than the approximate form $\id{\SpaceH} \sim \sum_{n=1}^\kappa \ket{f_n}\bra{f_n}$, when $\ket{f_n}_{n=1}^\kappa$ is the chosen ort'l basis.}.
\end{proof}

\noindent In fact, it is not hard to show that $\starHilbCategory$ is a dagger compact category.
\begin{theorem}\label{thmNS_CompactClosed}
	The category $\starHilbCategory$ is compact closed. The \textbf{dual} of an object $\SpaceH = (V,\ket{e_n}_{n=1}^\kappa)$ is defined by $\SpaceH^\ast := (V^\ast,\ket{\xi_n}_{n=1}^\kappa)$, where $\ket{\xi_n}$ is the adjoint of $\ket{e_n}$ seen as a state of $V^\ast$. The \textbf{cap} and \textbf{cup} on $\SpaceH$ are defined as follows:
	\begin{equation}\label{nonStdCupCap}
		\input{pictures/chapter3/nonStdCupCap.tikz}
	\end{equation}	
	More in general, any classical structure in $\starHilbCategory$ can be used to define a \textbf{(weak) symmetric cap} and a \textbf{(weak) symmetric cup}, satisfying (weak)\footnote{By a \textbf{weak} equation we will henceforth mean one which is satisfied only up to $\sim$.} yanking equations.
\end{theorem}
\begin{proof}
	Weak yanking equations follow from the Frobenius law and weak Unit laws of any classical structure in $\starHilbCategory$. When the classical structure is that of the chosen orthonormal basis, the strict Unit laws result in strict yanking equations, yielding legitimate cups and caps (again because of the exact resolution of the identity into $\id{\SpaceH} = \sum_{n=1}^\kappa \ket{f_n}\bra{f_n}$).  
\end{proof}

\noindent The compact closed structure gives rise to a \textbf{trace} in the usual way:
\begin{equation}\label{nonStdTrace}
	\input{pictures/chapter3/nonStdTrace.tikz}
\end{equation}	

\noindent In particular, we see that the notation $\dim{\SpaceH} := \kappa$ for $\SpaceH = (V,\ket{e_n}_{n=1}^\kappa)$ was well chosen: $\Trace{\id{\SpaceH}} = \sum_{n=1}^\kappa 1 = \kappa = \dim{\SpaceH}$. The trace can also be used to endow the homset $\Hom{\starHilbCategory}{\SpaceH}{\SpaceG}$, which we have already seen to be a $\starComplexs$-vector space, with the following $\starComplexs$-valued \textbf{Hilbert-Schmidt inner product}:
\begin{equation}
	\innerprod{\truncate{G}}{\truncate{F}} := \Trace{\truncate{G}^\dagger \truncate{F}} = \sum_{m=1}^{\dim{\SpaceG}} \sum_{n=1}^{\dim{\SpaceH}} \truncate{G}_{mn}^\star \truncate{F}_{mn}.
\end{equation}
This is exactly the inner product that one would get by enriching the category $\starHilbCategory$ in itself via compact closure.

\subsection{Wavefunctions with periodic boundaries}
\label{subsubsection_WavefunctionsPeriodic}

As a sample application of the structures presented above, we cover the theory of wavefunctions on a 1-dimensional space with periodic boundary conditions: these live in $\Ltwo{\reals/(L \integers)} \isom \Ltwo{S^1}$, where $L$ is the length of the underlying space. The \textbf{momentum eigenstates}, or \textbf{plane-waves}, form a countable orthogonal basis for $\Ltwo{\reals/(L \integers)}$, indexed by $n \in \integers$ (in this section, $n,m,k,h$ will range over integers):
\begin{equation}
	\chi_n :=  x \mapsto  e^{-i (2\pi/L)nx}.
\end{equation}
The plane-wave $\ket{\chi_n}$ is the eigenstate of momentum $n \hbar$. Let $\theta(n) := |2n| + \frac{1-\operatorname{sign}(n)}{2}$ (with $\operatorname{sign}(0):=-1$) be a bijection $\integers \rightarrow \naturals^+$. We can obtain a countable orthonormal basis $\ket{e_{l}}_{l \in \naturals^+}$ for $\Ltwo{\reals/(L\integers)}$ as follows: 
\begin{equation}
	\ket{e_{\theta(n)}} := \frac{1}{\sqrt{L}} \ket{\chi_{n}} \text{ for all } n \in \integers.
\end{equation}
Now we shift our attention to the object $(\Ltwo{\reals/(L\integers)}, \ket{e_l}_{l=1}^{\kappa})$ of $\starHilbCategory$, with $\kappa = 2 \omega + 1$ some odd infinite natural\footnote{The notions of oddness and evenness extend from $\naturals$ to $\starNaturals$ by the transfer theorem, and by saying that some infinite non-standard natural $\kappa \in \starNaturals$ is odd we mean exactly that $\kappa = 2 \omega + 1$ for some (necessarily infinite) non-standard natural $\omega \in \starNaturals$. Note that the infinite natural $\omega$ here has nothing to do with the ordinal $\omega$ from set theory.}. As a shorthand for $\sum_{l=1}^{\kappa} \ket{\chi_{\theta^{-1}(l)}}$, and other cases where the index is bijected to the integers, we will simply re-index over the non-standard integers $\{-\omega,...,+\omega\}$ (such as in $\sum_{n=-\omega}^{+\omega} \ket{\chi_n}$). In particular, we will write our chosen object as $(\Ltwo{\reals/(L\integers)},\frac{1}{\sqrt{L}}\ket{\chi_n}_{n=-\omega}^{+\omega})$, or simply $\Ltwo{\reals/(L\integers)}$ \vspace{-1mm} when no confusion can arise. Now that we established the role of momentum eigenstates in our framework, it is time to turn our attention to position eigenstates. On a continuous space, position eigenstates are given by Dirac delta functions, and as a consequence are not associated with well-defined standard vectors. Here, we will define them in terms of the basis of momentum eigenstates, and then show that they coincide with their more traditional formulation in terms of Dirac deltas. Let $x_0 \in \nonstd{\big(\reals/(L\integers)\big)}$ be a point of the underlying space, then we define the \textbf{position eigenstate} at $x_0$ to be the following non-standard state:
\begin{equation}
	\ket{\delta_{x_0}} := \frac{1}{\sqrt{L}}\sum_{n=-\omega}^{+\omega}  \chi_{n}(x_0)^\ast \frac{1}{\sqrt{L}}\ket{\chi_n}.
\end{equation}  

\begin{theorem}\label{thmNS_PositionEigenstates}
	The position eigenstates are weakly orthogonal at standard points. Furthermore, they behave as \textbf{Dirac deltas}, i.e. they satisfy $\braket{\delta_{x_0}}{f} \simeq f(0)$ for all standard smooth $f \in \Ltwo{\reals/(L\integers)}$ and all standard points $x_0 \in \reals/(L\integers)$. The position eigenstates are also unbiased with respect to the momentum eigenstates, in the sense that $|\braket{\delta_{x_0}}{\chi_n}|=1$ independently of $n$ or $x_0$.
\end{theorem} 
\begin{proof}
	The proof that the state $\ket{\delta_{x_0}}$ satisfies $\braket{\delta_{x_0}}{f} \simeq f(0)$ for all standard smooth $f \in \Ltwo{\reals/L\integers}$ hinges on the transfer theorem, together with the following standard result from Fourier theory:
	\begin{equation}
		\frac{1}{\sqrt{L}}\sum_{n=-N}^{N} e^{-i(2\pi/L)x_0n}\frac{1}{\sqrt{L}}\braket{\chi_n}{f} = \frac{1}{L}\int_{\reals/L\integers} \Big(\sum_{n=-N}^{N} e^{i(2\pi/L)(x-x_0)n}\Big) f(x) dx\; \stackrel{N \rightarrow \infty}{\longrightarrow}\; f(x_0).
	\end{equation}
	To show orthogonality, we repeat the reasoning above in the special case of $\ket{f} := \sum_{m=-M}^{M}\big(\frac{1}{L}e^{i(2\pi/L)x_1 m}\big)\ket{\chi_m}$, for some $x_1 \neq x_0$. Two limits and two applications of the transfer theorem yield the desired result (we cannot do $\braket{\delta_{x_0}}{\delta_{x_1}} \simeq 0$ directly because $\ket{\delta_{x_1}}$ is not a standard smooth function):
	\begin{align}
		&\frac{1}{\sqrt{L}}\sum_{n=-N}^{N}\frac{1}{\sqrt{L}}\sum_{m=-M}^{M} \big(e^{-i(2\pi/L)x_0n}\frac{1}{\sqrt{L}}\big)\braket{\chi_n}{\chi_m}\big(e^{i(2\pi/L)x_1 m} \frac{1}{\sqrt{L}}\big) =\nonumber \\
		= &\frac{1}{L^2}\int_{\reals/L\integers} \Big(\sum_{n=-N}^{N}\sum_{m=-M}^{M} e^{i(2\pi/L)\big((x-x_0)n+(x-x_1)m\big)}\Big) dx\; \stackrel{N,M \rightarrow \infty}{\longrightarrow}\; 0.
	\end{align}
	Finally, the position eigenstates are clearly unbiased for the momentum eigenstates: by the first part of this proof, any given position eigenstate $\ket{\delta_{x_0}}$ satisfies $|\braket{\delta_{x_0}}{\chi_n}|^2 \simeq 1$ for all momentum eigenstates $\ket{\chi_n}$, independently of $n$.
\end{proof}

In the first part of this Chapter, we have defined coherent groups starting from the position observable $\hbox{\input{symbols/ZbwdotSym.tex}}\!\!$ and translation symmetry, and we have proven that $\hbox{\input{symbols/DdotSym.tex}}\!\!$ is the momentum observable. In this Section, we will take the opposite approach: we will start from the momentum observable $\hbox{\input{symbols/DdotSym.tex}}\!\!$ given by the plane-waves, then define the boost symmetry on momentum eigenstates, and finally prove that there is a corresponding position observable $\hbox{\input{symbols/ZbwdotSym.tex}}\!\!$ given by the Dirac deltas defined above (strongly complementary to the momentum observable).

We begin by showing explicitly that momenta generate the translation action of $\nonstd{\big(\reals/(L\integers)\big)}$ on the Dirac deltas.

\begin{theorem}\label{thmNS_MomentaGenerateTranslations}
	Let $(\!\hbox{\input{symbols/DcomultSym.tex}}\!\!,\!\hbox{\input{symbols/DcounitSym.tex}}\!\!,\!\hbox{\input{symbols/DmultSym.tex}}\!\!,\!\hbox{\input{symbols/DunitSym.tex}}\!\!)$ be the classical structure for the chosen orthonormal basis of normalised momentum eigenstates. Then the monoid $(\!\hbox{\input{symbols/DmultSym.tex}}\!\!,\!\hbox{\input{symbols/DunitSym.tex}}\!\!)$ endows the set $\suchthat{\sqrt{L}\ket{\delta_x}}{x \in \nonstd{\big(\reals/(L\integers)\big)}}$ of position eigenstates with the abelian group structure of position-space translation $\Big(\nonstd{\big(\reals/(L\integers)\big)},\oplus,0\Big)$:
	\begin{equation}\label{nonStdMomentaGenerateTranslations}
		\input{pictures/chapter3/nonStdMomentaGenerateTranslations.tikz}
	\end{equation}
	This can be equivalently written in the following form:
	\begin{equation}\label{eqn_MomentaGenerateTranslationsExplicit}
		\ket{\delta_{x \oplus y}} = \Big[\frac{1}{\sqrt{L}^2}\sum_{n=-\omega}^{+\omega} \chi_{n}(x)^\ast \ket{\chi_n}\bra{\chi_n} \Big] \ket{\delta_y},
	\end{equation}
	Note that $\chi_{n}(x)^\ast = \braket{\chi_n}{\delta_x} $ is nothing but $ \exp[i \frac{x\,p}{\hbar}]$ for a given (quantised) momentum eigenvalue $p = (n \hbar)/L$ and corresponding momentum eigenstate $\ket{\chi_n}$. 
\end{theorem}
\begin{proof}
	Using the definition of the classical structure for momentum eigenstates, together with the definition of the position eigenstates, we obtain the desired equalities:
	\begin{align}
		\!\hbox{\input{symbols/DmultSym.tex}}\!\! \circ (\sqrt{L}\ket{\delta_x} \sqrt{L}\ket{\delta_y}) &= \Big[\frac{1}{\sqrt{L}^3}\sum_{n=-\omega}^{+\omega} \ket{\chi_n}\bra{\chi_n}\bra{\chi_n} \Big]\sqrt{L}\ket{\delta_x}\sqrt{L}\ket{\delta_y} \nonumber \\
		&= \Big[\frac{1}{\sqrt{L}^2}\sum_{n=-\omega}^{+\omega} \chi_{n}(x)^\ast \ket{\chi_n}\bra{\chi_n} \Big] \sqrt{L}\ket{\delta_y} \nonumber \\
		&= \frac{1}{\sqrt{L}}\sum_{n=-\omega}^{+\omega} \chi_{n}(x)^\ast \chi_{n}(y)^\ast \ket{\chi_n} \nonumber \\
		&= \frac{1}{\sqrt{L}}\sum_{n=-\omega}^{+\omega} \chi_{n}(x \oplus y)^\ast \ket{\chi_n} = \sqrt{L}\ket{\delta_{x \oplus y}}.
	\end{align}
\end{proof}

While the momentum observable is embedded in our very definition of the object $\Big(\Ltwo{\reals/(L\integers)},\ket{e_l}_{l=1}^\kappa\Big)$, the definition of a position observable is not as straightforward, because the position eigenstates don't form a \textit{countable} orthonormal basis. Instead of defining the observable directly, we appeal to our understanding of symmetry-observable duality in coherent groups: we first define the boost symmetry $(\classicalStates{\hbox{\input{symbols/DdotSym.tex}}\!\!},\!\hbox{\input{symbols/ZbwmultSym.tex}}\!\!,\!\hbox{\input{symbols/ZbwunitSym.tex}}\!\!)$ on momentum eigestates, and only then we show that it is a unital $\dagger$-qSCFA behaving as expected from the position observable. 

Consider the binary function $a \oplus b := \modclass{a+b}{2N+1}$, where representatives for the $2N+1$ remainder classes are chosen in the set $\{-N,...,+N\}$: for every $N$, this function is defined in the standard theory of $\integers$, and endows $\{-N,...,+N\}$ with the group structure of $\integersMod{2N+1}$. By transfer theorem, a similar group operation exists on the internal set $\{-\omega, ..., +\omega\}$ of $\starIntegers$, endowing it with the group structure of $\starIntegersMod{2\omega+1}$. Remarkably, for any two finite integers $n,m \in \integers$ we have $n \oplus m = n + m$ (because $n,m < \omega$ implies $n+m < \omega$, so no modular reduction occurs). Now consider the following morphisms of $\starHilbCategory$:
\begin{equation}\label{nonStdBoostStruct}
	\input{pictures/chapter3/nonStdBoostStruct.tikz}
\end{equation}

\begin{theorem}\label{thmNS_GroupAlgebra}
	$(\!\hbox{\input{symbols/ZbwcomultSym.tex}}\!\!,\!\hbox{\input{symbols/ZbwcounitSym.tex}}\!\!,\!\hbox{\input{symbols/ZbwmultSym.tex}}\!\!,\!\hbox{\input{symbols/ZbwunitSym.tex}}\!\!)$ is a unital commutative $\dagger$-Frobenius algebra, the \textbf{group algebra} of $\starIntegersMod{2\omega+1}$. It is quasi-special, with normalisation factor $N_{\hbox{\input{symbols/ZbwdotSym.tex}}\!\!}=(2\omega + 1)$. Furthermore, it coherently copies, adjoins and deletes the (rescaled) position eigenstates, as long as the position $x$ takes the form $x = j \frac{L}{2\omega+1} \in \stdpart{\reals/(L\integers)}$ for some $j \in \starIntegersMod{2\omega+1}$ (i.e. we have that $x \in \frac{L}{2\omega+1}\starIntegersMod{2\omega+1}$)\footnote{Note the very interesting duality which emerges between the large-scale cutoff on the momentum $k$ (which has magnitude bounded above by $\omega$) and the small-scale cutoff on the position $x$ (which must be an integer multiple of $\frac{L}{2\omega+1}$).}:
	\begin{equation}\label{nonStdPositionEigenstatesCopy}
		\input{pictures/chapter3/nonStdPositionEigenstatesCopy.tikz}
	\end{equation}	
	\begin{equation}\label{nonStdPositionEigenstatesAdj}
		\input{pictures/chapter3/nonStdPositionEigenstatesAdj.tikz}
	\end{equation}	
	\begin{equation}\label{nonStdPositionEigenstatesDelete}
		\input{pictures/chapter3/nonStdPositionEigenstatesDelete.tikz}
	\end{equation}	
	As a consequence, we will also refer to it as the \textbf{classical structure for position eigenstates}, or as the \textbf{position observable}.
\end{theorem}
\begin{proof}
	Commutative, Associative and Unit laws can be proven on the monoid using the corresponding laws for $(\oplus,0)$. Frobenius law follows from the following re-indexing, with $k' := k\oplus n$:
	\begin{align}
		(\!\hbox{\input{symbols/ZbwmultSym.tex}}\!\! \otimes \id{})\circ(\id{} \otimes \!\hbox{\input{symbols/ZbwcomultSym.tex}}\!\!) &= \frac{1}{\sqrt{L}^4}\sum_{n=-\omega}^{+\omega} \sum_{m=-\omega}^{+\omega} \Big[\;\,\sum_{k\oplus h = m} \ket{\chi_{n\oplus k}}\otimes\ket{\chi_{h}}\;\bra{\chi_n}\otimes\bra{\chi_m}\Big] = \nonumber \\
		&= \frac{1}{\sqrt{L}^4}\sum_{n=-\omega}^{+\omega} \sum_{m=-\omega}^{+\omega}\Big[\sum_{k'\oplus h = n \oplus m} \ket{\chi_{k'}}\otimes\ket{\chi_{h}}\;\bra{\chi_n}\otimes\bra{\chi_m}\Big] \nonumber \\
		& = \!\hbox{\input{symbols/ZbwcomultSym.tex}}\!\! \circ \!\hbox{\input{symbols/ZbwmultSym.tex}}\!\!
	\end{align}
	The algebra is obviously a group algebra, and hence it is quasi-special with normalisation factor $(2 \omega + 1)$.
	The fact that position eigenstates are copied is a straightforward check, with a re-indexing $n' := n \ominus k$ in the second-to-last step:
	\begin{align}
		\!\hbox{\input{symbols/ZbwcomultSym.tex}}\!\! \circ \big(\sqrt{L} \ket{\delta_x}\big) &= \frac{1}{\sqrt{L}^2}\sum_{n=-\omega}^{+\omega} \sum_{k=-\omega}^{+\omega} \ket{\chi_k}\otimes \ket{\chi_{n\ominus k}}\; \braket{\chi_n}{\delta_x} = \nonumber \\
		&= \frac{1}{\sqrt{L}^2}\sum_{n=-\omega}^{+\omega} \sum_{k=-\omega}^{+\omega} \ket{\chi_k}\otimes \ket{\chi_{n\ominus k}}\; \chi_{n}(x)^\ast = \hspace{2cm}\nonumber \\
		\hspace{2cm}&= \frac{1}{\sqrt{L}^2}\sum_{n=-\omega}^{+\omega} \sum_{k=-\omega}^{+\omega} \ket{\chi_k} \otimes\ket{\chi_{n\ominus k}}\; \chi_{k}(x)^\ast \chi_{n\ominus k}(x)^\ast e^{i2\pi \frac{2\omega+1}{L} sx} = \nonumber \\
		&= \Big[\frac{1}{\sqrt{L}}\sum_{n'=-\omega}^{+\omega} \chi_{n'}(x)^\ast \ket{\chi_{n'}}\Big] \otimes \Big[\frac{1}{\sqrt{L}}\sum_{k=-\omega}^{+\omega}  \chi_{k}(x)^\ast\ket{\chi_k}\Big]= \nonumber \\
		&= \sqrt{L}\ket{\delta_x} \otimes \sqrt{L} \ket{\delta_x}.
	\end{align}
	In the third line, the extra phase $e^{i2\pi \frac{2\omega+1}{L} s_{n,k}x}$ appears because $\chi_{n}$ is a character of $\integers$, not of $\starIntegersMod{2\omega+1}$: the value of $s_{n,k} \in \{-1,0,+1\}$ keeps track of whether some modular reduction was necessary to go from $k \oplus (n\ominus k)$ to $n$. It is cancelled out if and only if we require $x$ to be in the form $x = j \frac{L}{2\omega+1}$, for some $j \in \starIntegersMod{2\omega+1}$: hence a duality between the large-scale cutoff of momentum and the small-scale cutoff of position arises as a consequence of a purely algebraic requirement in the non-standard framework. The adjoint and delete conditions have proofs that go along similar lines. 
\end{proof}

\begin{remark} 
Because the position eigenstates act as Dirac deltas on the smooth standard functions, the delete condition above means that the rescaled counit $\sqrt{L}\!\hbox{\input{symbols/ZbwcounitSym.tex}}\!\!$ defines the \textbf{integral operator}: $\sqrt{L}\!\hbox{\input{symbols/ZbwcounitSym.tex}}\!\! = \ket{\truncate{f}} \mapsto \int_{\reals/(L\integers)} f(x) dx$. Furthermore, the position eigenstates are actually orthogonal (and not only for standard points): this is because they are the classical states of a quasi-special commutative $\dagger$-Frobenius algebra in a SMC with scalars forming a field \cite{Coecke2013b}.
\end{remark}

\begin{theorem}\label{thmNS_StrongComplementarity}
The position and momentum observables defined above form a doubly well-pointed coherent group $(\hbox{\input{symbols/ZbwdotSym.tex}}\!\!,\hbox{\input{symbols/DdotSym.tex}}\!\!)$. 
\end{theorem}
\begin{proof}
We begin by observing that both $\hbox{\input{symbols/DdotSym.tex}}\!\!$ and $\hbox{\input{symbols/ZbwdotSym.tex}}\!\!$ have enough classical states: the $\dagger$-SCFA $\hbox{\input{symbols/DdotSym.tex}}\!\!$ has enough classical states by construction, while the $\dagger$-qSCFA $\hbox{\input{symbols/ZbwdotSym.tex}}\!\!$ can be seen to have enough classical states because the position eigenstates are enough to distinguish all momentum eigenstates. Then the laws of complementarity and strong complementarity follow immediately from the fact that: (i) the momentum and position eigenstates are mutually unbiased (by Theorem \ref{thmNS_PositionEigenstates}), (ii) the momentum eigenstates form group under $(\!\hbox{\input{symbols/ZbwmultSym.tex}}\!\!,\!\hbox{\input{symbols/ZbwunitSym.tex}}\!\!)$ (by the very definition of the group algebra for $\starIntegersMod{2\omega+1}$), and (iii) position eigenstates form group under $(\!\hbox{\input{symbols/DmultSym.tex}}\!\!,\!\hbox{\input{symbols/DunitSym.tex}}\!\!)$ (by Thm \ref{thmNS_MomentaGenerateTranslations}).
\end{proof}

\begin{corollary}[\textbf{Weyl CCRs}]\hfill\\
For all $x \in \reals/(L \integers)$, let $U_x$ be the unitary on $\Ltwo{\reals/(L \integers)}$ corresponding to space-translation of wavefunctions (with periodic boundary conditions) by $x$. For all $k \in \integers$, let $V_k$ be the unitary corresponding to momentum-boost by $k \hbar$. Then the following braiding relations hold between the two unitaries:
\begin{equation}\nonumber 
V_k U_x = e^{i \frac{2 \pi}{L} k \cdot x} \; U_x V_k
\end{equation}
\end{corollary}
\begin{proof}
The proof is a straightforward consequences of Theorems \ref{thm_WeylCCRs} and \ref{thmNS_StrongComplementarity}, and exemplifies an application of tools from $\starHilbCategory$ to obtain a simple algebraic proof of an iconic result of standard quantum mechanics. Indeed, if $x' := j\frac{L}{2\omega+1} \in \nonstd{\reals/(L\integers)}$ is any near-standard point such that $x' \simeq x$, then the following holds:
\begin{equation}\label{WeylCCRProofWavefunctionsBoundaryConditions}
	\resizebox{\textwidth}{!}{\input{pictures/chapter3/WeylCCRProofWavefunctionsBoundaryConditions.tikz}}
\end{equation}
\end{proof}

\begin{remark}
The methods presented here can be extended to the case of wavefunctions on spaces with a compact or discrete abelian group of translations (such as tori or lattices). A further extension is possible to certain locally compact groups, such as $(\reals,+,0)$, albeit requiring some additional finesse. These developments are detailed in the recent \cite{Gogioso2017}.
\end{remark}








\newpage 
\setcounter{section}{5} 
\section{Quantum dynamics}
\label{section_finiteCyclic}

\subsection{A traditional perspective on quantum dynamics}

In the traditional formulation of quantum mechanics, a quantum dynamical system is a quantum system $\SpaceH$ equipped with a prescribed Hamiltonian $\textbf{H}$. By Stone's Theorem, this is the same as saying that it is a quantum system equipped with a 1-parameter unitary group $(U_t)_{t \in \reals}$: as such, dynamics can be treated as a special case of symmetry, where the symmetry group is chosen to be $(\reals,+,0)$. 

More general kinds of dynamics can be studied by considering different symmetry groups. The ones generally deemed of interest in physics and computer science are: 
\begin{enumerate}
	\item[(i)] \textbf{continuous} dynamics, corresponding to symmetry group $(\reals,+,0)$;
	\item[(ii)] \textbf{continuous periodic} dynamics, corresponding to symmetry group $(\reals/(T\integers),\oplus,0)$; 
	\item[(iii)] \textbf{discrete} dynamics, corresponding to symmetry group $(\integers,+,0)$; 
	\item[(iv)] \textbf{discrete periodic} dynamics, corresponding to symmetry group $(\integersMod{T},\oplus,0)$. 
\end{enumerate}
In this opening Subsection, we will briefly recap the basics of all four notions of dynamics in the traditional formulation of finite-dimensional quantum theory: when making general statements about all four notions, we will used $(G,\oplus,0)$ to denote the generic time-translation symmetry group.

\subsubsection{Quantum dynamical systems}

A quantum dynamical system is a finite-dimensional Hilbert space $\SpaceH$ endowed with a unitary representation $(U_t)_{t \in G}$ of the time-translation symmetry group $(G,\oplus,0)$, where $(U_t)_{t \in G}$ is a strongly continuous family (a condition which is trivially satisfied in the case of discrete dynamics). As mentioned above and discussed in further detail below, this perspective is entirely equivalent to the perspective involving Hamiltonians and Schr\"{o}dinger's equation. 

When talking about the dynamics of a system $\SpaceH$, we are often interested in the evolution of an initial state $\psi_0$ under time-translation. Classically, we can look at the trajectory of $\psi_0$ in $\SpaceH$ as the function $\Psi: G \rightarrow \SpaceH$ which traces the history of the initial state as it evolves in time, i.e. the one defined by $\Phi:= t \mapsto U_t \ket{\psi_0}$. 

The classical trajectory of $\psi_0$ under $\integersMod{T}$ is not just a function $\integersMod{T} \rightarrow \SpaceH$: it is an equivariant function, by which we mean that $U_{\delta t} \Phi(t) = \Phi(t+\delta t)$, and hence it is a structurally sound way of seeing the time-translation symmetry group (a dynamical system itself, under the regular action) into the dynamical system $\SpaceH$.

\subsubsection{Hamiltonian and energy measurement}

In the continuous case, the Hamiltonian is the unique self-adjoint operator, given by Stone's Theorem on 1-parameter unitary groups, such that $U_t = e^{-i\frac{\textbf{H}}{\hbar}t}$. In the previous Chapter, we have discussed a number of issues with the identification of quantum mechanical observables with self-adjoint operators: as a consequence, we will instead take the Hamiltonian to be a PVM (Projector-Valued Measure) $\big(\pi({\goodchi})\big)_{\goodchi \in \reals^\wedge}$, or equivalently a PVM $\big(\pi(E)\big)_{E \in \reals}$, where we have fixed an isomorphism $\reals^\wedge \isom \reals$ (by choosing a constant $\hbar$). Stone's Theorem in its PVM version takes the following form:
\begin{equation}
U_t = \int_{\reals^\wedge} \goodchi(t) d\pi(\goodchi) = \int_{\reals} e^{-i\frac{E}{\hbar}t} d\pi(E)
\end{equation}
The PVM $(\pi(E))_{E \in \reals}$ itself specifies the energy measurement for the system: if the quantum system is in state $\rho$ (pure or mixed) and $S \subseteq \reals$ is any measurable subset, then the probability of an energy measurement resulting in an outcome in $S$ is given by $\Trace{\!\pi(S)\rho}$ (where $\pi(S)$ is the projector on the subspace of pure states spanned by the energy eigenstates $\ket{\psi_E}$ with $E \in S$, as specified by the PVM).

Once we let go of the self-adjoint operator point of view on the Hamiltonian, all four cases of dynamics enumerated in the introduction to this Section can be tackled uniformly. The Hamiltonian is always PVM $\big(\pi({\goodchi})\big)_{\goodchi \in G^\wedge}$: for the continuous case, $G \isom \reals$ and $G^\wedge \isom \reals$; for the continuous periodic case, $G \isom \reals/(T\integers)$ and $G^\wedge \isom \integers$; for the discrete case, $G \isom \integers$ and $G^\wedge \isom \reals/\integers$; for the discrete periodic case, $G \isom \integersMod{T}$ and $G^\wedge \isom \integersMod{T}$. The PVM version of Stone's Theorem takes the same form for all four cases (in the periodic cases we can think of $E = n h$):
\begin{align}
U_t =& \int_{\reals^\wedge} \goodchi(t) d\pi(\goodchi) = \int_{\reals} e^{-i2\pi\frac{Et}{h}} d\pi(E) \text{ for all } t \in \reals\\
U_t =& \int_{(\reals/(T\integers))^\wedge} \hspace{-1.05cm}\goodchi(t) d\pi(\goodchi) = \sum_{n \in \integers} e^{-i2\pi\frac{nt}{T}} \pi(n) \text{ for all } t \in \reals/(T\integers)\\
U_t =& \int_{\integers^\wedge} \goodchi(t) d\pi(\goodchi) = \int_{\reals/\integers} e^{-i2\pi\frac{Et}{h}} d\pi(E) \text{ for all } t \in \integers\\
U_t =& \int_{\integersMod{T}^\wedge} \goodchi(t) d\pi(\goodchi) = \sum_{n \in \integersMod{T}} e^{-i2\pi\frac{nt}{T}} \pi(n) \text{ for all } t \in \integersMod{T}
\end{align}
The energy measurements can similarly be expressed in the same form for all four cases (the top expression is for the two continuous cases\footnote{There are some (weak) caveats on the PVM in order for the integral expression to apply.}, while the bottom expression applies to the two discrete cases, where again we think of $E = n h$):
\begin{align}
	\mathbb{P}[E \in S \vert \rho] =& \Trace{\!\pi(S)\rho} \stackrel{*}{=} \int_{S} \Trace{[\pi'(E)\rho]} dE\\
	\mathbb{P}[n \in S \vert \rho] =& \Trace{\!\pi(S)\rho} = \sum_{n \in S} \Trace{\pi(n)\rho} 
\end{align}

\subsubsection{Schr\"{o}dinger's Equation}

In its traditional formulation for continuous quantum dynamics, the \textbf{time-independent Schr\"{o}dinger Equation} for an energy eigenstate $\ket{\psi_E}$ of a quantum dynamical system can be written as follows, where $\textbf{H}$ is the traditional Hamiltonian observable (the self-adjoint operator given by Stone's Theorem) and $\ket{\psi_E}$ is an energy eigenstate for energy level $E \in \reals$:
\begin{equation}\label{traiditionalSchrodingerEquation}
\textbf{H} \ket{\psi_E} = E \ket{\psi_E}
\end{equation}
The full \textbf{Schr\"{o}dinger Equation} takes the following differential form:
\begin{equation}\label{traiditionalSchrodingerEquationTimeDep}
i\hbar \frac{d}{dt}\ket{\psi(t)} = \textbf{H} \ket{\psi(t)}
\end{equation}
Because we do not want to work with self-adjoint operators and differentiation, we instead consider the following, equivalent formulation of Equations \ref{traiditionalSchrodingerEquation} and \ref{traiditionalSchrodingerEquationTimeDep}, known as the \textbf{exponentiated Schr\"{o}dinger Equation} (here $h = 2 \pi \hbar$):
\begin{equation}\label{traiditionalExponentiatedSchrodingerEquation}
U_t \ket{\psi_E} = e^{-i2\pi\frac{Et}{h}} \ket{\psi_E}
\end{equation}
To a practising physicist, the exponentiated formulation may feel further from the spirit of dynamics than the traditional formulation, because it replaces instantaneous states and differential evolution with a global, non-differential description in terms of a 1-parameter unitary group. In this work, however, we choose to adopt a more holistic point of view: the two equations are after all mathematically equivalent, and preferring one formulation over the other depends on the specific application and on the mathematical tools available to solve the equation itself. Because this work is concerned with the study of dynamics as a symmetry, and not with the solution of the equations of motion for some specific dynamical system, the exponentiated formulation in terms of 1-parameter unitary groups will be the undisputed favourite.

Aside from conceptual stances, the exponentiated version of Schr\"{o}dinger Equation has another, clear-cut advantage over the differential one when it comes to this work: it has a direct translation to those dynamical systems, such as the discrete periodic ones that will occupy large parts of this section, which don't admit infinitesimal generators for their dynamics. Below we write the exponentiated Schr\"{o}dinger Equation for all four cases of dynamics considered in the introduction to this Section:
\begin{align}
U_t \ket{\psi_{\goodchi}} =& \goodchi(t) \ket{\psi_{\goodchi}} = e^{-i2\pi\frac{Et}{h}} \ket{\psi_{\goodchi}} \text{ for all } \goodchi \in \reals^\wedge \leftrightarrow E/h \in \reals\\
U_t \ket{\psi_{\goodchi}} =& \goodchi(t) \ket{\psi_{\goodchi}}  = e^{-i2\pi\frac{nt}{T}} \ket{\psi_{\goodchi}} \text{ for all } \goodchi \in (\reals/(T\integers))^\wedge \leftrightarrow n=E/h \in \integers\\
U_t \ket{\psi_{\goodchi}} =& \goodchi(t) \ket{\psi_{\goodchi}}  = e^{-i2\pi\frac{Et}{h}Et} \ket{\psi_{\goodchi}} \text{ for all } \goodchi \in \integers^\wedge \leftrightarrow E/h \in \reals/\integers\\
U_t \ket{\psi_{\goodchi}}  =& \goodchi(t) \ket{\psi_{\goodchi}} = e^{-i2\pi\frac{nt}{T}} \ket{\psi_{\goodchi}} \text{ for all } \goodchi \in \integersMod{T}^\wedge \leftrightarrow n \in \integersMod{T} \label{discretePeriodicExponentiatedSchrodingerEquation}
\end{align}
In all four cases, we have explicitly fixed a correspondence between the canonical energy levels $\goodchi \in G^\wedge$ and non-canonical values of more direct physical significance.

\subsubsection{von Neumann's mean ergodic theorem}

von Neumann's Mean Ergodic Theorem is a cornerstone result in quantum dynamics, and its generalisation provides a statement which is dual, in a very specific sense, to the PVM formulation of Stone's Theorem for discrete dynamics. 
\begin{theorem}[\textbf{von Neumann's Mean Ergodic Theorem (discrete)} \cite{V.Neumann1932a}]\label{thm_meanErgodicDiscrete}\hfill\\
Let $U : \SpaceH \rightarrow \SpaceH$ be a unitary operator on a Hilbert space $\SpaceH$, and let $(U_t)_{t \in \integers}$ be the discrete dynamics it generates, i.e. $U_t := U^t$. Let $P : \SpaceH \rightarrow \SpaceH$ be the orthogonal projector on the invariant subspace for $U$, i.e. the subspace given by those vectors $\ket{\phi}$ such that $U \ket{\phi} = \ket{\phi}$. Then the following limit holds in the strong operator topology:
\begin{equation}
\lim_{T \rightarrow \infty} \frac{1}{T}\sum_{t=0}^{T-1} U_t = P_1
\end{equation} 
\end{theorem}

\begin{corollary}\label{cor_meanErgodicDiscreteGeneralised} Let $U : \SpaceH \rightarrow \SpaceH$ be a unitary operator on a Hilbert space $\SpaceH$, and let $(U_t)_{t \in \integers}$ be the discrete dynamics it generates, i.e. $U_t := U^t$. For each $\goodchi \in \integers^\wedge$, write $P_{\goodchi} : \SpaceH \rightarrow \SpaceH$ be the orthogonal projector on the subspace given by those vectors $\ket{\phi}$ such that $U \ket{\phi} = \goodchi(1)\ket{\phi}$.\footnote{Note that the $\goodchi \in \integers^\wedge$ are exactly those in the form $\goodchi := t \mapsto \zeta^t$ for some complex phase $\zeta$, so that the values $\goodchi(1)$ cover exactly all the possible eigenvalues for unitary operators.} Then the following limit holds in the strong operator topology:
\begin{equation}
\lim_{T \rightarrow \infty} \frac{1}{T}\sum_{t=0}^{T-1} \goodchi(t)^\ast U_t = P_{\goodchi}
\end{equation} 
\end{corollary}
\begin{proof}
Apply Theorem \ref{thm_meanErgodicDiscrete} to $U' := \goodchi(1)^\ast U$, observing that $\goodchi(t) = \goodchi(1)^t$, and that the condition $U' \ket{\psi} = \ket{\psi}$ used in Theorem \ref{thm_meanErgodicDiscrete} is equivalent to the condition $U \ket{\psi} = \goodchi(1) \ket{\psi}$ used in this Corollary.
\end{proof}

From Corollary \ref{cor_meanErgodicDiscreteGeneralised}, we can see how von Neumann's Mean Ergodic Theorem is dual to Stone's Theorem on discrete dynamics. Stone's Theorem, in its PVM formulation for discrete dynamics, shows that unitary dynamics can be reconstructed by integrating the projectors of the Hamiltonian observable multiplied by phases given by the canonical energy spectrum; von Neumann's Theorem, in its generalised formulation, shows that the projectors of the Hamiltonian observable can be reconstructed by averaging the unitary dynamics across time, multiplied by phases given by the canonical energy spectrum. This suggests a somewhat different point of view on quantum dynamical ergodicity: rather than being about invariant measures and the coincidence of time and space averages, as is the case in classical dynamical systems, the ergodic theorem in quantum dynamical systems is a manifestation of symmetry-observable duality, proving how time-translation symmetry and the Hamiltonian observable can be reconstructed one from the other.

A generalised version of von Neumann's Mean Ergodic Theorem can be formulated for continuous dynamics \cite{V.Neumann1932a}, continuous periodic dynamics and discrete periodic dynamics (as corollaries), and is summarised below.
\begin{theorem}[\textbf{Generalised von Neumann's Mean Ergodic Theorem}]\label{thm_meanErgodicGeneralised}\hfill\\
Let $(G,\oplus,0)$ be one of $(\reals,+,0)$, $(\reals/(T\reals),\oplus,0)$, $(\integers,+,0)$, or $(\integersMod{T},\oplus,0)$. Let $(U_t)_{t \in G}$ be a unitary dynamic on a Hilbert space $\SpaceH$ (i.e. a strongly continuous 1-parameter unitary group). For all $\goodchi \in G^\wedge$, let $P_{\goodchi} : \SpaceH \rightarrow \SpaceH$ be the orthogonal projector on the invariant subspace for $U$, i.e. the subspace given by those vectors $\ket{\phi}$ such that $U_t \ket{\phi} = \goodchi(t)\ket{\phi}$. Then the following limit holds in the strong operator topology:
\begin{equation}\label{eqn_meanErgodicGeneralised}
\lim_{T \rightarrow \infty} \frac{1}{T}\int\limits_{t=0}^{T} \goodchi(t)^\ast U_t dt = P_{\goodchi}
\end{equation} 
\end{theorem}
\noindent It is immediately clear from Equation \ref{eqn_meanErgodicGeneralised} that the four generalised versions of von Neumann's Theorem---for continuous, continuous periodic, discrete and discrete periodic dynamics, respectively---provide exact duals for the four versions of Stone's Theorem presented earlier on.

\subsubsection{The issue with time observables}

The problem of time observables is a long standing open problem in the philosophy of quantum theory. The history of time observables is turbulent, and extremely interesting: we will only mention some of the headlines in the coming paragraphs, and we refer the interested reader to \cite{Hilgevoord2005a,Pashby2015a,Roberts2012a,Butterfield2014}. 

Our story begins in 1926, when Dirac introduces time in quantum mechanics as a dynamical variable $t$, with associated \inlineQuote{conjugate momentum} $W$ satisfying the Kennard-Weyl form $tW - Wt = -i \hbar$ of the CCRs (technically, the conjugate momentum is $-W$). Heisenberg follows in 1927 with the time-energy uncertainty principle $Et - tE = - i \hbar$ (note the sign!) for Stern-Gerlach experiments, and Bohr proposes in 1928 the uncertainty relation $\Delta t \Delta E \geq h$ for wave-packets (although he talks of complementarity, rather than uncertainty relations). Both Dirac and Heisenberg later revise their position, and start treating $t$ as a parameter.

The first real issues with the notion of time observable are raised by Schr\"{o}dinger in 1931: he posits that a quantum time observable $t$ would be measured by observing an ideal quantum clock, and concludes that the resulting state of the system would be \inlineQuote{physically meaningless}, as it would have completely uncertain energy. Pauli in 1933 makes this claim rigorous, in what would become known as \inlineQuote{Pauli's Theorem}\footnote{His remarks are not really a theorem: he originally used the Kennard-Weyl form of the CCR, while the Stone-von Neumann Theorem requires the Weyl form.}. Using the Stone-von Neumann Theorem, one deduces that a time observable $t$ and a Hamiltonian $H$ satisfying the Weyl CCRs for the group $(\reals,+,0)$ would force $H$ to have continuous, unbounded-below spectrum. Since this contradicts real-world observations, Pauli concludes time observables to be physically meaningless. 

Putting together Schr\"{o}dinger's remarks and Pauli's result, we immediately spot a problem: what they are measuring is the \textit{clock} time observable and the \textit{system} energy. Looking at a synchronised clock-system state, it is easy to see that the energy measurement on the system, obtained from the coherent Hamiltonian, always commutes with the clock time measurement. What \textit{really} comes out of Schr\"{o}dinger's and Pauli's arguments is an issue with the quantum clocks themselves, not with the quantum dynamical systems. 

When you think a little more about what the clock does, however, things are not as ludicrous as they might seem. It is true that if you \inlineQuote{freeze} the clock in a definite clock time state (by measuring it) you lose all information about the the clock energy. However, we have seen in the previous Subsection that the clock energy is only relevant for the dynamical systems governed by the quantum clock, not for the clock itself: if one wishes to see the clock as a dynamical system governed by some notion of time other than the one it is ticking itself, then the resulting Hamiltonian for the clock need not have anything to do with the clock energy observable. Indeed, this is exactly what happens when we try to apply Schr\"{o}dinger's and Pauli's arguments to the quantum clocks themselves, and the physical absurdity stems from the incorrect identification between two different \inlineQuote{energy} observables:
\begin{enumerate} 
	\item[(i)] the clock energy observable, which is only relevant for the quantum dynamical systems governed by the clock, and which becomes completely undetermined upon measurement in the clock time observable; 
	\item[(ii)] the physical Hamiltonian, which instead corresponds to seeing the quantum clock as a dynamical system governed by some other, larger clock, and which need not be affected at all by the measurement of the clock time observable.
\end{enumerate}

\noindent It would be tempting to conclude the discussion above by positing the existence of some large, universal quantum clock ticking time (with time-translation group $\reals$) for all quantum clocks, but attempts to consider such an object have so far fallen back into the usual routine: either (i) the clock is an accessible physical system, in which it also governs itself and has an unbounded below, physically meaningless Hamiltonian, or (ii) it is external to the theory, in which case time is an external parameter and we get to the same conclusion of Schr\"{o}dinger, Pauli and many quantum physicists after them. In fact, even if we accepted the existence of Hamiltonians unbounded below, the first option would be physically meaningless: there would exists a system which we can measure and which freezes time for the entire universe. What nonsense! In Subsection \ref{subsection_timeObservables} below, we will approach this problem from a different direction: we will develop the tool to define a notion of quantum dynamics based on hierarchies of (locally) synchronised quantum clocks, which can be made inaccessible without the need for time to reduce to an external parameter. We will argue that this can be used to construct a (toy?) model of emergent global time without any need for an accessible universal quantum clock, but we will leave the detailed construction to future work.

\subsubsection{Quantum dynamics within the coherent group framework}

The techniques we have developed in the previous Sections allow us to treat three special cases of quantum dynamics: those which are discrete, periodic, or both. These are associated, respectively, with the symmetry groups $\integers$, $\reals/T\integers$ (where time flows continuously, with period $T$), and $\integersMod{T}$ (where time ticks at regular discrete intervals, with period $T$). Unfortunately, it will not be possible to cover the continuous aperiodic case of dynamics governed by $\reals$ in this work: the treatment of $\Ltwo{\reals}$ in the non-standard framework was introduced only recently \cite{Gogioso2017}, and the techniques developed here will be extended to continuous dynamics in the near future. 

\paragraph{Discrete periodic case.} Discrete periodic quantum dynamics are unitary representations of the finite cyclic groups $\integersMod{T}$ (where $T$ is the period). In our coherent framework, they are the representations of those doubly well-pointed, doubly finite coherent groups $\mathbb{G} = (\hbox{\input{symbols/ZbwdotSym.tex}}\!\!,\hbox{\input{symbols/DdotSym.tex}}\!\!)$ in $\fdHilbCategory$ which have $\underlyingGroup{G} \isom \integersMod{T}$. 

\paragraph{Discrete case.} Discrete quantum dynamics are unitary representations of the group $\integers$. In order to deal with them within our coherent framework, we consider the object $\SpaceG := \Big(\reals/(T\integers),\frac{1}{\sqrt{T}}\ket{\goodchi_n}_{n=-\omega}^{\omega}\Big)$ of $\starHilbCategory$ constructed in Section \ref{section_compactAbelian}, together with the doubly well-pointed coherent group $\GroupG := (\hbox{\input{symbols/ZbwdotSym.tex}}\!\!,\hbox{\input{symbols/DdotSym.tex}}\!\!)$ on $\SpaceG$ corresponding to the momentum/position pair for wavefunctions in a $1$-dimensional box of side $T > 0$ with periodic boundary conditions. The underlying group $\underlyingGroup{\GroupG}$ for $\GroupG$ is $\starIntegersMod{2\omega+1}$, which has $\integers$ as its subgroup of standard elements. 

In order to talk about standard discrete quantum dynamics within the non-standard framework, we will be interested in those unitary representations $\alpha: \SpaceH \otimes \SpaceG \rightarrow \SpaceH$ of the coherent group $\GroupG$ such that $\bar{U}_t := \alpha \circ (\id{\SpaceH}\otimes\ket{t})$ is near-standard for all $t \in \integers$ (it is in fact enough to ask for $\bar{U}_1$ to be near-standard, as $\bar{U}_t = (\bar{U}_1)^t$). Given a standard discrete unitary dynamics $(U_{t})_{t \in \integers}$ on some separable Hilbert space $V$, the following result shows how to find an $\alpha$ such that $\stdpart{\bar{U}_t} = U_t$ for all $t \in \integers$.
\begin{theorem}
	Let $\SpaceH := (V,\ket{e_d}_{d=1}^{D})$ be a space in $\starHilbCategory$, with dimension $D \in \starNaturals$. Consider a discrete dynamic $(U_{t})_{t \in \integers}$ on the separable standard Hilbert space $V$, and denote by $(U_{t})_{t \in \starIntegers}$ its non-standard extension. Then any near-standard unitary  $W \sim \id{\SpaceH}$ such that $(W U_{1})^{2\omega+1} = \id{\SpaceH}$ induces a unitary representation $\alpha: \SpaceH \otimes \SpaceG \rightarrow \SpaceH$ of the coherent group $\GroupG$ as follows:
	\begin{equation}
		\input{pictures/chapter3/dynamics/discreteDynCoherentFromStd.tikz}
	\end{equation}
	The unitary representations $\alpha$ of $\GroupG$ which arise this way are exactly those such that $\alpha \circ (\id{\SpaceH}\otimes\ket{1})$ is near-standard. Equivalently, $\alpha \circ (\id{\SpaceH}\otimes\ket{t})$ is near-standard for all $t \in \integers$, and the original standard discrete dynamics is recovered as $U_t := \stdpart{\alpha \circ (\id{\SpaceH}\otimes\ket{t})}$.
\end{theorem}
\begin{proof}
	Any discrete standard dynamic $(U_t)_{t \in \integers}$ satisfies $U_t = (U_1)^t$ for all $t \in \integers$, and hence by Transfer Theorem its non-standard extension must satisfy  $U_t = (U_1)^t$ for all $t \in \{-\omega,...,+\omega\}$. But the family $(U_{t})_{t =-\omega}^{+\omega}$ is not in general a representation of $\starIntegersMod{2\omega+1}$: it need not satisfy $U_{2\omega+1} = \id{\SpaceH}$. We now construct some near-standard unitary $W \sim \id{V}$ commuting with $U_1$ and such that $W^{2\omega+1} = U_{-2\omega-1}$. Diagonalise the unitary $U_1$ and let its eigenvalues be $(e^{i2\pi\,\alpha_d})_{d=1}^D$: then $U_{-2\omega-1}$ is necessarily near-standard, with eigenvalues $(e^{i2\pi\,\beta_d})_{d=1}^D$, where $\beta_d$ is the unique non-standard real $0\leq\beta_d<1$ satisfying $\beta_d = \modclass{(2\omega+1)\alpha_d}{1}$. Some unitaries $W$ satisfying the requirements above are the ones with the same eigenvectors as $U_1$ and with eigenvalues $(e^{i2\pi\,\gamma_d})_{d=1}^D$, where $\gamma_d := (\beta_d+k)/(2\omega+1)$ and $k \in \starIntegers$ is such that $k/(2\omega+1)$ is infinitesimal. If we now let $U'_t := (W U_1)^t$ for one such $W$, then $(U'_t)_{t \in \starIntegersMod{2\omega+1}}$ is a representation of $\starIntegersMod{2\omega+1}$, because the following equation now holds:
	\begin{equation}
		U'_{2\omega+1} = (W U_{1})^{2\omega+1} = \id{\SpaceH}
	\end{equation}
	Now consider a unitary representation of $\GroupG$, write $U'_1 := \alpha \circ (\id{\SpaceH}\otimes\ket{1})$ and let $U$ be a standard unitary such that $U'_1 \sim U$. If we define $W := U'_1U^{-1}$, then $W$ is by definition a near-standard unitary such that $W \sim \id{\SpaceH}$ and satisfying:
	\begin{equation}
		(WU)^{2\omega+1} = (U'_1)^{2\omega+1} = \id{\SpaceH}
	\end{equation}
	Hence $\alpha$ is in the required form, for a $W$ satisfying the required conidtions. Finally, note that $\alpha \circ (\id{\SpaceH}\otimes\ket{t}) \sim U^t$ for finite integers $t \in \integers$, and hence letting $U_t := \stdpart{\alpha \circ (\id{\SpaceH}\otimes\ket{t})} = U^t$ defines a unitary representation of $\integers$ on $V$. 
\end{proof}

\paragraph{Continuous periodic case.} Continuous periodic quantum dynamics are unitary representations of the group $\reals/(T\integers)$. In order to deal with them within our coherent framework, we consider the object $\SpaceG := \Big(\reals/(T\integers),\frac{1}{\sqrt{T}}\ket{\goodchi_n}_{n=-\omega}^{\omega}\Big)$ of $\starHilbCategory$ constructed in Section \ref{section_compactAbelian}, together with the doubly well-pointed coherent group $\GroupG := (\hbox{\input{symbols/DdotSym.tex}}\!\!,\hbox{\input{symbols/ZbwdotSym.tex}}\!\!)$ on $\SpaceG$ corresponding to the position/momentum pair for wavefunctions in a $1$-dimensional box of side $T > 0$ with periodic boundary conditions. 

The underlying group $\underlyingGroup{\GroupG}$ for $\GroupG$ is $\frac{T}{2\omega+1}\starIntegersMod{2\omega+1}$, with elements in the form $x = \frac{kT}{2 \omega +1}$ for $k \in \starIntegersMod{2\omega+1}$, and taking the standard part corresponds to a quotient group homomorphism $\stdpartSym: \frac{T}{2\omega+1}\starIntegersMod{2\omega+1} \rightarrow \reals/(T\integers)$. This means that every element of $\reals/(T\integers)$ can be approximated to within infinitesimal distance by elements of $\underlyingGroup{\GroupG}$, with the elements $\stdpartSym^{-1}(y) \subset \underlyingGroup{\GroupG}$ approximating a given $y \in \reals/(T \integers)$ forming a coset of the elements $\stdpartSym^{-1}(0) \subset \underlyingGroup{\GroupG}$ approximating the group unit $0 \in \reals/(T \integers)$. 

In order to talk about standard continuous periodic quantum dynamics within the non-standard framework, we will be interested in those unitary representations $\alpha: \SpaceH \otimes \SpaceG \rightarrow \SpaceH$ of the coherent group $\GroupG$ such that $\bar{U}_t := \alpha \circ (\id{\SpaceH}\otimes\ket{t})$ is near-standard for all $t \in \frac{T}{2\omega+1}\starIntegersMod{2\omega+1}$. Given a standard continuous periodic unitary dynamics $(U_{t})_{t \in \integers}$ on some separable Hilbert space $V$, the following result shows how to find an $\alpha$ such that $\stdpart{\bar{U}_{t}} = U_{\stdpart{t}}$ for all $t \in \frac{T}{2\omega+1}\starIntegersMod{2\omega+1}$, where we used again the notation $\bar{U}_t := \alpha \circ (\id{\SpaceH}\otimes\ket{t})$.
\begin{theorem}
	Consider a continuous periodic dynamic $(U_{t})_{t \in \reals/(T\integers)}$ on a separable standard Hilbert space $V$, let $(P_n)_{n \in \integers}$ be the complete family of orthogonal projectors for the Hamiltonian observable, and let  $(P_n)_{n \in \starIntegers}$ be its non-standard extension. Consider the space $\SpaceH := (V,\ket{e_d}_{d=1}^{D})$ in $\starHilbCategory$, where without loss of generality we have picked the basis $\ket{e_d}_{d=1}^{D}$ to consist of energy eigenstates, and we have chosen the dimension $D \in \starNaturals$ such that $\sum_{n=-\omega}^{+\omega} P_n = \sum_{d=1}^D \ket{e_d}\bra{e_d}$. Then the following is a unitary representation $\alpha: \SpaceH \otimes \SpaceG \rightarrow \SpaceG$ of the coherent group $\GroupG$:
	\begin{equation}
		\input{pictures/chapter3/dynamics/contPeriodicDynCoherentFromStd.tikz}
	\end{equation}
	Furthermore, we have that $\stdpart{\bar{U}_t} = U_{\stdpart{t}}$ for all $t \in \frac{T}{2\omega+1}\starIntegersMod{2\omega+1}$.
\end{theorem}
\begin{proof}
	For each given $t \in \reals/(T\integers)$, by PVM version of Stone's Theorem we have that:
	\begin{equation}
		U_t = \sum_{n \in \integers} e^{i2\pi\frac{nt}{T}} P_n 	
	\end{equation} 
	As a consequence, it should not be too surprising that for a given ${t} \in \frac{T}{2\omega+1}\starIntegersMod{2\omega+1}$ we have defined:
	\begin{equation}
		\bar{U}_{{t}} := \sum_{n=-\omega}^{+\omega} e^{i2\pi\frac{n{t}}{T}} P_n
	\end{equation} 
	Because $(P_n)_{n\in\integers}$ is a complete family of orthogonal projectors, then so is $(P_n)_{n\in\starIntegers}$. As a consequence, we have that:
	\begin{align}
		\bar{U}_0 =& \sum_{n=-\omega}^{+\omega} P_n = \id{\SpaceH} \\
		\bar{U}_t \bar{U}_s =& \sum_{n=-\omega}^{+\omega}\sum_{m=-\omega}^{+\omega} e^{i2\pi\frac{nt+ms}{T}} P_n P_m = \sum_{n \in \integers} e^{i2\pi\frac{n(t+s)}{T}} P_n = \bar{U}_{t\oplus s}
	\end{align}
	Furthermore, for any infinite $m \in \starIntegers$ we must have that $\stdpart{P_m} = 0$, because $P_mP_n = 0$ for all finite $n \in \integers$. As a consequence, for any ${t} \in \frac{T}{2\omega+1}\starIntegersMod{2\omega+1}$ we have:
	\begin{equation}
		\stdpart{\bar{U}_{{t}}} = \sum_{n=-\omega}^{+\omega} e^{i2\pi\frac{n\stdpart{t}}{T}} \stdpart{P_n} =  \sum_{n \in \integers} e^{i2\pi\frac{n\stdpart{t}}{T}} {P_n} = U_{\stdpart{t}}
	\end{equation}
\end{proof}

\paragraph{The way forward.} For the remainder of this Section, we will try to deal with quantum dynamics in full generality, by working with some generic doubly well-pointed coherent group $\GroupG := (\hbox{\input{symbols/ZbwdotSym.tex}}\!\!,\hbox{\input{symbols/DdotSym.tex}}\!\!)$ which we interpret as encoding time-translation symmetry. We will take $\underlyingGroup{\GroupG} = (\classicalStates{\hbox{\input{symbols/ZbwdotSym.tex}}\!\!},\!\hbox{\input{symbols/DmultSym.tex}}\!\!,\!\hbox{\input{symbols/DunitSym.tex}}\!\!)$ to mark the states of definite time, and deduce that $(\classicalStates{\hbox{\input{symbols/DdotSym.tex}}\!\!},\!\hbox{\input{symbols/ZbwmultSym.tex}}\!\!,\!\hbox{\input{symbols/ZbwunitSym.tex}}\!\!)$ marks the states of definite energy, but for the most part we will not be concerned with the specific structure of $\GroupG$: as a consequence, our results will apply to all those dynamics which can be modelled by coherent groups within our framework\footnote{At present, this includes discrete, continuous periodic and discrete periodic dynamics. Continuous dynamics will be added to this list in the near future, thanks to the recent work of \cite{Gogioso2017}.}. We will make it explicitly clear when a specific underlying group structure is used to derive some result, or at certain points in the discussion.

\subsection{Quantum clocks}

We consider a doubly well-pointed coherent group $\mathbb{G} = (\hbox{\input{symbols/ZbwdotSym.tex}}\!\!,\hbox{\input{symbols/DdotSym.tex}}\!\!)$, and we interpret the underlying group $\underlyingGroup{\GroupG}$ to be the time-translation symmetry group of a classical \textbf{clock} governing the dynamical systems which we are interested in. From an operational perspective, when saying that a clock \inlineQuote{governs} a dynamical system we will merely mean that the two are \inlineQuote{perfectly synchronised}: this perspective will be covered in detail later in this Section.

If we understand the underlying group $\underlyingGroup{\GroupG}$ as a classical clock, then the coherent group $\mathbb{G}$ is exactly what Schr\"{o}dinger would refer to as a \textbf{quantum clock} \cite{Hilgevoord2005a}: the quantum system of wavefunctions over the classical clock states, together with the appropriate time-translation symmetry structure. By construction, the point structure $\hbox{\input{symbols/ZbwdotSym.tex}}\!\!$ of a quantum clock is associated with the \textbf{clock time} observable, and the group structure endows the set $\classicalStates{\hbox{\input{symbols/ZbwdotSym.tex}}\!\!}$ of \textbf{clock time states} with the relevant time-translation structure: after all, the quantum clock always governs its own dynamics, as it is necessarily synchronised with itself.  what is the physical meaning of the observable associated with the group structure~$\hbox{\input{symbols/DdotSym.tex}}\!\!$? Is it energy, in a certain sense? To understand its role, we need to look at the dynamical systems governed by the quantum clock.

\subsection{Quantum dynamical systems}

Because we understand dynamics as time-translation symmetry, a \textbf{quantum dynamical system} governed by a quantum clock $\mathbb{G}$ is simply a unitary representation $\alpha$ of $\mathbb{G}$. As a consequence, the unitary Eilenberg-Moore category $\UnitaryRepCategory{\mathbb{G}}$ is the category of quantum dynamical systems governed by $\mathbb{G}$. 

In the coherent perspective, the evolution of an initial state $\psi_0$ under time-translation is given by its \textbf{coherent history} in the quantum dynamical system $\alpha: \SpaceH \otimes \SpaceG \rightarrow \SpaceH$, which is defined to be the following process $\Psi: \SpaceG \rightarrow \SpaceH$:
\begin{equation}\label{coherentHistory}
	\input{pictures/chapter3/dynamics/coherentHistory.tikz}
\end{equation}
A similar construction can be done for arbitrary coherent groups, not necessarily with dynamical semantics, in which case we will say that $\Psi$ is the \textbf{coherent orbit} of $\psi_0$ in the symmetric system $\alpha$.

Just like classical trajectories can be characterised as certain equivariant functions (e.g. $\reals \rightarrow \SpaceH$, in the continuous case), so coherent histories can be characterised as certain Eilenberg-Moore morphisms. 
\begin{theorem}[\textbf{Coherent orbits are EM morphisms}]\label{thm_coherentOrbitsEMmorphisms}\hfill\\
Let $\mathbb{G} = (\hbox{\input{symbols/ZbwdotSym.tex}}\!\!,\hbox{\input{symbols/DdotSym.tex}}\!\!)$ be a coherent group on a system $\SpaceG$ of a $\dagger$-SMC, let $\alpha: \SpaceH \otimes \SpaceG \rightarrow \SpaceH$ be a unitary representation of $\mathbb{G}$. If $\Psi$ is the coherent orbit of an initial state $\psi_0$ in $\alpha$, then it is also an Eilenberg-Moore morphism $\Psi: \!\hbox{\input{symbols/DmultSym.tex}}\!\! \rightarrow \alpha$ from the coherent group (seen as the regular representation) to $\alpha$:
\begin{equation}\label{coherentOrbitIsEMmorphism}
	\input{pictures/chapter3/dynamics/coherentOrbitIsEMmorphism.tikz}
\end{equation}
Conversely, if $\Psi: \!\hbox{\input{symbols/DmultSym.tex}}\!\! \rightarrow \alpha$ is an Eilenberg-Moore morphism, then it is the coherent orbit of the following initial state $\psi_0$:
\begin{equation}\label{EMmorphismIsCoherentOrbit}
	\input{pictures/chapter3/dynamics/EMmorphismIsCoherentOrbit.tikz}
\end{equation}
\end{theorem}
\begin{proof}
First we prove that the coherent orbit of an initial state $\psi_0$ in $\alpha$ is an EM morphism:
\begin{equation}\label{coherentOrbitIsEMmorphismProof}
	\resizebox{\textwidth}{!}{\input{pictures/chapter3/dynamics/coherentOrbitIsEMmorphismProof.tikz}}
\end{equation}
Conversely, we prove that an EM morphism is the coherent orbit of the initial state $\psi_0$ specified by Equation \ref{EMmorphismIsCoherentOrbit}:
\begin{equation}\label{EMmorphismIsCoherentOrbitProof}
	\input{pictures/chapter3/dynamics/EMmorphismIsCoherentOrbitProof.tikz}
\end{equation}
\end{proof}

\subsection{The coherent Hamiltonian}

In order to understand the role of the group structure $\hbox{\input{symbols/DdotSym.tex}}\!\!$ in the quantum clock, we turn our attention to the states which are invariant under the coherent dynamics: from Theorem \ref{thm_invariantStates}, we know that an invariant state of a quantum dynamical system $\alpha$ is associated with a definite outcome $\goodchi^\dagger \in \classicalStates{\hbox{\input{symbols/DdotSym.tex}}\!\!}$ of the coherent measurement $\alpha^\dagger$, and that the phase at time $t \in \integersMod{T}$ of its evolution is given by the scalar $\goodchi \circ t$. But this is exactly what happens with energy eigenstates in traditional quantum mechanics!

The admissible energy levels for a generic continuous quantum dynamical system $(U_t)_{t \in \reals}$, are traditionally labelled by the real numbers, and the phase acquired over time $t \in \reals$ by an eigenstate $\ket{\psi_E}$ of energy $E \in \reals$ is given by $\goodchi_{E/h}(t) := e^{i2\pi\frac{Et}{h}}$. 
The admissible energy levels for a continuous periodic dynamical system $(U_t)_{t \in \reals/(T\integers)}$ are discretised by periodicity, and are traditionally labelled by $n h$, where $n \in \integers$. The phase acquired over time $t \in \reals/(T\integers)$ by an eigenstate $\ket{\psi_{n\hbar}}$ of energy $n h$ is given by $\goodchi_{n}(t) := e^{i2\pi\frac{nt}{T}}$.
The admissible energy levels for a discrete quantum dynamical system $(U_t)_{t \in \integers}$ are continuous, but they are made periodic by the discrete nature of the dynamics. If we label the energy levels as $E \in \reals/(h\integers)$, then the phase acquired over time $t \in \integers$ by an eigenstate $\ket{\psi_E}$ of energy $E$ is given by $\goodchi_{E/h}(t) := e^{i2\pi\frac{Et}{h}}$. 
The admissible energy levels for a discrete periodic quantum dynamical system $(U_t)_{t \in \integersMod{T}}$ are both discrete and periodic: we can label them by $n h$ as in the discrete case, but with $n \in \integersMod{T}$ in this case. The phase acquired over time $t \in \integersMod{T}$ by an eigenstate $\ket{\psi_{n h}}$ of energy $nh$ is given by $\goodchi_{n}(t) := e^{i2\pi\frac{nt}{T}}$. 

In all four cases above, we could equivalently label the energy levels for dynamics governed by a time-translation group $G$ in a canonical way by using the multiplicative characters in $G^\wedge$. Indeed, an energy level is always uniquely identified with the time evolution of phases for its eigenstates: the two notions can be made to coincide, and in doing so we obtain a labelling of energy levels which is independent of choices of units of measurement for energy (and in particular of Planck constant $h$).

From the discussion above, it is clear that the $\hbox{\input{symbols/DdotSym.tex}}\!\!$-classical states can be identified with the admissible energy levels for quantum dynamical systems governed by the given quantum clock $(\hbox{\input{symbols/ZbwdotSym.tex}}\!\!,\hbox{\input{symbols/DdotSym.tex}}\!\!)$. As a consequence, we will refer to $\hbox{\input{symbols/DdotSym.tex}}\!\!$ as the \textbf{clock energy} observable, and to $\alpha^\dagger$ as the \textbf{coherent Hamiltonian} of the quantum dynamical system $\alpha$. Hence, the non-demolition and demolition measurements associated with $\alpha^\dagger$ correspond to the non-demolition and demolition measurements for the energy of the quantum dynamical system $\alpha$, where $P_{\goodchi}$ is the projector for energy level $\goodchi$:
\begin{equation}
	\input{pictures/chapter3/dynamics/coherentHamiltonian.tikz}
\end{equation}

\subsection{Schr\"{o}dinger's Equation}

Just as a representation of a coherent group carries a lot more information than the corresponding representation of the underlying classical group, the coherent history of a state in a quantum dynamical system $\alpha$ carries a lot more information than the corresponding history under the classical clock $\underlyingGroup{\GroupG}$. When evaluating the coherent history $\Psi: \SpaceG \rightarrow \SpaceH$ of an initial state $\psi_0$ at a clock time state $t$, we obtain the state $\psi_t$ corresponding to the evolution of the system at that time:
\begin{equation}\label{coherentHistoryTimeState}
	\input{pictures/chapter3/dynamics/coherentHistoryTimeState.tikz}
\end{equation}
Thanks to the coherent approach, however, we could instead choose to evaluate the coherent history at a clock energy state $\goodchi^\dagger$. As the following result shows, this yields the component of $\psi_0$ corresponding to energy level $\goodchi^\dagger$.
\begin{lemma}
Let $\mathbb{G} = (\hbox{\input{symbols/ZbwdotSym.tex}}\!\!,\hbox{\input{symbols/DdotSym.tex}}\!\!)$ be a coherent group on a system $\SpaceG$ of a $\dagger$-SMC, and let $\alpha: \SpaceH \otimes \SpaceG \rightarrow \SpaceH$ be a unitary representation of $\mathbb{G}$. Let $\psi_0$ be a state of $\SpaceH$, and let $\Psi : \SpaceG \rightarrow \SpaceH$ be the associated coherent history. If $\goodchi^\dagger \in \hbox{\input{symbols/DdotSym.tex}}\!\!$, then the following holds:
\begin{equation}\label{coherentHistoryEnergyState}
	\input{pictures/chapter3/dynamics/coherentHistoryEnergyState.tikz}
\end{equation}
\end{lemma}
\begin{proof}
The proof is straightforward, by unpacking the definition of $\Psi$ and using idempotence of $\alpha$ and $\hbox{\input{symbols/DdotSym.tex}}\!\!$-classicality of $\goodchi^\dagger$:
	\begin{equation}\label{coherentHistoryEnergyStateProof}
		\resizebox{\textwidth}{!}{\input{pictures/chapter3/dynamics/coherentHistoryEnergyStateProof.tikz}}
	\end{equation}
\end{proof}

The same idea---using coherence to evaluate something that is classically a function of clock time states on a clock energy state instead---can be used to derive Schr\"{o}dinger's Equation for a quantum dynamical system $\alpha$ from the defining equation of Eilenberg-Moore morphisms $\!\hbox{\input{symbols/DmultSym.tex}}\!\! \rightarrow \alpha$. From a categorical perspective, this is an extremely neat result: the fundamental equation of traditional quantum dynamics finds its natural counterpart in the fundamental equation defining evolution of states within the categorical framework (see Theorem \ref{thm_coherentOrbitsEMmorphisms}).

\begin{theorem}[\textbf{Schr\"{o}dinger Equation}]\label{thm_schrodingersEquation}\hfill\\
Let $\mathbb{G} = (\hbox{\input{symbols/ZbwdotSym.tex}}\!\!,\hbox{\input{symbols/DdotSym.tex}}\!\!)$ be a coherent group on a system $\SpaceG$ of a $\dagger$-SMC, and let $\alpha: \SpaceH \otimes \SpaceG \rightarrow \SpaceH$ be a unitary representation of $\mathbb{G}$. Suppose that $\psi_{\goodchi}$ is an energy eigenstate of $\alpha$ corresponding to clock energy level $\goodchi^\dagger \in \classicalStates{\hbox{\input{symbols/DdotSym.tex}}\!\!}$:
\begin{equation}\label{energyEigenstate}
	\input{pictures/chapter3/dynamics/energyEigenstate.tikz}
\end{equation}
Then $\psi_{\goodchi}$ satisfies the following Equation :
\begin{equation}\label{schrodingersEquationEnergyEigenstateNaked}
	\input{pictures/chapter3/dynamics/schrodingersEquationEnergyEigenstateNaked.tikz}
\end{equation}
When evaluated on clock time states, Equation \ref{schrodingersEquationEnergyEigenstate} is easily seen to be an abstract counterpart to Equation \ref{traiditionalExponentiatedSchrodingerEquation}:
\begin{equation}\label{schrodingersEquationEnergyEigenstate}
	\input{pictures/chapter3/dynamics/schrodingersEquationEnergyEigenstate.tikz}
\end{equation}
From this point onwards, we will refer to Equation \ref{schrodingersEquationEnergyEigenstateNaked} as \textbf{Schr\"{o}dinger's Equation} in our framework. Now assume that $\mathbb{G}$ is doubly well-pointed, and consider any process $\Psi: \SpaceG \rightarrow \SpaceH$. Then the following two conditions are equivalent:
\begin{itemize}
	\item the states $\psi_{\goodchi} := \Psi \circ \goodchi^\dagger$ satisfy Schr\"{o}dinger's Equation for all $\goodchi^\dagger \in \classicalStates{\hbox{\input{symbols/DdotSym.tex}}\!\!}$:
		\begin{equation}\label{schrodingersEquationEnergyComponentPsi}
			\input{pictures/chapter3/dynamics/schrodingersEquationEnergyComponentPsi.tikz}
		\end{equation}
	\item the process $\Psi$ is a coherent history, i.e. it satisfies the defining equation for Eilenberg-Moore morphisms $\!\hbox{\input{symbols/DmultSym.tex}}\!\! \rightarrow \alpha$:		
		\begin{equation}\label{PsiEMequation}
			\input{pictures/chapter3/dynamics/coherentOrbitIsEMmorphism.tikz}
		\end{equation}
\end{itemize}
\end{theorem}
\begin{proof}
Proving that Schr\"{o}dinger's Equation (Equation \ref{schrodingersEquationEnergyEigenstateNaked}) is satisfied by energy eigenstates is a straightforward application of unitarity for $\alpha$ and $\hbox{\input{symbols/DdotSym.tex}}\!\!$-classicality of $\goodchi^\dagger$. Now we want to prove that Equation \ref{schrodingersEquationEnergyEigenstateNaked} holding for all $\psi_{\goodchi}$ is equivalent to $\Psi$ satisfying the defining equation for EM algebras. Because the coherent group is doubly well-pointed, the defining equation for EM algebras holds if and only if it holds when evaluated on all $\goodchi^\dagger \in \classicalStates{\hbox{\input{symbols/DdotSym.tex}}\!\!}$:
\begin{equation}\label{PsiEMequationProof}
	\input{pictures/chapter3/dynamics/PsiEMequationProof.tikz}
\end{equation}
If Equation \ref{PsiEMequationProof} holds for all $\goodchi^\dagger$, then so does Equation  \ref{schrodingersEquationEnergyEigenstateNaked}:
\begin{equation}\label{PsiEMequationProofb}
	\input{pictures/chapter3/dynamics/PsiEMequationProof2a.tikz}
\end{equation}
Conversely, if Equation \ref{schrodingersEquationEnergyEigenstateNaked} holds for all $\goodchi^\dagger$, then so does Equation  \ref{PsiEMequationProof}:
\begin{equation}\label{PsiEMequationProof2b}
	\input{pictures/chapter3/dynamics/PsiEMequationProof2b.tikz}
\end{equation}
\end{proof}

\subsection{von Neumann's mean ergodic theorem}
\label{subsection_MeanErgodicTheorem}

We will now use symmetry-observable duality for coherent quantum dynamics to provide concise proofs of von Neumann's mean ergodic theorem in the discrete periodic, discrete and continuous periodic cases (using the coherent groups we introduced at the beginning of this Section). The same proof method applies---essentially unchanged---to the continuous case, using the coherent group on $\Ltwo{\reals}$ introduced in \cite{Gogioso2017}; however, a fully detailed treatment of the continuous case is left to future work.

\begin{theorem}[\textbf{Mean Ergodic Theorem (discrete periodic)}]\label{thm_vNMETdp}\hfill\\
Let $(U_t)_{t \in \integersMod{T}}$ be a unitary representation of $\integersMod{T}$ on a finite-dimensional Hilbert space $\SpaceH$, and let $P_{\goodchi} : \SpaceH \rightarrow \SpaceH$ be the orthogonal projector on the energy eigenspace corresponding to energy level $\goodchi \in \big(\integersMod{T}\big)^\wedge$. Then the following equality holds:
\begin{equation}
\frac{1}{T}\sum_{t=0}^{T-1} \goodchi(t)^\ast U_t = P_{\goodchi}
\end{equation} 
\end{theorem}
\begin{proof}
Symmetry-observable duality for systems with coherent symmetries can be invoked to obtain the following one-line proof:
\begin{equation}\label{meanErgodicOneLineProofDiscretePeriodic}
	\resizebox{\textwidth}{!}{\input{pictures/chapter3/dynamics/meanErgodicOneLineProofDiscretePeriodic.tikz}}
\end{equation}
\end{proof}

\begin{theorem}[\textbf{Mean Ergodic Theorem (discrete)}]\label{thm_vNMETd}\hfill\\
Let $(U_t)_{t \in \integers}$ be a unitary representation of $\integers$ on a separable Hilbert space $V$, and let $P_{\goodchi} : V \rightarrow V$ be the orthogonal projector on the energy eigenspace corresponding to energy level $\goodchi \in \integers^\wedge$. Then the following equality holds:
\begin{equation}
\lim_{T \rightarrow \infty}\frac{1}{T}\sum_{t=0}^{T-1} \goodchi(t)^\ast U_t = P_{\goodchi}
\end{equation} 
\end{theorem}
\begin{proof}
Symmetry-observable duality for systems with coherent symmetries can be invoked to obtain the following chain of equations:
\begin{equation}\label{meanErgodicOneLineProofDiscrete}
	\resizebox{\textwidth}{!}{\input{pictures/chapter3/dynamics/meanErgodicOneLineProofDiscrete.tikz}}
\end{equation}
Because $\omega$ is an arbitrary infinite integers, we have that $\lim_{T \rightarrow \infty}\frac{1}{T}\sum_{t=0}^{T-1} \goodchi(t)^\ast U_t \sim \frac{1}{2\omega+1}\sum_{t=-\omega}^{+\omega} \goodchi(t)^\ast U_t$, and by taking the standard part we obtain the desired result.
\end{proof}

\begin{theorem}[\textbf{Mean Ergodic Theorem (continuous periodic)}]\label{thm_vNMETcp}\hfill\\
Let $(U_t)_{t \in \reals/(T\integers)}$ be a unitary representation of $\reals/(T\integers)$ on a separable Hilbert space $V$, and let $P_{\goodchi} : V \rightarrow V$ be the orthogonal projector on the energy eigenspace corresponding to energy level $\goodchi \in \big(\reals/(T\integers)\big)^\wedge$. Then the following equality holds:
\begin{equation}
\frac{1}{T}\int\limits_{t=0}^{T} \goodchi(t)^\ast U_t dt = P_{\goodchi}
\end{equation} 
\end{theorem}
\begin{proof}
Symmetry-observable duality for systems with coherent symmetries can be invoked to obtain the following chain of equations, where we used the non-standard extension $(U_t)_{t \in \nonstd{(\reals/(T\integers))}}$ and we had defined the shorthand $t:= \frac{kT}{2\omega+1}$ in the two leftmost expressions:
\begin{equation}\label{meanErgodicOneLineProofContinuousPeriodic}
	\resizebox{\textwidth}{!}{\input{pictures/chapter3/dynamics/meanErgodicOneLineProofContinuousPeriodic.tikz}}
\end{equation}
It is a standard result of non-standard analysis that the integral $\int_{t=0}^{T} \goodchi(t)^\ast U_t dt $ can be approximated, up to infinitesimals, by the infinite sum $\sum_{k=-\omega}^{+\omega} \big(\goodchi(t)^\ast U_{t} \frac{T}{2\omega+1}\big)$: hence the chain of equations above reads $\frac{1}{T}\int_{t=0}^{T} \goodchi(t)^\ast U_t dt \sim P_{\goodchi}$, and by taking the standard part we obtain the desired result.
\end{proof}

\subsection{Stone's Theorem}
\label{subsection_StoneTheorem}

We will use symmetry-observable duality for coherent quantum dynamics once more, this time to provide concise proofs of Stone's Theorem in the discrete periodic, discrete and continuous periodic cases (using the coherent groups we introduced at the beginning of this Section). These proofs are essentially the duals of the proofs for von Neumann's Mean Ergodic Theorem presented above, but they're presented in full for instructive reasons. Again, he same proof method applies---essentially unchanged---to the continuous case, using the coherent group on $\Ltwo{\reals}$ introduced in \cite{Gogioso2017}; however, a fully detailed treatment of the continuous case is left to future work.

\begin{theorem}[\textbf{Stone's Theorem (discrete periodic)}]\label{thm_STdp}\hfill\\
Let $(U_t)_{t \in \integersMod{T}}$ be a unitary representation of $\integersMod{T}$ on a finite-dimensional Hilbert space $\SpaceH$, and let $(P_{\goodchi})_{\goodchi \in (\integersMod{T})^\wedge}$ be the complete family of orthogonal projectors associated to the Hamiltonian observable. Then the following equality holds:
\begin{equation}
	U_t = \sum_{\goodchi \in (\integersMod{T})^\wedge} \goodchi(t) P_{\goodchi}
\end{equation} 
\end{theorem}
\begin{proof}
	Symmetry-observable duality for systems with coherent symmetries can be invoked to obtain the following on-line proof:
	\begin{equation}\label{eqn_stoneTheoremDiscretePeriodic}
	\resizebox{\textwidth}{!}{\input{pictures/chapter3/dynamics/stoneTheoremDiscretePeriodic.tikz}}
	\end{equation}
\end{proof}

\begin{theorem}[\textbf{Stone's Theorem (discrete)}]\label{thm_STd}\hfill\\
Let $(U_t)_{t \in \integers}$ be a unitary representation of $\integers$ on a separable Hilbert space $V$, and let $\big(\pi(S)\big)_{S \subseteq \integers^\wedge}$ be the PVM associated to the Hamiltonian observable\footnote{Note that $\integers^\wedge \isom \reals/\integers$, so we need to consider a PVM instead of a complete family of orthogonal projectors as we do in the other cases.}. Then the following equality holds:
\begin{equation}
	U_t = \int_{\integers^\wedge} \goodchi(t) d\pi(\goodchi)
\end{equation} 
\end{theorem}
\begin{proof}
	We have $\integers^\wedge \isom \reals/\integers$, which corresponds to $\starIntegersMod{2\omega+1}^\wedge \isom \frac{1}{2\omega+1}\starIntegersMod{2\omega+1}$ in the non-standard framework. For each $\frac{k}{2\omega+1} \in \frac{1}{2\omega+1}\starIntegersMod{2\omega+1}$, we write $\goodchi_k$ for the corresponding element of $\starIntegersMod{2\omega+1}^\wedge$. Symmetry-observable duality for systems with coherent symmetries can then be invoked to obtain the following chain of equations:
	\begin{equation}\label{eqn_stoneTheoremDiscrete}
	\resizebox{\textwidth}{!}{\input{pictures/chapter3/dynamics/stoneTheoremDiscrete.tikz}}
	\end{equation}
	The projectors $P_{\goodchi_k}$ are infinitesimal, i.e. they can be seen to satisfy the following property: if $x,y \in \reals/\integers$ and $h,k \in \starIntegersMod{2\omega+1}$ are such that $x \simeq \frac{k_x}{2\omega+2}$ and $y \simeq \frac{k_y}{2\omega+2}$, then we have $\int_x^y d\pi(\goodchi) \sim \sum_{k=k_x}^{k_y} P_{\goodchi_k}$. As a consequence, we have $\int_{\integers^\wedge} \goodchi(t) d\pi(\goodchi) \sim \sum\limits_{k=-\omega}^{+\omega} \goodchi(t) P_{\goodchi_k}$, and taking the standard part completes our proof.
\end{proof}

\begin{theorem}[\textbf{Stone's Theorem (continuous periodic)}]\label{thm_STcp}\hfill\\
Let $(U_t)_{t \in \reals/(T\integers)}$ be a unitary representation of $\reals/(T\integers)$ on a separable Hilbert space $V$, and let $(P_{\goodchi})_{\goodchi \in (\reals/(T\integers))^\wedge}$ be the complete family of orthogonal projectors associated to the Hamiltonian observable\footnote{Note that $\big(\reals/(T\integers)\big)^\wedge \isom \integers$, so we can work directly with a complete family of orthogonal projectors in this case.}. Then the following equality holds:
\begin{equation}
	U_t = \sum_{\goodchi \in (\reals/(T\integers))^\wedge} \goodchi(t) P_{\goodchi}
\end{equation} 
\end{theorem}
\begin{proof}
	Symmetry-observable duality for systems with coherent symmetries can be invoked to obtain the following chain of equations, where we have used the non-standard extension $(U_t)_{t \in \nonstd{(\reals/(T\integers))}}$ and the shorthand $t := \frac{kT}{2\omega+1} \in \frac{T}{2\omega+1}\starIntegersMod{2\omega+1}$:
	\begin{equation}\label{eqn_stoneTheoremProofContinuousPeriodic}
	\resizebox{\textwidth}{!}{\input{pictures/chapter3/dynamics/stoneTheoremProofContinuousPeriodic.tikz}}
	\end{equation}
	We obtain our desired result by taking the standard part of the leftmost and rightmost expression, by observing that $\stdpart{U_t} = U_{\stdpart{t}}$ and that $\stdpart{P_{\goodchi_n}} = 0$ for all infinite non-standard integers $n$.
\end{proof}

\subsection{Feynman's clock}

Given a quantum circuit composed of unitary gates, the Feynman clock construction \cite{Feynman1982,Feynman1986} provides a Hamiltonian with ground states characterising the entire computation. More precisely, if $(V^{(t)})_{t=0,...,n}$ is some finite sequence of unitary gates on a quantum system $\SpaceH$, then the construction produces a Hamiltonian with the following ground states:
\begin{equation}\label{FeynmanClockHistoryStates}
	\left[\sum\limits_{t=0,...,n} \ket{\psi_t} \tensor \ket{t} \right] \text{ s.t. } V^{(t)} \ket{\psi_t} = \ket{\psi_{t+1}}
\end{equation}
The problem of performing the quantum computation is then reduced to the problem of finding a ground state for the Hamiltonian. This construction can be straightforwardly applied to the parallel-in-time simulation of discrete quantum dynamics (i.e. the one-step computation of a coherent history for the system) by seeing $U^{(t)}$ as the time evolution operator from time $t$ to time $t+1$ \cite{McClean2013}.

The relation between the Feynman clock construction and discrete periodic dynamics comes from the following observation: any linear circuit $(V^{(t)})_{t=0,...,T-1}$ can be turned into an appropriate cyclic circuit $(U^{(t)})_{t \in \integersMod{2T}}$ by setting $U^{(t)} := V^{(t)}$ for all $t=0,...,T-1$ and $U^{(t)} := V^{(2T-t-1)}$ for $t=T,...,2T-1$, and the cyclic circuit can be seen as a discrete periodic quantum dynamical system possessing a time-dependent Hamiltonian. The problem of finding the ground energy state for the original linear circuit is evidently equivalent to the problem of finding the ground energy state for the cyclic circuit, and hence the two circuits can be used interchangeably for the purposes of the Feynman clock construction. We will henceforth be considering a generic cyclic circuit, i.e. some family $(U^{(t)})_{t \in \integersMod{T}}$ of unitaries such that $\prod_{t=0}^{T-1} U^{(t)} = \id{}$.

The main obstacle to treating the Feynman clock construction within our framework would appear to be that a generic cyclic circuit $(U^{(t)})_{t \in \integersMod{T}}$ need not correspond to discrete periodic dynamics, in the sense used in this work up to this moment: in a representation $(U_t)_{t \in \integersMod{T}}$ of the finite cyclic group $\integersMod{T}$ we have $U_t = (U_1)^t$, while the operators $U^{(t)}$ are completely arbitrary. From a physical perspective, our symmetry approach models the the time-translation symmetry of quantum dynamical systems with a time-independent Hamiltonian, while the Feynman clock construction allows for potentially different time evolution operators $U^{(t)}$ at each different time $t$. 

However, this does not turn out to be such a mighty obstacle after all, because time-dependent dynamics can be easily accommodated in the time-independent symmetry perspective. Instead of a representation of $\integersMod{T}$ on $\SpaceH$, we consider the following representation $(W_t)_{t \in \integersMod{T}}$ of $\integersMod{T}$ on $\SpaceH \otimes \complexs[\integersMod{T}]$:
\begin{equation}
	W_{\delta t}\big(\ket{\psi} \tensor \ket{t}\big) := \left[\left(\prod\limits_{j=t}^{t+\delta t-1}U^{(j)}\right)\ket{\psi} \right] \tensor \ket{t \oplus \delta t}
\end{equation}
where the product is expanded to the left. This representation is essentially the \textit{propagator} for the discrete quantum dynamical system: given an interval of time $\delta t$, the time-translation action of the propagator evolves the state $\psi_t$ at time $t$ to the corresponding state $\psi_{t \oplus \delta t}$ at time $t \oplus \delta t$. If we interpret $\SpaceH$ as the quantum system of wavefunctions over some space, then $\SpaceH \otimes \SpaceG$ (here $\SpaceG = \complexs[\integersMod{T}]$) can be interpreted as the quantum system of wavefunctions over the corresponding (non-relativistic) space-time. 

With a little work to formalise that expanding product, we can turn this iterated product construction into a general result about unitary representations of coherent groups, and as a consequence we will be able to model quantum dynamical systems with time-dependent Hamiltonians within our framework.
\begin{theorem}[\textbf{Propagators}]\label{thm_propagators}\hfill\\
Let $\mathbb{G} := (\hbox{\input{symbols/ZbwdotSym.tex}}\!\!,\hbox{\input{symbols/DdotSym.tex}}\!\!)$ be a coherent group on an system $\SpaceG$ of a $\dagger$-SMC $\CategoryC$. Consider a process $\Pi U: (\SpaceH \otimes \SpaceG) \otimes \SpaceG \rightarrow \SpaceG$, and construct a process $\beta: (\SpaceH \otimes \SpaceG) \otimes \SpaceG \rightarrow \SpaceH \otimes \SpaceG$ as follows:
\begin{equation}\label{propagatorRepresentation}
	\input{pictures/chapter3/dynamics/propagatorRepresentation.tikz}
\end{equation}
Then $\beta$ is a unitary representation of the coherent group $\mathbb{G}$ on the composite system $\SpaceH \otimes \SpaceG$ if and only if the process $\Pi U$ satisfies the following requirements:
\begin{equation}\label{requirementsCyclicFamilyofUnitaries1}
	\input{pictures/chapter3/dynamics/requirementsCyclicFamilyofUnitaries1.tikz}
\end{equation}
\begin{equation}\label{requirementsCyclicFamilyofUnitaries2}
	\input{pictures/chapter3/dynamics/requirementsCyclicFamilyofUnitaries2.tikz}
\end{equation}
\begin{equation}\label{requirementsCyclicFamilyofUnitaries3}
	\input{pictures/chapter3/dynamics/requirementsCyclicFamilyofUnitaries3.tikz}
\end{equation}
When  $\beta$ is a unitary representation the form above, we refer to it as a \textbf{propagator}.
\end{theorem}
\begin{proof}
The entire proof essentially depends on the fact that the process $\Pi U$ can be recovered from the propagator $\beta$ by coherently deleting the time output of the latter:
\begin{equation}\label{propagatorRepresentation2}
	\resizebox{\textwidth}{!}{\input{pictures/chapter3/dynamics/propagatorRepresentation2.tikz}}
\end{equation}
Checking the various implications (six in total) is a tedious but entirely straightforward application of the laws of strong complementarity and Frobenius algebras. 
\end{proof}
\noindent In $\fdHilbCategory$, the process $\Pi U$ captures all the possible iterated products of unitaries in the following way: 
\begin{equation}\label{propagatorRepresentationExplained}
	\input{pictures/chapter3/dynamics/propagatorRepresentationExplained.tikz}
\end{equation}
Equation \ref{requirementsCyclicFamilyofUnitaries1} (similarly evaluated on clock time states $t$ and $\delta t$) corresponds to the following property of the iterated products:
\begin{equation}
\prod_{j=t}^{t+\delta t + \delta t'-1} U^{(j)} =\Big( \prod_{j=t+\delta t}^{t+\delta t + \delta t'-1} U^{(j)}\Big)\Big( \prod_{j=t}^{t+\delta t-1} U^{(j)}\Big)
\end{equation}
Equation \ref{requirementsCyclicFamilyofUnitaries2} corresponds to the following property of the iterated products:
\begin{equation}
\prod_{j=t}^{t-1} U^{(j)} = \id{}
\end{equation}
Equation \ref{requirementsCyclicFamilyofUnitaries3} defines what it means to take an iterated product \inlineQuote{going backwards}:
\begin{equation}
\prod_{j=t}^{t-\delta t-1} U^{(j)} := \Big( \prod_{j=t-\delta t  }^{t-1} U^{(j)} \Big)^\dagger = \big(U^{(t-\delta t)}\big)^\dagger \circ \big(U^{(t-\delta t + 1)}\big)^\dagger ... \circ \big(U^{(t-2)}\big)^\dagger \circ \big(U^{(t-1)}\big)^\dagger
\end{equation}
Equation \ref{requirementsCyclicFamilyofUnitaries3} furthermore proves that the product is in fact a product of unitaries. Theorem \ref{thm_propagators} is stated for a generic coherent group, but in the special case of discrete periodic dynamics Equation \ref{requirementsCyclicFamilyofUnitaries3} also proves that the family of unitaries involved in the product is in fact a cyclic circuit.

\begin{corollary}[\textbf{Time-translationally invariant propagators}]\hfill\\
Let $\mathbb{G} := (\hbox{\input{symbols/ZbwdotSym.tex}}\!\!,\hbox{\input{symbols/DdotSym.tex}}\!\!)$ be a coherent group on an system $\SpaceG$ of a $\dagger$-SMC $\CategoryC$, and let $\alpha$ be a unitary representation of $\mathbb{G}$ on $\SpaceH$. Then the following process $\beta$ is a propagator:
\begin{equation}\label{timeTranslInvariantPropagator}
	\input{pictures/chapter3/dynamics/timeTranslInvariantPropagator.tikz}
\end{equation}
We refer to these as \textbf{time-translationally invariant propagators}.
\end{corollary}
\begin{proof}
The proof involves another series of straightforward checks using the laws of strong complementarity. For example, the following chain of equalities proves the multiplication condition for $\beta$ using that for $\alpha$:
\begin{equation}\label{timeTranslInvariantPropagatorProof}
	\resizebox{\textwidth}{!}{\input{pictures/chapter3/dynamics/timeTranslInvariantPropagatorProof1.tikz}}
\end{equation}
\end{proof}

We are now in a position to prove correctness of the Feynman Clock construction within our framework. The next result is proven for general coherent groups, but in the special case where $\mathbb{G}$ is a quantum clock we have that $\beta^\dagger$ is the coherent Hamiltonian, and that the states $\Psi$ of Equation \ref{FeynmanClockHistoryStates} are the ground energy eigenstates. The result can be understood by observing that Equation \ref{FeynmanClockHistoryStatesStationaryProperty} is the abstract version of the following generalisation of the condition appearing in Equation \ref{FeynmanClockHistoryStates}:
\begin{equation}
\forall t. \forall \delta t. \; \Big(\prod_{j=t}^{t+\delta t - 1} U^{(j)} \Big) \ket{\psi_t} = \ket{\psi_{t+\delta t}}
\end{equation}
Hence Theorem \ref{thm_FeynmanClock} yields, in the special case of discrete periodic quantum dynamics, a proof of correctness for the traditional Feynman clock construction.
\begin{theorem}[\textbf{Feynman's Clock}]\label{thm_FeynmanClock}\hfill\\
Let $\mathbb{G} := (\hbox{\input{symbols/ZbwdotSym.tex}}\!\!,\hbox{\input{symbols/DdotSym.tex}}\!\!)$ be a coherent group on an system $\SpaceG$, and let $\beta: (\SpaceH \otimes \SpaceG) \otimes \SpaceG \rightarrow \SpaceH \otimes \SpaceG$ be a propagator for $\mathbb{G}$. Then the eigenstates of $\beta^\dagger$ corresponding to the definite outcome $\!\hbox{\input{symbols/ZbwunitSym.tex}}\!\! \in \classicalStates{\hbox{\input{symbols/DdotSym.tex}}\!\!}$ are exactly the states $\Psi$ of $\SpaceH \otimes \SpaceG$ satisfying the following condition:
\begin{equation}\label{FeynmanClockHistoryStatesStationaryProperty}
	\input{pictures/chapter3/dynamics/FeynmanClockHistoryStatesStationaryProperty.tikz}
\end{equation}
\end{theorem}
\begin{proof}
If $\Psi$ is an eigenstate corresponding to definite outcome $\!\hbox{\input{symbols/ZbwunitSym.tex}}\!\!$, then we have the following equality (by the definition of $\beta$ as a propagator):
\begin{equation}\label{FeynmanClockProof1}
	\input{pictures/chapter3/dynamics/FeynmanClockProof1.tikz}
\end{equation}
As a consequence we also have the following chain of equalities, where the rightmost equality is obtained by applying the laws of strong complementarity (central rule of the bottom row) and Hopf's law:
\begin{equation}\label{FeynmanClockProof2}
	\resizebox{\textwidth}{!}{\input{pictures/chapter3/dynamics/FeynmanClockProof2.tikz}}
\end{equation}
\end{proof}

\subsection{Clock-system synchronisation}

Up to this point, we have described quantum dynamical systems in terms of unitary representations. Although appealing from an algebraic and categorical perspective, this formulation lacks an immediate physical interpretation. To fix this, we shift point of view from the \textit{action} of a quantum clock on a quantum system, to the \textit{synchronisation} of the quantum clock and quantum dynamical system. The algebraic perspective corresponds to saying that a clock time state $\ket{\delta t}$ sends state $\ket{\psi_t}$ of a quantum dynamical system $\alpha$ to the corresponding evolved state $\ket{\psi_{t \oplus \delta t}}$. The synchronisation perspective corresponds to saying that whenever the clock is measured to be in clock time state $\ket{t}$, the quantum dynamical system is necessarily in state $\ket{\psi_t}$.

When talking about a \textbf{synchronised clock-system state} for a unitary representation $\alpha$ of coherent group $\mathbb{G} = (\hbox{\input{symbols/ZbwdotSym.tex}}\!\!,\hbox{\input{symbols/DdotSym.tex}}\!\!)$, we will mean a state in the following form, and we will refer to $\psi_0$ as the \textbf{initial state} for $\alpha$:
\begin{equation}\label{syncClockSystemPair}
	\input{pictures/chapter3/dynamics/syncClockSystemPair.tikz}
\end{equation}
Note that a synchronised clock-system state for $\alpha$ is a stationary state for the time-translationally invariant propagator $\beta$ associated with $\alpha$ (the one given by Equation \ref{timeTranslInvariantPropagator}). There is an obvious generalisation to multiple quantum dynamical systems $\alpha^{(1)},...,\alpha^{(N)}$ governed by $\mathbb{G}$ and mutually synchronised:
\begin{equation}\label{syncClockSystemPairMultipleSystems}
	\input{pictures/chapter3/dynamics/syncClockSystemPairMultipleSystems.tikz}
\end{equation}
In fact, this is not really a generalisation of the notion of synchronised clock-system state, but rather a special case of it. Indeed, we can obtain Diagram \ref{syncClockSystemPairMultipleSystems} from Diagram \ref{syncClockSystemPair} by choosing $\alpha$ to be the following \textbf{joint dynamical system} of $\alpha^{(1)},...,\alpha^{(N)}$:
\begin{equation}\label{jointDynamicalSystem}
	\input{pictures/chapter3/dynamics/jointDynamicalSystem.tikz}
\end{equation}

A measurement of the quantum clock $\mathbb{G}$ in the clock time observable results in each systems $\alpha^{(j)}$ collapsing to state $\psi^{(j)}_t$, as one would expect:
\begin{equation}\label{syncClockSystemPairTime}
	\input{pictures/chapter3/dynamics/syncClockSystemPairTime.tikz}
\end{equation}
One should note that Equation \ref{syncClockSystemPairTime} is a post-selection on a definite clock time state, and does not necessarily reflect the intuition of looking at a clock to find out what time it is. In short, the situation can be summarised as follows: in the real world, both the quantum system and the quantum clock can be thought to be in turn synchronised with some inaccessible quantum clock ticking time for both of them. Real world clocks, for example, are finite: they model the discrete periodic time of $\integersMod{T}$, ticked at regular intervals and starting again from zero when the clock has gone through all its $T$ time states. With respect to that same regular interval, we could think of an inaccessible external quantum clock as ticking the discrete time of $\integers$, so that synchronisation between the clock and the external clock corresponds to the action $\integers \times \integersMod{T} \rightarrow \integersMod{T}$ given by $(\delta t, t_0) \mapsto t_0 \oplus q(\delta t)$, where $q: \integers \rightarrow \integersMod{T}$ is the group quotient homomorphism defined by $q(1) = \modclass{1}{T}$. Similarly, the continuous case would involve finite clocks ticking $\reals/(T\integers)$ time and the external clock ticking $\reals$ time, with synchronisation between the two given by the corresponding quotient group homomorphism $q:\reals \rightarrow \reals/(T\integers)$.

When thinking of a real world synchronised clock-system scenarios, we sometimes have to explicitly consider the inaccessible external clock. For example, consider the following situation (e.g. with the clock ticking $\integersMod{T}$ and the external clock ticking $\integers$ as before). At some point in the external clock time, the (internal) clock is measured to be in clock time state $t_0$, so that the synchronised system is inferred to be in state $\ket{\psi_{t_0}}$. A certain amount $\delta t$ of external clock time is then allowed to pass, and the clock is measured again: in a correct modelling of this situation, the clock should be found in clock time state $t_0 \oplus q(\delta t)$, and the system should be inferred to be in state $\ket{\psi_{t_0\oplus q(\delta t)}}$. This scenario cannot be modelled simply by post-selecting clock time states, because post-selection is a static process and the clock needs to evolve between successive measurements. The correct modelling of this situation goes as follows, where $\gamma$ is the quantum dynamical system corresponding to the discrete periodic clock (described by the action of $\integers$ on $\integersMod{T}$ above) and $\alpha$ is the quantum dynamical system synchronised with it:
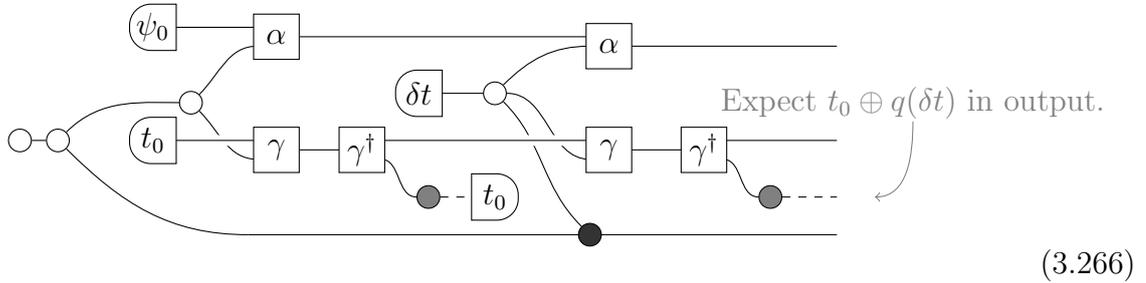
\begin{equation}
	\input{pictures/chapter3/dynamics/systemClockUniclockEvol.tikz}
\end{equation}

We have given an interpretation to measurement of the quantum clock $\mathbb{G}$ in clock time states. But what if instead we measured a quantum clock, synchronised with one or many systems, in the clock energy observable? We claim that this results in the synchronised systems finding themselves in  a global state of definite total energy $\goodchi_{tot}$ given by the outcome of the clock energy measurement:
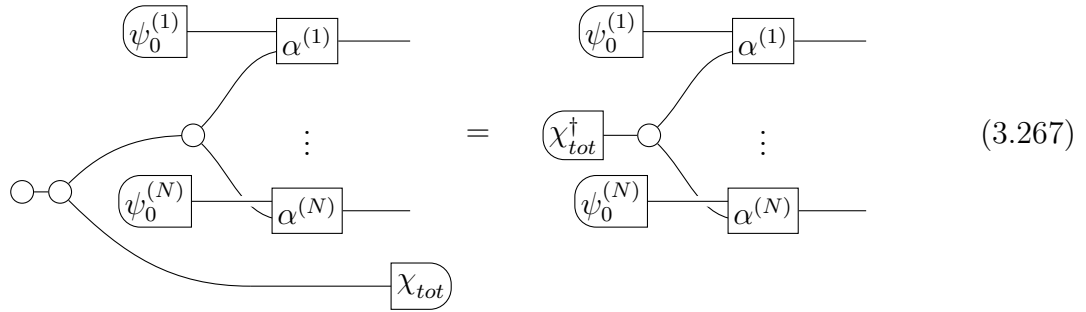
\begin{equation}\label{syncClockSystemPairEnergy}
	\input{pictures/chapter3/dynamics/syncClockSystemPairEnergy.tikz}
\end{equation}
To simplify our life in proving that this interpretation is sound in general, we will (this time only) assume that the energy levels are orthogonal. If we perform a Hamiltonian measurement on systems $\alpha^{(1)},...,\alpha^{(N)}$ and obtain energy levels $\goodchi^{(1)},...,\goodchi^{(N)}$, then a global state in the form of \ref{syncClockSystemPairEnergy} imposes the constraint $\goodchi^{(1)}\oplus ... \oplus \goodchi^{(N)} = \goodchi_{tot}$, proving that $\goodchi_{tot}$ behaves exactly like we would expect the total energy of the global dynamical system to behave:
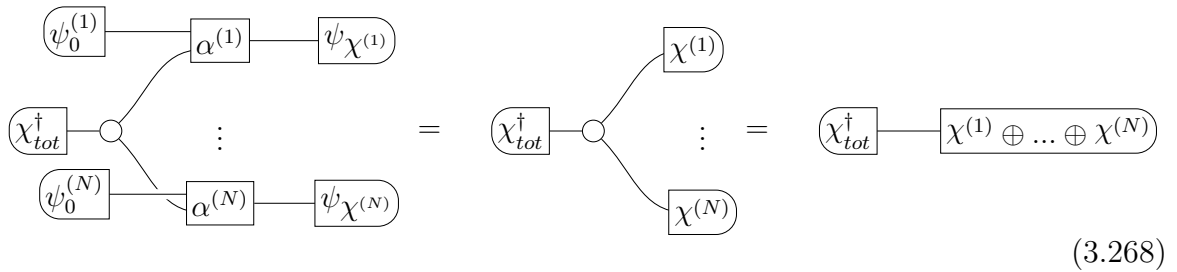
\begin{equation}\label{syncClockSystemPairEnergyProof}
	\resizebox{\textwidth}{!}{\input{pictures/chapter3/dynamics/syncClockSystemPairEnergyProof.tikz}}
\end{equation}
In $\fdHilbCategory$, the state of definite total energy given by Equation \ref{syncClockSystemPairEnergy} is a superposition of all possible combinations of states of definite energies $\goodchi^{(1)},...,\goodchi^{(N)}$ for the individual systems, exactly as would be expected:
\begin{equation}
	\input{pictures/chapter3/dynamics/totalEnergyStateSuperposition.tikz}
\end{equation}

This whole business of measuring a quantum clock in the clock energy observable also answer a pending question about the inaccessible external clocks we have talked so much about: what is a good way to denote their inaccessibility in the diagrammatic formalism? The answer turns out to lie in the clock energy measurement. For example, we should not have access to the time state of a universal clock (a particularly extreme case of external clock), but we sure can certainly impose the total energy that systems governed by said universal clock should have. Hence the act of making an external clock inaccessible coincides with the act of setting the total energy $\goodchi_{tot}$ for a group of synchronised dynamical systems, as done in Equation \ref{syncClockSystemPairEnergy}. Indeed, the time translations required to model the incessant marching of universal time can still be performed under post-selection on a total energy state, as long as we are willing to ignore the ensuing global phase:
\begin{equation}
	\input{pictures/chapter3/dynamics/systemClockUniclockEvolTimeTransl.tikz}
\end{equation}
This shows that there is a third way to address the issue of universal clocks: it may well be inevitable that they be made inaccessible in the modelling of any operational scenario, but the discussion above shows that this can be achieved without necessarily turning time into an external classical parameter.

\subsection{Time observables}
\label{subsection_timeObservables}

The introduction of synchronised clock-system states as modelling the relationship between a quantum dynamical system and the quantum clock governing its dynamics can be related to the problem of time observables. We have seen in the beginning of this Section that positing the existence of a universal quantum clock poses severe issues from both a philosophical and a physical perspective In the previous Subsection we have discussed how such a problem may be solved within our coherent framework, by positing the existence of inaccessible quantum clocks which govern the joint dynamics of the quantum dynamical systems in the various scenarios we might be interested in.

This approach leaves one important question open: how do quantum clocks emerge in the first place? How are they related between themselves? To answer these questions, we will show that certain quantum dynamical systems possess an \inlineQuote{internal} time observables, strongly complementary (in a suitable sense) to their Hamiltonian, and that these systems can be turned into quantum clocks governing all other systems in the global synchronised state. 

To begin with, if $\alpha: \SpaceH \otimes \SpaceG \rightarrow \SpaceG$ is a unitary representation of a coherent group $\mathbb{G}$ on an object $\SpaceH$, we say that a symmetric $\dagger$-qSFA $\hbox{\input{symbols/WbwdotSym.tex}}\!$ on $\SpaceH$ \textbf{internalises} $\alpha^\dagger$ if there is a $\hbox{\input{symbols/WbwdotSym.tex}}\!$-to-$\hbox{\input{symbols/DdotSym.tex}}\!\!$ classical process $s: \SpaceH \rightarrow \SpaceG$ such that:
\begin{equation}\label{internalisingCoherentHamiltonian}
	\input{pictures/chapter3/dynamics/internalisingCoherentHamiltonian.tikz}
\end{equation}
In fact, if one such $s$ exists then it is necessarily unique:
\begin{equation}\label{internalisingCoherentHamiltonianUnique}
	\input{pictures/chapter3/dynamics/internalisingCoherentHamiltonianUnique.tikz}
\end{equation}
In the context of quantum dynamical systems, $\hbox{\input{symbols/WbwdotSym.tex}}\!$ will called be an \textbf{internal Hamiltonian observable}, or \textbf{internal energy observable}: in $\fdHilbCategory$, this $\dagger$-qSFA corresponds to the traditional Hamiltonian observable. If $s$ is an isometry, then we will refer to $\hbox{\input{symbols/WbwdotSym.tex}}\!$ as a \textbf{non-degenerate} internal energy observable, because the isometry condition means that the quantum dynamical system has non-degenerate energy eigenspaces. The process $s$ simply maps the energy eigenstates, which act as internal labels for the energy levels of the dynamical system, to the corresponding clock energy states, which are the canonical labels for the energy levels of the dynamical system. 

Now that we have candidate clock energy observable for the dynamical system $\alpha$, we need to find an appropriate clock time observable to match it. 

\begin{theorem}[\textbf{Internal time observables}]\label{thm_internalTimeObservables}\hfill\\
Let $\mathbb{G} := (\hbox{\input{symbols/ZbwdotSym.tex}}\!\!,\hbox{\input{symbols/DdotSym.tex}}\!\!)$ be a doubly well-pointed coherent group on a system $\SpaceG$ of a $\dagger$-SMC $\CategoryC$, and let $\alpha$ be a unitary representation of $\mathbb{G}$ on a system $\SpaceH$. Assume that there is a non-degenerate internal energy observable $\hbox{\input{symbols/WbwdotSym.tex}}\!$ having enough classical states. Then the following implications both hold.
\begin{enumerate}
	\item[(i)] If the function $s : \classicalStates{\hbox{\input{symbols/WbwdotSym.tex}}\!} \rightarrow \classicalStates{\hbox{\input{symbols/DdotSym.tex}}\!\!}$ has image which is a subgroup $H$ of $(\classicalStates{\hbox{\input{symbols/DdotSym.tex}}\!\!},\!\hbox{\input{symbols/ZbwmultSym.tex}}\!\!,\!\hbox{\input{symbols/ZbwunitSym.tex}}\!\!)$, then there is a symmetric $\dagger$-qSCFA $\hbox{\input{symbols/YbwdotSym.tex}}\!\!$ on $\SpaceH$ such that $\mathbb{H}=(\hbox{\input{symbols/YbwdotSym.tex}}\!\!,\hbox{\input{symbols/WbwdotSym.tex}}\!)$ is a doubly well-pointed coherent group on system $\SpaceH$ and $s^\dagger$ is a coherent group homomorphism $s^\dagger:(\hbox{\input{symbols/ZbwdotSym.tex}}\!\!,\hbox{\input{symbols/DdotSym.tex}}\!\!) \rightarrow (\hbox{\input{symbols/YbwdotSym.tex}}\!\!,\hbox{\input{symbols/WbwdotSym.tex}}\!)$.
	\item[(ii)] Conversely, if  $\mathbb{H}=(\hbox{\input{symbols/YbwdotSym.tex}}\!\!,\hbox{\input{symbols/WbwdotSym.tex}}\!)$ is a doubly well-pointed coherent group on system $\SpaceH$ such that $s^\dagger$ is a coherent group homomorphism $s^\dagger:(\hbox{\input{symbols/ZbwdotSym.tex}}\!\!,\hbox{\input{symbols/DdotSym.tex}}\!\!) \rightarrow (\hbox{\input{symbols/YbwdotSym.tex}}\!\!,\hbox{\input{symbols/WbwdotSym.tex}}\!)$, then $s : \classicalStates{\hbox{\input{symbols/WbwdotSym.tex}}\!} \rightarrow \classicalStates{\hbox{\input{symbols/DdotSym.tex}}\!\!}$ has image which is a subgroup $H$ of $(\classicalStates{\hbox{\input{symbols/DdotSym.tex}}\!\!},\!\hbox{\input{symbols/ZbwmultSym.tex}}\!\!,\!\hbox{\input{symbols/ZbwunitSym.tex}}\!\!)$.
\end{enumerate}
In both cases, $s^\dagger$ restricts to a quotient group homomorphism $s^\dagger: (\classicalStates{\hbox{\input{symbols/ZbwdotSym.tex}}\!\!},\!\hbox{\input{symbols/DmultSym.tex}}\!\!,\!\hbox{\input{symbols/DunitSym.tex}}\!\!) \rightarrow (\classicalStates{\hbox{\input{symbols/YbwdotSym.tex}}\!\!},\hbox{\input{symbols/WbwmultSym.tex}}\!,\hbox{\input{symbols/WbwunitSym.tex}}\!)$, showing that $(\classicalStates{\hbox{\input{symbols/YbwdotSym.tex}}\!\!},\hbox{\input{symbols/WbwmultSym.tex}}\!,\hbox{\input{symbols/WbwunitSym.tex}}\!) \isom (\classicalStates{\hbox{\input{symbols/ZbwdotSym.tex}}\!\!},\!\hbox{\input{symbols/DmultSym.tex}}\!\!,\!\hbox{\input{symbols/DunitSym.tex}}\!\!)/H^\wedge$. 
\end{theorem}
\begin{proof}
We begin by proving implication (i). Define $\!\hbox{\input{symbols/YbwcomultSym.tex}}\!\!$ and $\!\hbox{\input{symbols/YbwcounitSym.tex}}\!\!$ as follows:
\begin{equation}\label{internalTimeObservable}
	\input{pictures/chapter3/dynamics/internalTimeObservable.tikz}
\end{equation}
The $\dagger$-qSCFA $\hbox{\input{symbols/WbwdotSym.tex}}\!$ has enough classical states, $s$ is a classical injection on $\hbox{\input{symbols/WbwdotSym.tex}}\!$-classical states (because it is an isometry) and it has a subgroup $H$ of $(\classicalStates{\hbox{\input{symbols/DdotSym.tex}}\!\!},\!\hbox{\input{symbols/ZbwmultSym.tex}}\!\!,\!\hbox{\input{symbols/ZbwunitSym.tex}}\!\!)$ as its image. As a consequence, it is immediate to check that  $\!\hbox{\input{symbols/YbwcomultSym.tex}}\!\!$ and $\!\hbox{\input{symbols/YbwcounitSym.tex}}\!\!$ form, together with their adjoints, a symmetric $\dagger$-qSFA $\hbox{\input{symbols/YbwdotSym.tex}}\!\!$: one evaluates the equations on $\hbox{\input{symbols/WbwdotSym.tex}}\!$-classical states, pushes the states through the classical injection, and checks the validity of the equations for $\hbox{\input{symbols/ZbwdotSym.tex}}\!\!$, which satisfies all of them because it is a $\dagger$-qSCFA. 

Since the image of the classical map $s$ is a subgroup $H$ of $(\classicalStates{\hbox{\input{symbols/DdotSym.tex}}\!\!},\!\hbox{\input{symbols/ZbwmultSym.tex}}\!\!,\!\hbox{\input{symbols/ZbwunitSym.tex}}\!\!)$, it is also immediate to check that $(\!\hbox{\input{symbols/YbwmultSym.tex}}\!\!,\!\hbox{\input{symbols/YbwunitSym.tex}}\!\!)$ endows $\classicalStates{\hbox{\input{symbols/WbwdotSym.tex}}\!}$ with the structure of a group: but $\hbox{\input{symbols/WbwdotSym.tex}}\!$ has enough classical states, and hence $(\hbox{\input{symbols/YbwdotSym.tex}}\!\!,\hbox{\input{symbols/WbwdotSym.tex}}\!)$ is a coherent group. A similar argument can be used to prove that $s$ is a coherent group homomorphism $s:(\hbox{\input{symbols/WbwdotSym.tex}}\!,\hbox{\input{symbols/YbwdotSym.tex}}\!\!) \rightarrow (\hbox{\input{symbols/DdotSym.tex}}\!\!,\hbox{\input{symbols/ZbwdotSym.tex}}\!\!)$, and as a consequence $s^\dagger$ is a coherent group homomorphism $s^\dagger:(\hbox{\input{symbols/ZbwdotSym.tex}}\!\!,\hbox{\input{symbols/DdotSym.tex}}\!\!) \rightarrow (\hbox{\input{symbols/YbwdotSym.tex}}\!\!,\hbox{\input{symbols/WbwdotSym.tex}}\!)$. Finally, $\hbox{\input{symbols/YbwdotSym.tex}}\!\!$ has enough classical states because: (a) $s^\dagger$ is $\hbox{\input{symbols/ZbwdotSym.tex}}\!\!$-to-$\hbox{\input{symbols/YbwdotSym.tex}}\!\!$ classical and surjective on $\hbox{\input{symbols/YbwdotSym.tex}}\!\!$-classical states; (b) $\hbox{\input{symbols/ZbwdotSym.tex}}\!\!$ has enough classical states; (c) $s$ is an isometry.

The proof of implication (ii) goes along similar lines: if  $(\hbox{\input{symbols/YbwdotSym.tex}}\!\!,\hbox{\input{symbols/WbwdotSym.tex}}\!)$ is a coherent group on $\SpaceH$ and $s^\dagger$ is a coherent group homomorphism $s^\dagger:(\hbox{\input{symbols/ZbwdotSym.tex}}\!\!,\hbox{\input{symbols/DdotSym.tex}}\!\!) \rightarrow (\hbox{\input{symbols/YbwdotSym.tex}}\!\!,\hbox{\input{symbols/WbwdotSym.tex}}\!)$, then $s$ must be a coherent group homomorphism $s:(\hbox{\input{symbols/WbwdotSym.tex}}\!,\hbox{\input{symbols/YbwdotSym.tex}}\!\!) \rightarrow (\hbox{\input{symbols/DdotSym.tex}}\!\!,\hbox{\input{symbols/ZbwdotSym.tex}}\!\!)$, and hence the image of $s$ must be a subgroup $H$ of $(\classicalStates{\hbox{\input{symbols/DdotSym.tex}}\!\!},\!\hbox{\input{symbols/ZbwmultSym.tex}}\!\!,\!\hbox{\input{symbols/ZbwunitSym.tex}}\!\!)$.

In both cases, $s^\dagger$ restricted to $\hbox{\input{symbols/ZbwdotSym.tex}}\!\!$-classical states must be a surjective group homomorphism $s^\dagger: (\classicalStates{\hbox{\input{symbols/ZbwdotSym.tex}}\!\!},\!\hbox{\input{symbols/DmultSym.tex}}\!\!,\!\hbox{\input{symbols/DunitSym.tex}}\!\!) \rightarrow (\classicalStates{\hbox{\input{symbols/YbwdotSym.tex}}\!\!},\hbox{\input{symbols/WbwmultSym.tex}}\!,\hbox{\input{symbols/WbwunitSym.tex}}\!)$, hence proving that  $(\classicalStates{\hbox{\input{symbols/YbwdotSym.tex}}\!\!},\hbox{\input{symbols/WbwmultSym.tex}}\!,\hbox{\input{symbols/WbwunitSym.tex}}\!) \isom (\classicalStates{\hbox{\input{symbols/ZbwdotSym.tex}}\!\!},\!\hbox{\input{symbols/DmultSym.tex}}\!\!,\!\hbox{\input{symbols/DunitSym.tex}}\!\!)/H^\wedge$.
\end{proof}

Theorem \ref{thm_internalTimeObservables} is an extremely important result for this framework: it gives a characterisation of certain quantum dynamical systems which can be taken to behave as quantum clocks, i.e. which can be endowed with the structure of a coherent group which is compatible (via the coherent group homomorphism $s^\dagger$) with the original dynamics. We work out the details of the most general example of this phenomenon in $\fdHilbCategory$, for discrete periodic dynamics. 

Consider a discrete periodic quantum clock $\GroupG$ in $\fdHilbCategory$, given in its most general form by the group algebra $\SpaceG := \complexs[\integersMod{T}]$ for some $T$, and let $\alpha: \SpaceH \otimes \SpaceG \rightarrow \SpaceH$ be a quantum dynamical system governed by $\GroupG$, associated to a unitary representation $(U_t)_{t \in \integersMod{T}}$ of $\integersMod{T}$ on $\SpaceH$. A non-degenerate internal energy observable $\hbox{\input{symbols/WbwdotSym.tex}}\!$ (a $\dagger$-qSCFA with normalisation factor $N_{\hbox{\input{symbols/WbwdotSym.tex}}\!} = |\integersMod{T}|=|H||H^\wedge|$ fixed by the requirement that $s^\dagger$ be an isometry) exists if we have the following decomposition for the representation, where $H \subseteq \integersMod{T}^\wedge$ is any non-empty subset and $(\ket{\psi_{\goodchi}})_{\goodchi\in H}$ is an orthogonal basis (the classical states for $\hbox{\input{symbols/WbwdotSym.tex}}\!$):
\begin{equation}
	U_t := \sum_{\goodchi \in H} \goodchi(t) \frac{1}{|H||H^\wedge|}\ket{\psi_{\goodchi}}\bra{\psi_{\goodchi}}
\end{equation}
When $H$ is a subgroup, we can consider the quotient $G' := \integersMod{T}/H^\wedge$, and we label the elements of the quotient by the cosets $t \oplus H^\wedge$ of $H^\wedge$ in $\integersMod{T}$. We can then consider the following family $(\ket{t\oplus H^\wedge})_{(t\oplus H^\wedge) \in G'}$ of states:
\begin{equation}
		\ket{t\oplus H^\wedge} := \sum_{\goodchi \in H} \goodchi(t) \frac{1}{|H||H^\wedge|} \ket{\psi_{\goodchi}}
\end{equation}
It is not hard to check that the family is an orthogonal basis corresponding to a $\dagger$-qSCFA $\hbox{\input{symbols/YbwdotSym.tex}}\!\!$ on $\SpaceH$ (with normalisation factor $N_{\hbox{\input{symbols/YbwdotSym.tex}}\!\!} = 1/|H^\wedge|$ fixed by the requirement that $s^\dagger$ be an isometry):
\begin{align}
		\braket{s \oplus H^\wedge}{t\oplus H^\wedge} &=  \sum_{\goodchi,\goodchi' \in H} (\goodchi')^\ast(s)\goodchi(t) \frac{1 }{|H||H^\wedge|}\frac{1}{N_{\hbox{\input{symbols/WbwdotSym.tex}}\!}}\braket{\psi_{\goodchi'}}{\psi_{\goodchi}} = \nonumber\\
		&= \frac{1}{|H||H^\wedge|}\sum_{\goodchi \in H} \goodchi(t\ominus s) = 
		\begin{cases}
			\frac{1}{|H^\wedge|} \text{ if } s \oplus H = t \oplus H \\
			0 \text{ otherwise}
		\end{cases}
\end{align}
We want the unit $\hbox{\input{symbols/WbwunitSym.tex}}\!$ to be the state $\ket{0\oplus H^\wedge}$, and indeed our choice of normalisation factors yields the following equality:
\begin{equation}
	|H^\wedge|\hbox{\input{symbols/WbwcounitSym.tex}}\! \circ \ket{t\oplus H^\wedge} = \frac{1}{|H|} \sum_{\goodchi \in H} \goodchi(t)  = 
	\begin{cases}
		1 \text{ if } t = 0 \\
		0 \text{ otherwise}
	\end{cases}
\end{equation}
The multiplication $\hbox{\input{symbols/WbwmultSym.tex}}\!$ acts as the group multiplication of $G'$ on the family:
\begin{align}
	\hbox{\input{symbols/WbwmultSym.tex}}\! \circ (\ket{t\oplus H^\wedge} \otimes \ket{s\oplus H^\wedge}) &= \frac{1}{|H||H^\wedge|}\sum_{\goodchi \in H}  \goodchi(s\oplus t) \ket{\psi_{\goodchi}} = \ket{(s\oplus t)\oplus H^\wedge}
\end{align}
As a consequence, the pair $\GroupG':=(\hbox{\input{symbols/YbwdotSym.tex}}\!\!,\hbox{\input{symbols/WbwdotSym.tex}}\!)$ is a doubly well-pointed coherent group on $\SpaceH$, with underlying group $\underlyingGroup{\GroupG'} = G'$. The map $s^\dagger$ specified by Equation \ref{internalisingCoherentHamiltonianUnique} is given explicitly as follows, and restricts to the quotient group homomorphism $\integersMod{T} \rightarrow \integersMod{T}/H^\wedge$ when evaluated on $\hbox{\input{symbols/ZbwdotSym.tex}}\!\!$-classical states:
\begin{align}
	s^\dagger &= \frac{1}{N_{\hbox{\input{symbols/WbwdotSym.tex}}\!}}\sum_{t \in \integersMod{T}}\sum_{\goodchi \in H} \goodchi(t) \ket{\psi_{\goodchi}}\bra{t} 
	= \frac{1}{ N_{\hbox{\input{symbols/WbwdotSym.tex}}\!}}\sum_{(t \oplus H^\wedge) \in G'}\sum_{t' \in (t \oplus H^\wedge)}\sum_{\goodchi \in H} \goodchi(t') \ket{\psi_{\goodchi}}\bra{t'} = \nonumber \\
	&= \sum_{(t \oplus H^\wedge) \in G'}\sum_{\goodchi \in H} \goodchi(t) \frac{1}{|H||H^\wedge|}\ket{\psi_{\goodchi}} \hspace{-3mm}\sum_{t' \in (t \oplus H^\wedge)}\hspace{-2mm} \bra{t'} \hspace{2mm}= \sum_{(t \oplus H^\wedge) \in G'}\hspace{-2mm}\ket{t \oplus H^\wedge}\Big(\sum_{t' \in (t \oplus H^\wedge)}\bra{t'}\Big)
\end{align}

Theorem \ref{thm_internalTimeObservables} is a rather general result, but by itself it does not cover all interesting cases of internal time observables: the relationship between external and internal time only involves a quotient, with no space for any kind of \inlineQuote{coarsening} of the time being ticket. Indeed, consider consider the following setup with two synchronised clocks: one is a wall clock, ticking 12 hours in intervals of one minute, and one is a chronograph, ticking 24 hours in intervals of 1/100 of a second. It is sensible to say that the wall clock is a dynamical system governed by the chronograph in some appropriate sense, but the relationship between wall clock (internal) time and chronograph (external) time is not simply given by a quotient: first one needs to discretise the chronograph to a digital clock, ticking 24 hours in intervals of one second, and only at that point a quotient can be taken. Mathematically, this amounts to first considering the subgroup $\integersMod{86,400} \normalSubgroup \integersMod{8,640,000}$ (going from $8,640,000$ 1/100 seconds in 24h to $86,400$ seconds in 24 hours), and then considering the quotient $\integersMod{86,400} \rightarrow \integersMod{43,200}$ (going from  $86,400$ seconds in 24 hours to $43,200$ seconds in 12 hours). 

Consider two coherent groups $\GroupG,\GroupG'$ on objects $\SpaceG,\SpaceG'$ of a $\dagger$-SMC, and a coherent group homomorphism $f : \GroupG' \rightarrow \GroupG'$. If $\alpha:\SpaceH \otimes \SpaceG \rightarrow \SpaceH$ is a unitary representation of $\GroupG$ on $\SpaceH$, then it is easy to check that $\gamma:=\alpha \circ (\id{\SpaceH} \otimes f)$ is a unitary representation of $\GroupG'$ on $\SpaceH$. It is possible, therefore, that the assumptions of Theorem \ref{thm_internalTimeObservables} might not apply to $\alpha$ and $\GroupG$, because $H$ is not a subgroup of $\underlyingGroup{\GroupG}$, but that the assumptions do apply once we move to $\gamma$ and $\GroupG'$: in the wall-clock vs chronograph example above, $\alpha$ is the wall clock seen as a system governed by the chronograph $\GroupG$, and $\gamma$ is the wall clock seen as a system governed by the digital clock $\GroupG'$, a coarsening of the chronograph. We work out the details of this phenomenon in $\fdHilbCategory$, for discrete periodic dynamics. 

Consider again a discrete periodic doubly well-pointed quantum clock $\GroupG$ in $\fdHilbCategory$, given in its most general form by the group algebra $\SpaceG := \complexs[\integersMod{T}]$ for some $T$, and let $\alpha: \SpaceH \otimes \SpaceG \rightarrow \SpaceH$ be a quantum dynamical system governed by $\GroupG$, associated to a unitary representation $(U_t)_{t \in \integersMod{T}}$ of $\integersMod{T}$ on $\SpaceH$. Consider again a non-degenerate internal energy observable $\hbox{\input{symbols/WbwdotSym.tex}}\!$ (a $\dagger$-qSCFA with normalisation factor $N_{\hbox{\input{symbols/WbwdotSym.tex}}\!}$) with $H \subseteq \integersMod{T}^\wedge$ any non-empty subset and  $(\ket{\psi_{\goodchi}})_{\goodchi\in H}$ the orthogonal basis of classical states for $\hbox{\input{symbols/WbwdotSym.tex}}\!$:
\begin{equation}
	U_t := \sum_{\goodchi \in H} \goodchi(t) \frac{1}{N_{\hbox{\input{symbols/WbwdotSym.tex}}\!}}\ket{\psi_{\goodchi}}\bra{\psi_{\goodchi}}
\end{equation}
In this more general case, $H$ is not necessarily a subgroup of $\integersMod{T}^\wedge$. Assume, however, that there is a subgroup injection $i: \integersMod{T'} \rightarrow \integersMod{T}$ (i.e. $T = m T'$ for some $m \in \naturals^+$), and assume that $H':=\suchthat{\goodchi \circ i}{\goodchi \in H}$ is a subgroup of $\integersMod{T'}^\wedge$, where we have:
\begin{equation}
	\big(\goodchi \circ i\big)(t') = \goodchi(m\,t')
\end{equation}
In this context, we could apply Theorem \ref{thm_internalTimeObservables} to the quantum dynamical system $\gamma:=\alpha \circ (\id{\SpaceH} \otimes i)$ governed by the quantum clock $\GroupG'$ specified by the group algebra $\complexs[\integersMod{T'}]$ (we have written $i:\GroupG'\rightarrow\GroupG$ for the coherent group homomorphism specified by the subgroup injection $i: \integersMod{T'} \rightarrow \integersMod{T}$). 

As a concrete example, we consider the quantum clocks associated with the wall-clock, chronograph and digital clock examples given above. The chronograph corresponds to a quantum clock $\GroupG$ with underlying group $\integersMod{T}$ for $T=8,640,000$, while the digital clock corresponds to a quantum clock $\GroupG'$ with underlying group $\integersMod{T'}$ for $T'=86,400$. The wall clock as a quantum dynamical system $\alpha$ governed by the chronograph is given by the following unitary representation $(U_t)_{t \in \integersMod{T}}$ (where $(\ket{\psi_k})_{k=0}^{43,200-1}$):
\begin{equation}
	U_t := \sum_{k=0}^{43,200-1} e^{i2\pi \frac{2k t}{8,640,000}} \ket{\psi_k}\bra{\psi_k}
\end{equation}
The subset $H:=\suchthat{t \mapsto e^{i2\pi \frac{2k t}{8,640,000}}}{k=0,...,43,200-1}$ is not a subgroup of $\integersMod{T}$. The wall clock as a quantum dynamical system $\beta$ governed by the chronograph is given by the following unitary representation $(V_{t'})_{t' \in \integersMod{T'}}$:
\begin{equation}
	V_{t'} := \sum_{k=0}^{43,200-1} e^{i2\pi \frac{2k t'}{86,400}} \ket{\psi_k}\bra{\psi_k}
\end{equation}
The subset $H':=\suchthat{t \mapsto e^{i2\pi \frac{2k t}{86,400}}}{k=0,...,43,200-1}$ is a subgroup of $\integersMod{T'}$, with $(H')^\wedge \isom \integersMod{2}$. We can apply Theorem \ref{thm_internalTimeObservables} to the wall clock $\gamma$ governed by the digital clock $\GroupG'$, and we conclude that the wall clock is a quantum clock itself, with underlying group $\integersMod{86,400}/\integersMod{2}\isom \integersMod{43,200}$.

We have established above that under certain conditions quantum dynamical systems are quantum clocks, but what use is a clock if it is not synchronised with other dynamical systems? In other words: we have established that certain quantum dynamical systems \textit{have the algebraic structure of} quantum clocks, but we have not yet shown that they \textit{behave operationally as quantum clocks}. Consider the general scenario of Equation \ref{syncClockSystemPairEnergy}, where quantum dynamical systems $\alpha^{(1)},...,\alpha^{(N)}$ are synchronised between themselves and governed by some inaccessible clock, which has been forgotten by setting a total energy $\goodchi_{tot}$ for the systems. If one of the dynamical systems, say $\alpha^{(N)}$, can be turned into a quantum clock by virtue of Theorem \ref{thm_internalTimeObservables}, then we would expect it to govern the remaining dynamical systems $\alpha^{(1)},...,\alpha^{(N-1)}$ (we can treat the latter as a single quantum dynamical system $\beta$, and we will simply write $\alpha$ for the quantum dynamical system $\alpha^{(N)}$ being promoted to the role of clock). Theorem \ref{thm_emergingQuantumClocks} below shows that this is indeed the case, and hence completes the picture on the emergence of quantum clocks: quantum clocks begin their lives as quantum dynamical systems, and rise to the challenge when the quantum clock originally governing their dynamics is forgotten (by imposing a total energy constraint). 
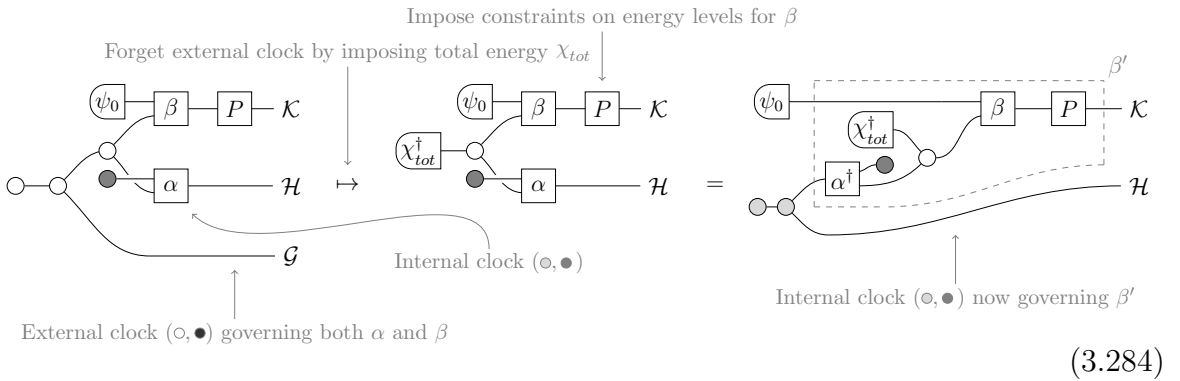
\begin{equation}
	\resizebox{\textwidth}{!}{\input{pictures/chapter3/dynamics/emergenceOfQuantumClocks.tikz}}
\end{equation}
In order for this to happen, however, one must first ensure that the quantum dynamical system $\beta$ is restricted to energy levels which are compatible with the internal clock. This is achieved by the introduction of a projector $P$ onto the subspace spanned by the energy eigenstates for $\beta$ corresponding to energy levels of the internal clock $\alpha$ (seen as a subgroup of the energy levels of the external clock):
\begin{equation}\label{emergingClockRepProjectorExplained}
	\input{pictures/chapter3/dynamics/emergingClockRepProjectorExplained.tikz}
\end{equation}
In order to treat the constraint enforced by the projector $P$ in a categorical fashion, we work in the \textbf{$\dagger$-Karoubi envelope} $\DaggerKaroubiEnvelope{\CategoryC}$ of our original $\dagger$-SMC $\CategoryC$: 
\begin{enumerate}
	\item objects are now pairs $(\SpaceH,p)$ where $\SpaceH$ is an object of the original category $\CategoryC$ and $p : \SpaceH \rightarrow \SpaceH$ is a self-adjoint idempotent (i.e. a projector);
	\item morphisms $f: (\SpaceH,p) \rightarrow (\SpaceK,q)$ are exactly the morphisms $f:\SpaceH \rightarrow\SpaceK$ in $\CategoryC$ which are invariant under the projectors, i.e. those satisfying $f = q \circ f \circ p$.
\end{enumerate}
The $\dagger$-Karoubi envelope is itself a $\dagger$-SMC, and we can see $\CategoryC$ canonically as the full sub-$\dagger$-SMC of $\DaggerKaroubiEnvelope{\CategoryC}$ given by objects in the form $(\SpaceH,\id{\SpaceH})$, which we will simply write as $\SpaceH$ when no confusion can arise. In particular, quantum dynamical systems and quantum clocks in $\CategoryC$ are also quantum dynamical systems and quantum clocks in $\DaggerKaroubiEnvelope{\CategoryC}$. Working in $\DaggerKaroubiEnvelope{\CategoryC}$ is the same as working in $\CategoryC$ while additionally enforcing constraints on states (or effects, or processes) specified by the projectors. 

\begin{theorem}[\textbf{Emerging quantum clocks}]\label{thm_emergingQuantumClocks}\hfill\\
Assume that the \inlineQuote{external} quantum clock $\mathbb{G} := (\hbox{\input{symbols/ZbwdotSym.tex}}\!\!,\hbox{\input{symbols/DdotSym.tex}}\!\!)$, the quantum dynamical system $\alpha$ and the associated \inlineQuote{internal} quantum clock $\mathbb{H} := (\hbox{\input{symbols/YbwdotSym.tex}}\!\!,\hbox{\input{symbols/WbwdotSym.tex}}\!)$ are as in Theorem \ref{thm_internalTimeObservables}. Let $\beta: \SpaceK \otimes \SpaceG \rightarrow \SpaceK$ be another quantum dynamical system governed by $\mathbb{G}$, and let $P:\SpaceH \rightarrow \SpaceH$ be the map defined by Diagram \ref{emergingClockRepProjectorExplained}. Then $P$ is a self-adjoint idempotent (i.e. a projector), which commutes with $\beta$. Furthermore, the following is a quantum dynamical system governed by $\mathbb{H}$ on the object $(\SpaceK,P)$ of the $\dagger$-Karoubi envelope, for all possible choices of total energy $\goodchi^\dagger \in \classicalStates{\hbox{\input{symbols/DdotSym.tex}}\!\!}$.
\begin{equation}\label{emergingClockRep}
	\input{pictures/chapter3/dynamics/emergingClockRep.tikz}
\end{equation}
\end{theorem}
\begin{proof}
We being by proving the following multiplicativity result:
\begin{equation}\label{emergingClockRepProof1mult}
	\resizebox{\textwidth}{!}{\input{pictures/chapter3/dynamics/emergingClockRepProof1mult.tikz}}
\end{equation}
\begin{equation}\label{emergingClockRepProof2mult}
	\resizebox{\textwidth}{!}{\input{pictures/chapter3/dynamics/emergingClockRepProof2mult.tikz}}
\end{equation}
Similarly, we can prove the following inversion result:
\begin{equation}
	\resizebox{\textwidth}{!}{\input{pictures/chapter3/dynamics/emergingClockRepProofInverse.tikz}}
\end{equation}
Using the results above (and the unit laws for $\hbox{\input{symbols/WbwdotSym.tex}}\!$), the map $P$ defined in Diagram \ref{emergingClockRepProjectorExplained} is seen to be a self-adjoint idempotent (i.e. a projector). Similarly, the projector $P$ can be seen to commute with the representation $\beta$, in the following sense:
\begin{equation}
	\input{pictures/chapter3/dynamics/emergingClockRepProofProjectorCommute.tikz}
\end{equation}
Putting the results above together, we conclude that the map of Diagram \ref{emergingClockRep} is a unitary representation of the quantum group $\mathbb{H}$ on the object $(\SpaceK,P)$ of the $\dagger$-Karoubi envelope.
\end{proof}

Once again, we work out the details of this phenomenon in $\fdHilbCategory$, for discrete periodic dynamics. Consider a discrete periodic quantum clock $\GroupG:=(\hbox{\input{symbols/ZbwdotSym.tex}}\!\!,\hbox{\input{symbols/DdotSym.tex}}\!\!)$ in $\fdHilbCategory$, given in its most general form by the group algebra $\SpaceG := \complexs[\integersMod{T}]$ for some $T$. Let $\alpha: \SpaceH \otimes \SpaceG \rightarrow \SpaceH$ be a quantum dynamical system governed by $\GroupG$, which is itself a quantum clock $\GroupH := (\hbox{\input{symbols/YbwdotSym.tex}}\!\!,\hbox{\input{symbols/WbwdotSym.tex}}\!)$ given by the group algebra $\SpaceH = \complexs[\integersMod{T'}]$ for $T=mT'$, and is related to $\GroupG$ by the quotient group homomorphism $\integersMod{T} \rightarrow \integersMod{T'}$. Let $\beta: \SpaceK \otimes \SpaceG \rightarrow \SpaceK$ be another quantum dynamical system governed by $\GroupG$, associated to a unitary representation $(U_t)_{t \in \integersMod{T}}$ which we write in its most general form as follows:
\begin{equation}
	U_t := \sum_{\goodchi \in \integersMod{T}^\wedge} \goodchi(t) P_{\goodchi}
\end{equation}
Now we look at the quantum dynamical system of Diagram \ref{emergingClockRep}, governed by the quantum clock $\mathbb{H}$; we set total energy $\goodchi_{tot} \in \integersMod{T}^\wedge$. In terms of energy eigenstates, the map $s: \SpaceH \rightarrow \SpaceG$  takes the following form:
\begin{equation}
	s := \frac{1}{T}\sum_{\goodchi' \in \integersMod{T'}^\wedge} \ket{\goodchi' \circ q}\bra{\goodchi}
\end{equation}
Seen another way, $s: \classicalStates{\hbox{\input{symbols/WbwdotSym.tex}}\!} \rightarrow \classicalStates{\hbox{\input{symbols/DdotSym.tex}}\!\!}$ corresponds to the injective group homomorphism $q^\wedge: \integersMod{T'}^\wedge \rightarrow \integersMod{T}^\wedge$ given by $q^\wedge := \goodchi' \mapsto \goodchi' \circ q$, where  $q: \integersMod{T} \rightarrow \integersMod{T'}$ is the quotient group homomorphism given by $q:=t \mapsto (\modclass{t}{T'})$ (recall that the normalisation factor for $\classicalStates{\hbox{\input{symbols/WbwdotSym.tex}}\!}$ is fixed to $N_{\hbox{\input{symbols/WbwdotSym.tex}}\!} = T$, so that $\frac{1}{T}\braket{\goodchi'}{\goodchi''} = \delta_{\goodchi',\goodchi''}$). 

Without loss of generality and to simplify the discussion, we will take $\goodchi_{tot}$ to be the ground energy level (the general case simply involves a translation $\goodchi' \circ q  \mapsto \goodchi' \circ q \oplus \goodchi_{tot}$ of the energy levels for $\beta$ by $\goodchi_{tot}$). Consider those energy level $\goodchi' \circ q$ of the external quantum clock $\GroupG$ which are compatible with some energy level $\goodchi' \in \integersMod{T'}^\wedge$ of the inner quantum clock $\mathbb{H}$. We have the following expression for the projectors onto the eigenspaces of the quantum dynamical system $\beta$ corresponding to those energy levels:
\begin{equation}
	\input{pictures/chapter3/dynamics/emergentQuantumClockExampleProjectors.tikz}
\end{equation}
As a consequence, it is immediate to see that the projector $P$ defined by Diagram \ref{emergingClockRepProjectorExplained} corresponds to the subspace spanned by all the eigenstates of $\beta$ corresponding to energy levels compatible with the inner quantum clock $\mathbb{H}$:
\begin{equation}
	\resizebox{\textwidth}{!}{\input{pictures/chapter3/dynamics/emergentQuantumClockExampleWholeProjector.tikz}}
\end{equation}
Now write $(V_{t'})_{t' \in \integersMod{T'}}$ for the unitary representation of $\integersMod{T'}$ on $\SpaceK$ corresponding to the quantum group representation given by Diagram \ref{emergingClockRep}:
\begin{equation}
	\input{pictures/chapter3/dynamics/emergentQuantumClockExampleUnitaries.tikz}
\end{equation}
Then we get the following explicit expression for the unitary representation $(V_{t'})_{t' \in \integersMod{T'}^\wedge}$ describing the quantum dynamical system as governed by the internal quantum clock:
\begin{equation}
	\input{pictures/chapter3/dynamics/emergentQuantumClockExampleUnitariesExplicit.tikz}
\end{equation}
Luckily, this is exactly what we would have expected from Stone's Theorem!

We believe that Theorems \ref{thm_internalTimeObservables} and \ref{thm_emergingQuantumClocks}, together with the worked out examples that follow them, provide an interesting new perspective on time observables in quantum theory: they show a sense in which time observables do indeed exist, and how quantum clocks can be seen to emerge from synchronised quantum dynamical systems under appropriate conditions. In particular, we believe to have shed some light on  Schr\"{o}dinger's 1931 conundrum: when a quantum clock $\mathbb{H}$ arises by synchronisation with a (\inlineQuote{forgotten}, or \inlineQuote{inaccessible}) external quantum clock $\mathbb{G}$, the dynamics of $\mathbb{H}$ are governed by the external clock energy observable on $\mathbb{G}$, which commutes with the internal time observable on $\mathbb{H}$. In particular, joint eigenstates for the two observables exist, so that the quantum clock $\mathbb{H}$ can be in a definite external clock energy state and a definite internal time state at the same time, exactly as Schr\"{o}dinger desired.

In describing the mechanism of emergence of quantum clocks, we have provided ways to relate synchronised quantum clocks by increased periodicity and increased discretisation. In our exemplification on finite-dimensional quantum systems with discrete periodic dynamics, increased periodicity corresponds to a group quotient $q:\integersMod{T} \rightarrow \integersMod{T'}$, while increased discretisation corresponds to a subgroup injection $i: \integersMod{T'} \rightarrow \integersMod{T}$. The techniques and results we have proven are fully general, and can be straightforwardly extended to other notions of dynamics (e.g. using the non-standard framework for infinite-dimensional CQM presented earlier in this chapter).

\begin{remark}
As the very last remark to this Section and Chapter, we sketch a brief argument in favour of the construction of a toy model of emergent quantum time for finite-dimensional quantum systems based on discrete periodic dynamics alone: we do so in the hope that it will provide further evidence that the discrete periodic case is worthy of study in itself. The detailed construction of this toy model, including its abstraction and categorification, will be the subject of future work. 

Consider a $d$-dimensional quantum dynamical system $\alpha$ governed by continuous dynamics, corresponding to a 1-parameter unitary group $(U_t)_{t \in \reals}$. Use Stone's Theorem to write the unitary group in the following form, for some $\omega_1,...,\omega_J \in \reals$ and a complete family of orthogonal projectors $(P_j)_{j=1}^J$:
\begin{equation}
	U_t = \sum_{j=1}^J e^{i2\pi \omega_j t} P_j
\end{equation}
Now consider two strictly increasing sequences of non-zero natural numbers $(Q^{(k)})_{k \in \naturals}$ and $(R^{(k)})_{k \in \naturals}$, such that $Q^{(k)} | Q^{(k+1)}$ and $R^{(k)} | R^{(k+1)}$ for all $k \in \naturals$, and define the following coefficients:
\begin{equation}
	a_j^{(k)} := \lfloor \omega_j Q^{(k)} \rfloor
\end{equation}
Now, consider the following families of unitary representations $(U_t^{(k)})_{t \in \integersMod{T^{(k)}}}$, $(V_t^{(k)})_{t \in \integersMod{S^{(k)}}}$ and  $(W_t^{(k)})_{t \in \integersMod{S^{(k)}}}$, where $T^{(k)} := Q^{(k)}R^{(k)}$ and $S^{(k)} := Q^{(k+1)}R^{(k)}$:
\begin{align}
	U_t^{(k)} :=  \sum_{j=1}^J e^{i2\pi \frac{a_j^{(k)}}{Q^{(k)}} \frac{t}{R^{(k)}} } P_j \\
	V_t^{(k)} :=  \sum_{j=1}^J e^{i2\pi \frac{a_j^{(k+1)}}{Q^{(k+1)}} \frac{t}{R^{(k)}} } P_j \\
	W_t^{(k)} :=  \sum_{j=1}^J e^{i2\pi \frac{a_j^{(k)}}{Q^{(k)}} \frac{t}{R^{(k)}} } P_j
\end{align}
Designate by $\alpha_U^{(k)}$, $\alpha_V^{(k)}$ and $\alpha_W^{(k)}$ the corresponding quantum dynamical system, governed by the quantum clocks $\complexs[\integersMod{T^{(k)}}]$ and $\complexs[\integersMod{S^{(k)}}]$. Then the quantum dynamical systems $\alpha_U^{(k)}$ are converging approximations of the quantum dynamical system $\alpha$, in the sense that as $k\rightarrow \infty$ and $t^{(k)}/R^{(k)} \rightarrow t \in \reals$ we get that $U_{t^{(k)}}^{(k)} \rightarrow U_t$ in operator norm.

Consider a scenario in which the quantum clock $\complexs[\integersMod{S^{(k)}}]$ is related to the quantum clock  $\complexs[\integersMod{T^{(k+1)}}]$ by increased discretisation, via the subgroup injection $\integersMod{ Q^{(k+1)}R^{(k)}} \rightarrow \integersMod{Q^{(k+1)}R^{(k+1)}}$. From the discussion following Theorems \ref{thm_internalTimeObservables} and \ref{thm_emergingQuantumClocks}, it is easy to check that quantum dynamical system $\alpha_U^{(k+1)}$, governed by quantum clock $\complexs[\integersMod{T^{(k+1)}}]$, descends to the quantum dynamical system $\alpha_V^{(k)}$, governed by the quantum clock $\complexs[\integersMod{S^{(k)}}]$. 

For each value of $t \in \integersMod{S^{(k)}}$, the operator norm distance between $V_t^{(k)}$ and $W_t^{(k)}$ is bounded above by $\epsilon^{(k)} := d\frac{1}{Q^{(k)}}$, and tends to zero as $k$ diverges. Hence we can think of $\alpha_W^{(k)}$ as an increasingly precise approximation of $\alpha_V^{(k)}$. Now consider a scenario in which the quantum clock $\complexs[\integersMod{T^{(k)}}]$ is related to the quantum clock  $\complexs[\integersMod{S^{(k)}}]$ by increased periodicity, via the group quotient $\integersMod{ Q^{(k+1)}R^{(k)}} \rightarrow \integersMod{Q^{(k)}R^{(k)}}$. Again from the discussion following Theorems \ref{thm_internalTimeObservables} and \ref{thm_emergingQuantumClocks}, it is easy to check that  quantum dynamical system $\alpha_W^{(k)}$, governed by the quantum clock $\complexs[\integersMod{S^{(k)}}]$ and $\epsilon^{(k)}$-close to $\alpha_V^{(k)}$, descends to quantum dynamical system $\alpha_U^{(k)}$ governed by quantum clock $\complexs[\integersMod{T^{(k)}}]$.

In conclusion, we have seen that---up to an error in operator norm which vanishes exponentially fast as $k$ diverges---the quantum dynamical systems $\alpha^{(k)}$ form a hierarchy of approximations to the original quantum dynamical system $\alpha$, governed by a hierarchy $\complexs[\integersMod{T^{(k)}}]$ of quantum clocks related by progressive increase in discretisation and periodicity, as described by Theorems \ref{thm_internalTimeObservables} and \ref{thm_emergingQuantumClocks}.
\end{remark}







\chapter{Strong Complementarity in Quantum Algorithms}
\label{chapter_algos}


In the previous Chapter, we have explored the relevance of strong complementarity and the coherent treatment of group and representation theory to the foundations of quantum mechanics; more specifically, we have focused our attention on the symmetry-observable duality following from strong complementarity. In this Chapter, we will develop two new facets of this versatile algebraic property, in their applications to quantum algorithms and protocols.

In Section \ref{section_algosHSP}, taken from \cite{Gogioso2016d}, we put the connection between strong complementarity and the quantum Fourier transform to work in the first fully diagrammatic, theory-independent proof of correctness for the quantum subroutine of the algorithm solving the Hidden Subgroup Problem (HSP). The abstract nature of our proof allows us to interpret it directly in categories other than $\fdHilbCategory$, and our results will provide compelling evidence that strong complementarity is the structural feature powering the quantum subroutine for the abelian HSP. We also obtain interesting new results by interpreting our proof in real quantum theory, hyperbolic quantum theory, finite-field quantum theory and non-standard infinite-dimensional quantum theory.

In Section \ref{section_algosMermin}, we investigate the special standing of the points $\classicalStates{\hbox{\input{symbols/DdotSym.tex}}\!\!}$ of a well-pointed coherent group $(\hbox{\input{symbols/DdotSym.tex}}\!\!,\hbox{\input{symbols/ZbwdotSym.tex}}\!\!)$ within the larger set of unbiased states for the group structure $\hbox{\input{symbols/ZbwdotSym.tex}}\!\!$. We put this to work in a broad generalisation of Mermin's non-locality argument for GHZ states, and we provide an exact characterisation of the connection between non-locality and phase groups, bringing the programme started by \cite{Coecke2012c,Coecke2010a} to a close. We relate our findings to the framework of All-vs-Nothing arguments \cite{Abramsky2015} (a different generalisation of Mermin's argument), and we formulate a non-trivial extension of the quantum-classical secret sharing scheme of Hillery, Bu\v{z}ek and Berthiaume \cite{Hillery1999} (together with novel device-independent guarantees).

\section{Hidden Subgroup Problem}
\label{section_algosHSP}


The advent of quantum computing promises to solve a number number of problems which have until now proven intractable for classical computers. Amongst these, one of the most famous is Shor's algorithm \cite{Shor1995,Ekert1996}: implemented on a quantum computer, it allows for an efficient solution of the integer factorisation problem and the discrete logarithm problem, the hardness of which underlies many of the cryptographic algorithms which we currently entrust with our digital security (such as RSA and Diffie-Hellman Key Exchange). Integer factorisation and the discrete logarithm, together with Simon's problem \cite{Simon1997}, Deutsch original algorithm and a number of other number-theoretic questions, turn out to be special cases of the much more general abelian Hidden Subgroup Problem (HSP) \cite{Jozsa2001}, and can all be tackled by quantum computers using the same strategy.

The reformulation of Shor's algorithm as a special case of the abelian HSP \cite{Jozsa2001} makes the core issue of order-finding pop out as a group-theoretic question, and highlights the role played by the quantum Fourier transform in solving it \cite{Jozsa1997}. However, it is only with the compelling diagrammatic work of \cite{Vicary2012a} that the structures and information flow behind the quantum solution to the HSP become apparent: the unitary oracle used in the algorithm is decomposed into its algebraic building blocks, namely certain $\dagger$-Frobenius algebras, providing a clear topological account of why the procedure works. In this Section, taken from \cite{Gogioso2016d}, we present (Theorem \ref{thm_AbelianHSP}) the first\footnote{The diagrammatics in the proof of \cite{Vicary2012a} are nothing but straightforward graphical transcriptions of results obtained via traditional representation theory.} fully diagrammatic\footnote{Technically, our proof involves a classically-indexed family of diagrams, but this is an accepted standard for fully-diagrammatic treatments of quantum protocols \cite{Coecke2016a}.} proof of correctness for the quantum algorithm solving the abelian HSP, providing compelling evidence that strong complementarity is indeed the structural feature powering the quantum algorithm.

Furthermore, we exploit the theory independent formulation of our proof to obtain new results in theories other than finite-dimensional quantum theory. We show that Simon's Problem can be efficiently solved in real quantum theory and in hyperbolic quantum theory: the latter result is perhaps more surprising than the former, as hyperbolic quantum theory is a local quasi-probabilistic theory. Theorem \ref{thm_integersHSP} uses the non-standard infinite-dimensional CQM framework to show that the infinite abelian HSP for $\integers^N$ can be efficiently solved (under suitable assumptions).

\subsection{The Hidden Subgroup Problem}
\label{section_HSP}

The \textbf{Hidden Subgroup Problem} (HSP) can be phrased as follows: 
\begin{enumerate}
\item[(i)] a finite group $G$ is fixed;
\item[(ii)] we are given an oracle implementing a \textbf{subgroup hiding function} $f : G \rightarrow \integersMod{2}^N$, which associates to each element of $G$ a \textbf{label} in the form of an $N$-bit string;
\item[(iii)] we are promised that the function is constant on (left) cosets of some subgroup $H \leq G$, and associates different labels to different cosets; equivalently, we are promised that $f$ factorises as follows for some injective function $s$ and the quotient group homomorphism $q$ (we refer to this as the \textbf{factorisation promise}):
\begin{equation}\label{diagram_quotientFactorisation}
\input{pictures/chapter4/hsp/quotientFactorisation.tikz}
\end{equation}
\item[(iv)] we are asked to find the \textbf{hidden subgroup} $H$.
\end{enumerate}
In the \textbf{abelian} HSP we are also promised that $G$ is abelian, while in the more general \textbf{normal} HSP we are promised that $H$ is a normal subgroup (a fact which always holds in the abelian HSP). In order for a quantum treatment to be possible at all, one imposes additional requirement on the oracle encoding the subgroup hiding function: 
\begin{enumerate}
	\item[(e)] the oracle is given \textit{coherently}, as the following unitary $U_f \in \UnitaryOps{\complexs[G] \otimes \complexs[\integersMod{2}^N]}$:
		\begin{equation}\label{eqn_unitaryOracle}
			U_f := \ket{g} \otimes \ket{t} \mapsto \ket{g} \otimes \ket{f(g) \oplus t}
		\end{equation}
	where by $\oplus$ we denoted the bit-wise XOR operation on $N$-bit strings.
\end{enumerate}

\noindent A number of important problems arise as special instances of the abelian HSP. In the Discrete Logarithm problem, one is given a prime number $p$, a primitive root $g$ mod $p$, and a number $a$ such that $a = \modclass{g^b}{p}$ for some unknown $b$ to be found. This is an instance of the abelian HSP with group $G = \integersMod{p-1} \times \integersMod{p-1}$, hidden subgroup $H = \integersMod{p-1} \cdot (b,1) \mathrel{\unlhd} G$ and subgroup hiding function $f(x,y) = \modclass{g^x a^{-y}}{p} $.

In the Integer Factorisation problem, one is given a composite number $N$ and is asked to provide a non-trivial factorisation for it. Shor's algorithm solves the problem efficiently on quantum computers: its core is the order-finding subroutine, which considers an integer $a$ coprime\footnote{If $a \in \{2,...,N-1\}$ is not coprime with $N$, then we already have a non-trivial factorisation of $N$.} with $N$, and asks for the order of $a$ as a multiplicative unit modulo $N$. The order-finding subroutine is an instance of the abelian HSP with group $G = \integersMod{N}^\times$ (the abelian group of multiplicative units modulo $N$)\footnote{In practice one uses $\integersMod{2^M}$: if $M \gg \log_2 N$, the errors due to the inexact period of $a$ will be small.}, hidden subgroup $H = \langle a \rangle \isom \integersMod{\operatorname{ord}(a)}$ and subgroup hiding function $f(x) = \modclass{a^x}{N}$.

In Simon's problem, one is given a function $f : \integersMod{2}^N \rightarrow \integersMod{2}^N$ with the promise that the stabilizer subgroup for $f$ has order $2$: there is a unique non-zero string $z \in \integersMod{2}^N$, which we are asked to find, such that for any two $N$-bit strings $x,y \in \integersMod{2}^N$ we have that $f(x) = f(y)$ if and only if $x = y$ or $x = y \oplus z$. The importance of Simon's problem in the  complexity of quantum computing lies in a result \cite{Bernstein1997} stating that, relative to oracles with the promise above, Simon's problem separates BQP (the class of bounded-error quantum polynomial time problems) from BPP (the class of bounded-error classical polynomial time problems). Simon's problem is clearly an instance of the abelian HSP, with $G = \integersMod{2}^N$, hidden subgroup $H = \langle z \rangle = \{0,z\}$ and subgroup hiding function $f$.

\subsection{The Quantum Algorithm}
\label{section_quantumAlgo}

For the fully diagrammatic version of the proof, we replace the concrete Hilbert space setting with four assumptions about the $\dagger$-SMC $\CategoryC$ we want to implement the protocol in, and the strongly complementary pairs that it possesses.

\begin{enumerate}
	\item[(a)] There exist strongly complementary pairs encoding the four relevant finite abelian groups: the group $G$, the hidden subgroup $H$, the quotient group $G/H$ and the group of $N$-bit strings $\integersMod{2}^N$. That is, there exists a strongly complementary pair $(\hbox{\input{symbols/ZbwdotSym.tex}}\!\!_K,\hbox{\input{symbols/DdotSym.tex}}\!\!_K)$ on object $\SpaceH_K$ such that $K \isom (\classicalStates{\hbox{\input{symbols/ZbwdotSym.tex}}\!\!_K},\!\hbox{\input{symbols/DmultSym.tex}}\!\!_K,\!\hbox{\input{symbols/DunitSym.tex}}\!\!_K)$, for each $K=G,H,G/H,\integersMod{2}^N$.
	\item[(b)] The $\dagger$-SMC we are working with has an absorbing scalar $0$, i.e. we can define a sensible notion of impossibility in it.
	\item[(c)] $\hbox{\input{symbols/ZbwdotSym.tex}}\!\!_{H}$ and $\hbox{\input{symbols/ZbwdotSym.tex}}\!\!_{G/H}$ have enough classical states.
	\item[(d)] $\hbox{\input{symbols/DdotSym.tex}}\!\!_{G}$ and $\hbox{\input{symbols/ZbwdotSym.tex}}\!\!_{\integersMod{2}^N}$ have enough classical states, and their classical states are orthogonal---this is so that measurement in either observable can be properly interpreted as a process with classical output in the CP* category $\CPStarCategory{\CategoryC}$ modelling the full quantum-classical theory\footnote{Because the procedure results in uniform sampling, we don't really need to ask for any specific semiring structure on the space of measurement outcomes.}.
\end{enumerate}
As a matter of convenience, we also assume the point structures $\hbox{\input{symbols/ZbwdotSym.tex}}\!\!_K$ to be special, though this is not crucial to the proof. We denote by $\xi_K^\dagger\xi_K$ the normalisation factor of the group structures $\hbox{\input{symbols/DdotSym.tex}}\!\!_K$ (which are quasi-special).

\newpage
Much of the proof relies on the fact that the quotient map $q : G \rightarrow G/H$ is a very specific group homomorphism, defined by the following three properties:
\begin{enumerate}
	\item[(a)] $q$ identifies enough elements of $G$ to send all of $H$ to the group unit;
	\item[(b)] $q$ does not send any elements other than those in $H$ to the group unit;
	\item[(c)] $q$ is surjective.
\end{enumerate}
As a consequence, we require that there exists a morphism $q : \SpaceH_G \rightarrow \SpaceH_{G/H}$ satisfying the three graphical properties below, where $r: \SpaceH_{G/H} \rightarrow \SpaceH_{G}$ is a  $\hbox{\input{symbols/ZbwdotSym.tex}}\!\!_{G/H}$-to-$\hbox{\input{symbols/ZbwdotSym.tex}}\!\!_{G}$ classical isometry (a section for $q$, witnessing its surjectivity), and $i_H : \SpaceH_{H} \rightarrow \SpaceH_{G}$ is a $\hbox{\input{symbols/ZbwdotSym.tex}}\!\!_{H}$-to-$\hbox{\input{symbols/ZbwdotSym.tex}}\!\!_{G}$ classical isometry (modelling the group homomorphism injecting $H$ into $G$):
\begin{equation}\label{eqns_quotientMap}
\input{pictures/chapter4/hsp/quotientMap.tikz}
\end{equation}
Again as a matter of notational convenience, we will henceforth choose coset representatives so that the map $r$ sends the $\hbox{\input{symbols/ZbwdotSym.tex}}\!\!_{G/H}$-classical state corresponding to coset $g_b H \in G/H$ to the $\hbox{\input{symbols/ZbwdotSym.tex}}\!\!_{G}$-classical state corresponding to element $g_b \in G$.

We begin by constructing an abstract version of the unitary oracle $U_f$ given in Equation \ref{eqn_unitaryOracle}. We replace the subgroup hiding function $f: G \rightarrow \integersMod{2}^N$ (or, to be precise, its linear extension $f: \complexs[G] \rightarrow \complexs[\integersMod{2}^N]$) with a $\hbox{\input{symbols/ZbwdotSym.tex}}\!\!_G$-to-$\hbox{\input{symbols/ZbwdotSym.tex}}\!\!_{\integersMod{2}^N}$-classical map $f: \SpaceH_G \rightarrow \SpaceH_{\integersMod{2}^N}$, which is required to satisfy an appropriate factorisation promise, where $q$ is the quotient map defined above and $s: \SpaceH_{G/H} \rightarrow \SpaceH_{\integersMod{2}^N}$ is a  $\hbox{\input{symbols/ZbwdotSym.tex}}\!\!_{G/H}$-to-$\hbox{\input{symbols/ZbwdotSym.tex}}\!\!_{\integersMod{2}^N}$-classical isometry (modelling the subgroup labelling function, which was originally injective)
\begin{equation}
\input{pictures/chapter4/hsp/quotientFactorisationCoherent.tikz}
\end{equation}
The unitary oracle $U_f$ can then be decomposed as follows, in terms of a coherent copy operations for $G$, a coherent multiplication operations for $\integersMod{2}^N$, and the coherent subgroup hiding function $f$:
\begin{equation}\label{eqn_unitaryOracle1}
\input{pictures/chapter4/hsp/unitaryOracle.tikz}
\end{equation}
The process $U_f$ defined above is always unitary, and on finite-dimensional quantum systems it coincides with the oracle we explicitly defined in Equation \ref{eqn_unitaryOracle} \cite{Vicary2012a}.

\newpage
The following diagram presents the quantum subroutine in its entirety: the initial state is prepared, the unitary oracle is applied, and two outcomes $b \in \integersMod{2}^N$ and $\goodchi \in G^\wedge$ are obtained from the measurements performed on the two parts of the resulting state:
\begin{equation}\label{eqn_quantumHSPsubroutine}
\input{pictures/chapter4/hsp/quantumHSPsubroutine.tikz}
\end{equation}
The diagram given above is a scalar $c_{b,\goodchi}$, and we interpret its square absolute value as the probability $\mathbb{P}(b,\goodchi) = c_{b,\goodchi}^\dagger c_{b,\goodchi}$ of obtaining the joint measurement outcome $(b,\goodchi)$. We will now provide a fully diagrammatic proof that the probability must be zero if $b \notin \im{s}$ or if $\goodchi \notin  \Annihil{H}$, and it must otherwise be non-zero and independent of $b$ and $\goodchi$. In other words, we wish to show that the procedure produces a uniformly random sampling of the annihilator of $H$.

\begin{theorem}[\textbf{Abelian HSP in $\dagger$-SMCs}]\hfill\\
\label{thm_AbelianHSP}
The quantum subroutine for the (abelian) HSP can be carried out in any $\dagger$-SMC satisfying the assumptions presented above:
\begin{equation}
	\input{pictures/chapter4/hsp/quantumHSPsubroutineValue.tikz}
\end{equation}
In the case of finite-dimensional quantum systems, we have $\xi_K^\dagger \xi_K = |K|$ for any finite group $K$, and hence the scalar appearing on the RHS is $|H|^2/|G|^2$.\footnote{ This is what we expect: there are $|G/H| = |G| / |H|$ distinct $b$ in the image of $s:G/H \rightarrow \integersMod{2}^N$ (because $s$ is injective), and the annihilator of $H$ itself has size $|G|/|H|$, leading to a total of $|G|^2/|H|^2$ possible joint measurement outcomes $(b,\goodchi)$.}
\end{theorem}
\begin{proof}
We divide the proof into several steps: (i) we use the factorisation promise to break $f$ into its constituent components $q$ and $s$, and we eliminate the coset label $b \in \integersMod{2}^N$ by assuming it is into the image of $s$; (ii) we show that non-annihilator outcomes $\goodchi \notin \Annihil{H}$ are impossible; (iii) in the case of annihilator outcomes $\goodchi \in \Annihil{H}$, we explicitly introduce the abstract equivalents of the coset states $\ket{\goodchi} := \sum_{g \in G} \goodchi(g) \ket{g}$ appearing in the original quantum version; (iv) we annihilate the coset states and obtain the final probabilities.

\noindent \textbf{Using the Factorisation Promise.}
As our first manipulation step, we can substitute the promised factorisation of $f$ into $s \circ q$, and use the unit law to remove the $\dagger$-qSCFA $\hbox{\input{symbols/DdotSym.tex}}\!\!_{\integersMod{2}^N}$ from the diagram:
\begin{equation}\label{eqn_quantumHSPsubroutine1}
\input{pictures/chapter4/hsp/quantumHSPsubroutine1.tikz}
\end{equation}
The property of process $s: \SpaceH_{G/H} \rightarrow \SpaceH_{\integersMod{2}^N}$ being an isometry can be readily formulated diagrammatically as follows:
\begin{equation}
\input{pictures/chapter4/hsp/isometry.tikz}
\end{equation}
Because $s$ is a $\hbox{\input{symbols/ZbwdotSym.tex}}\!\!_{G/H}$-to-$\hbox{\input{symbols/ZbwdotSym.tex}}\!\!_{\integersMod{2}^N}$ classical map, and because $\hbox{\input{symbols/ZbwdotSym.tex}}\!\!_{G/H}$ has enough classical states and the classical states of $\hbox{\input{symbols/ZbwdotSym.tex}}\!\!_{\integersMod{2}^N}$ are orthogonal, then we have the following: a $\hbox{\input{symbols/ZbwdotSym.tex}}\!\!_{\integersMod{2}^N}$-classical state $b$ is either in the image of $s$, i.e. $b = s \circ (g_b H)$ for some classical state $g_b H$, or we have that $b^\dagger \circ s = 0$ is the impossible process, i.e. $b$ is never observed as outcome. Our second manipulation step then assumes that the outcome $b \in \integersMod{2}^N$ is in the image of $s$, and uses the isometry property to remove $s$ from the diagram altogether:
\begin{equation}\label{eqn_quantumHSPsubroutine2}
\input{pictures/chapter4/hsp/quantumHSPsubroutine2.tikz}
\end{equation}

\noindent\textbf{Excluding the Non-Annihilator Outcomes.}
As our next step, we want to show that the diagram evaluates to $0$ whenever $\goodchi$ is not in the annihilator of $H$. To do so, we first need a graphical definition of what it means for a character $\goodchi$ to annihilate $H$:
\begin{equation}\label{eqn_annihilator}
\input{pictures/chapter4/hsp/annihilator.tikz}
\end{equation}
We begin by moving $q$ from the lower to the upper branch, by using the fact that it is $\hbox{\input{symbols/ZbwdotSym.tex}}\!\!_{G}$-to-$\hbox{\input{symbols/ZbwdotSym.tex}}\!\!_{G/H}$ classical:
\begin{equation}\label{eqn_quantumHSPsubroutine3a}
\input{pictures/chapter4/hsp/quantumHSPsubroutine3a.tikz}
\end{equation}
We then proceed to show, by case analysis, that either $\goodchi^\dagger \circ q^\dagger = 0$ (and hence that the entire diagram vanishes), or that the character $\goodchi$ is in the annihilator of $H$ (according to the graphical definition of Equation \ref{eqn_annihilator}): 
\begin{equation}\label{eqn_killingNonAnnihilator}
\input{pictures/chapter4/hsp/killingNonAnnihilator.tikz}
\end{equation}
The first equality is a consequence of the following equivalent reformulation of a defining property of the quotient map $q$:
\begin{equation}\label{eqn_shortExactPropertyAlt}
\input{pictures/chapter4/hsp/shortExactPropertyAlt.tikz}
\end{equation}
The equivalence between the two versions can be proven by using the same general technique which is used in Equation \ref{eqns_cosetStatesDecompositionProof2} below.

\noindent\textbf{Introducing the Coset States.}
Having excluded the case where $\goodchi \notin \Annihil{H}$ (a fact which won't really be needed until the next subsection), our third manipulation step goes as follows:
\begin{equation}\label{eqn_quantumHSPsubroutine3b}
\input{pictures/chapter4/hsp/quantumHSPsubroutine3b.tikz}
\end{equation}
We removed both $q^\dagger$ and the state $g_b H$ of $\SpaceH_{G/H}$ from the diagram, and replaced them with an abstract version of the \textbf{coset state} $\sum_{h \in H} \ket{g_b \cdot h}$ of $\SpaceH_{G}$, by using the following result:
\begin{equation}\label{eqn_cosetStatesDecompositionStatement}
\input{pictures/chapter4/hsp/cosetStatesDecompositionStatement.tikz}
\end{equation}
The equality above can be proven diagrammatically using the defining properties of $q$: 
\begin{equation}
\input{pictures/chapter4/hsp/cosetStatesDecompositionProof1.tikz}
\end{equation}
More in detail, the second equation in the chain is proven by using the fact that $q$ is a $(\hbox{\input{symbols/ZbwdotSym.tex}}\!\!_G,\hbox{\input{symbols/DdotSym.tex}}\!\!_G)$-to-$(\hbox{\input{symbols/ZbwdotSym.tex}}\!\!_{G/H},\hbox{\input{symbols/DdotSym.tex}}\!\!_{G/H})$ homomorphism, by replacing the antipode with its definition, and by appealing to the fact that the antipode is self-inverse:
\begin{equation}\label{eqns_cosetStatesDecompositionProof2}
\resizebox{\textwidth}{!}{\input{pictures/chapter4/hsp/cosetStatesDecompositionProof2.tikz}}
\end{equation}

\noindent\textbf{Annihilating the Coset States.}
We are now in the situation where $\goodchi$ is in the annihilator of $H$, and we have rewritten our diagram explicitly in terms of coset states. As our fourth manipulation step, we turn the character around to obtain (the adjoint of) a diagram involving a character evaluated on a coset state:
\begin{equation}\label{eqn_quantumHSPsubroutine4}
\input{pictures/chapter4/hsp/quantumHSPsubroutine4.tikz}
\end{equation}
Because the character is multiplicative (its adjoint is a $\hbox{\input{symbols/DdotSym.tex}}\!\!_G$-classical state), we can copy it through~$\!\hbox{\input{symbols/DcomultSym.tex}}\!\!_G$. Evaluating against $g_b^{-1}$ removes the first copy to give some phase $\goodchi(g)$, satisfying $\goodchi(g)^\dagger \goodchi(g) = 1$, while the definition of the annihilator removes the second copy together with $i_H^\dagger$:
\begin{equation}\label{eqn_quantumHSPsubroutine5}
\input{pictures/chapter4/hsp/quantumHSPsubroutine5.tikz}
\end{equation}
We are left with a bunch of explicit scalars, and we can finally evaluate the square absolute value of Diagram \ref{eqn_quantumHSPsubroutine} to obtain our desired result: 
\begin{equation}
\resizebox{\textwidth}{!}{\input{pictures/chapter4/hsp/quantumHSPsubroutineValueFinal.tikz}}
\end{equation}
The evaluation of the square norm $\!\hbox{\input{symbols/ZbwunitSym.tex}}\!\!_H^\dagger\circ\!\hbox{\input{symbols/ZbwunitSym.tex}}\!\!_H $ to $\xi_H^\dagger \xi_H$ comes from the fact that $\!\hbox{\input{symbols/ZbwunitSym.tex}}\!\!_H$ is $\hbox{\input{symbols/DdotSym.tex}}\!\!_H$-classical, and the latter has $\xi_H^\dagger\xi_H$ as normalisation factor.

\end{proof}

\subsection{Non-abelian HSP}
\label{section_nonAbelian}

Quantum algorithms to solve the HSP have been studied beyond the abelian case. An extension of the efficient quantum solution to the case of normal subgroups of non-abelian groups is given by \cite{Hallgren2000}, while \cite{Moore2008,Hallgren2010} provide a no-go theorem showing that the same techniques cannot be used to formulate an efficient quantum solution to the general non-abelian case. The general non-abelian case is important because two interesting problems of classical computational complexity arise as special cases: the Graph Isomorphism Problem arises as a special case of the HSP on symmetric groups \cite{Hallgren2000}, while the Unique Shortest Vector Problem (uSVP) arises as a special case of the HSP on dihedral groups \cite{Regev2004}. The latter forms the basis of a public key cryptosystem \cite{Regev2004a} which, subject to quantum intractability of the HSP on dihedral groups, is a candidate to replace RSA in post-quantum cryptography (as are many other lattice-based cryptographic algorithms).

Nowhere in our proof above we have explicitly used commutativity of the $\dagger$-qSCFAs (equivalently, the fact that $G$ and $H$ are abelian), and our approach naturally generalises to the case where $G$ is a finite group and $H$ a normal subgroup (a necessary requirement in this approach, which explicitly uses a group structure on $G/H$). For the sake of simplicity, and because no hard result will be proven, we will stick to the case of finite-dimensional Hilbert spaces for the remainder of this Section. 

Going from commutative to general quasi-special $\dagger$-Frobenius algebras ($\dagger$-qSFA) has the following implication: the classical states are still the multiplicative characters, and they are still orthogonal, but they no longer form a basis. Instead, the $\dagger$-qSFA is now associated with a potentially degenerate observable: sampling it will produce, as classical output, the character $\goodchi_\rho$ of an irreducible representation $\rho$ of $G$, with the following probability (where $d_\rho$ is the dimension of representation $\rho$):
\begin{equation}
	\mathbb{P}[b,\goodchi_\rho] = 
	\begin{cases}
		\frac{|H|^2}{|G|^2} d_\rho^2  & \text{ if } \rho(h) = 0 \text{ for all } h \in H \\
		0 & \text{ otherwise }
	\end{cases}
\end{equation}
For our graphical proof to go through, Diagram \ref{eqn_quantumHSPsubroutine} needs to be modified as follows:
\begin{equation}\label{eqn_quantumHSPsubroutineNonAbelian}
\input{pictures/chapter4/hsp/quantumHSPsubroutineNonAbelian.tikz}
\end{equation}
where we used the dagger-compact structure to take the trace of the irreducible representation $\rho: \complexs[G] \rightarrow V_\rho^\ast \otimes V_\rho$ (technically, its linear extension from $G$ to $\complexs[G]$).
The defining properties of multiplicative characters generalise to representations, as shown by \cite{Vicary2012a}:
\begin{equation}\label{eqn_irrep}
\resizebox{\textwidth}{!}{\input{pictures/chapter4/hsp/irrep.tikz}}
\end{equation} 
However, the generalisation from the abelian to the non-abelian case encounters a much bigger hurdle in the classical post-processing: the logarithmic dependency of the number of generators on the size of the group only need to hold in the abelian case, and the number of samples required is in general linear in the size of the group (and hence exponential in the size of its description) \cite{Hallgren2010}. This is a separate problem, interesting in its own right, and is beyond the scope of this work.

\subsection{Toy quantum theories}

The fully diagrammatic, abstract character of our approach means that our results can be directly applied to theories other than quantum theory, as long as they feature the relevant algebraic structure. Specifically, we consider theories of wavefunctions valued in some commutative involutive semirings $S$ (generalising the complex-valued wavefunctions of ordinary quantum theory). This is a large family, and a number of concrete examples were given in Chapter \ref{chapter_CQM}.

In order for the quantum HSP algorithm for some finite group $G$ to be implementable in $\RMatCategory{S}$, we need two conditions to be satisfied:
\begin{enumerate} 
	\item[(i)] we need $G$, $H$, $G/H$ and $\integersMod{2}^N$ to all admit strongly complementary pairs in $\RMatCategory{S}$, with enough points in the case of $H$, $G/H$ and $\integersMod{2}^N$; 
	\item[(ii)]  we need $\hbox{\input{symbols/DdotSym.tex}}\!\!_G$ to have enough classical states: this is a non-trivial condition, depending entirely on the structure of $S$-valued multiplicative characters for the group $G$ (it is still necessary for $G$ to be abelian, but no longer sufficient). 
\end{enumerate}
Even when the quantum algorithm can be implemented, it is worth noting that the $S$-valued multiplicative characters arising in $\RMatCategory{S}$ may be very different from the complex-valued ones arising in the traditional implementation, and it is in general non-trivial to check that the classical post-processing part of the algorithm will go through as expected.

\subsubsection{Real quantum theory}

Real quantum theory is the theory $\RMatCategory{\reals}$ of real-valued wavefunctions, where $\reals$ is equipped with the identity as its involution. Because $\reals$ is a field, every finite group $K$ admits a strongly complementary pair $(\hbox{\input{symbols/ZbwdotSym.tex}}\!\!_K,\hbox{\input{symbols/DdotSym.tex}}\!\!_K)$ in $\RMatCategory{\reals}$ with enough points. The real-valued characters of a group $G$ are a subset of the complex valued ones, an observation which has two consequences: (i) the finite abelian groups $G$ which admit an implementation of the quantum algorithm to solve the HSP are those possessing only real-valued multiplicative characters, i.e. those in the form $G \isom \integersMod{2}^N$; (ii) the classical post-processing goes through as in the traditional implementation, without additional issues. 

Hence real quantum theory admits efficient solutions to the HSP on $\integersMod{2}^N$, and in particular to Simon's problem. This furthermore implies that the class BQP$_\reals$ (by which we mean BQP for Real Quantum Theory) is separated from BPP relative to oracles with the appropriate promise (see \cite{DeBeaudrap2014} for a detailed study of computational complexity in some toy quantum theories modelled by $\RMatCategory{S}$ constructions).

\subsubsection{Hyperbolic quantum theory}

Hyperbolic quantum theory is the theory $\RMatCategory{\splitComplexs}$ of wavefunctions valued in the \textit{split complex numbers} $\splitComplexs := \reals[X]/(X^2-1)$, a two-dimensional real algebra. Split complex numbers take the form $(x+j y)$, where $x,y \in \reals$ and $j^2 = 1$; in particular, they have non-trivial zero-divisors in the form $a(1\pm j)$, because $(1+j)(1-j)=1-j^2 = 0$. They come with the involution $(x + j y)^\ast := x - jy$, making Hyperbolic Quantum Theory a quasi-probabilistic theory (i.e. it has signed probabilities) and a local theory (because its scalars form a field).

The split complex numbers contain $\reals$ as a subfield fixed by the involution, and hence all protocols which can be implemented in real quantum theory can also be implemented in hyperbolic quantum theory. In particular, the HSP on $\integersMod{2}^N$ can be efficiently solved in hyperbolic quantum theory. There are no other finite abelian groups for which the HSP can be solved efficiently; however, hyperbolic quantum theory admits---similarly to ordinary quantum theory and in contrast to real quantum theory---implementations of the HSP for the infinite abelian groups $\integers^N$.

\subsubsection{Finite-field quantum theory}

If $\finiteField{p^n}$ is a finite field (with $p$ odd) and $\varepsilon$ is a primitive element, then we can consider the ring $\finiteField{p^n}[\sqrt{\varepsilon}] := \finiteField{p^n}[X^2-\varepsilon]$, equipped with the involution $(x+y\sqrt{\varepsilon})^\ast := (x-y\sqrt{\varepsilon})$. Because $\varepsilon$ is a primitive element, $\finiteField{p^n}(\sqrt{\varepsilon}) \isom \finiteField{p^{2n}}$ is a field. We are thus working with the quadratic extension of fields $\finiteField{p^n}(\sqrt{\varepsilon}) / \finiteField{p^n}$, equipped with the usual involution and (squared) norm from Galois theory:
\begin{equation}
\big|x+y\sqrt{\varepsilon}\big|^2 = (x-y\sqrt{\varepsilon})(x+y\sqrt{\varepsilon}) = x^2 - \varepsilon y^2
\end{equation}
The finite abelian groups $K$ admitting a strongly complementary pair $(\hbox{\input{symbols/ZbwdotSym.tex}}\!\!_K,\hbox{\input{symbols/DdotSym.tex}}\!\!_K)$ with enough points in $\RMatCategory{\finiteField{p^n}(\sqrt{\varepsilon})}$ are exactly those with order not divisible by $p$. Furthermore, the group of phases in finite-field quantum theory is isomorphic to the finite abelian group $\integersMod{p^n+1}$: as a consequence, the finite abelian groups $G$ such that $\hbox{\input{symbols/DdotSym.tex}}\!\!_G$ has enough classical states are exactly those in the form $G \isom \prod_{k=1}^{K} \integersMod{p_k^{e_k}}$, with $p_k^{e_k} | p^n+1$ for all $k=1,...,K$ (which, in particular, have order not divisible by $p$). Finally, the $\finiteField{p^n}(\sqrt{\varepsilon})$-valued multiplicative characters can be easily interpreted as complex-valued multiplicative characters (using the subgroup $\integersMod{p^n+1}$ of the circle group given by the $(p^n+1)$-th roots of unity), and the classical post-processing phase of the algorithm goes through without additional issues.

\subsubsection{p-adic quantum theory}

The phases in $p$-adic quantum theory form a multiplicative abelian group $C_\varepsilon$ isomorphic to the additive group $\integersMod{p+1} \times p Z_p$, where $(\integersMod{p+1},+,0)$ are the integers modulo $p+1$, while $(p Z_p, +, 0)$ is the additive subgroup of $Z_p$ formed by those $p$-adic integers which are divisible by $p$. As a consequence, a finite abelian group $G$ gives rise to a group structure $\hbox{\input{symbols/DdotSym.tex}}\!\!_G$ with enough classical states if and only if $G \isom \prod_{k=1}^{K} \integersMod{p_k^{e_k}}$, with $p_k^{e_k} | p+1$ for all $k=1,...,K$---just as in the finite-field quantum theory $\RMatCategory{\finiteField{p}(\sqrt{\varepsilon})}$. Similarly, the $Q_p(\sqrt{\varepsilon})$-valued multiplicative characters can be easily interpreted as complex-valued, and the classical post-processing of the algorithm for the HSP goes through as in the traditional case.

\subsubsection{Relational quantum theory}

Relational quantum theory is the theory of boolean-valued wavefunctions, or more generally of wavefunctions valued in a locale $\Omega$ (with the identity as involution). It is modelled by the dagger compact category $\RMatCategory{\Omega}$, and in the boolean case it coincides with the category $\fRelCategory$ of finite sets and relations. For any finite group $K$, the scalar $|K|$ in Relational Quantum Theory is simply the scalar $1$ (because $|K|=1+1+1+....+1$, and $1$ is additively idempotent in any locale), and hence all finite groups admit a strongly complementary pair with enough points. However, the phase group in relational quantum theory is the trivial group $\{1\} \isom \integersMod{1}$, and the only group $K$ such that $\hbox{\input{symbols/DdotSym.tex}}\!\!_K$ has enough classical states is the trivial group $\integersMod{1}$ itself. Hence there are no non-trivial implementations of the quantum algorithm for the HSP in relational quantum theory. However, this does not necessarily mean that no efficient solution to the HSP can be obtained in relational quantum theory, and the study of relational formulations of quantum algorithms is a busy open field (especially in the context of Spekkens' Toy Model \cite{Spekkens2007,Coecke2012a,Backens2015,Zeng2015,Disilvestro2016,Catani2017}).

\subsubsection{Infinite-dimensional HSP}
\label{section_infiniteHSP}

The category $\HilbCategory$ of infinite-dimensional Hilbert spaces and bounded linear maps does not admit Frobenius algebras, and as a consequence it cannot be used to extend our abstract setup to infinite abelian groups. However, we have seen in Chapter \ref{chapter_CQM} that tools from non-standard analysis can be used to construct a well-defined category $\starHilbCategory$ of infinite-dimensional separable Hilbert spaces, including both bounded and unbounded linear maps, as well as a number of commonplace features of quantum mechanics (such as Dirac deltas and plane-waves). This opens the way to infinite generalisations of the abelian HSP.

\begin{theorem}[\textbf{HSP for $\integers^N$ with a particle in a box}]\label{thm_integersHSP}\hfill\\
Consider a particle in an $N$-dimensional box with periodic boundary conditions. Assume that preparations and measurements in the position observable can be performed with arbitrarily high precision. Then it is possible to efficiently solve the HSP for the infinite abelian group $\integers^N$.
\end{theorem} 
\begin{proof}
In Chapter \ref{chapter_CQM} we have seen that that the category $\starHilbCategory$ possesses suitable strongly complementary pairs corresponding to the discrete groups $\integers^N$ of translations of lattices and the compact groups $\torusGroup{N}$ of translations of tori; all the observables concerned are quasi-special or special, commutative, and have enough classical states. These observables have direct physical relevance, as they correspond to the momentum/position observable pairs for particles in $N$-dimensional boxes with periodic boundary conditions. As a consequence, our scheme straightforwardly extends to a quantum subroutine for the HSP on the infinite abelian groups $G = \integers^N$. The classical subroutine requires no adjustment for the case of $\integers^N$, because all the possible quotients $G/H$ are in the form $\integers^J \times \prod_{k=1}^K \integersMod{n_k}$: hence all the possible annihilators $\Annihil{H}$ are in the form $\torusGroup{J} \times \prod_{k=1}^K \integersMod{n_k}$, and can be efficiently sampled (infinite precision notwithstanding). 

For the sake of physical implementation, this setup corresponds to fixing the computational basis $\hbox{\input{symbols/ZbwdotSym.tex}}\!\!_G$ to be the basis of momentum eigenstates (valued in $\integers^N$) of a particle in an $N$-dimensional box with periodic boundary conditions, and performing the $\hbox{\input{symbols/DdotSym.tex}}\!\!_G$ measurement corresponding to its position observable (valued in $T^N$). The oracle is a standard unitary, but the particle needs to be prepared in the exact position eigenstate $\frac{\sqrt{L}^N}{\sqrt{(2\omega+1)^N}}\!\hbox{\input{symbols/ZbwunitSym.tex}}\!\!_G$, and then measured in the position observable with infinite precision (i.e. with continuous-valued outcomes in $\torusGroup{N}$). Hence we can efficiently solve HSPs on the infinite abelian groups $\integers^N$ as long as we can perform exact preparations and measurements in the position observable.

In fact, being able to perform preparations and measurements in the position observable with arbitrary finite precision turns out to be sufficient. Indeed, consider the family of processes, indexed by $t \in \{1,...,\omega\}$, which involves preparation in the following approximate position eigenstate:
\begin{equation}
	\frac{\sqrt{L}^N}{\sqrt{(2\omega+1)^N}}\ket{\delta_{0}^{(t)}} := \frac{1}{\sqrt{(2\omega+1)^N}}\sum_{k_1=-t}^{t}...\sum_{k_N=-t}^{t} \frac{1}{\sqrt{L}^N}\ket{\chi_{\underline{k}}}
\end{equation}  
This results in a sequence of conditional probability distributions $\mathbb{P}_t(x|b)$ on position measurement outcomes $x \in \frac{1}{(2\omega+1)}\starIntegersModPow{2\omega+1}{N}$, indexed by the parameter $t\in\{1,...,\omega\}$. The choice of infinite natural $\omega$ is arbitrary, and hence the main proof of this Section shows that $\mathbb{P}_t(x|b) \simeq 0$ for all $t$ infinite and all  $x \notin \Annihil{H}$. As a consequence, we must have have the following standard result, where now we restrict our attention to $t \in \naturals^+$:
\begin{equation}
	\lim_{t \rightarrow \infty}\mathbb{P}_t(\stdpart{x}|b) = 0 \text{ for all } x \notin \Annihil{H} 
\end{equation}
This shows that arbitrary precision position preparations plus exact position measurements are enough. However, all quotients of $\integers^N$ take the form $\torusGroup{J} \times \prod_{k=1}^K \integersMod{n_j} \leq \torusGroup{N}$, and as a consequence approximate position measurements with sufficiently high precision always suffice. b
\end{proof}

\newpage
\section{Mermin-type non-locality scenarios}
\label{section_algosMermin}


Non-locality is a defining feature of quantum mechanics, and its connection to the structure of phase groups is a key foundational question. A particularly crisp example of this connection is given by Mermin's argument for qubit GHZ states \cite{Mermin1990}, which finds practical application in the HBB quantum secret sharing protocol. 

In Mermin's argument, $N$ qubits are prepared in a Pauli $Z$ GHZ state, then a controlled phase gate is applied to each, followed by measurement in the Pauli $X$ observable. Even though the $N$ outcomes (valued in $\integersMod{2}$) are probabilistic, their parity turns out to satisfy certain deterministic equations. Mermin shows that the existence of a local hidden variable model would imply a joint solution for the equations: however, the latter form an inconsistent system, and Mermin concludes that the scenario is non-local. Mermin's argument has sparked a number of lines of enquiry, and this work is concerned with two in particular: one leading to All-vs-Nothing arguments, and the other investigating the role played by strong complementarity. All-vs-Nothing arguments \cite{Abramsky2015} arise in the context of the sheaf-theoretic framework for non-locality and contextuality \cite{Abramsky2011}, and generalise the idea of a system of equations which is locally consistent but globally inconsistent. The second line of research is brought forward within the framework of categorical quantum mechanics, and it focuses on the algebraic characterisation of the structures involved.

A detailed analysis of Mermin's argument shows that the special relationship between the Pauli $X$ and Pauli $Z$ observables powering the argument is nothing but strong complementarity. A pair of complementary observables corresponds to mutually unbiased orthonormal bases: for example, both Pauli $X$ and Pauli $Y$ are complementary to Pauli $Z$. Strong complementarity amounts to a strictly stronger requirement: if one observable is taken as the computational basis, the other must correspond to the Fourier basis for some finite abelian group. Pauli $X$ fits the bill, for the abelian group $\integersMod{2}$, but Pauli $Y$ doesn't (in fact, Pauli $X$ is the only one for qubits).

In \cite{Coecke2012c}, Mermin's argument is completely reformulated in terms of strongly complementary observables ($\dagger$-Frobenius algebras) and abstract phase gates. It can therefore be tested on theories different from quantum mechanics, to better understand the connection between non-locality and the structure of phase groups. A particularly insightful comparison is given by qubit stabiliser quantum mechanics \cite{Coecke2011,Backens2014} vs Spekkens' toy model \cite{Spekkens2007,Coecke2012a}: both theories sport very similar operational and algebraic features, but the difference in phase groups ($\integersMod{4}$ for the former vs $\integersMod{2} \times \integersMod{2}$ for the latter) results in the former being non-local and the latter being local (both models have $\integersMod{2}$ as group of measurement outcomes, like Mermin's original argument). The picture arising from comparing qubit stabiliser quantum mechanics and Spekkens' toy model is iconic, and provides a first real glimpse into the connection between phase groups and non-locality~\cite{Coecke2010a}. 

While presenting an extremely compelling case for stabiliser qubits and Spekkens' toy qubits, the work of \cite{Coecke2012c,Coecke2010a} does not treat the general case (i.e. beyond $\integersMod{2}$ as group of measurement outcomes), nor does it provide a complete algebraic characterisation of the conditions guaranteeing non-locality. In this Section, taken from \cite{Gogioso2016e}, we fully generalise Mermin's arguments (Definition \ref{def_generalisedMerminArgument}) from $\integersMod{2}$ to arbitrary finite abelian groups, in arbitrary theories and for arbitrary phase groups (we will refer to these as \textbf{generalised Mermin-type arguments}). We also provide exact algebraic conditions for non-locality to be exhibited by our generalised Mermin-type arguments (Theorem \ref{thm_contextuality}), thus completing the line of research on the connection of phases and non-locality initiated by \cite{Coecke2012b,Coecke2010a}.

We proceed to make contact with the All-vs-Nothing line of enquiry \cite{Abramsky2015}, showing that the non-local generalised Mermin-type arguments yield a new hierarchy of quantum-realisable All-vs-Nothing empirical models, and hence that they are strongly contextual (Theorem \ref{thm_quantumRealisability} proves that the arguments are quantum realisable, while Theorem \ref{thm_AvNMermin} proves the non-local arguments are All-vs-Nothing). As a consequence, we show that the hierarchy of quantum-realisable All-vs-Nothing models over finite fields does not collapse (Theorem \ref{thm_nonCollapsingHierarchy}).

Mermin's argument for the qubit GHZ states also finds practical application in the quantum-classical secret sharing scheme of Hillery, Bu\v{z}ek and Berthiaume \cite{Hillery1999}, and we provide a non-trivial extension of the scheme to our generalised Mermin-type arguments. We also use the strong contextuality deriving from our All-vs-Nothing characterisation to provide some device-independent security guarantees (Theorem \ref{thm_HBBdisec}), which apply to the original HBB scheme as a special case.

The results in this Section are also related to the recent work of \cite{Zukowski2013}, and previous work by \cite{Lee2006,Cerf2002,Kaszlikowski2002,Zukowski1999}. The construction adopted in \cite{Zukowski2013} is similar to the one we will derive in this Section for the special case of cyclic groups $K := \integersMod{D}$, and the relationship between the dimension $D$ of the quantum systems and the allowed numbers $N$ of parties are the same in both works (i.e. $\gcd(N,D) = 1$). The References presented above feature explicit constructions, but are not concerned with investigating the general connection between Mermin-type non-locality and the structure of phase groups (which is the stated aim of the line of research of \cite{Coecke2012b,Coecke2010a}, brought forward by this Section).

\subsection{Mermin's Original Argument}
\label{section_merminOriginalArg}

\subsubsection{The parity argument} 

\noindent In the original \cite{Mermin1990}, Mermin considers a 3-qubit GHZ state in the computational basis, the basis of eigenstates for the single-qubit Pauli $Z$ observable, together with the following four joint measurements\footnote{Where $X_j$ and $Y_j$ are the single-qubit Pauli $X$ and $Y$ observables on qubit $j$, for $j=1,2,3$.}:
\begin{enumerate}[\indent(a)]
	\item the GHZ state is measured in the observable $X_1 \tensor X_2 \tensor X_3$;
	\item the GHZ state is measured in the observable $\;Y_1 \tensor \;Y_2 \tensor X_3$;
	\item the GHZ state is measured in the observable $\;Y_1 \tensor X_2 \tensor \;Y_3$;
	\item the GHZ state is measured in the observable $X_1 \tensor \;Y_2 \tensor \;Y_3$.
\end{enumerate}
We will denote the eigenstates of the Pauli $Z$ observable by $\ket{z_0},\ket{z_1}$, the eigenstates of the Pauli $X$ observable by $\ket{\pm} := \frac{1}{\sqrt{2}}(\ket{z_0} \pm \ket{z_1})$ and the eigenstates of the Pauli $Y$ observable by $\ket{\pm i} := \frac{1}{\sqrt{2}}(\ket{z_0} \pm i \ket{z_1})$. Mermin's argument is a parity argument, where measurement outcomes are valued in the abelian group $\integersMod{2} = \{0,1\}$ according to the following bijections:
\begin{enumerate}[\indent(i)]
	\item for the $X$ observable, $\ket{+} \mapsto 0$ and $\ket{-} \mapsto 1$
	\item for the $Y$ observable, $\ket{+i} \mapsto 0$ and $\ket{-i} \mapsto 1$
\end{enumerate}

\noindent The argument then proceeds as follows. While the joint measurement outcomes are probabilistic, the $\integersMod{2}$ sum of the three outcomes turns out to be deterministic, yielding the following system of equations ($\oplus$ here denotes the sum in $\integersMod{2}$):
\begin{equation}
\label{eqn_MerminSystemZ2Equations}
\begin{cases}
X_1 \oplus X_2 \oplus X_3 &= 0 \\
\;Y_1 \oplus \;Y_2 \oplus X_3 &= 1 \\
\;Y_1 \oplus X_2 \oplus \;Y_3 &= 1 \\
X_1 \oplus \;Y_2 \oplus \;Y_3 &= 1 
\end{cases}
\end{equation}
If there was a non-contextual assignment of outcomes for all measurements (i.e. $X_1,X_2,X_3$, $Y_1,Y_2$ and $Y_3$), i.e. if there existed a non-contextual hidden variable model, then System \ref{eqn_MerminSystemZ2Equations} would have a solution in $\integersMod{2}$, and in particular it would have to be consistent. However, the sum of the left hand sides yields $0$ in $\integersMod{2}$:
\begin{equation}
	\label{eqn_MerminSystemZ2EquationsLHSSum}
	2 X_1 \oplus 2 X_2 \oplus ... \oplus 2 Y_3 = 0 X_1 \oplus ... \oplus 0 Y_3 = 0
\end{equation}
while the sum of the right hand sides yields $0 \oplus 1 \oplus 1 \oplus 1 = 3 = 1$ in $\integersMod{2}$. This shows the system to be inconsistent. Equivalently, one could observe that the sum of the LHSs from Equation \ref{eqn_MerminSystemZ2EquationsLHSSum} can be written as $2(Y_1 \oplus Y_2 \oplus Y_3)$, and that inconsistency of the system is witnessed by the fact that the equation $2 y = 1$ has no solution in~$\integersMod{2}$. 

The first point of view, where contextuality is witnessed by an inconsistent system where each equation individually admits a solution, is behind the generalisation of Mermin's argument to All-vs-Nothing arguments, presented in \cite{Abramsky2015}. The second point of view, where contextuality is witnessed by the single unsatisfiable equation $2 y = 1$, will inspire the generalisation presented in this work.

\subsubsection{The role of phases}

To understand the role played by the equation $2 y = 1$ in the original Mermin argument, we need to take a step back. First of all, we observe that the Pauli $Y$ measurement can be equivalently obtained as a Pauli $X$ measurement preceded by an appropriate unitary. A single-qubit \textbf{phase gate}, in the computational basis (the Pauli $Z$ observable), is a unitary transformation in the following form:
\begin{equation}
	\label{eqn_Z2PhaseGates}
	\phasegate{\alpha} := \left(
	\begin{array}{cc}
		1 & 0 \\
		0 & e^{i \alpha}
	\end{array}
	\right)
\end{equation}
where we eliminated global phases by setting the first diagonal element to 1. Measuring in the single-qubit $Y$ observable is equivalent to first applying the single-qubit phase gate $\phasegate{\frac{\pi}{2}}$, and then measuring in the Pauli $X$ observable.

Because they pairwise commute, phase gates come with a natural abelian group structure given by composition, resulting in an isomorphism $\alpha \mapsto P_\alpha$ between them and the abelian group\footnote{The abelian group $\reals/(2\pi\integers)$ is isomorphic to the circle group $S^1$. We prefer the former because of its additive notation, as opposed to the traditionally multiplicative notation of the latter (which is a subgroup of the non-zero multiplicative complex numbers $\complexs^\times$).} $\reals/(2\pi\integers)$. Of all the phase gates, $\phasegate{0}$ (the identity element of the group) and $\phasegate{\pi}$ stand out because of their well-defined action on the (unnormalised) eigenstates of the Pauli $X$ observable:
\begin{align}
	\label{eqn_PzeroPiAction}
	\phasegate{0} & = \ket{\pm} \mapsto \ket{\pm} \nonumber\\
	\phasegate{\pi} & = \ket{\pm} \mapsto \ket{\mp} 
\end{align}
If we see $\ket{\pm}$ as the subgroup\footnote{Corresponding to $\{\pm 1\} < S^1$ in the circle group.} $\{0,\pi\} < \reals / (2\pi\integers)$, then Equation \ref{eqn_PzeroPiAction} looks a lot like the regular action of $\{0,\pi\}$ on itself. This is not a coincidence. Each phase gate $\phasegate{\alpha}$ can be (faithfully) associated the unique \textbf{phase state} $\ket{\alpha} := \ket{z_0} + e^{i \alpha} \ket{z_1}$ obtained from its diagonal, and these phase states can be abstractly characterised in terms of the Pauli $Z$ observable, with no reference to the phase gates they came from (See Subsection \ref{section_phaseGroup}). The phase states inherit the abelian group structure of the phase gates, and their regular action coincides with the action of the group of phase gates on them. In particular, the phase gates $\phasegate{0}$ and $\phasegate{\pi}$ have orthogonal eigenstates of the Pauli $X$ observable as their associated phase states $\ket{0}$ and $\ket{\pi}$, which coincide with $\sqrt{2}\ket{+}$ and $\sqrt{2}\ket{-}$ respectively: this endows the outcomes of Pauli $X$ measurements with the natural $\integersMod{2}$ abelian group structure arising\footnote{Natural because there is a unique isomorphism $\integersMod{2} \isom \{0,\pi\}$.} from the inclusion $\{0,\pi\} < \reals/(2\pi\integers)$. We will henceforth refer to the group of phase states as the \textbf{group of $Z$-phase states}, and to the subgroup $\{0,\pi\}$ as the \textbf{subgroup of $X$-classical states}; the latter will also be used to label the corresponding measurement outcomes.

In order to pave the way to our generalisation, we now proceed to show how Mermin's original argument can be re-constructed from the following statement: 
\begin{displayquote}
	\textit{the equation $2 y = \pi$ has no solution in the subgroup $\{0,\pi\}$ of $X$-classical states, but a solution\footnote{Corresponding to $y = e^{i\frac{\pi}{2}} = +i$ in the circle group $S^1$.} $y = \frac{\pi}{2}$ can be found in the larger group $\reals/(2\pi\integers)$ of $Z$-phase states.}
\end{displayquote}
We begin by observing that tripartite qubit GHZ state used in Mermin's argument has a special property when it comes to the application of phase gates followed by measurements in the Pauli $X$ observable.
\begin{lemma}[\cite{Coecke2012c}]\label{lemma_qubitGHZphaseSum} If $\alpha_j \in \reals/(2\pi\integers)$, denote by $X_j^{\alpha_j}$ the measurement outcome on qubit $j$ obtained by first applying phase gate $\phasegate{\alpha_j}$, and then measuring in the Pauli $X$ observable. If $\alpha_1 \oplus \alpha_2 \oplus \alpha_3 = \modclass{0\text{ or }\pi}{2\pi}$, then $X_1^{\alpha_1} \oplus X_2^{\alpha_2} \oplus X_3^{\alpha_3} = \modclass{0 \text{ or }\pi}{2\pi}$ respectively.
\end{lemma}
\noindent Now consider System \ref{eqn_MerminSystemZ2Equations} again, with values on the the RHS now obtained by applying Lemma~\ref{lemma_qubitGHZphaseSum} to $X_j := X_j^{0}$ and $Y_j := X_j^{\frac{\pi}{2}}$ (and valued in $\{0,\pi\}$ instead of $\integersMod{2}$):
\begin{equation}
	\label{eqn_MerminSystemZ2EquationsPM}
	\begin{cases}
	X^{0}_1\; \oplus X^{0}_2\; \oplus X^{0}_3\; &= 0 \text{, the control}\\
	X^{\frac{\pi}{2}}_1 \oplus X^{\frac{\pi}{2}}_2 \oplus X^{0}_3\; &= \pi \text{, the first variation}\\
	X^{\frac{\pi}{2}}_1 \oplus X^{0}_2\; \oplus X^{\frac{\pi}{2}}_3 &= \pi \text{, the second variation}\\
	X^{0}_1\; \oplus X^{\frac{\pi}{2}}_2 \oplus X^{\frac{\pi}{2}}_3 &= \pi \text{, the third variation}
	\end{cases}
\end{equation}
There are two complementary parts to the Mermin non-locality argument: (i) System \ref{eqn_MerminSystemZ2EquationsPM} above must be inconsistent, to rule out the existence of a non-contextual hidden variable model, and (ii) joint measurements yielding the individual equations must be possible (in quantum theory). For the first part, inconsistency of the system is witnessed by the fact that the equation $2 y = \pi$ has no solution in the subgroup of $X$-classical states. For the second part, notice that only measurements in the $Y$ observable contribute to the sum for each equation, as measurements in the $X$ observable are associated with the group unit $0$ of the group of $Z$-phase states. As a consequence, the existence of measurements implementing each individual equation reduces to the existence of a $Z$-phase state $\ket{y}$ satisfying equation $2 y = \pi$: the $Y$ observable is chosen exactly because $y = \pi/2$ gives one such $Z$-phase state.

The following steps summarise the skeleton of the argument, and open the way to our generalisation: 
\begin{enumerate}[\indent 1.]
	\item consider a non-degenerate observable, call it $Z$, on an arbitrary quantum system;
	\item consider a non-degenerate observable, call it $X$, such that the $X$-classical states are a subgroup (call it $K$) of the abelian group of $Z$-phase states (call it $P$);
	\item consider an equation in the following form, generalising $2 y = \pi$:
	\begin{equation}\label{eqn_generalisedInconsistencyEquation}
		n_1 y_1 \oplus ... \oplus n_M y_M = a
	\end{equation}
	(here $a \in K$, $n_1,...,n_M$ are integers\footnote{This is a general equation in abelian groups, seen equivalently as $\integers$-modules.}, and $\oplus$ is the group addition in $P$);
	\item construct an appropriate system of equations, generalising System \ref{eqn_MerminSystemZ2EquationsPM}, with inconsistency witnessed by non-existence of solutions for  Equation \ref{eqn_generalisedInconsistencyEquation} in $K$, and consistency of the individual equations witnessed by the existence of solutions in $P$;
	\item a measurement scenario is implementable if and only if a solution exists in $P$;
	\item a measurement scenario is contextual if and only if no solutions exist in $K$.
\end{enumerate} 

\noindent To give a first example of how such an appropriate system of equations might be constructed, we consider the simple generalisation of the argument from a $3$-partite to an $N$-partite GHZ state, for appropriate values of $N \geq 2$. Our requirements are: 
\begin{enumerate}[\indent(i)]
	\item we want the phases in the control to sum to $0$, and hence we will take them all to be $0$ (i.e. $X$ measurements), just as in the original argument; 
	\item we also want the phases in each variation to sum to $\pi$, and hence we will take two measurements in each variation to be with phase $\pi/2$ (i.e. measurements in the $Y$ observable), and all the other ones to be with phase $0$; 
	\item we want an odd number $V$ of variations, in order to ensure that the RHSs will sum to $0 \oplus V \pi = \pi$;
	\item we want the LHSs to sum to an even multiple of $X_1^{\frac{\pi}{2}} \oplus ... \oplus X_N^{\frac{\pi}{2}}$;
\end{enumerate}
An appropriate choice is given by the following system of equations, where $V := N$ and all  variations are cyclic permutations of the first one:
\begin{equation}
	\begin{cases}
	X_1^{0}\; \oplus X_2^{0}\; \oplus X_3^{0}\; \oplus\;\;\;\, ...\;\;\;\hspace{0.2mm} \oplus X_{N-1}^{0}\; \oplus X_N^{0}\; &= 0 \text{, the control}\\
	X_1^{\frac{\pi}{2}} \oplus X_2^{\frac{\pi}{2}} \oplus X_3^{0}\; \oplus\;\;\;\, ...\;\;\;\hspace{0.4mm} \oplus X_{N-1}^{0}\; \oplus X_N^{0}\; &= \pi \text{, the $1^{st}$ variation}\\
	X_1^{\frac{\pi}{2}} \oplus X_2^{0}\; \oplus\;\, ...\;\hspace{0.2mm} \oplus X_{N-2}^{0}\; \oplus X_{N-1}^{0}\; \oplus X_N^{\frac{\pi}{2}} &= \pi \text{, the $2^{nd}$ variation}\\
	\hspace{29.6mm}\vdots \\
	X_1^{0}\; \oplus X_2^{\frac{\pi}{2}} \oplus X_3^{\frac{\pi}{2}} \oplus \;\;\;X_4^{0}\;\;\hspace{0.5mm} \oplus \;\;\;\, ...\;\;\;\hspace{0.3mm} \oplus X_{N}^{0}\; &= \pi \text{, the $N^{th}$ variation}\\
	\end{cases}
\end{equation}
As long as $N = \modclass{1}{k}$, where $k=2$ is the exponent\footnote{The smallest positive integer such that $k x = 0$ for all $x \in K$.} of $K$, the RHSs will sum to $\pi$ in $K$. Having \vspace{-1mm} chosen our variations by cyclic permutation also makes for the desired sum of the LHSs, since each $X_j^{\pi/2}$ will be counted exactly twice:
\vspace{3mm}
\begin{center}
\begin{tabular}{ccccc}
$\big(X_1^{0} \oplus ... \oplus X_N^{0}\big)$ & $\oplus$ & $2 \cdot \big(X_1^{\frac{\pi}{2}} \oplus ... \oplus X_N^{\frac{\pi}{2}})$ & $\oplus$ & $(N-2)\cdot \big(X_1^{0} \oplus ... \oplus X_N^{0})$ \\
control & & $X_j^{\frac{\pi}{2}}$s from the variations & &  $X_j^{0}$s from the variations\\ 
\end{tabular}
\end{center}
\vspace{3mm}
Writing $x$ for $X_1^{0} \oplus ... \oplus X_N^{0}$ and $y$ for $X_1^{\frac{\pi}{2}} \oplus ... \oplus X_N^{\frac{\pi}{2}}$, the sum above can be rearranged to take the form $(N-1) x \oplus 2 y$, which is equal to $2 y$ in $K$ (since $(N-1) = \modclass{0}{k}$)\footnote{In this specific case, it is also true that $2 = \modclass{0}{k}$, but this is not key to the argument.}. Hence summing all the LHSs and RHSs leaves us with the equation $2 y = \pi$, which we know to be unsatisfiable in $K$.

\subsection{The phase group}
\label{section_phaseGroup}

Mermin's parity argument is fundamentally group-theoretic, and it depends almost entirely on the special relationship between the Pauli $Z$ and Pauli $X$ observables. Fixing the eigenstates of the Pauli $Z$ observable as the computational basis, the requirement that the $X$-classical states are $Z$-phase states is satisfied by the Pauli $X$ observable, but also by the Pauli $Y$: in fact, the $Z$-phase states are exactly the \textbf{unbiased states} for the Pauli $Z$ observables, the states lying on the equator of the Bloch sphere, and hence any observable \textbf{complementary}, or \textbf{mutually unbiased}, to Pauli $Z$ would do the trick; because their eigenstates lie on the equator of the Bloch sphere, we refer to these as \textbf{equatorial observables}. 

Complementarity however, is not sufficient for Mermin's argument: Lemma \ref{lemma_qubitGHZphaseSum} only holds if we measure the GHZ state in the Pauli $X$ observable, not in any other equatorial observable. The algebraic relationship between the Pauli $X$, $Y$ and $Z$ observables is vividly captured by the ZX calculus \cite{Coecke2011}: there, the special property relating the Pauli $Z$ and $X$ observables is axiomatised under the name of \textbf{strong complementarity}, to distinguish it from the complementarity of Pauli $Z$ and any other equatorial observable (such as Pauli $Y$). Strong complementarity is behind the proof of Lemma \ref{lemma_qubitGHZphaseSum}, which lies at core of the fully diagrammatic treatment of Mermin's original argument appearing in \cite{Coecke2012c}. 

\begin{remark}\label{rmrk_XYstrongComplementarityConundrum}
The Pauli $X$ and $Y$ observables on physical qubits are physically indistinguishable, as they can be turned into one another by a unitary which fixes the Pauli $Z$ observable (i.e. a Pauli $Z$ phase gate). As a consequence, it seems somewhat disturbing that $X$ and $Y$ could be distinguished by an abstract, basis-independent property such as strong complementarity. 

An extremely detailed description of the (inexact) correspondence between quantum observables, orthonormal bases of vectors and classical structures is given in Section 5 of \cite{Coecke2011}, where the concept of strong complementarity was originally introduced. In particular, the discussion explains why strong complementarity picks Pauli $X$ and not Pauli $Y$. The point is that the classical structure $\hbox{\input{symbols/ZbwdotSym.tex}}\!\!$ for Pauli $Z$ picks out two vector representatives $\ket{z_0}, \ket{z_1}$ for the Pauli $Z$ eigenstates (which are complex projective lines), and in doing so it imposes a specific group structure $(\!\hbox{\input{symbols/ZbwmultSym.tex}}\!\!,\!\hbox{\input{symbols/ZbwunitSym.tex}}\!\!)$ on the set of equatorial states. 

Physically, Pauli $X$ can be turned into Pauli $Y$ by a rotation of $\pi/2$ around the positive $Z$ axis, and this transformation leaves the Pauli $Z$ physical observable invariant. However, it turns out to change the Pauli $Z$ classical structure, in the following way: keeping $\ket{z_0}$ fixed (without loss of generality), it sends $\ket{z_1}$ to $i\ket{z_1}$ (same physical state, different vector), so that the unit $\!\hbox{\input{symbols/ZbwunitSym.tex}}\!\!$ changes from the Pauli $X$ $+1$-eigenstate $\ket{+} := \ket{z_0} + \ket{z_1}$ to the Pauli $Y$ $+1$-eigenstate $\ket{+i} := \ket{z_0} + i \ket{z_1}$; the multiplication changes as well, to reflect the fact that the unit of the group of equatorial states has rotated from $\ket{+}$ to $\ket{+i}$.

\end{remark}

The importance of strong complementarity for the Hidden Subgroup Problem lied in its connection to the quantum Fourier transform. The situation with Mermin-type arguments, however, is different: the relevant facet of strong complementarity will be the special relationship between $\hbox{\input{symbols/DdotSym.tex}}\!\!$-classical points and $\hbox{\input{symbols/ZbwdotSym.tex}}\!\!$-phase states, explored in detail in this Subsection. If strong complementarity is the fundamental algebraic property at work in Mermin's argument, phase gates and GHZ states are the operational components key to its implementation. Phase gates arise in the context of quantum-to-classical transitions, where they provide a characterisation, in the spirit of groups and symmetries, of how much information is lost by performing a (demolition) measurement in a non-degenerate observable.
\begin{definition}
Let $\hbox{\input{symbols/ZbwdotSym.tex}}\!\!$ be a $\dagger$-qSFA on an object $\SpaceH$ of a dagger compact category. Then the \textbf{$\hbox{\input{symbols/ZbwdotSym.tex}}\!\!$-phase gates} are the unitaries $U: \SpaceH \rightarrow \SpaceH$ which are annihilated by the measurement in the $\hbox{\input{symbols/ZbwdotSym.tex}}\!\!$ observable:
\begin{equation}\label{eqn_ZphaseGateDef}
	\input{pictures/chapter4/mermin/ZphaseGateDef.tikz}
\end{equation} 
Equation \ref{eqn_ZphaseGateDef} can be unfolded into the following equivalent definition, which extends to an arbitrary $\dagger$-SMC:
\begin{equation}\label{eqn_ZphaseGateDefExplained}
	\input{pictures/chapter4/mermin/ZphaseGateDefExplained.tikz}
\end{equation}
\end{definition}

A simpler algebraic characterisation of phase gates is given by the following two equations, which are equivalent to Equation \ref{eqn_ZphaseGateDefExplained} (because $U$ is assumed to be unitary):
\begin{equation}\label{eqn_ZphaseGateConsequencesComonoid}
	\input{pictures/chapter4/mermin/ZphaseGateConsequencesComonoid.tikz}
\end{equation}
\begin{equation}\label{eqn_ZphaseGateConsequencesMonoid}
	\input{pictures/chapter4/mermin/ZphaseGateConsequencesMonoid.tikz}
\end{equation}
Both equations will play a pivotal role in this section: Equation \ref{eqn_ZphaseGateConsequencesComonoid} will features shortly in Lemma \ref{lem_GHZphaseGates}, the result relating phase gates and GHZ states, while Equation \ref{eqn_ZphaseGateConsequencesMonoid} will feature in Theorem \ref{thm_PhaseGatesUnbiasedStates}, the result relating phase gates and unbiased states. 

From Equation \ref{eqn_ZphaseGateDef}, it is not hard to see that $\hbox{\input{symbols/ZbwdotSym.tex}}\!\!$-phase gates form a group: we will refer to this as the \textbf{$\hbox{\input{symbols/ZbwdotSym.tex}}\!\!$-phase group}, and we will denote it by $\phaseGroup{\hbox{\input{symbols/ZbwdotSym.tex}}\!\!}$. If $\hbox{\input{symbols/ZbwdotSym.tex}}\!\!$ is a symmetric $\dagger$-SFA on a finite-dimensional Hilbert space $\SpaceH$, associated with a direct sum decomposition $\SpaceH = \bigoplus_j \SpaceH_j$, then the phase group $\phaseGroup{\hbox{\input{symbols/ZbwdotSym.tex}}\!\!}$ is given by the corresponding direct sum of unitary groups, modulo a global phase: 
\begin{equation}
\phaseGroup{\hbox{\input{symbols/ZbwdotSym.tex}}\!\!} = \Big(\bigoplus_j U(\SpaceH_j)\Big) \big/ S^1
\end{equation}
In the special case where $\hbox{\input{symbols/ZbwdotSym.tex}}\!\!$ is a $\dagger$-SCFA on $\SpaceH$, i.e. when all $\SpaceH_j$ subspaces are 1-dimensional, the phase group is abelian, the translation group of a torus: 
\begin{equation}
\phaseGroup{\hbox{\input{symbols/ZbwdotSym.tex}}\!\!} = \Big(\bigoplus_{j=1}^{\dim{\SpaceH}}U(1) \Big) \big/ S^1 \isom T^{\dim{\SpaceH}-1}
\end{equation}  
The connection between abelian phase groups and commutative Frobenius algebras generalises from $\fdHilbCategory$ to arbitrary dagger compact categories. The following result shows that the phase group of a commutative Frobenius algebra is always abelian, while the converse will be proven later on in Corollary \ref{cor_ZphaseGroupAbelianIff} (conditional to the existence of enough unbiased states)
\begin{lemma}\label{lem_ZphaseGroupAbelianProof}
Let $\hbox{\input{symbols/ZbwdotSym.tex}}\!\!$ be a $\dagger$-qSFA on an object $\SpaceH$ of a dagger compact category. If $\hbox{\input{symbols/ZbwdotSym.tex}}\!\!$ is commutative, then the $\hbox{\input{symbols/ZbwdotSym.tex}}\!\!$-phase group $\phaseGroup{\hbox{\input{symbols/ZbwdotSym.tex}}\!\!}$ is abelian.
\end{lemma}
\begin{proof} 
\begin{equation}
\input{pictures/chapter4/mermin/ZphaseGroupAbelianProof.tikz}
\end{equation}
The first equality is by unit law for $\hbox{\input{symbols/ZbwdotSym.tex}}\!\!$; the second equality is by Equation \ref{eqn_ZphaseGateConsequencesComonoid}; the third equality is some topological manipulation; the fourth equality (top right to bottom left) is by commutativity of $\hbox{\input{symbols/ZbwdotSym.tex}}\!\!$; the fifth equality is by Equation \ref{eqn_ZphaseGateConsequencesComonoid}; the sixth equality is commutativity of $\hbox{\input{symbols/ZbwdotSym.tex}}\!\!$; the seventh and last equality is by Equation \ref{eqn_ZphaseGateConsequencesComonoid}, followed by unit law for $\hbox{\input{symbols/ZbwdotSym.tex}}\!\!$. \end{proof}

Having defined the phase group and proven Lemma \ref{lem_ZphaseGroupAbelianProof}, we are now in a position to state the first important result of this Subsection. Lemma \ref{lem_GHZphaseGates} below will characterise the states that can be obtained by application of phases gates to a GHZ state: in the context of our generalised Mermin-type arguments, it will play the same role that Lemma \ref{lemma_qubitGHZphaseSum} played in Mermin's original argument.     
\begin{definition}\label{def_GHZ}
If $\hbox{\input{symbols/ZbwdotSym.tex}}\!\!$ is a $\dagger$-qSFA on an object $\SpaceH$ of a dagger compact category, the \textbf{$N$-partite $\hbox{\input{symbols/ZbwdotSym.tex}}\!\!$-GHZ state} is the following state of $\SpaceH^{\otimes N}$:
\begin{equation}
\input{pictures/chapter4/mermin/GHZstate.tikz}
\end{equation}
\end{definition}
\begin{lemma}\label{lem_GHZphaseGates}
Let $\hbox{\input{symbols/ZbwdotSym.tex}}\!\!$ be a $\dagger$-qSCFA on an object $\SpaceH$ of a dagger compact category. Then the state obtained by applying $\hbox{\input{symbols/ZbwdotSym.tex}}\!\!$-phase gates $U_1,...,U_N$ to the $N$-partite $\hbox{\input{symbols/ZbwdotSym.tex}}\!\!$-GHZ state only depends on the composition $U_1 \cdot ...\cdot U_N$ of the phase gates:
\begin{equation}
\input{pictures/chapter4/mermin/GHZstatePhaseGates.tikz}
\end{equation}
\end{lemma}
\begin{proof} Each $\hbox{\input{symbols/ZbwdotSym.tex}}\!\!$-phase gate is pushed down by using Equation \ref{eqn_ZphaseGateConsequencesComonoid} and commutativity of $\hbox{\input{symbols/ZbwdotSym.tex}}\!\!$. Formally, the proof is by induction, with inductive step given by the following equality:
\begin{equation}
\resizebox{\textwidth}{!}{\input{pictures/chapter4/mermin/GHZstatePhaseGatesProof.tikz}}
\end{equation}
\end{proof}


We have remarked before that the phase gates in Mermin's original argument are associated to certain phase states, extracted from their diagonalisation, which are also unbiased states for the relevant observable. As Theorem \ref{thm_PhaseGatesUnbiasedStates} below shows, the connection between $\hbox{\input{symbols/ZbwdotSym.tex}}\!\!$-phase gates and $\hbox{\input{symbols/ZbwdotSym.tex}}\!\!$-unbiased states holds true in full generality, and as a consequence we will also refer to $\hbox{\input{symbols/ZbwdotSym.tex}}\!\!$-unbiased states as \textbf{$\hbox{\input{symbols/ZbwdotSym.tex}}\!\!$-phase states}. In the case of $\fdHilbCategory$, the decomposition of a $\hbox{\input{symbols/ZbwdotSym.tex}}\!\!$-phase gate $U$ given by Equation \ref{eqn_ZphaseGateZphaseStateStatement} below for a $\dagger$-SCFA $\hbox{\input{symbols/ZbwdotSym.tex}}\!\!$ is equivalent to saying that $U$ is diagonal in the orthonormal basis $(\ket{x})_x$ associated with $\hbox{\input{symbols/ZbwdotSym.tex}}\!\!$, and has diagonal encoded by $\hbox{\input{symbols/ZbwdotSym.tex}}\!\!$-phase state $\ket{u}$ as $U_{xx} = \braket{x}{u}$.
\begin{theorem}[\textbf{Phase gates, phase states}]\label{thm_PhaseGatesUnbiasedStates}\hfill\\
Let $\hbox{\input{symbols/ZbwdotSym.tex}}\!\!$ be a $\dagger$-qSFA on an object $\SpaceH$ of a dagger compact category. Then the $\hbox{\input{symbols/ZbwdotSym.tex}}\!\!$-phase gates are exactly the maps $\phasegate{u}: \SpaceH \rightarrow \SpaceH$ taking the following form for some $\hbox{\input{symbols/ZbwdotSym.tex}}\!\!$-unbiased state $\psi_u$:
\begin{equation}\label{eqn_ZphaseGateZphaseStateStatement}
	\input{pictures/chapter4/mermin/ZphaseGateZphaseStateStatement.tikz}
\end{equation}
\end{theorem}
\begin{proof} First we prove that any phase gate $U$ takes the form above, for some $\hbox{\input{symbols/ZbwdotSym.tex}}\!\!$-unbiased state $\psi_{u}$.
An appropriate state $\psi_{u}$ can then be obtained by unit law for $\hbox{\input{symbols/ZbwdotSym.tex}}\!\!$: 
\begin{equation}
	\input{pictures/chapter4/mermin/ZphaseGateZphaseStateProof1.tikz}
\end{equation}
By using Equation \ref{eqn_ZphaseGateDefExplained}, we can prove that the state we obtained is $\hbox{\input{symbols/ZbwdotSym.tex}}\!\!$-unbiased:
\begin{equation}
	\input{pictures/chapter4/mermin/ZphaseGateZphaseStateProof2.tikz}
\end{equation}
Then we prove that any $U$ in the form above with $\psi_{u}$ a $\hbox{\input{symbols/ZbwdotSym.tex}}\!\!$-unbiased state is a unitary:
\begin{equation}
	\input{pictures/chapter4/mermin/ZphaseStateZphaseGateUnitaryProof.tikz}
\end{equation}
Finally, we prove that any unitary $U$ in the form above with $\psi_{u}$ a $\hbox{\input{symbols/ZbwdotSym.tex}}\!\!$-unbiased state is a $\hbox{\input{symbols/ZbwdotSym.tex}}\!\!$-phase gate:
\begin{equation}
	\input{pictures/chapter4/mermin/ZphaseStateZphaseGateProof.tikz}
\end{equation}
\end{proof}
\noindent Because of the correspondence above, we will adopt a uniform notation for phase gates and phase states, known in the literature as \textbf{decorated spider} notation \cite{Coecke2011,Coecke2016a}:
\begin{equation}
	\input{pictures/chapter4/mermin/ZphaseGatesStatesDecoratedSpiders.tikz}
\end{equation}

\begin{corollary}\label{cor_GHZphaseStates}
Let $\hbox{\input{symbols/ZbwdotSym.tex}}\!\!$ be a $\dagger$-qSCFA on an object $\SpaceH$ of a dagger compact category. Then the state obtained by applying $\hbox{\input{symbols/ZbwdotSym.tex}}\!\!$-phase gates $\phasegate{u_1},...,\phasegate{u_N}$ to the $N$-partite $\hbox{\input{symbols/ZbwdotSym.tex}}\!\!$-GHZ state takes the following form in terms of the corresponding $\hbox{\input{symbols/ZbwdotSym.tex}}\!\!$-phase states $\psi_{u_1},...,\psi_{u_N}$:
\begin{equation}
\input{pictures/chapter4/mermin/GHZstatePhaseStates.tikz}
\end{equation}
That is, the states that can be obtained by applying $\hbox{\input{symbols/ZbwdotSym.tex}}\!\!$-phase gates to the $N$-partite $\hbox{\input{symbols/ZbwdotSym.tex}}\!\!$-GHZ state are exactly those obtained by comultiplying $N$-times some $\hbox{\input{symbols/ZbwdotSym.tex}}\!\!$-unbiased state $u$ (specifically, above we have $u = u_1 \cdot ... \cdot u_N$, and all $\hbox{\input{symbols/ZbwdotSym.tex}}\!\!$-unbiased states can be obtained this way). 
\end{corollary}
\begin{proof} From Lemma \ref{lem_GHZphaseGates}, by re-writing each $\hbox{\input{symbols/ZbwdotSym.tex}}\!\!$-phase gate in terms of the corresponding $\hbox{\input{symbols/ZbwdotSym.tex}}\!\!$-phase state using Theorem \ref{thm_PhaseGatesUnbiasedStates}, and then using associativity to group the $\hbox{\input{symbols/ZbwdotSym.tex}}\!\!$-phase states together.
\end{proof}

The group structure of phase gates transfers to unbiased states via the correspondence given by Theorem \ref{thm_PhaseGatesUnbiasedStates}. Albeit not surprising, this result plays an important role in our generalisation of Mermin-type arguments, where it connects the operational side of phase gates and GHZ states to the algebraic side of strong complementarity (see Theorem \ref{thm_characterisingCandSC} below).
\begin{lemma}
Let $\hbox{\input{symbols/ZbwdotSym.tex}}\!\!$ be a $\dagger$-qSFA on an object $\SpaceH$ of a dagger compact category. Then $(\!\hbox{\input{symbols/ZbwmultSym.tex}}\!\!,\!\hbox{\input{symbols/ZbwunitSym.tex}}\!\!)$ endows the set of $\hbox{\input{symbols/ZbwdotSym.tex}}\!\!$-unbiased states with the structure of $\phaseGroup{\hbox{\input{symbols/ZbwdotSym.tex}}\!\!}$.
\end{lemma}
\begin{proof} The $\hbox{\input{symbols/ZbwdotSym.tex}}\!\!$-phase gate corresponding to the $\hbox{\input{symbols/ZbwdotSym.tex}}\!\!$-unbiased state $\!\hbox{\input{symbols/ZbwunitSym.tex}}\!\!$ is the identity, the unit of $\phaseGroup{\hbox{\input{symbols/ZbwdotSym.tex}}\!\!}$, so all we need to show is that composition of phase gates is the same as multiplication under $\!\hbox{\input{symbols/ZbwmultSym.tex}}\!\!$ of the corresponding $\hbox{\input{symbols/ZbwdotSym.tex}}\!\!$-unbiased states:
\begin{equation}
	\input{pictures/chapter4/mermin/ZphaseGroupProof.tikz}
\end{equation}
\end{proof}

\noindent As a bonus, the correspondence between the $\hbox{\input{symbols/ZbwdotSym.tex}}\!\!$-phase group and the group structure on $\hbox{\input{symbols/ZbwdotSym.tex}}\!\!$-unbiased states can be used to prove a converse to Lemma \ref{lem_ZphaseGroupAbelianProof}.
\begin{corollary}\label{cor_ZphaseGroupAbelianIff}
Let $\hbox{\input{symbols/ZbwdotSym.tex}}\!\!$ be a $\dagger$-qSFA on an object $\SpaceH$ of a dagger compact category, and assume that $\hbox{\input{symbols/ZbwdotSym.tex}}\!\!$ has \textbf{enough unbiased states}\footnote{Two morphisms $F,G: \SpaceH \rightarrow \SpaceK$ are equal whenever $F\circ \psi = G \circ \psi$ for all $\hbox{\input{symbols/ZbwdotSym.tex}}\!\!$-unbiased states $\psi$.}. Then $\hbox{\input{symbols/ZbwdotSym.tex}}\!\!$ is commutative iff $\phaseGroup{\hbox{\input{symbols/ZbwdotSym.tex}}\!\!}$ is abelian.
\end{corollary}
\begin{proof} We already know from Lemma \ref{lem_ZphaseGroupAbelianProof} that if $\hbox{\input{symbols/ZbwdotSym.tex}}\!\!$ is commutative then the $\hbox{\input{symbols/ZbwdotSym.tex}}\!\!$-phase group $\phaseGroup{\hbox{\input{symbols/ZbwdotSym.tex}}\!\!}$ must be abelian. Conversely, if $\phaseGroup{\hbox{\input{symbols/ZbwdotSym.tex}}\!\!}$ is abelian then so is the group structure induced by $(\!\hbox{\input{symbols/ZbwmultSym.tex}}\!\!,\!\hbox{\input{symbols/ZbwunitSym.tex}}\!\!)$ on the $\hbox{\input{symbols/ZbwdotSym.tex}}\!\!$-unbiased states. In particular, this means that $\!\hbox{\input{symbols/ZbwmultSym.tex}}\!\!$ is commutative whenever it is applied to $\hbox{\input{symbols/ZbwdotSym.tex}}\!\!$-unbiased states, and the existence of enough unbiased states allows us to conclude that $\hbox{\input{symbols/ZbwdotSym.tex}}\!\!$ is always commutative.
\end{proof}

With Theorem \ref{thm_PhaseGatesUnbiasedStates} we have proven a general correspondence between phase gates and unbiased states, while with Lemma \ref{lem_GHZphaseGates} and Corollary \ref{cor_GHZphaseStates} we have characterised the states that can be obtained by applying phase gates to GHZ states. Phase gates and the GHZ state for the Pauli $Z$ observable are the key operational ingredients for Mermin's original argument. However, just as important is the special algebraic standing of those phase gates derived from the eigenstates of the Pauli $X$ observable (an observable strongly complementary to Pauli $Z$), as opposed to the phase gates derived from other equatorial states (the eigenstates of observables complementary to Pauli $Z$). The last result of this section, Theorem \ref{thm_characterisingCandSC}, provides a general characterisation of complementarity and strong complementarity in terms of the relation between classical states of one observable and unbiased states of the other. Together with Theorem \ref{thm_PhaseGatesUnbiasedStates} and Corollary \ref{cor_GHZphaseStates}, it will form the basis for the formulation of our generalised Mermin-type arguments in the next Subsection.

\begin{theorem}[\textbf{Strong complementarity and phase groups}]\label{thm_characterisingCandSC}\hfill\\
Let $\hbox{\input{symbols/ZbwdotSym.tex}}\!\!$ and $\hbox{\input{symbols/DdotSym.tex}}\!\!$ be symmetric $\dagger$-qSFA on an object $\SpaceH$ of a $\dagger$-SMC. The following implications always hold:
\begin{enumerate}[(i)]
\item if $\hbox{\input{symbols/ZbwdotSym.tex}}\!\!$ and $\hbox{\input{symbols/DdotSym.tex}}\!\!$ are complementary, then the $\hbox{\input{symbols/DdotSym.tex}}\!\!$-classical states form a subset of the $\hbox{\input{symbols/ZbwdotSym.tex}}\!\!$-unbiased states, and viceversa; 
\item if $\hbox{\input{symbols/ZbwdotSym.tex}}\!\!$ and $\hbox{\input{symbols/DdotSym.tex}}\!\!$ are strongly complementary, then the $\hbox{\input{symbols/DdotSym.tex}}\!\!$-classical states form a subgroup of the $\hbox{\input{symbols/ZbwdotSym.tex}}\!\!$-unbiased states, and viceversa. 
\end{enumerate}
The converse implications hold if $\hbox{\input{symbols/DdotSym.tex}}\!\!$ has enough classical states: 
\begin{enumerate}[(i)]
\setcounter{enumii}{2}
\item if the $\hbox{\input{symbols/DdotSym.tex}}\!\!$-classical states form a subset of the $\hbox{\input{symbols/ZbwdotSym.tex}}\!\!$-unbiased states, then $\hbox{\input{symbols/ZbwdotSym.tex}}\!\!$ and $\hbox{\input{symbols/DdotSym.tex}}\!\!$ are complementary; 
\item if the $\hbox{\input{symbols/DdotSym.tex}}\!\!$-classical states form a subgroup of the $\hbox{\input{symbols/ZbwdotSym.tex}}\!\!$-unbiased states, then $\hbox{\input{symbols/ZbwdotSym.tex}}\!\!$ and $\hbox{\input{symbols/DdotSym.tex}}\!\!$ are strongly complementary.
\end{enumerate}
Note that the existence of enough $\hbox{\input{symbols/DdotSym.tex}}\!\!$-classical states implies the existence of enough $\hbox{\input{symbols/ZbwdotSym.tex}}\!\!$-unbiased states when the former are a subset/subgroup of the latter.
\end{theorem}
\begin{proof} 

Implication (i) is the statement of Lemma \ref{lem_complementarityUnbiased}, implication (ii) is the statement of Theorem \ref{thm_QuantumGroupAreGroupsOnPoints}, and implication (iii) is the statement of Lemma \ref{lem_complementarityUnbiasedConverse}. To prove implication (iv) we use the fact that $\hbox{\input{symbols/DdotSym.tex}}\!\!$ has enough classical states by hypothesis, and we work with the colour-swapped versions of the defining equation of strong complementarity (which imply the usual ones, see Remark \ref{rmrk_strongComplementarityColorSwappedEqns}). The colour-swapped top row of Equations \ref{eqn_strongComplementarityAltTopRow} simply states that the unit $\!\hbox{\input{symbols/ZbwunitSym.tex}}\!\!$ is a $\hbox{\input{symbols/DdotSym.tex}}\!\!$-classical state, something which is true when the $\hbox{\input{symbols/DdotSym.tex}}\!\!$-states are a subgroup of the $\hbox{\input{symbols/ZbwdotSym.tex}}\!\!$-unbiased states:
\begin{equation}\label{eqn_strongComplementarityAltTopRow}
	\input{pictures/chapter4/mermin/strongComplementarityAltTopRow.tikz}
\end{equation}
The colour-swapped bottom row of Equations \ref{eqn_strongComplementarityAltTopRow} holds applied to two $\hbox{\input{symbols/DdotSym.tex}}\!\!$-classical states if and only if the multiplication under $\!\hbox{\input{symbols/ZbwmultSym.tex}}\!\!$ of two $\hbox{\input{symbols/DdotSym.tex}}\!\!$-classical states is a $\hbox{\input{symbols/DdotSym.tex}}\!\!$-classical state, something which is always true when the $\hbox{\input{symbols/DdotSym.tex}}\!\!$-states are a subgroup of the $\hbox{\input{symbols/ZbwdotSym.tex}}\!\!$-unbiased states
\begin{equation}\label{eqn_strongComplementarityAltBottomRowApplied}
	\resizebox{\textwidth}{!}{\input{pictures/chapter4/mermin/strongComplementarityAltBottomRowApplied.tikz}}
\end{equation}
Conditional to $(\!\hbox{\input{symbols/ZbwmultSym.tex}}\!\!,\!\hbox{\input{symbols/ZbwunitSym.tex}}\!\!)$ endowing the $\hbox{\input{symbols/DdotSym.tex}}\!\!$-classical states with the structure of a monoid, Hopf's law applied to a $\hbox{\input{symbols/DdotSym.tex}}\!\!$-classical state is equivalent to the antipode acting as group inverse on $\hbox{\input{symbols/DdotSym.tex}}\!\!$-classical states. 
\begin{equation}\label{eqn_HopfLawApplied}
	\input{pictures/chapter4/mermin/HopfLawApplied.tikz}
\end{equation}
This concludes the proof of implication (iv).
\end{proof}

\subsection{Generalised Mermin-type Arguments}
\label{section_generalisedMerminArg}

Armed with the necessary results relating the classical and unbiased states of strongly complementary observables, we are now in a position to formulate our generalised Mermin-type arguments. To do so, we first review the ingredients of Mermin's original parity argument:
\begin{enumerate}[(a)]
\item a 3-partite qubit GHZ state for the Pauli $Z$ observable;
\item the abelian group $\phaseGroup{Z} \isom \reals/(2\pi\integers)$ of phase states for the Pauli $Z$ observable;
\item the finite subgroup $\{0,\pi\} \isom \integersMod{2}$ given by the eigenstates of the Pauli $X$ observable;
\item an equation $2 x = 1$ with no solution in the subgroup $\{0,\pi\}$ given by the Pauli $X$ eigenstates, but with a solution $\pi/2$ in the group $\reals/(2\pi\integers)$ of Pauli $Z$ phase states;
\item measurements in the Pauli $X$ observable.
\end{enumerate}
Similarly, our generalised Mermin-type arguments will involve the following ingredients:
\begin{enumerate}[(a)]
\item an $N$-partite GHZ state for a $\dagger$-qSCFA $\hbox{\input{symbols/ZbwdotSym.tex}}\!\!$;
\item the abelian group $(\phaseGroup{\hbox{\input{symbols/ZbwdotSym.tex}}\!\!}, \oplus, 0)$ of $\hbox{\input{symbols/ZbwdotSym.tex}}\!\!$-phase states\footnote{Isomorphic, by Theorem \ref{thm_PhaseGatesUnbiasedStates}, to the $\hbox{\input{symbols/ZbwdotSym.tex}}\!\!$-phase group, which we will denote by $(\phaseGroup{\hbox{\input{symbols/ZbwdotSym.tex}}\!\!}, \cdot, \id{})$.}; 
\item the subgroup $(\classicalStates{\hbox{\input{symbols/DdotSym.tex}}\!\!},\oplus,0)$, assumed to be finite, of $\hbox{\input{symbols/DdotSym.tex}}\!\!$-classical states for a symmetric $\dagger$-qSFA $\hbox{\input{symbols/DdotSym.tex}}\!\!$ strongly complementary to $\hbox{\input{symbols/ZbwdotSym.tex}}\!\!$;
\item a finite system of $\integers$-module equations, together with a solution in the group $\phaseGroup{\hbox{\input{symbols/ZbwdotSym.tex}}\!\!}$;
\item measurements in the $\hbox{\input{symbols/DdotSym.tex}}\!\!$ observable.
\end{enumerate}
The non-existence of a solution in the subgroup $\classicalStates{\hbox{\input{symbols/DdotSym.tex}}\!\!}$ of $\hbox{\input{symbols/DdotSym.tex}}\!\!$-classical states is not part of our generalised setup: it will be explicitly characterised as the necessary and sufficient condition for contextuality. Also, $N$ will not be a free parameter, being instead determined by the exponent of the finite abelian group $\classicalStates{\hbox{\input{symbols/DdotSym.tex}}\!\!}$. 

\begin{definition}
\label{def_generalisedMerminArgument}
Consider an $R$-probabilistic CPM~category~$\CategoryC$. A \textbf{generalised Mermin-type argument}~in~$\CategoryC$ is specified by the following data:
\begin{enumerate}[(i)]
	\item a strongly complementary pair $(\hbox{\input{symbols/ZbwdotSym.tex}}\!\!,\hbox{\input{symbols/DdotSym.tex}}\!\!)$ of a canonical $\dagger$-qSCFA $\hbox{\input{symbols/ZbwdotSym.tex}}\!\!$ and a canonical $\dagger$-SCFA~$\hbox{\input{symbols/DdotSym.tex}}\!\!$ on some object $\SpaceH$ of $\CategoryC$, such that $\hbox{\input{symbols/DdotSym.tex}}\!\!$ has enough classical states; we furthermore assume that the set $\classicalStates{\hbox{\input{symbols/DdotSym.tex}}\!\!}$ of $\hbox{\input{symbols/DdotSym.tex}}\!\!$-classical states is finite\footnote{This, together with commutativity of $\hbox{\input{symbols/ZbwdotSym.tex}}\!\!$, means that $(\classicalStates{\hbox{\input{symbols/DdotSym.tex}}\!\!},\oplus,0)$ is a finite abelian group.}, and that $|\classicalStates{\hbox{\input{symbols/DdotSym.tex}}\!\!}|$ is invertible as an element of the semiring $R$ of scalars of $\CategoryC$;
	\item a finite system $\mathcal{S}$ of $\integers$-module equations\footnote{I.e. equations with integer coefficients  $n_r^s \in \integers$ and valued in abelian groups (aka $\integers$-modules).}, with $a^1,...,a^S \in \classicalStates{\hbox{\input{symbols/DdotSym.tex}}\!\!}$:
	\begin{equation}\label{eqn_system}
	\mathcal{S} = \begin{cases}
		\bigoplus_{r=1}^{M} n^1_r \, y_r = a^1 \\
		\hspace{5mm}\vdots\\
		\bigoplus_{r=1}^{M} n^S_r \, y_r = a^S 
	\end{cases}
	\end{equation}  
	\item a given solution $(y_r := \beta_r)_{r=1}^M$ in the abelian group $\phaseGroup{\hbox{\input{symbols/ZbwdotSym.tex}}\!\!}$ of $\hbox{\input{symbols/ZbwdotSym.tex}}\!\!$-phase states;
	\item a positive integer $N$ such that $N \geq \sum_{r=1}^M n_r^s$ for all $s=1,...,S$, and satisfying $\gcd(N,\exp[\classicalStates{\hbox{\input{symbols/DdotSym.tex}}\!\!}]) = 1$, where $\exp[\classicalStates{\hbox{\input{symbols/DdotSym.tex}}\!\!}]$ is the exponent\footnote{The smallest positive integer $e$ such that $e \cdot g = 0$ for all $g \in \classicalStates{\hbox{\input{symbols/DdotSym.tex}}\!\!}$.} of $\classicalStates{\hbox{\input{symbols/DdotSym.tex}}\!\!}$.
\end{enumerate}
Therefore a generalised Mermin-type argument is specified by a quintuple $(\hbox{\input{symbols/ZbwdotSym.tex}}\!\!,\hbox{\input{symbols/DdotSym.tex}}\!\!, \mathcal{S}, \beta, N)$. 
\end{definition} 

\noindent The quintuple $(\hbox{\input{symbols/ZbwdotSym.tex}}\!\!,\hbox{\input{symbols/DdotSym.tex}}\!\!, \mathcal{S}, \beta, N)$ contains all the algebraic and operational ingredients we need to formulate a measurement scenario, which sees $N$ no-signalling parties sharing an $N$-partied $\hbox{\input{symbols/ZbwdotSym.tex}}\!\!$-GHZ state. Each party makes a measurement choice $m_j \in \{0,1,...,M\}$, applies the phase gate $\phasegate{\beta_{m_j}}$ to her system, and measures it in the $\hbox{\input{symbols/DdotSym.tex}}\!\!$ observable (i.e. measurement outcomes are valued in the set $\classicalStates{\hbox{\input{symbols/DdotSym.tex}}\!\!}$ of $\hbox{\input{symbols/DdotSym.tex}}\!\!$-classical states). 

Not all combinations of measurement choices are needed for the argument, and the measurement contexts will be determined by System \ref{eqn_system}. We begin by zero-padding the system as follows, so that exactly $N$ phase states are involved in each equation:
\begin{equation}\label{eqn_systemAlgExt}
\begin{cases}
	n^0_0 \,y_0 \oplus \;\;0 \,y_1 ... \oplus \;\;\;\,0 \,y_M = 0 \\
	n^1_0 \,y_0 \oplus n^1_1 \,y_1 ... \oplus n^1_M \,y_M = a^1 \\
	\hspace{10mm}\vdots\\
	n^S_0 \,y_0 \hspace{-0.5mm} \oplus n^S_1 \,y_1 ...\hspace{-0.5mm} \oplus n^S_M \,y_M = a^S 
\end{cases}
\end{equation}  
where we have defined $a^0 := 0$, $n_0^s := N - \sum_{r=1}^M n_r^s$ for all $s=1,...,S$, $n_0^0 := N$ and $n_r^0 := 0$ for all $r=1,...,M$; we will also extend the given solution by setting $\beta_0 := 0$. The first equation in System \ref{eqn_systemAlgExt} (which we will refer to by the special value $s=0$ of the parameter $s$) will contribute to a single measurement context, the \textbf{control}; each further equation (i.e. for each value $s=1,....,S$ of the parameter $s$) will give rise to $N$ measurement contexts, the \textbf{variations}, for a total of $1+S \cdot N$ measurement contexts involved in the scenario.

In the control, all parties choose $m_j^0 = 0$, i.e. perform no phase gate before measuring. They obtain the following global state (where $1 / |\classicalStates{\hbox{\input{symbols/DdotSym.tex}}\!\!}|^{N-1}$ is the normalisation factor required to obtain a $R$-distribution):
\begin{equation}\label{eqn_MerminControl}
	\input{pictures/chapter4/mermin/MerminControl.tikz}
\end{equation}
The first variation for each value $s=1,...,S$ is specified by the corresponding equation in System \ref{eqn_systemAlgExt}: the first $n^s_0$ parties choose $m_j^s=0$, the next $n^s_1$ parties choose $m_j^s = 1$, the next $n^s_2$ parties choose $m_j^s = 2$ and so on, until the last $n^s_M$ parties choose $m_j^s = M$:
\begin{equation}\label{eqn_measurementChoices}
	m_j^s := \text{the largest $m \in \{0,...,M\}$ such that }  j \geq \sum_{r=0}^{m-1} n^s_r
\end{equation}
They obtain the following global state, where the equality results from an application of Corollary \ref{cor_GHZphaseStates}, using the relevant equation from System \ref{eqn_systemAlgExt}:
\begin{equation}\label{eqn_MerminVariation1}
	\input{pictures/chapter4/mermin/MerminVariation1.tikz}
\end{equation}
For each fixed value of $s$, the next $N-1$ variations are cyclic permutations of the first. The measurement choice for the $j^{th}$ party at the $k^{th}$ variation of a given $s$ is $m_{j+(k-1)}^s$, where the sum $j+(k-1)$ is taken modulo $N$: 
\begin{equation}\label{system_variations}
\begin{array}{c|ccccc}
	\text{Parties:} & 1 & 2 & ... & N-1 & N \\
	\hline
	\text{ $1^{st}$ variation for } s & m_1^s & m_2^s & ... & m_{N-1}^s & m_N^s \\
	\text{ $2^{nd}$ variation for } s & m_2^s & m_3^s & ... & m_N^s & m_1^s \\
	\text{ $3^{rd}$ variation for } s & m_3^s & m_4^s & ... & m_1^s & m_2^s \\
	\vdots & \vdots & \vdots &  & \vdots & \vdots \\
	\text{ $N^{th}$ variation for } s  & m_N^s & m_1^s & ... & m_{N-2}^s & m_{N-1}^s
\end{array}
\end{equation} 
Because $\hbox{\input{symbols/ZbwdotSym.tex}}\!\!$ is commutative, the global state obtained is the same as that for the first variation for that value of $s$ (shown on the RHS of Equation \ref{eqn_MerminVariation1}).

By using strong complementarity and Theorem \ref{thm_characterisingCandSC}, we rewrite the global state obtained by the $N$ parties in the control and variations, obtaining an explicit $R$-distribution over the set $\classicalStates{\hbox{\input{symbols/DdotSym.tex}}\!\!}^N$ of joint measurement outcomes (from now on, the parameter $s$ can take any value in $\{0,1,...,S\}$, unless otherwise specified). 
\begin{lemma}
\begin{equation}\label{eqn_MerminVariationDistrib}
	\input{pictures/chapter4/mermin/MerminVariationDistrib.tikz}
\end{equation}
\end{lemma}
\begin{proof} 
Strong complementarity can be used to swap $\hbox{\input{symbols/ZbwdotSym.tex}}\!\!$ and $\hbox{\input{symbols/DdotSym.tex}}\!\!$, as shown in Corollary 4.1 of \cite{Coecke2012c}, and then $a^s$ can be pushed through because it is a $\hbox{\input{symbols/DdotSym.tex}}\!\!$-classical state (we have left normalisation aside, and we use $a^0:= 0$ to treat control and variations uniformly): 
\begin{equation}\label{eqn_MerminVariationSCDistrib1}
	\input{pictures/chapter4/mermin/MerminVariationSCDistrib1.tikz}
\end{equation}
Using fact that $\hbox{\input{symbols/DdotSym.tex}}\!\!$ has enough classical states, and recalling from Theorem \ref{thm_characterisingCandSC} that $(\!\hbox{\input{symbols/ZbwmultSym.tex}}\!\!,\!\hbox{\input{symbols/ZbwunitSym.tex}}\!\!)$ acts as the group multiplication of $\classicalStates{\hbox{\input{symbols/DdotSym.tex}}\!\!}$ when restricted to the $\hbox{\input{symbols/DdotSym.tex}}\!\!$-classical states, we can further decompose the state on the RHS of Equation \ref{eqn_MerminVariationSCDistrib1} into an $R$-distribution over the set $\classicalStates{\hbox{\input{symbols/DdotSym.tex}}\!\!}^N$:
\begin{equation}\label{eqn_MerminVariationSCDistrib2}
	\input{pictures/chapter4/mermin/MerminVariationSCDistrib2.tikz}
\end{equation}
\end{proof}

The joint outcome of measurements for the control is uniformly distributed over the subgroup $H_0 \normalSubgroup \classicalStates{\hbox{\input{symbols/DdotSym.tex}}\!\!}^N$ specified by $H_0 := \suchthat{(g_1,...,g_N)}{g_1 \oplus ... \oplus g_N = 0}$, while the joint outcome of any of the $N$ variations for each specific value of $s$ is uniformly distributed over the coset $H_{a_s} := (a^s,0,...,0) \oplus H_0$. For each $s,s' \in \{0,1,...,S\}$, the cosets $H_{a_s}$ and $H_{a_{s'}}$ are disjoint if and only if $a^s \neq a^{s'}$. All in all, we get the following empirical model for the generalised Mermin-type argument:
\begin{align}
\mathbb{P}[(g_1,...,g_N) | \text{control}] &= 
	\begin{cases}
		\frac{1}{|\classicalStates{\hbox{\input{symbols/DdotSym.tex}}\!\!}|^{N-1}} & \text{ if } g_1 \oplus ... \oplus g_N = 0 \\
		\hfill 0 \hfill & \text{ otherwise }
	\end{cases}\label{eqn_empiricalModelControl}\\
\mathbb{P}[(g_1,...,g_N) | \text{$k^{th}$ variation for $s$}] &= 
	\begin{cases}
		\frac{1}{|\classicalStates{\hbox{\input{symbols/DdotSym.tex}}\!\!}|^{N-1}} & \text{ if } g_1 \oplus ... \oplus g_N = a^s \\
		\hfill 0 \hfill & \text{ otherwise }
	\end{cases}\label{eqn_empiricalModelVariations}
\end{align} 

One of the catchy features of Mermin's original argument is that it is entirely deterministic: instead of relying on the violation of some probabilistic inequality, the proof of contextuality shows that the existence of a local hidden variable (LHV) model would lead to the existence of solutions to an unsatisfiable parity equation (i.e. one which doesn't admit solutions in the finite abelian group $\integersMod{2}$). The proof of contextuality for our generalised Mermin-type arguments goes by similar lines, showing that the existence of a LHV model is equivalent to System \ref{eqn_system} admitting solutions in the finite abelian group $\classicalStates{\hbox{\input{symbols/DdotSym.tex}}\!\!}$. 

\begin{theorem}[\textbf{Mermin-type contextuality}]\label{thm_contextuality}\hfill\\
Consider an $R$-probabilistic CPM~category~$\CategoryC$, and let  $(\hbox{\input{symbols/ZbwdotSym.tex}}\!\!,\hbox{\input{symbols/DdotSym.tex}}\!\!, \mathcal{S}, \beta, N)$ be a generalised Mermin-type argument in it. If the associated empirical model is contextual, then the system $\mathcal{S}$ admits no solution in the finite abelian group $\classicalStates{\hbox{\input{symbols/DdotSym.tex}}\!\!}$. Conversely, if the system $\mathcal{S}$ admits no solution in $\classicalStates{\hbox{\input{symbols/DdotSym.tex}}\!\!}$ and $R$ is a positive semiring, then the empirical model is contextual.
\end{theorem}
\begin{proof} 

The proof comes in two parts: ($\Rightarrow$) we show that any solution in $\classicalStates{\hbox{\input{symbols/DdotSym.tex}}\!\!}$ can be turned into a LHV model; ($\Leftarrow$) we show that, as long as $R$ is a positive semiring, any LHV model can be turned into a solution in $\classicalStates{\hbox{\input{symbols/DdotSym.tex}}\!\!}$. 

\noindent \textbf{Proof of ($\Rightarrow$).} Assume that the system $\mathcal{S}$ (in the form of System \ref{eqn_system}) admits a solution $(y_r := b_r)_{r=1}^{M}$, and define $b_0 := 0$. A LHV model can be obtained as follows:
\begin{enumerate}[(i)]
	\item the uniform $R$-distribution on $H_0 \normalSubgroup \classicalStates{\hbox{\input{symbols/DdotSym.tex}}\!\!}^N$ is taken as a shared classical state amongst the $N$ parties:
		\begin{equation}\label{eqn_MerminLHVDistrib1}
			\input{pictures/chapter4/mermin/MerminLHVDistrib1.tikz}
		\end{equation}
	\item upon measurement choice $m_j \in \{0,1,...,M\}$ for the $j^{th}$ party, a translation by $b_{m_j}$ in the group $\classicalStates{\hbox{\input{symbols/DdotSym.tex}}\!\!}$ is applied to the respective classical subsystem, independently of the measurement choices of the other parties:
		\begin{equation}\label{eqn_MerminLHVDistrib2}
			\input{pictures/chapter4/mermin/MerminLHVDistrib2.tikz}
		\end{equation}
\end{enumerate} 
All we need to show is that the procedure above produces the same $R$-distributions on $\classicalStates{\hbox{\input{symbols/DdotSym.tex}}\!\!}^N$ as those given by the empirical model of Equations \ref{eqn_empiricalModelControl} and \ref{eqn_empiricalModelVariations}. To do so, we simply observe that the global state obtained with the procedure above is the same as the global states obtained in the control \ref{eqn_MerminControl} and in the variations \ref{eqn_MerminVariation1} (which we treat uniformly by considering $s=0,1,...,S$), because $b_0,b_1,...,b_N$ satisfy the same equations satisfied by the phases $\beta_0,\beta_1,...,\beta_N$:
\begin{equation}\label{eqn_MerminLHVDistrib}
	\input{pictures/chapter4/mermin/MerminLHVDistrib.tikz}
\end{equation}

\noindent \textbf{Proof of ($\Leftarrow$).} Now assume that $R$ is a positive semiring, and that the scenario admits a LHV model:
\begin{enumerate}[(i)]
	\item there is a some finite set $\Lambda$, the set of values for the hidden variable, coming with an $R$-distribution $p : \Lambda \rightarrow R$;
	\item for each possible measurement choice $r=0,1,...,M$ that each party $i=1,...,N$ can make, there is a family $(c_{r}^{i,\lambda})_{\lambda \in \Lambda}$ of $\hbox{\input{symbols/DdotSym.tex}}\!\!$-classical states, the deterministic local outcomes for each value of the hidden variable;
	\item for each measurement context (either $s=0$, $k=1$ for the control, or $(s,k)\in \{1,...,S\}\times\{1,...,N\}$ for the $N \cdot S$ variations), a definite $\hbox{\input{symbols/DdotSym.tex}}\!\!$-classical outcome $d_{s,k}^{i,\lambda}$ is obtained by each party $i=1,...,N$ at each definite value $\lambda \in \Lambda$ of the hidden variable:
		\begin{equation}
			d_{s,k}^{i,\lambda} := c_{m_{i+(k-1)}^s}^{i,\lambda}
		\end{equation}
	\item if these definite $\hbox{\input{symbols/DdotSym.tex}}\!\!$-classical global states are weighted based on the $R$-distribution $p$ on $\Lambda$, one obtains the same $R$-distribution on joint measurement outcomes that would be expected from the measurement context:
		\begin{equation}\label{eqn_MerminLHV}
			\input{pictures/chapter4/mermin/MerminLHV.tikz}
		\end{equation} 
\end{enumerate}
Given a LHV model, we can sum up all $N$ outcomes of each side of Equation \ref{eqn_MerminLHV} in $(\classicalStates{\hbox{\input{symbols/DdotSym.tex}}\!\!},\oplus,0)$ to obtain an equation between $R$-distribution over $\classicalStates{\hbox{\input{symbols/DdotSym.tex}}\!\!}$:
\begin{equation}\label{eqn_MerminLHVsummed}
	\input{pictures/chapter4/mermin/MerminLHVsummed.tikz}
\end{equation}
The last equation used the fact that $\hbox{\input{symbols/DdotSym.tex}}\!\!$ was chosen to be special\footnote{The special $\hbox{\input{symbols/DdotSym.tex}}\!\!$ could have been replaced by a more general $\dagger$-qSCFA, but at the price of an additional normalisation factor in all global states.}, and hence the normalisation factor for the $\dagger$-qSCFA $\hbox{\input{symbols/ZbwdotSym.tex}}\!\!$ is $|\classicalStates{\hbox{\input{symbols/DdotSym.tex}}\!\!}|$ (because $\hbox{\input{symbols/DdotSym.tex}}\!\!$ has enough classical states)\footnote{The normalisation factor $|\classicalStates{\hbox{\input{symbols/DdotSym.tex}}\!\!}|$ refers to two wires: each additional wire is an additional copy of $|\classicalStates{\hbox{\input{symbols/DdotSym.tex}}\!\!}|$, for a total of $|\classicalStates{\hbox{\input{symbols/DdotSym.tex}}\!\!}|^{N-1}$ in the $N$-wire case here.}. Equation \ref{eqn_MerminLHVsummed} can be turned into the following conditions on the LHV: 
\begin{equation}
	\sum_{\lambda \text{ s.t. }\bigoplus_{i=1}^N d_{s,k}^{i,\lambda} = a^s} \hspace{-0.75cm} p(\lambda) = 1 \hspace{2cm}
	\sum_{\lambda \text{ s.t. }\bigoplus_{i=1}^N d_{s,k}^{i,\lambda} \neq a^s} \hspace{-0.75cm} p(\lambda) = 0
\end{equation}
Because $R$ is a positive semiring, $p(\lambda) = 0$ for any $\lambda$ such that $\oplus_{i=1}^N d_{s,k}^{i,\lambda} \neq a^s$ for some $s$. Conversely, picking any $\lambda_+$ such that $p(\lambda_+) > 0$ (and at least one such $\lambda_+$ exists, because $p$ is an $R$-distribution) yields a family $(d_{s,k}^{i,\lambda_+})_{s,k,i}$ such that $\oplus_{i=1}^N d_{s,k}^{i,\lambda_+} = a^s$ for all $s$ and $k$. For the control ($s=0$ and $k=1$), we obtain the following equation: 
\begin{equation}\label{eqn_LHVstatesSummedUpControl}
	\oplus_{i=1}^N c_{0}^{i,\lambda_+} = 0
\end{equation}
For each variation $(s,k) \in \{1,...,S\}\times\{1,...,N\}$, we obtain the following equation:
\begin{equation}
	\oplus_{i=1}^N c_{m_{i+(k-1)}^s}^{i,\lambda_+} = a^s
\end{equation}
If $c_r^{i,\lambda_+}$ were independent of the party $i$ for all $r=1,...,M$, this equation would yield a solution to system $\mathcal{S}$ in the form of $b_r := c_r^{i,\lambda_+}$ for any $i$; unfortunately, this need not be the case. This is where our cyclic definition of the $N$ variations for each value of $s$ comes into play. For each fixed value of $s$, we add up the $N$ equations for $k=1,...,N$:
\begin{equation}\label{eqn_LHVstatesSummedUp}
	\oplus_{k=1}^{N}\oplus_{i=1}^N c_{m_{i+(k-1)}^s}^{i,\lambda_+} = N a^s
\end{equation}
Because $\gcd(N, \exp[\classicalStates{\hbox{\input{symbols/DdotSym.tex}}\!\!}]) = 1$, the equation above has solutions if and only if the equation below does:
\begin{equation}\label{eqn_LHVstatesSummedUpEquiv}
	\oplus_{k=1}^{N}\oplus_{i=1}^N c_{m_{i+(k-1)}^s}^{i,\lambda_+} = a^s
\end{equation}
Now refer to the Table \ref{system_variations} defining the $N$ variations for $s$ and to the Equation \ref{eqn_measurementChoices} defining the measurement choices. The LHS of Equation \ref{eqn_LHVstatesSummedUp} is a sum by rows of the $N^2$ measurement choices in Table \ref{system_variations}: each $r=0,1,...,M$ appears $n_r^s$ times in each row, but the changing value of $i$ along each row stops us from turning it into a solution to system $\mathcal{S}$. However, we can switch the summations in Equation \ref{eqn_LHVstatesSummedUp} to obtain a sum by columns of the table, where each $r=0,1,...,M$ still appears $n_r^s$ times in each column (by the cyclic definition), but now $i$ is constant along each column:
\begin{equation}\label{eqn_LHVstatesSummedUp2}
	\oplus_{i=1}^N \oplus_{k=1}^{N} c_{m_{i+(k-1)}^s}^{i,\lambda_+} = \oplus_{i=1}^N \oplus_{r=0}^M n_r^s c_r^{i,\lambda_+}
\end{equation}
We can then sum up all $(c_r^{i,\lambda_+})_{i=1}^N$ for each $r=0,1,...,M$, and use Equation \ref{eqn_LHVstatesSummedUp} (together with Equation \ref{eqn_LHVstatesSummedUpControl} to cancel out the contribution from $r=0$) to finally obtain the desired solution $(b_r)_{r=1}^M$ to system $\mathcal{S}$:
\begin{equation}
	\input{pictures/chapter4/mermin/MerminLHVsolution.tikz}
\end{equation}
\end{proof}

\subsection{Quantum realisability}
\label{section_quantumRealisab}

In quantum theory, i.e. in the $\reals^+$-probabilistic CPM category $\CPMCategory{\fdHilbCategory}$, many of the requirements of generalised Mermin-type arguments are automatically satisfied: canonical $\dagger$-SCFA in $\CPMCategory{\fdHilbCategory}$ (i.e. $\dagger$-SCFA in $\fdHilbCategory$) always have enough classical states (and finitely many so), the semiring $\reals^+$ of scalars is positive, and any non-zero integer is invertible in it. Hence, only strong complementarity is required in point (i) of the definition generalised Mermin-type arguments, and Theorem \ref{thm_contextuality} establishes an unconditional equivalence between contextuality of a generalised Mermin-type argument $(\hbox{\input{symbols/ZbwdotSym.tex}}\!\!,\hbox{\input{symbols/DdotSym.tex}}\!\!, \mathcal{S}, \beta, N)$ and the existence of solutions to system $\mathcal{S}$ in the finite abelian group $\classicalStates{\hbox{\input{symbols/DdotSym.tex}}\!\!}$ of $\hbox{\input{symbols/DdotSym.tex}}\!\!$-classical states. 

The remarks above show that the correspondence between systems of equations in finite abelian groups and generalised Mermin-type arguments is particularly tight in the case of quantum theory, but an important question remains unanswered: which systems of $\integers$-module equations lead to arguments which can be realised in quantum theory? As it turns out, all of them (but an obvious caveat applies).

\begin{theorem}[\textbf{Quantum Realisability}]\label{thm_quantumRealisability}\hfill\\
Let $(K,\oplus,0)$ be a finite abelian group, and $\mathcal{S}$ be a finite system of $\integers$-module equations in the following form, with $a^1,...,a^S \in K$:
\begin{equation}\label{eqn_systemQuantumRealisability}
\mathcal{S} = \begin{cases}
	\bigoplus_{r=1}^{M} n^1_r \, y_r = a^1 \\
	\hspace{5mm}\vdots\\
	\bigoplus_{r=1}^{M} n^S_r \, y_r = a^S 
\end{cases}
\end{equation}  
Assume that the system is \textbf{consistent} in the following sense, where by $\underline{n}^s \in \integers^{M}$ we denoted the row vectors of System \ref{eqn_systemQuantumRealisability}:
\begin{equation}
	\bigoplus_{s=1}^{S} c_s \cdot \underline{n}^{s} =_{\integers^M} \underline{0} \implies \bigoplus_{s=1}^{S} c_s \cdot a^{s} =_{K} 0,
\end{equation}
Then for every $|K|$-dimensional quantum system $\SpaceH$ and every $\dagger$-qSCFA $\hbox{\input{symbols/ZbwdotSym.tex}}\!\!$ on $\SpaceH$ with normalisation factor $|K|$, there exists a generalised Mermin-type argument $(\hbox{\input{symbols/ZbwdotSym.tex}}\!\!,\hbox{\input{symbols/DdotSym.tex}}\!\!, \mathcal{S}, \beta, N)$ corresponding to System \ref{eqn_systemQuantumRealisability}, i.e. we can always find:
\begin{enumerate}[(i)]
	\item a $\dagger$-SCFA $\hbox{\input{symbols/DdotSym.tex}}\!\!$, strongly complementary to $\hbox{\input{symbols/ZbwdotSym.tex}}\!\!$, such that $(\classicalStates{\hbox{\input{symbols/DdotSym.tex}}\!\!},\!\hbox{\input{symbols/ZbwmultSym.tex}}\!\!,\!\hbox{\input{symbols/ZbwunitSym.tex}}\!\!) \isom (K,\oplus,0)$; 
	\item a solution $(y_r := \beta_r)_{r=1}^M$ to $\mathcal{S}$ in $\phaseGroup{\hbox{\input{symbols/ZbwdotSym.tex}}\!\!} \isom T^{|K|-1}$;
	\item a positive integer $N$ (infinitely many, in fact) such that $N \geq \sum_{r=1}^{M} n_r^s$ for all $s=1,...,S$, and such that $\gcd(N,\exp[\classicalStates{\hbox{\input{symbols/DdotSym.tex}}\!\!}]) = 1$. 
\end{enumerate}
\end{theorem}
\begin{proof} 
Point (iii) is trivial: there are infinitely many positive integers $N$ such that $\gcd(N,\exp[K])=1$, and hence we can always find one such that $N \geq \sum_{r=1}^{M} n_r^s$ for all $s=1,...,S$. Point (i) is more interesting, and relies on the characterisation of strong complementarity in $\fdHilbCategory$ and Pontryagin duality for finite abelian groups. Point (ii) is perhaps the most interesting, and relies on the possibility of solving consistent systems of $\integers$-module equations in the torus $T^{|K|-1}$.

\noindent \textbf{Proof of point (i).} Because $\hbox{\input{symbols/ZbwdotSym.tex}}\!\!$ is a $\dagger$-qSCFA with normalisation factor $|K|$ on a $|K|$-dimensional Hilbert space $\SpaceH$, it is associated with a basis of $|K|$ vectors, each having norm $\sqrt{|K|}$. Label the basis vectors by the $|K|$ multiplicative characters $\goodchi \in K^\wedge$ of the finite abelian group $K$, and construct an orthonormal basis by using the multiplicative characters $\tau \in (K^\wedge)^\wedge$ of the finite abelian group $K^\wedge$:
\begin{equation}\label{eqn_doubleDualBasis}
	\ket{\tau} := \frac{1}{|K|} \sum_{\goodchi \in K^\wedge} \tau(\goodchi) \ket{\goodchi}
\end{equation}
By Pontryagin duality, there is a canonical isomorphism $(K^\wedge)^\wedge \isom K$, so that the new orthonormal basis given by Equation \ref{eqn_doubleDualBasis} is canonically labelled by elements of $K$. Consider the $\dagger$-SCFA $\hbox{\input{symbols/DdotSym.tex}}\!\!$ associated to the orthonormal basis thus defined to obtain the desired $(\classicalStates{\hbox{\input{symbols/DdotSym.tex}}\!\!},\!\hbox{\input{symbols/ZbwmultSym.tex}}\!\!,\!\hbox{\input{symbols/ZbwunitSym.tex}}\!\!) \isom (K,\oplus,0)$.

\noindent \textbf{Proof of point (ii).} The phase group $\phaseGroup{\hbox{\input{symbols/ZbwdotSym.tex}}\!\!}$ for a canonical $\dagger$-qSCFA on a $|K|$-dimensional Hilbert space in $\CPMCategory{\fdHilbCategory}$ is isomorphic to the $(|K|-1)$-dimensional torus, an abelian Lie group. To find a solution $(y_r := \beta_r)_{r=1}^M$ to System \ref{eqn_systemQuantumRealisability}, we will show that one can always find solutions to arbitrary consistent systems of $\integers$-module equations in a torus.

While all $K$-valued systems with solutions in some super-group of $K$ must necessarily be consistent, the converse is not true in general: given a super-group $P$ of $K$ there may be consistent systems with no solutions in $P$. Certainly if $P$ is finite then at least one such system exists (because of the finite exponent), and certainly if $P=\rationals^d$ then no such system exists; in fact, every divisible torsion-free abelian group $P$ is canonically a $\rationals$-vector space, and thus every consistent system of $\integers$-modules equations  (and, in fact, of $\rationals$-vector space equations) valued in a divisible torsion-free abelian group $P$ has solutions in $P$ (e.g. by Gaussian elimination over the field $\rationals$). Unfortunately, while tori are divisible, they are not torsion-free, and in particular not $\rationals$-vector spaces: as a consequence, the reasoning above does not apply. 

However, a more general argument can be used to show that any consistent system of equations can be solved in any divisible abelian group, regardless of whether the group is torsion-free or not \cite{fuchs2015abelian} (although uniqueness of solution need not hold for systems with linearly independent row vectors). As tori are divisible abelian groups, all consistent systems of $\integers$-module equations can be solved in them, and in particular we can find our solution $(y_r := \beta_r)_{r=1}^M$ to System \ref{eqn_systemQuantumRealisability}.
\end{proof}

\newcommand{\eqnIndex}[1]{\operatorname{index}(#1)}
\newcommand{\RlinearTheory}[2]{\mathbb{T}_{#1}(#2)}
\newcommand{\AvN}[2]{\operatorname{AvN}_{#1,#2}}
\newcommand{\AvNring}[1]{\operatorname{AvN}_{#1}}

\subsection{All-vs-Nothing Arguments}
\label{section_AvN}

Strong contextuality can be reformulated directly in terms of the supports of the distributions. The supports of the global sections, i.e. the $d \in \presheafOfDistributions{\mathbb{B}}{\mathcal{X}}$ satisfying Equation \ref{eqn_StrongContextualityCondition} form a (possibly empty) lattice, and thus a probabilistic empirical model is strongly contextual iff the following set is empty:
\begin{equation}
	\supportSubpresheaf{\mathcal{X}} := \Big\{ s \in \sheafOfEvents{\mathcal{X}} \Big\vert \restrict{s}{C} \in \support{\zeta_C} \text{ for all } C \in \mathcal{M}\Big\}
\end{equation}
For a possibilistic (no-signalling) empirical model $(\zeta_C)_{C \in \mathcal{M}}$, we can define \cite{Abramsky2015} a \textbf{support subpresheaf} $\supportSubpresheafSym \subseteq \sheafOfEventsSym$ by setting:
\begin{equation}
	\supportSubpresheaf{U} := \suchthat{s \in \sheafOfEvents{U}}{ \restrict{s}{C \cap U} \in \support{\restrict{\zeta_C}{U \cap C}} \text{ for all } C \in \mathcal{M}}
\end{equation}
Then a possibilistic empirical model is strongly contextual if and only if $\supportSubpresheaf{\mathcal{X}} = \emptyset$.

The fundamental observation behind the \textbf{All-vs-Nothing arguments} of \cite{Abramsky2015} is that contextuality of Mermin's original argument follows from the existence of the system of $\integersMod{2}$ equations which has no global solution (corresponding to $\supportSubpresheaf{\mathcal{X}} = \emptyset$ in the sheaf-theoretic framework for contextuality \cite{Abramsky2011} we have previously summarised), but where each equation admits a solution (i.e. we have $\supportSubpresheaf{C} \neq \emptyset$ for the measurement context $C$ associated to each equation). In this Subsection we summarise the basic framework of All-vs-Nothing arguments from \cite{Abramsky2015}, taking the liberty of slightly generalising the definitions therein, from rings to modules over rings.

Let $\mathcal{R}$ be a commutative ring with unit: we will denote by $+$ the addition in the ring $\mathcal{R}$, and by $\oplus$ the addition in $\mathcal{R}$-modules. The ring $\mathcal{R}$ should not be confused with the semiring $R$ over which the distributions are taken (i.e. the semiring of scalars of the $R$-probabilistic CPM category which the arguments take place in). If $G$ is some $\mathcal{R}$-module, we will define an \textbf{$\mathcal{R}$-linear equation valued in $G$} to be a triple $\phi = (C,n,b)$ where:
\begin{enumerate}
	\item[(i)] $C$ is some finite set, and we define $\eqnIndex{\phi} := C$;
	\item[(ii)] $n: C \rightarrow \mathcal{R}$ is any function;
	\item[(iii)] $b \in G$ is a given element of $G$.
\end{enumerate}
If $\phi = (C,n,b)$ is an $\mathcal{R}$-linear equation valued in $G$, we will say that a function $s: C \rightarrow G$ (henceforth an \textbf{assignment}) \textbf{satisfies} $\phi$, written $s \models \phi$, if and only if the following equation holds in $G$:
\begin{equation}
	\bigoplus_{m \in C} n_m s_m = b
\end{equation}
where we denoted $n_m := n(m)$ and $s_m := s(m)$. Any set $W$ of assignments $C \rightarrow G$ can be associated a corresponding set $\RlinearTheory{\mathcal{R}}{W}$ of satisfied equations, which is itself an $\mathcal{R}$-module\footnote{This gives rise to some interesting results on affine closures, see \cite{Abramsky2015}.}:
\begin{equation}
	\label{eqn_RlinearTheory}
	\RlinearTheory{\mathcal{R}}{W} := \suchthat{\phi}{s \models \phi \text{ for all }s \in W}
\end{equation} 

Let $(\zeta_C)_{C \in \mathcal{M}}$ be a possibilistic empirical model for a measurement scenario $(\sheafOfEventsSym,\mathcal{M})$, such that all measurements have the same $\mathcal{R}$-module $G$ as their set of outcomes (for example we had $G = \integersMod{2}$, a $\integers$-module, for Mermin's original argument). Let $\supportSubpresheafSym \subseteq \sheafOfEventsSym$ be the support subpresheaf for the empirical model and define its \textbf{$\textbf{R}$-linear theory}:
\begin{equation}
	\RlinearTheory{\mathcal{R}}{\supportSubpresheafSym} := \bigcup_{C \in \mathcal{M}}  \RlinearTheory{\mathcal{R}}{\supportSubpresheaf{C}}
\end{equation} 
We say that a possibilistic empirical model is \textbf{All-vs-Nothing} with respect to ring $\mathcal{R}$ and $\mathcal{R}$-module $G$, written $\AvN{\mathcal{R}}{G}$, iff the $\mathcal{R}$-linear theory admits no solution in $G$, i.e. iff there exists no global assignment $s: \mathcal{X} \rightarrow G$ such that:
\begin{equation}
	\label{eqn_AvNDefinition}
	\restrict{s}{C} \models \phi \text{ for all } C \in \mathcal{M} \text{ and all } \phi \in \RlinearTheory{\mathcal{R}}{\supportSubpresheaf{C}}
\end{equation}
To connect back with the notation in \cite{Abramsky2015}, we will simply write $\AvNring{\mathcal{R}}$ for $\AvN{\mathcal{R}}{\mathcal{R}}$. 

A straightforward generalisation (from rings to modules) of a result by \cite{Abramsky2015} proves that any possibilistic empirical model which is $\AvN{\mathcal{R}}{G}$ for some ring $\mathcal{R}$ and some $\mathcal{R}$-module $G$ is strongly contextual: if the model weren't strongly contextual, then there would be some global section $s \in \supportSubpresheaf{\mathcal{X}}$, and this would imply $\restrict{s}{C} \in \supportSubpresheaf{C}$ for all $C \in \mathcal{M}$, which in turn would prove that global assignment $s$ satisfies Equation \ref{eqn_AvNDefinition} (by appealing to Equation \ref{eqn_RlinearTheory}). 

A result by \cite{Abramsky2011} shows that a probabilistic empirical model is strongly contextual if and only if it is maximally contextual, i.e. if and only if it lies on a face of the no-signalling polytope with no local vertices. As a consequence, showing that our generalised Mermin-type arguments are $\AvN{\mathcal{R}}{G}$ is a particularly neat way of proving that they are maximally contextual, a highly desirable property for the device-independent security of the quantum-classical secret sharing protocol we will present in the next Subs.

\begin{theorem}[\textbf{Mermin-type contextuality is AvN}]\label{thm_AvNMermin}\hfill\\
Consider a $R$-probabilistic CPM category $\CategoryC$, and let $(\hbox{\input{symbols/ZbwdotSym.tex}}\!\!,\hbox{\input{symbols/DdotSym.tex}}\!\!, \mathcal{S}, \beta, N)$ be a generalised Mermin-type argument in it. If the associated empirical model is contextual, then it is $\AvN{\integers}{K}$.
\end{theorem}
\begin{proof}
The associated probabilistic empirical model is given by Equations \ref{eqn_empiricalModelControl} and \ref{eqn_empiricalModelVariations}: the only scalars appearing are $0$ and the invertible $\frac{1}{|\classicalStates{\hbox{\input{symbols/DdotSym.tex}}\!\!}|}$, which are (necessarily) sent to $0$ and $1$ respectively in the passage to the possibilistic empirical model. The possibilistic empirical model is as follows:
\begin{align}
\mathbb{P}[(g_1,...,g_N) | \text{control}] &= 
	\begin{cases}
		1 & \text{ if } g_1 \oplus ... \oplus g_N = 0 \\
		0 & \text{ otherwise }
	\end{cases}\label{eqn_empiricalModelControlBool}\\
\mathbb{P}[(g_1,...,g_N) | \text{$k^{th}$ variation for $s$}] &= 
	\begin{cases}
		1 & \text{ if } g_1 \oplus ... \oplus g_N = a^s \\
		0 & \text{ otherwise }
	\end{cases}\label{eqn_empiricalModelVariationsBool}
\end{align} 
The possibilistic empirical model has the following support subpresheaf $\supportSubpresheafSym \subseteq \sheafOfEventsSym$:
\begin{align}
	\supportSubpresheaf{\text{control}} &= \suchthat{(c^{i}_{m^0_{i}})_{i=1}^N \in K^N}{\oplus_{i=1}^N c^{i}_{m^0_{i}} =_K 0} \\
	\supportSubpresheaf{\text{$k^{th}$ variation for $s$}} &= \suchthat{(c^{i}_{m^s_{i+(k-1)}})_{i=1}^N \in K^N}{\oplus_{i=1}^N c^{i}_{m^s_{i+(k-1)}} =_K a^s} 
\end{align} 
Amongst the (many) equations in $\RlinearTheory{\integers}{\supportSubpresheafSym}$ we can find the following $1+N\cdot S$ equations:
\begin{align}
	\bigoplus_{m} s_m &= 0 \text{, satisfied by all } s \in \supportSubpresheaf{\text{control}}\\
	\bigoplus_{m} s_m &= a^s \text{, satisfied by all } s \in \supportSubpresheaf{\text{$k^{th}$ variation for $s$}} 
\end{align} 
Any global assignment satisfying all equations in $\RlinearTheory{\integers}{\supportSubpresheafSym}$ would in particular satisfy the $1+N\cdot S$ equations above, and hence provide a solution in $K$ to the system $\mathcal{S}$, as shown in the proof of Theorem \ref{thm_contextuality}. If the empirical model is contextual, then by Theorem \ref{thm_contextuality} no such solution can exist: hence there can be no global assignment satisfying all equations in $\RlinearTheory{\integers}{\supportSubpresheafSym}$, proving that the model is in particular $\AvN{\integers}{K}$.
\end{proof}

\begin{corollary} The generalised Mermin-type arguments provide an infinite family of quantum realisable $\AvN{\integers}{K}$ empirical models, indexed by all finite abelian groups $K$ and all finite consistent systems $\mathcal{S}$ of $\integers$-module equations valued in $K$ which admit no solution in $K$. Furthermore, all $\AvN{\integers}{K}$ arguments for some fixed $K$ are equivalently $\AvN{\integersMod{n}}{K}$ for any positive integer $n$ divisible by the exponent of $K$: as a consequence, there are generalised Mermin-type arguments providing quantum realisable $\AvNring{\integersMod{n}}$ models for all positive integers $n \geq 2$.
\end{corollary}
\begin{proof} 
The first part is a straightforward consequence of Theorems \ref{thm_contextuality}, \ref{thm_quantumRealisability}, \ref{thm_AvNMermin} and \ref{thm_nonCollapsingHierarchy} below. The second part is a consequence of the fact that any $\integers$-module equation valued in a finite abelian group $K$ is equivalent to a $\integersMod{\exp[K]}$-module equation (by taking remainders modulo $\exp[K]$ of all coefficients), and hence also to a $\integersMod{n}$-module equation for any $n$ divisible by the exponent $\exp[K]$ (by taking reminders modulo $n$ of all coefficients). The last part is the special case where we consider the finite abelian group $K=\integersMod{n}$ as a module over the ring $\mathcal{R} = \integersMod{n}$.
\end{proof}

One open question about All-vs-Nothing arguments asks whether all quantum realisable $\AvNring{\integers}$ models are in fact $\AvNring{\integersMod{2}}$. The following result answers the question negatively, showing that the infinite family of $\AvNring{\integers}$ models provided by the previous corollary form a non-collapsing hierarchy of $\AvNring{\integersMod{p}}$ models for all $n \geq 2$.

\begin{theorem}[\textbf{Non-collapsing AvN hierarchy over finite fields}]\label{thm_nonCollapsingHierarchy} \hfill\\
For each $n \geq 2$, there is a quantum realisable $\AvNring{\integersMod{n}}$ (and hence also $\AvN{\integers}{\integersMod{n}}$) empirical model which is not $\AvN{\integersMod{m}}{K'}$ for any $m \geq 2$ coprime with $n$ and any non-trivial abelian group $K'$ with exponent dividing $m$; in particular, it is not $\AvNring{\integersMod{m}}$.
\end{theorem}
\begin{proof}
The next Section fully works out the example of $K := \integersMod{n}$ with the system $\mathcal{S}$ consisting of a single $\integers$-module equation $t y = 1$. If we pick a $t \in \{2,...,n-1\}$ which divides $n$, the equation cannot be satisfied for $K = \integersMod{n}$, giving rise to a model which is both $\AvN{\integers}{\integersMod{n}}$ and $\AvNring{\integersMod{n}}$ (because the equation can be replaced by an equivalent $\integersMod{n}$-module equation). 
Now consider some $m$ coprime with $n$, and some abelian group $K'$ with exponent dividing $m$. Then the equation has solutions in $K'$, giving rise to a model which is not $\AvN{\integers}{K'}$ nor $\AvN{\integersMod{m}}{K'}$ (nor $\AvNring{\integersMod{m}}$, in the case $K' := \integersMod{m}$). Indeed, we must have $K' \isom \prod_{l=1}^{L} \integersMod{p_l^{e_l}}$ for some primes $p_l$ not dividing $n$ and some exponents $e_l \geq 1$, and the equation has solutions in $\integersMod{p_l^{e_l}}$ for all $l$ (because $t$ has the same prime factors of $n$, and hence no $p_l$ can divide $t$). 
\end{proof}

\subsection{A fully worked-out example}
\label{section_example}

In this Section, we fully work out a generalised Mermin-type argument, for the group $K := \integersMod{d}$ and the system $\mathcal{S}$ consisting of a single $\integers$-module equation $t y = 1$ (i.e. we have $S=M=1$), where $d \geq 2$ and $t \in \{1,...,d-1\}$. This can equivalently be seen as a $\integersMod{d}$-module equation $t y = \modclass{1}{d}$. We will go through the following stages: (i) we will present the measurement scenario and empirical model explicitly; (ii) we will characterise local hidden variable models; (iii) we will discuss the equations turning the model into an All-vs-Nothing argument; (iv) we will give a concrete realisation in terms of GHZ states and phase gates on qudits (i.e. $d$-dimensional quantum systems). 

\subsubsection{Measurement scenario.} 
Firstly, the exponent of $\integersMod{d}$ is $k := d$, and we fix a number of parties $N = 1 \mod{d}$ (e.g. $N = d+1$). Each party $i=1,...,N$ can make a measurement choice $m_i$ in the set $\{0,1\}$, and the measurement contexts take the following form. In the control, all parties make measurement choice $0$, while the variations are $N$ cyclic permutations, each one featuring $N-t$ contiguous parties making measurement choice $0$ and $t$ parties making measurement choice $1$:
\begin{equation}
\begin{array}{c|ccccccccc}
	\text{Party:} & 1 & 2 & ... & N-t-1 & N-t & N-t+1 &... & N-1 & N \\
	\hline
	\text{ control} 				& 0 & 0 & ... & 0 & 0 & 0 & ...  & 0 & 0 \\
	\text{ $1^{st}$ variation} 		& 0 & 0 & ... & 0 & 0 & 1 & ...  & 1 & 1 \\
	\text{ $2^{nd}$ variation} 		& 0 & 0 & ... & 0 & 1 & 1 & ...  & 1 & 0 \\
	\text{ $3^{rd}$ variation} 		& 0 & 0 & ... & 1 & 1 & 1 & ...  & 0 & 0 \\
	\vdots & \vdots & \vdots &  & \vdots & \vdots & \vdots & & \vdots & \vdots \\
	\text{ $N^{th}$ variation}  & 1 & 0 & ... & 0 & 0 & 0 & ...  & 1 & 1
\end{array}
\end{equation}

\subsubsection{Empirical model.}
The joint measurement outcomes $(g_1,...,g_N)$ for the $N$ parties are valued in $\integersMod{d}^N$, and the generalised Mermin-type argument is associated with the following probabilistic empirical model:
\begin{equation}\label{explicitEmpiricalModel}
\begin{array}{c||c|c|c}
	 & g_1\oplus...\oplus g_N = 0  & g_1 \oplus ... \oplus g_N = 1 & g_1 \oplus ... \oplus g_N \neq 0,1\\
	\hline
	\text{ control} 				& \frac{1}{d^{N-1}} & 0 & 0 \\
	\text{ $1^{st}$ variation} 		& 0 & \frac{1}{d^{N-1}} & 0 \\
	\text{ $2^{nd}$ variation} 		& 0 & \frac{1}{d^{N-1}} & 0 \\
	\text{ $3^{rd}$ variation} 		& 0 & \frac{1}{d^{N-1}} & 0 \\
	\vdots & \vdots & \vdots & \vdots  \\
	\text{ $N^{th}$ variation}  & 0 & \frac{1}{d^{N-1}} & 0
\end{array}
\end{equation}

\subsubsection{Local hidden variable models.}
When $t$ and $d$ are coprime, the equation $t y = \modclass{1}{d}$ has a (unique) solution $y :=  \modclass{t^{-1}}{d}$, and a local hidden variable model for the empirical model \ref{explicitEmpiricalModel} can be obtained as follows. 

Consider the set $\Lambda$ of all the $(g_1,...,g_N) \in \integersMod{d}^n$ such that $g_1\oplus...\oplus g_N = 0$, together with the uniform probability distribution $p: \Lambda \rightarrow \reals^+$ on $\Lambda$ (i.e. $p(g_1,...,g_N) = \frac{1}{d^{N-1}}$). Also, consider deterministic local outcomes for each fixed value $\underline{g} \in \Lambda$ of the hidden variable such that, upon measurement choice $m_i$ for party $i$, the measurement outcome is $g_i$ whenever $m_i = 0$ and $g_i \oplus t^{-1}$ whenever $m_i = 1$. 

In the control, all parties $i=1,...,N$ will choose $m_i = 0$, and the joint measurement outcome will be uniformly distributed over the subgroup $\Lambda \subset \integersMod{d}^N$. In any variation, $t$ parties will choose $m_i = 1$ and $N-t$ parties will choose $m_i = 0$, and the joint measurement outcome will be uniformly distributed over the coset  $(1,0,...,0) \oplus \Lambda \subset \integersMod{d}^N$ (using the fact that $t \cdot t^{-1} = 1$ in $\integersMod{d}$). Hence this really defines a local hidden variable model for the empirical model \ref{explicitEmpiricalModel} associated with the generalised Mermin-type argument.

\subsubsection{All-vs-Nothing arguments.}
When $t$ and $d$ are not coprime, the equation $t y = \modclass{1}{d}$ cannot have solutions in $K = \integersMod{d}$ (by a standard argument from number theory). The possibilistic empirical model associated with the argument has the following support subpresheaf $\supportSubpresheafSym \subseteq \sheafOfEventsSym$ (the control and the first three variations are shown here, to exemplify the pattern):
\begin{align}
	\supportSubpresheaf{\text{control}} &= \text{ the set of all } (g^{1}_{0},g^{2}_{0},...,g^{N-t-1}_{0},g^{N-t}_{0},g^{N-t+1}_{0},...,g^{N-1}_{0},g^{N}_{0}) \in \integersMod{d}^N\nonumber\\
	& \text{ such that } \bigoplus_{i=1}^{N} g^{i}_{0} = 0 \label{AvNequationControl}
\end{align}
\begin{align}
	\supportSubpresheaf{\text{$1^{st}$ var'n}} &= \text{ the set of all } (g^{1}_{0},g^{2}_{0},...,g^{N-t-1}_{0},g^{N-t}_{0},g^{N-t+1}_{1},...,g^{N-1}_{1},g^{N}_{1}) \in \integersMod{d}^N\nonumber\\
	& \text{ such that } \Big(\bigoplus_{i=1}^{N-t} g^{i}_{0}\Big) \oplus \Big(\bigoplus_{i=N-t+1}^{N} g^{i}_{1}\Big)  = 1 \label{AvNequationVar1}
\end{align}
\begin{align}
	\supportSubpresheaf{\text{$2^{nd}$ var'n}} &= \text{ the set of all } (g^{1}_{0},g^{2}_{0},...,g^{N-t-1}_{0},g^{N-t}_{1},g^{N-t+1}_{1},...,g^{N-1}_{1},g^{N}_{0}) \in \integersMod{d}^N\nonumber\\
	& \text{ such that } \Big(g^{N}_{0} \oplus \bigoplus_{i=1}^{N-t-1} g^{i}_{0}\Big) \oplus \Big(\bigoplus_{i=N-t}^{N-1} g^{i}_{1}\Big)  = 1 \label{AvNequationVar2}
\end{align}
\begin{align}
	\supportSubpresheaf{\text{$3^{rd}$ var'n}} &= \text{ the set of all } (g^{1}_{0},g^{2}_{0},...,g^{N-t-1}_{1},g^{N-t}_{1},g^{N-t+1}_{1},...,g^{N-1}_{0},g^{N}_{0}) \in \integersMod{d}^N\nonumber\\
	& \text{ such that } \Big(g^{N-1}_{0} \oplus g^{N}_{0} \oplus \bigoplus_{i=1}^{N-t-2} g^{i}_{0}\Big) \oplus \Big(\bigoplus_{i=N-t-1}^{N-2} g^{i}_{1}\Big)  = 1 \label{AvNequationVar3}
\end{align}
Amongst the (many) equations in $\RlinearTheory{\integers}{\supportSubpresheafSym}$ we can find the $N+1$ equations equations above, one for the control (Equation \ref{AvNequationControl}) and $N$ for the variations (Equations \ref{AvNequationVar1}, \ref{AvNequationVar2} and \ref{AvNequationVar3}, corresponding to the first three variations, exemplify the pattern), and any global assignment $(g^{i}_{r})^{i=1,...,N}_{r=0,1}$ which satisfies all equations in  $\RlinearTheory{\integers}{\supportSubpresheafSym}$ would in particular satisfy those $N+1$ equations. However, adding up the $N$ equations corresponding to the variations yields, after a bit of rearranging, the following equation (recall that $N= \modclass{1}{d}$):
\begin{equation}
(N-t)\Big(\bigoplus_{i=1}^{N} g^{i}_{0}\Big) \oplus t\Big(\bigoplus_{i=1}^{N} g^{i}_{1}\Big)  = N \cdot 1 = 1
\end{equation}
Taking this together with the equation $\bigoplus_{i=1}^{N} g^{i}_{0} = 0$ associated with the control then results in the following equation:
\begin{equation}
t\Big(\bigoplus_{i=1}^{N} g^{i}_{1}\Big) = 1
\end{equation}
But this means that setting $y := \bigoplus_{i=1}^{N} g^{i}_{1}$ would yield a solution to the equation $t y = 1$ in $\integersMod{d}$, which we assumed not to exist. Hence we cannot have any global assignment satisfying all equations in $\RlinearTheory{\integers}{\supportSubpresheafSym}$, and the model is both $\AvN{\integers}{\integersMod{n}}$ and $\AvNring{\integersMod{n}}$.

\subsubsection{Quantum realisation}

We now give a concrete realisation of this generalised Mermin-type argument in quantum-like theories of wavefunctions over commutative involutive semirings $S$, i.e. in $R$-probabilistic CP* categories $\CPStarCategory{\RMatCategory{S}}$ (where $R$ is the sub-semiring of positive elements in $S$). Let $(P,\cdot,1):= \suchthat{x \in S}{x^\ast x = 1}$ be the multiplicative group of phases in $S$, and let $\hbox{\input{symbols/ZbwdotSym.tex}}\!\!$ be the $\dagger$-SCFA corresponding to the standard orthonormal basis $(\ket{j})_{j\in \integersMod{d}}$ of the quantum system $S^{\integersMod{d}}$: 
\begin{equation}
\!\hbox{\input{symbols/ZbwcomultSym.tex}}\!\! := \sum_{j \in \integersMod{d}}\ket{j}\ket{j} \bra{j} \hspace{2cm} \!\hbox{\input{symbols/ZbwcounitSym.tex}}\!\! := \sum_{j \in \integersMod{d}} \bra{j}
\end{equation}
The $\hbox{\input{symbols/ZbwdotSym.tex}}\!\!$-phase states take the form $\ket{\alpha} := \sum_{j \in \integersMod{d}} \alpha_j \ket{j}$ for $\alpha_j \in P$, where without loss of generality we can set $\alpha_0 := 1$, and hence the group $(\phaseGroup{\hbox{\input{symbols/ZbwdotSym.tex}}\!\!},\!\hbox{\input{symbols/ZbwmultSym.tex}}\!\!,\!\hbox{\input{symbols/ZbwunitSym.tex}}\!\!)$ of $\hbox{\input{symbols/ZbwdotSym.tex}}\!\!$-phase gates is isomorphic to $P^{d-1}$ (we will write its elements as $(\alpha_1,...,\alpha_{d-1})$). 

\textbf{Assumption (i): the scalar $d := |\integersMod{d}|$ is invertible in $R$.} If $|\integersMod{d}|$ is invertible in $R$, then the following defines a $\dagger$-qSCFA on $S^{\integersMod{d}}$: 
\begin{equation}
\!\hbox{\input{symbols/DmultSym.tex}}\!\! := \sum_{i,j \in \integersMod{d}}\ket{i \oplus j}\bra{i} \bra{j} \hspace{2cm} \!\hbox{\input{symbols/DunitSym.tex}}\!\! := \ket{0}
\end{equation}
Then $(\hbox{\input{symbols/ZbwdotSym.tex}}\!\!,\hbox{\input{symbols/DdotSym.tex}}\!\!)$ form a strongly complementary pair, with $(\classicalStates{\hbox{\input{symbols/ZbwdotSym.tex}}\!\!},\!\hbox{\input{symbols/DmultSym.tex}}\!\!,\!\hbox{\input{symbols/DunitSym.tex}}\!\!) \isom \integersMod{d}$.

\textbf{Assumption (ii): the group $P$ contains some element $\zeta$ of order $d$.} If an element $\zeta \in P$ of order $d$ exists, we can define the following group homomorphisms $\goodchi_k: \integersMod{d} \rightarrow P$ for all $k \in \integersMod{d}$:
\begin{equation}
\goodchi_k: j \mapsto \zeta^{jk}
\end{equation}
Then these are exactly the $S$-valued multiplicative characters of $\integersMod{d}$, and correspond to the $\hbox{\input{symbols/DdotSym.tex}}\!\!$-classical states $\ket{\goodchi_k} := \sum_{j \in } \zeta^{jk} \ket{j}$. There is a group isomorphism between the group of $S$-valued multiplicative characters and the group of complex multiplicative characters (e.g. given by $\zeta \leftrightarrow e^{i \frac{2 \pi}{d}}$): as a consequence the $S$-valued multiplicative characters are enough to discriminate between elements of $\integersMod{d}$, and thus the observable $\hbox{\input{symbols/DdotSym.tex}}\!\!$ has enough classical states (and we can legitimately measure in it).

\textbf{Assumption (iii): the group $P$ contains some element $\xi$ of order $dt$, and we picked $\zeta := \xi^t$ to satisfy Assumption (ii) above.} The $\hbox{\input{symbols/ZbwdotSym.tex}}\!\!$-classical states form the subgroup $(\classicalStates{\hbox{\input{symbols/DdotSym.tex}}\!\!},\!\hbox{\input{symbols/ZbwmultSym.tex}}\!\!,\!\hbox{\input{symbols/ZbwunitSym.tex}}\!\!) \isom \integersMod{d}$ of the group of $\hbox{\input{symbols/ZbwdotSym.tex}}\!\!$-phase gates, having elements in the form $\goodchi_k \equiv (\zeta^{k},\zeta^{2k},...,\zeta^{(d-1)k})$. In this subgroup, the equation $t y = \goodchi_1 = (\zeta,\zeta^2,...,\zeta^{d-1})$ does not admit any solution, because it doesn't in $\integersMod{d}$. However, a solution $y := \beta$ exists in the larger group $(\phaseGroup{\hbox{\input{symbols/ZbwdotSym.tex}}\!\!},\!\hbox{\input{symbols/ZbwmultSym.tex}}\!\!,\!\hbox{\input{symbols/ZbwunitSym.tex}}\!\!)$ of $\hbox{\input{symbols/ZbwdotSym.tex}}\!\!$-phase states, in the form $\beta := (\xi,\xi^2,...,\xi^{d-1})$. The corresponding $\hbox{\input{symbols/ZbwdotSym.tex}}\!\!$-phase gate $P_{\beta} := \sum_{j \in \integersMod{d}} \xi^{j} \ket{j}\bra{j}$ can be then used to implement our generalised Mermin-type argument in $\CPStarCategory{\RMatCategory{S}}$. 

\noindent Now we cover some specific quantum-like theories of interest:

\begin{enumerate}

\item[(i)] In the case of ordinary quantum theory, $d$ is always invertible. We have that $P = S^1$, and we can always take $\xi:=e^{i \frac{2\pi}{dt}}$ and $\zeta := e^{i \frac{2\pi}{d}}$. 

\item[(ii) ]In the case of real quantum theory, $d$ is always invertible. However, we have that $P = \{\pm 1\}$, and hence the only argument allowed is the trivial one with $\integersMod{2}$ and the equation $1 \cdot y = 1$ (which has solution $y := 1$ in $\integersMod{2}$). 

\item[(iii)] In the case of hyperbolic quantum theory, $d$ is always invertible. However, we have $P=SO(1,1) \isom \integersMod{2} \times \reals$, and again the only argument allowed is the trivial one with $\integersMod{2}$ and the equation $1 \cdot y = 1$. Contrary to real quantum theory, non-trivial arguments would be allowed in hyperbolic quantum theory for infinite groups such as $\integers$; their implementation in the non-standard framework is left to future work.

\item[(iv)] In the case of finite-field quantum theory, the phase group takes the form $P \isom \integersMod{p^n+1}$: an element $\zeta$ of order $d$ exists if and only if $d | p^n+1$, and an element $\xi$ of order $dt$ exists if and only if $dt | p^n+1$. When this is the case, $d$ is necessarily an invertible scalar (because $d$ divides $p^n+1$, we cannot have that $p$ divides $d$). 

\item[(v)] In the case of relational quantum theory, parity quantum theory and tropical quantum theory we have $P=\{1\}$, and no value of $d$ is admissible.

\end{enumerate}

\subsection{Quantum-classical Secret Sharing}
\label{section_QSS}

In contrast to other information security protocols, classical secret sharing comes with the intrinsic assumption that some participants cannot, to some extent, be trusted. A \textit{dealer} is interested in sharing some \textit{secret} with a number of \textit{players}, with the caveat that the secret be revealed to the players only when all players agree to cooperate\footnote{More in general, a minimum number of cooperating players can be specified.}. Integrity and availability of communications is guaranteed by the existence of authenticated classical channels between dealer and players, and the protocol is only concerned with confidentiality, defined as the impossibility of recovering the secret unless all players cooperate.

The quantum-classical scheme of Hillery, Bu\v{z}ek and Berthiaume \cite{Hillery1999} introduces a new layer of security to secret sharing, employing entangled states and non-commuting observables to detect eavesdropping. The HBB scheme is based on the same measurement contexts of Mermin's original parity argument: a dealer and $N-1$ players share $N$ qubits in a GHZ state (with respect to the computational basis associated with the Pauli $Z$ observable), and randomly choose to measure their qubit in either of the mutually unbiased Pauli $X$ or Pauli $Y$ observables. It can be shown \cite{Zamdzhiev2012} that confidentiality is an immediate consequence of strong complementarity of the Pauli $Z$ and $X$ observables, while eavesdropping detection follows from mutual unbias of the Pauli $X$ and $Y$ observables. 

We extend the HBB scheme from Mermin's original parity argument to our generalised Mermin-type arguments, and we use our result on contextuality to provide a number of device-independent security guarantees. For the remainder of this section, we will consider a generalised Mermin-type argument $(\hbox{\input{symbols/ZbwdotSym.tex}}\!\!,\hbox{\input{symbols/DdotSym.tex}}\!\!, \mathcal{S}, \beta, N)$, on an object $\SpaceH$ of a $R$-probabilistic CP* category.

Consider a \textbf{dealer}, call her Alice, who wishes to share a \textbf{secret} with $N'$ \textit{players}, where $2 \leq N' < N$. As the owner of the secret, Alice is always a trusted party, the \textit{only} trusted party in the protocol. The secret is assumed to take the form of a string of elements of $\classicalStates{\hbox{\input{symbols/DdotSym.tex}}\!\!}$, the \textbf{plaintext} (at most one element of $\classicalStates{\hbox{\input{symbols/DdotSym.tex}}\!\!}$, the \textbf{round plaintext}, transmitted for each round of the protocol). We wish to ensure that the plaintext can be decoded from the information Alice sends, the \textbf{cyphertext}, if and only if all players agree to cooperate (by which we mean that they all reveal their secret keys to some party in possession of the cyphertext). Alice and the players are given $N$ devices (one per player, and $N-N'$ for Alice): at each round $w$, each device $B_j$ is fed an \textbf{input} $m_j^w \in \{0,1,...,M\}$ and returns an \textbf{output} $g_j^w \in \classicalStates{\hbox{\input{symbols/DdotSym.tex}}\!\!}$ (we also refer to the outputs $g_1^w,...,g_{N'}^w$ as the \textbf{secret keys} of the players for round $w$).
We furthermore assume the following \textbf{security conditions} to hold. 
\begin{enumerate}
	\item[(i)] Alice and the players share an authenticated classical channel, ensuring integrity and availability of all classical communications involved in the protocol.
	\item[(iia)] Alice and the players are in possession of $N$ secure independent classical sources of randomness, to generate independent inputs at each round which are uniformly distributed in $\{0,1,...,M\}$.
	\item[(iib)] Alice is in possession of a secure classical source of randomness, independent from all other, to decide which rounds will be \textbf{secret rounds} (with probability $(1-\tau)>0$) and which rounds will be \textbf{test rounds} (with probability $\tau > 0$).
	\item[(iii)] During step 2 of the protocol below, no signalling is possible between distinct parties/devices\footnote{This can be achieved, for example, by ensuring the devices are operated in conditions controlled by Alice (trusted laboratories, synchronized time-stamp servers, etc).}.
	\item[(iv)] We will assume that in step 3 Alice is communicated the measurement choices faithfully\footnote{This can be achieved by entrusting the laboratory setup with the communication of the random measurement choices to Alice, the player and the device.}.
\end{enumerate}
Because tampering can only be determined after the protocol has ended and the entirety (or an otherwise significant portion) of the plaintext has been transmitted, we distinguish between the \textbf{plaintext}, the data that can be decoded using the secret keys, and the actual \textbf{secret} that Alice wants the players to share. Before the protocol begins, Alice will obtain the plaintext by encrypting the secret with a secure symmetric encryption protocol\footnote{If the secret is in the form of a string of elements of $\classicalStates{\hbox{\input{symbols/DdotSym.tex}}\!\!}$, the natural choice for this protocol, then the plaintext can be obtained by generating a string of uniformly random $k^w$ elements of $\classicalStates{\hbox{\input{symbols/DdotSym.tex}}\!\!}$, obtaining the round plaintext $p^w$ from the corresponding \inlineQuote{round secret} $q^w$ as $p^w = q^w \oplus k^w$. Once the string of random elements is broadcast, upon successful completion of the protocol, the secret can be recovered from the decoded plaintext as $q^w = p^w \ominus k^w$.}, using a freshly generated ephemeral key which she will broadcast only if the protocol is successful. If the protocol fails, the random key will not be broadcast and the secret will be unrecoverable even if the plaintext is decoded. 

\begin{figure}
	\input{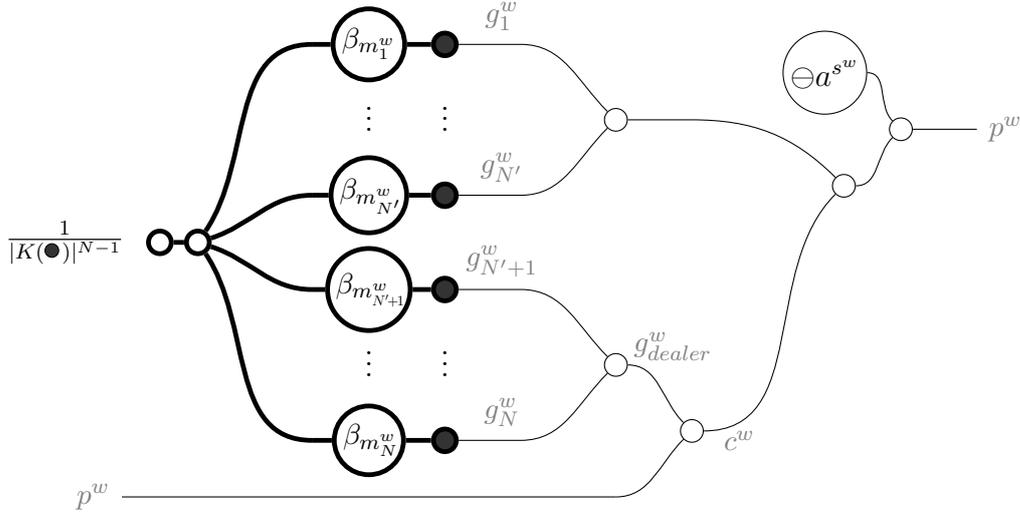}
	\caption{Graphical presentation of a noiseless, trusted implementation. 
	\label{fig_trustedProtocol}}
\end{figure}

The quantum-classical secret sharing protocol then proceeds as follows for each round $w=1,...,W$, until the entire secret has been transmitted. An individual round for a noiseless, trusted implementation is presented in Figure \ref{fig_trustedProtocol}. Throughout the protocol, Alice keeps a count of occurrences of joint outputs $g_1,...,g_N$ conditional to each joint input $m_1,...,m_N$ that she observes in test rounds.
\begin{enumerate}
	\item[1.] Alice and the players share $N$ subsystems of a state $\rho$: each player has an individual subsystem and Alice keeps the remaining $N-N'$ subsystems. In a noiseless, trusted implementation, $\rho$ is the $N$-partite $\hbox{\input{symbols/ZbwdotSym.tex}}\!\!$-GHZ state. For the purposes of a device-independent security analysis, $\rho$ can be potentially any state (pure or mixed).
	\item[2.] Alice and the players each sample their classical source of randomness and obtain inputs $m_1^w,...,m_N^w$ which are passed to the devices $B_1,...,B_N$ and result in outputs $g_1^w,...,g_{N'}^w \in \classicalStates{\hbox{\input{symbols/DdotSym.tex}}\!\!}$ for the players (the secret keys for the round) and $g_{N'+1}^w,...,g_{N}^w \in \classicalStates{\hbox{\input{symbols/DdotSym.tex}}\!\!}$ for Alice. In a noiseless, trusted implementation, $B_j$ with input $m_j^w$ applies the phase gate $\phasegate{\beta_{m_j^w}}$ to the subsystem $j$ and then measures it in the $\hbox{\input{symbols/DdotSym.tex}}\!\!$ observable.
	\item[3.] The inputs for the players are communicated to Alice. She checks that $m_1^w,...,m_N^w$ define a valid \textbf{measurement context} (either the control ($s=0$) or a variation for some $s=1,...,S$).
	\item[4.] Alice samples her source of randomness to decide whether the round will be a test round or a secret round.
	\item[4a.] If the round is a test round, Alice requests all players to communicate their secret keys, and she increases the occurrence count for joint output $(g_1^w,...,g_N^w)$ conditional to joint input $(m_1^w,...,m_N^w)$.
	\item[4b.] If the round is a secret round, Alice computes $g_{dealer}^w := \bigoplus_{j=N'+1}^N g_j^w$ and broadcasts the \textbf{round ciphertext} $c^w:= p^w \oplus g_{dealer}^w$ to the players, where the \textbf{round plaintext} $p^w$ is the next element of the plaintext to be sent. She also broadcasts the relevant value $s^w \in \{0,1,...,S\}$ obtained from the joint inputs $m_1^w,...,m_N^w$.
	\item[5.] Anyone in possession of $s^w$, the round ciphertext $c^w$, and all secret keys $g_1^w,...,g_{N'}^w$ can obtain the round plaintext $p^w$ by computing $p^w = (c^w\, \oplus g_1^w \oplus ... \oplus g_{N'}^w)\ominus a^{s^w}$, where $s^w$ is the value broadcast in Step 3. 
\end{enumerate}
The chosen generalised Mermin-type argument determines the following \textbf{promised conditional distribution} $\mathbb{P}_{promised}\big[\,\underline{g}\,\big\vert\, \underline{m}\,\big]$, the one which Alice and the players expect to observe (asymptotically) in a trusted noiseless implementation (we use the more compact notation $\underline{g} := (g_1,...,g_N)$ for the joint output and $\underline{m} := (m_1,...,m_N)$ for the joint input): 
\begin{equation}\label{Ppromised}
\mathbb{P}_{promised}\big[\,\underline{g}\,\big\vert\, \underline{m}\,\big] = 
\begin{cases}
	\frac{1}{|\classicalStates{\hbox{\input{symbols/DdotSym.tex}}\!\!}|^{N-1}} &\text{ if } g_1 \oplus ... \oplus g_N = \beta_{m_1} \oplus ... \oplus \beta_{m_N} \\
	0 &\text{ otherwise}
\end{cases}
\end{equation}
At the end of the protocol, Alice normalises her joint output counts for each joint input to obtain the \textbf{observed conditional distribution} $\mathbb{P}_{observed}\big[\,\underline{g}\,\big\vert\, \underline{m}\,\big]$ (which need not be no-signalling). She then computes the \textbf{noise parameter} $\epsilon$ as follows: 
\begin{equation}
\epsilon := 1 - |\classicalStates{\hbox{\input{symbols/DdotSym.tex}}\!\!}|^{N-1} \min \Big\{ \mathbb{P}_{observed}\big[\,\underline{g}\,\big\vert\, \underline{m}\,\big] \Big\vert g_1 \oplus ... \oplus g_N = \beta_{m_1} \oplus ... \oplus \beta_{m_N} \Big\} 
\end{equation}
The error parameter as defined above is the smallest $\epsilon \in [0,1]$ such that the observed conditional distribution can be decomposed as the following convex combination of  promised conditional distribution and some \textbf{noise conditional distribution} $\mathbb{P}_{noise}\big[\,\underline{g}\,\big\vert\, \underline{m}\,\big]$:
\begin{equation}
\mathbb{P}_{observed}\big[\,\underline{g}\,\big\vert\, \underline{m}\,\big] = (1 - \epsilon) \; \mathbb{P}_{promised}\big[\,\underline{g}\,\big\vert\, \underline{m}\,\big] + \epsilon \;\mathbb{P}_{noise}\big[\,\underline{g}\,\big\vert\, \underline{m}\,\big]
\end{equation} 
Before a run of the protocol begins, Alice sets a maximum $\epsilon_{max}$ that she is going to accept for the noise parameter. Alice chooses as low an $\epsilon_{max}$ as possible compatibly with the specifications of the device provider (and any other beliefs she might have) on the amount of noise she should expect from the devices and states in the absence of any tampering from Eve. At the end of the protocol run, Alice compares the noise parameter $\epsilon$ she computed with the maximum $\epsilon_{max}$ she decided to accept: if $\epsilon \leq \epsilon_{max}$, she declares the protocol run a success and broadcasts the ephemeral key she used to encode the secret into the plaintext; if $\epsilon > \epsilon_{max}$, she declares the protocol run a failure and she destroys the ephemeral key, rendering the secret unrecoverable even if the plaintext is at some point obtained by the players or by Eve.

The HBB quantum-classical secret sharing protocol comes with two security guarantees: (i) ignorance about any one secret key for a round denies knowledge about the plaintext for that round; (ii) successful, undetected eavesdropping has low probability. It can be shown \cite{Zamdzhiev2012} that in a noiseless and trusted implementation the first guarantee follows abstractly from strong complementarity of the Pauli $Z$ and $X$ observables, and the proof straightforwardly transfers to the strongly complementary pairs $(\hbox{\input{symbols/ZbwdotSym.tex}}\!\!,\hbox{\input{symbols/DdotSym.tex}}\!\!)$ appearing in our generalised protocol. Instead of treating eavesdropping directly, we will present a more general, device-independent proof of security, based solely on contextuality of the generalised Mermin-type argument used by the protocol. 

Works on device-independent security (such as \cite{Barrett2005,Vazirani2014} on quantum key distribution) usually posit Eve to be an adversary who can arbitrarily tamper with the shared state and measurement devices, and is only bound in her attempts by the physical theory under consideration\footnote{Eve is often assumed to be bound by the laws of quantum theory, but sometimes super-quantum attackers are also considered, bound only by causality and no-signalling.} and by the security conditions explicitly enforced by the protocol (including no-signalling). Examples of things that the Eve can to do include:
\begin{enumerate}[(i)]
\item the measurement outcomes broadcast at a test round can reveal to Eve information about measurement outcomes in previous secret rounds;
\item Eve can keep a subsystem of the shared state to herself, which she can optimally measure, once all inputs and test round outputs have been broadcast, to obtain information about the secret keys.
\end{enumerate}
Our choice of a device-independent setting comes from the more modest desire to show that the security guarantees follow from contextuality of the generalised Mermin-type argument, regardless of the specific implementation; as a consequence, we will be content with a more restricted model of attack. We assume that Alice and the players might be provided with noisy or imperfect states and devices, which might give Eve a variety of security loopholes to exploit. However, we assume that the device provider shows no malice: 
\begin{enumerate}[(i)]
\item the devices are memoryless and operate independently at each round;
\item the states used at different rounds are independent and identical; 
\item the states are not entangled with any additional system.
\end{enumerate}
However, Eve might possess classical information about the states which is unavailable to the players (such as information leaked through noise or side channels, information acquired via eavesdropping, etc).

\begin{figure}
	\input{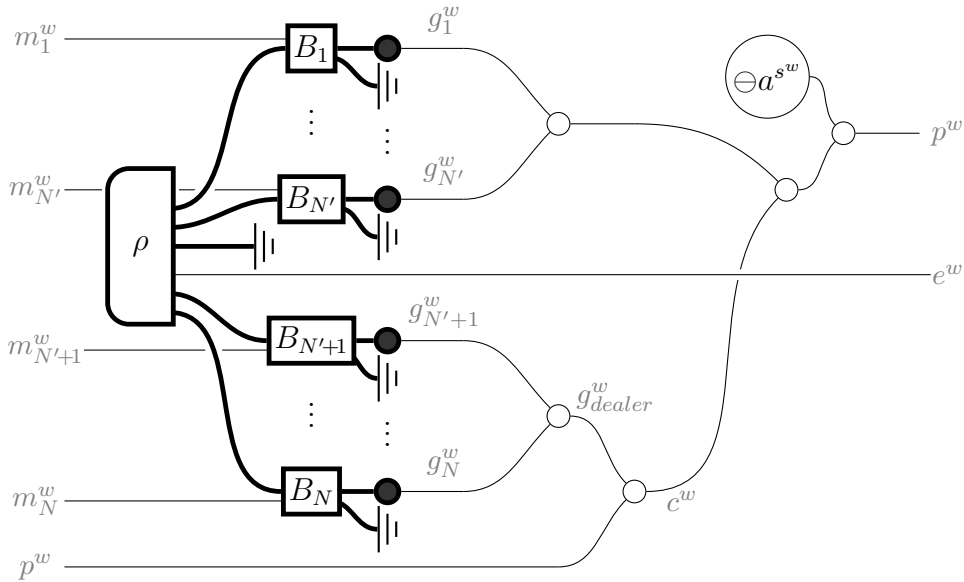}
	\caption{Graphical presentation of a generic, untrusted implementation at a single round of the protocol. Eve might have some classical information $e^w$ about the states which Alice and the players don't know. The classical side of the protocol is entirely in the hands of Alice and the players, and proceeds as in the trusted noiseless case.
	\label{fig_untrustedProtocol}}
\end{figure}

Although not fully general, this setup subsumes a variety of more specialised security scenarios that are of interest in classical and quantum cryptography:
\begin{enumerate}[(i)]
\item Real-world implementations are unavoidably noisy, and one should consider any noise as a potential source of cryptophthora\footnote{Secret degradation, usually due to side-channel leakage.}. Our setup allows for the possibility that both the shared state and the measurement devices be noisy, with no dependence on a specific model of noise; it also allows for the possibility that what looks like random noise to Alice and the players might actually carry side-channel information to Eve.
\item Eavesdropping detection is a typical desideratum in quantum cryptography, where Eve intercepts the local state of a player\footnote{In our secret sharing protocol, a single player's secret key is all that Eve needs to break confidentiality, as we may freely assume that the remaining players are colluding with Eve.}, measures it in some basis to obtain classical information, and forwards the resulting collapsed state to the player. Our setup allows for the possibility of eavesdropping\footnote{However, it does not cover a more advanced attack in which Eve sends through a subsystem of an entangled state, keeping the rest of the state to herself and measuring it in the future to obtain more information about the player's outcome.}: the classical information that Eve possesses about the state can be used to model the information she acquired by eavesdropping. Our security proof then has eavesdropping detection as a special case of protocol failure.
\end{enumerate} 

\noindent Figure \ref{fig_untrustedProtocol} displays a single round $w$ of the protocol in a generic, untrusted implementation. An $N$-partite state $\rho$ is shared between Alice and the players at a given round of the protocol, with no additional subsystem accessible to Eve (who might however be in possession of classical information $e^w$ about it). The measurement devices $B_1,...,B_N$ operate independently at each round, with no memory or shared resource other than the state $\rho$. At each round $w$, device $B_j$ takes measurement choice $m_j^w$ as a classical input and returns measurement outcome $g_j^w$ as a classical output. The rest of the protocol is entirely in the hands of Alice and the players, and proceeds as in the trusted noiseless case.

Our first result shows that lack of contextuality implies the existence of a scenario in which a perfect undetectable attack may take place. In fact, the scenario is not particularly remote: it might well happen happen that the device provider inadvertently chose phase states $\beta_1,...,\beta_M$ which happen to be $\hbox{\input{symbols/DdotSym.tex}}\!\!$-classical states (maybe she did not notice, maybe she was tricked by Eve into choosing them), and that the GHZ state decoheres (spontaneously or with a malicious helping hand) in the $\hbox{\input{symbols/DdotSym.tex}}\!\!$ observable. In that case, Alice and the players will notice nothing wrong with their protocol, and Eve will obtain the entirety of the secret all by herself.

\begin{theorem}[\textbf{Perfect undetectable attack}]\hfill\\
Consider a quantum-classical secret sharing protocol based on a generalised Mermin-type argument $(\hbox{\input{symbols/ZbwdotSym.tex}}\!\!,\hbox{\input{symbols/DdotSym.tex}}\!\!, \mathcal{S}, \beta, N)$, in a $R$-probabilistic CP* category with a positive semiring $R$ of scalars. If the associated empirical model is non-contextual, then there is a shared state $\rho$ and measurement devices $B_1,...,B_N$ such that test rounds will succeed with certainty, and Eve will always know all the secret keys.
\end{theorem}
\begin{proof}
By Theorem \ref{thm_contextuality}, if the empirical model is non-contextual then there exists a solution $(y_r := b_r)_{r=1}^M$ in $\classicalStates{\hbox{\input{symbols/DdotSym.tex}}\!\!}$ to the system $\mathcal{S}$ (which we take to be in the form of System \ref{eqn_system}). For each round $w$, Eve samples a random variable uniformly distributed over the following set:
\begin{equation}
\suchthat{(h_1^w,...,h_N^w) \in \classicalStates{\hbox{\input{symbols/DdotSym.tex}}\!\!}^N}{h_1^w \oplus ... \oplus h_N^w = 0}
\end{equation}
Now assume that the separable pure state $\ket{h_1^w} \otimes ... \otimes \ket{h_N^w}$ is given in input to the measurement devices $B_1,...,B_N$ at round $w$, and that the devices are designed so that $B_j$ returns $g_j^w := h_j^w \oplus b_{m_j^w}$ upon measurement choice $m_j^w$ (i.e. applies a phase $b_{m_j^w}$ which happens to be $\hbox{\input{symbols/DdotSym.tex}}\!\!$-classical). The state seen by Alice and the players is following round-independent mixed state $\rho$, but Eve at each round has additional information $e^w$ which helps her identify which pure component of $\rho$ will actually be sent to the parties at that specific round:
\begin{equation}
\rho := \sum_{h_1 \oplus ... \oplus h_N = 0} \frac{1}{|\classicalStates{\hbox{\input{symbols/DdotSym.tex}}\!\!}|^{N-1}} \ket{h_1}\bra{h_1} \otimes ... \otimes \ket{h_N}\bra{h_1}  
\end{equation}
Once the measurement $(m_j^w)_{j=1}^{N'}$ choices for the players are broadcast, Eve can compute all the secret keys $(g_j^w)_{j=1}^{N'}$. Furthermore, since $(b_r)_{r=1}^{M}$ is a solution to $\mathcal{S}$, the measurement outcomes obtained from this setup will have the same distribution as the ones from a noiseless trusted implementation, and all test rounds will succeed with certainty.
\end{proof}

Our second result is restricted to probabilistic theories, i.e. distributively $\CMonCategory$-enriched CPM categories having $\reals^{+}$ as their semiring of scalars. Consider the no-signalling polytope associated with the measurement scenario of a contextual generalised Mermin-type argument $(\hbox{\input{symbols/ZbwdotSym.tex}}\!\!,\hbox{\input{symbols/DdotSym.tex}}\!\!, \mathcal{S}, \beta, N)$, and let $F$ be the face of the polytope specified by the support of the empirical model (the one defined by Equation \ref{Ppromised}). For each vertex $v \in F$ of that face, corresponding to empirical model $\mathbb{P}_{v}\big[\,\underline{g}\,\big\vert\,\underline{m}\,\big]$, let $H_v$ be the average entropy across all measurement contexts:
\begin{equation}
H_v := \frac{1}{1+N \cdot S} \sum_{\underline{m} \in \mathcal{M}} H\Big[\, \mathbb{P}_{v}\big[\,\emptyArg\,\big\vert\,\underline{m}\,\big] \,\Big]
\end{equation}
Let $H_{promised}^{(min)} := \min_{v \in F} H_v$ be the minimum average entropy across all vertices of the face: because the generalised Mermin-type argument is strongly contextual, the face cannot contain any local vertices, and hence the minimum average entropy $H_{promised}^{(min)}$ is always strictly positive; a tighter estimation of this quantity is left to future work. Call $\eta := \Big(1-\frac{H_{promised}^{(min)}}{|\classicalStates{\hbox{\input{symbols/DdotSym.tex}}\!\!}|^{N-1}}\Big) \in [0,1)$ the \textbf{information leakage fraction} for the face: it is the maximum fraction of plaintexts that Eve can expect to decipher when the empirical model she sees lies on face $F$.

We will now show that protocols based on contextual generalised Mermin-type arguments always provide a certain amount of security: for observed noise parameter $\epsilon$ small enough, the maximum expected fraction of plaintexts that Eve can expect to decipher is sharply peaked somewhere between $\eta$ and $c \cdot \epsilon$, where $c$ is some constant depending on the geometry of the no-signalling polytope. In one extreme, we may have $\eta = 0$, i.e. all empirical model on the face carry the same maximal amount of entropy. In this case, Eve's chances of learning some parts of the secret rely entirely on the noise parameter $\epsilon$: in her best case scenario, she observes a deterministic empirical model for some fraction $\epsilon$ of rounds, in which case she can gain complete knowledge about the round plaintext. In the other extreme, we have $\eta \gg \epsilon$, i.e. there are empirical models on the face $F$ which might lead to more leakage of plaintext information than any number of deterministic model which might be lurking in the noise $\epsilon$. In this case, Eve's best bet might just be to exploit the empirical models on the face $F$ itself.

\begin{theorem}[\textbf{Device-independent security}]\hfill\\
\label{thm_HBBdisec}
Consider a quantum-classical secret sharing protocol based on a generalised Mermin-type argument $(\hbox{\input{symbols/ZbwdotSym.tex}}\!\!,\hbox{\input{symbols/DdotSym.tex}}\!\!, \mathcal{S}, \beta, N)$, in a probabilistic CP* category $\CPStarCategory{\CategoryC}$ (with $\reals^{+}$ as its positive semiring of scalars). Consider a run of the protocol with a large number $W$ of rounds, of which $P$ secret rounds and $T$ test rounds (with $P \rightarrow (1-\tau)W$ and $T \rightarrow \tau W$ almost certainly as $W \rightarrow \infty$). Let $\epsilon$ be the noise parameter observed by Alice at the end (a random variable), and let $P_{Eve}$ be maximum number of round plaintexts that Eve expects to successfully decipher (another random variable). Then the maximum fraction of plaintexts $P_{Eve}/P$ that Eve expects to successfully decipher is sharply peaked around some value between $\eta$ and $O(\epsilon)$, with variance bounded above by $O(\frac{\tau(1-\tau)}{W})$ almost certainly for $W \rightarrow \infty$ (where the big-$O$ notation hides a constant depending on the geometry of the polytope alone).
\end{theorem}
\begin{proof}
As part of this proof, a number of different conditional distributions will be considered: 
\begin{enumerate}[(i)]
\item the no-signalling conditional distribution $\mathbb{P}_{true}(e)\big[\,\underline{g}\,\big\vert\, \underline{m}\,\big]$ determined by $\rho$ and the devices $B_1,...,B_N$ conditional to Eve obtaining information $e$ (this is the conditional distribution as seen from Eve's vantage point);
\item the no-signalling conditional distribution $\mathbb{P}_{true}\big[\,\underline{g}\,\big\vert\, \underline{m}\,\big] := \sum_{e}\mathbb{P}[e] \cdot \mathbb{P}_{true}(e)\big[\,\underline{g}\,\big\vert\, \underline{m}\,\big]$ determined by $\rho$ and the devices $B_1,...,B_N$, averaged over Eve's information (this is the \textbf{true conditional distribution} as seen from Alice's vantage point, which her tests will estimate);
\item the no-signalling conditional distribution $\mathbb{P}_{promised}\big[\,\underline{g}\,\big\vert\, \underline{m}\,\big]$ derived from the generalised Mermin-type argument (this is what Alice would expect to estimate in the absence of any noise or tampering);
\item the conditional distribution $\mathbb{P}_{observed}\big[\,\underline{g}\,\big\vert\, \underline{m}\,\big]$ estimated by Alice.
\end{enumerate}
Alice's estimate of the true conditional distribution $\mathbb{P}_{true}\big[\,\underline{g}\,\big\vert\, \underline{m}\,\big]$ can be modelled by considering the vector-valued random variables $\underline{X}^w := \big(X_{(\underline{g},\underline{m})}^w\big)$ for all test rounds $w$, where $X_{(\underline{g},\underline{m})}^w$ is the real-valued random variable defined as follows (note that $\underline{g}^w$ is a random element of $\classicalStates{\hbox{\input{symbols/DdotSym.tex}}\!\!}^N$, and $\underline{m}^w$ is a uniformly random element of the set of $1+NS$ measurement contexts):
\begin{equation}
X_{(\underline{g},\underline{m})}^w = 
\begin{cases}
1 & \text{ if } \underline{g} = \underline{g}^w \text{ and } \underline{m} = \underline{m}^w\\
0 & \text{ otherwise}
\end{cases}
\end{equation} 
The vector $\underline{X}^w$ takes the value $1$ over the joint input/joint output pair  recorded by Alice for round $w$, and $0$ everywhere else: Alice's estimate of the true conditional distribution is then obtained from the average random variable $\frac{1}{T}\!\!\sum\limits_{w \text{ test}}\!\! \underline{X}^w$. By the central limit theorem, Alice's estimate $\mathbb{P}_{observed}\big[\,\underline{g}\,\big\vert\, \underline{m}\,\big]$ will be normally distributed around the true conditional distribution, with variance $O(\frac{1}{T})$; because the noise parameter $\epsilon$ observed by Alice is obtained from this estimate, it will similarly be distributed around the true noise parameter $\epsilon_{true}$ defined below, with variance bounded above by $O(\frac{1}{T})$ (almost certainly for $T \rightarrow \infty$).

We define the \textbf{true noise parameter} $\epsilon_{true}$ to be obtained from the conditional distribution $\mathbb{P}_{true}\big[\,\underline{g}\,\big\vert\, \underline{m}\,\big]$ in the same way that $\epsilon$ is obtained from the conditional distribution $\mathbb{P}_{observed}\big[\,\underline{g}\,\big\vert\, \underline{m}\,\big]$. This means $\epsilon_{true}$ is the largest such that $\mathbb{P}_{true}\big[\,\underline{g}\,\big\vert\, \underline{m}\,\big]$ decomposes as follows, for some conditional distribution $\mathbb{P}_{true,noise}\big[\,\underline{g}\,\big\vert\, \underline{m}\,\big]$:
\begin{equation}
(1-\epsilon_{true}) \, \mathbb{P}_{promised}\big[\,\underline{g}\,\big\vert\, \underline{m}\,\big] + (\epsilon_{true}) \mathbb{P}_{true,noise}\big[\,\underline{g}\,\big\vert\, \underline{m}\,\big]
\end{equation}
For each value $e \in E$ that Eve's information can take, we define the parameter $\xi(e) \in [0,1]$ to be the smallest possible such that the conditional distribution $\mathbb{P}_{true}(e)\big[\,\underline{g}\,\big\vert\, \underline{m}\,\big]$ decomposes as follows:
\begin{equation}
(1-\xi(e))\mathbb{P}_{F}(e)\big[\,\underline{g}\,\big\vert\, \underline{m}\,\big] + \xi(e) \mathbb{P}_{F,noise}(e)\big[\,\underline{g}\,\big\vert\, \underline{m}\,\big]
\end{equation}
for some distribution $\mathbb{P}_{F}(e)\big[\,\underline{g}\,\big\vert\, \underline{m}\,\big]$ lying on the face $F$ and some distribution $\mathbb{P}_{F,noise}(e)\big[\,\underline{g}\,\big\vert\, \underline{m}\,\big]$ lying outside of face $F$. To Eve, in possession of information $e$, the conditional distribution $\mathbb{P}_{true}(e)\big[\,\underline{g}\,\big\vert\, \underline{m}\,\big]$ looks like a biased coin deciding between the two following scenarios:
\begin{enumerate}
\item[(a)] with probability $(1-\xi(e))$, she observes a distribution $\mathbb{P}_{F}(e)\big[\,\underline{g}\,\big\vert\, \underline{m}\,\big]$ lying on face $F$, which means that the fraction of the round plaintext that she expects to learn is bounded above by $\eta$;  
\item[(b)] with probability $\xi(e)$, she observes some other distribution $\mathbb{P}_{F,noise}(e)\big[\,\underline{g}\,\big\vert\, \underline{m}\,\big]$, which in the best case scenario could give her full knowledge of the round plaintext.
\end{enumerate}
Because marginalising over Eve's knowledge\footnote{I.e. taking the convex combination of the conditional distributions $\mathbb{P}_{true}(e)\big[\,\underline{g}\,\big\vert\, \underline{m}\,\big]$ with respect to the probability distribution $\mathbb{P}[e]$ of Eve's side-channel information.} must result in the distribution $\mathbb{P}_{true}\big[\,\underline{g}\,\big\vert\, \underline{m}\,\big]$, the geometry of the polytope implies that the convex combination $\sum_e \mathbb{P}[e] \xi(e)$ must go to zero as $O(\epsilon_{true})$ (i.e. there must be some constant $c > 0$ such that $\sum_e \mathbb{P}[e] \xi(e) \leq c \cdot \epsilon_{true}$).

It should be noted that the information $e$ obtained by Eve is random to Eve herself: sometimes she will obtain information giving her better guessing probability, sometimes she will obtain information giving her worse guessing probability.  When the distribution of $e$ is taken into account, the fraction of round plaintexts that Eve can expect to decipher is bounded above by the following value, falling somewhere between $\eta$ and $O(\epsilon_{true})$:
\begin{equation}
\sum_{e} \mathbb{P}[e]\,\Big((1-\xi(e)) \eta + \xi(e)\Big) 
\end{equation}
Again by central limit theorem, the maximum fraction $P_{Eve}/P$ of round plaintexts that Eve expects to successfully decipher is normally distributed around the value above, with variance $O(\frac{1}{P})$ (almost certainly for $P \rightarrow \infty$).

Finally, because $P_{Eve}/P$ is sharply peaked around some value between $\eta$ and $O(\epsilon_{true})$, with variance $O(\frac{1}{P})$, and because $\epsilon$ is sharply peaked around $\epsilon_{true}$, with variance bounded above by $O(\frac{1}{T})$, we can conclude that $P_{Eve}/P$ is sharply peaked around some value between $\eta$ and $O(\epsilon)$ , with variance bounded above by $O(\frac{1}{T}+\frac{1}{P})$ (which tends to $O(\frac{1}{\tau (1-\tau) W})$ almost certainly as $W \rightarrow \infty$).
\end{proof}







\chapter*{Conclusions and future work}  
\label{section_conclusions}

\section*{Categorical Quantum Dynamics}

Throughout Chapter \ref{chapter_CQD}, we have seen how strong complementarity can be used to provide a compelling abstract description of the fundamental structural and operational features of quantum symmetries and dynamics. 

We have started our journey from the familiar case of wavefunctions on periodic lattices, where we have identified the potential for strong complementarity to provide an abstract description of the relationship between the position and momentum observables. In line with our proposed coherent approach to group theory and quantum symmetries, we have defined a new notion of quantum group. Having proven a minimal set of result relating quantum groups to their classical counterparts, we have gone back to wavefunctions on periodic lattices, and we have embarked on a quest to prove that the strongly complementary observables of a quantum group truly model a sensible notion of position-momentum duality; we have shown that momentum eigenstates generate the translation symmetry, and dually that position eigenstates generate the boost symmetry; we have shown that the bialgebra law yields the Weyl form of the Canonical Commutation Relations; we have shown that putative position-momentum pair satisfies a suitably weak version of the uncertainty principle. Although narrated through the lens of periodic lattices, the results we obtained are fully general, and apply to all quantum groups.

Satisfied with our description of quantum groups as position-momentum pairs, we have shifted our attention towards more general symmetric systems. We have defined a notion of unitary representations for quantum groups, as the coherent counterparts of unitary symmetries for classical groups. Just like a classical group can be though of as a physical system exerting classical control over the symmetric system, a quantum group can be though of as a physical system exerting coherent control. We have characterised representations of quantum groups categorically as the algebras in the Eilenberg-Moore category for a certain monad, with equivariant maps as Eilenberg-Moore morphisms. We have extended our results on symmetry-observable duality to unitary representations, and we have provided a suitable reformulation of Stone's Theorem to match them.

In order to treat the textbook case of 1-dimensional wavefunctions with periodic boundary conditions, we have introduced a new approach to infinite-dimensional separable Hilbert spaces based on non-standard analysis. Contrary to previous approaches, the category $\starHilbCategory$ of separable Hilbert spaces we introduced is compact closed, and has unital $\dagger$-Frobenius algebras. We have then proceeded to construct a doubly well-pointed quantum group corresponding to the position-momentum pair for 1-dimensional wavefunctions with periodic boundary conditions. We have also remarked that our methods extend to all compact and discrete abelian groups.

In the final section of the Chapter, we have applied the tools developed in the remainder of the chapter to the coherent treatment of quantum dynamics. Armed with all the necessary results, we have quickly ploughed through quantum clocks and dynamical systems, we have identified a suitable coherent Hamiltonian, and we have shown that Schr\"{o}dinger's Equation corresponds exactly to the defining equation for Eilenberg-Moore algebras. Using our previous results on symmetry-observable duality, we have provided simple diagrammatic proof for Stone's Theorem on 1-parameter unitary groups and von Neumann's Mean Ergodic Theorem, in the case of discrete periodic, discrete and continuous periodic dynamics. We have provided an abstract characterisation of the Feynman clock construction, and proven its validity in our framework (for arbitrary quantum groups). Finally, we have tackled the issue of synchronisation of dynamical systems, provided conditions for the existence of internal time observables, and proven sufficient conditions for the emergence of quantum clocks amongst synchronised systems. 

Chapter \ref{chapter_CQD} sure contains a lot of material, but a lot of work remains to be done. To begin with, we don't have a satisfactory characterisation of non-well-pointed quantum groups in $\fdHilbCategory$, other than \inlineQuote{they sort of look like other definitions of quantum groups}. A structural theorem, akin to the one for well-pointed quantum groups, would make for a rounder picture, especially in connection with non-commutative geometry.

As far as the characterisation of quantum groups as position-momentum pairs is concerned, the desirable results are all there, with the possible exception of the uncertainty principle. While it is true that the full uncertainty principle is undesirably strong, the version we have proven might be seen as excessively weak, and a middle ground could perhaps be reached. 

The state of symmetry-observable duality for general symmetric systems is also pretty satisfactory, but their categorical characterisation as Eilenberg-Moore algebras is open territory. Some additional results on the monadic approach to dynamics has been obtained in \cite{Gogioso2015c}, but have not yet been adapted to the quantum group framework presented in this work. 

Infinite-dimensional categorical quantum mechanics is perhaps the youngest addition here, and certainly requires more work and thought. While the techniques we exemplified extend straightforwardly to other compact and discrete abelian groups, it would greatly benefit the have a number of other examples of interest fully worked out. Extensions of the framework to locally compact symmetries and quantum field theory are currently in the making.

Finally, three main avenues of research are currently open in the applications to dynamics. Firstly, one would like to extend the results to the real-world case of continuous dynamics, governed by the symmetry group $(\reals,+,0)$. The main challenge, the derivation of a suitable coherent group, has already been solved in recent work, so this is mostly a matter of adapting the results where necessary, and reap the rewards. Secondly, our results on internal time observable have already answered some questions in the context of time observables in quantum theory, and we expect that techniques and ideas derived from them will provide a significant contribution to the debate in the near future. Thirdly, the very last results in the chapter point towards the possibility of formulating a toy model for emergent time in quantum theory based solely on hierarchies of mutually synchronised discrete periodic quantum clocks: a brief argument in favour of this construction has already been sketched, but the full development of such a model is left to future work.

\section*{Hidden Subgroup Problem}

The abelian Hidden Subgroup Problem comprises many of the problems successfully tackled by quantum algorithms as special instances, but the traditional presentation of the quantum solution is too heavily algebraic to clearly show the key structures at work. In Chapter \ref{chapter_algos}, we improved upon previous work by presenting the first fully graphical proof of correctness for the algorithm, proving that strong complementarity is the key algebraic feature behind the quantum advantage in the abelian HSP.

We have remarked that our diagrammatic treatment naturally extends to the non-abelian case, and that the known intractability of the problem is more a matter of classical post-processing than an issue with the quantum part itself. We have also remarked that our approach immediately transfers to other theories possessing the required algebraic structures, and as a corollary of our work we have shown that Simon's Problem can be efficiently solved in Real Quantum Theory.

A number of questions remain open. Firstly, the group theoretic nature of the Hidden Subgroup Problem begs the question of whether strong complementarity is somehow also a necessary condition for the implementation of a suitable quantum subroutine. Secondly, it would be interesting to look at concrete implementations of our results in other theories, such as Fermionic Quantum Theory or Spekkens' Toy Model. Finally, the relationship between strong complementarity and the quantum Fourier transform prompts further investigation of the role that these algebraic structures might be playing in a number of other quantum algorithms and protocols.

One might think that a similar physical setup, with position and momentum swapped, could be used to tackle the $G = T^N$ case. However, the annihilators $\Annihil{H} \leq \integers^N$ are all infinite sub-lattices of $\integers^N$, and the classical post-processing is left with the daunting task of reconstructing one such lattice in polynomial time from polynomially many random samples. This seems to be sufficiently close to the Shortest Independent Vectors Problem---a known hard lattice problem \cite{Blomer1999}, related to other quantum-resistant lattice problems \cite{Regev2004,Regev2004a}---to suggest that solving the HSP for compact Lie subgroups of $\torusGroup{N}$ might be beyond current quantum approaches; however, a thorough investigation of this issue is left to future work.

Another research direction for the infinite abelian HSP using non-standard methods lies in its application to infinite-dimensional hyperbolic quantum theory, a non-standard model of which can be easily constructed on the same lines of $\starHilbCategory$. We remarked that hyperbolic quantum theory does not have admit enough multiplicative characters for finite abelian groups other than $\integersMod{2}^M$. However, it does admit enough multiplicative characters for the infinite abelian groups $\integers^N$, and this indicates that there could be a fully local toy model of infinite-dimensional (separable) quantum theory in which the HSP for $\integers^N$ can be efficiently solved without the requirement of non-locality. However, reasoning about non-locality in the infinite-dimensional setting is likely to be trickier than it might seem at first glance, and further pursuit of this observation is left to future work.

One might think that a physical setup similar to the one used for $G = \integers^N$, but with position and momentum swapped, could be used to tackle the $G = \torusGroup{N}$ case. However, the annihilators $\Annihil{H} \leq \integers^N$ are all infinite sub-lattices of $\integers^N$, and the classical post-processing is left with the daunting task of reconstructing one such lattice in polynomial time from polynomially many random samples. This seems to be sufficiently close to the Shortest Independent Vectors Problem---a known hard lattice problem \cite{Blomer1999}, related to other quantum-resistant lattice problems \cite{Regev2004,Regev2004a}---to suggest that solving the HSP for compact Lie subgroups of $\torusGroup{N}$ might be beyond current quantum approaches; however, a thorough investigation of this issue is left to future work.

\section*{Generalised Mermin-type non-locality}

Using phase groups and strongly complementary observables, we have fully generalised Mermin-type non-locality arguments in Chapter \ref{chapter_algos}, and we have provided the exact group-theoretic conditions required for non-locality to arise. Our results complete the line of enquiry on the connection between phase groups and non-locality started in \cite{Coecke2010a,Coecke2012c}. We have furthermore shown that all our generalised arguments can be realised in quantum mechanics, using GHZ states and appropriate phase gates.  

We have then proceeded to investigate the empirical models arising from our generalised arguments, using the sheaf-theoretic framework for non-locality and contextuality. We have shown the models to provide new instances of All-vs-Nothing arguments, and in particular to be strongly contextual. As a consequence, we have shown that the hierarchy of quantum-realisable All-vs-Nothing arguments over finite fields does not collapse.

Finally, our generalisations lead us to an extension of the quantum-classical secret sharing scheme of Hillery, Bu\v{z}ek and Berthiaume, which was originally based on Mermin's non-locality argument for qubit GHZ states. Using our results on strong contextuality, we have been able to provide device-independent security guarantees for our generalised protocol (and for the original HBB scheme as a special case).

A number of questions are left open to future investigation. Firstly, our generalised arguments are formulated for finite abelian groups, encoded by an orthonormal basis of unbiased states: an extension to arbitrary finite groups will be of interest, and more general subsets of the phase group could be considered. 

Secondly, we have restricted ourselves to the case in which one structure is commutative and has enough points. Treatment of the more general case, where both structures are allowed to be possibly non-commutative, would extend our result from traditional groups to certain quantum groups.

Thirdly, we have shown that our generalised Mermin-type arguments are All-vs-Nothing, but the converse is not true in general. It would be interesting to investigate which modifications would be necessary to extend our techniques to other families of All-vs-Nothing arguments.

Finally, the model of attack we used to provide device-independent security guarantees is somewhat more restricted than the gold standard employed in device-independent quantum cryptography. A more complete proof of security should be a priority for future developments.





\chapter*{Acknowledgements} 

The writing of this Thesis was made possible by the loving, unconditional support of my wife Sukrita, my mum and dad, my little brother Niccol\`{o} and my best friend Nicol\`{o}. The content and presentation have greatly benefited from my many conversation with Bob Coecke, Aleks Kissinger, Dominic Horsman, Fabrizio Romano Genovese, Carlo Maria Scandolo, William Zeng, Ross Duncan, Chris Heunen, Dan Marsden, Amar Hadzihasanovic, David Reutter, Samson Abramsky and Paolo Perinotti. The very possibility of pursuing this line of research was entirely thanks to Bob Coecke---who took me in despite my loud, braggadocious Italian character---and to the generous funding provided by EPSRC and Trinity College's Williams Scholarship.


\newpage

\addcontentsline{toc}{chapter}{Bibliography}

\bibliographystyle{alpha}
\newcommand{\etalchar}[1]{$^{#1}$}

\end{document}